\documentclass[a4paper, 11pt]{book}
\usepackage[utf8x]{inputenc}
\usepackage{amsmath,amssymb}
\usepackage[mathscr]{euscript}
\usepackage{graphicx}
\usepackage[all]{xy}
\usepackage{tabularx}
\makeindex

\newcommand{\zdd}{a_{12}}

\newcommand{\zcd}{a_{c2}}
\newcommand{\zdc}{a_{1c}}
\newcommand{\zcc}{a_{cc}}

\newtheorem{prop}{Proposition}[section]

\newtheorem{defi}{Definition}[section]
\newtheorem{lemm}{Lemma}[section]
\newtheorem{thm}{Theorem}[section]
\newtheorem{coro}{Corollary}[section]
\newtheorem{ex}{Example}[section]

\newcommand{\bbox}{\normalsize {}%
        \nolinebreak \hfill $\blacksquare$ \medbreak \par}
\newenvironment{proof}{\noindent\emph{Proof} ---}{\bbox\vspace{0,15cm}}
\newlength{\aaa}
\title{Renormalization of quantum field theory on curved space-times, a causal approach.}
\author{Nguyen Viet Dang}

\begin{document}
\maketitle

% Dedicace et epigraphe
%\include{dedicace}

% Remerciements
%\include{remerciements}
\pagenumbering{roman}
Tout d'abord, je tiens à remercier mon directeur de thèse 
Fr\'ed\'eric H\'elein sans qui cette thèse 
n'existerait pas, 
pour 
sa grande gentillesse, 
sa disponibilité constante 
alors qu'il est assailli 
par les 
t\^aches 
administratives,
et
sa grande 
pédagogie. 
En effet, 
je suis toujours ébloui 
par 
sa capacité à rendre les concepts 
de géométrie et 
d'analyse très difficiles 
complètement limpides. 
Je le remercie aussi 
pour son insistance 
à comprendre
les chose en profondeur, 
grâce à ça, 
la récurrence 
centrale de cette thèse 
en a bien profité.
Enfin je voudrais souligner le fait 
que Fr\'ed\'eric 
fait preuve d'un  
optimisme indéfectible 
et 
d'une grande 
joie de vivre mathématique  
puisqu'il a toujours cru que 
j'avais 
les bonnes idées et que mes fautes 
pouvaient être rattrapées 
moyennant du travail et de la technique.

Ensuite, je tiens à remercier Christian 
Brouder, 
le deuxième père spirituel de la thèse, 
pour avoir cru en moins dès le jour 
où j'ai parlé d'un article 
complètement visionnaire de Borcherds, 
et pour avoir porté ce projet 
d'éclaicir 
l'article de Borcherds 
jusqu'au bout, alors que moi 
même 
je n'y croyais plus. 
Le nombre d'heures que nous avons 
passé à faire des maths,
éplucher des articles,
faire des calculs,
est énorme.
C'est lui qui m'a initié véritablement 
à cette magnifique approche causale de
la théorie des champs.
J'aimerais aussi remercier les membres
de l'IMPMC pour avoir rendu agréable 
mes passages fréquents à Jussieu
pour travailler avec Christian Brouder.

S'il y a un autre physicien que je tiens à
remercier, c'est Bertrand Delamotte
qui m'a appris les bases de la renormalisation
perturbative et aussi non perturbative à la Wilson
et qui est l'un des seuls à expliquer la théorie
quantique des champs de façon compréhensible
et exploitable par un pauvre matheux comme moi.

Je voudrais remercier les rapporteurs Christian Gérard 
et Klaus Fredenhagen pour l'énorme travail qu'ils ont abattu,
c'est quand même pas de la tarte 
de se coltiner les 200 et quelques
pages de cette thèse, et d'avoir
accepté d'être présents le
jour de la soutenance.
Je remercie également Sylvie Paycha
pour m'avoir invité à
participer à une conférence passionnante 
à Potsdam, de s'intéresser
à mon travail, de me proposer 
une collaboration qui s'annonce 
passionante
et enfin d'avoir accepté d'être dans mon jury de thèse.

Je tiens à remercier Louis Boutet de Monvel de
m'avoir appris tant de choses en M2, lorsque je suivais
son cours, qui se sont révélées cruciales
pour la thèse: 
analyse microlocale, 
opérateurs pseudodifférentiels,
techniques de la théorie de l'index,
l'importance énorme 
de la géométrie symplectique
en analyse et bien d'autres conseils
qu'il m'a donné lorsque je passais 
lui poser des questions. 
C'est lui qui est 
à l'origine 
du Chapitre 9 de la thèse
où on fait des développements asymptotiques de 
distributions et qu'on 
régularise méromorphiquement.

Je voudrais aussi remercier mes autres maîtres en mathématiques:
Alain Chenciner pour m'avoir expliqué que les maths, 
c'est dur pour tout le monde,
m'avoir appris à dessiner et penser géométrique (j'ai encore en mémoire ce fabuleux
cours de mécanique avec ces fameux tores lagrangiens...), aussi Daniel
Bennequin pour m'avoir initié à la géométrie symplectique, qui
m'a suggéré d'étudier les équations
Fuchsiennes que l'on retrouve au Chapitre 9 et qui a conjecturé l'un 
des 
résultats 
les plus original 
de la thèse: 
le front d'onde des amplitudes de Feynman 
est Lagrangien. Je lui dédie 
le Chapitre 7 de ma thèse
qui confirme cette intuition.

Je remercie tous les membres
de l'équipe géométrie et dynamique, membres
du GDR renorm, chercheurs, 
avec lesquels 
j'ai pu avoir des discussions intéressantes:
Alessandra Frabetti, Frédéric Paugam, Sergei Barannikov,
Paul Laurain, Thierry de Pauw, Mathieu Stiénon, Ping Xu,
Frédéric Patras, Frédéric Menous, Fabien Vignes-Tourneret,
Dominique Manchon, Volodya Roubtsov, Marc Bellon, 
Susama Argawala, Mathieu Lewin, Miguel Bermudez, 
Thierry de Pauw
et aussi les camarades et copains, copines 
avec qui j'ai passé des moments rigolos:
Pierre, Kodjo, Marie, Guillaume, Yohann, 
Hélène, Christophe, Elise, Jean Maxime, 
Elsa, Jérémy, Vincent, Violette, Romain, 
Olivier, 
va cac ban Viet:
Bang, Trung, Van, Linh, Linh, Huyen, Trang, 
Truong, anh Hung, anh Nam.

Je remercie les professeurs, instituteurs
que j'ai eu dans ma scolarité:
Monique Viot, Monsieur Griffaut,
Nicolas Choquet, 
Gilles Alozy,
Michel Wirth.

Je voudrais remercier mes parents, 
surtout
ma mère pour le sacrifice et la lutte permanente 
qu'elle mène depuis 
son départ du Vietnam parmi les boat people
jusqu'à aujourd'hui, 
pour m'avoir nourri, élevé, 
appris à parler le Vietnamien, 
pour me pousser 
à donner le meilleur de moi
même.
Je remercie mon oncle
pour 
s'être occupé de moi 
quand
j'étais petit,
pour m'avoir fait
comprendre la chance que 
j'ai d'être né en France et
pour m'avoir appris
à être drôle alors que son métier 
n'est pas toujours marrant, puis
merci à ma tante Mo Dao 
pour son humour qui déchire
et à mon cousin Vo An
qui est la seule personne au monde
à comprendre mon goût prononcé
pour le cinéma d'Art et d'Essai
(Kung Fu Panda, Transformers 123, 
Pacific Rim,
Expendables 12).  
Merci à Bac et Thi pour avoir 
toujours été des exemples
pour moi alors que c'est moi l'a\^iné,
pour tous les bons moments qu'on a passé 
ensemble à mater des séries, téléréalité, 
à faire les malins et les rigolos.

 Enfin à Elise pour tous les bons moments
passés ensembles et avec
sa merveilleuse famille  
et surtout à Tho (cop cai) 
que j'aime, qui cuisine tellement 
bien, 
qui me supporte dans 
les 
p\'eriodes
difficiles où
je perd les pédales,
que j'adore taquiner 
et 
qui donne un sens à
ma vie... anh iu em.

\tableofcontents
\cleardoublepage
\pagebreak
\chapter*{Introduction.}

In this thesis, we study
and solve the problem
of the renormalization 
of 
a perturbative quantum field 
theory
of interacting scalar
fields 
on curved space times
following the causal approach.

Quantum field theory is one of the
greatest and most successfull achievements of modern physics, 
since its numerical predictions
are probed by experiments with incredible accuracy.
Furthermore, QFT 
can be applied to many fields 
ranging from condensed matter theory, 
solid state physics to particle physics.
One of the greatest challenges for
modern mathematical physics is to unify 
quantum field theory and Einstein's 
general relativity.
This program seems
today
out of reach,
however
we can
address
the more recent
question
to first try 
to 
\textbf{define and construct}
quantum field theory 
on curved Lorentzian space times.
This problem was solved 
in the groundbreaking work
of Brunetti and Fredenhagen 
\cite{BF} in 2000.  

Their work was
motivated by the
observation that
both the conventional 
axiomatic approach to quantum field theory 
following Wightman's axioms or the usual textbook 
approach in momentum space failed to be generalized 
to curved space-times for several
obvious reasons:\\
- there is no Fourier transform
on curved space time\\
- the space time is no longer Lorentz invariant.\\
Indeed, 
the starting point of
the work \cite{BF}
was to follow 
one of the very 
first
approach
to QFT 
due to
Stueckelberg, 
which is based 
on the concept of causality.

The ideas of Stueckelberg 
were first 
understood
and developed 
by Bogoliubov (\cite{Bogoliubov}) 
and then 
by Epstein-Glaser (\cite{Epstein}, \cite{EGS}) 
(on flat space time).
In these approaches, one works
directly
in spacetime 
and the renormalization 
is formulated as a problem
of extension of distributions.
Somehow, this point of view based 
on the S-matrix formulation of QFT 
was almost
completely forgotten
by people working on QFT 
at the exception of
few people as e.g. Stora,
Kay, Wald
who made important contributions
to the topic (\cite{Popineau},\cite{Stora02}).
However, in 1996, 
a student of Wightman,
M.\,Radzikowsky 
revived the subject. 
In his thesis, 
he used microlocal analysis
for the first time 
in this context
and introduced the concept
of 
\emph{microlocal spectrum condition},
a condition on the wave front set
of the distributional
two-point function
which represents the quantum states,
which characterizes the 
quantum states of positive energy 
(named Hadamard states)
on curved space times.
In 2000, in a breakthrough paper, 
Brunetti and Fredenhagen
were able to generalize 
the Epstein-Glaser theory
on curved space times 
by relying on the fundamental contribution of 
Radzikowski. 
These results were further extended by 
Fredenhagen, Brunetti,
Hollands, Wald, Rejzner, etc. 
to Yang-Mills fields and the gravitation.
 
Let us first explain what do we mean by ``a quantum field
theory''.
\paragraph{The input data of a quantum field theory.}
Our data are a 
smooth globally hyperbolic 
oriented and time oriented manifold
$(M,g)$ and an algebra bundle $\underline{H}$
(called bundle of local fields) over $M$. 
Smooth sections of $\underline{H}$
represent polynomials of the scalar fields
with coefficients in $C^\infty(M)$.
$\underline{H}$ has in fact the structure of a Hopf algebra bundle, 
i.e. a vector bundle the fibers of which are Hopf algebras.
The natural causality structure on $M$ 
induces a 
natural partial order relation 
for elements of $M$: $x\leqslant y$
if $y$ lives in the causal future of $x$.
The metric $g$ gives a natural d'Alembertian operator 
$\square$ and we choose some distribution 
$\Delta_+\in \mathcal{D}^\prime(M^2)$ 
in such a way that:
\begin{itemize}
\item the distribution $\Delta_+$
is a bisolution
of $\square$,
\item the wave front set
and the singularity of 
$\Delta_+$ satisfy some 
specific constraints
(actually, $WF(\Delta_+)$ satisfies
the microlocal spectrum condition).
\end{itemize}
\paragraph{From the input data to 
modules living on configuration spaces and the 
$\star$ product.}
For each finite subset $I$ of the integers, 
we define the configuration space $M^I$ 
as the set
of maps from $I$ to $M$
figuring a cluster of points
in $M$ labelled
by indices of $I$. 
From the algebra bundle $\underline{H}$,
we construct
a natural infinite 
collection of $C^\infty(M^I)$-modules 
$(\mathcal{H}^I)_I$
(each $\mathcal{H}^I$ 
containing products
of fields at 
points labelled
by $I$) and define a collection
of subspaces $(V^I)_I$ of distributions 
on $M^I$ indexed by finite subsets $I$
of $\mathbb{N}$
(each $V^I$
contains the Feynman 
amplitudes).
The collections 
$(M^I)_I,(\mathcal{H}^I)_I,(V^I)_I$
enjoy the following simple property:
for each inclusion 
of finite sets of integers 
$I\subset J$ we have a corresponding 
projection 
$M^J\mapsto M^I$ and inclusions
$\mathcal{H}^I\hookrightarrow \mathcal{H}^J$, $V^I\hookrightarrow V^J$.
We can define a product $\star$
(``operator product of fields''), 
which to a pair of elements $A,B$
in a subset of 
$\left(\mathcal{H}^I\otimes_{C^\infty(M^I)} V^I\right)\times\left(\mathcal{H}^J\otimes_{C^\infty(M^J)}V^J\right)$ 
where $I,J$
are disjoint finite subsets of $\mathbb{N}$, assigns
an element in $\mathcal{H}^{I\cup J}\otimes_{C^\infty(M^{I\cup J})}V^{I\cup J}$.
The product
$\star$ is defined 
by some combinatorial formula 
(which translates the ``Wick theorem'' and is 
equivalent 
to a Feynman diagrammatic expansion)
which involves powers of
$\Delta_+$.
The partial order on $M$ induces
a partial order $\leqslant$ between elements
$A,B$ in $\mathcal{H}^I\times\mathcal{H}^J$ for all $I,J$.
\paragraph{The definition of a quantum field theory.}
A quantum field theory is 
a collection $T_I$ of morphisms of $C^\infty(M^I)$-modules:
$$T_I:\mathcal{H}^I\otimes_{C^\infty(M^I)} V^I\mapsto \mathcal{H}^I\otimes_{C^\infty(M^I)} V^I ,$$
which satisfies
the following axioms
\begin{enumerate}
\item $\forall \vert I\vert\leqslant 1, T_I $ is the identity map, 
\item \emph{the Wick expansion property} 
which generalizes the Wick theorem,
\item \textbf{the causality equation} which reads
$\forall A,B$ s.t. $B\nleqslant A$ 
\begin{equation}\label{caus}
T(AB)=T(A)\star T(B).
\end{equation}
\end{enumerate}
The maps $T_I$ can be interpreted
as the time ordering operation of 
Dyson.
The main problem is to find a solution 
of the equation (\ref{caus}). This solution
turns out to be non unique, actually
all solutions of this
equation
are related by the
renormalization
group of
Bogoliubov (\cite{Bogoliubov},\cite{BrouderQFT}).
\paragraph{Renormalization as the problem of making sense of the above definition.}
We denote by $d_n$ the thin diagonal in $M^n$
corresponding to $n$ points collapsing over one
point.
From the previous 
axioms, we prove that  
$T_n|_{M^n\setminus d_n}$ is a linear combination of products  
of $T_I,I\varsubsetneq \{1,\cdots,n\}$
with coefficients in $C^\infty(M^n\setminus d_n)$.
So we encounter two problems:\\
1) Since the elements $T_I$ are $\mathcal{H}$-valued distributions, we must 
justify that these products of distributions make sense in
$M^n\setminus d_n$.\\
2) Even if the product makes sense $T_n$ is still not defined over $M^n$, thus we must extend $T_n$ on $M^n$.

\paragraph{Contents of the Thesis.}
In Chapter 1, 
we address the
second of the previous questions
of defining $T_n$ on $M^n$, 
which amounts to extend 
a distribution $t$ defined on
$M\setminus I$ where $M$ is a smooth manifold and
$I$ is a closed 
embedded submanifold. 
We give a geometric definition
of scaling transversally to the submanifold $I$ 
and 
of a weak homogeneity which are completely intrinsic
(i.e. they do not depend on the choice of
local charts). 
Our definition of weak homogeneity 
follows \cite{KK} and \cite{Meyer} 
and slightly differs from the definition
of \cite{BF} which uses the Steinman scaling degree.
We prove that if a distribution 
$t$ is in $\mathcal{D}^\prime(M\setminus I)$ 
and is weakly homogeneous 
of degree $s$ then it
has an extension 
$\overline{t}\in\mathcal{D}^\prime(M)$ 
which is 
weakly homogeneous of degree $s^\prime$ for all $s^\prime <s$.
The extension sometimes 
requires 
a renormalization 
which is a subtraction of
distributions
supported on $I$ i.e.
local counterterms. 
The main difference with the work \cite{BF} 
is that 
we only have one definition
of weak homogeneity 
and we use a continuous partition of unity.  
This chapter does not rely
on microlocal analysis.

In Chapter 2, in order to solve the first problem 
of defining $T_n$,
we must explain
why the product of the
$T_I$'s in the formula
which gives $T_n$
makes sense and  
this is possible 
under some specific conditions
on the wave front sets
of the coefficients of 
the $T_I$'s.
So we are led to study
the
wave front sets
of
the extended distributions
defined
in Chapter 1.
We find a
geometric condition 
on $WF(t)$
named \emph{soft landing condition} 
which
ensures that
the wave front of the extension
is controlled.
However this geometric condition
is not sufficient and we explain 
this by a counterexample.
We also give a geometric
definition of local counterterms
associated to a distribution $t$,
which 
generalizes 
the counterterms
of QFT textbooks in
the context of 
curved space times.  
We show that the soft landing condition is equivalent to
the fact that the local counterterms
of $t$ are smooth functions
multiplied by distributions localized on the diagonal,
i.e. they have a 
specific structure 
of finitely generated
module over the ring 
$C^\infty(I)$.
The new features of this Chapter
are the soft landing
condition which does not exist in
the literature (only implicit in \cite{BF}),
the definition
of local counterterms
associated to $t$ and our
theorem 
which proves
that under certain conditions
local counterterms
are conormal distributions. 
Finally,
our counterexample
explains
why in \cite{BF}, 
the authors
impose certain
microlocal conditions
on the unextended distribution
$t$ in order to
control the wave front set
of the extension.

In chapter 3, we prove that if 
we add one supplementary boundedness 
condition on $t$ i.e. if
$t$ is weakly homogeneous in some
topological space
of distributions 
with prescribed wave front set, 
then the 
wave front $WF(\overline{t})$ 
of the extension is contained 
in the smallest possible set
which is the union of the closure of the
wave front of 
the unextended distribution $\overline{WF(t)}$
with the conormal $C$ of $I$.
Chapter 3 differs from \cite{BF} 
by the 
fact that 
we estimate 
$\overline{WF(t)}$ also in 
the case of renormalization 
with counterterms and
our proof is much more detailed.

In chapter 4, we manage to prove 
that
the conditions of Chapter 3
can be made
completely
geometric and
coordinate invariant. 
We also prove the 
boundednes
of the product
and the pull-back
operations on
distributions
in suitable microlocal topologies. 
Then we conclude Chapter 4
with the following theorem:
if $t$ is microlocally
weakly
homogeneous of degree $s\in\mathbb{R}$
then a ``microlocal extension'' $\overline{t}$
exists with minimal wave front set 
in $\overline{WF(t)}\cup C$
and $\overline{t}$ is microlocally
weakly
homogeneous of degree $s^\prime$
for all $s^\prime< s$.
Chapter 4 improves
the results of H\"ormander 
on products and pull-back of distributions
since we prove that these operators 
are bounded
maps for the 
suitable microlocal topologies. 
This seems to be a new result since in the literature
only the sequential continuity of 
products and 
pull-back are proved.

In Chapter 5, 
we construct the
two point function $\Delta_+$
which is a distributional 
solution of the wave equation
on $M$. We prove that $WF(\Delta_+)$
satisfies the microlocal
spectrum condition
of Radzikowski and finally
we establish that 
$\Delta_+$ is ``microlocally weakly
homogeneous'' of degree $-2$.
Chapter 5 contains a 
complete mathematical
justification of the 
Wick
rotation 
for which an explicit
reference is missing although
the idea of its proof 
is sketched in \cite{Taylor}.
We also explicitly 
compute the
wave front set
of the
holomorphic family
$Q^s(\cdot+i0\theta)$ which
cannot be found in \cite{Hormander},
(we only found
a computation of 
the \textbf{analytic}
wave front set
--in the sense of Sato-Kawai-Kashiwara--
of $Q^s(\cdot+i0\theta)$ 
in \cite{KKK} p.~90 example 2.4.3).
Finally, our  
proof that the
wave front set
of $\Delta_+$
(constructed as a 
perturbative 
series \`a la Hadamard)
satisfies the
microlocal spectrum condition
seems to be missing 
in the literature.
The construction appearing in \cite{GW}
is not
sufficient to
prove that
$\Delta_+$
is microlocally weakly
homogeneous of degree 
$-2$.

Chapter 6 is the final piece of this
building. 
Inspired
by the work of Borcherds,
we quickly give our definition
of a quantum field theory
using the convenient language of
Hopf algebras then we 
state the problem of defining
a quantum field theory
as equivalent
to the problem
of solving the equation
(\ref{caus}) in $T$ 
recursively in $n$
on all configuration spaces $M^n$.
We prove this recursively 
using all tools developed 
in the previous chapters,
a careful 
partition of the
configuration space
generalizing ideas
of R.\, Stora to the case
of curved space times 
and an idea of 
polarization of wave front 
sets which translates
microlocally the idea 
of positivity of energy.

Chapter 7 solves a
conjecture
of Bennequin and
gives a nice
geometric interpretation
of
the wave front set of any
Feynman amplitude:
\begin{itemize}
\item it is parametrized
by a Morse family,
\item it is a union of
smooth Lagrangian
submanifolds of the
cotangent space of configuration space.
\end{itemize}

In Chapter 8, 
which can be read independently 
of the rest
except Chapter 1, 
using the language of currents, we treat
the problem
of preservation of symmetries 
by the extension procedure.
Indeed, renormalization can break 
the symmetries 
of the
unrenormalized objects
and the fact that
renormalization does not commute
with the action
of vector fields from
some Lie algebra of symmetries 
is called anomaly
and is measured by the 
appearance of local counterterms,
which are far reaching generalizations
of the notion of residues coming
from algebraic geometry,
(but generalized 
here to the current theoretic setting).

Finally, in chapter 9 
we revisit the extension problem
from the point of view of meromorphic regularization.
We prove that under certain conditions on distributions,
they can be meromorphically regularized then the extension
consists in a subtraction of poles which are also local counterterms.
To conclude this last Chapter, 
we 
introduce a lenght scale 
$\ell$
in the meromorphic 
renormalization and 
we prove that scaling  
in $\ell$
only gives polynomial 
divergences in
$\log\ell$.

\chapter{The extension of distributions.}
\pagenumbering{arabic}
\section{Introduction.}
In the Stueckelberg  (\cite{StueckelbergP})
approach
to quantum field theory,
renormalization  
was formulated as a problem
of division
of distributions.
For Epstein--Glaser (\cite{Epstein}, \cite{EGS})
, Stora (\cite{Popineau},\cite{Stora02}), 
and implicitly in
Bogoliubov (\cite{Bogoliubov}),
it was formulated as a problem 
of extension of distributions,
the latter
approach
is more general
since the
ambiguity
of the extension
is described
by the renormalization
group.
This procedure 
was implemented 
on arbitrary manifolds 
(hence for curved Lorentzian spacetimes) 
by Brunetti and Fredenhagen 
in their groundbreaking paper of 
$2000$ \cite{BF}. 
However, in the mathematical literature, 
the problem of extension of distributions 
goes back at least to the work of 
Hadamard and Riesz on hyperbolic equations 
(\cite{Riesz},\cite{Hadamard}).
It became a central argument 
for the proof of a conjecture of Laurent Schwartz (\cite{Schwartz} p.~126,\cite{Malbour}): 
the problem was to find a fundamental solution $E$ for a linear PDE with constant coefficients in $\mathbb{R}^n$, which means solving the equation $PE=\delta$ in the distributional sense. By Fourier transform, this is equivalent to the problem of extending $\widehat{P}^{-1}$ which is a honest smooth function on $\mathbb{R}^n\setminus\{\widehat{P}=0\}$ as a distribution on $\mathbb{R}^n$,
in such a way that $\widehat{P}\widehat{P}^{-1}=1$ which makes
the division a particular case
of an extension.
This problem 
set 
by Schwartz
was solved 
positively 
by Lojasiewicz and H\"ormander (\cite{Hormander},\cite{Loja}). 
Recently, the more general 
extension problem 
was revisited in mathematics 
by Yves Meyer in his wonderful book \cite{Meyer}. 
In \cite{Meyer}, Yves Meyer 
also explored some deep relations 
between the extension problem and 
Harmonic analysis 
(Littlewood--Paley and Wavelet decomposition). 
The extension problem was solved in \cite{Meyer} on 
$\left(\mathbb{R}^n\setminus\{0\}\right)$.
For the need of quantum field theory, 
we will extend 
his method 
to manifolds. 
In order to renormalize, one should find 
some way of measuring the wildness 
of the singularities of distributions. 
Indeed, we need to impose 
some growth 
condition on distributions 
because distributions cannot be extended in general!
We estimate the 
wildness 
of the singularity 
by first defining 
an adequate notion of 
scaling with respect 
to a closed embedded
submanifold $I$ of a given manifold $M$, 
as done by Brunetti--Fredenhagen \cite{BF}.
On $\mathbb{R}^{n+d}$ viewed 
as the cartesian product 
$\mathbb{R}^n\times \mathbb{R}^d$, 
the scaling is clearly defined by homotheties 
in the variables corresponding to the second factor $\mathbb{R}^d$. 
We adapt the definition of Meyer \cite{Meyer} in these variables
and define the space of 
weakly homogeneous distributions of degree $s$ which we 
call $E_s$.
 
 We are able to represent all elements of 
$E_s$ which are
defined on $M\setminus I$ 
through a decomposition formula 
by a family 
$\left(u^\lambda\right)_{\lambda\in(0,1]}$ satisfying some specific hypothesis. 
The distributions $\left(u^\lambda\right)_{\lambda\in(0,1]}$ are the building blocks of the $E_s$
and are the key for the renormalization.
We establish
the following
correspondence
\begin{eqnarray}
\left(u^\lambda\right)_{\lambda\in(0,1]}\longmapsto \int_0^1 \frac{d\lambda}{\lambda} \lambda^s(u^\lambda)_{\lambda^{-1}} + \text{nice terms},\\
t\in E_s\longmapsto \left(u^\lambda\right)_{\lambda\in(0,1]} \text{ where } u^\lambda=\lambda^{-s}t_\lambda\psi,
\end{eqnarray}
the nice terms are distributions
supported on the complement
of $I$.

 However this scaling is only defined in local charts and a scaling around a 
submanifold $I$
in a manifold $M$ depends on the choice of an Euler 
vector field. 
Thus we propose a geometrical definition 
of a class of Euler vector fields:
to any closed embedded submanifold $I\subset M$, 
we associate the 
$\textbf{ideal }\mathcal{I}$ of smooth functions vanishing on $I$. 
A vector field $\rho$ is called Euler vector field if 
\begin{equation}
\forall f\in\mathcal{I},  \rho f-f\in\mathcal{I}^2. 
\end{equation} 
This definition is clearly intrinsic.   
We prove that all scalings are equivalent hence all spaces 
of weakly homogeneous distributions are equivalent 
and that our definitions are in fact independent 
of the choice of Euler vector fields. Actually, 
we prove that 
all Euler vector fields are locally conjugate
by a local diffeomorphism which fixes the submanifold $I$.
So it is enough to study both $E_s$ and the 
extension problem in 
a local chart. 
Meyer and Brunetti--Fredenhagen make use of a dyadic decomposition. 
We use instead a
$\textbf{continuous partition of unity}$ 
which is a continuous analog 
of the Littlewood--Paley 
decomposition.   
The continuous partition of unity 
has many advantages 
over the discrete approaches:
1) it provides a direct connection with the theory of Mellin transform, 
which allows to easily 
define meromorphic regularizations;
2) it gives elegant formulas 
especially for the poles and residues
appearing in the 
meromorphic regularization (see Chapter $7$);
3) it is well suited
to the study of anomalies (see Chapter $6$). 
\paragraph{Relationship with other work.}
In Brunetti--Fredenhagen 
\cite{BF}, 
the scaling 
around manifolds
was also defined
but they used
two different 
definitions of scalings,
then they
showed that these
actually
coincide,
whereas 
we only give one
definition
which is geometric.
In mathematics,
we also found
some interesting
work by Kashiwara--Kawai,
where
the concept 
of weak homogeneity
was also defined 
(\cite{KK} Definition $(1.1)$ p.~22).
\section{Extension and renormalization.}
\subsection{Notation, definitions.}
We work in $\mathbb{R}^{n+d}$ with coordinates $(x,h)$, $I=\mathbb{R}^n\times\{0\}$ is the linear subspace $\{h=0\}$.
For any open set $U\subset\mathbb{R}^{n+d}$, we denote by $\mathcal{D}(U)$ the space of 
test functions supported on $U$ and for all compact $K\subset U$,
we denote by $\mathcal{D}_K(U)$ the subset of all test functions in $\mathcal{D}(U)$ supported on $K$.
We also use the seminorms:
$$\forall \varphi\in\mathcal{D}(\mathbb{R}^{n+d}), \pi_k(\varphi):=\sup_{\vert\alpha\vert\leqslant k} \Vert \partial^\alpha\varphi\Vert_{L^\infty(\mathbb{R}^{n+d})},$$
$$\forall \varphi\in C^\infty(\mathbb{R}^{n+d}),\forall K\subset \mathbb{R}^d, \pi_{k,K}(\varphi):=\sup_{\vert\alpha\vert\leqslant k} \sup_{x\in K}\vert \partial^\alpha\varphi(x)\vert.$$
We denote by $\mathcal{D}^\prime(U)$ the space of distributions defined on $U$.
The duality
pairing
between
a
distribution $t$ and 
a test function $\varphi$
is denoted
by 
$\left\langle t,\varphi\right\rangle$.
For a function,
we define
$\varphi_\lambda(x,h)
=\varphi(x,\lambda h)$. 
For the
vector field
$\rho=h^j\frac{\partial}{\partial h^j}$,
the
following
formula
$$\varphi_\lambda=e^{(\log\lambda)\rho\star}\varphi,$$
shows the relation
between $\rho$
and the scaling. 
Once we have defined 
the scaling
for test functions,
for any distribution $f$, 
we define the scaled
distribution $f_\lambda$:
$$\forall\varphi\in\mathcal{D}(\mathbb{R}^{n+d}),\left\langle f_\lambda,\varphi \right\rangle=\lambda^{-d}\left\langle f,\varphi_{\lambda^{-1}} \right\rangle.$$
If $f$ were a function, this definition would coincides with the naive scaling
$f_\lambda(x,h)=f(x,\lambda h)$.

We give a definition of 
weakly homogeneous 
distributions in flat space
following \cite{Meyer}.
We call a subset $U\subset \mathbb{R}^{n+d}$ $\rho$-convex
if $(x,h)\in U \implies \forall\lambda \in(0,1], (x,\lambda h)\in U$.
We insist on the fact that since we pick $\lambda>0$, a $\rho$-convex domain
may have \emph{empty intersection} with $I$. 
\begin{defi}
Let $U$ be an arbitrary $\rho$-convex open subset of $\mathbb{R}^{n+d}$.  
$E_s(U)$ is defined as the space of distributions $t$ such that $t\in \mathcal{D}^\prime(U)$ and 
$$\forall\varphi\in \mathcal{D}(U),\exists C(\varphi), \sup_{\lambda\in(0,1]} \vert\lambda^{-s}\left\langle t_\lambda,\varphi \right\rangle \vert\leqslant C(\varphi).$$ 
\end{defi}
In the
quantum field theory
litterature, the wildness
of distributions is measured
by the Steinman
scaling degree.
We prefer the definition
of Meyer, which exploits
the properties of bounded
sets 
in the space of distributions
(this is related to bornological properties
of $\mathcal{D}^\prime(U)$).

We denote by $\frac{d\lambda}{\lambda}$ the multiplicative measure on $[0,1]$. 
We shall now give a definition of a
class of maps $\lambda\mapsto u^\lambda$ with value in
the space of distributions.
\begin{defi}
For all $1\leqslant p\leqslant \infty$, we define $L_{\frac{d\lambda}{\lambda}}^p([0,1],\mathcal{D}^\prime(U))$ as the space of families $(u^\lambda)_{\lambda\in(0,1]}$ of distributions such that
\begin{equation}
\forall\varphi\in \mathcal{D}(U), \lambda\mapsto  \left\langle u^\lambda,\varphi\right\rangle \in L_{\frac{d\lambda}{\lambda}}^p([0,1],\mathbb{C}). 
\end{equation}
\end{defi}
\paragraph{The H\"ormander trick.}
We recall here the basic idea of Littlewood--Paley analysis (\cite{Meyer} p.~14). Pick a function $\chi$ which depends only on $h$ such that $\chi=1$ when $\vert h\vert\leqslant 2$ and $\chi=0$ for $\vert h\vert\geqslant 3$.
\begin{figure} %on ouvre l'environnement figure
\begin{center}
\includegraphics[width=8cm]{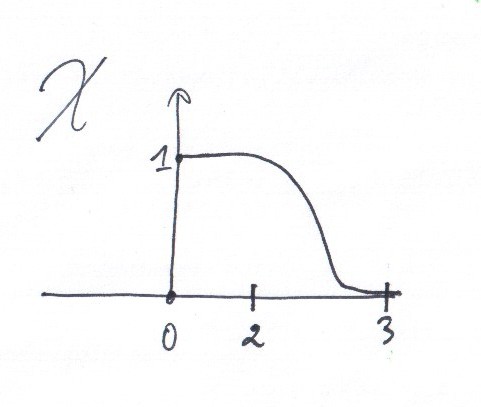} %ou image.png, .jpeg etc.
\caption{The function $\chi$ of Littlewood--Paley theory.} %la légende
%l'étiquette pour faire référence à cette image
\end{center}
\end{figure} %on ferme l'environnement figure 
The Littlewood--Paley function $\psi(\cdot)=\chi(\cdot)-\chi(2\cdot)$ is supported on the annulus $1\leqslant \vert h\vert\leqslant 3$.
Then the idea is to rewrite the plateau function $\chi$ using the trick of the telescopic series
$$\chi=  \chi(\cdot)-\chi(2\cdot)+\cdots+\chi(2^j\cdot)-\chi(2^{j+1}\cdot)+\cdots $$
and deduce a dyadic partition of unity $$1=\left(1-\chi\right) + \sum_{j=0}^\infty \psi(2^j.)$$
Our goal in this paragraph is to derive a continuous analog of the dyadic partition of unity.
% In this section, we work in a given fixed compact subset of the form $K=K_1\times \{ \vert h\vert\leqslant a\}\subset \mathbb{R}^{n+d}$, the compact set $K$ is $\rho$-convex. 
%Warning: All test functions $\chi,\varphi$ presented in this section are supported in the respective compact sets $K^\prime,K$ where $K\subset K^\prime$ because we only treat the local problem in this part.
Let $\chi\in C^\infty(\mathbb{R}^{n+d})$ such that $\chi=1$ in a neighborhood $N_1$ of $I$ and $\chi$ vanishes outside a neighborhood $N_2$ of $N_1$. This implies $\chi$ satisfies the following constraint: 
for all compact set $K\subset \mathbb{R}^{n},\exists (a,b)\in \mathbb{R}^2$ such that $b>a>0$ and $\chi|_{(K\times\mathbb{R}^d)\cap \{\vert h\vert\leqslant a\}}=1, \chi|_{(K\times\mathbb{R}^d)\cap \{\vert h\vert\geqslant b\}}=0$.
\begin{figure} %on ouvre l'environnement figure
\begin{center}
\includegraphics[width=10cm]{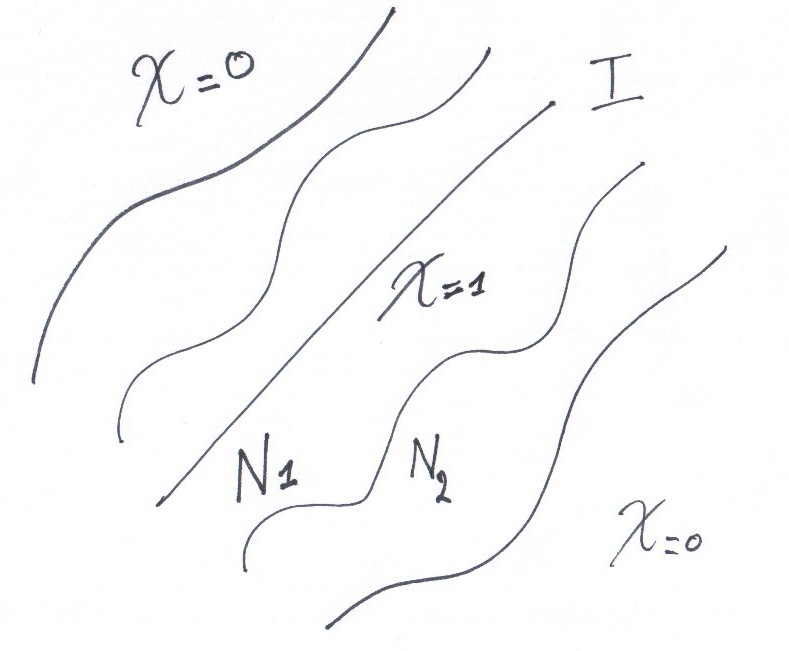} %ou image.png, .jpeg etc.
%l'étiquette pour faire référence à cette image
\end{center}
\end{figure} %on ferme l'environnement figure 
%
% We fix the compact set $K_2=\{\vert h\vert\leqslant a\}$ and $K_1$ is any closed ball in $\mathbb{R}^n$.
We find a convenient formula (inspired by \cite{Hormanderwave} equation (8.5.1) p.~200 and \cite{Meyer}
Formula (5.6) p.~28) for $\chi$ as an integral over a scale space indexed by $\lambda\in (0,1]$.
First notice that $\chi(x,\frac{h}{\lambda})\rightarrow_{\lambda\rightarrow 0} 0$ in $L^1_{loc}$. 
We repeat the Littlewood Paley trick in the continuous setting: 
$$\chi(x,h)=\chi(x,h)-0= \int_{0}^1 \frac{d\lambda}{\lambda}\lambda\frac{d}{d\lambda}\left[\chi(x,\lambda^{-1}h)\right]= \int_{0}^1\frac{d\lambda}{\lambda} \left(-\rho\chi\right)(x,\lambda^{-1}h) $$
Set 
\begin{equation}
\psi=-\rho\chi.
\end{equation}
Notice an important property of $\psi$: on each compact set $K\subset\mathbb{R}^{n}$, $\exists (a,b)\in\mathbb{R}^2$ such that $\psi|_{(K\times\mathbb{R}^d)}$ is
supported on the annulus $(K\times\mathbb{R}^d)\cap \{a\leqslant \vert h\vert\leqslant b \}$.
\begin{figure} %on ouvre l'environnement figure
\begin{center}
\includegraphics[width=10cm]{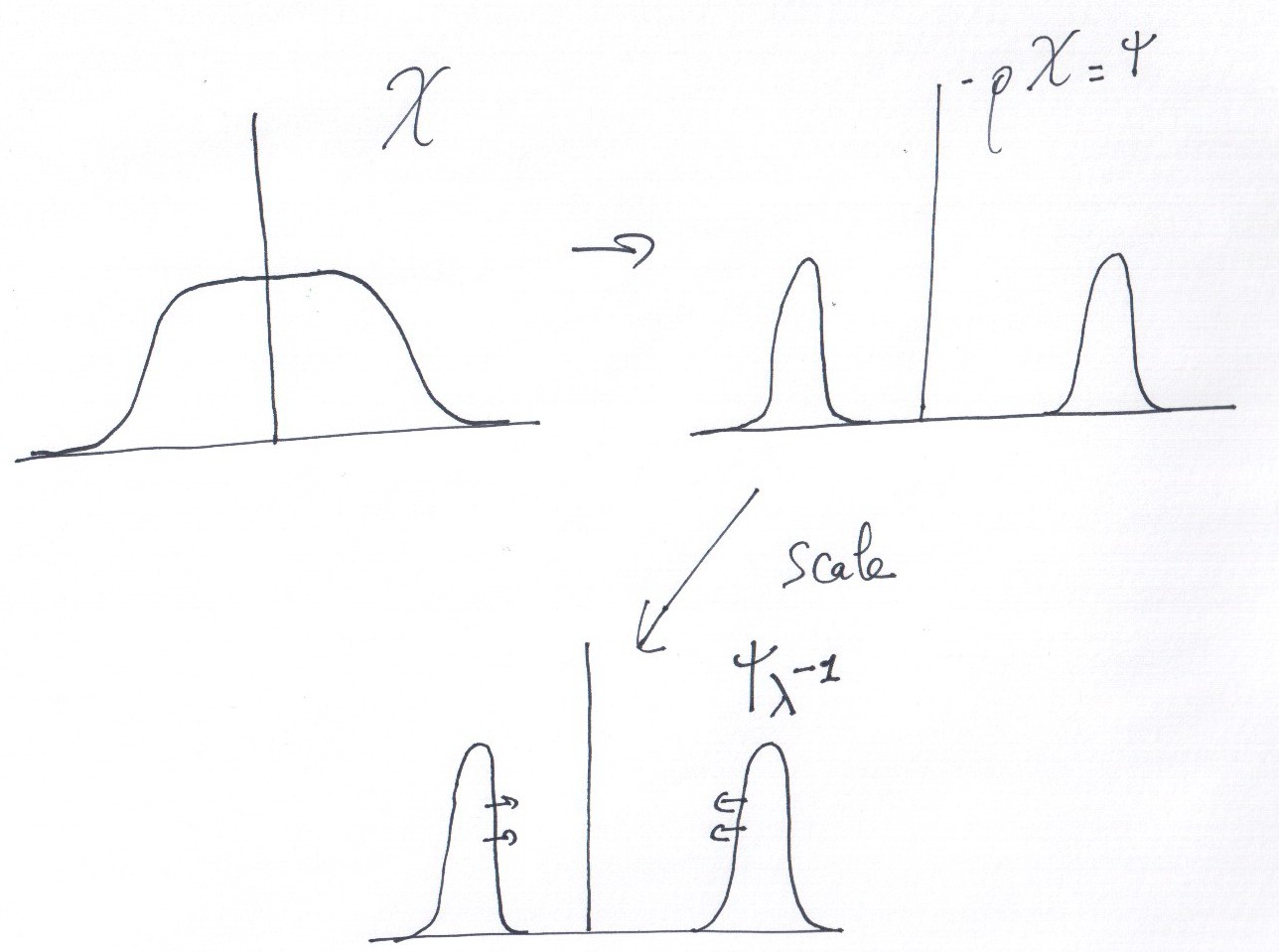} %ou image.png, .jpeg etc.
\caption{The function $\chi$, the function $\psi$ and the scaled $\psi_{\lambda^{-1}}$.} %la légende
%l'étiquette pour faire référence à cette image
\end{center}
\end{figure} %on ferme l'environnement figure 
We obtain the formula
\begin{equation}
1=(1-\chi)+\int_{0}^1\frac{d\lambda}{\lambda} \psi_{\lambda^{-1}},  
\end{equation}
which for the moment only has a heuristic meaning.
\begin{figure} %on ouvre l'environnement figure
\begin{center}
\includegraphics[width=8cm]{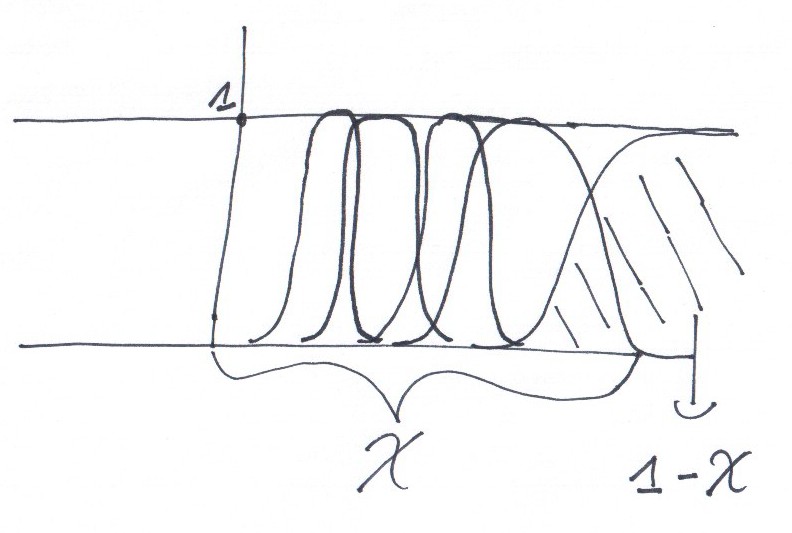} %ou image.png, .jpeg etc.
\caption{Partition of unity.} %la légende
%l'étiquette pour faire référence à cette image
\end{center}
\end{figure} %on ferme l'environnement figure 
The next proposition gives a precise meaning to the heuristic formula and gives a candidate formula for the extension problem.
\begin{prop}
Let $\chi\in C^\infty(\mathbb{R}^{n+d})$ such that $\chi=1$ in a neighborhood $N_1$ of $I$ and $\chi$ vanishes outside a neighborhood $N_2$ of $N_1$ and let $\psi=-\rho\chi$. Then for all
$\varphi\in \mathcal{D}(\mathbb{R}^{n+d})$ such that $\varphi=0$ in a neighborhood of $I=\{h=0\}$, 
we find
\begin{equation}
\left\langle t,\varphi\right\rangle =\int_{0}^1\frac{d\lambda}{\lambda} \left\langle t \psi_{\lambda^{-1}} ,\varphi\right\rangle + \left\langle t,(1-\chi)\varphi\right\rangle. 
\end{equation}
\end{prop}
The formula 
$t=\int_{0}^1\frac{d\lambda}{\lambda} \left\langle t \psi_{\lambda^{-1}} ,\varphi\right\rangle + \left\langle t,(1-\chi)\varphi\right\rangle$
was inspired 
by Formula
(5.8), (5.9)
in \cite{Meyer}
p.~28.

\begin{proof}
Let $\delta>0$ such that $\varphi=0$ when $\vert h\vert\leqslant \delta$.
We can find $0<a<b$ such that
$\left[\vert h\vert>b\implies \chi=0\right]$ and $\left[\vert h\vert>b\implies-\rho\chi=\psi=0\right]$.
Hence $ \text{supp }\psi(x,\frac{h}{\lambda})\subset \{\vert h\vert\leqslant \lambda b\} $
which implies $\forall\lambda \leqslant \frac{\delta}{b}, \varphi(x,h)\psi(x,\frac{h}{\lambda}) =0$.
We have the relation $\varphi=\varphi(1-\chi)+\varphi\chi=\int_{\frac{\delta}{b}}^1 \frac{d\lambda}{\lambda}\psi_{\lambda^{-1}}\varphi + \varphi(1-\chi)$ where the integral is well defined, we thus
deduce $\forall \varepsilon\in[0,\frac{\delta}{b}]$
$$\varphi\chi= \int_\varepsilon^1\frac{d\lambda}{\lambda}  \underset{=0 \text{ for }\lambda\in[\varepsilon,\frac{\delta}{b}]}{\underbrace {\psi_{\lambda^{-1}}\varphi}}=\int_{\frac{\delta}{b}}^1 \frac{d\lambda}{\lambda}\psi_{\lambda^{-1}}\varphi  $$
where the product makes perfect sense as a product of smooth functions, hence
$$\left\langle t\chi,\varphi\right\rangle =\left\langle t,\chi\varphi\right\rangle=\left\langle t,\int_\varepsilon^1  \frac{d\lambda}{\lambda}\psi_{\lambda^{-1}}\varphi\right\rangle= \int_{\varepsilon}^1\frac{d\lambda}{\lambda} \left\langle t\psi(\frac{h}{\lambda}) ,\varphi\right\rangle$$ 
$$=\int_{\frac{\delta}{b} }^1\frac{d\lambda}{\lambda} \left\langle t\psi(\frac{h}{\lambda}) ,\varphi\right\rangle=\int_{0 }^1\frac{d\lambda}{\lambda} \left\langle t\psi(\frac{h}{\lambda}) ,\varphi\right\rangle $$ where we can safely interchange the integral 
and the duality bracket.
\end{proof}

\paragraph{Another interpretation of the H\"ormander formula.}
The H\"ormander formula gives a convenient way to write $\chi-\chi_{\varepsilon^{-1}}$.
$$\chi-\chi_{\varepsilon^{-1}}=\int_\varepsilon^1 \frac{d\lambda}{\lambda} \psi_{\lambda^{-1}} $$
then noticing that when $\varepsilon>0$, for all $\lambda\in[\varepsilon,1]$, $\psi_{\lambda^{-1}}$ is supported 
on the complement of a neighborhood of $I$, this implies that for all test functions $\varphi\in\mathcal{D}(\mathbb{R}^{n+d})$, for all $\varepsilon>0$,
we have the nice identity:
$$\int_\varepsilon^1 \frac{d\lambda}{\lambda} \left\langle t\psi_{\lambda^{-1}},\varphi\right\rangle =\left\langle t\left(\chi-\chi_{\varepsilon^{-1}}\right),\varphi\right\rangle .$$
\begin{figure} %on ouvre l'environnement figure
\begin{center}
\includegraphics[width=8cm]{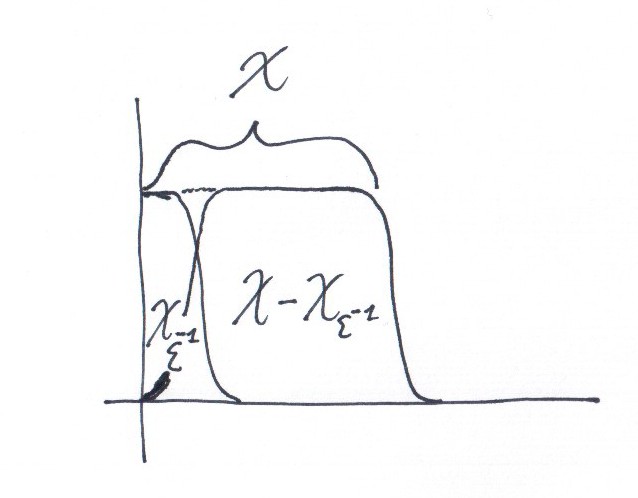} %ou image.png, .jpeg etc.
\caption{$\chi-\chi_{\varepsilon^{-1}}$.} %la légende
%l'étiquette pour faire référence à cette image
\end{center}
\end{figure} %on ferme l'environnement figure 
Now if the function $\left\langle t\psi_{\lambda^{-1}},\varphi\right\rangle $ is integrable on $[0,1]$ w.r.t. the measure $\frac{d\lambda}{\lambda}$, 
the \textbf{existence} of the integral $\int_0^1 \frac{d\lambda}{\lambda} \left\langle t\psi_{\lambda^{-1}},\varphi\right\rangle $ will \textbf{imply} that the \textbf{limit}
\begin{equation}\label{limitformula}
\lim_{\varepsilon\rightarrow 0} \left\langle t\left(\chi-\chi_{\varepsilon^{-1}}\right),\varphi\right\rangle 
\end{equation}
exists. In the next sections, 
we prove that 
when the distribution $t$ is in $E_s$ for $s+d>0$, the integral formula $\int_\varepsilon^1 \frac{d\lambda}{\lambda} \left\langle t\psi_{\lambda^{-1}},\varphi\right\rangle $ converges when $\varepsilon\rightarrow 0$. 
Thus the limit (\ref{limitformula}) exists.
However, when $t\in E_s$ when $s+d<0$, 
we must modify the formula $\int_\varepsilon^1 \frac{d\lambda}{\lambda} \left\langle t\psi_{\lambda^{-1}},\varphi\right\rangle $, which is divergent when $\varepsilon\rightarrow 0$, by subtracting a \textbf{local counterterm}
$\left\langle c_\varepsilon,\varphi \right\rangle$ where $(c_\varepsilon)_\varepsilon$ is a family of distribution \textbf{supported} on $I$ such that
the limit
\begin{equation}\label{renormlimitformula}
\lim_{\varepsilon\rightarrow 0} \left(\left\langle t\left(\chi-\chi_{\varepsilon^{-1}}\right),\varphi\right\rangle-\left\langle c_\varepsilon,\varphi \right\rangle \right),
\end{equation}
makes sense.
Notice that the renormalization does not affect the original distribution $t$ on $M\setminus I$
since $c_\varepsilon$ is supported on $I$. 
\subsection{From bounded families to weakly homogeneous distributions.}

We construct an algorithm 
which starts from an arbitrary
family
of bounded
distributions $(u^\lambda)_{\lambda\in(0,1]}$ 
supported 
on some annular domain,
and builds a weakly homogeneous
distribution of degree $s$.
Actually, any distribution
which is weakly
homogeneous of degree $s$
can be reconstructed from our algorithm
as we will see in the next section.
This is the key remark
which allows us 
to solve the problem
of extension
of distributions.
In this part, we make essential use of 
the Banach Steinhaus theorem 
on the dual of a Fr\'echet space 
recalled in appendix.    
We use the notation $t_\lambda(x,h)=t(x,\lambda h)$ and $U$ is a $\rho$-convex open subset in $\mathbb{R}^{n+d}$.
\begin{defi}
A family of distributions $(u^\lambda)_{\lambda\in(0,1]}$ is called uniformly supported on an annulus domain
of $U$ if
for all compact set $K\subset \mathbb{R}^n$, there exists $0<a<b$ such that $\forall\lambda, u^\lambda|_{(K\times \mathbb{R}^d)\cap U}$ is supported in a fixed annulus $\{(x,h)|x\in K, a\leqslant\vert h\vert\leqslant b\}\cap U$.
\end{defi}
The structure theorem gives us an algorithm to construct distributions in $E_s(U)$ 
given any family of distributions $(u^\lambda)_{\lambda\in (0,1]}$ bounded in $\mathcal{D}^\prime(U\setminus I)$ and uniformly supported on an annulus domain of $U$. 

\begin{lemm}
Let $(u^\lambda)_{\lambda\in(0,1]}$ be a bounded family in $\mathcal{D}^\prime(U\setminus I)$
which is uniformly supported on an annulus domain
of $U$. 
Then the family
$\left(\lambda^{-d}u^\lambda_{\lambda^{-1}}\right)_{\lambda\in(0,1]}$ is bounded in $\mathcal{D}^\prime(U)$.
\end{lemm}
\begin{proof}
If the family $(u^\lambda)_{\lambda\in(0,1]}$ is uniformly supported on an annulus domain
of $U$, then
for all compact set $K\subset \mathbb{R}^n$, there exists $0<a<b$ such that $\forall\lambda, u^\lambda|_{(K\times \mathbb{R}^d)\cap U}$ is supported in a fixed annulus $\mathcal{A}=\{a\leqslant\vert h\vert\leqslant b\}\cap ((K\times \mathbb{R}^d)\cap U)$.
If $u^\lambda|_{(K\times \mathbb{R}^d)\cap U}$ is a bounded family of distributions supported on the $\textbf{fixed annulus}$ $\mathcal{A}=\underset{\text{compact in }\mathbb{R}^{n+d}}{\underbrace{\{a\leqslant\vert h\vert\leqslant b\}\cap (K\times \mathbb{R}^d)\cap U}}$, then the family $u^\lambda$ satisfies the following estimate by Banach Steinhaus:
$$\forall K^\prime\subset\mathbb{R}^{n+d} \text{compact},\exists (k, C), \forall \varphi\in \mathcal{D}_{K^\prime}\left(U\right), \sup_{\lambda\in(0,1]} \vert \left\langle u^\lambda, \varphi \right\rangle \vert\leqslant C  \pi_k(\varphi),$$ 
and we notice that the estimate is still valid for test functions in $C^\infty((K\times \mathbb{R}^{d})\cap U)$ (by compactness of $\mathcal{A}$):
\begin{equation}\label{est1}
\exists (k, C), \forall \varphi\in C^\infty((K\times \mathbb{R}^{d})\cap U), \sup_{\lambda\in(0,1]} \vert \left\langle u^\lambda, \varphi \right\rangle \vert\leqslant C\pi_{k\mathcal{A}}(\varphi),
\end{equation}
because $u^\lambda$ is compactly supported in the $h$ variables and $\varphi$ is compactly supported
in the $x$ variables.
For any test function $\varphi\in\mathcal{D}(U)$: $$\lambda^{-d}\vert\left\langle u^\lambda_{\lambda^{-1}},\varphi \right\rangle\vert=\lambda^{-d}\lambda^d\vert\left\langle u^\lambda,\varphi(.,\lambda.) \right\rangle\vert\leqslant  C\pi_{k\mathcal{A}}(\varphi_\lambda)$$ 
thus
\begin{equation}\label{est2}
\lambda^{-d}\vert\left\langle u^\lambda_{\lambda^{-1}},\varphi \right\rangle\vert\leqslant C\pi_{k}(\varphi)
\end{equation}
because of the estimate (\ref{est1}) on the family $(u^\lambda)_\lambda$. 
This proves that 
the family
$\left(\lambda^{-d}u^\lambda_{\lambda^{-1}}\right)_{\lambda\in(0,1]}$ 
is bounded in $\mathcal{D}^\prime(U\setminus I)$. 
\end{proof} 
\begin{coro}\label{ren1}
Let $(u^\lambda)_{\lambda\in(0,1]}$ be a bounded family in $\mathcal{D}^\prime(U\setminus I)$
which is uniformly supported on an annulus domain
of $U$.
If $s+d>0$, then the integral
\begin{equation}\label{formulaalgo1} 
\int_{0}^1\frac{d\lambda}{\lambda} \lambda^s u^\lambda_{\lambda^{-1}}  
\end{equation}
converges in $\mathcal{D}^\prime(U)$.
\end{coro}
\begin{proof}
When $s+d>0$, $\lambda\mapsto \lambda^su^\lambda_{\lambda^{-1}}=\underset{\text{integrable}}{\underbrace{\lambda^{s+d}}} \underset{\text{bounded}}{\underbrace{\lambda^{-d}u^\lambda_{\lambda^{-1}}}}\in L_{\frac{d\lambda}{\lambda}}^{1}([0,1],\mathcal{D}^\prime(U\setminus I)) $ and the integral
$t=\int_{0}^1\frac{d\lambda}{\lambda} \lambda^{s+d} \lambda^{-d}u^\lambda_{\lambda^{-1}}$
converges in $L_{\frac{d\lambda}{\lambda}}^1([0,1],\mathcal{D}^\prime(U\setminus I))$! 
By the estimate (\ref{est2}) on the bounded family $\lambda^{-d}u^\lambda_{\lambda^{-1}}$, we also have the estimate:
$$\vert \left\langle t,\varphi \right\rangle \vert=\vert \int_0^1\frac{d\lambda}{\lambda}\lambda^{s} \left\langle u^\lambda_{\lambda^{-1}},\varphi \right\rangle\vert $$ $$\leqslant \int_0^1\frac{d\lambda}{\lambda}\lambda^{s+d}\underset{\leqslant  C\pi_{k}(\varphi)\text{ by \ref{est2}} }{\underbrace{\vert \lambda^{-d}\left\langle u^\lambda_{\lambda^{-1}},\varphi \right\rangle\vert}} \leqslant  C\pi_{k}(\varphi) \int_{0}^1\frac{d\lambda}{\lambda} \lambda^{s+d}=\frac{C}{s+d}\pi_{k}(\varphi).$$
\end{proof}
\begin{prop}\label{scal1}
Under the assumptions of Corollary (\ref{ren1}),
$\int_{0}^1\frac{d\lambda}{\lambda} \lambda^s u^\lambda_{\lambda^{-1}}\in E_s(U)$.
\end{prop}
\begin{proof}
Recall that we proved that the integral
$t= \int_{0}^1\frac{d\lambda}{\lambda} \lambda^{s} u^\lambda_{\lambda^{-1}} $
converges in $\mathcal{D}^\prime(U)$ and we would like to prove that $t\in E_s(U)$.
We try to bound the quantity $\mu^{-s} t_\mu$:
%$$\mu^{-s} \left\langle t_\mu ,\varphi \right\rangle=\int_{0}^1\frac{d\lambda}{\lambda} \mu^{-s}\lambda^{s} \left\langle c_\lambda(\mu\lambda^{-1}.),\varphi \right\rangle=\int_{0}^1\frac{d\lambda}{\lambda} \frac{\lambda^{s}}{\mu^s} \left\langle c_\lambda(\frac{\mu}{\lambda}.),\varphi \right\rangle=\int_{0}^{\frac{1}{\mu}}\frac{d\lambda}{\lambda} \lambda^{s} \left\langle c_{\lambda\mu}(\lambda^{-1}.),\varphi \right\rangle $$
$$\forall 0< \mu\leqslant 1  , \mu^{-s} \left\langle t_\mu ,\varphi \right\rangle=\mu^{-s-d} \left\langle t ,\varphi_{\mu^{-1}} \right\rangle=\int_{0}^1\frac{d\lambda}{\lambda} \mu^{-s-d}\lambda^{s} \left\langle u^\lambda_{\lambda^{-1}},\varphi_{\mu^{-1}} \right\rangle$$ $$=\int_{0}^1\frac{d\lambda}{\lambda} \left(\frac{\lambda}{\mu}\right)^{s+d} \left\langle u^\lambda,\varphi_{\frac{\lambda}{\mu}} \right\rangle=\int_{0}^{\frac{1}{\mu}}\frac{d\lambda}{\lambda} \lambda^{s+d} \left\langle u^{\lambda\mu} ,\varphi_\lambda \right\rangle.$$
We use the fact that there exists $R>0$ such that $\varphi\in\mathcal{D}(U)$ is supported inside the domain $\{\vert h\vert\leqslant R\}$. Then $\varphi_\lambda=\varphi(.,\lambda.)$ is supported in $\{\vert h\vert\leqslant \lambda^{-1}R\}$.
We denote by $\pi_1$ the projection 
$\pi_1:=(x,h)\in\mathbb{R}^{n+d}\mapsto 
(x,0)\in\mathbb{R}^n\times\{0\}$
and we make the
notation abuse
$\pi_1(x,h)=(x)$. 
Then $K=\pi_1(\text{supp }\varphi)$ 
is compact in $\mathbb{R}^n$ thus, 
by assumption on the
family $u$, $u^{\lambda\mu}|_{(K\times\mathbb{R}^d)\cap U}$ is supported in $\{a\leqslant \vert h\vert\leqslant b\}$ for some $0<a<b$
and $\left\langle u^{\lambda\mu} ,\varphi_\lambda \right\rangle $ must vanish when
$\lambda^{-1}R\leqslant a \Leftrightarrow \lambda\geqslant \frac{R}{a} $. Finally:
$$\mu^{-s} \left\langle t_\mu ,\varphi \right\rangle=\int_{0}^{\frac{R}{a}}\frac{d\lambda}{\lambda} \lambda^{s+d} \left\langle u^{\lambda\mu} ,\varphi_{\lambda} \right\rangle.$$
Since $\varphi_\lambda\in C^\infty((K\times\mathbb{R}^d)\cap U)$, 
by estimate (\ref{est1}), we have $\vert \left\langle u^{\lambda\mu},\varphi_\lambda\right\rangle \vert\leqslant C\pi_{k,\mathcal{A}}(\varphi)\leqslant C\pi_{k}(\varphi)$ and
$$\vert \mu^{-s} \left\langle t_\mu ,\varphi \right\rangle\vert \leqslant \left(\frac{R}{a}\right)^{s+d}\frac{C}{s+d}\pi_{k}\left(\varphi\right).$$ 
\end{proof}

\begin{prop}\label{ren2}
Let $(u^\lambda)_{\lambda\in(0,1]}$ be a bounded family in $\mathcal{D}^\prime(U\setminus I)$
which is uniformly supported on an annulus domain
of $U$.
If $-m-1< s+d \leqslant -m, m\in\mathbb{N}$, then the integral $\int_{0}^1\frac{d\lambda}{\lambda} \lambda^s u^\lambda_{\lambda^{-1}}$ needs a renormalization. There is a family $(\tau^\lambda)_{\lambda\in(0,1]}$ of distributions supported on $I$ 
such that the renormalized integral
\begin{equation}\label{formulaalgo2} 
\int_{0}^1\frac{d\lambda}{\lambda} \lambda^s \left(u^\lambda_{\lambda^{-1}}-\tau^\lambda\right)   
\end{equation}
converges in $\mathcal{D}^\prime(U)$.
\end{prop}
\begin{proof}
If $-m-1< s+d\leqslant -m$, 
then we repeat 
the previous proof 
except we have to subtract 
to $\varphi$ its Taylor polynomial $P_m$ of order $m$ in $h$. 
We call $I_m$ the Taylor remainder.
Then $\varphi-P_m=I_m$. 
In coordinates, we get
$$\varphi(x,h)-\underset{P_m}{\underbrace{\sum_{\vert i\vert \leqslant m} 
\frac{h^i}{i!}\frac{\partial^i\varphi}{\partial h^i}(x,0)}}
=I_m(x,h)=\sum_{\vert i\vert =m+1} h^i H_i(x,h)$$ 
where $(H_i)_i$ are smooth functions.
$R_\lambda(x,h)=R(x,\lambda h)=\lambda^{m+1}\sum_{\vert i\vert =m+1} h^i H_i(x,\lambda h)$. 
We define a distribution supported on $I$, which we call ``counterterm'':
\begin{equation} 
\left\langle \tau^\lambda ,\varphi \right\rangle = \left\langle u^\lambda_{\lambda^{-1}},\sum_{\vert i\vert \leqslant m} \frac{h^i}{i!}\frac{\partial^i\varphi}{\partial h^i}(\cdot,0)\right\rangle
\end{equation} 
where we abusively denoted the
expression
$\frac{\partial^i\varphi}{\partial h^i}\circ \pi_1$ by 
$\frac{\partial^i\varphi}{\partial h^i}(\cdot,0)$ .
We take into account the counterterm
$$\lambda^{s}\left\langle u^\lambda_{\lambda^{-1}}-\tau^\lambda,\varphi \right\rangle=\lambda^{s}\left\langle u^\lambda_{\lambda^{-1}},\varphi(x,h)-\sum_{\vert i\vert \leqslant m} \frac{h^i}{i!}\frac{\partial^i\varphi}{\partial h^i}(\cdot,0) \right\rangle $$
$$=\lambda^{s}\left\langle u^\lambda_{\lambda^{-1}},\sum_{\vert i\vert =m+1} h^i H_i(x,h)\right\rangle=\lambda^{s+d}\left\langle u^\lambda ,\lambda^{(m+1)}\sum_{\vert i\vert =m+1} h^i H_i(x,\lambda h)\right\rangle $$
$$=\lambda^{(d+s+m+1)}\left\langle u^\lambda ,\sum_{\vert i\vert =m+1} h^i H_i(x,\lambda h)\right\rangle   $$
Hence
$$\int_0^1\frac{d\lambda}{\lambda} 
\lambda^{s}
\left\langle u^\lambda_{\lambda^{-1}}-
\tau^\lambda,\varphi \right\rangle=\int_0^1\frac{d\lambda}{\lambda} \underset{\text{integrable}}{\underbrace{\lambda^{(d+s+m+1)}}}\underset{\text{bounded}}{\underbrace{\left\langle u^\lambda, \sum_{\vert i\vert =m+1} h^iH_i(x,\lambda h)\right\rangle}} $$ 
since $\forall \lambda\in(0,1], h^iH_i(x,\lambda h)\in C^\infty((K\times \mathbb{R}^{d})\cap U)$,
we can use estimate (\ref{est1})
$$\left|\int_0^1\frac{d\lambda}{\lambda} 
\lambda^{s}
\left\langle u^\lambda_{\lambda^{-1}}-
\tau^\lambda,\varphi \right\rangle \right| \leqslant \frac{C}{d+s+m+1}\sup_{\lambda\in(0,1]}\pi_{k,\mathcal{A}}(\underset{\text{derivatives of }\varphi\text{ order }m+1}{\underbrace{\sum_{\vert i\vert =m+1}h^iH_i(x,\lambda h)}}) $$
$$\int_0^1\frac{d\lambda}{\lambda} \vert\lambda^{s}\left\langle u^\lambda_{\lambda^{-1}}-\tau^\lambda,\varphi \right\rangle\vert \leqslant \frac{\tilde{C}}{d+s+m+1} \pi_{ k+m+1}\left(\varphi\right)$$
where the constant $\tilde{C}$ does not depend on $\varphi$ and can be estimated by the 
integral remainder formula. 
\end{proof}
\begin{prop}\label{scal2}
Under the assumptions of proposition (\ref{ren2}),
if $s$ is not an integer then $\int_{0}^1\frac{d\lambda}{\lambda} \lambda^s \left(u^\lambda_{\lambda^{-1}}-\tau^\lambda\right)\in E_s(U)$. 
\end{prop}
\begin{prop}\label{scal3noninteger}
Under the assumptions of proposition (\ref{ren2}), if $s+d$ is a $\textbf{non positive integer}$ then $\int_{0}^1\frac{d\lambda}{\lambda} \lambda^s \left(u^\lambda_{\lambda^{-1}}-
\tau^\lambda\right)\in E_{s^\prime}(U),\forall s^\prime<s$, and
$t=\int_{0}^1\frac{d\lambda}{\lambda}
\lambda^s \left(u^\lambda_{\lambda^{-1}}-
\tau^\lambda\right)$ 
satisfies the estimate
\begin{equation}
\forall \varphi \in \mathcal{D}(U), \exists C, \vert\mu^{-s}\left\langle t_\mu,\varphi \right\rangle\vert\leqslant C\left(1+\vert \log\mu\vert  \right). 
\end{equation}
\end{prop}
\begin{proof}
To check the homogeneity of the renormalized integral is a little tricky since we have to take the scaling of counterterms into account.
When we scale the smooth function then we should scale simultaneously the Taylor polynomial and the remainder 
$$\varphi_\lambda=P_\lambda+R_\lambda $$  
  We want to know to which scale space $E_{s^\prime}$ the distribution $t$ belongs:
$$\mu^{-s^\prime}\left\langle t_\mu,\varphi \right\rangle
=\mu^{s-s^\prime}\mu^{-s-d}\left\langle t,\varphi_{\mu^{-1}} \right\rangle
=\mu^{s-s^\prime}\int_0^1\frac{d\lambda}{\lambda} \lambda^{s}\left\langle u^\lambda_{\lambda^{-1}}-\tau^\lambda,\mu^{-d-s} \varphi_{\mu^{-1}} \right\rangle$$
$$=\mu^{s-s^\prime}\int_0^1\frac{d\lambda}{\lambda}\left(\frac{\lambda}{\mu}\right)^{s}\mu^{-d}\left\langle u^\lambda_{\lambda^{-1}},\varphi(x,\frac{h}{\mu})-\sum_{\vert i\vert \leqslant m} \frac{h^i}{\mu^ii!}\frac{\partial^i\varphi}{\partial h^i}(x,0) \right\rangle $$
$$=\mu^{s-s^\prime}\int_0^1\frac{d\lambda}{\lambda}\left(\frac{\lambda}{\mu}\right)^{s+d}\left\langle u^\lambda ,\varphi(x,\frac{\lambda}{\mu}h)-\sum_{\vert i\vert \leqslant m} \frac{h^i}{\mu^ii!}\frac{\partial^i\varphi}{\partial h^i}(x,0) \right\rangle .$$
$\varphi_{\frac{\lambda}{\mu}}$ 
is supported on $\vert h\vert\leqslant \frac{\mu R}{\lambda}$
thus when $\frac{R\mu}{\lambda}\leqslant a\Leftrightarrow \frac{R\mu}{a}\leqslant \lambda$, 
the support of $\varphi_{\frac{\lambda}{\mu}}$
does not meet 
the support
of $u^\lambda$
because $u^\lambda$
is supported on $a\geqslant\vert h\vert$,
whereas $\sum_{\vert i\vert \leqslant m} \frac{(\lambda h)^i}{i!}\frac{\partial^i\varphi}{\partial h^i}(x,0)$ is supported everywhere because it is a Taylor polynomial. 
Consequently, we must split the integral in two parts
$$\mu^{-s}\left\langle t_\mu,\varphi \right\rangle=I_1+I_2$$ 
$$I_1=\int_0^{\frac{R\mu}{a}}\frac{d\lambda}{\lambda}\left(\frac{\lambda}{\mu}\right)^{s+d}\left\langle u^\lambda , I_{m,\frac{\lambda}{\mu}} \right\rangle $$ $$=\int_0^{\frac{R\mu}{a}}\frac{d\lambda}{\lambda} \left(\frac{\lambda}{\mu}\right)^{(d+s+m+1)}\left\langle u^\lambda, \sum_{\vert i\vert =m+1} h^iH_i(x,\frac{\lambda}{\mu} h)\right\rangle $$ 
$$I_2= \int_{\frac{R\mu}{a}}^1\frac{d\lambda}{\lambda}\left(\frac{\lambda}{\mu}\right)^{s+d}\left\langle u^\lambda , I_{m,\frac{\lambda}{\mu}} \right\rangle$$ 
$$=\int_{\frac{R\mu}{a}}^1\frac{d\lambda}{\lambda}\left(\frac{\lambda}{\mu}\right)^{s+d}\underset{\text{no contribution of $\varphi_{\frac{\lambda}{\mu}}$ since }\frac{R\mu}{a}\leqslant \lambda}{\left\langle u^\lambda ,\varphi(x,\frac{\lambda}{\mu}h)-\sum_{\vert i\vert \leqslant m} \frac{(\lambda h)^i}{\mu^ii!}\frac{\partial^i\varphi}{\partial h^i}(x,0) \right\rangle} $$  
and we apply a variable change for $I_1$: 
$$I_1=\int_0^{\frac{R}{a}}\frac{d\lambda}{\lambda} \lambda^{(d+s+m+1)}\left\langle u^\lambda_\mu, \sum_{\vert i\vert =m+1} h^iH_i(x,\lambda h)\right\rangle\hfill$$ again by estimate (\ref{est1}) $$\hfill\leqslant \left(\frac{R}{a}\right)^{-(d+s+m+1)} \frac{C}{s+d+m+1} \sup_{\lambda\in(0,1]} \pi_{k,\mathcal{A}}\left(\sum_{\vert i\vert =m+1} h^iH_i(x,\lambda h) \right)$$
and each $H^i$ is a term in the Taylor remainder $I_m$ of $\varphi$,
$$I_1 \leqslant  C_1  \pi_{k+m+1}( \varphi).$$
Notice that in the second term only the counterterm contributes 
$$I_2=\int_{\frac{R\mu}{a}}^1\frac{d\lambda}{\lambda}
\left(\frac{\lambda}{\mu}\right)^{s+d}\left\langle u^\lambda ,-\sum_{\vert i\vert \leqslant m} \frac{(\lambda h)^i}{\mu^i i!}\frac{\partial^i\varphi}{\partial h^i}(x,0) \right\rangle $$
$$=\int_{\frac{R\mu}{a}}^1\frac{d\lambda}{\lambda}\left\langle u^\lambda ,-\sum_{\vert i\vert \leqslant m} \left(\frac{\lambda}{\mu}\right)^{s+d+i}\frac{h^i}{i!}\frac{\partial^i\varphi}{\partial h^i}(x,0) \right\rangle.$$ 
Then notice that by assumption $s+d\leqslant -m$ and 
$\vert i\vert$ ranges from $0$ to $m$ 
which implies $s+d+\vert i\vert\leqslant 0$. 
When $s+d+\vert i\vert < 0$:  
$$\int_{\frac{R\mu}{a}}^1\frac{d\lambda}{\lambda}\left|\left\langle u^\lambda , \left(\frac{\lambda}{\mu}\right)^{s+d+i}\frac{h^i}{i!}\frac{\partial^i\varphi}{\partial h^i}(x,0) \right\rangle  \right| \leqslant  \underset{\text{no blow up when }\mu\rightarrow 0}{\underbrace{C_2\left|\left(\frac{1}{\mu}\right)^{s+d+i}-\left(\frac{R}{a}\right)^{s+d+i}\right|}}\pi_k( \varphi).$$
If $s+d<-m$ then $s+d+\vert i\vert$ is always strictly negative
and there is no blow up when $\mu\rightarrow 0$, thus $t\in E_s$.  
If $s+d+m=0$ and for $\vert i\vert= m$:
$$\int_{\frac{R\mu}{a}}^1\frac{d\lambda}{\lambda}\left|\left\langle u^\lambda , \left(\frac{\lambda}{\mu}\right)^{s+d+i}\frac{h^i}{i!}\frac{\partial^i\varphi}{\partial h^i}(x,0) \right\rangle \right| \vert \leqslant  C_2\vert\log(\frac{R\mu}{a})\vert\pi_k( \varphi)$$   
and the only term which blows up when $\mu\rightarrow 0$ is the logarithmic term.
If $s+d=-m$ then $t\in E_{s^\prime}$ for all $s^\prime < s$ and $\vert\mu^{-s}\left\langle t_\mu,\varphi \right\rangle\vert$ has at most $\textbf{logarithmic blow up}$:
$$\exists (C_1,C_2) \,\ \vert\mu^{-s^\prime}\left\langle t_\mu,\varphi \right\rangle\vert\leqslant \underset{\text{bounded when }s^\prime<s}{\underbrace{\mu^{s-s^\prime}\left( C_1  \pi_{k+m+1}(\varphi) + C_2\vert\log(\frac{R\mu}{a})\vert\pi_k(\varphi)\right)}}.$$ 
\end{proof}
\section{Extension of distributions.}
Conversely, if we start from any distribution $t$ in $E_s\left(U\setminus I \right)$, 
then we can associate to it a bounded family $\left(u^\lambda\right)_{\lambda\in(0,1]}$.
Then application of the previous results on the family $(u^\lambda)_\lambda$ allows to construct a distribution $\overline{t\chi}$ in $E_s(U)$. But the resulting
distribution
given by formulas 
(\ref{formulaalgo1}) (\ref{formulaalgo2}) coincides exactly with the extension formula $\int_0^1 \frac{d\lambda}{\lambda} t\psi_{\lambda^{-1}}$ on
$U\setminus I$. 
Hence $\overline{t\chi}$ 
is an extension of $t\chi$.
Moreover, if we started from a distribution $t\in E_s(U)$ then  
the reconstruction theorem provides us
with a distribution which is equal
to $t\chi$ up to a distribution supported on $I$, 
except for the case $s+d>0$ 
where 
the extension is unique 
if we do not want to increase the degree of divergence.
\begin{prop}
Let $t\in E_s\left(U\setminus I \right)$ and let $\psi=-\rho\chi$ where $\chi\in C^\infty(\mathbb{R}^{n+d})$, $\chi=1$ in a neighborhood $N_1$ of $I$ and $\chi=0$ outside $N_2$ a neighborhood of $N_1$, then 
\begin{equation}
u^\lambda=\lambda^{-s}t_\lambda\psi 
\end{equation}
is a bounded family in $\mathcal{D}^\prime(U\setminus I)$
which is uniformly supported on an annulus domain
of $U$.
\end{prop}
\begin{proof}
Consider the function $\psi=-\rho\chi$ 
used in our construction of the partition of unity of H\"ormander. 
By construction, it is supported on an annulus domain
of $U$.
By definition, $t\in E_s(U\setminus I)$ 
implies $\lambda^{-s}t_\lambda$ is a bounded family of distributions in $\mathcal{D}^\prime(U\setminus I)$, 
hence $u^\lambda=\lambda^{-s}t_\lambda\psi$ is a bounded family of distributions uniformly supported in $\text{supp }\psi$.
\end{proof} 
Once we notice $$\int_0^1 \frac{d\lambda}{\lambda} \lambda^s u^\lambda_{\lambda^{-1}}=\int_0^1  \frac{d\lambda}{\lambda} \lambda^s\left(\lambda^{-s}t_\lambda\psi\right)_{\lambda^{-1}}=\int_0^1  \frac{d\lambda}{\lambda} t\psi_{\lambda^{-1}} ,$$ the formula of the construction algorithm exactly coincides with the extension formula of H\"ormander. Then we can deduce
all the results listed below from simple applications of results derived for the family
$u^\lambda$:
\begin{thm}\label{thm1} 
Let $t\in E_s\left(U\setminus I \right)$, if
$s+d>0 $ then 
\begin{equation}
\forall\varphi\in\mathcal{D}(U),\overline{t}(\varphi)=\lim_{\varepsilon\rightarrow 0}\left\langle t(1-\chi_{\varepsilon^{-1}}),\varphi \right\rangle
\end{equation}
exists and defines an extension $\overline{t}\in\mathcal{D}^\prime(U)$ 
and $\overline{t}$ is in $E_s(U)$.
\end{thm}
The proof relies on the first identification 
$$\int_0^1\frac{d\lambda}{\lambda} \lambda^s u^\lambda_{\lambda^{-1}}=\int_0^1 \frac{d\lambda}{\lambda} t \psi(\frac{h}{\lambda})=\lim_{\varepsilon\rightarrow 0}\int_\varepsilon^1 \frac{d\lambda}{\lambda} t \psi_{\lambda^{-1}}=\lim_{\varepsilon\rightarrow 0} \left\langle t\left(\chi-\chi_{\varepsilon^{-1}}\right),\varphi\right\rangle,$$
where $\psi=-\rho\chi$. 
Then by definition of $\overline{t}$:
$$\overline{t}= \int_0^1 \frac{d\lambda}{\lambda} t \psi(\frac{h}{\lambda})+\left\langle t(1-\chi),\varphi \right\rangle$$ $$=\lim_{\varepsilon\rightarrow 0} \left\langle t\left(\chi-\chi_{\varepsilon^{-1}}\right),\varphi\right\rangle+\left\langle t(1-\chi),\varphi \right\rangle=\lim_{\varepsilon\rightarrow 0}\left\langle t(1-\chi_{\varepsilon^{-1}}),\varphi \right\rangle.$$
In the case $s+d>0$,
the last formula $\lim_{\varepsilon\rightarrow 0}\left\langle t(1-\chi_{\varepsilon^{-1}}),\varphi \right\rangle$ also appears 
in the very
nice
recent work
\cite{WB} 
(but with different hypothesis 
and interpretation)
and in fact 
goes back
to Meyer \cite{Meyer} Definition 1.7 p.~15 
and formula (3.16) p.~15.
\begin{thm}\label{thm2}
Let $t\in E_s\left(U\setminus I \right)$, if $-m-1< s+d \leqslant -m\leqslant 0$
then 
\begin{eqnarray}
\overline{t}= \lim_{\varepsilon\rightarrow 0} \left(  \left\langle t\left(\chi-\chi_{\varepsilon^{-1}}\right),\varphi\right\rangle -\left\langle c_\varepsilon,\varphi\right\rangle \right)+ \left\langle t(1-\chi),\varphi \right\rangle
\end{eqnarray} 
exists and defines an extension $\overline{t}\in\mathcal{D}^\prime(U)$
where the \textbf{local counterterms} $c_\varepsilon$ is defined by 
\begin{eqnarray}
\left\langle c_\varepsilon,\varphi\right\rangle = \left\langle t\left(\chi-\chi_{\varepsilon^{-1}}\right),\sum_{\vert i\vert\leqslant m}\frac{h^i}{i !}\varphi^i(x,0)\right\rangle.
\end{eqnarray}
If $s$ is not an integer then the extension $\overline{t}$ is in $E_s(U)$,
otherwise $\overline{t}\in E_{s^\prime}(U),\forall s^\prime<s$. 
\end{thm}
The last case is treated by \cite{WB} and \cite{BF} 
in a slightly different way,
they introduce a projection $P$
from the space of $C^\infty$ functions 
to the m-th power $\mathcal{I}^m$
of the ideal of smooth functions 
(of course by definition
the restriction of this projection to
$\mathcal{I}^m$ is the identity), 
and to construct
this projection 
one has to subtract
local counterterms as Meyer does.
\subsubsection{A converse result.}
Before we move on, let us prove a general converse theorem, namely that given any distribution $t\in \mathcal{D}^\prime\left(U\right)$, we can find $s_0\in \mathbb{R}$ such that for all $s\leqslant s_0$, $t\in E_s(U)$ (we believe such sort of theorems were first proved by Lojasiewicz and Alberto Calderon, \cite{Weinsteinext}), this means
morally
that any distribution
has ``finite
scaling degree''
along an arbitrary
vector subspace. 
We also have the property that $\forall s_1\leqslant s_2, t\in E_{s_2}\implies t\in E_{s_1}$. This means that the spaces $E_s$ are $\textbf{filtered}$. 
We work in $\mathbb{R}^{n+d}$ where $I=\mathbb{R}^n\times\{0\}$ and $\rho=h^j\frac{\partial}{\partial h^j}$:
\begin{thm}
Let $U$ be a $\rho$-convex open set
and $t\in \mathcal{D}^\prime(U)$.
If
$t$ is of order $k$,
then $t\in E_s(U)$
for all
$s\leqslant d+k$, 
where $d$ is the \textbf{codimension}
of $I\subset\mathbb{R}^{n+d}$. 
In particular
any compactly 
supported
distribution
is in
$E_s(\mathbb{R}^{n+d})$
for some $s$.
\end{thm}
\begin{proof}
First notice if a function $\varphi\in \mathcal{D}(U)$, then the family of scaled functions $(\varphi_{\lambda^{-1}})_{\lambda\in(0,1]}$ has support contained in a compact set $K=\{(x,\lambda h) \vert (x,h)\in\text{supp }\varphi,\lambda\in (0,1] \}$.
We recall that for any distribution $t$, there exists $k,C_K$ such that
$$\forall\varphi \in \mathcal{D}_K(U), \vert\left\langle t,\varphi \right\rangle\vert \leqslant C_K \pi_{K,k}(\varphi).$$ 
$$\vert\left\langle t_{\lambda},\varphi \right\rangle\vert=\vert\lambda^{-d}\left\langle t,\varphi_{\lambda^{-1}} \right\rangle\vert\leqslant C_K \lambda^{-d}\pi_{K,k}(\varphi_{\lambda^{-1}})\leqslant C_K \lambda^{-d-k}\pi_{K,k}(\varphi).$$
So we find that $\lambda^{d+k}\left\langle t_\lambda,\varphi \right\rangle $ is bounded which yields the conclusion.
\end{proof}
 
\subsection{Removable singularity theorems.}
Finally, we would like to conclude this section by a simple removable singularity theorem in the spirit of Riemann, (compare with Harvey-Polking \cite{HP} theorems $(5.2)$ and $(6.1)$). 
In a renormalization procedure there is always an ambiguity which is the ambiguity of the extension of the distribution. Indeed, two extensions always differ by a distribution supported on $I$. The removable singularity theorem states that if $s+d>0$ and if we demand that $t\in E_s(U\setminus I)$ should extend to $\overline{t}\in E_s(U)$, then the extension is \textbf{unique}.
Otherwise, if $-m-1 < s+d \leqslant -m $, then we bound the transversal order of the ambiguity.
We fix the coordinate system $(x^i,h^j)$ in $\mathbb{R}^{n+d}$ and $I=\{h=0\}$. The collection of coordinate functions $(h^j)_{1\leqslant j\leqslant d}$ defines a canonical collection of transverse vector fields $(\partial_{h^j})_j$.
We denote by $\delta_I$ the unique distribution such that $\forall\varphi\in \mathcal{D}(\mathbb{R}^{n+d})$,$$\left\langle \delta_I,\varphi\right\rangle=\int_{\mathbb{R}^n} \varphi(x,0)d^nx .$$
If $t\in \mathcal{D}^\prime(\mathbb{R}^{n+d})$ with $\text{supp } t \subset I$, then there exist
unique distributions (once the system of transverse vector fields $\partial_{h^j}$ is fixed) $t_\alpha\in \mathcal{D}^\prime\left(\mathbb{R}^n\right)$, where each compact intersects
$\text{supp }t_\alpha$ for a finite number of multiindices $\alpha$, such that $t(x,h)=\sum_\alpha t_\alpha(x) \partial_h^\alpha\delta_I(h)$
(see \cite{Schwartz} theorem $(36)$ and $(37)$ p.~101--102 or \cite{Hormander} theorem $(2.3.5)$) where the $\partial_h^\alpha$ are derivatives in the \textbf{transverse} directions. 
\begin{thm}\label{removsing}
Let $t\in E_s(U\setminus I)$ and $\overline{t}\in E_{s^\prime}(U\setminus I)$
its extension given by Theorem (\ref{thm1}) and Theorem (\ref{thm2}) $s^\prime=s$ when $-s-d\notin\mathbb{N}$ or $\forall s^\prime<s$ otherwise.
Then $\tilde{t}$ is an extension in $E_{s^\prime}(U)$ 
if and only if 
$$\tilde{t}(x,h)=\overline{t}(x,h)+\sum_{\alpha\leqslant m} t_\alpha(x) \partial_h^\alpha\delta_I(h),$$
where $m$ is the integer part of $-s-d$. 
In particular when $s+d>0$ the extension is unique.
\end{thm}
Remark: when $-s-d$ is a nonnegative integer,
the counterterm is in
$E_s$ whereas the extension
is in $E_{s^\prime},\forall s^\prime<s$.

\begin{proof}
We scale an elementary distribution $\partial_h^\alpha\delta_I$:
$$\left\langle (\partial_h^\alpha\delta_I)_\lambda, \varphi\right\rangle=\lambda^{-d}\left\langle \partial_h^\alpha\delta_I, \varphi_{\lambda^{-1}}\right\rangle=(-1)^{\vert\alpha\vert}\lambda^{-d-\vert\alpha\vert} \left\langle \partial_h^\alpha\delta_I, \varphi\right\rangle $$
hence $\lambda^{-s}(\partial^\alpha\delta_I)_\lambda=\lambda^{-d-\vert\alpha\vert-s}\partial_h^\alpha\delta_I$ is bounded iff $d+s+\vert\alpha\vert\leqslant 0 \implies d+s\leqslant -\vert\alpha\vert $. 
When $s+d>0$, $\forall\alpha,\partial_h^\alpha\delta_I\notin E_s$ hence
any two extensions in $E_s(U)$ cannot differ
by a local counterterm of the form $\sum_\alpha t_\alpha \partial_h^\alpha\delta_I$.
When $-m-1<d+s\leqslant -m$ then $\lambda^{-s}(\partial_h^\alpha\delta_I)_\lambda$ is bounded iff
$s+d+\vert\alpha\vert\leqslant 0\Leftrightarrow -m\leqslant -\vert\alpha\vert\Leftrightarrow \vert\alpha\vert\leqslant m $. 
We deduce that $\partial_h^\alpha\delta_I \in E_{s}$ for all $\alpha\leqslant m$
which means that the scaling degree \textbf{bounds} the order $\vert\alpha\vert$ 
of the derivatives in the transverse directions. 
Assume there are two extensions in $E_{s}$, their difference is of the form $u=\sum_\alpha u_\alpha \partial_h^\alpha\delta_I$ by the structure theorem $(36)$ p.~101 in \cite{Schwartz} and is also in $E_{s}$
which means their difference equals $u=\sum_{\vert\alpha\vert\leqslant m} u_\alpha \partial_h^\alpha\delta_I$.
\end{proof}

\section{Euler vector fields.}
We want to solve the extension problem
for distributions
on manifolds, in order to do so
we must give a geometric definition
of scaling transversally
to a submanifold $I$ closely embedded 
in a given 
manifold $M$.
We will define a class of 
Euler vector fields 
which scale transversally
to a given fixed submanifold $I\subset M$.
Let $M$ be a smooth manifold and $I\subset M$ an embedded submanifold without boundary. 
For the moment, all discussions are purely local. 
A classical result in differential geometry associates to each submanifold $I\subset M$ the $\textbf{sheaf of ideal }\mathcal{I}$ of functions vanishing on $I$.
\begin{defi}
Let $U$ be an open subset of $M$ and $I$ a submanifold of $M$, then we define
the ideal $\mathcal{I}(U)$ 
as the collection of functions $f\in C^\infty(U)$ 
such that 
$f|_{I\cap U}=0$.
We also define the ideal $\mathcal{I}^2(U)$
which consists of functions $f\in C^\infty(U)$ such that
$f=f_1f_2$ where $(f_1,f_2)\in \mathcal{I}(U)\times \mathcal{I}(U)$.
\end{defi}
\begin{defi}
A vector field $\rho$ is locally defined on an open set $U$ is called Euler if 
\begin{equation}
\forall f\in\mathcal{I}(U),  \rho f-f\in\mathcal{I}^2(U). 
\end{equation} 
\end{defi}  
\begin{ex}\label{fundexEuler}
$h^i\partial_{h^i}$ is Euler by application of Hadamard lemma, if
$f$ in $\mathcal{I}$ then 
$f=h^iH_i$ where the $H_i$
are 
smooth functions,
which implies $\rho f= f + h^ih^j\partial_{h^j}H_i\implies \rho f-f=h^ih^j\partial_{h^j}H_i$.
\end{ex} 
In this definition, $\rho$ is defined by testing against arbitrary restrictions of smooth functions $f|_U$ vanishing on $I$.
Let $G$ be the pseudogroup of local diffeomorphisms of $M$ (i.e. an element $\Phi$ in $G$ is defined over an open set $U\subset M$ and maps it diffeomorphically to an open set $\Phi(U)\subset M$)
% isotopic to identity $M\mapsto M$ 
such that
$\forall p\in I\cap U, \forall \Phi\in G, \Phi(p)\in I$.
\begin{prop}
Let $\rho$ be $\textbf{Euler}$, then $\forall \Phi\in G$, $\Phi_*\rho $ is $\textbf{Euler}$.
\end{prop}
\begin{proof} 
For this part, see \cite{Lee} p.~92 
for the definition 
and properties of the pushforward
of a vector field:
if $Y=\Phi_*X$ then $L_Y f=L_X(f\circ\Phi)\circ \Phi^{-1} $.
We may write the last expression in terms of pull-back 
\begin{equation}\label{pushforward}
L_{\Phi_*X} f=L_X(f\circ\Phi)\circ \Phi^{-1} =\Phi^{-1*}\left(L_X \left(\Phi^*f\right)\right).
\end{equation} 
Then we apply the identity to $X=\rho,Y=\Phi_*\rho$,
setting $L_{\Phi_*\rho}f=\Phi_*\rho f$ and $L_\rho f=\rho f$ for shortness:
$$\left(\left(\Phi_*\rho\right)f\right)=\Phi^{-1*}\left(\rho \left(\Phi^*f\right)\right). $$
Now since $\Phi\in G$, $\rho$ is Euler and $f$ an arbitrary function in $\mathcal{I}$.
$$\forall \Phi\in G,\forall f\in \mathcal{I}, \left(\Phi_*\rho\right)f-f
=\Phi^{-1*}\left(\rho \left(\Phi^*f\right)\right)-\Phi^{-1*} \left(\Phi^*f\right)=\Phi^{-1*} \left(\rho \left(\Phi^*f\right)- \left(\Phi^*f\right)\right).$$
Since $\Phi(I)\subset I$, we have actually 
$\Phi^*f\in\mathcal{I}$ hence $\left(\rho \left(\Phi^*f\right)- \left(\Phi^*f\right)\right)\in\mathcal{I}^2$ and we deduce
that
$\Phi^{-1*} \left(\rho \left(\Phi^*f\right)- \left(\Phi^*f\right)\right)\in\Phi^{-1*}\mathcal{I}^2$. We will prove that $\Phi^*\mathcal{I}(U)=\mathcal{I}(\Phi(U))$.
$$f\in \mathcal{I} \Leftrightarrow f|_I=0 \Leftrightarrow f|_{\Phi(I)}=0 \text{ since } \Phi(I)\subset I  \Leftrightarrow (f\circ \Phi)|_I=0 \text{ thus }\Phi^*f\in \mathcal{I}.$$ 
Hence $\rho \left(\Phi^*f\right)- \left(\Phi^*f\right)\in\mathcal{I}^2$ by definition of $\rho$, finally we use the fact $$\Phi^*\left(\mathcal{I}^2\right)= \{(fg)\circ \Phi; (f,g)\in \mathcal{I}^2 \}=\{(f\circ \Phi)(g\circ \Phi); (f,g)\in \mathcal{I}^2 \}=(\Phi^*\mathcal{I})^2=\mathcal{I}^2 $$ since $\Phi^*\mathcal{I}=\mathcal{I}$ to deduce: 
$$\Phi^{-1*} \left(\rho \left(\Phi^*f\right)- \left(\Phi^*f\right)\right)\in\mathcal{I}^2 $$
which completes the proof.
\end{proof}
$\textbf{Euler vector fields}$ form a \textbf{sheaf} (check the definitions p.~289 in \cite{Lee})  with the following nice additional properties:
\begin{itemize}
\item Given $I$, the set of $\emph{global}$ Euler vector fields defined on some open neighborhood of $I$ is $\textbf{nonempty}$.
\item For any local Euler vector field $\rho|_U$, for any open set $V\subset U$ there is a Euler vector field $\rho^\prime$ defined on a $\textbf{global neighborhood}$ of $I$ such that $\rho^\prime|_V=\rho|_V$. 
\end{itemize} 
\begin{proof}
These two properties result from the fact that one can glue together
Euler vector fields by a partition of unity subordinated
to some cover of some neighborhood $N$ of $I$.
By paracompactness of $M$, we can pick an arbitrary locally finite open cover $\cup_{i\in I} V_i$ of $I$ by open sets $V_i$, such that for each $i$, there is a local chart $(x,h):V_i\mapsto \mathbb{R}^{n+d}$ where the image of $I$ by the local chart is the subspace $\{h^j=0\}$. 
We can define a Euler vector field $\rho|_{V_i}$, it suffices to pullback the vector field $\rho=h^j\partial_{h^j}$ in each local chart for $V_i$ and by the example \ref{fundexEuler} this is a Euler vector field. 
The vector fields $\rho_i=\rho|_{V_i}$ do not necessarily coincide on the overlaps $V_i\cap V_j$.
For any partition of unity $(\alpha_i)_i$ subordinated to this subcover, $\alpha_i\geqslant 0$, $\sum_i \alpha_i=1$, consider the vector field $\rho$ defined by the formula
\begin{equation}
\rho=\sum \alpha_i\rho_i  
\end{equation} 
then $ \forall f\in\mathcal{I}(U), \rho f-f= \sum \alpha_i\rho_i f- \sum \alpha_i f=\sum\alpha_i \left(\rho_i f-f \right)\in\mathcal{I}^2(U)$.
\end{proof}
We can find the general form for all possible Euler vector fields $\rho$ in arbitrary coordinate system $(x,h)$ where $I=\{h=0\}$. 
\begin{lemm}
$\rho|_U$ is $\textbf{Euler}$
if and only if for all $p\in I\cap U$, in $\textbf{any arbitrary}$ local chart $(x,h)$ centered at $p$ where $I=\{h=0\}$, $\rho$ has the standard form 
\begin{equation}\label{generalform}
\rho=h^j\frac{\partial}{\partial h^j} + h^iA_i^j(x,h)\frac{\partial}{\partial x^j} + h^ih^jB_{ij}^k(x,h)\frac{\partial}{\partial h^k}
\end{equation}
where $A,B$ are smooth functions of $(x,h)$.
\end{lemm} 
\begin{proof} 
We use the sum over repeated index convention. Let us start with an arbitrary $f\in \mathcal{I}(U)$.
%$\rho=B_i(x)\partial_{h^i}+ B_i^j(x)h^i \partial_{h^j}+ h^ih^jB_{ij}^k(x,h)\partial_{h^k} + L_i(x)\partial_{x^i}+ L_i^j(x,h)h^i \partial_{x^j} $
Set $\rho=B^i(x,h)\partial_{h^i}+L^i(x,h)\partial_{x^i}$ and we use $$f\in \mathcal{I}\implies f=h^j\frac{\partial f}{\partial h^j}(0,0)  + x^ih^j \frac{\partial^2f}{\partial x^i\partial h^j}(0,0) + O(\vert h\vert^2)  $$ 
First compute $\rho f$ up to order two in $h$:
$$\rho f= B^j(x,h) \partial_{h^j}f + L^i(x,h) \partial_{x^i}f$$ 
$$=B^j(x,h) \frac{\partial f}{\partial h^j}(0,0)+ B^j(x,h)x^i\frac{\partial^2 f}{\partial h^j\partial x^i}(0,0)+h^jL^i(x,h) \frac{\partial^2 f}{\partial h^j\partial x^i}(0,0)+O(\vert h\vert^2)$$
then the condition $\rho f-f\in\mathcal{I}^2$ reads
$\forall f\in \mathcal{I},$ $$ B^j(x,h) \frac{\partial f}{\partial h^j}(0,0)+ \left(B^j(x,h)x^i+h^jL^i(x,h) \right)\frac{\partial^2 f}{\partial h^j\partial x^i}(0,0)$$ $$=h^j\frac{\partial f}{\partial h^j}(0,0)  + x^ih^j \frac{\partial^2f}{\partial x^i\partial h^j}(0,0) + O(\vert h\vert^2)  $$
Now we set $f(x,h)=h^j$ which is an element of $\mathcal{I}$, and substitute it in the previous equation, by uniqueness of the Taylor expansion
$$B^j(x,h)=h^j + O(\vert h\vert^2) $$
but this implies
$$h^j \frac{\partial f}{\partial h^j}(0,0)+ h^jx^i\frac{\partial^2 f}{\partial h^j\partial x^i}(0,0)+h^jL^i(x,h) \frac{\partial^2 f}{\partial h^j\partial x^i}(0,0)$$ $$=h^j\frac{\partial f}{\partial h^j}(0,0)  + x^ih^j \frac{\partial^2f}{\partial x^i\partial h^j}(0,0) + O(\vert h\vert^2)$$ $$\implies h^jL^i(x,h) \frac{\partial^2 f}{\partial h^j\partial x^i}(0,0)= O(\vert h\vert^2)\implies L^i\in\mathcal{I} $$
finally $\rho=B^i(x,h)\partial_{h^i}+L^i(x,h)\partial_{x^i}$ where $B^j(x,h)=h^j+\mathcal{I}^2$ and $L^i\in\mathcal{I}$ which gives the final generic form. 
\end{proof} 
 
Fix $N$ an open neighborhood of $I$ with smooth boundary $\partial N$, the
boundary $\partial N$ forms a tube around $I$. If the Euler
$\rho$ restricted to $\partial N$ points outward, this means that the Euler $\rho$ can be exponentiated to generate a one-parameter group of local diffeomorphism: $t\mapsto e^{-t\rho}:N\mapsto N$, $N$ is thus \textbf{$\rho$-convex}. 
$I$ is the fixed point set of this dynamical system. 
The one parameter family acts on any section of a natural bundle functorially defined over $M$, hence on smooth compactly supported sections of the tensor bundles over $M$ particularly on $\Omega_c^d(M)$.
\begin{ex}
Choose a local chart $(x,h):U\mapsto \mathbb{R}^{n+d}$ where $I$ is given by $\{h=0\}$, 
the scaling $\left(e^{\log\lambda \rho *}f\right)$ 
satisfies the differential identity
\begin{equation} 
\lambda\frac{d}{d\lambda}\left(e^{\log\lambda \rho *}f\right)=\rho\left(e^{\log\lambda \rho *}f\right). 
\end{equation} 
In the case of the canonical Euler $\rho=h^j\frac{\partial}{\partial h^j}$, we also 
have identity: 
$$\lambda\frac{d}{d\lambda}f(x,\lambda h)=\left(h^j\frac{df}{d h^j}\right)(x,\lambda h)=\left(\rho f\right)(x,\lambda h),$$
from which we deduce that
$\left(e^{\log\lambda \rho *}f\right)(x,h)=f(x,\lambda h)$
which is true because both the l.h.s. and r.h.s. satisfy the differential equation
$(\lambda\frac{d}{d\lambda}-\rho) f=0 $
and coincide at $\lambda=1$.
\end{ex}
We generalize the definition
of weakly homogeneous distributions to the case
of manifolds
but this definition is $\rho$
dependent:
\begin{defi}
Let $U$ be $\rho$-convex open set. The set 
$E^\rho_s(U)$ is defined as the set of distributions $t\in \mathcal{D}^\prime(U)$ such that
$$\forall\varphi\in \mathcal{D}(U),\exists C(\varphi), \sup_{\lambda\in(0,1]} \vert\left\langle\lambda^{-s}t_\lambda,\varphi\right\rangle \vert\leqslant C(\varphi).$$ 
\end{defi}
\subsection{Invariances}
We gave a global definition of the space $E^\rho_s$ but this definition depends on the Euler $\rho$. 
%Set $I=\mathbb{R}^{n}\times\{0\}\subset \mathbb{R}^{n+d}$.
Recall that $G$ is the group of local diffeomorphisms preserving $I$.
On the one hand, we saw that the class of Euler vector fields is invariant by the action of $G$
on the other hand
it is not obvious 
that for any two Euler vector fields $\rho_1,\rho_2$, 
there is an element $\Phi\in G$ such that $\Phi_*\rho_1=\rho_2$. 
Denote by $S(\lambda)=e^{\log\lambda \rho}$ the scaling operator 
defined by the Euler $\rho$.
$S(\lambda)$ is a multiplicative group homomorphism, it satisfies the identity $S(\lambda_1)S(\lambda_2)=S(\lambda_1\lambda_2)$.
\begin{prop}\label{propositionvariablefamily}
Let $p$ in $I$, let $U$ be an open set containing $p$
and let $\rho_1,\rho_2$ be two Euler vector fields defined on $U$ and $S_a(\lambda)=e^{\log\lambda \rho_a},a=1,2$ 
the corresponding scalings.
Then there exists a neighborhood $V\subset U$ of $p$ and 
a one-parameter family of diffeomorphisms 
$\Phi\in C^\infty([0,1]\times V,V)$ 
such that, if for all $\lambda\in [0,1]$, $\Phi(\lambda)=\Phi(\lambda,.):V\mapsto V$, 
then $\Phi(\lambda)$
satisfies the equation:
$$ S_2(\lambda)=  S_1(\lambda)\circ \Phi(\lambda) .$$ 
% Let $X,H$ denote the components of $\Phi$ in $\mathbb{R}^{n+d}$, $\Phi:(x,h)\mapsto (X(x,h),H(x,h))$
%
% In coordinates, $\Phi(0):(x,h)\mapsto (X(x,0),\frac{DH}{Dh}(x,0)h)$ is linear in $h$.
\end{prop}
\begin{proof}
We use a local chart $(x,h):U\mapsto \mathbb{R}^{n+d}$ centered at $p$, where
$I=\{h=0\}$.
We set $\rho=h^j\partial_{h^j}$ and $S(\lambda)=e^{\log\lambda\rho}$ and we try to solve the two problems
$S(\lambda)^*t=\Phi_a(\lambda)^* \left( S_a(\lambda)^*t\right)$ for $a=1$ or $2$. We must have the following equation $$\Phi_a(\lambda)^*=S(\lambda)^*S_a^{-1}(\lambda)^* \implies \Phi_a(\lambda)=S_a^{-1}(\lambda)\circ S(\lambda) .$$ 
If so, the map $\Phi_a(\lambda)$ satisfies the differential equation
$$\lambda\frac{\partial}{\partial\lambda}\Phi_a(\lambda)^*=\lambda\frac{\partial}{\partial\lambda}S(\lambda)^*S_a^{-1}(\lambda)^*$$ 
$$=\rho S(\lambda)^*S_a^{-1}(\lambda)^*+ S(\lambda)^*(-\rho_a)S_a^{-1}(\lambda)^*$$
$$=\rho S(\lambda)^*S_a^{-1}(\lambda)^*+ S(\lambda)^*(-\rho_a)S^{-1}(\lambda)^*S(\lambda)^*S_a^{-1}(\lambda)^*$$ 
$$=\left(\rho - Ad_{S^{-1}(\lambda)}\rho_a\right)\Phi_a(\lambda)^* $$
$$\implies \lambda\frac{\partial}{\partial\lambda}\Phi_a(\lambda)=\left(\rho - S^{-1}(\lambda)_{\star}\rho_a \right)\left(\Phi_a(\lambda)\right) \text{ with } \Phi_a(1)=Id $$
where we used the Lie algebraic formula (\ref{pushforward}): $\left(\left(\Phi_*\rho\right)f\right)= \Phi^{-1*}\left(\rho \left(\Phi^*f\right)\right)=\left(Ad_{\Phi}\rho\right) f $. Let $f$ be a smooth function and $X$ a vector field. 
We use formula (\ref{pushforward}) to compute 
the pushforward of $fX$ by a diffeomorphism $\Phi$:
\begin{equation}\label{pushforwardgeneral}
L_{\Phi_*\left(fX\right)}\varphi=(\Phi^{-1*}f) L_{\Phi_*X}\varphi. 
\end{equation} 
We use the general form (\ref{generalform}) for a Euler vector field:  
$$\rho_a=h^j\frac{\partial}{\partial h^j} + h^iA_i^j(x,h)\frac{\partial}{\partial x^j} + h^ih^jB_{ij}^k(x,h)\frac{\partial}{\partial h^k}$$ 
hence we apply formula \ref{pushforwardgeneral}: $$S^{-1}(\lambda)_{*}\rho_a
= S^{-1}(\lambda)_{*}\left(h^j\frac{\partial}{\partial h^j}\right) + S^{-1}(\lambda)_{*}\left(h^iA_i^j\frac{\partial}{\partial x^j}\right) + S^{-1}(\lambda)_{*}\left(h^ih^jB_{ij}^k\frac{\partial}{\partial h^k}\right)$$
$$=(S(\lambda)^{*}h^j)S^{-1}(\lambda)_{*}\frac{\partial}{\partial h^j}+ S(\lambda)^*(h^iA_i^j)S^{-1}(\lambda)_{*}\frac{\partial}{\partial x^j} + S(\lambda)^*(h^ih^jB_{ij}^k)S^{-1}(\lambda)_{*}\frac{\partial}{\partial h^k} $$
$$=\lambda h^j \lambda^{-1}\partial_{h^j}+ \lambda h^iA_i^j(x,\lambda h)\frac{\partial}{\partial x^j} + \lambda^2h^ih^jB_{ij}^k(x,\lambda h)\lambda^{-1}\frac{\partial}{\partial h^k} $$
$$=h^j\partial_{h^j}+ \lambda h^iA_i^j(x,\lambda h)\frac{\partial}{\partial x^j} + \lambda h^ih^jB_{ij}^k(x,\lambda h)\frac{\partial}{\partial h^k}$$
$$\implies \rho - S^{-1}_{*}(\lambda)\rho_a =-\lambda\left( h^iA_i^j(x,\lambda h)\frac{\partial}{\partial x^j} + h^ih^jB_{ij}^k(x,\lambda h)\frac{\partial}{\partial h^k}\right).$$ 
If we define the vector field $X(\lambda)= -\left( h^iA_i^j(x,\lambda h)\frac{\partial}{\partial x^j} + h^ih^jB_{ij}^k(x,\lambda h)\frac{\partial}{\partial h^k}\right)$ then 
\begin{equation}\label{ODEhomotopy}
\frac{\partial \Phi_a}{\partial\lambda}(\lambda)=X\left(\lambda,\Phi_a(\lambda)\right) 
\text{ with }  \Phi_a(1)=Id. 
\end{equation}
$\Phi_a(\lambda)$ satisfies a non autonomous ODE, the vector field 
$$X(\lambda)=-\left( h^iA_i^j(x,\lambda h)\frac{\partial}{\partial x^j} 
+ h^ih^jB_{ij}^k(x,\lambda h)\frac{\partial}{\partial h^k}\right)$$ 
depends smoothly on $(\lambda,x,h)$.
We have to prove that by choosing a suitable neighborhood of $p\in I$, there is  
always a solution of (\ref{ODEhomotopy}) 
on the interval $[0,1]$ 
in the sense that 
there is no blow up at $\lambda=0$. 
For any compact $K\subset \{\vert h\vert\leqslant \varepsilon_1\}$, we have the estimates $\forall (x,h)\in K,\forall\lambda\in[0,1],  \vert h^ih^jB_{ij}(x,\lambda h)\vert\leqslant b\vert h\vert^2$ and
$\vert h^iA_i(x,\lambda h)\vert\leqslant a\vert h\vert$. 
Hence as long as $\vert h\vert\leqslant \varepsilon_1$, we have $\vert \frac{dh}{d\lambda}\vert\leqslant b\vert h\vert^2 \leqslant b\varepsilon_1 \vert h\vert$.
Then for any Cauchy data $(x(1),h(1))\in K$ such that $\vert h(1)\vert\leqslant \varepsilon_2 $, we compute the maximal interval $I=(\lambda_0,1]$ such that for all $\lambda\in [\lambda_0,1]$ we have $\vert h(\lambda)\vert \leqslant\varepsilon_1$. 
An application of Gronwall lemma (\cite{Tao} Theorem 1.17 p.~14)
to the differential inequality
$\vert \frac{dh}{d\lambda}\vert\leqslant b\varepsilon_1 \vert h\vert$
yields $\forall\lambda\in I, \vert h(\lambda)\vert\leqslant e^{(1-\lambda)\varepsilon_1 b}\vert h(1)\vert$.
Hence, if we choose $\lambda$ in such a way that $e^{(1-\lambda)\varepsilon_1 b}\varepsilon_2 \leqslant \varepsilon_1 $ then $ \vert h(\lambda)\vert\leqslant e^{(1-\lambda)\varepsilon_1 b}\vert h(1)\vert\leqslant e^{(1-\lambda)\varepsilon_1 b}\varepsilon_2 \leqslant \varepsilon_1 $ thus $\lambda\in I$ by definition.
Hence, we conclude that if we choose $\varepsilon_2\leqslant \frac{\varepsilon_1}{e^{\varepsilon_1 b}}$ then $$ 
 [0,1]=\{ \lambda |e^{(1-\lambda)\varepsilon_1 b}\frac{\varepsilon_1}{e^{\varepsilon_1 b}} \leqslant \varepsilon_1\}\subset 
\{ \lambda |e^{(1-\lambda)\varepsilon_1 b}\varepsilon_2 \leqslant \varepsilon_1\}\subset I $$ 
and by classical ODE theory the equation (\ref{ODEhomotopy}) always has a smooth solution $\lambda\mapsto \Phi_a(\lambda)$ on the interval $[0,1]$, the open set $V$ 
on which this existence result holds is the restriction 
of the chart $U\cap \{\vert h\vert\leqslant \varepsilon_2\}$.   
% $ S_\lambda^*(\Phi^*t)=\Phi_\lambda^* (S_\lambda^*t)$. From $$S_\lambda \circ \Phi_\lambda=\Phi\circ S_\lambda\Rightarrow \Phi_\lambda=S_{\lambda^{-1}}\circ \Phi \circ S_\lambda ,$$
%hence it is natural to set $\Phi_\lambda=S_{\lambda^{-1}}\circ \Phi \circ S_\lambda$ and see if $\Phi_\lambda$ fulfills our requirements. In coordinates $\Phi(x,h)=(X(x,h),H(x,h))$, this means $\Phi_\lambda(x,h)=(X(x,\lambda h),\lambda^{-1}H(x,\lambda h) )$. Because $\Phi$ sends $\{h=0\}$ to $\{ H=0 \}$,  
%$H(x,0)=0$.
%By Hadamard lemma, there is a smooth function $g$ such that
%$H=hg(x,h) $. Now we deduce $\Phi_\lambda(x,h)=(X(x,\lambda h),hg(x,\lambda h) )$ has a limit $\Phi_0=(X(x,0),hg(x,0))$. 
Now, to conclude properly in the case both $\rho_1,\rho_2$ are not equal to $\rho=h^j\frac{\partial}{\partial h^j}$ then we apply the previous result 
$$S(\lambda) = S_1(\lambda)\circ \Phi_1(\lambda)= S_2(\lambda)\circ \Phi_2(\lambda)\implies S_2(\lambda)=S_1(\lambda)\circ \Phi_1(\lambda)\circ \Phi^{-1}_2(\lambda) $$
hence $ S_2(\lambda)^*t= \left(\Phi_1(\lambda)\circ \Phi^{-1}_2(\lambda)\right)^*S_1(\lambda)^*t  $
\end{proof}
We keep the notations and assumptions of the above proposition and proof, we give an elementary proof of the conjugation without the use of Sternberg Chen theorem:
\begin{coro}\label{conjugcoro}
Let $\rho_a,a=(1,2)$ be two Euler vector fields and $S_a(\lambda)=e^{\log\lambda \rho_a},a=(1,2)$ the two corresponding scalings.
In the chart $(x,h), I=\{h=0\}$ around $p$, let $\rho=h^j\partial_{h^j}$ be the canonical Euler vector field and $S(\lambda)=e^{\log\lambda \rho}$ the corresponding scaling and $\Phi_a(\lambda)$ be the family of diffeomorphisms $\Phi_a(\lambda)=S_a^{-1}(\lambda)\circ S(\lambda)$ 
which
has a $\textbf{smooth limit}$ $\Psi_a=\Phi_a(0)$ 
when $\lambda\rightarrow 0$.
Then
$\Psi_a\in G$ locally 
conjugates the hyperbolic dynamics:
\begin{eqnarray}
\forall\mu, \Psi_a\circ S(\mu)\circ\Psi^{-1}_a=S_a(\mu)\\
\Phi_a(\mu)=\Psi_a\circ S(\mu^{-1}) \circ \Psi_a^{-1}\circ S(\mu)\\
\rho_a=\Psi_{a\star}\rho.
\end{eqnarray}
\end{coro}
Hence in any coordinate chart, in a neighborhood of any point $(x,0)\in I$, 
all Euler are locally conjugate by an element of $G$ to the standard Euler $\rho=h^j\partial_{h^j}$.
\begin{proof}
The map $\lambda\mapsto S(\lambda)$ is a group homomorphism from $(\mathbb{R}^*,\times)\mapsto (G,\circ)$:
$$\Phi_a(\lambda)\circ S(\mu)=\left(S_a^{-1}(\lambda)\circ S(\lambda)\right)\circ S(\mu)=S_a^{-1}(\lambda)\circ S(\lambda\mu)$$ $$=S_a(\mu) \circ S_a^{-1}(\mu)\circ S_a^{-1}(\lambda)\circ S(\lambda\mu)=S_a(\mu)\circ S_a^{-1}(\lambda\mu)\circ S(\lambda\mu)=S_a(\mu)\circ\Phi_a(\lambda\mu)$$ 
finally $\forall (\lambda,\mu) $, we find $\Phi_a(\lambda)\circ S(\mu)=S_a(\mu)\circ\Phi_a(\lambda\mu)\implies  \Phi_a(0)\circ S(\mu)=S_a(\mu)\circ\Phi_a(0)$ at the limit when $\lambda\rightarrow 0$ where the limit makes sense because $\Phi_a$ is smooth in $\lambda$ at $0$. 
To obtain the pushforward equation $\rho_a=\Psi_{a\star}\rho$, just differentiate the last identity w.r.t. $\mu$.
\end{proof}
Beware that the conjugation theorem 
is only true 
in a neighborhood $V_p$
of some given point $p\in I$.
$\rho_1,\rho_2 $ are \emph{not necessarily conjugate globally} 
in a neighborhood of $I$. 
There might be topological obstructions 
for a global conjuguation. 
The local diffeomorphism $\Psi=\Phi_a(0)$ makes the following diagram
$$
\begin{array}{ccccc}
&V &\overset{S(\lambda)}{\rightarrow} & V & \\
\Psi&\downarrow & & \downarrow & \Psi\\
 &V &\underset{S_a(\lambda)}{\rightarrow} & V&
\end{array}
$$
commute. 
We keep the notational conventions of the above corollary:
\begin{lemm}\label{localthmrho}
Let $p$ in $I$ and $U$ be an open set containing $p$, 
let $\rho_1,\rho_2$ be 
two Euler vector fields 
defined on $U$ 
%and 
%$\Psi:U\mapsto U$
%which makes the diagram
%$$
%\begin{array}{ccc}
%U &\overset{S_1(\lambda)}{\rightarrow} & U\\
%\Psi\downarrow & & \downarrow\Psi\\
%U &\underset{S_2(\lambda)}{\rightarrow} & U
%\end{array}
%$$
%commutes, 
then there exists 
an open neighborhood 
$V$ of $p$
on which $\forall s, E_s^{\rho_1}(V)= E_s^{\rho_2}(V)$. 
\end{lemm}
\begin{proof}
%But $\bigcover_{p\in I}V_p $ covers $I$ hence we can extract a locally finite subcover $(V_i)_i$, in each $V_i$ the vector fields $\rho_1,\rho_2$ are conjuguated by a local diffeomorphism $g_i\in G$ preserving $I$. We pick a partition of unity $\varphi_i$ subordinated to this cover and $\sum\varphi_i g_i$ is a global diffeomorphism in the neighborhood $U^\prime=\bigcup_{i}V_i$ which conjuguates $\rho_1$ and $\rho_2$.
%
% Then on the smaller neighborhood $U^\prime$ of $I$ such that $U^\prime\subset U$, the vector fields $\rho_1$ and $\rho_2$ are conjuguated.
%
Set $\Phi(\lambda)=S_1^{-1}(\lambda)\circ S_2(\lambda)$, 
$\Phi$ depends smoothly in $\lambda$ by 
Proposition
\ref{propositionvariablefamily} and 
$V=\bigcap_{\lambda\in[0,1]}\Phi^{-1}(\lambda)(U)$.
$$\forall \varphi\in\mathcal{D}(V),\lambda^{-s}\left\langle S_2(\lambda)^*t,\varphi \right\rangle
=\lambda^{-s}\left\langle \Phi^{*}(\lambda)\left(S_1(\lambda)^*t\right),\varphi \right\rangle$$
$$=\lambda^{-s}\left\langle S_1(\lambda)^*t , \underset{\text{bounded in }\mathcal{D}(U)}{\underbrace{\left(\Phi(\lambda)^{-1*}\varphi\right) \vert \det(D\Phi(\lambda)^{-1}) \vert}}\right\rangle $$
which is bounded by the hypothesis $t\in E_s^{\rho_1}$ which 
means by definition that $\lambda^{-s}S_1(\lambda)^*t$ is bounded in $\mathcal{D}^\prime(U)$. 
%then we notice that if $\varphi$ is supported in $K\subset V$ then $\forall \lambda\in [0,1]$, we find $$\forall\lambda, \text{supp }\Phi^{-1}(\lambda)^*\varphi \vert \det(D\Phi^{-1}) \vert(\lambda)\subset \bigcup_{\lambda\in [0,1]}\Phi(\lambda)(K) \subset U$$
%where $K^\prime=\bigcup_{\lambda\in [0,1]}\Phi(\lambda)(K)$ is a compact set because  $\Phi([0,1]\times K)$ is the image of a compact set by a smooth map $\Phi$, hence compact, $K^\prime$ is the projection of $\Phi([0,1]\times K)$ on $\mathbb{R}^{n+d}$ hence it is compact.
%
% Use the Banach Steinhaus theorem for the bounded family $\lambda^{-s} S_1(\lambda)^*t$:
%$$\vert\lambda^{-s}\left\langle S_1(\lambda)^*t , \Phi^{-1}(\lambda)^*\varphi \vert \det(D\Phi^{-1}) \vert(\lambda) \right\rangle \vert\leqslant C_{K^\prime} \sup_{\vert\alpha\vert\leqslant m} \vert \partial^\alpha \Phi^{-1}(\lambda)^*\varphi \vert \det(D\Phi^{-1}) \vert(\lambda) \vert_{L^\infty}  $$
%then we conclude by noticing the map $\Phi$ is smooth in $[0,1]\times \mathbb{R}^{n+d}$ hence the function $ \Phi^{-1}(\lambda)^*\varphi \vert \det(D\Phi^{-1}) \vert(\lambda)$ is smooth in $[0,1]\times \mathbb{R}^{n+d}$ hence the family of test functions $ \Phi^{-1}(\lambda)^*\varphi \vert \det(D\Phi^{-1}) \vert(\lambda)$ on the right hand side is bounded in $\mathcal{D}(K^\prime)$.
\end{proof}
We illustrate the previous method in the following example:
\begin{ex}
We work in $\mathbb{R}^2,n=d=1$ with coordinates $(x,h)$, let $\rho_1=h\partial_h, \rho_2=h\partial_h+h\partial_x$. 
Let $t(x,h)=f(x)g(h)$  
where $f$ is an arbitrary distribution and $g$ is homogeneous of degree $s$: 
$$ \lambda^{-s}g(\lambda h)=g(h) .$$
Then $t$ is homogeneous of degree $s$ with respect to $\rho_1$ thus $t\in E_s^{\rho_1}$. 
We will study the scaling behaviour when we scale with $\rho_2$, 
$S_2(\lambda)(x,h)=e^{\log\lambda\rho_2}(x,h)=(x+(\lambda-1)h,\lambda h)$:
$$\int_{\mathbb{R}^2} S_2(\lambda)^*\left( f(x) g(h)\right) \varphi(x,h) dxdh=\int_{\mathbb{R}^2}  f\left(x + (\lambda-1) h\right)g(\lambda h) \varphi(x,h) dxdh $$
Use Proposition (\ref{propositionvariablefamily}) and first determine $\Phi(\lambda)$ in such a way that the equation $\forall\lambda, S_2(\lambda)=S_1(\lambda) \circ \Phi(\lambda)  $ is satisfied.
We find $\Phi(\lambda)(x,h)=S^{-1}_1(\lambda)\circ S_2(\lambda)=S_1^{-1}(\lambda)(x+(\lambda-1)h,\lambda h)=(x+(\lambda-1)h, h)$.
Applying the previous result to our example reduces to a 
simple change of variables
in the integral:
$$\int_{\mathbb{R}^2} S_2(\lambda)^*\left( f(x) g(h)\right) \varphi(x,h) dxdh=\int_{\mathbb{R}^2} f(x)g(\lambda h) \varphi(x+(1-\lambda)h,h) dxdh$$ 
$$=\lambda^s\int_{\mathbb{R}^2}  f(x)g(h) \underset{\text{bounded family of test functions}}{\underbrace{\varphi(x+(1-\lambda)h,h)}} dxdh  .$$
Then the result is straightforward 
and we can conclude $t\in E_s^{\rho_2}$. 
\end{ex}
\paragraph{Local invariance}
\begin{defi}
A distribution $t$ is said to be locally $E_s^\rho$ at $p$ if
there exists
an open $\rho$-convex set $U\subset M$ such that $\overline{U}$ is a neighborhood of $p$ and such that $t\in E_s^{\rho}(U)$.
\end{defi}
Corollary (\ref{conjugcoro}) and 
lemma (\ref{localthmrho}) 
imply the following local statement:
\begin{thm}\label{locthmGOOD}
Let $p\in I$, if $t$ is locally $E_s^\rho$ at $p$ for some Euler vetor field $\rho$, then it is so
for any other Euler vector field.
\end{thm}
A comment on the statement of the theorem, first the definition
of $\rho$-convexity allows $U$ to have \emph{empty intersection} with $I$, 
because
the definition of $\rho$-convexity is 
$\forall p\in U,\forall \lambda\in(0,1], S(\lambda)[p]\in U$, 
the fact that $\lambda>0$ allows the case of empty intersection with $I$.
The previous theorem allows to give a definition of the space of distributions $E_s(U)$ that are weakly homogeneous of degree $s$ which makes no mention of the choice of Euler vector field:
\begin{defi}\label{defEs} 
A distribution $t$ is weakly homogeneous of degree $s$ at $p$ if $t$ is locally $E_s^\rho$ at $p$ for some $\rho$. $E_s(U)$ is the space of all distributions $t\in\mathcal{D}^\prime(U)$ such that
$\forall p\in (I\cap \overline{U})$, $t$ is weakly homogeneous of degree $s$ at $p$.
\end{defi}
If we look at the definition \ref{defEs}, 
and we take into account
that the space of distributions 
on open sets of $M$ forms a sheaf,
we deduce the following gluing property:
if there is a collection $U_i$ and a collection $t_i\in\mathcal{D}^\prime(U_i)$ 
such that $\forall i,$
$t_i\in E_s(U_i)$ and $t_i=t_j$ on every intersection $U_i\cap U_j$,
then for $U=\cup_i U_i$ there is a unique $t\in \mathcal{D}^\prime(U)$
which lives in $E_s(U)$ and coincides with $t_i$ on $U_i$ for all $i$.
From this gluing property, 
and since the property of being weakly 
homogeneous of degree $s$ at
$p$ is \emph{open}, 
we can deduce that
it is sufficient
to check the
property on a
cover $(U_i)_i$ of $U$
by 
local charts $(x,h)_i:U_i\subset N \mapsto \Omega_i\subset\mathbb{R}^{n+d}$
where $t|_{U_i}$
is in $E_s^{\rho_i}(U_i)$
for the canonical
Euler
$\rho_i$
given by the
chart.
\begin{thm}\label{thmfi}
Let $U$ be an open neighborhood of $I\subset M$, if $t\in E_s(U\setminus I)$ then there exists
an extension $\overline{t}$ in $E_{s^\prime}(U)$ where $s^\prime=s$ if $-s-d\notin\mathbb{N}$ 
and
$s^\prime<s$ otherwise.
\end{thm}
Apply Theorem \ref{locthmGOOD}, restrict to local charts $(x,h)_i:(U_i\setminus I)\mapsto (\Omega_i\setminus I)$ where
$t|_{U_i\setminus I}=t_i\circ(x,h)_i$ where $t_i\in E_s(\Omega_i\setminus I)$, then extend each $t_i$ on $\Omega_i$, $\overline{t_i}\in E_s(\Omega_i)$, pullback the extension denoted by $\overline{t|_{U_i}}\in E_s(U_i)$ on $U_i$, then glue together all $\overline{t|_{U_i}}$ (they coincide on $(U_i\cap U_j)\setminus I $ but might not coincide on $I$ but this does not matter !) 
by a partition of unity $(\varphi_i)_i$
subordinated to the cover $(U_i)_i$.
The extension reads $\overline{t}=\sum_i \varphi_i\overline{t|_{U_i}}$.

\paragraph{The extension depends only on $\rho,\chi$.}
Instead of using the Taylor expansion in local coordinates,
we can use the identity  
$$\sum_{\vert\alpha\vert=n} \frac{h^\alpha}{\alpha!}\partial_h^\alpha f(x,0)=\frac{1}{n!} \left(\left(\frac{d}{dt}\right)^n e^{\log t\rho*}f\right)|_{t=0}(x,h) $$ 
%
% We notice that the flow $e^\log t \rho$ commutes with the action of the Lie derivatives $\rho$, hence $\forall f, e^{\log t \rho*}\left(\rho f\right)=\rho \left(e^{\log t \rho*} f\right)  $. 
% 
We can define the counterterms 
and the 
renormalized distribution 
by the equations:
\begin{eqnarray}\label{equationscounterterms}
\left\langle \tau^\lambda,\varphi \right\rangle=\lim_{t\rightarrow 0}\left\langle te^{-\log\lambda \rho*}\left(-\rho\chi \right),\sum_0^{m} \frac{1}{n!}\left(\frac{d}{dt}\right)^ne^{\log t\rho*}\varphi  \right\rangle \hfill \\
\left\langle \overline{t},\varphi \right\rangle= \left\langle te^{-\log\lambda \rho*}\left(-\rho\chi \right), I_m\left(\varphi\right)  \right\rangle+\left\langle t(1-\chi),\varphi  \right\rangle\\
I_m(\varphi)= \int_0^1\frac{d\lambda}{\lambda}   \frac{1}{m!}\int_0^1 ds(1-s)^{m}\left(\frac{\partial}{\partial s}\right)^{m+1}e^{\log s\rho\star}\varphi
\end{eqnarray}
where we made an effort to convince 
the reader that the formulas 
only depend on $\rho$ and $\chi$.

\section{Appendix.}
\paragraph{The Banach--Steinhaus theorem.}
We will frequently use the Banach--Steinhaus theorem in more general spaces than Banach spaces.
We recall here basic results about Fr\'echet spaces using Gelfand--Shilov \cite{GS2} as our main reference.
Let $E$ be a locally convex topological vector space where the topology is 
given by a countable family of norms, 
ie $E$ is a Fr\'echet space in modern terminology 
and ``countably normed space'' in Gelfand--Shilov terminology.
Hence it is a $\textbf{complete metric space}$ 
(the topology induced by the metric 
is exactly the same as the topology induced by the family of norms) 
(section $3.4$ in \cite{GS2}).
Following \cite{GS2}, we assume 
the family of norms 
defining the topology are ordered
$\Vert .\Vert_p \leqslant \Vert .\Vert_{p+1} ,$
where we denote by $E_p$ the completion of $E$ 
with respect to the norm $\Vert.\Vert_p$ 
which makes $E_p$ a \textbf{Banach space}. 
Then we have the sequence of continuous inclusions 
$E= ...\subset E_{p+1} \subset E_p\subset ...$ and $E=\bigcap_p E_p$.  

 A $\textbf{complete metric space}$ satisfies the Baire property: 
any countable union of closed sets with empty interior has empty interior. 
A consequence of the Baire property is that if a set $U\subset E$ is closed, convex, centrally symmetric ($U=-U$) and absorbant, then it must contain a neighborhood of the origin for the Fr\'echet topology of $E$.
In 4.1 of \cite{GS2}, 
starting from the definition of the continuity of a linear map $\ell$ on $E$, the authors deduced the existence of $p$ and the corresponding seminorm $\Vert.\Vert_p$ such that $\forall x\in E, \ell(x)\leqslant C\Vert x\Vert_p$. 
Following the interpretation of 4.3 in \cite{GS2}, if we denote by $E_p$ the completion of $E$ relative to the norm $\Vert.\Vert_p$ then $\ell$ defines by \textbf{Hahn--Banach} a non unique element of $E_p^\prime$, the topological dual of $E_p$. 
Then the main theorem of 5.3 in \cite{GS2} 
characterizes strongly bounded sets in the topological dual $E^\prime$ of $E$. A set $B\subset E^\prime$ is strongly bounded iff there is $p$ such that $B\subset E_p^\prime$ and elements of $B$ are bounded in the norm of $E_p^\prime$.
$$\exists C, \forall f\in B, \sup_{\Vert \varphi\Vert_p\leqslant 1} \vert\left\langle f, \varphi \right\rangle\vert\leqslant C .$$
The weak topology in $E^\prime$ is generated by the collection of open sets 
$$\{f ;  \vert \left\langle f , \varphi\right\rangle \vert <\varepsilon \}$$
By definition, if $A$ is a weakly bounded set, then:
$$\forall \varphi, \sup_{f\in A}\vert \left\langle f , \varphi\right\rangle \vert < \infty .$$
In $5.5$ it is proved that weakly bounded sets of $E^\prime$ are in fact strongly bounded in $E^\prime$.
Let $A$ be a weakly bounded set in $E^\prime$. Then the set 
$B=\{ \varphi ; \forall f\in A, \vert \left\langle f , \varphi\right\rangle \vert<1  \} $ 
is closed, convex, centrally symmetric ($U=-U$) and absorbant 
therefore 
it must contain a neighborhood 
of the origin by lemma of section $3.4$.
$$ \{ \Vert \varphi\Vert_p\leqslant C \} \subset B$$
for a certain seminorm $\Vert.\Vert_p$ by definition 
of a neighborhood of the origin in a Fr\'echet space.
By definition elements of $A$ are bounded on this neighborhood of the origin
$$\forall f\in A, \varphi\in B, \vert \left\langle f , \varphi\right\rangle \vert<1 $$
$$\implies \forall f\in A, \Vert\varphi\Vert_p\leqslant C, \vert \left\langle f , \varphi\right\rangle \vert<1 $$
$$\implies  \forall f\in A, \vert \left\langle f , \varphi\right\rangle \vert\leqslant C^{-1} \Vert\varphi\Vert_p .$$

 Now we will apply these 
abstract results in the case of bounded families of distributions:  
\begin{thm}
Let $U\subset\mathbb{R}^d$ be an open subset.
If $A$ is a weakly bounded family of distributions in $\mathcal{D}^\prime(U)$ :
$$\forall \varphi\in \mathcal{D}(U), \sup_{t\in A}\left\langle t,\varphi \right\rangle <\infty $$
then for all compact subset $K\subset U$:
$$\exists p,\exists C_K, \forall t\in A, \forall \varphi\in \mathcal{D}_K(U), \vert \left\langle t, \varphi\right\rangle \vert\leqslant C_K \pi_p(\varphi) .$$
\end{thm}
\begin{proof}
Set $\Vert \varphi\Vert_p= \pi_p( \varphi)$, it is well known this is a norm.
The family $A$ is $\textbf{weakly bounded}$ in the dual $\mathcal{D}^\prime(K)$ of the Fr\'echet space $\mathcal{D}(K)=\bigcap_{k}C_0^k(K)$ ie the intersection of all spaces of $C^k$ functions supported in $K$. It is thus strongly bounded in the dual space $\mathcal{D}^\prime(K)$ and translating the strong boundedness into estimates yields the result.
\end{proof}
\begin{thm}
Let $K$ be a fixed compact subset of $\mathbb{R}^d$.
If $A$ is a family of distributions in $\mathcal{D}_K^\prime\left(U\right)$ $\textbf{supported}$ on $K\subset U$ and
$$\forall \varphi\in C^\infty\left(U\right), \sup_{t\in A}\left\langle t,\varphi \right\rangle <\infty ,$$
then $\forall K_2 $ which is a compact neighborhood of $K$, $\exists p,\exists C,$
$$ \forall t\in A, \forall \varphi\in C^\infty(U), \vert \left\langle t , \varphi\right\rangle \vert\leqslant C\pi_{p,K_2}(\varphi).$$
\end{thm}
\begin{proof} 
In the second case, first we find a compact set $K_2$ such that $K_2$ is a neighborhood of $K$. We set the Fr\'echet $E=\bigcap_k C_0^k(K_2)$ which is the intersection of all $C^k$ functions supported in $K_2$. These functions should not necessarily vanish on the complement of $K$. Then we pick any plateau function $\chi$ 
such that $\chi|_K=1$ and $\chi=0$ on the complement of $K_2$. $t\in A$ is supported on $K$ thus $\forall t\in A, \forall \varphi\in C^\infty(U), \vert \left\langle t , \varphi\right\rangle \vert=\vert \left\langle t , \chi\varphi\right\rangle \vert $ then we reduce to the previous
theorem:
$\forall t\in A, \forall \varphi\in C^\infty(U), \vert \left\langle t , \varphi\right\rangle \vert=\vert \left\langle t , \chi\varphi\right\rangle \vert\leqslant C_{K_2}\sup_{\vert\alpha\vert\leqslant p} \vert \partial^\alpha \chi\varphi \vert_{L^\infty}\leqslant C\sup_{\vert\alpha\vert\leqslant p} \vert \partial^\alpha \varphi \vert_{L^\infty(K_2)}$. 
\end{proof}

\begin{coro}
Let $U$ be an arbitrary open domain, $t\in E_s(U)$ iff $t\in \mathcal{D}^\prime(U)$ is a distribution on $U$ 
$$\forall\varphi\in \mathcal{D}(U),\exists C(\varphi), \sup_{\lambda\in[0,1]} \vert\lambda^{-s}t_\lambda,\varphi \vert\leqslant C(\varphi)$$ 
$$\Leftrightarrow \forall K\subset U,  \exists (p,C_K), 
\forall \varphi\in \mathcal{D}_K(U),  
\sup_{\lambda\in[0,1]} \vert\lambda^{-s}t_\lambda,\varphi \vert\leqslant C_K \pi_p(\varphi).
$$
\end{coro}

\chapter{A prelude to the microlocal extension.}
\subsection{Introduction.}
First, let us recall the problem which was solved in Chapter $1$. 
We started from a smooth
manifold $M$ and
a closed embedded submanifold $I\subset M$.
We defined
a general setting 
in which
we could scale
transversally
to $I$ using
the flow generated by 
a class of 
vector fields
called \textbf{Euler} 
vector fields. 
Then for each distribution
$t\in\mathcal{D}^\prime(M\setminus I)$ 
which was weakly homogeneous of degree $s$
in some precise sense 
(we called $E_s(M\setminus I)$ the space of such distributions):\\
- the notion of 
weak homogeneity
was made 
independent
of the choice of Euler 
$\rho$,\\
- we proved that $t$
has an extension 
$\overline{t}\in E_{s^\prime}(M)$
for some $s^\prime$.
We also understood that
the problem
of extension is
essentially a local problem
and that everything
can be reduced
to the extension problem
in $\mathbb{R}^{n+d}$ with coordinates $(x,h)$, 
$I=\mathbb{R}^n\times\{0\}=\{h=0\}$
and where the
scaling
is defined by 
$\rho=h^j\frac{\partial}{\partial h^j}$.
All the ``geometry''
is somehow contained
in the possibility
of choosing another 
Euler vector field.
In fact, 
the pseudogroup $G$ 
of local
diffeomorphisms
of $\mathbb{R}^{n+d}$ preserving
$I$
acts on the
space of
Euler vector fields.

 However 
this gives no a priori information 
on the wave front set 
of the extension $\overline{t}$. 
But in QFT, 
we need conditions on $WF(\overline{t})$ 
in order to define products of distributions.
By the pull-back theorem of H\"ormander
(\cite{Hormander} thm $8.2.4$), 
there is no reason for $WF(t_\lambda)$ to be equal to $WF(t)$. Hence in order
to control the wave front set of $\overline{t}$, 
the first step is to build some cone $\Gamma$ 
which bounds 
the wave front set of 
all scaled distributions $t_\lambda$ 
and a natural candidate for $\Gamma$
is 
$\Gamma=\bigcup_{\lambda\in(0,1]}WF(t_\lambda) $.  
% More precisely, we characterize $\Gamma$ with the following \textbf{universal property}:
%\begin{equation}\label{universalequation}
%WF(t)\subset \Gamma\implies\forall\lambda\in(0,1], WF(t_\lambda)\subset\Gamma 
%\end{equation}
%$\Gamma$ must solve this geometric equation.  
%
%
% 
% Then we describe the solutions of this geometric equation (\ref{universalequation}).
%
% If a generalized Euler $\rho$ is given, there is a local foliation of a neighborhood of $I$ which leaves is the set of all points $p$ which have same endpoint by the flow $\lim_{t\rightarrow \infty} e^{-t\rho}(p)$. To each foliation, we associate the union of the conormal bundle of each leaves and call it $C_\rho$.
We denote by $(x,h;k,\xi)$ the coordinates
in $T^\star \mathbb{R}^{n+d}, 
(x;k)\in T^\star\mathbb{R}^{n},
(h;\xi)\in T^\star\mathbb{R}^{d}$.
We use the notation $T^\bullet M$
for the cotangent bundle $T^\star M$
with the zero section
removed.
Denote by $C_\rho$ the set 
$\{(x,h;k,0)|k\neq 0\}\subset T^\bullet \mathbb{R}^{n+d}$. 
We call $C=\{(x,0;0,\xi)|\xi\neq 0\}$ 
the intersection of the conormal bundle of $I$ 
with $T^\bullet \mathbb{R}^{n+d}$.
In the first part of this Chapter, 
we will explain the geometric
interpretation
of the set $C_\rho$
and how it
depends on the choice 
of Euler $\rho$.
$C_\rho$ plays an important
role for the determination 
of the
analytical 
structure of 
local counterterms: 
if $WF(t)$ does not meet 
$C_\rho=\{(x,h;k,0)|k\neq 0\}$, 
then the \textbf{local counterterms} 
constructed from $t$
(\ref{equationscounterterms}) 
are distributions with 
wave front set 
in the \textbf{conormal} 
(we meet  
them again 
in Chapter $6$
under the form of
anomaly counterterms). 
Whereas the condition
$WF(t)\cap C_\rho=\emptyset$ 
depends on the choice 
of $\rho$,
the stronger
condition 
$\overline{WF(t)}|_I\subset C$
does not depend on
$\rho$
and
implies that
for any choice of Euler  
$\rho$, 
$WF(t)\cap C_\rho=\emptyset$
in some neighborhood
of $I$. 
\paragraph{The problem of the closure of $\Gamma$ over $I$.}
So we are led 
to study 
under which conditions on $WF(t)$ 
the cone $\Gamma$ defined by 
$\Gamma=\bigcup_{\lambda\in(0,1]}WF(t_\lambda) $ 
satisfies the constraint $\overline{\Gamma}|_I\subset C$,
where $\overline{\Gamma}$ is the closure of 
$\Gamma\subset T^\bullet\left(M\setminus I\right)$ in $T^\bullet M$.
Then we find a necessary and sufficient condition
on $WF(t)$ which 
we call
\textbf{soft landing condition}
for 
$\overline{\Gamma}|_I$ to lie in 
$C$.
%\begin{thm}\label{keypropforcut}
% 
% Let $\rho$ be a generalized Euler vector field. Let $WF(t)$ be given, if $WF(t)$ satisfies the \textbf{metric condition}
% For all compact set $K$ which is geodesically convex with respect to scaling by $\rho$,
%\begin{equation}\label{Broudermetric}
%\exists\varepsilon>0,\exists\delta>0, WF(t)|_{\vert h\vert\leqslant \varepsilon}\subset\{\vert k\vert\leqslant \delta\vert h\vert\vert\xi\vert \}
%\end{equation}     
%is \emph{equivalent} to the fact that the conic set $\Gamma=\{(x,\lambda^{-1}h,k,\lambda\xi)|\lambda\in(0,1],\vert\lambda^{-1}h\vert\leqslant\varepsilon ,(x,h,k,\xi)\in WF(t)   \}$ satisfies
%\begin{eqnarray}
%\overline{\Gamma|_I }\subset C=(TI)^\perp
%\end{eqnarray}
%which is turn equivalent to the fact that 
%for any generalised Euler vector field $\rho^\prime$ where $\rho^\prime$ is not necessarily equal to $\rho$, 
%\begin{equation}
%\Gamma\cap C_{\rho^\prime}=\emptyset
%\end{equation}
The fact that
$WF(t)$ satisfies the soft landing condition guarantees that whatever generalized Euler vector field $\rho$ we choose, the counterterms are conormal distributions supported on $I$. Furthermore, it is a condition which allows 
to control 
the wave front set of the 
extension 
as we will see in Chapter $3$.
\paragraph{The soft landing condition 
is not sufficient 
in order to control the wave front set.}
Assume that $t\in E_s(M\setminus I)$ 
and $WF(t)$ satisfies
the soft landing
condition.
Under these 
assumptions, 
we address the question: 
in which sense 
$\lim_\varepsilon\int^1_\varepsilon \frac{d\lambda}{\lambda} t\psi(\frac{h}{\lambda})$ 
converges to $\overline{t}$ ? 
More precisely for 
what
topology on
$\mathcal{D}^\prime(M)$ 
do
we have convergence ?
We already know from
Theorems \ref{thm1} and \ref{thm2} 
in Chapter $1$ 
that the integral converges 
in the \textbf{weak topology} of $\mathcal{D}^\prime$ but
this is not sufficient 
since 
it 
does not
imply 
the convergence
in stronger topologies which
control wave front sets
as the following 
examples show:
indeed in (\ref{contrexamplequitue}), 
we construct 
a distribution $t$ 
such that 
$t(1-\chi_{\varepsilon^{-1}})\underset{\varepsilon\rightarrow 0}{\rightarrow}t$
in $\mathcal{D}^\prime$,
whereas 
$\forall\varepsilon\in(0,1],t(1-\chi_{\varepsilon^{-1}})$ 
is smooth in 
$M\setminus I$, 
the wave front of 
$t$
can contain any ray $p\in T^\bullet M|_I$ 
in the cotangent cone over $I$.
Our example shows that generically,
we cannot control the wave
front set 
of $\lim_{\varepsilon\rightarrow 0}t(1-\chi_{\varepsilon^{-1}})$ 
even if the limit exists in $\mathcal{D}^\prime$ 
and each $t(1-\chi_{\varepsilon^{-1}})\in \mathcal{D}^\prime_\Gamma$ 
has wave front set in a given cone $\Gamma$.
Thus our assumptions that
$t\in E_s(M\setminus I)$ and $WF(t)$ 
satisfies the soft landing condition
are \textbf{not sufficient 
to control the wave front set of the extension $\overline{t}$}.
We will later prove in Chapter $3$, 
that the supplementary condition 
that $\lambda^{-s}t_\lambda$ be \textbf{bounded in $\mathcal{D}^\prime_\Gamma(M\setminus I)$}
(see Definition \ref{dprimegam}) is sufficient to have the estimate $WF(\overline{t})\subset \overline{WF(t)}\cup C$.  
%\paragraph{Converse result.}  
%% First, we find an example where we have to add the conormal bundle $C$ to $\overline{\Gamma}$: $WF(t)\subset \overline{\Gamma}\bigcup C$.
%% 
%% Secondly, we find an example where we have to add both the conormal bundle $C$ and also a convex sum $L+C$ where $L$ is a ray in the cotangent cone.
%% 
%% Finally, we find examples where the Wavefront of the sum is not contained in $\overline{\Gamma}\bigcup C\bigcup\left(\overline{\Gamma}+C\right)$:
%%$WF(t)\varsubsetneq \overline{\Gamma}\bigcup C\bigcup\left(\overline{\Gamma}+C\right)$. 
%However, we prove a converse result: if we start from a $\textbf{global}$ distribution $t\in \mathcal{D}^\prime(M)$, 
%the integral $\int_0^1\frac{d\lambda}{\lambda}t\psi(\frac{h}{\lambda})$ actually converges in $D_{WF(t)\cup C \cup (WF(t) + C)}^\prime$
%and the bound on the WF is sharp.
\paragraph{Notation and preliminary definitions.}
In this paragraph, we recall results on distribution spaces that we will use in the proof of our main theorem which controls the wave front set of the extension. Furthermore the seminorms that we define here allow to write proper estimates.
%
% We recall the meaning of the strong topology of $D_m^\prime(K)$ which is the subspace of distribution $D_m^\prime(K)$ compactly supported in $K$ of order $m$.
%
% By a deep structure theorem of Laurent Schwartz, each element $t\in D_m^\prime(K)$ can be $\emph{uniquely}$ represented as a finite sum
%\begin{equation}
%\left\langle t , \varphi\right\rangle=\sum_{\vert\alpha\vert\leqslant m} \mu_\alpha(\partial^\alpha \varphi)
%\end{equation} 
%where $\mu_\alpha$ are $\textbf{Radon measures}$ (not necessarily positive but linear continuous on $C^0(K)$) supported on $K$.
%
%
% A basis of neighborhood of the origin in $D_m^\prime(K)$ for this topology is of the form $\{t \vert \vert\left\langle t,\varphi \right\rangle\vert\leqslant \varepsilon \sup_{x\in K,\vert\alpha\vert\leqslant m} \vert \partial^\alpha\varphi(x)\vert \}$.
% 
% 
% 
For any cone $\Gamma\subset T^\bullet\mathbb{R}^d$, we let $\mathcal{D}^\prime_\Gamma$ be the set of distributions with wave front set in $\Gamma$.
We define the set of seminorms $\Vert .\Vert_{N,V,\chi}$ 
on $\mathcal{D}^\prime_\Gamma$. 
\begin{defi}
For all $\chi\in \mathcal{D}(\mathbb{R}^{d})$, for all closed cone $V\subset(\mathbb{R}^{d}\setminus \{0\})$ such that
$\left(\text{supp }\chi\times V\right)\cap \Gamma=\emptyset$,  
$\Vert t\Vert_{N,V,\chi}= \sup_{\xi\in V}\vert(1+\vert\xi\vert)^N\widehat{t\chi}(\xi)\vert$.
\end{defi}
We recall the definition
of the topology $\mathcal{D}^\prime_\Gamma$ (see \cite{Alesker} p.~14 and \cite{Geomasympt} Chapter $6$ p.~333),
\begin{defi}\label{dprimegam}
The topology of $\mathcal{D}^\prime_\Gamma $
is the weakest topology
that makes all seminorms
$\Vert .\Vert_{N,V,\chi}$
continuous
and which is stronger
than the weak topology
of $\mathcal{D}^\prime(\mathbb{R}^d)$.
%and also for any compactly supported distribution in $D^\prime(K)$ of order $m$,
%\begin{equation}
%\Vert (1+\vert k\vert+\vert\xi\vert)^{-m} \widehat{t} \Vert_{L^\infty}<\infty
%\end{equation}
Or it can be formulated as the topology which makes
all seminorms
$\Vert .\Vert_{N,V,\chi}$
and the seminorms of the weak topology:
\begin{equation}
\forall \varphi\in \mathcal{D}\left(\mathbb{R}^{d}\right) ,\vert\left\langle t,\varphi \right\rangle\vert=P_\varphi\left(t\right)
\end{equation}
continuous.
\end{defi}
%\begin{defi}
%$\lambda\mapsto c_\lambda \in L^p_{\frac{d\lambda}{\lambda}}([0,1],D_\Gamma^\prime),p\in [1,\infty)$ iff 
%\\for all $ \chi\in C_c^\infty(\mathbb{R}^{n+d})$,  for all closed cone $V\subset\mathbb{R}^{n+d *}$
%$$\left(\text{supp }\chi\times V\right)\cap \Gamma=\emptyset \implies  \int_0^1 \frac{d\lambda}{\lambda}\Vert c_\lambda\Vert^p_{N,V,\chi}<\infty  $$
%\end{defi}
We say that $B$ is bounded in $\mathcal{D}^\prime_\Gamma$ if $B$ is bounded in $\mathcal{D}^\prime$ and if for all seminorms $\Vert .\Vert_{N,V,\chi}$ defining the topology of $\mathcal{D}^\prime_\Gamma$,
$$\sup_{t\in B} \Vert t\Vert_{N,V,\chi}<\infty .$$
%We also use the seminorms:
%$$\forall \varphi\in\mathcal{D}(\mathbb{R}^d), \pi_m(\varphi):=\sup_{\vert\alpha\vert\leqslant m} \Vert \partial^\alpha\varphi\Vert_{L^\infty(\mathbb{R}^d)},$$
%$$\forall \varphi\in\mathcal{E}(\mathbb{R}^d),\forall K\subset \mathbb{R}^d, \pi_{m,K}(\varphi):==\sup_{\vert\alpha\vert\leqslant m} \Ver \partial^\alpha\varphi\Vert_{L^\infty(K)}.$$
\section{Geometry in cotangent space.}

We will denote by 
$C=\left(TI\right)^\perp\cap T^\bullet M$ 
the 
intersection 
of the conormal bundle $\left(TI\right)^\perp$
with the cotangent cone $T^\bullet M$.
\begin{figure} %on ouvre l'environnement figure
\begin{center}
\includegraphics[width=8cm]{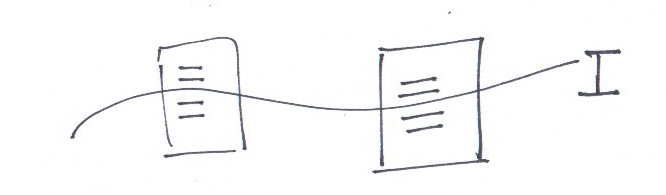} %ou image.png, .jpeg etc.
\caption{The conormal bundle to $I$.} %la légende
%l'étiquette pour faire référence à cette image
\end{center}
\end{figure} %on ferme l'environnement figure
For any subset $\Gamma$ of $T^\bullet M$
and for any subset $U$ of $M$
we denote by $\Gamma|_U$ 
the set $\Gamma\cap T^\bullet U$
where $T^\bullet U$ is 
the restriction of the cotangent
cone over $U$.
\paragraph{Associating a fiber bundle to a generalized Euler $\rho$.} 
We work with Euler vector fields $\rho$ defined on a neighborhood $\mathcal{V}$ of $I$ then $\mathcal{V}$ fibers over $I$ in such a way that the leaves of these fibrations are the set of all flow lines ending at a given of point of $I$, these leaves are invariant by the flow of $\rho$. 
\begin{defi}
Define the map $\pi^\rho:p\in \mathcal{V} \mapsto\lim_{t\rightarrow \infty} e^{-t\rho}(p)\in I $.
\end{defi}
\begin{figure} %on ouvre l'environnement figure
\begin{center}
\includegraphics[width=8cm]{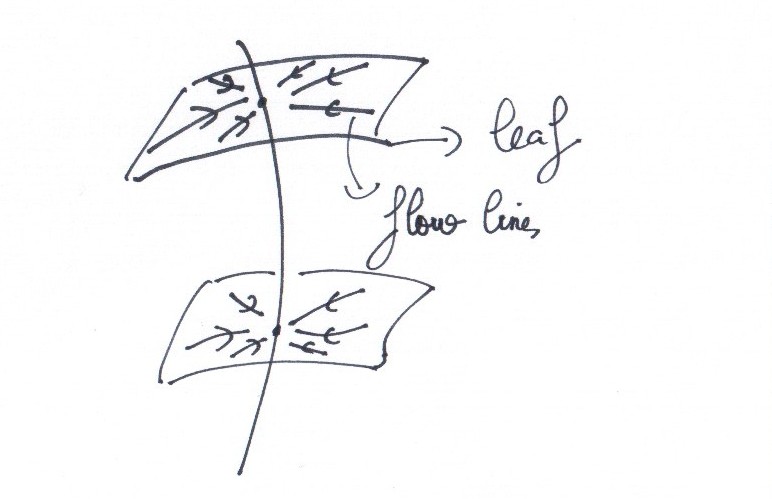} %ou image.png, .jpeg etc.
\caption{The foliation, endpoints of flow lines and leaves.} %la légende
%l'étiquette pour faire référence à cette image
\end{center}
\end{figure} %on ferme l'environnement figure
\begin{prop}
Let $\rho$ be a generalized Euler vector field defined on a neighborhood $\mathcal{V}$ of $I$, then $\mathcal{V}$ fibers over $I$, $\pi^\rho:\mathcal{V}\mapsto I$. 
\end{prop}
\begin{proof}
It is sufficient to check
that
the fibration is trivial over an open neighborhood of any $p\in I$ (\cite{Lee} Definition 6.1 p.~257 ). 
We proved that for any $p\in I$, there is a local chart $(x,h)$ of $M$ around $p$ where $I=\{h=0\}$ and the vector field $\rho$ writes $h^j\partial_{h^j}$. 
In this chart, the fibration takes the trivial form
$$(x,h)\in\mathbb{R}^{n+d}\mapsto (x)\in\mathbb{R}^n.$$
\end{proof}
\begin{defi}
We define a subset $C_\rho$ as the union of the conormals of the leaves of the fibration $\pi^\rho:\mathcal{V} \mapsto I$. $C_\rho$ is 
a coisotropic set of $T^\star M$.
\end{defi}
\begin{figure} %on ouvre l'environnement figure
\begin{center}
\includegraphics[width=6cm]{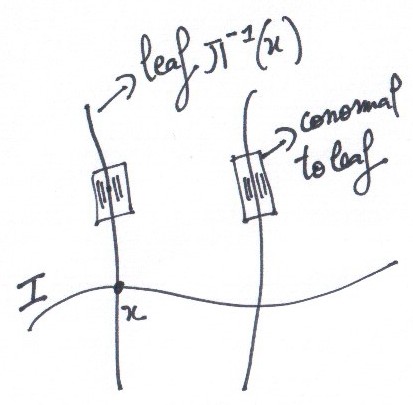} %ou image.png, .jpeg etc.
\caption{The representation of $C^\rho$ as a union of conormal bundles of the leaves of the foliation.} %la légende
%l'étiquette pour faire référence à cette image
\end{center}
\end{figure} %on ferme l'environnement figure
\paragraph{$C,C_\rho$ in local coordinates.}
In the sequel, we always work in local charts $(x,h)\in\mathbb{R}^{n+d}$ where $I=\{h=0\}$. 
We denote by $(x,h;k,\xi)$ the coordinates 
in cotangent space $T^\star\mathbb{R}^{n+d}$, 
where $k$ (resp $\xi$) is dual to $x$ (resp $h$). 
The scaling is defined by the Euler vector field 
$\rho=h^j\partial_{h^j}$. 
There is no loss of generality 
in reducing to this case 
because we proved that locally 
we can always reduce to this canonical situation (cf Chapter $1$).
In local coordinates $C=\{(x,0;0,\xi)| \xi\neq 0\}$ and $C_\rho=\{(x,h;k,0)| k\neq 0\}$. 
\begin{lemm}\label{lemmconorm}
Let $t\in\mathcal{D}^\prime(M\setminus I)$.
If $\overline{WF(t)}|_{I}\subset C$ 
then for any Euler $\rho$, 
there exists a neighborhood
$\mathcal{V}$ of $I$ 
for which $WF(t)|_{\mathcal{V}}\cap C_\rho=\emptyset$.
\end{lemm}
\begin{proof}
Since the property we want to prove is open,
it is sufficient to establish it on some open neighborhood 
of any $p\in I$. So consider a local chart $(x,h):\Omega\mapsto \mathbb{R}^{n+d}$ where $p=(0,0),I=\{h=0\}$,
$\rho=h^j\partial_{h^j}$
and 
$\Omega$ is a compact set.
By a simple contradiction argument,
if for all $\vert h\vert\leqslant\varepsilon$,
$WF(t)|_{\Omega\cap\{0<\vert h\vert\leqslant \varepsilon\}}\cap C_\rho\neq\emptyset$,
we can find a sequence 
$(x_n,h_n;\frac{k_n}{\vert k_n\vert},0)$ in $WF(t)$ such that $(x_n,h_n)\in \Omega,h_n\rightarrow 0$,
then extracting a convergent 
subsequence yields a contradiction 
with the assumption
$\overline{WF(t)}|_{I}\subset C$. 
\end{proof}
\paragraph{Lifted flows on cotangent space.}
It will be crucial in the proof of Theorem \ref{mainthm}
to control 
the wave front of the extension 
to understand the dynamics of 
the lift of the Euler flow on cotangent space.
When we scale a distribution $t$ by the one-parameter family $\Phi_\lambda=e^{\log\lambda\rho\star}$, 
we need to compute  the wave front of $\Phi_\lambda^*t$. This is described by the pull-back theorem of H\"ormander (see \cite{Hormander} Theorem 8.2.4) as the image of $WF(t)$ by the flow $T^*\Phi^{-1}_\lambda$.
\paragraph{Two interpretations of the lifted flow in cotangent space.}
We give here two points of view on this lifting. In the first one, the sections of the cotangent bundle are viewed as sections of the bundle of one forms $\Omega^1(M)$.
The second interpretation is more in the spirit of symplectic geometry and will be useful for the microlocal interpretation of the flow (see Chapter $5$). 
\begin{enumerate}
\item $\rho$ defines a flow on $M$ and, as any diffeomorphism, this flow can be lifted to the cotangent space $T^*M$.
Actually any diffeomorphism $\Phi:M\mapsto M$ lifts by the formula
\begin{equation}  
T^\star \Phi:(x,\eta) \mapsto \left(\Phi(x),  \eta\circ d\Phi^{-1}|_{\Phi(x)} \right)
\end{equation}
which in coordinates representation $(x,h)\mapsto (x,\lambda h)$ in $\mathbb{R}^{n+d}$ 
reads:
$$ (x,h;k,\xi)\in T^\star \mathbb{R}^{n+d}\mapsto (x,\lambda h;k,\lambda^{-1}\xi)\in T^\star \mathbb{R}^{n+d}.$$
\item The symbol of the differential operator $\rho$ is $\sigma(\rho)=-ih^j\xi_j$. 
We compute
its symplectic gradient  $\sigma(\rho)\in C^\infty (T^\star M)$ for the symplectic form $i(dk\wedge dx + d\xi\wedge dh)$
$$h^j\partial_{h^j}-\xi_j\partial_{\xi_j} ,$$
and we take the flow of this vector field (for more on the symbol map see \cite{Eskin} p.~198) .   
\end{enumerate}
Experts in microlocal analysis use this lifted flow 
in the ``Change-of-variables formula'' for pseudodifferential operators, 
see the formula at the bottom of p.~222 in \cite{Eskin} and Formula 61.20 p.~334 in \cite{Eskin}.
\section{Geometric and metric topological properties of $\Gamma$.}
We work in $\mathbb{R}^{n+d}$ with coordinates $(x,h)$, $I=\mathbb{R}^n\times\{0\}$ is the linear subspace $\{h=0\}$, the scaling
is given by the vector field
$\rho=h^j\frac{\partial}{\partial h^j}$
and we use the notation $f_\lambda(x,h)=f(x,\lambda h)$.
We restrict to a compact set $K$ which is $\rho$-convex.
The goal of the first part is to find conditions on $\Gamma$ so that $\forall\lambda\in(0,1], WF(t_\lambda)\subset \Gamma$.
We first use the pull-back theorem of H\"ormander to describe $WF(t_\lambda)$.
\paragraph{The pull back theorem of H\"ormander.}
Recall the definition of $\Phi^*\Gamma$ for $\Phi:X\mapsto Y$ a smooth diffeomorphism beetween two smooth manifolds $(X,Y)$ and $\Gamma \in T^\bullet Y$, $$\Phi^*\Gamma=\{ (x;\xi\circ d\Phi_x) |  (\Phi(x);\xi)\in\Gamma \}.$$
In the case $\Phi$ is a diffeomorphism, $\Phi$ is invertible and we have the simpler formula:
$$\Phi^*\Gamma=\{\Phi^{-1}(y);\xi \circ D\Phi_{\Phi^{-1}(y)} \vert \;  (y;\xi)\in \Gamma \} .$$
For $\Phi(\lambda): (x,h)\mapsto (x,\lambda h)$, we thus have
$$\Phi(\lambda)^{ *}\Gamma=\{(x,\lambda^{-1}h,k,\lambda\xi) \vert (x,h;k,\xi)\in\Gamma   \} $$
and also
$\Phi(\lambda)^{ *}\Gamma|_K= \{(x,h;k,\xi) \vert   (x,\lambda h, k, \lambda^{-1}\xi)\in \Gamma, (x,h)\in K  \}=\Phi(\lambda)^{ *}\Gamma\cap \left(K\times (\mathbb{R}^{n+d})^\star\right).$
%$$K_\lambda\times V_\lambda= \{(x,h;k,\xi) \vert   (x,\lambda h, k, \lambda^{-1}\xi)\in K\times V \}=\{(x,\lambda^{-1}h,k,\lambda\xi) \vert   (x, h, k, \xi)\in K\times V\}.$$
If $t\in \mathcal{D}^\prime_\Gamma$ then $\Phi^\star t\in  \mathcal{D}^\prime_{\Phi^\star\Gamma}$
by 
application of the pull-back theorem of H\"ormander (8.2.4 in \cite{Hormander} or \cite{Eskin} theorem 63.1) where H\"ormander uses the notation $^td\Phi_x\xi$ for $\xi\circ d\Phi_x$.
% It is immediate that for
%$\Gamma_\lambda$ and
%$K_\lambda\times V_\lambda$
%$$\forall \lambda>0 ,  \left(K_\lambda\times V_\lambda\right) \cap \Gamma_\lambda=\emptyset$$
%because if two sets $A$ and $B$ are disjoint then their image $\Phi(A),\Phi(B)$ under any diffeomorphism $\Phi$
%are disjoint. 

% As a matter of fact, we are more interested in $\Gamma_\lambda|_K$ 
%
% If $K$ is geodesically convex, then it is interesting to notice $\forall \lambda\in[0,1],K\subset K_\lambda, \Gamma_\lambda|_K\subset \Gamma_\lambda$
% $$K\times V_\lambda \cap \Gamma_\lambda =\emptyset\implies K\times V_\lambda \cap \Gamma_\lambda|_K=\emptyset  $$

%\begin{lemm}
%Let $K$ a compact set and $V\subset \mathbb{R}^{n+d*}$ be such that $\left(K\times V\right) \cap \Gamma|_K=\emptyset$, then set 
%$\Gamma_\lambda|_K= \{(x,h,k,\xi) \vert   (x,\lambda h, k, \lambda^{-1}\xi)\in \Gamma, (x,h)\in K  \}$, we claim if $K$ is geodesically convex then
%\begin{equation}
%\forall\lambda , K\times V_\lambda\cap \Gamma_\lambda|_K=\emptyset
%\end{equation}
%\end{lemm}
\paragraph{The fundamental equation.} 
We wish actually to compute $\bigcup_{\lambda\in(0,1]}WF(t_\lambda)$.
Let $U$ be any $\rho$-convex subset of $M$. 
We construct a geometric upper bound $\Gamma_M(WF(t))$ such that $\bigcup_{\lambda\in(0,1]}WF(t_\lambda)|_U\subset\Gamma_M(WF(t))$, where $\Gamma_M(WF(t))$ has a transparent geometrical meaning.
\begin{defi}\label{defigammamajorant}
Let $\rho$ be a Euler vector field and $U$ a $\rho$-convex subset of $M$.
Let $WF(t)$ be given, then the set $\Gamma_M(WF(t))|_U$ is defined as the union of all curves of the flow $\lambda\mapsto T^\star(e^{\log\lambda\rho})$ which intersect 
$WF(t)$ and the projection on the base space of which lie in $U$.
Let $T$ be the maximal time 
of existence of the flow $e^{\log\lambda\rho}$
\begin{equation}\label{geomcharac}
\Gamma_M(WF(t))|_U=\{T^\star e^{\log\lambda\rho}(p) \vert p\in WF(t),\lambda\in(0,T) \}\cap T^\bullet U. 
\end{equation}
% The idea of the construction of $\Gamma_M$ can be adapted to construct a simple geometric cone $\Gamma_{WF(t)}$ from the data of $WF(t)$.
%\begin{equation}\label{fromWFt}
%\Gamma_{WF(t)}|_U=\{\Phi_\lambda(p) \vert p\in WF(t),\lambda\in(0,1],\Phi_\lambda(p)\in T^\bullet U \} 
%\end{equation}
\end{defi}
\begin{figure} %on ouvre l'environnement figure
\begin{center}
\includegraphics[width=6cm]{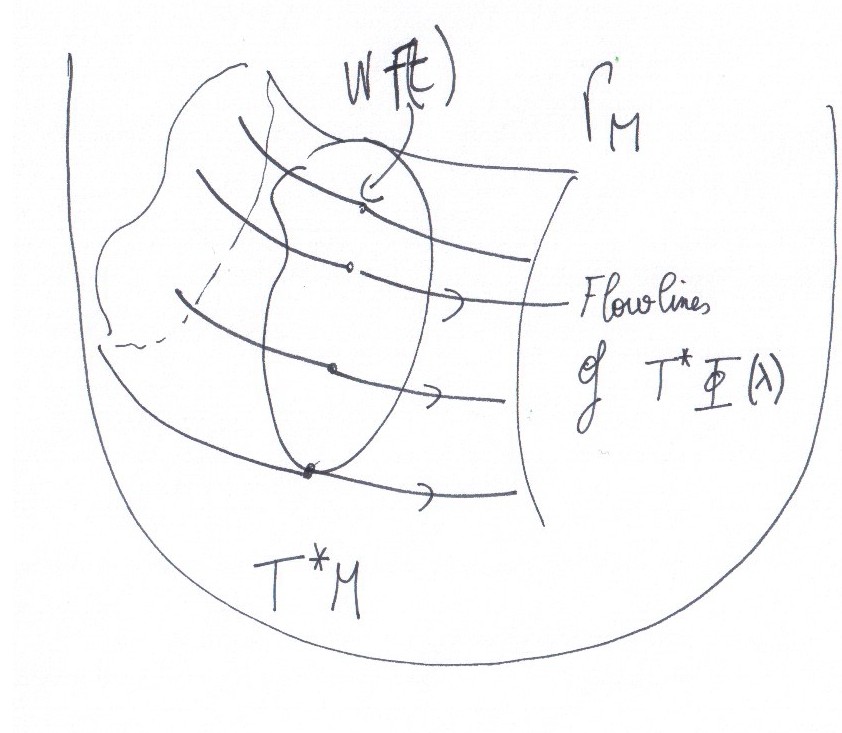} %ou image.png, .jpeg etc.
\caption{$\Gamma_M$ as the union of all flowlines intersecting $WF(t)$.} %la légende
%l'étiquette pour faire référence à cette image
\end{center}
\end{figure} %on ferme l'environnement figure
$\Gamma_M(WF(t))|_U$ is also
defined as the smallest subset of $T^\star_U M$
which contains $WF(t)\cap T^\star_U M$ and
which is stable by $T^\star e^{\log\lambda\rho}$ for
$\lambda\in(0,1]$.
It is \textbf{entirely determined} by $\rho$ and $WF(t)$.
\begin{prop}
For all $\lambda\in(0,1]$, $WF(t_\lambda)|_U\subset\Gamma_M(WF(t))|_U$.
\end{prop}
This is immediate from the definition of $\Gamma_M(WF(t))$ and the pullback theorem.
In the sequel, we use a local chart to 
identify 
a neighborhood of $p\in I$ with 
the $\left(h^j\frac{\partial}{\partial h^j}\right)$-convex
set
$U=\{ 0<\vert h\vert \leqslant \varepsilon,x\in K \}$ 
for some $\varepsilon$ and
where $K$ is a compact set of $\mathbb{R}^n$.
We want to describe geometrically the set $\Gamma_M(WF(t))$. 
The intuitive idea 
is that it is enough to know $\Gamma_M(WF(t))$ 
on a vertical slice 
$\{\vert h\vert=\varepsilon\}$ just by following 
the integral curves of the flow intersecting 
$\Gamma_M(WF(t))|_{\vert h\vert=\varepsilon}$.   
We solve a Cauchy problem for the set $\Gamma_M$, in the sense that we fix some geometric Cauchy data $\Gamma_M|_{\vert h\vert=\varepsilon}$ on the boundary $\{\vert h\vert=\varepsilon\}$ of the domain then we use the geometric characterization of $\Gamma_M|_U$ given by equation (\ref{geomcharac}). It is a geometric version of \emph{the method of characteristics} 
in PDE.
\begin{figure} %on ouvre l'environnement figure
\begin{center}
\includegraphics[width=6cm]{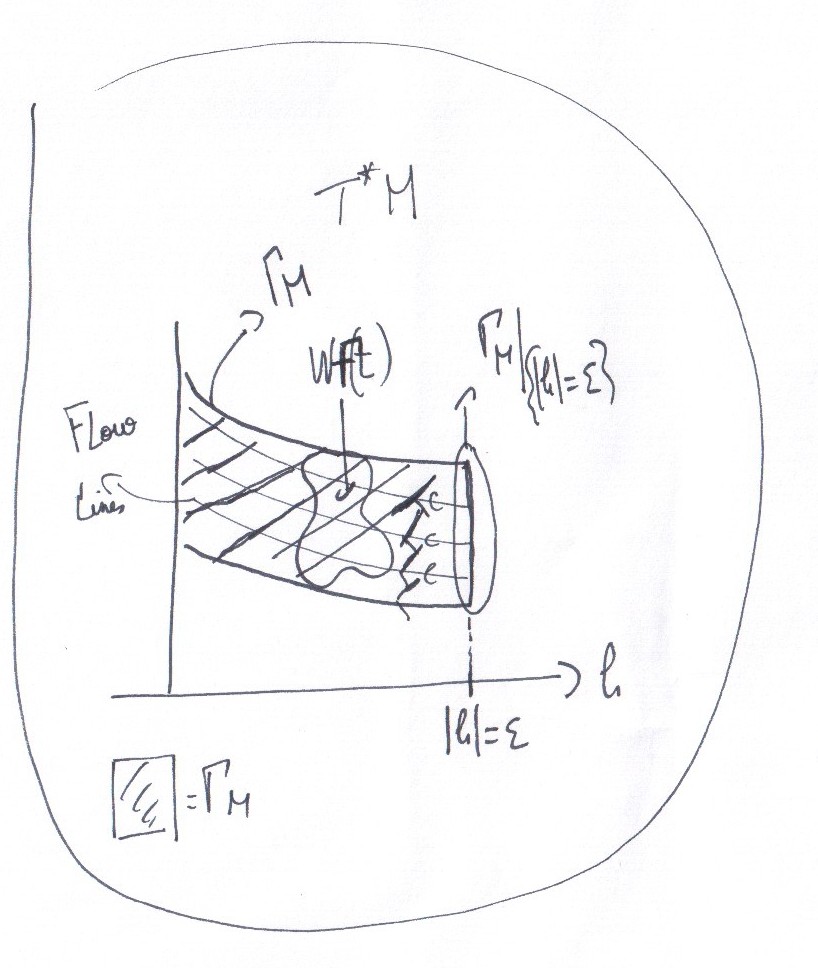} %ou image.png, .jpeg etc.
\caption{The WF(t), the foliation of $\Gamma_M$ by flowlines and the restriction over $\vert h\vert=\varepsilon$.} %la légende
%l'étiquette pour faire référence à cette image
\end{center}
\end{figure} %on ferme l'environnement figure
\begin{prop}\label{majorantlimits}
Let $U=\{(x,h)| 0<\vert h\vert \leqslant \varepsilon,x\in K \}\subset \mathbb{R}^{n+d}$ 
where $K$ is a compact subset of $\mathbb{R}^n$ 
and for some $\varepsilon>0$.
Let $\Gamma_M(WF(t))|_U$ be defined by Definition (\ref{defigammamajorant}). 
Then
$\Gamma_M(WF(t))|_{U\cap \{\vert h\vert=\varepsilon\}}$ entirely determines $\Gamma_M(WF(t))|_{U}$ by the equation:
% and also $\overline{\Gamma_M|_{0<\vert h\vert\leqslant\varepsilon}}$. 
\begin{eqnarray}
\Gamma_M(WF(t))|_{U}=\{T^\star\Phi_\lambda(p) \vert p\in\Gamma_M(WF(t))|_{U\cap \{\vert h\vert=\varepsilon\}}, 0<\lambda\leqslant 1 \}. 
%\\ \overline{\Gamma_M|_{0<\vert h\vert\leqslant\varepsilon}}=\{T^\star\Phi_\lambda(p) \vert p\in\Gamma_M|_{\vert h\vert=\varepsilon},\lambda\in[0,1] \} 
\end{eqnarray}
\end{prop}
\begin{proof}
By definition, $\Gamma_M(WF(t))|_{U}$ is fibered by curves
$\Gamma_M(WF(t))|_{U}=\{\Phi_\lambda(p)| p\in \Gamma_M(WF(t))|_{U},\lambda\in(0,1]\}\cap T^\bullet \{0<\vert h\vert\leqslant\varepsilon\}$. 
Each of these curves must intersect the boundary $\vert h\vert=\varepsilon$ in $T^\bullet U$ hence $\Gamma_M(WF(t))|_{U}$ is the set of all curves $\left(T^\star\Phi_\lambda(p)\right)_{0<\lambda \leqslant 1}$ for $p\in\Gamma_M(WF(t))|_{U\cap \{\vert h\vert=\varepsilon\}}$.
% We work in the cosphere bundle $\mathbb{S}^*U$ over $U$ and we therefore have an induced flow $S\Phi:(\lambda,p)\in [0,1]\times\mathbb{S}^*U\mapsto S\Phi_\lambda(p)\in\mathbb{S}^*U$.
% 
%  In coordinates $S\Phi:(\lambda,x,h;k,\xi)\mapsto (x,\lambda h;\frac{k}{\sqrt{\vert k\vert^2 +\vert\lambda^{-1}\xi\vert^2}},\frac{\lambda^{-1}\xi}{\sqrt{\vert k\vert^2 +\vert\lambda^{-1}\xi\vert^2}})$, $\Phi$ is a continuous map. 
% 
% Then for any convergent sequence of points $(p_n)_n\in \Gamma_M|_{0<\vert h\vert\leqslant\varepsilon}^{\mathbb{N}}$ such that $p_n\rightarrow p\in \mathbb{S}^* U|_I$, we associate the sequence $q_n$ of Cauchy data in $\Gamma_M|_{\vert h\vert=\varepsilon}^{\mathbb{N}}$ such that there exists a sequence of positive numbers $(\lambda_n)_{n\in\mathbb{N}},\lambda_n\rightarrow 0$ such that $p_n=\Phi_{\lambda_n}(q_n)$. 
%  Then we can extract a subsequence $q_n$ which converges to $q$ in $\mathbb{S}^*U$.   
%
% By continuity $p=\lim_{n\rightarrow\infty}\Phi_{\lambda_n}(q_n)=\Phi_0(q)$. So the limit point $p$ can be written as 
%$$p=\lim_{\lambda\rightarrow 0} \Phi_\lambda(q)$$ 
\end{proof}
For a given cone $WF(t)$ 
and $\Gamma_M(WF(t))$ defined 
by the equation (\ref{geomcharac}),
we believe it is natural
to demand that $\overline{\Gamma_M}|_I$
is contained in the conormal $C$
because this
ensures that $\Gamma_M(WF(t))$ 
never meets $C_\rho$ 
for arbitrary choices of 
generalized Euler 
vector fields $\rho$.
This condition 
is crucial for QFT 
because it ensures 
that counterterms are 
conormal distributions supported on $I$,
we will discuss
this in Theorem 
(\ref{Grosthmconormcontretermes}).  
We introduce a local condition
on $WF(t)$ named
local
\textbf{soft landing condition} 
at $p$
which
ensures
that for some 
neighborhood $V_p$ 
of $p$,
$\overline{\Gamma_M(WF(t))}|_{I\cap V_p}\subset C$:
\begin{defi}\label{Broudermetric}
$WF(t)$ satisfies the 
\textbf{soft landing condition} 
at $p$ 
if there exists $\rho$ and a local chart $(x,h)\in C^\infty(U,\mathbb{R}^{n+d}),I=\{h=0\}$ 
at $p\in U$ for which $\rho=h^j\frac{\partial}{\partial h^j}$ and such that 
\begin{equation}
\exists\varepsilon>0,\exists\delta>0, WF(t)|_{U\cap\{\vert h\vert\leqslant \varepsilon\}}\subset\{\vert k\vert\leqslant \delta\vert h\vert\vert\xi\vert \}.
\end{equation}    
\end{defi}
\begin{figure} %on ouvre l'environnement figure
\begin{center}
\includegraphics[width=6cm]{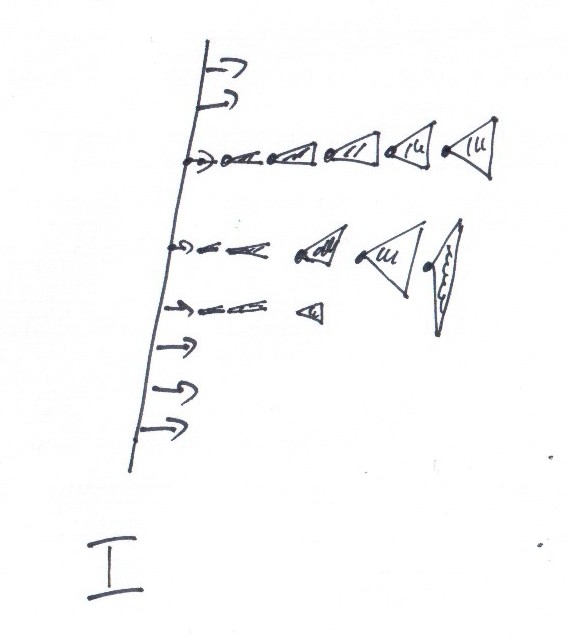} %ou image.png, .jpeg etc.
\caption{The soft landing condition forces the elements of WF to converge to the conormal of $I$.} %la légende
%l'étiquette pour faire référence à cette image
\end{center}
\end{figure} %on ferme l'environnement figure
Notice that the scale invariance
of estimate 
$\vert k\vert\leqslant \delta\vert h\vert\vert\xi\vert$
implies 
the stability
of the soft landing condition
by scaling with 
$\rho=h^j\frac{\partial}{\partial h^j}$.
The above definition
depends
on the choice of $\rho$,
however 
since by \ref{conjugcoro},
two Eulers $\rho_1,\rho_2$
are always locally
conjugated
by an element $\Psi$
of the pseudogroup 
$G$,
$\Psi$
transforms
the Euler 
by pushforward,
$\Psi_\star \rho_1=\rho_2$,
and the local chart
by pullback.
To prove that
the local soft 
landing condition
does not depend 
on the choice of
Euler vector field,
it suffices
to prove
$\Gamma$ satisfies
the local 
soft landing condition
at $p$ implies
$\Psi(\Gamma)$
satisfies the soft landing
condition
at $\Psi(p)$
for all $\Psi\in G$.

\subsubsection{The soft landing condition is stable by action of $G$.}
We prove
in Propositions \ref{HeleinmicroscopicI} 
that
the soft landing condition
is locally
stable by the action
of the pseudogroup
$G$ of local
diffeomorphisms
fixing $I$.
\paragraph{The geometric reformulation of the soft landing
condition.}
We are led to reformulate 
the local soft landing condition 
in a more geometric flavor which, 
once established,
makes the claim
of stability rather
trivial.
We denote by 
$\mathbb{U}^\star \mathbb{R}^{n+d}$
the unit cosphere bundle.
Let $\pi_1:(x,h;k,\xi )\in\mathbb{U}^\star \mathbb{R}^{n+d}\mapsto (x,h)\in\mathbb{R}^{n+d}$ and $\pi_2:(x,h;k,\xi )\in\mathbb{U}^\star \mathbb{R}^{n+d}\mapsto (k,\xi)\in\mathbb{U}^{n+d-1}$.
We introduce the following distance on the cosphere bundle $d_{\mathbb{U}^\star \mathbb{R}^{n+d}}\left(p,q\right)=d_{\mathbb{R}^{n+d}}(\pi_1(p),\pi_1(q))+d_{\mathbb{U}^{n+d-1}}(\pi_2(p),\pi_2(q))$.
Let us consider $\mathbb{U}\Gamma$ 
the trace of $\Gamma$ on 
$\mathbb{U}^\star \mathbb{R}^{n+d}$ 
and also $\mathbb{U}C$ 
the trace of the 
conormal bundle of $I$ 
in $\mathbb{U}^\star \mathbb{R}^{n+d}$.
\begin{defi}\label{intrinsicslc}
The set $\Gamma$ satisfies the local soft landing condition 
on $U$ 
if and only if for any element $p\in\mathbb{U}\Gamma$ 
such that $\pi_1(p)\in U$, 
the distance of $p$ with 
the conormal trace $\mathbb{U}C$ 
is controlled by 
the distance beetween $\pi_1(p)$ and $I$:
$$\forall K\subset U,\exists \delta,\forall p\in \Gamma_S,\pi_1(p)\in K, d_{\mathbb{S}^\star \mathbb{R}^{n+d}}(p,C_S)\leqslant \delta d_{\mathbb{R}^{n+d}}(\pi_1(p),I) .$$
\end{defi}
We will quickly explain the equivalence of this definition 
with the definition (\ref{Broudermetric}),
$$\vert k\vert\leqslant \delta\vert h\vert\vert\xi\vert\Leftrightarrow  
\frac{\vert k\vert}{\vert\xi\vert}
\leqslant \delta\vert h\vert
$$
$$\Leftrightarrow \vert\tan \theta((k,\xi);(0,\xi))\vert\leqslant \delta\vert h\vert \implies  \vert \theta((k,\xi);(0,\xi))\vert\leqslant \delta^\prime\vert h\vert  $$
$$\implies d_{\mathbb{S}^{n+d-1}}(\pi_2(p),\pi_2(C))\leqslant \delta^\prime d_{\mathbb{R}^{n+d}}(\pi_1(p),I) $$
$$\implies d_{\mathbb{S}^\star \mathbb{R}^{n+d}}\left(p,C_S\right)=d_{\mathbb{S}^{n+d-1}}(\pi_2(p),\pi_2(C))+d_{\mathbb{R}^{n+d}}(\pi_1(p),I) $$ 
$$\leqslant (1+\delta^\prime) d_{\mathbb{R}^{n+d}}(\pi_1(p),I) .$$
Conversely,
$$ d_{\mathbb{S}^\star \mathbb{R}^{n+d}}(p,C_S)\leqslant \delta d_{\mathbb{R}^{n+d}}(\pi_1(p),I) $$
$$\implies d_{\mathbb{S}^{n+d-1}}(\pi_2(p),\pi_2(C)) \leqslant \delta d_{\mathbb{R}^{n+d}}(\pi_1(p),I)$$
$$\implies\vert \theta((k,\xi);(0,\xi))\vert\leqslant \delta\vert h\vert $$
$$\implies \vert \tan\theta((k,\xi);(0,\xi))\vert\leqslant \delta^\prime\vert h\vert 
\implies \frac{\vert k\vert}{\vert\xi\vert}\leqslant \delta^\prime\vert h\vert .$$
\paragraph{The invariance by $G$.}
The geometrical reformulation
in terms of distance
combined with Proposition
\ref{Idlift}
makes obvious the following claim:
\begin{prop}\label{HeleinmicroscopicI} 
Let $\Psi: U \mapsto U$ be a local diffeomorphism in $G$, $\sigma=T^\star \Psi$ be 
the corresponding lift on $T^\star U$ and $\Gamma$ be a closed conic set in $T^\bullet M$.
Then if $\Gamma$ satisfies the local \textbf{soft landing condition} at $\pi_1\circ\sigma(p)\in U$, then $\sigma\circ \Gamma$ satisfies the local \textbf{soft landing condition} at $\pi_1\circ\sigma(p)$.
\end{prop} 
By \ref{conjugcoro}, this implies:
\begin{prop}
If $WF(t)$ satisfies the soft landing condition locally at $p$ for some $\rho$ and some associated chart, then
for any local chart $(x,h)\in C^\infty(U,\mathbb{R}^{n+d}),I=\{h=0\}$ and associated Euler
$\rho=h^j\frac{\partial}{\partial h^j}$,
$WF(t)$ satisfies the soft landing condition locally at $p$.
\end{prop}
\begin{defi}
$WF(t)$ satisfies the soft landing condition if for all $p\in I$, it satisfies the soft landing condition locally at $p$.
\end{defi}

\paragraph{Consequences of the soft landing condition.}
\begin{lemm}
Let $t\in\mathcal{D}^\prime(M\setminus I)$.
If $WF(t)$ satisfies the soft landing condition, then 
$\overline{WF(t)}|_I\subset C$. In particular, this implies for all Euler $\rho$,
there exists a neighborhood $\mathcal{V}$ of $I$ such that $WF(t)\cap C^\rho=\emptyset$.
\end{lemm}
\begin{proof} 
By definition of the soft landing condition, 
it suffices to work locally at each 
$p\in I$. For each $p$, there exists
some open set $U$ s.t.
$\exists\delta>0,WF(t)|_{U\cap\{\vert h\vert\leqslant \varepsilon\}}
\subset\{\vert k\vert\leqslant \delta\vert h\vert\vert\xi\vert \}$ 
which
implies 
$\overline{WF(t)}|_{U\cap\{h=0\}}\subset \{k=0\}\implies \overline{WF(t)}|_{I\cap U}\subset C$. 
Actually $\overline{WF(t)}|_{I\cap U}\subset C\implies WF(t)|_{U\cap\{\vert h\vert\leqslant \varepsilon\}}\cap C_\rho=\emptyset$
for $\varepsilon$ small enough
by Lemma \ref{lemmconorm}. 
For each $p$, we were able 
to find an open set $U_p$ 
and $\varepsilon_p>0$
such that 
$WF(t)|_{U_p\cap\{\vert h\vert\leqslant \varepsilon_p\}}\cap C_\rho=\emptyset$
then $\cup_{p\in I}U_p\cap\{\vert h\vert\leqslant \varepsilon_p\}$
forms an open cover of $I$ 
and extracting a subcover $\mathcal{V}=\cup_{n\in\mathbb{N}}U_{p_n}\cap\{\vert h\vert\leqslant \varepsilon_{p_n}\}$ 
gives a neighborhood $\mathcal{V}$ of $I$ such that $WF(t)\cap C^\rho=\emptyset$.
\end{proof}
\begin{thm}\label{gammaconorm}
Let $t\in\mathcal{D}^\prime(M\setminus I)$.
$WF(t)$ satisfies the soft landing condition
if and only if 
\begin{eqnarray} 
\overline{\Gamma_M(WF(t))}|_I\subset C=(TI)^\perp,
\end{eqnarray}
where $\Gamma_M(WF(t))$ is defined by Equation (\ref{geomcharac}).
\end{thm} 
\begin{proof} 
It suffices to work locally at each 
$p\in I$.
The sense $\Rightarrow$ is simple. The set $\{\vert k\vert\leqslant \delta\vert h\vert\vert\xi\vert \}$ is clearly invariant by the flow $(x,h;k,\xi)\rightarrow (x,\lambda h;k,\lambda^{-1}\xi)$. If $p\in WF(t)$ then by hypothesis $p\in \{\vert k\vert\leqslant \delta\vert h\vert\vert\xi\vert \}$, 
hence the whole curve $\lambda\mapsto \Phi_\lambda(p)$ lies in $\{\vert k\vert\leqslant \delta\vert h\vert\vert\xi\vert \}$ thus by definition $\Gamma_M=\{\Phi_\lambda(p)|p\in WF(t),\lambda\in(0,\infty),\Phi_\lambda(p)\in T^\bullet\left(0<\vert h\vert\leqslant\varepsilon\right)\}\subset \{\vert k\vert\leqslant \delta\vert h\vert\vert\xi\vert \}$.
%$$\overline{\Gamma_M|_{0<\vert h\vert\leqslant\varepsilon}}\cap C_\rho=\emptyset\implies 
%WF(t)|_{0<\vert h\vert\leqslant\varepsilon}}\cap C_\rho=\emptyset$$ because of the inclusion $WF(t)|_{0<\vert h\vert\leqslant\varepsilon}}\subset \overline{\Gamma_M|_{0<\vert h\vert\leqslant\varepsilon}}\implies \overline{WF(t)|_{0<\vert h\vert\leqslant\varepsilon}}}\subset \overline{\Gamma_M|_{0<\vert h\vert\leqslant\varepsilon}}$.
Since $\{\vert k\vert\leqslant \delta\vert h\vert\vert\xi\vert \}$ is closed then 
$\overline{\Gamma_M}\subset \{\vert k\vert\leqslant \delta\vert h\vert\vert\xi\vert \}$ and on $I=\{h=0\}$
we must have $k=0$ thus $\overline{\Gamma_M}|_I\subset C$.
Hence $\Gamma_{M}|_I\subset C$.
To establish the converse sense $\Leftarrow$, we use the proposition (\ref{majorantlimits}). 
If $\overline{\Gamma_M(WF(t))}|_I\subset C$ then
by Lemma \ref{lemmconorm}, 
$\Gamma_M(WF(t))|_{\vert h\vert=\varepsilon}\cap \{(x,h;k,0)|k\neq 0\}=\emptyset$ for $\varepsilon$ small
enough. 
This implies that $\exists \delta>0$
s.t. $\Gamma|_{\vert h\vert=\varepsilon} 
\subset \{\vert k\vert\leqslant \delta\varepsilon\vert\xi\vert\}$.
Indeed let us proceed by contradiction.
Assume the contrary,
then
for any
$n\in \mathbb{N}^*$, there
exist $(x_n,h_n;k_n\xi_n)\in \Gamma|_{\vert h\vert=\varepsilon}$ s.t.
$k_n>n\vert\xi_n\vert $ and w.l.g.
$\vert k_n\vert=1$.
By compactness, we can extract a
subsequence which converges
to $(x,h;k,0)$.
This hypothesis translates 
in an estimate $\Gamma_M|_{\vert h\vert=\varepsilon}\subset\{\vert k\vert\leqslant \delta\varepsilon\vert\xi\vert  \}$
for a certain
$\delta>0$. 
Now the idea is to scale 
this estimate in order to have a general estimate for all $h$.
$$ p\in \Gamma_M|_{\vert h\vert=\varepsilon}\implies p=(x,h;k,\xi)\in \{\vert k\vert\leqslant \delta\varepsilon\vert\xi\vert  \}\subset\{\vert k\vert\leqslant \delta\vert h\vert\vert\xi\vert  \}$$ 
by the estimate $\Gamma_M|_{\vert h\vert=\varepsilon}\subset\{\vert k\vert\leqslant \delta\varepsilon\vert\xi\vert  \}$ and because $\vert h\vert=\varepsilon$,
$$\implies \forall \lambda\in(0,1], \Phi_\lambda(p)=(x,\lambda h;k,\lambda^{-1}\xi)\in  \{\vert k\vert\leqslant \delta\lambda\vert h\vert\lambda^{-1}\vert\xi\vert  \}=\{\vert k\vert\leqslant \delta\vert h\vert\vert\xi\vert \}$$ 
Hence by proposition (\ref{majorantlimits}) we find
\begin{eqnarray}
\Gamma_M|_{0<\vert h\vert\leqslant\varepsilon}=\{\Phi_\lambda(p) \vert p\in\Gamma_M|_{\vert h\vert=\varepsilon}, 0<\lambda\leqslant 1 \} \subset \{\vert k\vert\leqslant \delta\vert h\vert\vert\xi\vert \}
\\ \overline{\Gamma_M|_{0<\vert h\vert\leqslant\varepsilon}}=\{\Phi_\lambda(p) \vert p\in\Gamma_M|_{\vert h\vert=\varepsilon},\lambda\in[0,1] \}\subset \{\vert k\vert\leqslant \delta\vert h\vert\vert\xi\vert \} 
\end{eqnarray}
and we proved the claim because $WF(t)|_{0<\vert h\vert\leqslant\varepsilon}\subset \overline{\Gamma_M|_{0<\vert h\vert\leqslant\varepsilon}}$.   
%   $p=(x,h;k,\xi), 0<\vert h\vert\leqslant\varepsilon, \xi\neq 0 $ because $p\notin C_\rho$, $\Phi_\lambda(p)=(x,\lambda h,k,\lambda^{-1}\xi)\rightarrow C$ and never meets $C_\rho$.
%  
%  So the curve $\{\Phi_\lambda(p) \vert \lambda\in[0,1]\}$ never meets $C_\rho$ and ends in $C$.
%
%
% Any limit point $p\in \overline{\Gamma_M|_{0<\vert h\vert\leqslant\varepsilon}}$ over $I=\{h=0\}$ can be obtained as a limit  
%$$p=\lim_{\lambda\rightarrow 0} \Phi_\lambda(q),q \in \overline{WF(t)|_{0<\vert h\vert\leqslant\varepsilon}}}$$ 
%by the same argument as in proposition (\ref{majorantlimits}). But $q$ does not meet $C_\rho$ by assumption hence $p=\lim_{\lambda\rightarrow 0} \Phi_\lambda(q)\in C$
%because $C$ is the stable manifold.
%  
\end{proof} 
%  Let $U$ be any open set of $\mathbb{R}^{n+d}$ which is stable by the flow $\forall\lambda\in (0,1], \Phi_\lambda\left(U\right)\subset U$.
%
% We call $\pi_1:T^\star\mathbb{R}^{n+d}\mapsto [0]\simeq\mathbb{R}^{n+d}$ the projector on the zero section $[0]\subset T^\star\mathbb{R}^{n+d}$ of the bundle $T^\star\mathbb{R}^{n+d}$ viewed as a submanifold and identified with the base space $\mathbb{R}^{n+d}$. 
\subsubsection{A counterexample which shows the optimality of the soft landing condition.}
We give a counterexample 
which proves
$\overline{WF(t)}|_I\subset C $ does not imply
$\overline{\Gamma_M(WF(t))}|_I\subset C$
ie 
the soft landing condition 
(\ref{Broudermetric}) 
is in fact optimal.
We work in $\mathbb{R}^2$ with coordinates $(x,h)$.
The Euler vector field writes $\rho=h\partial_{h}$. 
If $ WF(t)=\{(x,h;\lambda 1,\lambda h^{-\frac{1}{2}})|\lambda\in\mathbb{R}_+ \}$ then it is immediate that $\overline{WF(t)}|_I\subset C=\{(x,0;0,\xi)\}$. 
However $WF(t)$ does not satisfy the 
soft landing condition since  
we find that the sequence of points $(x,\frac{1}{n^2};1,n)$ belongs to $WF(t)$.
By definition of $\Gamma=\bigcup_{\lambda\in(0,1]}WF(t_\lambda)$, we find that
$$\Gamma=\{(x,\lambda^{-1}h,k,\lambda\xi)| (x,h,k,\xi)\in WF(t),\lambda\in(0,1]\} $$
thus setting $\lambda_n=\frac{1}{n}$, we find that the sequence $(x,n\frac{1}{n^2};1,\frac{n}{n})=(x,\frac{1}{n};1,1)$ belongs to $\Gamma$ thus $\lim_{n\rightarrow \infty}(x,\frac{1}{n};1,1)=(x,0;1,1)\in\overline{\Gamma}|_I$ which does not live in the conormal. 
\section{The counterterms are conormal distributions.}
We fix the coordinate system $(x^i,h^j)$ in $\mathbb{R}^{n+d}$ and $I=\{h=0\}$.
We first recall a deep theorem of Schwartz (see \cite{Schwartz} Theorems 36 p.~101) 
about the structure of distributions supported on  $I\subset \mathbb{R}^{n+d}$.
We denote by $\delta_I$ the unique distribution such that $\forall\varphi\in \mathcal{D}(\mathbb{R}^{n+d})$,$$\left\langle \delta_I,\varphi\right\rangle=\int_{\mathbb{R}^n} \varphi(x,0)d^nx .$$
The collection of coordinate functions $(h^j)_{1\leqslant j\leqslant d}$ defines a canonical collection of transverse vector fields $(\partial_{h^j})_j$.
If $t\in \mathcal{D}^\prime(\mathbb{R}^{n+d})$ with $\text{supp } t \subset I$, then there exist
a unique family of distributions (once the system of transverse vector fields $\partial_{h^j}$ is fixed) $t_\alpha\in \mathcal{D}^\prime\left(\mathbb{R}^n\right)$, with $\{\text{supp }t_\alpha\}$ locally finite, such that $t(x,h)=\sum_\alpha t_\alpha(x) \partial_h^\alpha\delta_I(h)$
(see \cite{Schwartz} Theorem 36 p.~101-102 or \cite{Hormander} theorem 2.3.5)) where the $\partial_h^\alpha$ are derivatives in the \textbf{transverse} directions. 
\paragraph{What happens in the case of manifolds ?}
From the point of view of 
L. Schwartz,
the only thing to keep in mind
is that a distribution
supported on a submanifold $I$ 
is always
well defined locally
and the representation 
of this distribution is 
unique once we fix
a system
of coordinate functions
$(h^j)_j$
which are transverse to $I$
(\cite{Schwartz} Theorem 37 p.~102).
For any distribution $t_\alpha\in\mathcal{D}^\prime(I)$, if we denote by 
$i:I\hookrightarrow M$ the canonical embedding of $I$ in $M$ then 
$i_\star t_\alpha$ is the push-forward of $t_\alpha$ in $M$:
$$\forall\varphi\in\mathcal{D}(M), \left\langle i_\star t_\alpha, \varphi \right\rangle=\left\langle t_\alpha, \varphi\circ i \right\rangle.$$ 
The next lemma completes Theorem \ref{removsing} proved in Chapter $1$. 
Here the idea is that we add a constraint on the \textbf{local counterterm }$t$, namely that $WF(t)$ is contained in the conormal of $I$. Then we prove that the coefficients $t_\alpha$ appearing in the Schwartz representation are in fact \textbf{smooth} functions.
\begin{lemm}\label{lemmbound}
Let $t\in\mathcal{D}^\prime(M)$ such that
$t$ is supported on $I$, then\\
1) $t$ has a unique 
decomposition as locally finite
linear combinations of transversal derivatives
of push-forward to $M$ of distributions $t_\alpha$ in $\mathcal{D}^\prime (I)$:
$t=\sum_{\alpha} \partial^\alpha_h\left(i_\star t_\alpha\right)$,\\
and 2) $WF(t)$ is contained in the conormal of $I$ if and only if
$\forall \alpha$, $t_\alpha$ is smooth.
\end{lemm}
\begin{proof}
In local coordinates,
let $$t(x,h)=\sum_\alpha \partial_h^\alpha\left(t_\alpha(x) \delta_I(h)\right)=\sum_\alpha t_\alpha(x) \partial_h^\alpha\delta_I(h).$$
Assume $t_\alpha$ is not smooth 
then $WF(t_\alpha)$ would be \textbf{non empty}. 
Then $WF(t_\alpha)$ 
contains an element $(x_0;k_0)$. 
Pick $\chi\in \mathcal{D}(R^n)$ such that $\chi(x_0)\neq 0$ 
then
$$\mathcal{F}(t_\alpha\chi \partial_h^\alpha\delta_I)(k,\xi)=\widehat{t_\alpha\chi}(k)(-i\xi)^\alpha, $$
hence we find a codirection $(\lambda k_0,\lambda \xi),k_0\neq 0$ in which the product  $\widehat{t_\alpha\chi}\widehat{\partial_h^\alpha\delta_I}$ is not rapidly decreasing, hence there is a point $(x,0)$ such that $(x,0;k_0,\xi_0)\in WF(t)$ (by lemma $8.2.1$ in \cite{Hormander}) which is in contradiction with the fact that $WF(t)\subset C=\{(x,0,0,\xi)|\xi\neq 0\}$.
The reader can use Theorem 8.1.5 in \cite{Hormander} for the converse.
\end{proof}
Combining with Theorem \ref{removsing}, 
we obtain:
\begin{coro}
Let $t\in\mathcal{D}^\prime(\mathbb{R}^{n+d})$ and $\text{supp }t\subset I$.
If $WF(t)\subset C$ and $t\in E_s(\mathbb{R}^{n+d}), -m-1<s+d\leqslant -m$, then $t(x,h)=\sum_\alpha t_\alpha(x) \partial_h^\alpha\delta_I(h)$, where $\forall\alpha$, $t_\alpha\in C^\infty\left(\mathbb{R}^n\right)$ and $\vert\alpha\vert\leqslant m$.
\end{coro}
\begin{coro}
Let $M$ be a smooth manifold
and $I$ a closed embedded submanifold.
For $-m-1< s+d\leqslant-m $, the space of distributions $t\in E_s(M)$ 
such that $\text{supp }t\in I$ and $WF(t)$ is contained in the conormal of $I$ is a finitely generated module of \textbf{rank} $\frac{m+d!}{m!d!}$ over the ring $C^\infty(I)$. 
\end{coro}
\begin{proof}
In each local chart $(x,h)$ where $I=\{h=0\}$,
$t=\sum_\alpha t_\alpha(x) \partial_h^\alpha\delta_I(h)$ where the lenght $\vert\alpha\vert$ is bounded by $m$ by the above corollary and $\forall\alpha$, $t_\alpha\in C^\infty\left(I\right)$. This improves on the result given 
by the structure theorem of Laurent Schwartz 
since we now know that the $t_\alpha$ are smooth.
\end{proof}
Recall $\pi$ is the fibration which
in local coordinates where $\rho=h^j\frac{\partial}{\partial h^j}$ 
writes $\pi:(x,h)\mapsto x$ and $i$ is the embedding
of $I$ in $M$.
Recall the formula \ref{equationscounterterms} for the counterterms  
which are used to renormalize the H\"ormander extension formula:
\begin{equation}\label{counterterms}
\left\langle\tau_\lambda,\varphi\right\rangle= \left\langle t \psi(\frac{h}{\lambda}),\sum_{\vert\alpha\vert\leqslant m}\frac{h^\alpha}{\alpha !}\pi^{\rho\star}i^\star\left(\partial^\alpha_h\varphi\right) \right\rangle.
\end{equation}
We give here a general definition of local counterterms of $t$ 
that covers the counterterms of Chapter $1$, the anomaly counterterms
of Chapter $6$ and 
the poles
of the meromorphic regularization
of Chapter $7$:
\begin{defi}
Let us fix a system $(h^j)_{1\leqslant j\leqslant d}$ 
of coordinate functions transverse to $I$.
The vector space of local counterterms
of $t\in\mathcal{D}^\prime(M\setminus I)$ is defined as the vector space 
generated by all distribution $\tau$
supported on $I$ 
which can be represented by the formula:
\begin{equation}
\forall \varphi\in\mathcal{D}(M),\left\langle\tau,\varphi\right\rangle=\left\langle t \psi,\pi^{\rho\star}i^\star\left(\partial^\alpha_h\varphi\right) \right\rangle,
\end{equation}
where $\psi$ vanishes in a neighborhood of $I$ and $\pi:\text{supp }\psi \mapsto I$
is a proper mapping.
\end{defi}
The next theorem
we will
prove is very simple
yet extremely
important
conceptually 
for QFT in curved
space times.
In classical
QFT textbooks,
one should subtract
polynomials
of momenta
to renormalize divergent
integrals.
By inverse Fourier
transform
these counterterms 
become sums
of derivatives
of delta functions
supported
on vector subspaces
of configuration
space.
In curved space
times, there is no concept
of polynomials of momenta
but the notion
of conormal
distribution
supported on a submanifold
still makes sense
and replaces the concept 
of
polynomials
of momenta.
We start by a simple lemma:
\begin{lemm}\label{lemm3}
Let $t\in\mathcal{D}^\prime(M\setminus I)$ and $\tau$ be a distribution defined by the formula
\begin{equation}
\forall \varphi\in\mathcal{D}(M),\left\langle\tau,\varphi\right\rangle=\left\langle t \psi,\left(\partial^\alpha_h\varphi\right)\circ i\circ \pi \right\rangle,
\end{equation}
where $\psi$ 
vanishes in a 
neighborhood of $I$ and 
$\pi:\text{supp }\psi \mapsto I$
is a proper mapping.
If $WF(t\psi)\cap C_\rho=\emptyset$ 
then $WF(\tau)$ is contained in the conormal $C$.
\end{lemm}
\begin{proof}
We can prove our claim in local charts and reduce to the flat case $\mathbb{R}^{n+d}$.
$\tau$ can be reformulated as a product of the pushforward of $t\psi$ by the \textbf{fibration} $\pi:(x,h)\in\mathbb{R}^{n+d}\mapsto x\in\mathbb{R}^n$ with a derivative of delta distribution.
The idea of the proof is to use the Fubini theorem where integration is performed in a specific order. To clearly understand the strategy, let us write $\left\langle t \psi,\partial^\alpha\varphi(x,0) \right\rangle$ in integral form
$$\int_{\mathbb{R}^{n+d}} d^nx d^dh t(x,h)\psi(x,h) \partial^\alpha\varphi(x,0) $$ 
$$=\int_{\mathbb{R}^{n}} d^nx \left(\int_{\mathbb{R}^d}d^dh t(x,h)\psi(x,h)\right) \partial^\alpha\varphi(x,0)$$
$$=\int_{\mathbb{R}^{n}} d^nx \underset{\text{integrated along fibers}}{\underbrace{\left(\int_{\pi^{-1}(x)}d^dh t(x,h)\psi(x,h)\right)}} \partial^\alpha\varphi(x,0).$$
This formula suggests the coefficient $t_\alpha(x)$ 
in the Schwartz representation formula
is just equal to the integral
$\left(\int_{\pi^{-1}(x)}d^dh t(x,h)\psi(x,h)\right)$.
%Recall in coordinates $C_\rho=\{(x,h;k,0)| k\neq 0\}$ and 
%$C_\rho$ is interpreted as the 
%coisotropic set 
%associated with the locally trivial fibration $\pi:(x,h)\mapsto x$ 
%which is the union of conormals of the leaves $\pi^{-1}(x)$. 
%$t\psi$ is compactly supported in the fibers of the fibration $\pi$. 
Then the distribution
$x\mapsto t_\alpha(x)=\int_{\pi^{-1}(x)}d^dh t(x,h)\psi(h)$ 
is the $\textbf{pushforward}$ $\pi_*\left(t\psi\right)$ 
where we \textbf{integrated $t\psi$ along the fibers of the fibration} $\pi$.
The wave front set of $\pi_*\left(t\psi\right)$
can be computed by proposition $(1.3.4)$ page $20$ of \cite{DuistermaatFIO}.
$WF(\pi_*\left(t\psi\right))=\{(x;k)|\exists h, (x,h;k,0)\in WF(t\psi)\}$, since $WF(t\psi)\cap C_\rho=\emptyset$ then  
$WF(\pi_*\left(t\psi\right))$ is empty 
hence $\pi_*\left(t\psi\right)\in C^\infty(I)$.
Finally, if we set $t_\alpha=\pi_*\left(t\psi\right)$ then the counterterm $\tau$
writes $\tau(x,h)=t_\alpha(x)\partial_h^\alpha\delta_I(h)$ where $t_\alpha\in C^\infty(I)$ and is a conormal distribution in the terminology
of H\"ormander (see \cite{Hormander} 8.1.5).
\end{proof}
Combining Lemmas \ref{lemm3}, \ref{lemmbound}, \ref{lemmconorm} 
and fixing a 
system
of coordinates 
functions $(h^j)_j$
transversal to $I$
yields the
theorem:
\begin{thm}\label{Grosthmconormcontretermes}
Let $t\in\mathcal{D}^\prime(M\setminus I)$. 
If $\overline{WF(t)}|_I\subset C$, then
there exists a neighborhood $\mathcal{V}$ of $I$ 
such that for all $\tau$ 
defined by the formula
\begin{equation}
\forall \varphi\in\mathcal{D}(M),\left\langle\tau,\varphi\right\rangle=\left\langle t \psi,\pi^{\rho\star}i^\star\left(\partial^\alpha_h\varphi\right) \right\rangle,
\end{equation}
where $\psi$ 
vanishes in a 
neighborhood of $I$ and 
$\pi:\text{supp }\psi \mapsto I$
is a proper mapping and 
$\text{supp }\psi\subset\mathcal{V}$,
$WF(\tau)\subset C$. 
In particular,
$\tau$ is represented in a unique way by
$\tau=\sum_\alpha \partial_h^\alpha\left(i_\star\tau_\alpha\right)$
where $\forall\alpha,\tau_\alpha\in C^\infty(I)$.
\end{thm}
%
%If a counterterm is a conormal distribution then $WF(t)$ does not meet some set $C_\rho$ which is a subcone of the cotangent space. We combine this constraint with
%the solution of the geometric equation \ref{ } and we deduce the following result
%\begin{thm}
%If $\forall \lambda\in[0,1], t_\lambda\in D^\prime_\Gamma$ then the following two claims are equivalent
%\\ there exists $V_1,V_2$ such that $I\subset V_1\subset V_2$ where the sequence of inclusions is in the sense of $\textbf{neighborhoods}$, $$\Gamma|_{V_2\setminus V_1}\cap C_\rho=\emptyset $$
%$$\Leftrightarrow $$
%$$\overline{\Gamma}|_I\subset C$$
%\end{thm}
\section{Counterexample.} 
We work in $T^\star\mathbb{R}^{n+d}$ 
with coordinates $(x,h;k,\xi)$ 
and $I=\{h=0\}$. 
In this section, we prove that for any $p\in  T^\bullet \mathbb{R}^{n+d}|_{I}$, we
can construct $t\in C^\infty(\mathbb{R}^{n+d}\setminus\{h=0\})\cap L^\infty(\mathbb{R}^{n+d})$ in such a way that $p\in WF(t)$. $t$ is a bounded function hence defines a unique element $t\in \mathcal{D}^\prime(\mathbb{R}^{n+d})$.
\begin{lemm}
For all $p=(x_0,0;k,\xi)\in T^\bullet \mathbb{R}^{n+d}|_{I}$,
there exists $t\in C^\infty(\mathbb{R}^{n+d}\setminus\{h=0\})\cap L^\infty(\mathbb{R}^{n+d})$ such that
$p\in WF(t)$. In particular, when $p=(0,0;\epsilon,0)$ then
we can choose
$$t(x,h)=\int_{\mathbb{R}^{n+d}} d\xi dk e^{i(x.k+h.\xi)}a(k,\xi)\left(1+\vert k\vert+\vert\xi\vert\right)^{-n-d-1},$$
where $a(k,\xi)=e^{-\frac{\vert k \vert^2+\vert\xi\vert^2-(k.\epsilon)^2}{(k.\epsilon)}}(1-\alpha(k,\xi))$ when $k.\epsilon>0$ and $0$ otherwise, where $\alpha=1$
in a neighborhood of $0$. 
\end{lemm}
The contruction of $t$ was inspired by \cite{Hormanderwave} Example 8.2.4 p.~188 and
the lecture notes of Louis Boutet de Monvel \cite{BoutetDuke} (8.7) p.~80.

\begin{proof}
Without loss of generality, we can reduce to the specific case where $\varepsilon=(1,0,\dots,0)$ and $\xi=0$ by coordinate change.
Notice $t\in L^\infty(\mathbb{R}^{n+d})$, 
$$\vert t\vert\leqslant
\int_{\mathbb{R}^{n+d}}
d\xi dk\left
(1+\vert k\vert+\vert\xi\vert\right)^{-n-d-1}$$ 
and 
$$\widehat{t}(k,\xi)=e^{-\frac{\sum_{i=2}^n k_i^2+\vert\xi\vert^2}{k^2_1}}\left(1+\vert k\vert+\vert\xi\vert\right)^{-n-d-1}(1-\alpha)$$ 
does not decrease faster than any polynomial inverse when $k_2=\dots=k_n=\xi_1=\dots=\xi_d=0, k_1>0$ which implies by Proposition 8.1.3 p.~254 in \cite{Hormander} that $WF(t)$ is
\textbf{nonempty}. $\widehat{t}$ is a smooth symbol on $T^\bullet\mathbb{R}^{n+d}$ (\cite{RS} p.~98--99) which does not depend on $(x,h)$ and the Fourier phase $(x.k+h.\xi)$ has critical points
only at $x=h=0$ thus by Theorem 9.47 p.~102--103 in \cite{RS}, 
we find that the singular support
of $t$ reduces to $(0,0)$ thus
$WF(t)\subset T_{(0,0)}^\bullet\mathbb{R}^{n+d}$ and $t\in C^\infty(\mathbb{R}^{n+d}\setminus\{h=0\})\cap L^\infty(\mathbb{R}^{n+d})$. 
But $WF(t)$ should be non empty
and the projection
on the second factor $(x,h;k,\xi)\in T^\bullet\mathbb{R}^{n+d} \mapsto (k,\xi)\in\mathbb{R}^{n+d}$ should
be contained in $\{k_2=\dots=k_n=\xi_1=\dots=\xi_d=0, k_1> 0\}$
so $WF(t)=(0,0;\lambda \varepsilon,0),\lambda>0$.
\end{proof}
The distribution $t$ is bounded hence weakly homogeneous of degree $0$, thus
the extension $\lim_{\varepsilon\rightarrow 0}\int_\varepsilon^1 \frac{d\lambda}{\lambda} t\psi_{\lambda^{-1}}=\lim_{\varepsilon\rightarrow 0}t(1-\chi_{\varepsilon^{-1}}) $
exists in $\mathcal{D}^\prime(\mathbb{R}^{n+d})$ by Theorem \ref{thm1}, 
is unique in $E_0(\mathbb{R}^{n+d})$ by Theorem \ref{removsing} and just corresponds to the extension of $t$
in $\mathcal{D}^\prime$
by integration against
test functions.
However, $\forall \varepsilon,\int_\varepsilon^1 \frac{d\lambda}{\lambda} t\psi_{\lambda^{-1}}=t(1-\chi_{\varepsilon^{-1}})\in C^\infty(\mathbb{R}^{n+d})$:
\begin{thm}\label{contrexamplequitue}
For all $p=(x_0,0;k,\xi)\in T^\bullet \mathbb{R}^{n+d}|_{I}$,
there exists a smooth function $t\in E_0(\mathbb{R}^{n+d}\setminus I)$ (thus $WF(t)=\emptyset$) which has
a unique extension 
$\overline{t}$
in $E_0(\mathbb{R}^{n+d})$ such that $p\in WF(\overline{t})$.
\end{thm}

\section{Appendix.}
\paragraph{The module structure of distributions supported on $I$.}
The concept of delta distribution $\delta_I$ of a 
submanifold $I$ is \emph{not intrinsically}
defined but 
a certain sheaf associated
to $I$ is canonically
defined: let $U$ be an open set
of $M$ and 
$(h^j)_{j=1,\cdots,d}\in \mathcal{I}(U)^d$ a collection of 
sections of the sheaf $\mathcal{I}$ 
of functions vanishing
on $I\cap U$ such that the
differentials
$dh^j,j=1,\cdots,d$ are
linearly independent ($(h^j)_{1\leqslant j\leqslant d}$
are transversal 
coordinates of 
a local chart). 
The map $h:U\mapsto \mathbb{R}^d$ 
allows to pullback $\delta^{\mathbb{R}^d}_0\in \mathcal{D}^\prime(\mathbb{R}^d)$
on $U$, and we denote this pullback $h^\star\delta^{\mathbb{R}^d}_{0}$ 
by $\delta_{\{h=0\}}$.
If we chose another system
of defining functions
$h^\prime$ for $I$, then
$\delta_{\{h^\prime=0\}}=\underset{\in C^\infty(I)}{\vert\frac{dh}{dh^\prime}\vert}\delta_{\{h=0\}}$,
where $\vert\frac{dh}{dh^\prime}\vert=\det(\frac{dh^j}{dh^{\prime i}})_{ij}$.
Thus the left module 
$C^\infty(I)\delta_{\{h=0\}}$ defined
over $U$ has
intrinsic meaning
(analoguous to the
space of sections 
of a vector bundle). 
Patching by
a partition
of unity gives a
sheaf of
modules of rank $1$
over $C^\infty(I)$.
Acting on the sections 
of this sheaf 
by differential operators
of order $k$ 
defines
a module of rank $\frac{d+k!}{d!k!}$ 
over $C^\infty(I)$.

\chapter{The microlocal extension.}
\paragraph{Introduction.}
Let $M$ be a smooth manifold and $I\subset M$ be a closed embedded submanifold of $M$.
In Chapter 2, we gave a necessary and sufficient condition on $WF(t),t\in\mathcal{D}^\prime(M\setminus I)$ that ensured that the union $\Gamma=\bigcup_{\lambda\in(0,1]}WF(t_\lambda)$ of the wave front sets of 
all scaled distribution $t_\lambda$ has the property $\overline{\Gamma}|_{I}\subset C$ where $C$ is the conormal of $I$. 
We saw this condition named \textbf{soft landing condition} (Definitions \ref{Broudermetric} and \ref{intrinsicslc}) was not sufficient to control the wave front set of the extension $\overline{t}$.
Our goal in this chapter is to add a boundedness condition which ensures the control of the wave front set of the extension.
Our plan starts with a 
geometric investigation 
of the dynamical properties of the scaling flow $e^{\log\lambda\rho}$ in cotangent space and show certain asymptotic behaviour of this flow. 
\section{Dynamics in cotangent space.}
In this section, we use the terminology 
and notation of
section $1$ of Chapter $2$.
We investigate 
the $\textbf{asymptotic behaviour}$ of the lifted flow $T^\star\Phi_\lambda$.
\paragraph{Decomposition in stable and unstable sets.}
We interpret $C,C_\rho$ as stable and unstable sets for the lifted flow $T^\star e^{t\rho}$ in cotangent space.
We work locally, let $p\in I$ and $V_p$ a neighborhood of $p$ in $M$, we fix a chart
$(x,h):V_p\mapsto \mathbb{R}^{n+d}$ in which $\rho=h^j\frac{\partial}{\partial h^j}$.
\begin{prop}\label{dynamicprop}
The flow $T^*e^{t\rho}$ lifted to the cotangent cone $T^\bullet V_p$ has the following property:
\begin{eqnarray}
\lim_{t\rightarrow +\infty} T^*e^{t\rho}(p)\in \left(C_\rho\cap T^\bullet V_p\right) 
\\ \lim_{t\rightarrow -\infty} T^*e^{t\rho}(p)\in \left(C \cap T^\bullet V_p\right)
\end{eqnarray}
in an \textbf{open dense} subset $T^\bullet V_p$.
\end{prop}
\begin{proof}
In coordinates $(x,h)$ 
in which $I=\{h=0\}$ 
and the flow has simple form  
$(x,h)\mapsto (x,e^th), $
the action lifts to 
$(x,h;k,\xi)\in T^\star \mathbb{R}^{n+d}\mapsto (x,e^th;k,e^{-t}\xi)\in T^\star \mathbb{R}^{n+d}$.
%in logarithmic parameterization $e^{\log\lambda \rho}$ acts via
%$$ (x,h,k,\xi)\mapsto (x,\lambda h,k,\lambda^{-1}\xi) .$$
We study the limit $t\rightarrow -\infty$, two cases arise:
\begin{itemize}
\item $\textbf{generically}$ $\xi\neq 0$, then $(x,e^t h;k,e^{-t}\xi)\sim (x,e^t h;e^t k,\xi) $ (because it is a cotangent cone) converges to $(x,0;0,\xi)$, it is immediate to deduce
$\{(x,0;0,\xi) \vert \xi\neq 0  \}=\left(T I\right)^\perp=C $
is the stable set of the flow. 
Notice the conormal bundle is an \textbf{intrinsic geometric object} and \textbf{does not depend on the choice} of vector field $\rho$.
\item Otherwise $\xi=0$, $(x,\lambda h; k,0)\rightarrow (x,0;k,0)$, the limit must lie in 
$\{(x,0;k,0)\vert k\neq 0\} \subset C^\rho$
which we will later see belongs to the unstable set.
\end{itemize}
% For instance, the fact that $p\notin C_\rho$ is $\textbf{preserved}$ by the flow, hence it is a robust fact. Furthermore, $\lim_{\lambda\rightarrow 0} T^*\Phi_\lambda(p)\in C$ and this is the crucial ingredient to prove that $\overline{\Gamma}|_I\subset C$. 
Conversely if $t\rightarrow\infty$:
\begin{itemize}
\item $\textbf{generically}$ $k\neq 0$, then $(x,e^t h;k,e^{-t}\xi)$ converges to $(x,0;k,0)$, it is immediate to deduce
$\{(x,h;k,0) \vert k\neq 0  \}=C_\rho$
is the unstable cone.
\end{itemize}
The flow $\lim_{t\rightarrow \infty}Te^{t\rho}$ sends all \textbf{conic sets in the complement} of $C$ to the coisotropic set $C_\rho$.
\end{proof}
Beware that the wave front set $WF\left(\Phi^*u\right)$ is the image of $WF(u)$ by the map $T^\star\Phi^{-1}$. If $\Phi=e^{\log\lambda\rho}$ then the interesting flow for the pull back will be $T^\star e^{-\log\lambda\rho}$ when $\lambda\rightarrow 0$.
This is why the properties established in the proposition \ref{dynamicprop} are crucial in the proof of the main theorem. Especially, we will use the fact that the flow $Te^{-\log\lambda\rho}$, when $\lambda\rightarrow 0$ sends all \textbf{conic sets in the complement} of $C$ to the coisotropic set $C_\rho$.
\subsection{Definitions.}
In this subsection, we recall results on distribution spaces that we will use in our proof of the main theorem which controls the wave front set of the extension. Furthermore the seminorms that we define here allow to write proper estimates.
%
% We recall the meaning of the strong topology of $D_m^\prime(K)$ which is the subspace of distribution $D_m^\prime(K)$ compactly supported in $K$ of order $m$.
%
% By a deep structure theorem of Laurent Schwartz, each element $t\in D_m^\prime(K)$ can be $\emph{uniquely}$ represented as a finite sum
%\begin{equation}
%\left\langle t , \varphi\right\rangle=\sum_{\vert\alpha\vert\leqslant m} \mu_\alpha(\partial^\alpha \varphi)
%\end{equation} 
%where $\mu_\alpha$ are $\textbf{Radon measures}$ (not necessarily positive but linear continuous on $C^0(K)$) supported on $K$.
%
%
% A basis of neighborhood of the origin in $D_m^\prime(K)$ for this topology is of the form $\{t \vert \vert\left\langle t,\varphi \right\rangle\vert\leqslant \varepsilon \sup_{x\in K,\vert\alpha\vert\leqslant m} \vert \partial^\alpha\varphi(x)\vert \}$.
We denote by $\theta$
the weight function $\xi\mapsto (1+\vert\xi\vert)$. 
For any cone $\Gamma\subset T^\star\mathbb{R}^d$, let $\mathcal{D}^\prime_\Gamma$ be the set of distributions with wave front set in $\Gamma$.
We define the set of seminorms $\Vert .\Vert_{N,V,\chi}$ 
on $\mathcal{D}^\prime_\Gamma$. 
\begin{defi}
For all $\chi\in \mathcal{D}(\mathbb{R}^{d})$, for all closed cone $V\subset\mathbb{R}^{d}\setminus \{0\}$ such that
$\left(\text{supp }\chi\times V\right)\cap \Gamma=\emptyset$,  
$\Vert t\Vert_{N,V,\chi}= \sup_{\xi\in V}\vert(1+\vert\xi\vert)^N\widehat{t\chi}(\xi)\vert$.
\end{defi}
We recall the definition
of the topology $\mathcal{D}^\prime_\Gamma$ (see \cite{Alesker} p.~14),
\begin{defi}
The topology of $\mathcal{D}^\prime_\Gamma $
is the weakest topology
that makes all seminorms
$\Vert .\Vert_{N,V,\chi}$
continuous
and which is stronger
than the weak topology
of $\mathcal{D}^\prime(\mathbb{R}^d)$.
%and also for any compactly supported distribution in $D^\prime(K)$ of order $m$,
%\begin{equation}
%\Vert (1+\vert k\vert+\vert\xi\vert)^{-m} \widehat{t} \Vert_{L^\infty}<\infty
%\end{equation}
Or it can be formulated as the topology defined
by
all seminorms
$\Vert .\Vert_{N,V,\chi}$
and the seminorms of the weak topology:
\begin{equation}
\forall \varphi\in \mathcal{D}\left(\mathbb{R}^{d}\right) ,\vert\left\langle t,\varphi \right\rangle\vert=P_\varphi\left(t\right).
\end{equation}
\end{defi}
%\begin{defi}
%$\lambda\mapsto c_\lambda \in L^p_{\frac{d\lambda}{\lambda}}([0,1],D_\Gamma^\prime),p\in [1,\infty)$ iff 
%\\for all $ \chi\in C_c^\infty(\mathbb{R}^{n+d})$,  for all closed cone $V\subset\mathbb{R}^{n+d *}$
%$$\left(\text{supp }\chi\times V\right)\cap \Gamma=\emptyset \implies  \int_0^1 \frac{d\lambda}{\lambda}\Vert c_\lambda\Vert^p_{N,V,\chi}<\infty  $$
%\end{defi}
We say that $B$ is bounded in $\mathcal{D}^\prime_\Gamma$, if $B$ is bounded in $\mathcal{D}^\prime$ and if for all seminorms $\Vert .\Vert_{N,V,\chi}$ defining the topology of $\mathcal{D}^\prime_\Gamma$,
$$\sup_{t\in B} \Vert t\Vert_{N,V,\chi}<\infty .$$
We also use the seminorms:
$$\forall \varphi\in\mathcal{D}(\mathbb{R}^d), \pi_m(\varphi)=\sup_{\vert\alpha\vert\leqslant m} \Vert \partial^\alpha\varphi\Vert_{L^\infty(\mathbb{R}^d)},$$
$$\forall \varphi\in\mathcal{E}(\mathbb{R}^d),\forall K\subset \mathbb{R}^d, \pi_{m,K}(\varphi)=\sup_{\vert\alpha\vert\leqslant m} \Vert \partial^\alpha\varphi\Vert_{L^\infty(K)}.$$
\section{Main theorem.} 
In this section, we prove the main theorem of this chapter which gives a sufficient condition to control the wave front set of the extension $\overline{t}$. 
The condition is as follows: Let $t\in E_s(M\setminus I)$ and assume $WF(t)$ satisfies 
the \textbf{soft landing condition},
and assume that $\lambda^{-s}t_\lambda$ is \textbf{bounded} in $\mathcal{D}^\prime_\Gamma$ where $\Gamma=\bigcup_{\lambda\in(0,1]}WF(t_\lambda)$. 
Then our theorem claims that $WF(\overline{t})\subset \overline{WF(t)}\cup C$ for the extension $\overline{t}$. 
\begin{thm}\label{mainthm} 
Let $s\in\mathbb{R}$ such that $s+d>0$, $\mathcal{V}$ be a $\rho$-convex neighborhood of $I$
and $t\in \mathcal{D}^\prime(\mathcal{V}\setminus I)$. 
Assume that $WF(t)$ satisfies the soft landing condition
and that $\lambda^{-s}t_\lambda$ is \textbf{bounded} in $\mathcal{D}^\prime_\Gamma(\mathcal{V}\setminus I)$ where 
$\Gamma=\bigcup_{\lambda\in(0,1]} WF(t_\lambda)\subset T^\bullet \left(M\setminus I\right)$. 
Then the wave front set  
of the extension
$\overline{t}$ 
of $t$ given by Theorem \ref{thm1}
%$$ \lambda^{-(s+d)}t\psi_{\lambda^{-1}}\in L_{\frac{d\lambda}{\lambda}}^\infty((0,1],D^\prime_{\overline{WF(t)}\cup C})$$
%and 
is such that $WF(\overline{t})\subset WF(t)\cup C$. 
\end{thm}
We saw in Chapter $2$ that the hypothesis that $WF(t)$ satisfies 
the \textbf{soft landing condition} is \emph{equivalent} to the requirement that $ \overline{\Gamma}|_I\subset C$ in particular, this implies that $\Gamma\cap C_\rho=\emptyset$ in a sufficiently small neighborhood of $I$ and $\overline{WF(t)}|_I\subset \overline{\Gamma}|_I\subset C$. Hence we have the relation $WF(\overline{t})\subset WF(t)\cup C=\overline{WF(t)}\cup C$.
\subsection{Proof of the main theorem.}
For the proof,
it suffices 
to work in flat space 
$\mathbb{R}^{n+d}$ 
with coordinates 
$(x,h)\in\mathbb{R}^n\times\mathbb{R}^d$ 
where $I=\{h=0\}$ and 
$\rho=h^j\frac{\partial}{\partial h^j}$,
since the hypothesis 
of the theorem and 
the result are 
local 
and open
properties. 

\begin{proof}
We denote by $\Xi$ the set $\overline{WF(t)}\cup C$.
The weight function $(1+\vert k\vert+\vert\xi\vert)$
is denoted by $\theta$.
In order to establish the inclusion
$WF(\overline{t})\subset \Xi$,
it suffices to prove that
for all
$p=(x_0,h_0;k_0,\xi_0)\notin \Xi$,
there exists $\chi$ s.t. $\chi(x_0,h_0)\neq 0$, $V$ a closed conic neighborhood of $(k_0,\xi_0)$
such that $\Vert \overline{t}\Vert_{N,V,\chi}<+\infty$ for all $N$. 
Let $p=(x_0,h_0;k_0,\xi_0)\notin \Xi$, then:

 Either $h_0\neq 0$, and we choose $\chi$ in such a way that $\chi=0$ on $I$ thus $t\chi=\overline{t}\chi$ and we are done since $\Vert \overline{t}\Vert_{N,V,\chi}=\Vert t\Vert_{N,V,\chi}<+\infty$.

 Either $h_0=0$ thus $k_0\neq 0$ since $p\notin C$. 
Since $\vert k_0\vert>0$, 
there exists $\delta^\prime>0$ s.t.
$$\vert k_0\vert\geqslant 2\delta^\prime\vert\xi_0\vert .$$
We set $V=\{(k,\xi)| \vert k\vert\geqslant \delta^\prime\vert\xi\vert\}$.
By the soft landing condition, 
$$\exists \varepsilon_1>0,\exists\delta>0,\, WF(t)|_{\vert h\vert\leqslant \varepsilon_1}\subset  
\{\vert k\vert\leqslant \delta \vert h\vert\vert \xi\vert \},$$
$$\text{and } \Gamma|_{\vert h\vert\leqslant \varepsilon_1}\subset  
\{\vert k\vert\leqslant \delta \vert h\vert\vert \xi\vert \}.$$
If we choose $\varepsilon>0$ in such a way that $\delta\varepsilon<\delta^\prime$ and $\varepsilon<\varepsilon_1$,
then for any function $\chi$ s.t. 
$\text{supp }\chi\subset \{\vert h\vert\leqslant \varepsilon\}$,
by the previous steps, we obtain that $\left(\text{supp }\chi\times V\right)\cap \Gamma=\emptyset$.
From now on, $\chi$ and $V$ are given. 
\begin{enumerate}
\item Recall $\psi=-\rho\chi^\prime$ is the Littlewood--Paley function on $\mathbb{R}^{n+d}$, and $\text{supp }\psi=\{a\leqslant \vert h\vert\leqslant 1\},0<a<1$ 
does not meet $I=\{h=0\}$. 
$\psi$ is defined on 
$\mathbb{R}^{n+d}$ 
but is not compactly 
supported in the $x$ variable.
We start from the definition of scaling given in Meyer (\cite{Meyer}) Definition 2.1 p.~45 Definition 2.2 p.~46:
$$\left\langle t_\lambda\psi,g \right\rangle=\lambda^{-d}\left\langle t\psi_{\lambda^{-1}},g_{\lambda^{-1}} \right\rangle  .$$ 
We pick the test functions $g$ defined by:
$$ g(x,h)=e^{-i(kx+\xi h)}\chi(x,h), $$
then application of the identity which 
defines the scaling gives: 
$$\widehat{t\psi_{\lambda^{-1}}\chi}=\lambda^{d}\widehat{t_\lambda\chi_\lambda\psi}(k,\lambda\xi) $$
The trick is to notice that 
$\psi\chi_\lambda$ has a compact support which does not meet $I=\{h=0\}$, because $\text{supp }\psi\subset \{a\leqslant\vert h\vert\leqslant b\}$ and $\chi(x,\lambda h)$ is compactly supported in $x$
uniformly in $\lambda$.
Thus we can find a compact subset $K\subset\mathbb{R}^{n+d}$ 
such that $\forall \lambda, \text{ supp }\chi_\lambda\psi\subset K$ 
and $K\cap I=\emptyset$ hence the above Fourier transforms
are well defined. 
Set the family of cones 
$V_\lambda=\{(k,\lambda\xi)\vert (x,\xi) \in V \}$. 
By definition of the seminorms $\Vert .\Vert_{N,V,\chi}$, 
we get
$$\Vert t\psi_{\lambda^{-1}} \Vert_{N,V,\chi}=\sup_{(k,\xi)\in V}(1+\vert k\vert+\vert\xi\vert)^N\vert \widehat{t\psi_{\lambda^{-1}}\chi}\vert$$ 
$$=\sup_{(k,\xi)\in V}(1+\vert k\vert+\vert\xi\vert)^N\lambda^{d}\vert\widehat{t_\lambda\chi_\lambda\psi}\vert(k,\lambda\xi),$$ 
we isolate the interesting term 
$$(1+\vert k\vert+\vert\xi\vert)^N\lambda^{d}\vert\widehat{t_\lambda\chi_\lambda\psi}\vert(k,\lambda\xi)=\frac{(1+\vert k\vert+\vert\xi\vert)^N}{(1+\vert k\vert+\lambda\vert\xi\vert)^N}(1+\vert k\vert+\lambda\vert\xi\vert)^N\lambda^{d}\vert\widehat{t_\lambda\chi_\lambda\psi}\vert(k,\lambda\xi).$$ 
We also have
$$\sup_{(k,\xi)\in V}(1+\vert k\vert+\lambda\vert\xi\vert)^N\lambda^{d}\vert\widehat{t_\lambda\chi_\lambda\psi}\vert(k,\lambda\xi)\leqslant \Vert \lambda^{d}t_\lambda\psi \Vert_{N,V_\lambda,\chi_\lambda},$$
by definition of $V_\lambda=\{(k,\lambda\xi)|(k,\xi)\in V\}$. 
\item Hence, we are reduced to prove that the quantity
$\frac{(1+\vert k\vert+\vert\xi\vert)^N}{(1+\vert k\vert+\lambda\vert\xi\vert)^N}$ remains bounded
for $(k,\xi)\in V$.
If so, we are able to apply estimates in Step 2 
to bound $\Vert t\psi_{\lambda^{-1}}\Vert_{N,V,\chi}$ 
in function of 
$\Vert \lambda^dt_\lambda\psi \Vert_{N,V_\lambda,\chi_\lambda}$.
The difficulty comes from the values of $\lambda$ close to $\lambda=0$.
But we find the following condition
\begin{equation}\label{monequa} 
\sup_{\lambda\in(0,1],(k,\xi)\in V}\frac{(1+\vert k\vert+\vert\xi\vert)^N}{(1+\vert k\vert+\lambda\vert\xi\vert)^N}<(1+\delta^{\prime-1})^N,
\end{equation}
this follows from:
$$(k,\xi)\in V\implies\delta^\prime\vert \xi\vert\leqslant \vert k\vert  $$
$$\implies 1\leqslant\frac{1+\vert k\vert+\vert\xi\vert}{1+\vert k\vert + \lambda\vert\xi\vert} \leqslant \frac{1+(1+\delta^{\prime -1})\vert k\vert}{1+\vert k\vert} \leqslant(1+\delta^{\prime -1}),$$
and implies the estimate 
$$\Vert t\psi_{\lambda^{-1}}\Vert_{N,V,\chi}\leqslant\lambda^{d}C\Vert 
t_\lambda\psi \Vert_{N,V_\lambda,\chi_\lambda},
$$
where $C=(1+\delta^{\prime-1})^N$.
By rescaling, we also have
\begin{equation}\label{aprioriestimate}
\forall\varepsilon>0, \, \Vert t\psi_{\lambda^{-1}}\Vert_{N,V,\chi}\leqslant\left(\frac{\lambda}{\varepsilon}\right)^{d}C\Vert 
t_{\frac{\lambda}{\varepsilon}}\psi_{\varepsilon^{-1}} \Vert_{N,V_{\frac{\lambda}{\varepsilon}},\chi_{\frac{\lambda}{\varepsilon}}}. 
\end{equation}
\item We return to
$V\subset \{\vert k\vert\geqslant \delta^\prime \vert\xi\vert \}$
thus
$$\text{supp }\chi \times V\subset \{\vert k\vert\geqslant \delta^\prime \vert h\vert\vert\xi\vert \} $$
since $\text{supp }\chi\subset\{\vert h\vert\leqslant\varepsilon\}$ 
and $\varepsilon$ can always be chosen $\leqslant 1$. 
For all $\lambda\leqslant\varepsilon$, 
we have the sequence of inclusions:
$$\text{supp }(\chi_{\frac{\lambda}{\varepsilon}}\psi_{\varepsilon^{-1}}) \times V_{\frac{\lambda}{\varepsilon}}\subset \text{supp }\chi\psi_{\varepsilon^{-1}} \times V \subset \{\vert k\vert\geqslant \delta^\prime \vert h\vert\vert\xi\vert \}, $$
from which we deduce an improvement of the rescaled estimate 
(\ref{aprioriestimate}):
$$\forall\lambda\leqslant \varepsilon, \Vert 
t_{\frac{\lambda}{\varepsilon}}\psi_{\varepsilon^{-1}} \Vert_{N,V_{\frac{\lambda}{\varepsilon}},\chi_{\frac{\lambda}{\varepsilon}}}\leqslant \Vert 
t_{\frac{\lambda}{\varepsilon}}\psi_{\varepsilon^{-1}}\chi_{\frac{\lambda}{\varepsilon}}\Vert_{N,V,\varphi^\prime}$$
for some function $\varphi^\prime\in\mathcal{D}(\mathbb{R}^{n+d})$ 
s.t. $\varphi^\prime=1$ 
on $\text{supp }\chi\psi_{\varepsilon^{-1}}$, 
$\varphi^\prime=0$ in a neighborhood of $I$ 
and $\left(\text{supp }\varphi^\prime\times V\right)\cap \Gamma=\emptyset$ 
(such $\varphi^\prime$ always exists by choosing $\varepsilon$ small enough in the first step of the proof
and by choosing $\text{supp }\varphi^\prime$ slightly larger than $\text{supp }\chi\psi_{\varepsilon^{-1}}$). 
We have gained the fact that the term $ \Vert 
t_{\frac{\lambda}{\varepsilon}}\psi_{\varepsilon^{-1}}\chi_{\frac{\lambda}{\varepsilon}}\Vert_{N,V,\varphi^\prime}$ on the r.h.s. is expressed in terms of a seminorm $\Vert.\Vert_{N,V,\varphi^\prime}$ where the cone $V$ does not depend on $\lambda$.
We still have to get rid of the dependance of the function $\psi_{\varepsilon^{-1}}\chi_{\frac{\lambda}{\varepsilon}}$ in $\lambda$. 
We use our estimates for the product 
of a smooth function and a distribution 
(see Estimate \ref{corollaire en or wanted by Christian}), 
for any $\textbf{arbitrary}$ cone $W$ 
which is a $\textbf{neighborhood}$ of $V$:
\begin{equation}\label{est45}
\begin{array}{c}
\Vert 
t_{\frac{\lambda}{\varepsilon}}\psi_{\varepsilon^{-1}}\chi_{\frac{\lambda}{\varepsilon}}\Vert_{N,V,\varphi^\prime}\\
\leqslant C \pi_{2N}\left(\psi_{\varepsilon^{-1}}\chi_{\frac{\lambda}{\varepsilon}}\right)\left(\Vert t_{\frac{\lambda}{\varepsilon}} \Vert_{N,W,\varphi^\prime} + \Vert\theta^{-m} t_{\frac{\lambda}{\varepsilon}}\varphi^\prime\Vert_{L^\infty} \right),
\end{array}
\end{equation}
where $\Vert . \Vert_{N,W,\varphi^\prime}$
is a seminorm of $\mathcal{D}^\prime_\Gamma$.
By using the hypothesis of the theorem
that $\lambda^{-s}t_\lambda$ is bounded in $\mathcal{D}^\prime_\Gamma$, we deduce that
$$\sup_{\lambda\in(0,\varepsilon]}\left(\frac{\lambda}{\varepsilon}\right)^{-s}\Vert t_{\frac{\lambda}{\varepsilon}} \Vert_{N,W,\varphi^\prime}<+\infty.$$
The above inequality
combined with the estimate (\ref{est45}), 
the estimate \ref{aprioriestimate}
and Theorem \ref{boundfourier} 
applied to the bounded family $\left(\lambda^{-s}t_\lambda\right)_{\lambda\in(0,1]}$
gives us: 
$$\forall\lambda\leqslant \varepsilon,\exists C^\prime, \Vert t\psi_{\lambda^{-1}}\Vert_{N,V,\chi} \leqslant C^\prime\left(\frac{\lambda}{\varepsilon}\right)^{s+d}.$$
\item This suggests we should decompose the integral $\int_0^1\frac{d\lambda}{\lambda}  t\psi_{\lambda^{-1}}$ in two parts:
$$\Vert\overline{t}\Vert_{N,V,\chi}=\Vert\int_0^1\frac{d\lambda}{\lambda}  t\psi_{\lambda^{-1}}\Vert_{N,V,\chi} $$ 
$$\leqslant \Vert \int_0^{\varepsilon}\frac{d\lambda}{\lambda}  t\psi_{\lambda^{-1}}\Vert_{N,V,\chi} + \Vert\int_{\varepsilon}^1\frac{d\lambda}{\lambda}  t\psi_{\lambda^{-1}}\Vert_{N,V,\chi} $$
$$\leqslant \int_0^{\varepsilon}\frac{d\lambda}{\lambda} \Vert t\psi_{\lambda^{-1}}\Vert_{N,V,\chi}+ \underset{<+\infty}{\underbrace{\Vert t(\chi-\chi_{\varepsilon^{-1}})\Vert_{N,V,\chi}}}, $$
because $t(\chi-\chi_{\varepsilon^{-1}})$ is supported away from $\{h=0\}$.
This reduces the study to 
$\int_0^{\varepsilon}\frac{d\lambda}{\lambda} \Vert t\psi_{\lambda^{-1}}\Vert_{N,V,\chi}$
which is bounded by $C^\prime\int_0^{\varepsilon}\frac{d\lambda}{\lambda} \left(\frac{\lambda}{\varepsilon}\right)^{s+d}<+\infty$.
\end{enumerate} 
We try to give an explicit
bound which ``summarizes'' 
all our previous arguments:
\begin{equation}\label{goodbound}
\begin{array}{c}
 \int_0^{\varepsilon}\frac{d\lambda}{\lambda} \Vert t\psi_{\lambda^{-1}}\Vert_{N,V,\chi}\\ 
\leqslant \frac{C\varepsilon^{s+d}}{2^{s+d}(s+d)}
\underset{\lambda\in(0,\varepsilon]}{\sup}
\left(\frac{\lambda}{\varepsilon}\right)^{-s}\pi_{2N}(\psi_{\varepsilon^{-1}}\chi_{\frac{\lambda}{\varepsilon}})\left(\Vert t_{\frac{\lambda}{\varepsilon}} \Vert_{N,W,\varphi^\prime} 
+ \Vert\theta^{-m}
\widehat{t_{\frac{\lambda}{\varepsilon}} } \varphi^\prime\Vert_{L^\infty}\right).
\end{array}
\end{equation}
\end{proof}

\paragraph{What do we need to reproduce the estimate (\ref{goodbound}) for families?}
We keep the same notation as in the proof 
and statement of theorem
(\ref{mainthm}).
The previous proof works for a fixed distribution $t$. 
We would like to reconsider 
the proof of the main theorem 
for a family $(t_\mu)_{\mu}$ 
of distributions bounded in $\mathcal{D}^\prime_\Gamma$. 
The validity of the previous theorem 
relied on the final estimate (\ref{goodbound}):
\begin{equation}
\begin{array}{c}
\int_0^{\varepsilon}\frac{d\lambda}{\lambda} \Vert t\psi_{\lambda^{-1}}\Vert_{N,V,\chi} \\
\leqslant \frac{C\varepsilon^{s+d}}{2^{s+d}(s+d)}
\underset{\lambda\in(0,\varepsilon]}{\sup}
\left(\frac{\lambda}{\varepsilon}\right)^{-s}\pi_{2N}(\psi_{\varepsilon^{-1}}\chi_{\frac{\lambda}{\varepsilon}})\left(\Vert t_{\frac{\lambda}{\varepsilon}} \Vert_{N,W,\varphi^\prime} + 
\Vert \theta^{-m} 
\widehat{t_{\frac{\lambda}{\varepsilon}} } \varphi^\prime\Vert_{L^\infty}\right).
\end{array}
\end{equation}
where the constants 
of the inequality
are \textbf{independent} 
of $t$. 
Hence the proof and the final estimate still works 
for the family of distributions 
$\mu^{-s}t_\mu$ since
the family $\lambda^{-s}(\mu^{-s}t_\mu)_\lambda=(\lambda\mu)^{-s}t_{\lambda\mu}$ is bounded in $\mathcal{D}^\prime_\Gamma(\mathcal{V}\setminus I)$ uniformly in $(\lambda,\mu)$.
Thus we have the proposition:
%\begin{thm} 
%Let $t_i\in E_s(M\setminus I)$.
%Assume the cone $\Gamma=\bigcup_{\lambda\in(0,1],i} WF(t_{i\lambda})\subset T^\bullet \left(M\setminus I\right) $ satisfies the constraint $\overline{\Gamma}|_I\subset C$. 
%And assume there is a $\rho$-convex neighborhood $V$ of $I$ such that the family $\left(\lambda^{-s}(t_{i\lambda}\right)_{(\lambda,i)\in [0,1]\times I}$ is \textbf{bounded} in $\mathcal{D}^\prime_\Gamma$ for $s+d>0$ then: 
% 
%
%$$ (\lambda^{-(s+d)}t_i\psi_{\lambda^{-1}})_{\lambda,i}$$
%is bounded in $D^\prime_{\overline{\Gamma}\cup C}$ and
%
%for each $i\in I$ the integral $\overline{t_i}=\int_0^1\frac{d\lambda}{\lambda} t_i\psi(\frac{h}{\lambda}) $ converges in $D^\prime_{\overline{WF(t_i)}\cup C}$ and the family $(\overline{t_i})_{I}$ is bounded in $D^\prime_{\overline{\Gamma}\cup C}$.
%
% 
%\end{thm}
\begin{prop}
If $t$ satisfies the assumptions of theorem (\ref{mainthm}),
then
the family $(\mu^{-s}\overline{t}_\mu)_{\mu\in(0,1]}$ 
is bounded in $\mathcal{D}^\prime_{\Gamma\cup C}(\mathcal{V})$.
\end{prop}

%\begin{prop}
% Now, if we are given only $\Gamma|_{1 \leqslant\vert h\vert\leqslant 3}$ then set the family:
%\begin{eqnarray}
%\Gamma_j=\Gamma|_{\{ 2^{-j}\leqslant\vert h\vert\leqslant 32^{-j} \}}=\{(x,2^{-j}h,k,2^j\xi)\vert (x,h,k,\xi)\in \Gamma|_{1 \leqslant\vert h\vert\leqslant 3}\}  
%\end{eqnarray}
%and the formula
%\begin{equation}
% \Gamma= \bigcup_{j=0}^\infty \Gamma_j 
%\end{equation}
%allows to reconstruct $\Gamma$ on $\{ 0 < \vert h\vert\leqslant 3 \}$.
%
% Finally, if $K\times V\cap \Gamma=\emptyset$ then set $K_{2^j}=\{x,2^{-j}h \vert (x,h)\in K\}$:
%
%$$\bigcup_{j=0}^\infty \left(K_{2^j}\setminus K_{2^{j+1}} \times V_j \right) \cap \Gamma=\emptyset $$
%\end{prop}
\subsection{The renormalized version of the main theorem.}
\paragraph{What do we need to extend the proof of the main theorem
to the case with counterterms ?}
In the course of the proof of \ref{mainthm}, we used that $\lambda^{-s}t_\lambda$ is bounded in $\mathcal{D}^\prime_\Gamma$. When $-m-1<s+d\leqslant m$, we need to introduce counterterms in the H\"ormander formula. 
We outline the proof 
of the renormalized case
following the main steps 
of the proof of Theorem \ref{mainthm}. 
We will sometimes denote by $\mathcal{F}\left[f\right]$, the Fourier transform $\widehat{f}$ of a Schwartz distribution $f$ and we denote by $e_{k,\xi}$ the Fourier character $e_{k,\xi}:(x,h)\mapsto e^{i(kx+\xi h)}$.
\begin{itemize}
\item The first step is identical, for 
$p=(x_0,0;k_0,\xi_0)\notin \overline{WF(t)}\cup C$, $k_0\neq 0$ 
we find a neighborhood $\text{supp }\chi\times V$ of $p$
such that $\text{supp }\chi\times V\cap \Gamma=\emptyset$
where $V\subset \{\vert k\vert\geqslant \delta^\prime \vert\xi\vert\}$
and $\text{supp }\chi\subset \{\vert h\vert\leqslant \varepsilon\}$ 
for some $\varepsilon,\delta^\prime>0$.
\item For the computational step,
we must use the Taylor formula
with integral 
remainder to take into account
the subtraction of counterterms:
$$ \mathcal{F}\left[\left(t\psi_{\lambda^{-1}}-\tau_\lambda\right)\chi\right](k,\xi)=\left\langle t \psi_{\lambda^{-1}},\underset{\text{subtraction of local counterterm}}{\underbrace{\left(1-\sum_{\vert\alpha\vert\leqslant m}\frac{h^\alpha}{\alpha !}(-\partial)^\alpha \delta_{h=0} \right)}} e_{k,\xi}\chi\right\rangle$$
$$=\left\langle t \psi_{\lambda^{-1}}, \underset{\text{Taylor remainder}}{\underbrace{\frac{1}{m!}\int_0^1 du(1-u)^{m}\left(\frac{\partial}{\partial u}\right)^{m+1} e_{k,u\xi}\chi_u}}\right\rangle $$
$$=\frac{1}{m!}\int_0^1 du(1-u)^{m}\left(\frac{\partial}{\partial u}\right)^{m+1}\widehat{t \psi_{\lambda^{-1}}\chi_u}(k,u\xi)$$
$$=\lambda^d\frac{1}{m!}\int_0^1 du(1-u)^{m}\left(\frac{\partial}{\partial u}\right)^{m+1}\widehat{t_\lambda \psi\chi_{\lambda u}}(k,u\lambda\xi) $$  
$$= \lambda^{d+m+1}\frac{1}{m!}\int_0^\lambda \frac{du}{\lambda}(1-\frac{u}{\lambda})^{m}\left(\frac{\partial}{\partial u}\right)^{m+1}\widehat{t_\lambda \psi\chi_{u}}(k,u\xi).$$ 
by variable change. 
We also 
introduce a 
rescaled version
of the previous
identity
with a 
variable parameter 
$\varepsilon>0$
in such a way that 
the cut-off function $\psi_{\varepsilon^{-1}}$
on the r.h.s.
restrict the expression
under the Fourier symbol
to the domain $\vert h\vert\leqslant\varepsilon$:
$$\forall\varepsilon>0,  \mathcal{F}\left[\left(t\psi_{\lambda^{-1}}-\tau_\lambda\right)\chi\right](k,\xi)$$ 
$$=\left(\frac{\lambda}{\varepsilon}\right)^{d+m}\frac{1}{m!}\int_0^{\frac{\lambda}{\varepsilon}}du(1-\frac{\varepsilon u}{\lambda})^{m}\left(\frac{\partial}{\partial u}\right)^{m+1}\mathcal{F}\left(t_{\frac{\lambda}{\varepsilon}} \psi_{\varepsilon^{-1}}\chi_{u}\right)(k,u\xi)  .$$
Since $\psi_{\varepsilon^{-1}}\subset \{\vert h\vert\leqslant \varepsilon\}$, 
we have the estimate
$$\partial_u^{m+1}\mathcal{F}\left(t_{\frac{\lambda}{\varepsilon}} \psi_{\varepsilon^{-1}}\chi_{u}\right)(k,u\xi)\leqslant (1+\varepsilon\vert\xi\vert)^{m+1}
\underset{0\leqslant j\leqslant m+1}{\sup}\left|\mathcal{F}(t_{\frac{\lambda}{\varepsilon}} \psi_{\varepsilon^{-1}}\partial_u^j\chi_{u})(k,u\xi)\right| ,$$
by Leibniz rule.

$$\vert (1+\vert k\vert+\vert\xi\vert)^N\left(\frac{\partial}{\partial u}\right)^{m+1}\mathcal{F}\left(t_{\frac{\lambda}{\varepsilon}} \psi_{\varepsilon^{-1}}\chi_{u}\right)(k,u\xi)  \vert$$
$$\leqslant (1+\vert k\vert+\vert\xi\vert)^{N+m+1}
\underset{0\leqslant j\leqslant m+1}{\sup}\left|\mathcal{F}(t_{\frac{\lambda}{\varepsilon}} \psi_{\varepsilon^{-1}}\partial_u^j\chi_{u})(k,u\xi)\right|$$
$$\leqslant\frac{(1+\vert k\vert+\vert\xi\vert)^{N+m+1}}{(1+\vert k\vert+u\vert\xi\vert)^{N+m+1}}(1+\vert k\vert+u\vert\xi\vert)^{N+m+1}
\underset{0\leqslant j\leqslant m+1}{\sup}\left|\mathcal{F}(t_{\frac{\lambda}{\varepsilon}} \psi_{\varepsilon^{-1}}\partial_u^j\chi_{u})(k,u\xi)\right|.$$
\item Following the proof of Theorem \ref{mainthm}, 
we find that the hypothesis (\ref{monequa}) $V\subset \{ \delta^\prime\vert \xi\vert\leqslant  \vert k\vert \}$ implies the estimate $$\sup_{(k,\xi)\in V}\frac{(1+\vert k\vert+\vert\xi\vert)^{N+m+1}}{(1+\vert k\vert+u\vert\xi\vert)^{N+m+1}}\leqslant  (1+\delta^{\prime -1})^{N+m+1} $$
from which we deduce:
$$\forall (k,\xi)\in V,\exists C,\, \vert (1+\vert k\vert+\vert\xi\vert)^N\left(\frac{\partial}{\partial u}\right)^{m+1}\mathcal{F}\left(t_{\frac{\lambda}{\varepsilon}} \psi_{\varepsilon^{-1}}\chi_{u}\right)(k,u\xi)   \vert$$
$$\leqslant C(1+\vert k\vert+u\vert\xi\vert)^{N+m+1}
\underset{0\leqslant j\leqslant m+1}{\sup}\left|\mathcal{F}(t_{\frac{\lambda}{\varepsilon}} \psi_{\varepsilon^{-1}}\partial_u^j\chi_{u})(k,u\xi)\right| .$$
\item Thus $\forall u\leqslant\frac{\lambda}{\varepsilon}$:
$$\Vert  \theta^N\left(\frac{\partial}{\partial u}\right)^{m+1}\mathcal{F}\left(t_{\frac{\lambda}{\varepsilon}} \psi_{\varepsilon^{-1}}\chi_{u}\right)(k,u\xi) \Vert_{L^\infty(V)}$$
$$\leqslant C\underset{0\leqslant j\leqslant m+1}{\sup} \Vert t_{\frac{\lambda}{\varepsilon}}\psi_{\varepsilon^{-1}}\Vert_{N+m+1,V_u,\partial_u^{j}\chi_{u}}$$
where $V_u=\{(k,u\xi) | (k,\xi)\in V \}$.
If we denote by $\chi^{(j)}_u=\partial_u^{j}\chi_{u}$,
by the same argument as in the proof
of Theorem \ref{mainthm}, for all $u\leqslant \frac{\lambda}{\varepsilon},\lambda\leqslant\varepsilon$,
we have the inclusion $\text{supp}\left(\psi_{\varepsilon^{-1}}\chi^{(j)}_{u}\right)\times V_{u}\subset \text{supp}\left(\psi_{\varepsilon^{-1}}\chi^{(j)}_{\frac{\lambda}{\varepsilon}}\right)\times V_{\frac{\lambda}{\varepsilon}}$
where 
$\text{supp}\left(\psi_{\varepsilon^{-1}}\chi^{(j)}_{\frac{\lambda}{\varepsilon}}\right)\times V\cap \Gamma=\emptyset ,$
which implies the estimate
$$\Vert t_{\frac{\lambda}{\varepsilon}}\psi_{\varepsilon^{-1}}\Vert_{N+m+1,V_u,\chi^{(j)}_{u}}\leqslant \Vert t_{\frac{\lambda}{\varepsilon}}\psi_{\varepsilon^{-1}}\chi^{(j)}_{u}\Vert_{N+m+1,V,\varphi^\prime} $$
where
$\varphi^\prime$ 
is any function in $\mathcal{D}(\mathbb{R}^{n+d})$
such that $\varphi^\prime=1$ on $\text{supp}\left(\psi_{\varepsilon^{-1}}\chi\right)$
and $\text{supp }\varphi^\prime\times V\cap \Gamma=\emptyset$.
Finally, we find that 
$$\Vert\left(t\psi_{\lambda^{-1}}-\tau_\lambda\right)\Vert_{N,V,\chi}$$ $$\leqslant C\left(\frac{\lambda}{\varepsilon}\right)^{d+m}\frac{1}{m!}\int_0^{\frac{\lambda}{\varepsilon}}du(1-\frac{\varepsilon u}{\lambda})^{m} \underset{u\in(0,1],0\leqslant j\leqslant m+1}{\sup}\Vert t_{\frac{\lambda}{\varepsilon}}\psi_{\varepsilon^{-1}}\chi^{(j)}_{u}\Vert_{N+m+1,V,\varphi^\prime}$$ 
$$\leqslant C \left(\frac{\lambda}{\varepsilon}\right)^{d+m+1}\frac{1}{m+1!}\underset{u\in(0,1],0\leqslant j\leqslant m+1}{\sup}\Vert t_{\frac{\lambda}{\varepsilon}}\psi_{\varepsilon^{-1}}\chi^{(j)}_{u}\Vert_{N+m+1,V,\varphi^\prime}$$
where we use the simple identity
$\frac{1}{m+1}=\int_0^1 du(1-u)^{m}$.
Then we use the estimates 
(\ref{corollaire en or wanted by Christian}) 
for the product of the bounded family of smooth functions $\psi_{\varepsilon^{-1}}\chi^{(j)}_{u}$ and the family of distributions $t_{\frac{\lambda}{\varepsilon}}$ 
and the assumption that $\lambda^{-s}t_\lambda$
is bounded in $\mathcal{D}^\prime_\Gamma$
to establish the estimate 
$$\sup_{u\leqslant 1}\Vert t_{\frac{\lambda}{\varepsilon}}\psi_{\varepsilon^{-1}}\chi^{(j)}_{u}\Vert_{N+m+1,V,\varphi^\prime}\leqslant C^\prime \left(\frac{\lambda}{\varepsilon}\right)^s$$ 
for all $0\leqslant j\leqslant m+1$. 
Then we can conclude in 
the same way as in the proof
of Theorem \ref{mainthm}:
$$\Vert\int_0^1\frac{d\lambda}{\lambda} \left(t\psi_{\lambda^{-1}}-\tau_\lambda\right)\Vert_{N,V,\chi}$$ 
$$\leqslant\Vert \underset{\in \mathcal{D}^\prime_{WF(t)}}{\underbrace{t(\chi-\chi_{\varepsilon^{-1}})}} \Vert_{N,V,\chi}
+ \Vert \underset{\in \mathcal{D}^\prime_C}{\underbrace{\int_\varepsilon^1 \tau_\lambda}}\Vert_{N,V,\chi}
+\int_0^\varepsilon\frac{d\lambda}{\lambda} \underset{\text{integrable}}{\left(\frac{\lambda}{\varepsilon}\right)^{s+d+m+1}}\frac{C}{m+1!}C^\prime,$$
where the last term is finite.
\end{itemize}
\begin{thm}\label{mainthm22}
Theorem \ref{mainthm} holds under the weaker assumption $s\in\mathbb{R}$.
Moreover
if $-s-d\in \mathbb{N}$ then
$\lambda^{-s^\prime}\overline{t}_\lambda$ is bounded in $\mathcal{D}^\prime_{\Gamma\cup C}(\mathcal{V})$ for all $s^\prime<s$,
if $-s-d\notin \mathbb{N}$
then $\lambda^{-s}\overline{t}_\lambda$ is bounded in $\mathcal{D}^\prime_{\Gamma\cup C}(\mathcal{V})$. 
\end{thm}
\section{Appendix}
\subsection{Estimates for the product of a distribution and a smooth function.}
\begin{thm}\label{producttheoremeskin}
Let $m\in\mathbb{N}$ 
and $\Gamma\subset T^\bullet(\mathbb{R}^d)$.  
Let $V$ be a closed cone in $\mathbb{R}^d\setminus {0}$ 
and $\chi\in \mathcal{D}(\mathbb{R}^d)$. 
Then for every $N$ and every closed conical neighborhood $W$ of $V$ 
such that $\left(\text{supp }\chi\times W\right)\cap \Gamma=\emptyset$, 
there exists a constant $C$ such that for all $\varphi \in \mathcal{D}(\mathbb{R}^d)$
and
for all $t\in \mathcal{D}^\prime_\Gamma(\mathbb{R}^d)$ such that $\Vert \theta^{-m} \widehat{t\chi} \Vert_{L^\infty}<+\infty$:
\begin{equation}\label{corollaire en or wanted by Christian}
\Vert t \varphi\Vert_{N,V,\chi} \leqslant C\pi_{2N,K}(\varphi) (\Vert t\Vert_{N,W,\chi}+\Vert \theta^{-m} \widehat{t\chi} \Vert_{L^\infty}).
\end{equation}
\end{thm}
\begin{proof}
We denote by $\theta$ the weight function
$\xi\mapsto (1+\vert\xi\vert)$ and $e_\xi:=x\mapsto e^{-ix.\xi}$ the Fourier character.
If the cone $V$ is given, we can always define a thickening $W$ of the cone $V$ such that $W$ is a closed conic neighborhood of $V$:
$$W=\{\eta\in\mathbb{R}^d\setminus \{0\} |\exists\xi\in V, \vert\frac{\xi}{\vert \xi\vert}- \frac{\eta}{\vert \eta\vert}\vert \leqslant \delta \},$$
intuitively this means that 
small angular perturbations of covectors in $V$ 
will lie in the neighborhood $W$. 
If 
$\left(\text{supp }\chi\times V\right)\cap \Gamma=\emptyset$ 
then $\delta$ 
can be chosen $\textbf{arbitrarily small}$ 
in such a way that
$\left(\text{supp }\chi\times W\right)\cap \Gamma=\emptyset$.  
We compute the Fourier transform
of the product:
$$\vert\widehat{t\varphi\chi}(\xi)\vert=
\vert \left\langle t\varphi , e_\xi \chi  \right\rangle\vert=\vert \widehat{t\chi} \star \widehat{\varphi} \vert(\xi) $$
$$\leqslant \int_{\mathbb{R}^d} \vert \widehat{\varphi}(\xi-\eta)    \widehat{t\chi} (\eta) \vert d\eta .$$
We reduce to the estimate $$ \int_{\mathbb{R}^d} \vert \widehat{\varphi}(\xi-\eta)    \widehat{t\chi} (\eta) \vert d\eta $$
$$\leqslant \underset{I_1(\xi)}{\underbrace{\int_{ \vert\frac{\xi}{\vert \xi\vert}- \frac{\eta}{\vert \eta\vert}\vert \leqslant \delta  } \vert \widehat{\varphi}(\xi-\eta)    \widehat{t\chi} (\eta) \vert d\eta}} + \underset{I_2(\xi)}{\underbrace{\int_{ \vert\frac{\xi}{\vert \xi\vert}- \frac{\eta}{\vert \eta\vert}\vert \geqslant \delta  }\vert \widehat{\varphi}(\xi-\eta)    \widehat{t\chi} (\eta) \vert d\eta}},$$
we will estimate separately the two terms $I_1(\xi),I_2(\xi)$. 
Start with $I_1(\xi)$, if $\xi\in V$ then $\vert\frac{\xi}{\vert \xi\vert}- \frac{\eta}{\vert \eta\vert}\vert \leqslant \delta \implies \eta\in W$ and by definition of the seminorms, we have the estimate
$$\forall N, \vert \widehat{t\chi}(\eta)\vert\leqslant \Vert t\Vert_{N,W,\chi} (1+\vert \eta\vert)^{-N} $$
then we use a trick due to Eskin, 
since $\varphi \in \mathcal{D}(\mathbb{R}^d)$, 
we also have 
$\vert\widehat{\varphi} (\xi-\eta) \vert\leqslant \Vert \theta^{2N}\widehat{\varphi}\Vert_{L^\infty}(1+\vert \xi-\eta\vert)^{-2N}\leqslant C\pi_{2N}(\varphi)(1+\vert \xi-\eta\vert)^{-2N}$ where $C=d^N\text{Vol (supp $\varphi$)}$ depends on $N$ and on the volume
of $\text{supp }\varphi$.
Hence
$$ \int_{ \vert\frac{\xi}{\vert \xi\vert}- \frac{\eta}{\vert \eta\vert}\vert \leqslant \delta  } \vert \widehat{\varphi}(\xi-\eta)    \widehat{t\chi} (\eta) \vert d\eta$$ $$\leqslant C\pi_{2N}(\varphi)\Vert t\Vert_{N,W,\chi}  (1+\vert \xi\vert)^{-N}  \int_{\mathbb{R}^d} \frac{(1+\vert \xi\vert)^{N}}{(1+\vert \eta\vert)^{N}  (1+\vert\xi-\eta\vert)^{2N}}d\eta $$
$$\leqslant C\pi_{2N}(\varphi) \Vert t\Vert_{N,W,\chi}  (1+\vert \xi\vert)^{-N} C_1$$ where $C_1=\sup_{\vert \xi\vert} \int_{\mathbb{R}^d} \frac{(1+\vert \xi\vert)^{N}}{(1+\vert \eta\vert)^{N}  (1+\vert\xi-\eta\vert)^{2N}}d\eta $ is finite
when $N\geqslant d+1$.
To estimate the second term $I_2(\xi)$, we use the inequality  $\vert\frac{\xi}{\vert \xi\vert}- \frac{\eta}{\vert \eta\vert}\vert \geqslant \delta$ which implies the angle beetween covectors is bounded below by an angle $\alpha=2\arcsin\frac{\delta}{2}>0$. 
By definition $\frac{\eta}{\vert \eta\vert} $ is in $\mathbb{R}^d\setminus \left(W\cup \{0\}\right)$, 
and $\frac{\xi}{\vert \xi\vert}\in V\subset W$ hence the angle between
$\frac{\xi}{\vert\xi\vert},\frac{\eta}{\vert\eta\vert}$
must be larger than $\alpha=2\arcsin\frac{\delta}{2}$. 
Then the trick
is to deduce lower bounds from the identity 
$a^2+b^2-2ab\cos c=(a-b\cos c)^2+b^2\sin^2c =(b-a\cos c)^2+a^2\sin^2c$, thus 
$$\forall (\xi,\eta)\in \left(V\times W^c\right),\vert( \sin\alpha) \eta\vert  \leqslant \vert \xi-\eta\vert, \vert(\sin\alpha) \xi\vert  \leqslant \vert \xi-\eta\vert .$$
We start again from the estimate on the Fourier tranform of $\varphi$, $\forall N$:
$$\vert\widehat{\varphi} (\xi-\eta) \vert\leqslant C\pi_{2N}(\varphi)(1+\vert \xi-\eta\vert)^{-2N} \leqslant C\pi_{2N}(\varphi)(1+\vert(\sin\alpha) \eta\vert )^{-N} (1+ \vert(\sin\alpha) \xi\vert)^{-N}$$
$$\quad \leqslant C\pi_{2N}(\varphi)\vert\sin\alpha\vert^{-2N}(1+\vert \eta\vert )^{-N} (1+ \vert \xi\vert)^{-N} $$
$$\int_{ \vert\frac{\xi}{\vert \xi\vert}- \frac{\eta}{\vert \eta\vert}\vert \geqslant \delta  } \vert \widehat{\varphi}(\xi-\eta)    \widehat{t\chi} (\eta) \vert d\eta$$
$$\leqslant C\pi_{2N}(\varphi)\vert\sin\alpha\vert^{-2N}(1+ \vert\xi\vert)^{-N} \int_{\mathbb{R}^d}(1+\vert\eta\vert )^{-N} \vert\widehat{t\chi} (\eta) \vert d\eta$$ $$\leqslant C\pi_{2N}(\varphi)\vert\sin\alpha\vert^{-2N}(1+ \vert\xi\vert)^{-N} \int_{\mathbb{R}^d}(1+\vert\eta\vert )^{-N} \Vert \theta^{-m} \widehat{t\chi} \Vert_{L^\infty}(1+\vert \eta \vert)^m d\eta $$
where $m$ is the order of the distribution, 
%and $\Vert (1+\vert\xi\vert)^{-m} \widehat{t\chi} \Vert_{L^\infty}$ is defined as the seminorm in the space of distributions:
%$$\Vert (1+\vert\xi\vert)^{-m} \widehat{t\chi} \Vert_{L^\infty}=\inf \{C | \forall\varphi\in C_c^\infty(K) , \vert\left\langle t,\varphi \right\rangle\vert \leqslant C\pi_m(\varphi)   \} $$
%$$\forall\varphi\in C_c^\infty(K), \vert\left\langle t,\varphi \right\rangle\vert \leqslant \Vert (1+\vert\xi\vert)^{-m} \widehat{t\chi} \Vert_{L^\infty}\pi_m(\varphi) $$
finally
$$I_2(\xi)\leqslant C_2\pi_{2N}(\varphi) (1+ \vert \xi\vert)^{-N} \Vert \theta^{-m} \widehat{t\chi} \Vert_{L^\infty}$$
where $C_2=C\vert\sin\alpha\vert^{-2N}
\int_{\mathbb{R}^d}(1+\vert\eta\vert )^{-N} (1+\vert \eta \vert)^md\eta$
is finite when 
$N\geqslant m+d+1$.
Gathering the two estimates, we have
$$ \int_{\mathbb{R}^d} \vert \widehat{\varphi}(\xi-\eta)    \widehat{t\chi} (\eta) \vert d\eta $$
$$\leqslant C\pi_{2N}(\varphi) (1+\vert \xi\vert)^{-N} \left(C_1\Vert t\Vert_{N,W,\chi}+C_2\Vert \theta^{-m} \widehat{t\chi} \Vert_{L^\infty}\right) $$ but recall the estimate on the right hand side is relevant 
provided $\delta>0$ which implies $\alpha>0$, 
$\delta$ depends 
on the choice of the cone $W$, 
the estimate is true for any cone $W$ 
such that $\text{dist }(W^c\cap\mathbb{S}^{d-1},V\cap \mathbb{S}^{d-1}) \geqslant \delta$.
%$$ \int_{\mathbb{R}^d} \vert \widehat{\varphi}(\xi-\eta)    \widehat{t\chi} (\eta) \vert d\eta $$
%$$\leqslant C\pi_{2N}(\varphi) (1+\vert \xi\vert)^{-N} \left(\Vert t\Vert_{N,W,\chi}   C_1 +  \Vert (1+\vert\xi\vert)^{-m} \widehat{t\chi} \Vert_{L^\infty}\pi_m(\chi) C_2(\delta) \right) $$
We have a final estimate 
$$\Vert t \varphi\Vert_{N,V,\chi} \leqslant C\pi_{2N}(\varphi)(\Vert t\Vert_{N,W,\chi}+\Vert \theta^{-m} \widehat{t\chi} \Vert_{L^\infty})$$
where $C$ is a constant which depends
on $N,V,W$ and the volume of $\text{supp }\varphi$.
\end{proof}

\chapter{Stability of the microlocal extension.}
\paragraph{Introduction.}
In Chapter $3$, 
we saw that there is 
a subspace of distributions of $\mathcal{D}^\prime(M\setminus I)$ 
for which we could 
control 
the wave front set 
of the extension 
$\overline{t}\in \mathcal{D}^\prime(M)$. 
In fact, 
we proved that if  
$WF(t)$ satisfies the soft landing condition 
and $\lambda^{-s}t_\lambda $ is bounded in $ \mathcal{D}^\prime_\Gamma$,
then $WF(\overline{t})\subset \overline{WF(t)}\bigcup C$.
Our assumptions obviously
depend on the choice of
some
Euler vector field
$\rho$.
Actually, our objective 
in this technical part 
is to investigate 
the dependence
of these conditions 
on the choice of 
$\rho$, 
their stability
when we pull-back by 
diffeomorphisms 
and when we multiply 
distributions
both satisfying
these hypotheses. 
This is absolutely necessary 
in order to prove by recursion that 
all 
vacuum expectation values 
$\left\langle 0\vert T(a_1(x_1)...a_n(x_n))  \vert 0 \right\rangle$ 
are well defined in the distributional sense.   
%\begin{thm}
%Let $\Gamma_1,\Gamma_2$ be two cones such that $\Gamma_1\cap-\Gamma_2=\emptyset$.
%Let $(f_{\mu})_{\mu},(g_{\nu})_{\nu}$ be two bounded families in $D^\prime_{\Gamma_1},D^\prime_{\Gamma_2}$ respectively. Then the family $(f_\mu g_\nu)_{(\mu,\nu)}$ is well defined in the sense of distributions and the family
%is bounded in $D^\prime_{\Gamma_1\cup\Gamma_2\cup \Gamma_1+\Gamma_2}$.
%\end{thm}
%
%
%
%\begin{thm}
%Let $\lambda\in[0,1]\mapsto \left(\Phi_{\lambda}(.):U\mapsto U\right)$ a family of diffeomorphisms of an open set $U$ which depends smoothly on a parameter $\lambda\in[0,1]$.
%$(t_\mu)_\mu$ be a bounded family in $D^\prime_\Gamma$, $\Gamma\subset T^\bullet U$.
%Then there exists a nested family of cones $\tilde{\Gamma}_\varepsilon$
%\begin{eqnarray}
%\varepsilon_1\leqslant\varepsilon_2\implies \tilde{\Gamma}_{\varepsilon_1}\subset \tilde{\Gamma}_{\varepsilon_2}
%\\ \bigcap_{\varepsilon} \tilde{\Gamma}_{\varepsilon}=\Gamma
%\end{eqnarray}
%such that $\forall \varepsilon >0$ the family $\left(\Phi_\lambda^*t_\mu\right)_{\lambda\in[0,\varepsilon],\mu}$
%is bounded in $D^\prime_{\tilde{\Gamma}_{\varepsilon}}$.
%\end{thm}
\section{Notation, definitions.}
We denote by $\theta$
the weight function $\xi\mapsto (1+\vert\xi\vert)$. 
We recall a theorem
of Laurent Schwartz
(see \cite{Schwartz} p.~86 Theorem (22))
which gives a concrete
representation
of bounded families
of distributions.
\begin{thm}\label{LaurentSchwartzbounded}
For a subset $B\subset \mathcal{D}^\prime(\mathbb{R}^d)$
to be bounded it is neccessary and sufficient 
that
for any domain $\Omega$ with compact closure,
there is a multi-index $\alpha$
such that
$\forall t\in B,\exists f_t\in C^0(\Omega)$ where
$t|_\Omega=\partial^\alpha f_t$
and $\sup_{t\in B}\Vert f_t\Vert_{L^\infty(\Omega)}<\infty$.
\end{thm}
We give 
an equivalent formulation
of the theorem
of Laurent Schwartz
in terms of Fourier
transforms:
\begin{thm}\label{boundfourier}
Let $B\subset \mathcal{D}^\prime(\mathbb{R}^d)$.
$$ \forall \chi\in \mathcal{D}(\mathbb{R}^d),\exists m\in\mathbb{N},\quad \sup_{t\in B} \Vert \theta^{-m}\widehat{t\chi}\Vert_{L^\infty}<+\infty$$ 
$$\Leftrightarrow B \text{ weakly bounded in }\mathcal{D}^\prime(\mathbb{R}^d)\Leftrightarrow B \text{ strongly bounded in }\mathcal{D}^\prime(\mathbb{R}^d).$$
\end{thm}
We refer the reader
to the appendix of this chapter
for a proof of the above theorem. 
For any cone $\Gamma\subset T^\star\mathbb{R}^d$, 
let $\mathcal{D}^\prime_\Gamma$ 
be the set of distributions with wave front set in $\Gamma$.
We define the set of seminorms $\Vert .\Vert_{N,V,\chi}$ 
on $\mathcal{D}^\prime_\Gamma$. 
\begin{defi}
For all $\chi\in \mathcal{D}(\mathbb{R}^{d})$, for all closed cone $V\subset(\mathbb{R}^{d}\setminus \{0\})$ such that
$\left(\text{supp }\chi\times V\right)\cap \Gamma=\emptyset$,  
$\Vert t\Vert_{N,V,\chi}= \sup_{\xi\in V}\vert(1+\vert\xi\vert)^N\widehat{t\chi}(\xi)\vert$.
\end{defi}
We recall the definition
of the topology $\mathcal{D}^\prime_\Gamma$ (see \cite{Alesker} p$14$),
\begin{defi}
The topology of $\mathcal{D}^\prime_\Gamma $
is the weakest topology
that makes all seminorms
$\Vert .\Vert_{N,V,\chi}$
continuous
and which is stronger
than the weak topology
of $\mathcal{D}^\prime(\mathbb{R}^d)$.
Or it can be formulated as the topology which makes
all seminorms
$\Vert .\Vert_{N,V,\chi}$
and the seminorms of the weak topology:
\begin{equation}
\forall \varphi\in \mathcal{D}\left(\mathbb{R}^{d}\right) ,\vert\left\langle t,\varphi \right\rangle\vert=P_\varphi\left(t\right)
\end{equation}
continuous.
\end{defi}
We say that $B$ is bounded in $\mathcal{D}^\prime_\Gamma$, if $B$ is bounded in $\mathcal{D}^\prime$ and if for all seminorms $\Vert .\Vert_{N,V,\chi}$ defining the topology of $\mathcal{D}^\prime_\Gamma$,
$$\sup_{t\in B} \Vert t\Vert_{N,V,\chi}<\infty .$$
We also use the seminorms:
$$\forall \varphi\in\mathcal{D}(\mathbb{R}^d), \pi_m(\varphi)=\sup_{\vert\alpha\vert\leqslant m} \Vert \partial^\alpha\varphi\Vert_{L^\infty(\mathbb{R}^d)},$$
$$\forall \varphi\in\mathcal{E}(\mathbb{R}^d),\forall K\subset \mathbb{R}^d, \pi_{m,K}(\varphi)=\sup_{\vert\alpha\vert\leqslant m} \Vert \partial^\alpha\varphi\Vert_{L^\infty(K)}.$$

\paragraph{Warning!}
%We would like
%to thank
%Christian Brouder 
%for pointing
%out an
%important error
%in a draft
%version
%of this
%chapter and we would like to
%thank Professor Semyon
%Alesker
%for helping us 
%to improve
%the clarity
%of our statements
%concerning products
%of distributions.
In this chapter, 
we will prove 
that if
$\Gamma_1,\Gamma_2$ are two
closed conic sets
in $T^\bullet\mathbb{R}^d$
such that $\Gamma_1\cap-\Gamma_2=\emptyset$, if we set $\Gamma=\Gamma_1\cup\Gamma_2\cup(\Gamma_1+\Gamma_2)$,
then the product
$(t_1,t_2)\in \mathcal{D}^\prime_{\Gamma_1}\times \mathcal{D}^\prime_{\Gamma_2} \mapsto t_1t_2\in \mathcal{D}^\prime_{\Gamma}$
is \textbf{jointly and separately sequentially continuous} and \textbf{bounded} 
for 
the topology of $\mathcal{D}^\prime_{\Gamma_1}\times \mathcal{D}^\prime_{\Gamma_2}$.
In fact, Professor Alesker informed us
that
he found a counterexample
which proves
that the product is \textbf{not
topologically
bilinear continuous}. 
This comes from the fact that
the space $\mathcal{D}^\prime_\Gamma$
is not bornological (see \cite{BrouderD}), for instance a
bounded linear map from
$\mathcal{D}^\prime_\Gamma$ to $\mathbb{C}$
\textbf{may not be continuous}.
We also prove that the pull-back
by a smooth diffeomorphism
$t\in \mathcal{D}^\prime_{\Gamma}\mapsto t\circ\Phi \in \mathcal{D}^\prime_{\Phi^\star \Gamma} $
is sequentially continuous and \textbf{bounded} 
from $\mathcal{D}^\prime_{\Gamma}$
to $\mathcal{D}^\prime_{\Phi^\star \Gamma}$. 
\section{The product of distributions.}
\subsection{Approximation and coverings.}
%We discuss approximation of sets in the cotangent cone $T^\bullet\mathbb{R}^n$.
In order to prove various theorems on the product of distributions and to discuss the action of Fourier integral operators on distributions, we should be able to approximate 
any conic set of $T^\bullet \mathbb{R}^d$ by some union of 
simple cartesian products of the form $K\times V\subset T^\bullet \mathbb{R}^d$ 
where $K$ is a compact set in space and 
$V$ is a closed cone in $\mathbb{R}^{d\bullet}$. 
We denote by $\mathbb{R}^d\overset{\pi_1}{\leftarrow} T^\star \mathbb{R}^d \overset{\pi_2}{\rightarrow} \mathbb{R}^{d*}$ the two projections on the base space $\mathbb{R}^d$ and the momentum space $\mathbb{R}^{d*}$ respectively. 
\begin{lemm}\label{approxlemma1}
Let $\Gamma_1,\Gamma_2$ be two \textbf{non intersecting} closed conic sets in $T^\bullet \mathbb{R}^d$. Then there is a family of closed cones $(V_{j1},V_{j2})_{j\in J}$ and a cover $(U_j)_{j\in J}$ of $\mathbb{R}^d$
such that
$$\Gamma_k\subset \bigcup_{j\in J} U_j\times V_{jk} $$
and $\forall j\in J, V_{j1}\cap V_{j2}=\emptyset$.
\end{lemm}
\begin{proof}
For all $x\in\mathbb{R}^d$, 
let $U_x(\varepsilon)$ be an open ball of radius 
$\varepsilon$ around $x$ and
$\Gamma_k|_x=\Gamma_k\cap T_x^\bullet \mathbb{R}^d$.
Let 
$V_{kx}(\varepsilon)=\pi_2\left(\Gamma_k|_{\overline{U_x(\varepsilon)}}\right)$ 
be a closed cone
which contains $\Gamma_k|_x$. 
We first establish that since 
$\Gamma_1|_x\cap \Gamma_2|_x=\emptyset$
and $\underset{\varepsilon>0}{\cap}\pi_2\left(\Gamma_k|_{U_x(\varepsilon)}\right)=\Gamma_k|_x$
we may assume that we can 
choose $\varepsilon$
small enough in such a way that
$V_{1x}\cap V_{2x}=\emptyset$:
assume that 
there exists a decreasing sequence
$\varepsilon_n\rightarrow 0$
such that
$$\forall n, V_{1x}(\varepsilon_n)\cap V_{2x}(\varepsilon_n)=\emptyset,$$
then let $\eta_n\in  V_{1x}(\varepsilon_n)\cap V_{2x}(\varepsilon_n)$ 
for all $n$
where we may assume that 
$\vert\eta_n\vert=1$.
Using the definition
of $V_{kx}(\varepsilon_n)$, there is a 
sequence $x_{kn}$ s.t. $(x_{kn};\eta_n)\in \Gamma_k|_{\overline{U_x(\varepsilon_n)}}$. 
$(x_{kn};\eta_n)$ 
lives in the compact set 
$\overline{U_x(\varepsilon_0)}\times \mathbb{S}^{d-1}$
and we can therefore
extract a convergent subsequence
which converges to $(x_k;\eta_k)\in\Gamma_k$
since $\Gamma_k$ is closed.
Furthermore $\eta_1=\eta_2=\eta$
and $x_{kn}\in \overline{U_x(\varepsilon_n)}$
implies $\lim_{n\rightarrow \infty} x_{kn}=x$
thus $(x;\eta)\in\Gamma_1\cap\Gamma_2$,
contradiction ! 
For all $x$, 
we thus have 
$\Gamma_k|_{U_x}\subset U_x\times V_{kx}$. 
Since $(U_x)_{x\in\mathbb{R}^d}$
forms an open cover
of $\mathbb{R}^d$,
we can extract a locally finite subcover
$(U_j)_{j\in J}$
and $\Gamma_k\subset \bigcup_{j\in J} U_j\times V_{jk}$.
\end{proof}
\begin{lemm}\label{approxlemma2} 
Let $\Gamma$ be a closed
conic set
in $T^\bullet\mathbb{R}^d$.
For every partition of unity $(\varphi^2_j)_{j\in J}$
of $\mathbb{R}^d$ 
and family of functions 
$(\alpha_j)_{j\in J}$ in $C^\infty(\mathbb{R}^d\setminus 0)$, homogeneous of degree $0$, 
$0\leqslant \alpha_j\leqslant 1$ 
such that
$\Gamma\bigcap \left(\bigcup_{j\in J}\text{supp }\varphi_j\times\text{supp }(1-\alpha_j)\right)=\emptyset$,
we have
$$\forall t\in \mathcal{D}^\prime_{\Gamma}, t=\sum_{j\in J} \underset{\text{singular part}}{\underbrace{\varphi_j\mathcal{F}^{-1}\left(\alpha_j\widehat{t\varphi_j}\right)}}
+\underset{\text{smooth part}}{\underbrace{\varphi_j\mathcal{F}^{-1}\left((1-\alpha_j)\widehat{t\varphi_j}\right)}}.$$
\end{lemm}
\begin{proof}   
Let $\mathcal{D}^\prime_{\Gamma}$ denote the 
set of all distributions with wave front set in $\Gamma$.
We use the highly non trivial 
lemma $8.2.1$ of \cite{Hormander}:
Let $t\in \mathcal{D}^\prime_{\Gamma}$, 
for any $\varphi\in \mathcal{D}(\mathbb{R}^d)$, 
for any $V$
such that 
$\left(\text{supp }\varphi\times V\right)\cap\Gamma=\emptyset$, 
we have
$\forall N, \Vert t\Vert_{N,V,\varphi}<\infty .$ 
Set the family of functions $V_j=\text{supp }(1-\alpha_j)$ then
$(\text{supp }\varphi_j\times\text{supp }(1-\alpha_j))\cap\Gamma=\emptyset$ hence $(1-\alpha_j)\widehat{t\varphi_j} $ has fast decrease at infinity 
and its inverse Fourier transform is a smooth function which yields
the result.
\end{proof}
\subsection{The product is bounded.}
A relevant example of
products of distributions first appeared in the work of Alberto Calderon in $1965$. 
A nice exposition of this work can be found in the article \cite{Meyerprod} by Yves Meyer. 
Actually, Meyer defines $\Gamma$-\emph{holomorphic} distributions as Schwartz distributions in $S^\prime\left(\mathbb{R}^d\right)$ the
Fourier transform of which 
is supported 
on a cone $\overline{\Gamma}\subset \mathbb{R}^d$
where $\Gamma\subset\mathbb{R}^d\setminus 0$ 
is defined by the inequality
$0<\vert\xi\vert\leqslant \delta \xi_d$ where $\delta>1$.
Notice that $\xi_d$ 
must be positive and that $0\notin \Gamma +\Gamma$. 
Then Meyer defines the functional spaces $L^p_\alpha$ 
which are analogs of the
classical Sobolev spaces $W^{\alpha,p}$ 
for positive $\alpha$,
and proves that for any pair 
$(t_1,t_2)\in L^p_\alpha\times L^q_\beta$
the product $t_1t_2$ makes sense, $t_1t_2$ is 
$\Gamma$-\emph{holomorphic} and belongs to the functional
space $L^r_{\alpha+\beta}$ where $r^{-1}=p^{-1}+q^{-1}$. Most importantly, Meyer proves there is a \emph{bilinear continuous mapping} $P_\Gamma$ which satisfies a H\"older like estimate and coincides with the product when $t_1,t_2$ are $\Gamma$-\emph{holomorphic}.

 In the same spirit, we will prove bilinear estimates for the product of distributions. 
The bilinear estimates 
are formulated 
in terms of 
the seminorms $\Vert.\Vert_{N,V,\chi}$ 
defining the topology of $\mathcal{D}^\prime_\Gamma$
and
the seminorms:
\begin{equation}
\Vert \theta^{-m} \widehat{t\chi} \Vert_{L^\infty}.
\end{equation}
which control boundedness in 
$\mathcal{D}^\prime$ (but they
do not
define the weak
topology of $\mathcal{D}^\prime$).
We closely follow the exposition of \cite{Eskin} thm $(14.3)$.
\begin{lemm}\label{lemm0}
Let $\Gamma_1,\Gamma_2$ be two conic sets in $T^\bullet\mathbb{R}^d$. 
If $\Gamma_1\cap-\Gamma_2=\emptyset$,
then there exists a partition of unity $(\varphi^2_j)_{j\in J}$ 
and a family of closed cones $(W_{j1},W_{j2})_{j\in J}$ in $\mathbb{R}^d\setminus {0}$ such that $\forall j\in J, W_{j1} \cap -W_{j2}=\emptyset$ and $\Gamma_k\subset \left(\bigcup_{j\in J} supp(\varphi_j)\times W_{jk} \right),(k=1,2)$.
\end{lemm}
\begin{proof}   
% We cover $K\times V$ in $\textbf{cotangent}$ space $K\times V\subset\bigcup_{j\in J}K_j\times V$ such that $$\forall j\in J, \left( K_j\times V\right) \cap \left(\Sigma(\varphi_jt_1)\cup\Sigma(\varphi_jt_2)\cup
%\left(\Sigma(\varphi_jt_1)+\Sigma(\varphi_jt_2)\right)\right)=\emptyset$$ 
% 
We use our approximation lemma 
for $\Gamma_1$ and $-\Gamma_2$.
The approximation lemma gives us a pair of covers
$$ \Gamma_k \subset \bigcup_{j\in J} U_j \times W_{jk}, k\in\{1,2\} ,$$ 
then pick a partition of unity $(\varphi^2_j)_{j\in J}$ 
subordinated to the cover $\bigcup_{j\in J}U_j$
and we are done.
% For any compactly supported distribution $t$, we denote $\Sigma(t)$ 
%the complement of the set $$\Sigma(t)=&^c\{\xi_0| \exists \delta, \forall \xi, \vert \frac{\xi}{\vert\xi\vert}-\frac{\xi_0}{\vert\xi_0\vert}\vert\leqslant \delta , \widehat{t}(\tau \xi)=O(\tau^{-\infty})   \}$$ 
% Set $\Sigma(\varphi_jt_k)=W_{jk},k=1,2$.
% 
% For $K_j$ sufficiently small, we can make
%$W_{j1} \cap -W_{j2}=\emptyset $ for all $j$. 
%
% To prove it, let us recall that forall $x$ in $K$, $\Gamma_{1x}\cap -\Gamma_{2x}=\emptyset$ means $\forall x\in K, \forall (\xi,\eta) \in \Gamma_{1x}\times\Gamma_{2x}, \vert\frac{\xi}{\vert\xi\vert}-\frac{\eta}{\vert\eta\vert}\vert  \geqslant \delta$, we have to choose the support of $\varphi_j$ small enough so that $\forall x,j, \varphi_j(x)\neq 0,  \forall \xi \in  W_{jk}, \text{dist}(\frac{\xi}{\vert\xi\vert},\Gamma_x)  \leqslant \frac{\delta}{3} $. Then $\forall (\xi,\eta)\in W_{j1}\times W_{j2}$ rewrite $\frac{\xi}{\vert\xi\vert}-\frac{\eta}{\vert\eta\vert}=\frac{\xi}{\vert\xi\vert}-\xi_1+\xi_1-\xi_2+\xi_2-\frac{\eta}{\vert\eta\vert}$ for $\xi_1,\xi_2$ in $\Gamma_{x1},\Gamma_{x2}$. Hence $\vert\frac{\xi}{\vert\xi\vert}-\frac{\eta}{\vert\eta\vert}\vert \geqslant \vert\vert\xi_1-\xi_2\vert-\vert \frac{\xi}{\vert\xi\vert}-\xi_1+ \xi_2-\frac{\eta}{\vert\eta\vert} \vert\vert \geqslant \frac{\delta}{3} $ because we choose $\xi_k$ such that $\vert \frac{\xi}{\vert\xi\vert}-\xi_1\vert= \text{dist }(\frac{\xi}{\vert\xi\vert},W_{j1})$ and $\vert \frac{\eta}{\vert\eta\vert}-\xi_2\vert= \text{dist }(\frac{\eta}{\vert\eta\vert},W_{j2})$.
% 
\end{proof}
\begin{lemm}\label{lemm1}
Let $\Gamma_1,\Gamma_2$ be two cones in $T^\bullet \mathbb{R}^d$ and
let $m_1,m_2$ be given non negative integers.
Assume $\Gamma_1\cap -\Gamma_2=\emptyset$  
then for all $\chi\in  \mathcal{D}(\mathbb{R}^d)$, 
for all $N_2\geqslant N_1+d+1$
there exists $C$ 
such that for all 
$(t_1,t_2)\in \mathcal{D}_{\Gamma_1}^\prime(\mathbb{R}^d)\times \mathcal{D}_{\Gamma_2}^\prime(\mathbb{R}^d)$ 
satisfying $\Vert \theta^{-m_1} \widehat{t_{1}\chi\varphi_j}\Vert_{L^\infty}<+\infty$ and
$\Vert \theta^{-m_2} \widehat{t_{2}\chi\varphi_j}\Vert_{L^\infty}<+\infty$,
we have the bilinear estimate:  
$$\Vert \theta^{-(m_1+m_2+d)} \widehat{t_1t_2\chi^2}(\xi)\Vert_{L^\infty}$$ 
$$\leqslant C\sum_{j\in J} \left(\Vert \theta^{-m_1} \widehat{t_{1}\chi\varphi_j}\Vert_{L^\infty} +\Vert t_{1}\chi \Vert_{N_1,V_{j1},\varphi_j}\right)\left(\Vert \theta^{-m_2} \widehat{t_{2}\chi\varphi_j}\Vert_{L^\infty}+\Vert t_{2}\chi \Vert_{N_2,V_{j2},\varphi_j} \right)$$ 
for some seminorms $\Vert . \Vert_{N_k,V_{jk},\varphi_j}$ of $\mathcal{D}^\prime_{\Gamma_k},k=1,2$.
\end{lemm}
Before we prove the lemma, let us explain the crucial consequence
of this lemma for the product of distributions.
Let $B_k,k\in\{1,2\}$ be bounded subsets of $\mathcal{D}_{\Gamma_k}^\prime(\mathbb{R}^d),k\in\{1,2\}$.
Then for each fixed $\chi$, there exists a pair $m_1,m_2$ such that the r.h.s. of the bilinear estimate
is bounded for all $t_1,t_2$ describing $B_1\times B_2$ by theorem (\ref{vouluChristian}). 
Thus for each fixed $\chi^2\in\mathcal{D}(\mathbb{R}^d)$, there exists
an integer $m_1+m_2+d$ such that $\Vert \theta^{-(m_1+m_2+d)} \widehat{t_1t_2\chi^2}(\xi)\Vert_{L^\infty}$ is bounded for all $t_1,t_2$ describing $B_1\times B_2$. Then this implies again by (\ref{vouluChristian}) 
that $t_1t_2$ is bounded in $\mathcal{D}^\prime(\mathbb{R}^d)$. So the consequence of this lemma
can be summarized as follows
\begin{coro}
Let $\Gamma_1,\Gamma_2$ be two cones in $T^\bullet \mathbb{R}^d$.
Assume $\Gamma_1\cap -\Gamma_2=\emptyset$.
Then the product 
$(t_1,t_2)\in \mathcal{D}_{\Gamma_1}^\prime(\mathbb{R}^d)\times \mathcal{D}_{\Gamma_2}^\prime(\mathbb{R}^d) \mapsto t_1t_2\in \mathcal{D}^\prime(\mathbb{R}^d) $
is well defined and bounded.
\end{coro}
Now let us return to the proof of lemma (\ref{lemm1}).

\begin{proof} 
By Lemma \ref{lemm0}
$\Gamma_k\subset \bigcup_{j\in J}\text{supp }\varphi_j\times W_{jk},k\in\{1,2\}$ for a partition of unity $(\varphi^2_j)_{j\in J}$ and
for a family of closed cones $(W_{j1},W_{j2})_{j\in J}$ in $\mathbb{R}^d\setminus {0}$ such that $\forall j\in J, W_{j1} \cap -W_{j2}=\emptyset$.
In a similar way to the construction of the approximation lemma, we have $$t_1t_2\chi^2=\sum_{j\in J} (\chi\varphi_jt_1)(\chi\varphi_jt_2)=\sum_{j\in J} t_{j1}t_{j2}.$$ 
where we set $t_{jk}=(\chi\varphi_jt_k)$. 
Set $\alpha_{jk},k\in\{1,2\}$ a smooth function on $\mathbb{R}^d\setminus\{0\}$, $\alpha_{jk}=1 $ on $W_{jk}$, homogeneous of degree $0$ such that $\text{supp }(\alpha_{j1})\cap -\text{supp }(\alpha_{j2})=\emptyset$.    
We decompose the convolution product $I(\xi)=\int_{\mathbb{R}^d}d\eta t_{j1}(\xi-\eta)t_{j2}(\eta)$ into four parts:
\begin{eqnarray} 
I_1=\int_{\mathbb{R}^d}d\eta \alpha_{j1}\widehat{t_{j1}}(\xi-\eta)\alpha_{j2}\widehat{t_{j2}}(\eta) 
\\ I_2=\int_{\mathbb{R}^d}d\eta (1-\alpha_{j1})\widehat{t_{j1}}(\xi-\eta)\alpha_{j2}\widehat{t_{j2}}(\eta)  
\\ I_3=\int_{\mathbb{R}^d}d\eta \alpha_{j1}\widehat{t_{j1}}(\xi-\eta)(1-\alpha_{j2})\widehat{t_{j2}}(\eta) 
\\ I_4=\int_{\mathbb{R}^d}d\eta (1-\alpha_{j1})\widehat{t_{j1}}(\xi-\eta)(1-\alpha_{j2})\widehat{t_{j2}}(\eta) 
\end{eqnarray} 
We would like to estimate $I(\xi)$ for \textbf{arbitrary} $\xi$.
Let us first discuss the more singular term $I_1$. The key point is that
its integrand vanishes outside the domain
$\vert\eta\vert\leqslant \frac{\vert\xi\vert}{\sin\delta} $
for some $\delta$. 
Indeed, we observe that $\text{supp }\alpha_{j1} \cap -\text{supp }\alpha_{j2}=\emptyset$ means that for any $(\zeta_1,\zeta_2) \in \text{supp }\alpha_{j1} \times \text{supp }\alpha_{j2}$, 
the angle $\theta$ between $\zeta_1$ and $\zeta_2$ is less than $\pi-\delta$ for a given $\delta>0$. 

Hence if $\zeta_1=\xi-\eta\in \text{supp }\alpha_{j1}$  and $\zeta_2=\eta \in  \text{supp }\alpha_{j2}$
the angle between
$\zeta_1$ and $\zeta_2$ is bounded from below: 
$$\vert \zeta_1+\zeta_2 \vert^2=\left\langle\zeta_1+\zeta_2 ,\zeta_1+\zeta_2  \right\rangle =\vert\zeta_1\vert^2+\vert\zeta_2\vert^2 + 2\cos\theta \vert \zeta_1\vert\vert\zeta_2\vert $$ 
$$=(\vert\zeta_1\vert+\cos\theta\vert\zeta_2\vert)^2 + \sin^2\theta\vert\zeta_2\vert^2 \geqslant \sin^2\theta\vert\zeta_2\vert^2\geqslant \sin^2\delta\vert\zeta_2\vert^2, $$
hence $\vert\sin\delta\vert \vert\eta\vert\leqslant \vert \xi\vert $ 
and 
$\vert\sin\delta\vert \vert\xi-\eta\vert\leqslant  \vert\xi\vert$ by symmetry between $\zeta_1,\zeta_2$.
Thus
$$\vert I_1\vert\leqslant \int_{\vert \xi\vert \geqslant \vert\sin \delta\vert \vert\eta\vert} d\eta \Vert \theta^{-m_1} \widehat{t_{j1}}\Vert_{L^\infty} \Vert \theta^{-m_2} \widehat{t_{j2}}\Vert_{L^\infty}(1+\vert\xi-\eta\vert)^{m_1}(1+\vert\eta\vert)^{m_2}$$
if $\vert\xi\vert$ is fixed we integrate a rational function over a ball
$$\vert I_1\vert\leqslant \vert\sin\delta\vert^{-m_1-m_2} \Vert \theta^{-m_1} \widehat{t_{j1}}\Vert_{L^\infty} \Vert \theta^{-m_2} \widehat{t_{j2}}\Vert_{L^\infty}  (1+\vert\xi\vert)^{m_1+m_2} \int_0^{\frac{\vert\xi\vert}{\vert\sin\delta\vert}} r^{d-1}dr$$
 $$ \leqslant C_1\Vert \theta^{-m_1} \widehat{t_{j1}}\Vert_{L^\infty} \Vert \theta^{-m_2} \widehat{t_{j2}}\Vert_{L^\infty} (1+\vert\xi\vert)^{m_1+m_2+d}  $$
where $C_1=\frac{2\pi^{\frac{d}{2}}}{\Gamma(\frac{d}{2})}\vert(\sin\delta)^{-d-m_1-m_2}\vert$
does not depend
on $t_1,t_2$.
We have estimated the more singular term, set $\text{supp }(1-\alpha_{jk})=V_{jk}$, we choose $\alpha_{jk}$ in such a way that $V_{jk}=\overline{W^c_{jk}}$. The estimation of others terms is simple and
relies on the 
key inequalities $ \frac{(1+\vert \eta\vert)}{(1+\vert\xi\vert)(1+\vert\xi-\eta\vert)}\leqslant
1$ and 
$\frac{(1+\vert\xi-\eta\vert)}{(1+\vert\xi\vert)(1+\vert \eta\vert)}\leqslant 1$.
We gather 
all results:
$$ I_1 \leqslant \frac{2\pi^{\frac{d}{2}}}{\Gamma(\frac{n}{2})}\Vert \theta^{-m_1} \widehat{t_{j1}}\Vert_{L^\infty} \Vert \theta^{-m_2} \widehat{t_{j2}}\Vert_{L^\infty}  (1+\vert\xi\vert)^{m_1+m_2+d}$$
$$ I_2 \leqslant \Vert t_{1}\chi \Vert_{m_2+d+1,V_{j1},\varphi_j} \Vert \theta^{-m_2} \widehat{t_{j2}}\Vert_{L^\infty} \int_{\mathbb{R}^d} d\eta (1+\vert\xi-\eta\vert)^{-(m_2+d+1)} (1+\vert \eta\vert)^{m_2}$$ 
$$\leqslant \Vert t_{1}\chi \Vert_{m_2+d+1,V_{j1},\varphi_j}  \Vert \theta^{-m_2} \widehat{t_{j2}}\Vert_{L^\infty}(1+\vert\xi\vert)^{m_2}  \int_{\mathbb{R}^d} d\eta \frac{(1+\vert \eta\vert)^{m_2}}{(1+\vert\xi\vert)^{m_2}(1+\vert\xi-\eta\vert)^{(m_2+d+1)}} $$
$$ I_3 \leqslant  \Vert \theta^{-m_1} \widehat{t_{j1}}\Vert_{L^\infty}\Vert t_{2}\chi \Vert_{m_1+d+1,V_{j2},\varphi_j}\int_{\mathbb{R}^d} d\eta (1+\vert\xi-\eta\vert)^{m_1} (1+\vert \eta\vert)^{-(m_1+d+1)}  $$
$$\leqslant \Vert \theta^{-m_1} \widehat{t_{j1}}\Vert_{L^\infty}\Vert t_{2}\chi \Vert_{m_1+d+1,V_{j2},\varphi_j}(1+\vert\xi\vert)^{m_1}\int_{\mathbb{R}^d} d\eta \frac{(1+\vert\xi-\eta\vert)^{m_1}}{(1+\vert\xi\vert)^{m_1}(1+\vert \eta\vert)^{m_1+d+1}}  $$
$$ I_4 \leqslant \Vert t_{1}\chi \Vert_{N_1,V_{j1},\varphi_j}\Vert t_{2}\chi \Vert_{N_2,V_{j2},\varphi_j}(1+\vert\xi\vert)^{-N_1}\int_{\mathbb{R}^d} d\eta   \frac{(1+\vert\xi\vert)^{N_1}}{(1+\vert\xi-\eta\vert)^{N_1}(1+\vert \eta\vert)^{N_2}} .$$
%$\frac{(1+\vert\xi\vert)^{N_1}}{(1+\vert\xi-\eta\vert)^{N_1}(1+\vert \eta\vert)^{N_2}}$
%The integral $\int_{\mathbb{R}^d} d\eta   \frac{(1+\vert\xi\vert)^{N_1}}{(1+\vert\xi-\eta\vert)^{N_1}(1+\vert \eta\vert)^{N_2}}$
%is finite if $N_2\geqslant N_1+d+1$
%since by the triangular inequality:
%$$\int_{\mathbb{R}^d} d\eta   \frac{(1+\vert\xi\vert)^{N_1}}{(1+\vert\xi-\eta\vert)^{N_1}(1+\vert \eta\vert)^{N_2}}\leqslant \int_{\mathbb{R}^d} d\eta   \frac{((1+\vert\xi-\eta\vert)+(1+\vert\eta\vert))^{N_1}}{(1+\vert\xi-\eta\vert)^{N_1}(1+\vert \eta\vert)^{N_2}}$$
%then by the inequality $\forall (a,b)>1, (a+b)^{N_1}\leqslant (ab)^{N_1}$,
%$$\leqslant  \int_{\mathbb{R}^d} d\eta   \frac{(1+\vert\xi-\eta\vert)^{N_1}(1+\vert\eta\vert)^{N_1}}{(1+\vert\xi-\eta\vert)^{N_1}(1+\vert \eta\vert)^{N_2}}$$ 
%which finally gives: 
%$$=\int_{\mathbb{R}^d} d\eta   \frac{(1+\vert\eta\vert)^{N_1}}{(1+\vert \eta\vert)^{N_2}}=
%\int_{\mathbb{R}^d} d\eta (1+\vert\eta\vert)^{N_1-N_2}<+\infty.$$
We write the estimates 
in a more compact form 
where we replaced 
the integrals by constants $(C_i)_{1\leqslant i\leqslant 4}$:
\begin{eqnarray}
I_1 \leqslant C_1\Vert \theta^{-m_1} \widehat{t_{j1}}\Vert_{L^\infty} \Vert \theta^{-m_2} \widehat{t_{j2}}\Vert_{L^\infty}  (1+\vert\xi\vert)^{m_1+m_2+d}\\
I_2 \leqslant C_2 \Vert t_{1}\chi \Vert_{m_2+d+1,V_{j1},\varphi_j}  \Vert \theta^{-m_2} \widehat{t_{j2}}\Vert_{L^\infty}(1+\vert\xi\vert)^{m_2}\\      
I_3   \leqslant C_3 \Vert \theta^{-m_1} \widehat{t_{j1}}\Vert_{L^\infty}\Vert t_{2}\chi \Vert_{m_1+d+1,V_{j2},\varphi_j}(1+\vert\xi\vert)^{m_1}\\
I_4 \leqslant C_4 \Vert t_{1}\chi \Vert_{N_1,V_{j1},\varphi_j}\Vert t_{2}\chi \Vert_{N_2,V_{j2},\varphi_j}(1+\vert\xi\vert)^{-N_1}
\end{eqnarray}
then we summarize the whole estimate, if $N_2\geqslant N_1+d+1$:  
$$(1+\vert\xi\vert)^{-m_1-m_2-d}\vert I\vert$$ 
$$\leqslant C  \left(\Vert \theta^{-m_1} \widehat{t_{j1}}\Vert_{L^\infty} +\Vert t_{1}\chi \Vert_{N_1,V_{j1},\varphi_j}\right)\left(\Vert \theta^{-m_2} \widehat{t_{j2}}\Vert_{L^\infty}+\Vert t_{2}\chi \Vert_{N_2,V_{j2},\varphi_j} \right) .$$  
\end{proof}

\begin{lemm}
Let $\Gamma_1,\Gamma_2$ be two cones in $T^\bullet \mathbb{R}^d$ and
$m_1,m_2$ some non negative integers.
Assume $\Gamma_1\cap -\Gamma_2=\emptyset$.  
Set $\Gamma=\Gamma_1\cup\Gamma_2\cup\Gamma_1+\Gamma_2$.
Then for all seminorm $\Vert .\Vert_{N,V,\chi^2}$ of 
$\mathcal{D}^\prime_\Gamma$ where
$N\geqslant \sup_{k=1,2} m_k+d+1$, there exists $C$ such that  
for all $(t_1,t_2)\in \mathcal{D}_{\Gamma_1}^\prime(\mathbb{R}^d)\times \mathcal{D}_{\Gamma_2}^\prime(\mathbb{R}^d)$
satisfying $\Vert \theta^{-m_1}\widehat{t_1\chi}\Vert_{L^\infty}<\infty, \Vert\theta^{-m_2}\widehat{t_2\chi}\Vert_{L^\infty}<\infty$,
we have the bilinear estimate:
$$\Vert t_1t_2\Vert_{N,V,\chi^2} 
 \leqslant C \sum_{j\in J} \Vert t_2\chi \Vert_{2N,V_{j2},\varphi_j} \Vert\theta^{-m_1}\widehat{t_1\varphi_j\chi}\Vert_{L^\infty} $$ $$ + \Vert t_1\chi \Vert_{2N,V_{j1},\varphi_j}  \Vert\theta^{-m_2}\widehat{t_2\varphi_j\chi}\Vert_{L^\infty}+\Vert t_1 \Vert_{2N,V_{j1},\varphi_j} \Vert t_2 \Vert_{N,V_{j2},\varphi_j}$$
for some seminorms $\Vert . \Vert_{N,V_{jk},\varphi_j}$ of $\mathcal{D}^\prime_{\Gamma_k},k=1,2$.
\end{lemm}
\begin{proof}
Let $V$ be a closed cone of $\mathbb{R}^d$ such that
$\text{supp }\chi\times V$ does not meet $\Gamma_1\cup\Gamma_2\cup\Gamma_1+\Gamma_2$.
Now, it is always possible to use the cover 
given by the approximation lemma \emph{fine enough} 
so that for all $j\in J$, $V$ will not meet $W_{j1}\cup W_{j2}\cup (W_{j1}+W_{j2})$. 
We would like to estimate $I(\xi)$ for $\xi \notin W_{j1}\cup W_{j2}\cup (W_{j1}+W_{j2}) $.
But $\alpha_{j2}(\eta)\alpha_{j1}(\xi-\eta)\neq 0 \implies (\eta,\xi-\eta) \in W_{j2}\times W_{j1}\implies \xi=(\xi-\eta) + \eta\in W_{j1}+W_{j2}$. Thus if $\xi\notin W_{j1}+W_{j2}$ then $\alpha_{j2}(\eta)\alpha_{j1}(\xi-\eta)=0$ for all $\eta$, 
hence $I_1(\xi)=0$ when $\xi\in V$. 
We set $\text{supp }(1-\alpha_{jk})=V_{jk}$ which is a cone in which $t_{jk}$ decreases faster than any inverse of polynomial function. 
By definition:
$$\vert(1-\alpha_{jk})\widehat{t}_{jk}\vert (\xi) \leqslant \Vert t_k\chi \Vert_{N,V_{jk},\varphi_j}\left(1 + \vert \xi\vert \right)^{-N} $$
also for $\alpha_{jk}\widehat{t}_{jk}$ where $t_{jk}=(t_k\chi)\varphi_j$, we have:
$$\vert\alpha_{jk}\widehat{t}_{jk} \vert(\xi)\leqslant \Vert (1+\vert\xi\vert)^{-m_k}\widehat{t_{jk}}\Vert_{L^\infty}(1+\vert\xi\vert)^{m_k} $$
where $m_k$ is the order of the compactly supported distribution $t_k\chi$.
We can estimate $I_4$ in a simple way:
$$\vert I_4\vert(\xi) \leqslant  \Vert t_1\chi \Vert_{2N,V_{j1},\varphi_j} \Vert t_2\chi \Vert_{N,V_{j2},\varphi_j}(1+\vert\xi\vert)^{-N} \int_{\mathbb{R}^d}d\eta\frac{\left(1 + \vert \xi \vert \right)^{N}}{\left(1 + \vert \xi-\eta \vert \right)^{2N}\left(1 + \vert \eta\vert \right)^{N}} $$ 
$$\vert I_4\vert(\xi) \leqslant C_N \Vert t_1\chi \Vert_{2N,V_{j1},\varphi_j} \Vert t_2\chi \Vert_{N,V_{j2},\varphi_j}(1+\vert\xi\vert)^{-N},  $$
where $C_N=\int_{\mathbb{R}^d}d\eta\frac{\left(1 + \vert \xi \vert \right)^{N}}{\left(1 + \vert \xi-\eta \vert \right)^{2N}\left(1 + \vert \eta\vert \right)^{N}}\leqslant \int_{\mathbb{R}^d}d\eta\left(1 + \vert \eta\vert \right)^{-N}$.
 
To estimate $I_2$, 
let us first notice that if $\alpha_{jk}$ were smooth at $0$ then we could identify the ``good function'' $(1-\alpha_{j1})\widehat{t}_{j1}(\eta)$ with the Fourier transform of a Schwartz function and ''the bad function'' $\alpha_{j2}\widehat{t}_{j2}(\eta)$ with the Fourier transform of a distribution.
%$$I_2(\xi)=\widehat{\mathcal{F}^{-1}(1-\alpha_{j1})t_{j1}} \star \widehat{\mathcal{F}^{-1}(\alpha_{j2}t_{j2})}=\widehat{\mathcal{F}^{-1}(1-\alpha_{j1})t_{j1}\mathcal{F}^{-1}(\alpha_{j2}t_{j2})} $$
%and this kind of estimate has been already dealt in the estimation of $\widehat{t\varphi}$ (see Chapter $3$).
%% To get back to the setting of our previous proof, we will multiply $\alpha(\xi)$ by a function $\chi(\vert\xi \vert)=1,\vert\xi\vert\geqslant 2 ,\chi(\vert\xi\vert)=0, \vert\xi\vert\leqslant 1$.  
%So we use the same idea as in the previous estimate and 
Denoting by $\theta(\xi,\eta)$ the angle between $\xi$ and $\eta$, we cut $I_2$ into two parts:$$I_2(\xi)=\int_{\theta(\xi,\eta)\leqslant \delta}(1-\alpha_{j1})\widehat{t}_{j1}(\xi-\eta)\alpha_{j2}\widehat{t}_{j2}(\eta) + \int_{\theta(\xi,\eta)\geqslant \delta}(1-\alpha_{j1})\widehat{t}_{j1}(\xi-\eta)\alpha_{j2}\widehat{t}_{j2}(\eta) $$
We set the cone $W_{kj}^\prime=\{\xi \vert \text{dist }(\frac{\xi}{\vert\xi\vert},W_{kj})\leqslant\delta \} $
for some $\delta>0$ in such a way that
the following sequence of inclusions holds:  
$$W_{kj} \subset \text{supp }\alpha_{jk} \subset W_{kj}^\prime .$$   
The restrictions $\xi\in V,\eta \in \text{supp }\alpha_{j2}$ impose the angle $\theta(\xi,\eta)$ between them satisfies the bound $\theta\geqslant dist(V\cap\mathbb{S}^{d-1}, \text{supp }\alpha_{j2}\cap\mathbb{S}^{d-1})>0,$
hence if $ \delta < dist(V\cap\mathbb{S}^{d-1},W_{j2}\cap\mathbb{S}^{d-1})$
then 
$$\forall \xi \in V, I_2(\xi)=\int_{\theta(\xi,\eta)\geqslant \delta}(1-\alpha_{j1})\widehat{t}_{j1}(\xi-\eta)\alpha_{j2}\widehat{t}_{j2}(\eta) ,$$
but the estimate $\theta(\xi,\eta)\geqslant \delta$ exactly means 
that the angle between $\xi,\eta$ is bounded
from below hence we use the bounds
$$\vert\xi-\eta  \vert\geqslant \sin\delta \vert \xi\vert, \vert\xi-\eta  \vert\geqslant \sin\delta\vert \eta\vert  $$
which implies $$(1+\vert\xi-\eta  \vert)^{-2N} \leqslant (1+ \sin\delta\vert \xi\vert )^{-N} (1+ \sin\delta\vert \eta\vert )^{-N}\leqslant (\sin\delta)^{-2N}(1+\vert\xi\vert)^{-N}(1+\vert\eta\vert)^{-N} $$
which implies the following bounds for $I_2$:
$$\forall\xi\in V, \vert I_2\vert(\xi)$$ 
$$ \leqslant \int_{\theta(\xi,\eta)\geqslant \delta} d\eta \Vert t_1\chi \Vert_{2N,V_{j1},\varphi_j} (1+\vert\xi-\eta  \vert)^{-2N} \Vert \theta^{-m_2}\widehat{t_{j2}}\Vert_{L^\infty} (1+\vert \eta\vert)^{m_2}$$
$$\leqslant \Vert t_1\chi \Vert_{2N,V_{j1},\varphi_j} \Vert \theta^{-m_2}\widehat{t_{j2}}\Vert_{L^\infty}(1+\vert \xi\vert )^{-N}\vert\sin\delta\vert^{-2N}\int_{\mathbb{R}^d}d\eta (1+\vert \eta\vert )^{-N} (1+\vert \eta\vert)^{m_2}.$$
Provided
that 
$dist(V\cap\mathbb{S}^{d-1}, W_{j2}\cap\mathbb{S}^{d-1}) >\delta>0$ 
and $N\geqslant m_2+d+1$, 
the integral on the right hand side absolutely converges.
Setting $C_2=\vert\sin\delta\vert^{-2N}\int_{\mathbb{R}^d}d\eta (1+\vert \eta\vert )^{-N} (1+\vert \eta\vert)^{m_2}$ yields the estimate
$$\forall\xi\in V, \vert I_2\vert(\xi) \leqslant C_2\Vert t_1\chi \Vert_{2N,V_{j1},\varphi_j}  \Vert \theta^{-m_2}\widehat{t_{j2}}\Vert_{L^\infty}(1+\vert \xi\vert )^{-N}.$$
Now for $I_3(\xi)$, after the variable change $$\int_{\mathbb{R}^d}d\eta \vert\alpha_{j1}t_{j1}(\xi-\eta) (1-\alpha_{j2})t_{j2}(\eta)\vert=\int_{\mathbb{R}^d}d\eta \vert\alpha_{j1}t_{j1}(\eta) (1-\alpha_{j2})t_{j2}(\xi-\eta)\vert ,$$ we repeat the exact same proof as above with the roles of the indices $1,2$ exchanged.  
$$\forall\xi\in V, \vert I_3\vert(\xi) \leqslant C_{3} \Vert t_2\chi \Vert_{2N,V_{j2},\varphi_j}  \Vert \theta^{-m_1}\widehat{t_{j1}}\Vert_{L^\infty}(1+\vert \xi\vert )^{-N}$$
where $C_{3}=\vert\sin\delta\vert^{-2N}\int_{\mathbb{R}^d}d\eta (1+ \vert\eta\vert )^{-N} (1+\vert \eta\vert)^{m_1}$.
Gathering the three terms, we obtain:
$$\forall\xi\in V, \vert I\vert(\xi)\leqslant C( \Vert t_2\chi \Vert_{2N,V_{j2},\varphi_j} \Vert \theta^{-m_1}\widehat{t_{j1}}\Vert_{L^\infty} 
$$ $$+ \Vert t_1\chi \Vert_{2N,V_{j1},\varphi_j} \Vert \theta^{-m_2}\widehat{t_{j2}}\Vert_{L^\infty} +\Vert t_1\chi \Vert_{2N,V_{j1},\varphi_j} \Vert t_2\chi \Vert_{N,V_{j2},\varphi_j} ) (1+\vert\xi\vert)^{-N} .$$
\end{proof}
Let us explain the boundedness properties
of the product.
Let $B_k,k\in\{1,2\}$ be bounded subsets of $\mathcal{D}_{\Gamma_k}^\prime(\mathbb{R}^d),k\in\{1,2\}$.
Then for each $V$ satisfying the hypothesis of the lemma for each $\chi$, there exists a pair $(m_1,m_2)$ such that the r.h.s. of the bilinear estimate
is bounded for all $t_1,t_2$ describing $B_1\times B_2$ by theorem (\ref{vouluChristian}). 
Thus the seminorm $\Vert t_1t_2\Vert_{N,V,\chi^2}$ 
is bounded for all $t_1,t_2\in B_1\times B_2$.
The joint 
and partial sequential continuity
of the product simply
follows 
from the above arguments.
As a corollary of the previous lemmas, we deduce the following important
\begin{thm}\label{productbounded}
Let $\Gamma_1,\Gamma_2$ be two cones in $T^\bullet \mathbb{R}^d$.
Assume $\Gamma_1\cap -\Gamma_2=\emptyset$.
Set $\Gamma=\left(\Gamma_1\cup\Gamma_2\cup(\Gamma_1+\Gamma_2)\right)$, where $x,\xi\in \Gamma_1+\Gamma_2$ means that $\xi=\xi_1+\xi_2$ for some $(x,\xi_1)\in\Gamma_1,(x,\xi_2)\in \Gamma_2$.
Then the product 
$$(t_1,t_2)\in D_{\Gamma_1}^\prime\times D_{\Gamma_2}^\prime \mapsto t_1t_2\in D_{\Gamma}^\prime $$
is well defined and \textbf{bounded}.
\end{thm}
\subsection{The soft landing condition is stable by sum.}
We have studied the boundedness properties of the product. The main theorem of Chapter $3$ singled out an essential property of the wave front set of distributions which was the \textbf{soft landing condition}. Our goal in this subsection will be to check that this condition on wave front sets is \emph{stable} by products. If $WF(t_i)_{\in\{1,2\}}$ satisfies the soft landing condition and $WF(t_1)\cap (-WF(t_2))=\emptyset$ on $M\setminus I$, then what happens to $WF(t_1t_2)$ ?
%Recall that $\Gamma_i,i\in\{1,2\}$ satisfies the
%\textbf{soft landing condition}
%if for all $\rho$-convex compact set $K$,
%\begin{equation}\label{Broudermetric}
%\forall i\in \{1,2\},\exists\varepsilon_i>0,\exists\delta_i>0, \Gamma_i|_{K\cap\vert h\vert\leqslant \varepsilon}\subset\{\vert k\vert\leqslant \delta\vert h\vert\vert\xi\vert \}
%\end{equation}   
\begin{prop}\label{sumstable}
Let $\Gamma_1,\Gamma_2$ be two closed conic sets which both satisfy the soft landing condition
and $\Gamma_1,\Gamma_2$ are such that $\Gamma_1\cap (-\Gamma_2)=\emptyset$.
Then the cone $\Gamma_1\cup \Gamma_2\cup \Gamma_1+\Gamma_2$ satisfies the \textbf{soft landing condition}.
\end{prop}
\begin{proof}
We just have to prove that $\Gamma_1+\Gamma_2$ satisfies the soft landing condition because taken individually, $\Gamma_i,\in\{1,2\}$ already satisfy the soft landing condition.
We denote $(x_i,h_i;k_k,\xi_i)$ a point in $\Gamma_i,\in\{1,2\}$. We also denote $\eta_i=(k_i,\xi_i)$. In the course of the proof, 
we use the norm $\vert\eta\vert=\vert k\vert+\vert\xi\vert$
and the result 
does not depend on
the choice of
this norm since all norms
are equivalent.
\begin{enumerate}
\item We start from the hypothesis that $\Gamma_i,\in\{1,2\}$ both satisfy the soft landing condition
$$\forall i\in \{1,2\},\exists\varepsilon_i>0,\exists\delta_i>0, \Gamma_i|_{K\cap\vert h\vert\leqslant \varepsilon}\subset\{\vert k\vert\leqslant \delta\vert h\vert\vert\xi\vert \}$$
but this implies that for the points of the form $(x,h;\eta_1)+(x,h;\eta_2)=(x,h;\eta_1+\eta_2)\in (\Gamma_1+\Gamma_2)|_{(x,h)}$, we have the inequality
$$\vert k_1+k_2\vert\leqslant \sup_{\in\{1,2\}}\delta_i\vert h\vert\left(\vert\xi_1\vert+\vert \xi_2\vert\right) ,$$
from now on, we set $\sup_{\in\{1,2\}}\delta_i=\delta$.
\item In order to estimate the sum $\left(\vert\xi_1\vert+\vert \xi_2\vert\right) $, we will use the fact that $\Gamma_1\cap -\Gamma_2=\emptyset$. This can be translated in the estimate
$$\forall (x,h;\eta_i)\in \Gamma_i|_K,\exists \delta^\prime>0 , \delta^\prime \left(\vert\eta_1\vert+\vert \eta_2\vert \right)\leqslant \vert\eta_1 +\eta_2\vert $$
$$\implies \delta^\prime \left(\vert k_1\vert+\vert k_2\vert+\vert\xi_1\vert+\vert \xi_2\vert \right)\leqslant \vert k_1+k_2\vert+ \vert \xi_1 + \xi_2\vert $$
$$\implies\vert\xi_1\vert+\vert \xi_2\vert \leqslant \frac{1-\delta^\prime}{\delta^\prime}\vert k_1+k_2\vert+ \frac{1}{\delta^\prime}\vert \xi_1 + \xi_2\vert  ,$$
where we can always assume we chose $\delta^\prime<1$.
\item Combining the two previous estimates, we obtain
$$\vert k_1\vert+\vert k_2\vert\leqslant \delta\vert h\vert\left(\vert\xi_1\vert+\vert \xi_2\vert\right)  \leqslant \delta\vert h\vert \left(\frac{1-\delta^\prime}{\delta^\prime}\vert k_1+k_2\vert+ \frac{1}{\delta^\prime}\vert \xi_1 + \xi_2\vert \right) .$$
Now we choose $\varepsilon^\prime$ small enough in such a way that $\forall \vert h\vert\leqslant \varepsilon^\prime$ $0<\delta\varepsilon^\prime\frac{1-\delta^\prime}{\delta^\prime} <1$.
Then this implies the final estimate
$$\forall \vert h\vert\leqslant \varepsilon^\prime,  \vert k_1+ k_2\vert\leqslant \frac{\delta\vert h\vert}{\delta^\prime} (1-\delta\varepsilon^\prime\frac{1-\delta^\prime}{\delta^\prime} )^{-1}\vert\xi_1+\xi_2\vert $$
\end{enumerate}
which means $\Gamma_1+\Gamma_2$ satisfies the \textbf{soft landing condition}.
\end{proof}

\section{The pull-back by diffeomorphisms.}
Our goal in this part consists in studying the lift to $T^\star M$ of diffeomorphisms of $M$ fixing $I$ since the symplectomorphisms of $T^\star M$ will determine the action on wave front sets. In this section, we will work in a local chart of $M$ in $\mathbb{R}^{n+d}$ with coordinates $(x,h)$ where $I$ is given by the equation $\{h=0\}$.
\subsection{The symplectic geometry of the vector fields tangent to $I$ and of the diffeomorphisms leaving $I$ invariant.}
We will work at the infinitesimal level 
within the class $\mathfrak{g}$ 
of vector fields tangent to $I$ 
defined by H\"ormander (\cite{Hormander} vol 3 Lemma (18.2.5)). First recall 
their definition
in coordinates $(x,h)$ where $I=\{h=0\}$:  
the vector fields $X$ tangent to $I$ are of the form
$$ h^ja_j^i(x,h)\partial_{h^i} + b^i(x,h)\partial_{x^i} $$
and they form an infinite dimensional Lie algebra denoted by $\mathfrak{g}$ which is a Lie subalgebra of $Vect(M)$. 
Actually, these vector fields form a module over the ring $C^\infty(M)$ finitely generated by the vector fields $h^i\partial_{h^j},\partial_{x^i}$. This module was defined by Melrose and is associated to a vector bundle called the Tangent Lie algebroid of $I$. 
This module is naturally filtered by the vanishing order of the vector field on $I$.
\begin{defi}
Let $\mathcal{I}$ be the ideal of functions vanishing on $I$. For $k\in\mathbb{N}$,
let $F_k$ be the submodule of vector fields tangent to $I$ defined as follows, $X\in F_k$ if 
$X\mathcal{I}\subset\mathcal{I}^{k+1}$.
\end{defi}
This definition of the filtration is completely coordinate invariant. We also immediately have $F_{k+1}\subset F_k$. Note that $F_0=\mathfrak{g}$.
\subsubsection{Cotangent lift of vector fields.} 
 We recall the following fact, 
any vector field $X\in Vect(M)$
lifts functorially to a
\emph{Hamiltonian vector field} $X^\star\in Vect(T^\star M)$
(for more on Hamiltonian vector fields, see \cite{CanaWeinstein} $3.5$ page $14$) by the following procedure  
$$X=X^i\frac{\partial}{\partial z^i}\in Vect(M) \overset{\sigma}{\mapsto} \sigma(X)=X^i\xi_i\in C^\infty(T^\star M)$$ 
$$ \mapsto X^\star=\left\lbrace \sigma(X),.  \right\rbrace= X^i\frac{\partial}{\partial z^i}-\xi_i\frac{\partial X^i}{\partial z^i}\frac{\partial}{\partial \xi_i},$$
where $\left\lbrace .,.  \right\rbrace $ is the Poisson bracket of $T^\star M$. 
Notice the projection on $M$ of $X^\star$ is $X$ and $X^\star$ is \textbf{linear} in the cotangent fibers.
This means the action of vector fields 
is lifted to an action by
Hamiltonian symplectomorphisms of $T^\star M$. 
The map $X\in\mathfrak{g}\mapsto \sigma(X)\in C^\infty(T^\star M)$ 
from the Lie algebra $\mathfrak{g}$ 
to the \emph{Poisson ideal} $\mathcal{I}_{(TI)^\perp}\subset C^\infty(T^\star M)$ 
can be interpreted as a
``universal'' \textbf{moment map} 
in Poisson geometry since to each element $X$ of the Lie algebra $\mathfrak{g}$ which acts symplectically as a vector field $X^\star\in Vect(T^\star M)$,
we associate a function which is the Hamiltonian of $X^\star$ 
(as explained to us by Mathieu Sti\'enon).
\begin{lemm}\label{lemmlift}
Let $X$ be a vector field in $\mathfrak{g}$.
Then $X^\star$ is tangent to the conormal $(TI)^\perp$ of $I$ and the symplectomorphism $e^{X^\star}$ leaves the conormal \textbf{globally invariant}. In particular,
%\item If $X\in F_1$ then the vector $X^\star$ fixes $T^\star M|_I$ and is vertical over $I$.
if $X\in F_2$, then $X^\star$ vanishes on the conormal $(TI)^\perp$ of $I$ and $(TI)^\perp$ is contained in the set of fixed points of the symplectomorphism $e^{X^\star}$.
\end{lemm}
\begin{proof} Any vector in $\mathfrak{g}$ admits the decomposition $ h^ja_j^i(x,h)\partial_{h^i} + b^i(x,h)\partial_{x^i} $. Thus the symbol map $\sigma(X)\in C^\infty(T^\star M)$ equals $ h^ja_j^i(x,h)\xi_i + b^i(x,h)k_i $. This function vanishes on the conormal bundle $(TI)^\perp$ which is a Lagrangian submanifold. Now we are reduced to the following problem:
given a function $f$ in a symplectic manifold which vanishes along a Lagrangian submanifold $C$, what can be said about the symplectic gradient $\nabla_\omega f$ along $C$ ?
%We give a quick argument due to Frederic Helein.
%By the Darboux Weinstein theorem, we can always locally 
%reduce by symplectomorphism to the case $\Lambda$ 
%is the zero section of a cotangent bundle. Then  
%we can conclude by a simple computation that 
%the vector field $\nabla_\omega f$ must be tangent
%to $\Lambda$. 
Since $f|_L=0$, for all $v\in TL$, $df(v)=0$. But $\forall v\in TL, 0=df(v)=\omega(\nabla_\omega f,v)$
which means that $\nabla_\omega f$ is in the orthogonal of $TL$ for the symplectic form
$\omega$. Since $L$ is a \textbf{Lagrangian} submanifold of $T^\star M$, this orthogonal is equal to $TL$, finally $\nabla_\omega f\in TL$.
If $X\in F_1$, then $\sigma(X)= h^jh^ia^l_{ji}(x,h)\xi_l + h^ib_{i}^l(x,h)k_l $ by the Hadamard lemma.
%Notice $\sigma(F_2)$ is a \emph{Poisson ideal} because $F_2$ is stable by Lie bracket, ie it is a Lie subalgebra of the Lie algebra of tangent vector fields. 
The symplectic gradient $X^\star$ is given by the formula
$$X^\star= \frac{\partial \sigma(X)}{\partial k_i}\partial_{x^i} - \frac{\partial \sigma(X)}{\partial x^i}\partial_{k_i}+\frac{\partial \sigma(X)}{\partial \xi_i}\partial_{h^i} - \frac{\partial \sigma(X)}{\partial h^i}\partial_{\xi_i},$$ 
%$$= -h^i\partial x^i(h^ja^l_{ji}(x,h)\xi_l - b_{i}^l(x,h)k_l)\partial_{k_i}+h^ib_{i}^l(x,h)\partial_{x^l} $$
thus $X^\star=0$ when $k=0,h=0$ which means $X^\star=0$ on the conormal $(TI)^\perp$ of $I$.
\end{proof}
\begin{prop}\label{Idlift}
Let $\rho_1,\rho_2$ be two 
Euler vector fields and $\Phi(\lambda)=e^{-\log\lambda \rho_1}\circ e^{\log\lambda \rho_2}$.
Then the cotangent lift $T^\star \Phi(\lambda)$
restricted to $(TI)^\perp$ is the identity map:
$$T^\star \Phi(\lambda)|_{(TI)^\perp}=Id|_{(TI)^\perp}.$$
In particular, the diffeomorphism 
$\Psi=\Phi(0)$ (Corollary \ref{conjugcoro}) 
which conjugates $\rho_1$ with
$\rho_2$ satisfies the same property.  
\end{prop}
\begin{proof}
Let us set
\begin{equation}\label{conjugchapt1}
\Phi(\lambda)=e^{-\log\lambda \rho_1}\circ e^{\log\lambda \rho_2} 
\end{equation}
which is a family of diffeomorphisms
which depends smoothly in $\lambda\in[0,1]$
according to \ref{propositionvariablefamily},
then $\Phi(0)$ 
is the diffeomorphism
which locally conjugates 
$\rho_1$ and $\rho_2$ (Corollary \ref{conjugcoro}).
The proof is similar to the proof of proposition \ref{propositionvariablefamily},
$\Phi(\lambda)$ satisfies the differential equation:
\begin{equation}
\lambda\frac{d\Phi(\lambda)}{d\lambda}=e^{-\log\lambda\rho_1}\left(\rho_2-\rho_1\right)e^{\log\lambda\rho_1}\Phi(\lambda) \text{ where } \Phi(1)=Id
\end{equation}
we reformulated this differential equation as
\begin{equation}
\frac{d\Phi(\lambda)}{d\lambda}=X(\lambda)\Phi(\lambda),\Phi(1)=Id
\end{equation}
where the vector field $X(\lambda)=\frac{1}{\lambda}e^{-\log\lambda\rho_1}\left(\rho_2-\rho_1\right)e^{\log\lambda\rho_1}$ depends smoothly in $\lambda\in [0,1]$.
The cotangent lift $T^\star \Phi_\lambda$ satisfies the differential equation
\begin{equation}
\frac{dT^\star\Phi(\lambda)}{d\lambda}=X^\star(\lambda)T^\star\Phi(\lambda),T^\star \Phi(1)=Id
\end{equation} 
Notice that $\forall \lambda\in[0,1], X(\lambda)\in F_1$ which implies that for all $\lambda$ the lifted Hamiltonian vector field $X^\star(\lambda)$ will vanish on $(TI)^\perp$ by the lemma (\ref{lemmlift}). Since $T^\star \Phi(1)=Id$ obviously fixes the conormal, this immediately implies that $\forall\lambda, T^\star\Phi(\lambda)|_{(TI)^\perp}=Id|_{(TI)^\perp}$.
\end{proof}
\subsection{The pull-back is bounded.}
\paragraph{The problem we solve.}
We start from a distribution $t\in \mathcal{D}^\prime(M\setminus I)$ 
such that $WF(t)$
satisfies the soft landing condition. 
We assumed that 
there exists a generalized Euler 
$\rho_1$ and a small neighborhood $\mathcal{V}$ of $I$
such that $\lambda^{-s}e^{-\log\lambda\rho_1*}t $ is bounded in $\mathcal{D}^\prime_\Gamma(\mathcal{V}\setminus I)$ where $\Gamma=\bigcup_{\lambda\in(0,1]}WF(e^{\log\lambda\rho_1\star}t)$. 
%We want to stress the fact that the boudedness hypothesis is a property of $t$ and does not only depend on $WF(t)$. For instance, this property should be checked for the Wightman function $\Delta_+$ and 
%this will be done in Chapter $5$.
Under these conditions, by the main theorem of Chapter $3$, we know that 
the extension $\overline{t}$ is well defined, $WF(\overline{t})\subset \overline{WF(t)}\cup C$ and for every $s^\prime < s$, 
$\lambda^{-s^\prime}e^{\log\lambda\rho_1}\overline{t} $ is bounded in $\mathcal{D}^\prime_{\overline{\Gamma}\cup C}(\mathcal{V})$.
% avec ce champ d'Euler pour
%fabriquer un cône \Gamma défini sur M.
% Du coup, si t est dans un E_s,
%son prolongement à M existe et le front d'onde de
%ce prolongement est dans l'union de la fermeture de WF(t)
%et du conormal de I.
We proved (Proposition \ref{propositionvariablefamily} 
Chapter $1$) that 
when we change 
the Euler vector field from 
$\rho_1$ to $\rho_2$, 
we have:
$$\lambda^{-s}e^{\log\lambda\rho_2*}t=\Phi(\lambda)^\star\underset{\text{bounded in }\mathcal{D}^\prime_{\Gamma_1}}{\underbrace{\left(\lambda^{-s}e^{\log\lambda\rho_1*}t\right)}}.$$ 
The above equation
motivates us to study
a more general question,
is the image of a bounded 
set
in $\mathcal{D}^\prime_\Gamma$
by a diffeomorphism $\Phi$
still
a bounded family in 
$\mathcal{D}^\prime_{\Phi^\star\Gamma}$?

\subsection{The action of Fourier integral operators.}
Fourier integral operators are abbreviated FIO.
In this section, we will work 
exclusively in $\mathbb{R}^d$ 
since 
our problem is local.
To solve our problem, we will have to revisit a deep theorem of H\"ormander (see \cite{Hormander} theorem $8.2.4$) which describes the wave front set of distributions under pull back.
However, we will reprove a variant of this theorem which is tailored for applications in QFT.
First, we prove the theorem for a specific subclass of FIO (as discussed in \cite{Eskin}) which contains the space of diffeomorphisms and we also 
give
explicit bounds for
the
seminorms of $\mathcal{D}^\prime_\Gamma$. 
We deliberately choose to discuss everything in the language of canonical relations and symplectomorphisms since these are at the core of the geometric ideas involved in the proof.
\subsubsection{A quick reminder about the formalism of FIO.}
 We recall the definition of a specific class of FIO following \cite{Eskin}. And we will frequently use several notions that can be found in \cite{Eskin}.
\paragraph{The definition of Eskin's FIO.}
We adapt the definition of \cite{Eskin} to our context, 
we consider operators 
of the form: 
$$U: \mathcal{D}(\mathbb{R}^d)\times \mathcal{D}^\prime(\mathbb{R}^d) \mapsto \mathcal{D}^\prime(\mathbb{R}^d)$$
\begin{equation}\label{Eskinform}
(\varphi,t)\mapsto U_\varphi t=\frac{1}{(2\pi)^d}\int_{\mathbb{R}^d}d\eta e^{iS(x,\eta)}a(x,\eta)\widehat{t\varphi}(\eta) 
\end{equation}
where $S$
is smooth, homogeneous of degree $1$ in $\eta$
and $\det \frac{\partial^2 S}{\partial x\partial \eta}\neq 0$, 
we do not assume $a=0$ if $\vert\eta\vert<1$ since for diffeomorphisms
$a=1$, and this 
does only change the FIO modulo smoothing operator (see \cite{Eskin} p.~330).
The Schwartz kernel of $U_\varphi$ is
the Fourier distribution which 
by a slight abuse of notation reads:
$$U_\varphi(x,y)=
\frac{1}{(2\pi)^d}\int_{\mathbb{R}^d}d\eta e^{iS(x,\eta)-iy.\eta}a(x,\eta)\varphi(y).$$
See \cite{Eskin} p.~341.
\begin{lemm}
Let $\Phi$
be a diffeomorphism of $\mathbb{R}^d$ and 
$\varphi\in  \mathcal{D}(\mathbb{R}^d)$. Then there exists an operator $U_\varphi$ as in \ref{Eskinform} such that $\forall t\in \mathcal{D}^\prime(\mathbb{R}^d)$, 
$U_\varphi(t)=\Phi^\star(t\varphi)$. 
\end{lemm}
We will later choose $\varphi$ as an element of an ad hoc partition of unity defined by the approximation lemmas (\ref{approxlemma1},\ref{approxlemma2}).
\begin{proof}
Our proof follows the strategy outlined in \cite{DuistermaatFIO} proposition $(1.3.3)$. 
The idea is to write down $t\varphi$ as the inverse Fourier transform of $\widehat{t\varphi}$. 
$$t\varphi=\mathcal{F}^{-1}\left( \widehat{t\varphi} \right)=\frac{1}{(2\pi)^{d}}\int_{\mathbb{R}^{d}} d\eta e^{ix.\eta}\widehat{t\varphi}(\eta) $$ 
Now, we pull-back $t\varphi$ 
by the diffeomorphism
$\Phi$ : 
$$\Phi^*\left(t\varphi\right)(x)= \Phi^*\mathcal{F}^{-1}\left(\widehat{t\varphi}\right)(x)=\frac{1}{(2\pi)^{d}}\int_{\mathbb{R}^{d}}d\eta e^{i \Phi(x).\eta} \widehat{t\varphi}(\eta) $$
Now setting $S(x;\eta)=\Phi(x).\eta$, we recognize the phase function $S$ appearing in  (\ref{Eskinform}).
\end{proof}
In the following, given 
a generating function
$S$, we denote by $\sigma$
the canonical transformation
defined by:
\begin{equation}\label{canorel}
\\ \sigma:(y;\eta)\mapsto (x;\xi) ,\xi=\frac{\partial S}{\partial x}(x,\eta),y=\frac{\partial S}{\partial \eta}(x,\eta),
\end{equation}
see Equation (61.2) p.~330 in \cite{Eskin}.
\begin{thm}\label{Fouriervariablephase}
Let $(t_\mu)_\mu$ be bounded in 
$\mathcal{D}^\prime_\Gamma(\Omega),\Omega\subset \mathbb{R}^d$. 
Let $U$
be a \textbf{proper} operator
as defined in (\ref{Eskinform}) 
with amplitude $a=1$ 
and generating function
$S$ and $\sigma$
the corresponding canonical relation.
Then $\left(U t_\mu\right)_{\mu}$ 
is bounded 
in $D_{\sigma\circ\Gamma}^\prime(\Omega)$. 
\end{thm}
We will decompose the proof of the theorem in many different lemmas.
Our strategy goes as follows, we have some bounds on $\widehat{t\varphi}$ where $\varphi\in  \mathcal{D}(\mathbb{R}^d)$ because we know that $t\in\mathcal{D}^\prime_\Gamma$ by the hypothesis of the theorem and we want to deduce from these bounds some estimates on the Fourier transform $\mathcal{F}\left(\chi U\left(t\varphi\right)\right)$. We first prove a lemma which
gives an estimate of $WF(U\left(t\varphi\right))$.
\begin{lemm}
Let $U$
be a \textbf{proper} operator
as defined in (\ref{Eskinform}) 
with amplitude $a=1$ 
and generating function
$S$, $\sigma$
the corresponding canonical transformation
and 
$\varphi\in  \mathcal{D}(\mathbb{R}^d)$. 
Then for all $t\in \mathcal{D}^\prime_\Gamma$, 
$WF(U_\varphi t)\subset\sigma\circ\Gamma$.
\end{lemm}
\begin{proof}
We denote by $(y;\eta)$ and $(x;\xi)$ the coordinates in $T^\star\mathbb{R}^d$.
Let $t$ be a distribution and $U$ a FIO of the form (\ref{Eskinform}) with phase function $S(x;\eta)-\left\langle y,\eta \right\rangle$. 
Then Theorem 63.1 in Eskin (see \cite{Eskin} p.~340) expresses $WF(U_\varphi t)$ in terms of
the image $\sigma\circ WF(t\varphi)$ of $WF(t\varphi)$ by the canonical relation $\sigma$ generated 
by $S$.
To apply the theorem
of Eskin,
we use the fact 
that
$t\varphi$ compactly supported 
$$\implies\Vert\theta^{-m}\widehat{t\varphi} \Vert_{L^\infty} < +\infty\implies \theta^{-m-\frac{d+1}{2}}\widehat{t\varphi}\in L^2(\mathbb{R}^d)
\Leftrightarrow \widehat{t\varphi}\in H^{-m-\frac{d+1}{2}}.$$
% For the image $Ut$ of a distribution $t$ by a FIO $U$ , we know how to locate $WF(Ut)$. It is by definition the image of $WF(t)$ by the canonical relation $\sigma_\lambda$ defined by the generating function $S_\lambda$ (see Eskin $61.3$).
\begin{equation} 
U_\varphi t(x)=\frac{1}{(2\pi)^{d}}\int_{\mathbb{R}^{2d}}dyd\eta e^{i[S(x;\eta)- y.\eta]} t\varphi(y)
\end{equation}
\begin{equation}
\\ \sigma:(y;\eta)\mapsto (x;\xi) ,\xi=\frac{\partial S}{\partial x}(x,\eta),y=\frac{\partial S}{\partial \eta}(x,\eta).
\end{equation}
The canonical transformation is the same as equation 61.2 p.~330 in \cite{Eskin}.
For convenience, we will write in local coordinates $\sigma(y,\eta)=(x(y,\eta),\xi(y,\eta))$.
In the particular case of a diffeomorphism $x\mapsto \Phi(x)$,
$$\frac{\partial S}{\partial \eta}(x,\eta)=\Phi(x),\frac{\partial S}{\partial x}(x,\eta)=\eta\circ d\Phi  $$   
and the corresponding family of canonical relations is 
\begin{equation} 
\sigma:(y,\eta)\mapsto (\Phi^{-1}(y),\eta\circ d\Phi). 
\end{equation} 
\end{proof}
 Motivated by this result, we will test $\Phi^*(t\varphi)$ on seminorms $\Vert.\Vert_{N,V,\chi}$, for a cone $V$ and test function $\chi$ such that $\text{supp }\chi\times V$ does not meet  $\sigma\circ\Gamma$.
\begin{lemm}
Let $U$ be given by \ref{Eskinform}, $\sigma$ the corresponding
canonical relation, $m$ a nonnegative integer, $\alpha\in C^\infty(\mathbb{R}^{d}\setminus{0})$, homogeneous of degree $0$, $\varphi\in \mathcal{D}(\mathbb{R}^d)$, $\chi\in  \mathcal{D}(\mathbb{R}^d)$ and $V\subset(\mathbb{R}^d\setminus{0})$ a closed cone.
If $\left(\text{supp }\chi\times V \right)\bigcap \sigma\circ\left(\text{supp }\varphi\times \text{supp }\alpha \right)=\emptyset$ and $\left(\text{supp }\varphi\times \text{supp }(1-\alpha)\right)\bigcap \Gamma=\emptyset$
then for all $N$, there exists
$C_N$ s.t. for all $t\in \mathcal{D}^\prime_\Gamma $ satisfying $\Vert \theta^{-m} \widehat{t\varphi_j}\Vert_{L^\infty}<+\infty$: 
\begin{equation}\label{est3}
 \Vert U\left(t\varphi\right)\Vert_{N,V,\chi}\leqslant C_N(1+\vert\xi\vert)^{-N} \left(\Vert \theta^{-m} \widehat{t\varphi}\Vert_{L^\infty} +
\Vert t\Vert_{N+d+1,W,\varphi}\right) 
\end{equation}
where $W=supp(1-\alpha)$.
\end{lemm} 
\begin{proof} 
Our method of proof is based on the method of stationary phase and a geometric interpretation. In the course of our proof, we will explain why constants appearing in all our estimates do not depend on $t$ but only on $U$ and $\Gamma$. 
This is the only way to obtain an estimate which is valid for families $(t_\mu)_\mu$ bounded in
$\mathcal{D}^\prime_\Gamma$.
In order to bound $\Vert U(t\varphi)\Vert_{N,V,\chi} $, we must first compute the Fourier transform of $\chi U(t\varphi)$:
\begin{equation}
\mathcal{F}\left(\chi U(t\varphi) \right)(\xi)=\frac{1}{(2\pi)^d}\int_{\mathbb{R}^{2d}}dxd\eta \chi(x) e^{i[S(x;\eta)-x.\xi]} \widehat{t\varphi}(\eta)
\end{equation} 
We then extract the $\emph{oscillatory integral}$ on which we will apply the method of stationary phase:
$$I(\xi,\eta)= \int_{\mathbb{R}^{d}}dx e^{i[S(x;\eta)- x.\xi]}\chi(x)=\int_{\mathbb{R}^{d}}dx  e^{i\psi(x,\xi,\eta)}\chi(x) ,$$    
where the phase $\psi(x,\xi,\eta)=[S(x;\eta)-x.\xi ]$.  
%\begin{equation}
%I_\lambda(r,\theta_1,\theta_2)= \int_{\mathbb{R}^{n}}dx  e^{ir[S_\lambda(x;\theta_1)-\left\langle x,\theta_2\right\rangle ]}\chi(x)=\int_{\mathbb{R}^{n}}dx  e^{ir\psi_\lambda(x,\theta_1,\theta_2)}\chi(x)
%\end{equation}
We reformulate the expression giving $\mathcal{F}\left(\chi U\left(t\varphi\right)\right)(\xi)$ in terms of the oscillatory integral $I(\xi,\eta)$:
$$\mathcal{F}\left(\chi U\left(t\varphi\right)\right)(\xi)=\int_{\mathbb{R}^{d}}d\eta  I(\xi,\eta) \widehat{t\varphi}(\eta) .$$ 
Then the idea is to split the integral in two parts, in one part the oscillatory integral $I(\xi,\eta)$ behaves nicely and decreases fastly at infinity, ie $\forall N, (1+\vert\xi\vert+\vert\eta\vert)^NI(\xi,\eta)$ is bounded. 
In the second part, the oscillatory integral is bounded but this domain corresponds to the codirections in which $\widehat{t\varphi}$ has fast decrease at infinity.
%$$\int_{\mathbb{R}^+\times \mathbb{S}^{n-1}} d\theta_2  I_\lambda(r,\theta_1,\theta_2)\widehat{t\varphi_j}(r\theta_2) $$
The method of stationary phase states (see \cite{Stein2} p.~330,341) that the integral $I$ is rapidly decreasing in the codirections $(\xi,\eta)$ for which $\psi$ is \textbf{noncritical}, i.e. $d_x\psi(x;\xi,\eta)\neq 0$. We compute the critical set of the phase
$$d_x\psi(x;\xi,\eta)=d_xS(x,\eta)-\xi.$$
Hence the critical set $d_x\psi=0$ is given by the equations 
\begin{equation}
\{(\eta,\xi)\vert d_xS(x,\eta)-\xi = 0,x\in\text{supp }\chi\}, 
\end{equation}
we thus naively set
\begin{equation}\label{sigmaxi}
\forall\xi,\, \Sigma(\xi):=\{(y,\eta)| \exists x\in \text{supp }\chi, d_xS(x,\eta)-\xi = 0 , y=\frac{\partial S}{\partial\eta}(x,\eta) \}.
\end{equation}
Motivated by the geometric relation
between 
the 
generating function $S$ 
and the canonical relation $\sigma$ 
(by Equation (\ref{canorel})), 
we interpret $\Sigma(\xi)$ 
in terms of the
canonical transformation $\sigma$:
\begin{eqnarray}
\Sigma(\xi)=\{(y,\eta)\vert \exists x\in \text{supp }\chi,  \sigma(y,\eta)=(x,\xi)\}\\
\text{or }\quad\Sigma(\xi)=\sigma^{-1}\circ\left(\text{supp }\chi\times \{\xi\}\right).  
\end{eqnarray} 
Hence $\Sigma(\xi)$ is the inverse image of $\text{supp }\chi\times \{\xi\}$ by the canonical relation $\sigma$.
Let us recall that $\pi_2$ projects $T^\bullet\mathbb{R}^d$ on the second factor $\mathbb{R}^{d\star}$. We define $$R(\xi)=\pi_2\left(\Sigma(\xi)\right)=\{\eta| \exists x\in \text{supp }\chi, d_xS(x,\eta)-\xi = 0  \}$$
which has the following analytic interpretation, for fixed $\xi$, 
$R(\xi)$ contains the critical set (``bad $\eta$'s'')  
of $I(\xi,\eta)$.
We \textbf{admit temporarily} that
$$\sigma\circ\left(\text{supp }\varphi\times \text{supp }\alpha\right)\bigcap\left(\text{supp }\chi\times V \right)=\emptyset$$
implies $\text{supp }\alpha$ does not meet
$\bigcup_{\xi\in V}R(\xi)$ (we will prove this claim in Lemma (\ref{Retsigma})).
We are led to define a neighborhood $R_\varepsilon(\xi)$ of $R(\xi)$
for which $\forall\xi\in V, R_\varepsilon(\xi)\cap \text{supp }\alpha=\emptyset$:
$$R_\varepsilon(\xi)=\{\eta| \exists x\in \text{supp }\chi, \vert d_xS(x,\eta)-\xi\vert \leqslant\varepsilon  \} .$$
%
% We now give a geometric interpretation of $R$. We use the fact that the generating function $S_\lambda$ parametrizes the canonical transformation $\sigma_\lambda=(y_\lambda,\xi_\lambda)$, hence we consider the image of the set $\text{supp }\chi\times V$ in cotangent space under the family of canonical transformations $\sigma_\lambda,\lambda \in [0,\delta]$:
%$$R=\{\eta\vert \exists (x,\xi,\lambda)\in \text{supp }\chi\times V\times[0,\delta],  \sigma_\lambda(y,\eta)=(x,\xi)  \} $$ 
%and if we define $V_\varepsilon=\{\xi |\exists \xi_0\in V, \vert\xi-\xi_0\vert\leqslant\varepsilon \}$ $$R_\varepsilon=\{ \eta \vert (x,\xi,\lambda)\in \text{supp }\chi\times V_\varepsilon\times[0,\delta], \sigma_\lambda(y,\eta)=(x,\xi) \} $$
% 
% 
Denote by $R_\varepsilon^c(\xi)$ the complement of $R_\varepsilon(\xi)$. 
$$R_\varepsilon^c(\xi)=\{\eta \vert \forall (x,\xi) \in \text{supp }\chi\times V ,\vert d_xS(x;\eta)-\xi\vert> \varepsilon \}$$ 
$$R_\varepsilon^c(\xi)=\{\eta \vert \forall (x,\xi) \in \text{supp }\chi\times V  ,\vert d_x\psi(\xi,\eta)\vert> \varepsilon \} .$$
We use the following result in Duistermaat, $\forall N, \exists C_N$ s.t.
\begin{equation}\label{Statphase}
\forall (\xi,\eta)\in V\times R^c_\varepsilon(\xi),\, \vert I(\xi,\eta)\vert\leqslant C_N\left(1+\vert \eta\vert+\vert\xi\vert \right)^{-N}.
\end{equation} 
The proof of this result is based on the fact that
we are away from the critical set $R(\xi)$ and from application of the stationary phase (\cite{DuistermaatFIO} Proposition 2.1.1 p.~11). The constant $C_N$ depends only on $N$, $\chi$, $S$, $\varepsilon$.

 Recall we made the assumption there is a function $\alpha\in C^\infty(\mathbb{R}^{n}\setminus{0})$, homogeneous of degree $0$ such that $\forall\xi\in V, R_\varepsilon(\xi)$ does not meet $\text{supp }\alpha$, and $\text{supp }\varphi\times \text{supp }(1-\alpha)$ does not meet $\Gamma$.
We cut the Fourier transform in two pieces:
$$I(\xi)=\mathcal{F}\left(\chi U\left(t_\mu\varphi_j\right)\right)(\xi)=I_1+I_2$$
where 
\begin{eqnarray}\label{I1I2} 
I_1(\xi)=\int_{R_\varepsilon(\xi)}d\eta I(\xi,\eta) \widehat{t\varphi}(\eta)\\
I_2(\xi)=\int_{R_\varepsilon^c(\xi)}d\eta I(\xi,\eta) \widehat{t\varphi}(\eta).
\end{eqnarray} 
Observe $I_1(\xi)=\int_{R_\varepsilon(\xi)}d\eta I(\xi,\eta) \alpha\widehat{t\varphi}(\eta)+\int_{R_\varepsilon(\xi)}d\eta I(\xi,\eta) (1-\alpha)\widehat{t\varphi}(\eta)=\int_{R_\varepsilon(\xi)}d\eta I(\xi,\eta) (1-\alpha)\widehat{t\varphi}(\eta)$
since we assumed $\forall\xi\in V,\text{supp }\alpha\cap R_\varepsilon(\xi)=\emptyset$.
By Paley--Wiener theorem, we know that $\exists m, \Vert \theta^{-m} \widehat{t\varphi}\Vert_{L^\infty}<\infty $. We use this inequality and stationary phase estimate (\ref{Statphase})
$$\vert I_2\vert(\xi)=\vert\int_{R_\varepsilon^c(\xi)}d\eta I(\xi,\eta) \widehat{t\varphi}(\eta)\vert
\leqslant C_{N+m+d+1}\int_{R_\varepsilon^c(\xi)}d\eta (1+\vert\eta\vert+\vert\xi\vert)^{-N-m-d-1}\vert \widehat{t\varphi}(\eta)\vert $$
$$\leqslant  C_{N+m+d+1}\int_{R_\varepsilon^c(\xi)}d\eta (1+\vert\eta\vert+\vert\xi\vert)^{-N-m-d-1}(1+\vert\eta\vert)^{m}\Vert \theta^{-m} \widehat{t\varphi}\Vert_{L^\infty} $$
$$\leqslant C_{N+m+d+1}(1+\vert\xi\vert)^{-N} \Vert \theta^{-m} \widehat{t\varphi}\Vert_{L^\infty}\int_{\mathbb{R}^d}d\eta(1+\vert\eta\vert)^{-d-1}$$ 
hence $I_2(\xi)\leqslant C_{N+m+d+1}^\prime(1+\vert\xi\vert)^{-N} \Vert \theta^{-m} \widehat{t\varphi_j}\Vert_{L^\infty} $ where $C_{N+m+d+1}^\prime$ is a constant which depends only on $N$, $\chi$, $S$, $\varepsilon$.
Now to estimate $I_1$, set $W:=\text{supp }(1-\alpha)$:
$$I_1(\xi)=\int_{R_\varepsilon(\xi)}d\eta I(\xi,\eta) (1-\alpha)\widehat{t\varphi}(\eta)$$ 
by a change of variable in (\ref{I1I2}) so that $\eta$ does appear on the right hand side, 
$$\vert I_1(\xi)\vert\leqslant \int_{\mathbb{R}^d}dx \vert\chi(x)\vert \int_{R_\varepsilon(\xi)}d\eta \vert(1-\alpha)\widehat{t\varphi}(\eta)\vert $$
because $\vert I(\xi,\eta)\vert\leqslant  \int_{\mathbb{R}^d}dx\vert\chi(x)\vert$,
$$\vert I_1(\xi)\vert\leqslant \int_{\mathbb{R}^d}dx \vert\chi(x)\vert \int_{R_\varepsilon(\xi)}d\eta \Vert t\Vert_{N,W,\varphi}\left(1+\vert\eta \vert\right)^{-N}. $$ 
%
% But this estimate is not good enough since we don't have a good a priori control on the asymptotic behaviour of $I_1$ when $\xi\rightarrow\infty$. 
% 
% Somehow, it is not obvious to obtain the decrease in $(1+\vert\xi\vert)^{-N}$.
Recall the definition of $R_\varepsilon(\xi)=\{\eta| \exists x\in \text{supp }\chi, \vert d_xS(x,\eta)-\xi\vert \leqslant\varepsilon  \} $.
The defining inequality $\vert d_xS(x,\eta)-\xi\vert \leqslant\varepsilon $ implies that on $R_\varepsilon(\xi)$: 
$$\vert d_xS(x;\eta)-\xi\vert\leqslant \varepsilon\implies \vert\xi\vert-\varepsilon \leqslant \vert d_xS(x;\eta)\vert\leqslant \vert\xi\vert+\varepsilon.$$ 
This estimate is relevant if $\vert\xi\vert >\varepsilon$.
Then we use the fact that $\eta\mapsto  d_xS(x,\eta)$ does not meet the zero section when $\eta\neq 0$ and depends smoothly on $x\in \text{supp }\chi$ (in the case of a diffeomorphism, we find $d_xS(x,\eta)=\eta\circ d\Phi(x)$),
so there is a constant $c>0$ such that
\begin{equation}\label{ineqsuite}
\forall (x,\eta)\in \text{supp }\chi\times \mathbb{R}^{d},c^{-1}\vert\eta\vert\leqslant\vert  d_xS(x,\eta)\vert\leqslant c\vert\eta\vert .
\end{equation} 
Combining with the previous estimate gives $\forall\xi\in V,\forall\eta\in R_\varepsilon(\xi), \vert\xi\vert-\varepsilon \leqslant c\vert\eta\vert$ which
%and $c^{-1}\vert\xi\vert\leqslant \vert\eta\vert+\varepsilon $.
%$$\implies  c^{-1}(\vert\eta\vert-\varepsilon) \leqslant\vert\xi\vert\leqslant c(\vert\eta\vert+\varepsilon)  $$
% This means for $\vert\eta\vert\geqslant r+\varepsilon$ for a given $r>0$,
can be translated as the inclusion of sets 
\begin{equation}
R_\varepsilon(\xi)\subset \{c^{-1}\left(\vert\xi\vert-\varepsilon\right) \leqslant \vert\eta\vert\}=\mathbb{R}^d\setminus B\left(0,\frac{\vert\xi\vert-\varepsilon}{c}\right)
\end{equation}
%$\forall \vert\eta\vert\geqslant r+\varepsilon$
$$I_{1}(\xi)\leqslant  \int_{\mathbb{R}^d}dx \vert\chi(x)\vert \int_{c^{-1}\left(\vert\xi\vert-\varepsilon\right) \leqslant \vert\eta\vert}d\eta \Vert t\Vert_{N+d+1,W,\varphi}\left(1+\vert\eta\vert\right)^{-N-d-1}$$
$$=\frac{2\pi^{d/2}}{\Gamma(d/2)}\left(\int_{\mathbb{R}^d}dx \vert\chi(x)\vert\right)\Vert t\Vert_{N+d+1,W,\varphi} \int_{c^{-1}\left(\vert\xi\vert-\varepsilon\right)}^\infty \left(1+r \right)^{-N-d-1}r^{d-1}dr  $$
$$\leqslant \frac{2\pi^{d/2}}{\Gamma(d/2)}\left(\int_{\mathbb{R}^d}dx \vert\chi(x)\vert\right)\Vert t\Vert_{N+d+1,W,\varphi} \int_{c^{-1}\left(\vert\xi\vert-\varepsilon\right)}^\infty r^{-N-2}dr$$ 
$$=\frac{2\pi^{d/2}}{\Gamma(d/2)}\left(\int_{\mathbb{R}^d}dx \vert\chi(x)\vert\right)\Vert t\Vert_{N+d+1,W,\varphi} \frac{\left(c^{-1}\left(\vert\xi\vert-\varepsilon\right) \right)^{-N-1}}{N+1} $$
 $$\leqslant C_{N+1} \Vert t\Vert_{N+d+1,W,\varphi}(1+\vert\xi\vert)^{-N-1}.$$ 
where $C_{N+1}$ does not depend on $t$ but only on $\Gamma$. 
\end{proof}
In the previous lemma, we made two assumptions that we are going to prove,
we recall some useful definitions:
$$\forall\xi\in V, \Sigma(\xi)=\sigma^{-1}\circ\left(\text{supp }\chi\times \{\xi\}\right) , R(\xi)=\pi_2(\Sigma(\xi))$$ and $R_\varepsilon(\xi)$ is a family
of neighborhoods of $R(\xi)$ which tends to $R(\xi)$ as $\varepsilon\rightarrow 0$.
\begin{lemm}\label{Retsigma}
For any closed conic set $V$ and $\chi\in  \mathcal{D}(\mathbb{R}^d)$ such that $\left(\text{supp }\chi \times V \right)\cap \left( \sigma\circ \Gamma\right)=\emptyset$, 
there exists a pseudodifferential partition of unity $(\alpha_j,\varphi_j)_j$ such that 
\begin{eqnarray}\label{systemapprox}
\forall\xi\in V, R_\varepsilon(\xi)\cap\text{supp }\alpha_j=\emptyset\\
\Gamma\subset\bigcup_{j\in J}\text{supp }\varphi_j\times \text{supp }\alpha_j.
\end{eqnarray} 
\end{lemm}
\begin{proof}
$\chi$ and $V$ are given in such a way that 
$$ \left(\text{supp }\chi \times V \right)\cap \left( \sigma\circ \Gamma\right)=\emptyset\underset{\sigma\text{ diffeo}}{\Leftrightarrow} 
\sigma^{-1}\left(\text{supp }\chi\times V \right)\cap \Gamma=\emptyset .$$ 
We then use Lemma \ref{approxlemma1}, \ref{approxlemma2} 
to cover $\Gamma$
by $\left(\bigcup_j\text{supp }\varphi_j\times \text{supp }\alpha_j\right) $ where $\alpha_j\in C^\infty(\mathbb{R}^d\setminus \{0\})$ is homogeneous of degree $0$ and we choose the cover fine enough
in such a way that
$$\left( \sigma^{-1}\circ\left(\text{supp }\chi \times V\right) \right)\cap \left(\bigcup_j\text{supp }\varphi_j\times \text{supp }\alpha_j\right)=\emptyset.$$
But this implies
$$\forall j, \left(\bigcup_{\xi\in V}\sigma^{-1}\left(\text{supp }\chi \times \{\xi\}\right) \right)\cap \left(\text{supp }\varphi_j\times \text{supp }\alpha_j\right)=\emptyset$$ 
$$\Leftrightarrow \left(\bigcup_{\xi\in V}\Sigma(\xi) \right)\cap \left(\text{supp }\varphi_j\times \text{supp }\alpha_j\right)=\emptyset\implies \left(\bigcup_{\xi\in V}R(\xi)\right)\cap \text{supp }\alpha_j=\emptyset,$$
the last line follows by projecting with $\pi_2$. Finally by choosing $\varepsilon$ small enough, we can always assume
$\forall\xi\in V, R_\varepsilon(\xi)\cap \text{supp }\alpha_j=\emptyset$:
assume the converse holds, 
i.e. $\forall n$, $\exists \xi_n\in V$, 
$\exists x_n \in\text{supp }\chi$, $\exists \eta_n\in R_{\frac{1}{n}}(\xi_n)\cap \text{supp }\alpha_j $
w.l.g. assume $\vert\eta_n\vert=1$ 
then by definition
of $R_{\frac{1}{n}}(\xi_n)$, we find that
$$\vert \xi_n-\frac{\partial S}{\partial x}(x_n,\eta_n)\vert<\frac{1}{n} $$
and estimate (\ref{ineqsuite}) $\implies \vert  d_xS(x_n,\eta_n)\vert\leqslant c\vert\eta_n\vert=c\implies 
\vert \xi_n\vert<c+\frac{1}{n}.$
This means 
the sequence $(x_n,\xi_n,\eta_n)$
lives in a compact set, thus
we can extract a subsequence 
which converges
to $(x,\xi,\eta)\in \text{supp }\chi\times V \times  \text{supp }\alpha_j$
and $\eta\in R(\xi)\cap \text{supp }\alpha_j$, contradiction !  
%We reduce the computation to one term of this pseudodifferential partition of unity 111$\chi \Phi_\lambda^*(t\varphi_i)=\chi\mathcal{F}^{-1}(\alpha_i \widehat{t\varphi_i})+\chi\mathcal{F}^{-1}((1-\alpha_i) \widehat{t\varphi_i})$. $\text{supp }\varphi_i\times\text{supp }(1-\alpha_i)$ is an approximation of $\Gamma$ over $\text{supp }\varphi_i$, $$\Gamma|_{\text{supp }\varphi_i}\subset \text{supp }\varphi_i\times\text{supp }(1-\alpha_i)$$
%
%
% Notice the following crucial fact, if (\ref{equationavoiding}) is satisfied then for $\varepsilon$ small enough, $R_\varepsilon\cap \text{supp }\varphi_j\times\text{supp }\alpha_j=\emptyset$.
% 
\end{proof}
Then we give the final lemma which concludes the proof of theorem (\ref{Fouriervariablephase}).
\begin{lemm}
Let $U$ be an operator given in (\ref{Eskinform})
with symbol $a=1$
and $\sigma$ the corresponding canonical transformation.
For any closed conic set $V$ and $\chi\in  \mathcal{D}(\mathbb{R}^d)$ such that $\left(\text{supp }\chi \times V \right)\cap \left(\sigma\circ \Gamma\right)=\emptyset$, there exists a 
finite family of seminorms $(\Vert.\Vert_{N,W_j,\varphi_j})_{j\in J^\prime}$ 
for $\mathcal{D}^\prime_\Gamma$ such that 
$\forall N,\exists C_N, \forall t\in \mathcal{D}^\prime_\Gamma$ s.t. $\forall j\in J^\prime, \Vert \theta^{-m} \widehat{t\varphi_j}\Vert_{L^\infty}<+\infty$:
$$\Vert Ut \Vert_{N,V,\chi} \leqslant
\underset{j\in J^\prime}{\sum} C_N \left(\Vert \theta^{-m} \widehat{t\varphi_j}\Vert_{L^\infty} +  \Vert t\Vert_{N+2d+1,W_j,\varphi_j}\right).$$
\end{lemm}
\begin{proof}
There is still a problem 
due to the
noncompactness of the support of $t$, 
there is no reason the sum $\sum_{j\in J} t\varphi_j$ 
($(\varphi_j)_{j\in J}$ is a partition of unity
of $\mathbb{R}^d$
given by Lemma \ref{Retsigma})
should be 
finite thus we do not
necessarily
have one fixed $m$ for which
$\forall j\in J, 
\Vert\theta^{-m} \widehat{t\varphi_j}\Vert_{L^\infty}<+\infty$.
However, 
$\chi Ut=\sum_{j\in J^\prime} \chi Ut\varphi_j$
where $J^\prime$ is any subset of $J$
such that
$\sum_{j\in J^\prime} \varphi_j=1$
on the \textbf{compact} set 
$\pi_1\left(\sigma^{-1}\left(\text{supp }\chi\times V\right)\right)$,
thus $J^\prime$ can be chosen finite.
Now we use the pseudodifferential partition of unity
indexed by $J^\prime$ 
to patch
everything together:
$$\forall\xi\in V, \vert\mathcal{F}\left(\chi U t\right)\vert(\xi)
\leqslant
\sum_{j\in J^\prime}
\vert\int_{\mathbb{R}^{2d}}dxd\eta  e^{i[S(x;\eta)-x.\xi]}\chi(x) \widehat{t\varphi_j} (\eta)\vert $$
$$  \leqslant \sum_{j\in J^\prime} C_N(1+\vert\xi\vert)^{-N} \left(\Vert \theta^{-m} \widehat{t\varphi_j}\Vert_{L^\infty} +  \Vert t\Vert_{N+2d+1,W_j,\varphi_j}\right)$$ 
by estimate (\ref{est3}) where $W_j=\text{supp }(1-\alpha_j)$.
And this final estimate 
generalizes directly 
to families
of distributions $(t_\mu)_\mu$:
$$\Vert Ut_\mu \Vert_{N,V,\chi} \leqslant \sum_{j\in J^\prime} C_N \left(\Vert \theta^{-m} \widehat{t_\mu\varphi_j}\Vert_{L^\infty} +  \Vert t_\mu\Vert_{N+2d+1,W_j,\varphi_j}\right).$$
\end{proof}
For $t_\mu$
in a bounded family
of distributions,
there is 
a finite integer
$m$
(which depends on the finite partition
of unity $\varphi_j$)
such that the r.h.s. of the above inequality is bounded
thus all seminorms
$\Vert.\Vert_{N,V,\chi}$
for $\mathcal{D}^\prime_{\sigma\circ\Gamma}$ are bounded.
Finally, it remains 
to check 
that
the pull-back by a diffeomorphism
of a weakly bounded
family
of distributions
is weakly bounded,
the proof is a
simple application of 
the variable change formula
for distributions (\cite{Eskin} formula (3.7) p.~10).

\paragraph{Consequences for the scaling with different Eulers.}
\begin{defi}
$t$ is microlocally weakly homogeneous of degree $s$ at $p\in I$ for $\rho$
if $WF(t)$ satisfies the local soft landing condition at $p$, 
there exists a $\rho$-convex open set $V_p$ 
such that $(\lambda^{-s}e^{\log\lambda\rho*}t)_{\lambda\in(0,1]}$ is bounded in $\mathcal{D}^\prime_{\Gamma}(V_p\setminus I)$ 
for some $\Gamma\subset T^\bullet V_p$
which satisfies the soft landing condition.
\end{defi}
In particular, if $(\lambda^{-s}e^{\log\lambda\rho*}t)_{\lambda\in(0,1]}$ is bounded in $\mathcal{D}^\prime_{\Gamma}(V_p\setminus I)$ for $\Gamma=\bigcup_{\lambda\in(0,1]}WF(t_\lambda)$
then
$t$ is microlocally weakly homogeneous of degree $s$
since $WF(t)$ 
satisfies the soft landing condition 
implies $\Gamma=\bigcup_{\lambda\in(0,1]}WF(t_\lambda)$ 
also does.
\begin{thm}\label{thminvmuloc}
Let $t\in \mathcal{D}^\prime(M\setminus I)$. If $t$ is microlocally weakly homogeneous of degree $s$ at $p\in I$ for some $\rho$ then it is so for any $\rho$.
\end{thm}
\begin{proof}
Let $\rho_1,\rho_2$ 
be two Euler vector fields and 
$t$ is microlocally weakly homogeneous 
of degree $s$ at $p\in I$ for $\rho_1$.
We use Proposition \ref{propositionvariablefamily} which states that locally 
there exists a smooth family of diffeomorphisms 
$\Phi(\lambda): V_p\mapsto V_p$ such that 
$\forall\lambda\in[0,1],\Phi(\lambda)(p)=p$ 
and 
$\Phi(\lambda)$ relates the two scalings: 
$$e^{\log\lambda\rho_2*}=\Phi(\lambda)^*e^{\log\lambda\rho_1*} .$$
Then $\Phi(\lambda)^\star$ 
is a Fourier integral operator
which depends
smoothly on 
a parameter 
$\lambda\in [0,1]$. 
$\lambda^{-s} e^{\log\lambda\rho_1*}t$ is bounded in $\mathcal{D}^\prime_{\Gamma_1}(V\setminus I)$, 
then we apply Theorem (\ref{Fouriervariablephase}) to deduce that the family
$$\Phi(\lambda)^*\left(\lambda^{-s}e^{\log\lambda\rho_1*}t\right)_\lambda=\left(\lambda^{-s}e^{\log\lambda\rho_2*}t\right)_\lambda $$
is in fact bounded in $\mathcal{D}^\prime_{\Gamma_2}(V_p)$, with $\Gamma_2$ given by the equation
$$ \Gamma_2= \bigcup_{\lambda\in[0,1]}\sigma_\lambda \circ \Gamma_1 $$
where $\sigma_\lambda=T^\star \Phi^{-1}(\lambda)$. 
\end{proof}
The previous theorem allows us 
to define a space of distributions $E_s(U)$ that are microlocally weakly homogeneous of degree $s$,  the definition being independent 
of the choice of Euler vector field $\rho$:
\begin{defi}\label{defEsmuloc} 
$t$ is microlocally weakly homogeneous of degree $s$ at $p$ if $t$ is microlocally weakly homogeneous of degree $s$ at $p$ for some $\rho$. $E^{\mu}_s(U)$ is the space of all distributions $t\in\mathcal{D}^\prime(U)$ such that
$\forall p\in (I\cap \overline{U})$, $t$ is microlocally weakly homogeneous of degree $s$ at $p$.
\end{defi}
We now state
a general theorem which summarizes all our investigations
in the first four chapters of this thesis and is a microlocal analog 
of Theorem \ref{thmfi},
\begin{thm}\label{thmfin2}
Let $U$ be an open neighborhood of $I\subset M$, if $t\in E^\mu_s(U\setminus I)$ then there exists
an extension $\overline{t}$ in $E^\mu_{s^\prime}(U)\cap \mathcal{D}^\prime_{WF(t)\cup C}\left(U\right)$ 
where $s^\prime=s$ if $-s-d\notin\mathbb{N}$ 
and
$s^\prime<s$ otherwise.
\end{thm}

\section{Appendix.}
%\begin{lemm}\label{theFourierversionofEs}
%$t\in E_s(U) \Leftrightarrow \forall \chi\in C_c^\infty(U),\exists m\in\mathbb{N}, C<\infty$ $$\sup_{\lambda\in[0,1]} \Vert (1+\vert k\vert+\lambda\vert\xi\vert)^{-m}\lambda^{-s-d} \widehat{t \chi_{\lambda^{-1}}}(k,\xi)\Vert_{L^\infty(\mathbb{R}^{n+d})} \leqslant C $$
%\end{lemm}
We recall a deep theorem
of Laurent Schwartz
(see \cite{Schwartz} p.~86 theorem (22))
which gives a concrete
representation
of bounded families
of distributions.
\begin{thm}\label{LaurentSchwartzbounded}
For a subset $B\subset \mathcal{D}^\prime(\mathbb{R}^d)$
to be bounded it is neccessary and sufficient 
that
for any domain $\Omega$ with compact closure,
there is an multiindex $\alpha$
such that
$\forall t\in B,\exists f_t\in C^0(\Omega)$ where
$t|_\Omega= \partial^\alpha f_t$
and $\sup_{t\in B}\Vert f_t\Vert_{L^\infty(\Omega)}<\infty$.
\end{thm}
We give 
an equivalent formulation
of the theorem
of Laurent Schwartz
in terms of Fourier
transforms:
\begin{thm}\label{vouluChristian}
Let $B\subset \mathcal{D}^\prime(\mathbb{R}^d)$.
$$ \forall \chi\in \mathcal{D}(\mathbb{R}^d),\exists m\in\mathbb{N},\, \sup_{t\in B} \Vert (1+\vert\xi\vert)^{-m}\widehat{t\chi}\Vert_{L^\infty}<+\infty$$ 
$$\Leftrightarrow B \text{ weakly bounded in }\mathcal{D}^\prime(\mathbb{R}^d)\Leftrightarrow B \text{ strongly bounded in }\mathcal{D}^\prime(\mathbb{R}^d).$$
\end{thm}
\begin{proof}
We will not recall here the proof
that $B$ is weakly bounded is
equivalent to
$B$ is strongly bounded 
(by Banach Steinhaus see the appendix of Chapter $1$).
Assume $\forall \chi\in \mathcal{D}^\prime(\mathbb{R}^d),\exists m\in\mathbb{N},\sup_{t\in B} \Vert (1+\vert\xi\vert)^{-m}\widehat{t\chi}\Vert_{L^\infty}<+\infty$.
We fix an arbitrary test function
$\varphi$. There is
a function $\chi\in\mathcal{D}(\mathbb{R}^d)$ 
such that $\chi=1$
on the support of $\varphi$.
Then 
$$\vert\left\langle t , \varphi \right\rangle\vert=\vert\left\langle t\chi , \varphi \right\rangle\vert=\vert\left\langle \widehat{t\chi} , \widehat{\varphi} \right\rangle\vert$$ $$=\vert\int_{\mathbb{R}^d} d^d\xi (1+\vert\xi\vert)^{-d-1}(1+\vert\xi\vert)^{-m}\widehat{t\chi}(\xi)(1+\vert\xi\vert)^{m+d+1}\widehat{\varphi}(\xi) \vert $$
$$\leqslant \underset{\text{integrable}}{\underbrace{\int_{\mathbb{R}^d} d^d\xi (1+\vert\xi\vert)^{-d-1}}}\vert(1+\vert\xi\vert)^{-m}\widehat{t\chi}(\xi)\vert \vert(1+\vert\xi\vert)^{m+d+1} \widehat{\varphi}(\xi) \vert$$
$$\leqslant C\Vert (1+\vert\xi\vert)^{-m}\widehat{t\chi}\Vert_{L^\infty} \pi_{m+d+1}(\varphi), $$
finally 
$$\sup_{t\in B} \vert\left\langle t , \varphi \right\rangle\vert \leqslant C\pi_{m+d+1}(\varphi)\sup_{t\in B}\Vert (1+\vert\xi\vert)^{-m}\widehat{t\chi}\Vert_{L^\infty}<+\infty .$$
Conversely, we can always assume $B$ to be strongly bounded, then
for all $\chi\in \mathcal{D}_K(\mathbb{R}^d)$,
the family $\left(\chi e_\xi\right)_{\xi\in\mathbb{R}^d}$ where $e_\xi(x)=e^{-ix.\xi}$ 
has fixed compact support $K$. 
Then there exists
$m$ and a universal constant $C$ such that
$$\forall t\in B, \forall\varphi \in \mathcal{D}(K), \vert\left\langle t , \varphi \right\rangle\vert\leqslant C\pi_m(\varphi)$$
thus 
$$\forall t\in B, \vert \widehat{t\chi}\vert(\xi)=\vert\left\langle t , \chi e_\xi  \right\rangle\vert\leqslant C\pi_m( \chi e_\xi ), $$
now notice that $\pi_m( \chi e_\xi)$ is polynomial in $\xi$ of degree $m$ thus $\sup_{\xi} \vert(1+\vert\xi\vert)^{-m}\pi_m( \chi e_\xi)$ is bounded.
But then $(1+\vert\xi\vert)^{-m}\vert \widehat{t\chi}(\xi)\vert\leqslant C\underset{\text{bounded in }\xi}{\underbrace{\vert(1+\vert\xi\vert)^{-m}\pi_m( \chi e_\xi )\vert}}$ and thus $\sup_{t\in B} \Vert \theta^{-m}\widehat{t\chi}\Vert_{L^\infty}<+\infty$.
\end{proof}

\chapter{The two point function 
$\left\langle 0|\phi(x)\phi(y)|0\right\rangle$.}
\paragraph{Introduction.}
Hadamard states are nowadays widely accepted as possible physical states of the free 
quantum field theory on a curved space-time. 
The Hadamard condition plays an essential role in the perturbative
construction of interacting quantum field theory \cite{BF}. 
Since the work of Radzikowski \cite{R}, the ``Hadamard condition'' (renamed microlocal
spectrum condition) is formulated as a requirement on 
the wave front set of the
associated two-point function $\Delta_+$ which is necessarily a bisolution of the
wave equation in the globally hyperbolic space-time.
The construction of solutions of the wave equation in a globally hyperbolic space-time
by the parametrix method, following Hadamard \cite{Hadamard} and Riesz \cite{Riesz}, 
is by now classical in the mathematical literature. 
For space-times of the form 
$\mathbb{R}\times M$
where
$M$ is a compact Riemannian manifold,
it is well known that 
$\Delta_+=\frac{e^{it\sqrt{-\Delta}}}{\sqrt{-\Delta}}$
where $e^{it\sqrt{-\Delta}}$ is a Fourier
integral operator constructed in \cite{DG} theorem $(1.1)$ p.~43 
with 
the wave front set satisfying the Hadamard condition 
(see also \cite{CDV} th\'eor\`eme 1 p.~2).
However, to our knowledge, only the
recent work of C. G\'erard and M. Wrochna 
\cite{GW} treats the non 
static space-times case 
(although
\cite{Junker} 
constructed
Hadamard states
on space-times 
with compact Cauchy surfaces). 
Furthermore, 
for the purpose 
of renormalizing 
interacting quantum field theory,
we need to establish 
that $\Delta_+$
has finite
``\textbf{microlocal scaling degree}''
(following the terminology of \cite{BF}),
which is a \textbf{stronger assumption} 
than establishing that 
$WF(\Delta_+)$ satisfies 
the Hadamard condition.
  
 The goal of this chapter is to prove that 
$\Gamma=WF(\Delta_+)$ satisfies the microlocal spectrum condition
and that 
$\Delta_+$ is
microlocally weakly homogeneous
of degree $-2$ in the sense of 
Chapter 
4 (means 
in the notation 
of Chapter 4
that
$\Delta_+\in E^\mu_{-2}$).
Although 
our goal is 
not to construct $\Delta_+$
on flat space, 
as preliminary, we spend some time 
to present various 
different mathematical interpretations of the Wightman function
$\Delta_+$ in the flat case and give many formulas 
that are scattered in the mathematical literature.
We provide 
proofs (or give precise references whenever we do not give all
the details) of so called 
``well known facts'', 
as
for instance the $\emph{Wick rotation}$, 
which cannot be easily found in the mathematical 
literature. 
In fact, 
our work done in the flat case 
will be useful 
when we pass to the curved case.

\paragraph{Our plan and some historical comments.}

 The first section deals with the Wightman 
function $\Delta_+$ in $\mathbb{R}^{n+1}$. 
We start with the expression of the Wightman function 
given by Reed and Simon \cite{RS}: $\Delta_+$ is the inverse Fourier transform $\mathcal{F}^{-1}\left(\mu\right)$ of a Lorentz invariant measure $\mu$ supported by the positive mass hyperboloid in momentum space.
This beautiful interpretation also appears in the book of Laurent Schwartz \cite{Schwartzpart}. This gives a first proof that $\Delta_+$ is a tempered distribution.
The formalism of functional calculus immediately allows us to relate $\mathcal{F}^{-1}\left(\mu\right)$ with the function $\frac{e^{it\sqrt{-\Delta}}}{\sqrt{-\Delta}}$ of the
Laplace operator $-\Delta$, 
$\frac{e^{it\sqrt{-\Delta}}}{\sqrt{-\Delta}}$ is a solution
in the space of
operators 
of the
wave equation.
From the inverse Fourier transform formula, 
$\Delta_+$ is interpreted as an oscillatory integral (\cite{RS}), hence 
by a theorem of H\"ormander, this gives
us a first possible way to compute the WF of $\Delta_+$.
 
 Then we give a second approach to the Wightman function: 
we notice the striking similarity of 
$\frac{e^{it\sqrt{-\Delta}}}{\sqrt{-\Delta}}$ with 
the Poisson kernel $\frac{e^{-\tau \sqrt{-\Delta}}}{\sqrt{-\Delta}}$, 
and the fact 
that they
should
be the 
same formula 
if we could treat the time variable $t$
as a complex variable. To carry out this program, 
we first compute the inverse Fourier transform 
w.r.t. to the variables $\xi$ of the Poisson kernel 
$\frac{e^{-\tau\vert\xi\vert}}{\vert\xi\vert}$, 
we obtain the function $\frac{C}{\tau^2+\sum_{i=1}^n (x^i)^2}$ 
which can be viewed as the Schwartz kernel of the operator $e^{-t\sqrt{-\Delta}}\sqrt{-\Delta}^{-1}$.  
This computation relies on the beautiful 
\textbf{subordination identity} connecting 
the Poisson operator and the Heat kernel.
Then we show how to make sense of the analytic continuation in time of the Poisson kernel $\frac{C}{\tau^2+\sum_{i=1}^n x_i^2}$, 
called the wave Poisson kernel and
which 
corresponds to the operator
$e^{i(t+i\tau)\sqrt{-\Delta}}\sqrt{-\Delta}^{-1} $. 
This allows to recover $\Delta_+$ when the complexified time $(\tau-it)$ becomes \textbf{purely imaginary}, justifying the famous \textbf{Wick rotation} 
and 
giving a third proof that $\Delta_+$ is a distribution.
In fact, to generalize this idea to static space-times of the form $\mathbb{R}\times M$ where $M$ 
is a noncompact Riemannian manifolds, we can use the machinery of functional calculus defined in
the monograph \cite{Taylor} (see also \cite{Taylor-Pinsky}), from the relation 
$$f(\sqrt{-\Delta}_g)=\frac{1}{(2\pi)^{\frac{1}{2}}}\int_{-\infty}^{+\infty} \widehat{f}(t) e^{it\sqrt{-\Delta}_g}  ,$$ 
one can easily define the analytic continuation in time of the Poisson kernel $e^{-\tau \sqrt{-\Delta}_g}\sqrt{-\Delta}_g^{-1}$, hence define the Wick rotation of $e^{-\tau \sqrt{-\Delta}_g}\sqrt{-\Delta}_g^{-1}$ where $\Delta_g$ denotes the Laplace--Beltrami operator on the noncompact Riemannian manifold, 
this will be the object of future investigations. 
Finally, we arrive at the formula 
which expresses the kernel 
of the Wightman function 
as a distribution defined as the boundary value of a holomorphic function
$$\frac{C}{Q(.+i0\theta)}=\lim_{\varepsilon\rightarrow 0^+} \frac{C}{(x^0\pm i\varepsilon)^2-\sum_{i=1}^n (x^i)^2} ,$$ 
where $Q(x)=(x^0)^2-\sum_{i=1}^n (x^i)^2$ and $\theta=(1,0,0,0)$.
Applying general theorems of H\"ormander, this gives a fourth proof of the fact that $\Delta_+$ is a distribution and a second way to estimate the wave front set of $\Delta_+$. 
Along the way, we prove that 
$\log\left((x^0+i0)^2-\sum_1^n (x^i)^2 \right)$ 
and the family $\left((x^0+i0)^2-\sum_{i=1}^n (x^i)^2\right)^s$ 
are distributions with wave front set 
satisfying the microlocal condition condition.
\paragraph{Going to the curved case.}
There are two conceptual difficulties
when we pass to the curved case, the first
is to intrinsically define objects
on $M^2$ which generalize 
the singularity 
$Q^{-1}(\cdot + i0\theta)$ of $\Delta_+$
and the powers of $Q$ in general. 
The starting
point 
is to pull back 
distributions and functions
defined on $\mathbb{R}^{n+1}$ by a map $F:\mathcal{V}\subset M^2\mapsto \mathbb{R}^{n+1}$ constructed by
inverting
the exponential 
geodesic map.
This well-known technique was already used in \cite{Hadamard} and \cite{Riesz} 
and is expounded 
in many recent works (\cite{Bar-Ginoux-Pfaffle}, \cite{Waldmann}), 
however none of these works present a computation of the wave front set
of the pulled back singular term $F^\star Q^{-1}(\cdot + i0\theta)$. 
Here we prove
that the wave front set of the singular term 
$F^\star Q^{-1}(\cdot + i0\theta)$ 
satisfies the Hadamard condition as stated in \cite{R}.
 
 The second step consists
in pulling back
certain distributions 
in 
$\mathcal{D}^\prime(M)$ 
on $\mathbb{R}^{n+1}$ in order  
to set 
and solve
the system
of transport equations. 
For all $p\in M$, 
we define a map 
$E_p:\mathbb{R}^{n+1}\mapsto M$
which allows to pull-back functions, 
differential operators and the metric
on $\mathbb{R}^{n+1}$ 
($\mathbb{R}^{n+1}$ 
is identified with the exponential
chart centered at $p$).

 Once these two difficulties
are solved, and all proper
geometric objects are defined, 
it is simple to follow 
the classical construction
of Hadamard \cite{Hadamard} 
to obtain a 
parametrix with 
suitable
wave front set.
\section{The flat case.} 
Fix the Lorentz invariant quadratic form $Q(x^0,x^1,\dots,x^n)=(x^{0})^{2}-\sum_1^n (x^{i})^2$ in $\mathbb{R}^{n+1}$.
In the book of Laurent Schwartz \cite{Schwartzpart}, the study of particles is related to the problem of finding Lorentz invariant tempered distributions of positive type on $\mathbb{R}^{n+1}$. By Fourier transform and application of the Bochner theorem (p.~60,66 in \cite{Schwartzpart}), 
it is equivalent to the problem of finding positive Lorentz invariant measures $\mu\in \left(C^0(\mathbb{R}^{n+1})\right)^\prime$ in momentum space.
Then $\mu$ is called a scalar particle.
If the particle is \emph{elementary}, 
it is required that $\mu$ is extremal which means that 
$\mu=\sum \alpha_i\mu_i$ holds iff 
$\mu_i$ are proportional to $\mu$. 
This notion of extremal measure is 
the analogue in functional analysis 
of the notion of irreducible 
representations of a group in representation theory.
We also require that $\mu$ has \emph{positive energy} 
i.e. $\mu$ is supported on $\{x^0\geqslant 0\}$.
Before we discuss Lorentz invariant measures, we would like to give a simple formula which is a reinterpretation of the usual Lebesgue integration in $\mathbb{R}^{n+1}$ in terms of $\emph{slicing}$ by the orbits of the Lorentz group:
\begin{equation}
\int_{\mathbb{R}^{n+1}} f \wedge_{\mu=0}^n dx^\mu = \int_{-\infty}^\infty dm \int_{Q=m} f\frac{\wedge_{\mu=0}^n dx^\mu}{dQ} 
\end{equation} 
as a consequence of the
coarea formula 
of Gelfand Leray (\cite{Gelfand}, \cite{Zoladek}). 
Notice that we can produce natural Lorentz invariant measures by modifying this integral, instead of integrating over the Lebesgue measure $dm$ over the real line, we integrate against an arbitrary measure $\rho(m)$:
\begin{prop} 
Any Lorentz invariant measure of positive energy $\mu$ can be represented by the formula
\begin{equation}
\mu(f)= \int_{-\infty}^\infty \rho(m) \int_{Q=m} f\frac{\wedge_{\nu=0}^n dx^\nu}{dQ}+cf(0) 
\end{equation}  
where the measure $\rho$ is in fact the push-forward of $\mu$:
$$\rho=Q_*(\mu) .$$
In particular, by Bochner theorem, any tempered positive distribution $\mu$ invariant by $O(n,1)_+^{\uparrow}$ can be represented by
\begin{equation}
\mu(f)=\int_{-\infty}^\infty \rho(m) \int_{Q=m} \widehat{f}\frac{\wedge_{\nu=0}^n dx^\nu}{dQ} + c\int_{\mathbb{R}^{n+1}} d^{n+1}x f(x).
\end{equation}  
\end{prop}
\begin{proof} 
The proof is given in full detail in \cite{RS} Theorem 9.33 p.~75
and also the classification of all Lorentz invariant distributions was given by M\'eth\'ee.
\end{proof}
From now on, we assume $\mu$ has positive energy. 
Inspired by the previous proposition, we claim
\begin{prop} 
Any extremal measure of positive energy $\mu$ in $\mathbb{R}^{n+1}$ which is invariant by the group
$O(n,1)_+^{\uparrow}$ 
of time and orientation preserving Lorentz transformations is supported on one orbit of $O(n,1)_+^{\uparrow}$. 
\end{prop}
\begin{proof} 
It was proved in a very general setting in \cite{Schwartzpart} p.~72.
The orbits of $O(n,1)_+^{\uparrow}$ in the positive energy region $\{x^0\geqslant 0\}$ are connected components of constant mass hyperboloids for $m>0$,
the half null cone $ (x^0)^2-\vert x\vert^2=0, x^0>0 $
and the fixed point $\{0\}$ of the group action:
$$\bigcup_{m>0} \underset{\text{positive mass hyperboloids}}{\{ (x^0)^2-\vert x\vert^2=m^2, x^0>0\}}\bigcup \underset{\text{halfcone}}{\{(x^0)^2-\vert x\vert^2=0, x^0>0\}}\bigcup \underset{\text{origin}}{\{0\}} .$$
Let $\mu$ be an $O(n,1)_+^{\uparrow}$ invariant measure on $\mathbb{R}^{n+1}$. Let $Q$ be 
the canonical 
$O(n,1)$ invariant quadratic form. 
Then the push-forward $Q_*\mu$ is a well defined measure on $\mathbb{R}^+$ (since $\mu$ has positive energy) because $Q$ is smooth and the support of $Q_*\mu$ contains the masses of the particles.
Assume the support of $\mu$ contains two points which are in disjoint orbits of $O(n,1)_+^{\uparrow}$, then the push-forward $Q_*\mu$ is supported at two different points $m_1,m_2$. Then pick a smooth function $0\leqslant\chi\leqslant 1$ such that $\chi(m_1)=1$ and $\chi(m_2)=0$ and consider the pair of push pull measures 
$$ Q^*\left(\chi Q_*\mu\right), Q^*\left((1-\chi)Q_*\mu\right) .$$ 
These are measures with different supports, hence linearly independent, and
$$\mu= H(x^0)Q^*\left(\chi Q_*\mu\right) + H(x^0)Q^*\left((1-\chi)Q_*\mu\right) $$ 
which contradicts the extremality of $\mu$.    
\end{proof}
%
%
%
% 
%
% We recall here a definition relating the Wightman functions $W_2$ and $\Delta_+$ taken from Reed Simon vol $2$, theorem $9.33$, theorem $9.34$ page $70$.
% 
% This definition defines $\Delta_+$ as the distribution representing the linear form $W_2$.
% 
%\begin{defi}
%There exists a polynomially bounded measure $\rho$ on $[0,\infty)$ such that
%$$W_2(f)=\int_0^\infty d\rho(m) \left(\int_{H_m}\hat{f}d\Omega_m\right)=\int_0^\infty d\rho(m)\frac{1}{i}\int dx f(x)\Delta_{2,+}(x;m)$$ $$=\int_0^\infty d\rho(m) \frac{1}{i}\left\langle f,\Delta_{2,+}(m)\right\rangle ,$$ which means the linear map $W_2$ is represented as the composition of three operations: first the Fourier transform of $f$ then the restriction of $\hat{f}$ on the mass hyperboloid of $\textbf{positive energy}$, finally the averaging of the restriction by a Lorentz invariant measure. 
%\end{defi}
%
Now, let $\mu$ be an extremal measure of positive energy. We already saw the support of $\mu$ is one orbit of $O(n,1)_+^{\uparrow}$, a hyperboloid of mass $m>0$.
Here we give an interpretation of the $O(n,1)^\uparrow_+$ invariant measure $\mu$ supported by the mass shell $m$ of positive energy (which is unique by theorem $9.37$ in \cite{RS}) in terms of the Gelfand--Leray distributions (see \cite{Gelfand}).
We introduce the
following notations:
$$\xi=(\xi^\mu)_{0\leqslant \mu\leqslant n}=(\xi^0,\xi^i_{1\leqslant i\leqslant n})
=(\xi^0,\overrightarrow{\xi}).$$
\begin{prop}
Let $\Omega=d\xi^0\wedge d^n\overrightarrow{\xi}$ be the canonical measure in $\mathbb{R}^{n+1}$ 
and $Q= (\xi^{0})^2-\sum_{i=1}^n (\xi^{i})^2$.
Then we can construct an $O(n,1)^\uparrow_+$ invariant measure $\mu$ supported by the component of positive energy of $Q=m$ given by the formulas:
\begin{equation}
\mu(f)=\left\langle\delta_{m},\left(\int_{Q=m}f\frac{\Omega}{dQ}\right)\right\rangle=\int_{\mathbb{R}^{n}} \frac{d^n\overrightarrow{\xi}}{2\sqrt{m^2+|\overrightarrow{\xi}|^2}}  f((m^2+\vert\overrightarrow{\xi}\vert^2)^{\frac{1}{2}},\overrightarrow{\xi}).
\end{equation}
\end{prop}
\begin{proof}
 Let us remark that the Lebesgue measure in momentum space $\Omega=d\xi^0\wedge d^n\overrightarrow{\xi}$ is $O(n,1)$ invariant because the determinant of any element in $O(n,1)$ equals $1$.
Let us compute the $\delta$ function $\delta_{\{(\xi^0)^2-|\overrightarrow{\xi}|^2=m,\xi^0\geqslant 0\}}(\Omega)$ as defined in Gelfand--Shilov \cite{Gelfand} :
$$\delta_{\{(\xi^0)^2-|\overrightarrow{\xi}|^2=m,\xi^0\geqslant 0\}}(d\xi^0\wedge d^{n}\overrightarrow{\xi})= \int_{\{ \xi^0=\sqrt{m^2+|\overrightarrow{\xi}|^2}\} }\frac{d\xi^0\wedge d^{n}\overrightarrow{\xi}}{d((\xi^0)^2-(m^2+|\overrightarrow{\xi}|^2))} $$
The Gelfand-Leray form 
$\frac{d\xi^0\wedge d^{n}\overrightarrow{\xi}}{d((\xi^0)^2-(m^2+|\overrightarrow{\xi}|^2))}$ 
is the ratio of two Lorentz invariant forms. More explicitely, we compute this ratio in the parametrization 
$\overrightarrow{\xi}\in\mathbb{R}^n\mapsto 
((m^2+\vert\overrightarrow{\xi}\vert^2)^{\frac{1}{2}},\overrightarrow{\xi})\in\mathbb{R}^{n+1}$
of the mass hyperboloid:
 $$\frac{d\xi^0 \wedge d^n\overrightarrow{\xi}}{d((\xi^0)^2-(m^2+|\overrightarrow{\xi}|^2))}|_{\xi^0=\sqrt{m^2+|\overrightarrow{\xi}|^2}}=\frac{d\xi^0 \wedge d^n\overrightarrow{\xi}}{2(\xi^0 d\xi^0-\left\langle\overrightarrow{\xi}, d\overrightarrow{\xi}\right\rangle)}|_{\xi^0=\sqrt{m^2+|\overrightarrow{\xi}|^2}}$$ 
$$=\frac{d^n\overrightarrow{\xi}}{2\xi^0}|_{\xi^0=\sqrt{m^2+|\overrightarrow{\xi}|^2}}$$ 
because $\frac{d^n\overrightarrow{\xi}}{2\xi^0}\wedge 2(\xi^0 d\xi^0-\left\langle\overrightarrow{\xi}, d\overrightarrow{\xi}\right\rangle)=d\xi^0 \wedge d^n\overrightarrow{\xi}$
$$=\frac{d^n\overrightarrow{\xi}}{2\sqrt{m^2+|\overrightarrow{\xi}|^2}},$$ 
we thus connect with the formula found in \cite{RS} p.~70,74.
\end{proof}
Once we have this measure $\mu$ in momentum space, 
we would like to recover the distribution
it defines by computing the inverse Fourier
transform $\mathcal{F}^{-1}(\mu)$
in $\mathbb{R}^{n+1}$. 
\begin{prop}
Assume $\Delta_+=\mathcal{F}^{-1}\left(\mu\right)$ where 
$\mu$ is an extremal measure of mass $m$, 
$O(n,1)_+^{\uparrow}$ invariant and of positive energy, 
then $\Delta_+$ is given by the formula
\begin{equation}\label{delta_+formula}
\Delta_+(x;m)=\frac{1}{2(2\pi)^{n+1}}\int_{\mathbb{R}^n} \frac{e^{-ix^0(m^2+|\overrightarrow{\xi}|^2)^{\frac{1}{2}} +i\overrightarrow{x}.\overrightarrow{\xi}}}{(m^2+|\overrightarrow{\xi}|^2)^{\frac{1}{2}}} d^n\overrightarrow{\xi}.
\end{equation}
\end{prop}
\begin{proof}
To prove the claim, we use the Gelfand--Leray notation and the beautiful identity $e^{i\tau f}\omega=e^{i\tau t} dt\int_{t=f} \frac{\omega}{df} $ (\cite{Zoladek} page $124$ lemma $(5.12)$), which allows to rewrite the Reed Simon formula:
$$\delta_{\{(\xi^0)^2-|\overrightarrow{\xi}|^2=m,\xi^0\geqslant 0\}}(e^{i(x^0\xi^0+\overrightarrow{x}.\overrightarrow{\xi})}\Omega) $$
$$=\int e^{i(x^0\xi^0+\overrightarrow{x}.\overrightarrow{\xi})} d\xi^0 \int_{\xi^0=\sqrt{m^2+|\overrightarrow{\xi}|^2}} \frac{d\xi^0 \wedge d^n\overrightarrow{\xi}}{d(\xi^{02}-(m^2+|\overrightarrow{\xi}|^2))}$$ 
$$=\int_{\mathbb{R}^n} e^{i(x^0\sqrt{m^2+|\overrightarrow{\xi}|^2} +\overrightarrow{x}.\overrightarrow{\xi})}\frac{d^n\overrightarrow{\xi}}{2\sqrt{m^2+|\overrightarrow{\xi}|^2}},$$
we recognize the inverse Fourier transform of a distribution supported 
by the positive
sheet of the 
hyperboloid.
\end{proof}
If we provisionnally 
call $t$ the variable
$x^0$
then the above proposition allows to 
interpret $\Delta_+$ 
as the Schwartz kernel
of the operator $\frac{e^{it\sqrt{-\Delta + m^2}}}{\sqrt{-\Delta + m^2}}$
where $\Delta$ is the Laplace operator acting on $\mathbb{R}^n$.
Also notice that the evolution operator $t\mapsto U(t)=\frac{e^{it\sqrt{-\Delta+m^2}}}{\sqrt{-\Delta+m^2}}$ 
satisfies the square root 
Klein--Gordon equation:
$\partial_t-i\sqrt{-\Delta+m^2} U=0 ,$  
thus $\Delta_+(x;m)$ is a solution of the Klein Gordon equation and for any $u\in H^s(\mathbb{R}^n)$, $u_+=\Delta_+(t;m)*u$ is a solution of the Klein Gordon equation 
which has \textbf{positive energy} i.e. its 
Fourier transform is supported in the positive hyperboloid.

\subsection{The Poisson kernel, the Wick rotation and the subordination identity.}
To define $\Delta_+$ as the inverse Fourier transform of 
the measure $\mu$ is not very satisfactory
since it does not give  
an explicit formula for $\Delta_+$
in space variables.
We will prove that $\Delta_+=C((x^0+i0)^2-\vert x\vert^2)^{-1}$ where we explain how to make sense 
of the term on the right hand side as a tempered distribution by the process of 
\textbf{Wick rotation}.
\begin{lemm}
The family of Schwartz distributions 
\begin{equation}\label{fundformula2}
\frac{1}{(2\pi)^n}\int_{\mathbb{R}^n} \frac{e^{ix.\xi-y\vert \xi\vert}}{|\xi|}d^n\xi
%=\mathcal{F}_\xi^{-1}\left( \frac{1}{(2\pi)^{\frac{n}{2}}} \frac{e^{-\vert\xi \vert y}}{\vert\xi\vert} \right)(x)
=\frac{e^{-y\sqrt{-\Delta}}}{\sqrt{-\Delta}}\delta(x)=\frac{\pi^{\frac{n+1}{2}}}{\Gamma(\frac{n-1}{2})}  \frac{1}{(y^2+\vert x\vert^2)^{\frac{n-1}{2}}}
\end{equation}
is holomorphic in $y\in \{y | Re(y)> 0 \}$ and \textbf{continuous} in $y\in\{y | Re(y)\geqslant 0 \}$ with values in $S^\prime(\mathbb{R}^n)$.
\end{lemm}
Similar computations of Poisson integrals are presented in \cite{Stein} p.~60, 130, \cite{Eskin} and \cite{Taylor} (3.5).

\begin{proof}
Our proof follows \cite{Taylor} (3.5).
Everything relies on the following  identity (see the identity $\beta$ in \cite{Stein} p.~61)
\begin{equation}\label{suborviet}
\frac{e^{-Ay}}{A}=\frac{1}{\pi^{1/2}}\int_0^\infty e^{\frac{-y^2}{4t}}e^{-A^2t}t^{-\frac{1}{2}}dt
\end{equation}
which is derived from the \textbf{subordination identity} $(5.22)$ in \cite{Taylor}
\begin{equation}
e^{-Ay}=\frac{y}{2\pi^{1/2}}\int_0^\infty e^{\frac{-y^2}{4t}}e^{-A^2t}t^{-\frac{3}{2}}dt
\end{equation}
by integrating w.r.t. $y$ and by noticing that when $y=0$ our formula 
(\ref{suborviet}) coincides with the Hadamard--Fock--Schwinger formula:
$$\int_0^\infty t^{-\frac{1}{2}} e^{-tA^2} dt=\int_0^\infty t^{\frac{1}{2}} e^{-tA^2} \frac{dt}{t}$$
$$=A^{-1}\int_0^\infty t^{\frac{1}{2}} e^{-t} \frac{dt}{t}= A^{-1} \Gamma(\frac{1}{2})=A^{-1}\pi^{\frac{1}{2}} $$
since $\Gamma(\frac{1}{2})=\pi^{\frac{1}{2}}$.
%which is related to the beautiful Hadamard-Fock-Schwinger formula
%\begin{equation}
%A^{-s-1}\Gamma(s+1)=\int_0^\infty t^s e^{-tA}dt
%\end{equation}
%and 
%\begin{equation}
%\int_{-\infty}^\infty (\lambda^2+|\xi|^2)^{-1}e^{ix\xi}d\xi=\frac{\pi}{\lambda} e^{-\lambda x}
%\end{equation}
%There is a problem here. We need instead to consider $\int_{\R^3} (\lambda^2+|\xi|^2)^{-1}e^{ix\xi}d\xi$.
Next we need 
functional calculus in our proof since we 
want to apply 
the subordination 
identity
with 
$A=\sqrt{-\Delta}$.
We then get 
an identity for 
functions of the 
\textbf{operator} 
$\sqrt{-\Delta}$. 
We apply these operators to the delta function supported at $0$:
$$\frac{e^{-\sqrt{-\Delta}y}}{\sqrt{-\Delta}}\delta_0=\left(\frac{1}{\pi^{1/2}}\int_0^\infty e^{\frac{-y^2}{4t}}e^{t\Delta}t^{-\frac{1}{2}}dt\right)\delta_0,$$
we recognize on the left hand side a distributional solution of the Poisson operator $\partial_y^2+\Delta_x$ and on the right hand side, we recognize the Heat kernel $e^{t\Delta}\delta_0=\frac{1}{(4\pi t)^{\frac{n}{2}}}e^{-\frac{\vert x\vert^2}{4t}} .$
Substituting 
in the previous formula,
$$\frac{e^{-\sqrt{-\Delta}y}}{\sqrt{-\Delta}}\delta_0=\frac{1}{\pi^{1/2}}\int_0^\infty e^{\frac{-y^2}{4t}}\frac{1}{(4\pi t)^{\frac{n}{2}}}e^{-\frac{\vert x\vert^2}{4t}}t^{-\frac{1}{2}}dt = \frac{1}{(4\pi)^{\frac{n}{2}}\pi^{\frac{1}{2}}} \int_0^\infty dt e^{-\frac{y^2+|x|^2}{4t}} \frac{1}{t^{\frac{n+1}{2}}} $$
set $t=\frac{1}{4s}$ and we get $$\frac{1}{(4\pi)^{\frac{n}{2}}\pi^{\frac{1}{2}}} \int_0^\infty \frac{ds}{4s^2} e^{-(y^2+|x|^2)s} (4s)^{\frac{n+1}{2}}=\frac{1}{2 \pi^{\frac{n+1}{2}}}\int_0^\infty dse^{-(y^2+|x|^2)s} s^{\frac{n-3}{2}}$$
finally by a variable change in the formula of the Gamma function $$\frac{e^{-\sqrt{-\Delta}y}}{\sqrt{-\Delta}}\delta_0=\frac{\Gamma(\frac{n-1}{2})}{2 \pi^{\frac{n+1}{2}}}\frac{1}{(y^2+|x|^2)^{\frac{n-1}{2}}} $$
\end{proof}
We give an interpretation
of 
$((t\pm i0)^2-\vert x\vert^2)^{-\frac{n-1}{2}}$
as an oscillatory 
integral.
\begin{thm}
The limit $\lim_{\varepsilon\rightarrow 0}((t\pm i\varepsilon)^2-\vert x\vert^2)^{-\frac{n-1}{2}}$ makes sense in $S^\prime(\mathbb{R}^n)$ and satisfies the identity: 
\begin{equation}\label{Taylormagicidentity}
((t\pm i0)^2-\vert x\vert^2)^{-\frac{n-1}{2}}=\frac{(-1)^{\frac{n-1}{2}}\pi^{\frac{n+1}{2}}}{\Gamma(\frac{n-1}{2})(4\pi)^{\frac{n-1}{2}}}\int_{\mathbb{R}^n} d^n\xi \frac{1}{|\xi|}e^{\pm it|\xi|}e^{ ix.\xi}
\end{equation}
\end{thm}
\begin{proof}
The key argument of the proof is to justify the analytic continuation of the Poisson kernel, this is called Wick rotation in physics textbooks. 
% We start with
%the fundamental solution of 
%$$\frac{d^2f}{dy^2}(x,y) + \Delta_x f(x,y)=0,\{x\in\R^n, y > 0 \}$$
%which in partial Fourier transform is given by formula
%$e^{-t|\xi|} $, by the inverse Fourier tranform
%$$P(y,x)=\mathcal{F}^{-1}\left(\frac{1}{(2\pi)^{n/2}}e^{-y|\xi|}\right)=\frac{1}{(2\pi)^n}\int_{\R^n}d^n\xi e^{-y|\xi|}e^{ix.\xi}  =e^{-t\Delta}\delta_{x=0} ,$$
%
%
% this fundamental solution can be interpreted as a function of the Laplace operator acting on the $\delta$ function. This fundamental solution is analytic in $t$ for $t>0$ and hence can be analytically continued in the complex domain $Re(t)>0$. Finally $\lim_{Re(y)\rightarrow 0^+}P(x,y)=\mathcal{F}^{-1}(\lim e^{-y|\xi|})=\mathcal{F}^{-1}(1)$ exists in the sense of Schwartz distributions by \textbf{continuity of the Fourier transform on } $S^\prime(\mathbb{R}^n)$ as the inverse Fourier transform of the constant function $1$ which is a well defined tempered distribution.
% 
Notice that $\frac{e^{-y\sqrt{-\Delta}}}{\sqrt{-\Delta}}\delta_0$ is the Schwartz kernel of a well defined operator $\frac{e^{-\sqrt{-\Delta}y}}{\sqrt{-\Delta}}$. 
Through the partial Fourier transform  w.r.t. 
the variable $x$, the operator 
$\frac{e^{-y\sqrt{-\Delta}}}{\sqrt{-\Delta}}$  
corresponds to
the multiplication 
by $\frac{e^{-y\vert\xi\vert}}{\vert\xi\vert}$.  
Consider now the function $\frac{1}{|\xi|}e^{-y|\xi|} $, \textbf{when $n\geqslant 2$}
this function is analytic in $\{y,Re(y)>0\}$ with value 
Schwartz distribution in $\xi$ because 
$$\forall y\in\{Re(y)\geqslant 0\}, \vert \frac{1}{|\xi|}e^{-y|\xi|}\vert \leqslant \frac{1}{|\xi|}\in L_{loc}^1(\mathbb{R}^n).$$ 
Notice
that 
the above estimate is still true
when $Re(y)\rightarrow 0^+$
hence $\frac{1}{|\xi|}e^{-y|\xi|} $ is a well defined Schwartz distribution in $\xi$ for $Re(y)\geqslant 0$ (it is continuous in $y$ with value tempered distribution).
Finally, we can continue this operator in the $y$ variable in the domain $Re(y)\geqslant 0$ , set $y=\tau+it$ and let $\tau$ tends to zero in $\mathbb{R}^+$. Set $\frac{e^{-\sqrt{-\Delta}(\tau\pm it)}}{\sqrt{-\Delta}}\delta_0=\frac{\Gamma(\frac{n-1}{2})}{2 \pi^{\frac{n+1}{2}}}\frac{1}{((\tau\pm it)^2+|x|^2)^{\frac{n-1}{2}}} $ then at the limit we find
$$\begin{array}{cc}
\frac{1}{(2\pi)^n}\int_{\mathbb{R}^n} \frac{e^{ix.\xi-(\tau\pm it)\vert \xi\vert}}{|\xi|}d^n\xi & = \frac{\Gamma(\frac{n-1}{2})}{2 \pi^{\frac{n+1}{2}}}\frac{1}{(-(t\mp i\tau)^2+\vert x\vert^2)^{\frac{n-1}{2}}} \\ 
\downarrow_{\tau\rightarrow 0^+} & \downarrow_{\tau\rightarrow 0^+} \\
\frac{1}{(2\pi)^n}\int_{\mathbb{R}^n} \frac{e^{ix.\xi\mp it\vert \xi\vert}}{|\xi|}d^n\xi  &=\frac{\Gamma(\frac{n-1}{2})}{2 \pi^{\frac{n+1}{2}}}\frac{1}{(\vert x\vert^2-(t\mp i0)^2)^{\frac{n-1}{2}}}\end{array}$$
\end{proof}

\subsection{Oscillatory integral.}
For QFT, we are interested in the formula (\ref{Taylormagicidentity}) for $n=3$. 
\begin{equation}
((t\pm i0)^2-\vert x\vert^2)^{-1}=C_n\int_{\mathbb{R}^n} d^n\overrightarrow{\xi} \frac{1}{|\overrightarrow{\xi}|}e^{\pm it|\overrightarrow{\xi}|}e^{-i\overrightarrow{x}.\overrightarrow{\xi}}, C_n=\frac{(-1)^{\frac{n-1}{2}}\pi^{\frac{n+1}{2}}}{\Gamma(\frac{n-1}{2})(4\pi)^{\frac{n-1}{2}}}.
\end{equation}
It provides a 
$\textbf{definition}$ of 
$((t\pm i0)^2-\vert\overrightarrow{\xi}\vert^2)^{-1}$ 
as an oscillatory integral or 
Lagrangian distribution in $\mathbb{R}^{n+1}$,  
\begin{equation}\label{oscilldef}
C_n\int_{\mathbb{R}^n} d^n\overrightarrow{\xi} e^{i\phi_{\pm}(t,\overrightarrow{x};\overrightarrow{\xi}) }\frac{1}{|\overrightarrow{\xi}|} 
\end{equation}
with phase function $\phi_\pm(t,\overrightarrow{x};\overrightarrow{\xi})
=\sum_{i=1}^n -x^i\xi_i \pm t\sqrt{\sum_1^n (\xi_i)^2} = -\overrightarrow{x}.\overrightarrow{\xi} \pm t\vert\overrightarrow{\xi}\vert$. 
The idea is to use the interpretation of $((t\pm i0)^2-\vert x\vert^2)^{-1}$ as an oscillatory integral
to compute $WF((t\pm i0)^2-\vert x\vert^2)^{-1}$.  
\begin{prop}\label{oscilldefi}
%Let us consider the identity $\frac{1}{2\pi}\int_{\mathbb{R}^3} d^3\theta e^{i(-x.\theta \pm t\vert\theta\vert) }\frac{1}{|\theta|}=((t\pm i0)^2-\vert\overrightarrow{\xi}\vert^2)^{-1}  $
%in the \textbf{sense of distributions}.
We claim  $$WF \left(C_n\int_{\mathbb{R}^n} d^n\xi e^{i(-x.\xi \pm t\vert\xi\vert) }\frac{1}{|\xi|}\right)=
\{ (0,0;\pm |\xi|,-\overrightarrow{\xi}) \}\cup  \{(\vert x\vert,x_i;\pm \lambda,-\frac{\lambda x_i}{\vert x\vert} ) \vert  \lambda > 0, \vert x\vert\neq 0 \}.$$
\end{prop}
\begin{proof}
This computation can be found 
in \cite{RS} example 7 p.~101.
The function 
$\phi=t\vert\overrightarrow{\xi} \vert - x.\overrightarrow{\xi}$.
satisfies 
the axioms of 
a phase function 
because it is homogeneous of degree $1$ in $\xi$, 
smooth outside $\overrightarrow{\xi}=0$ and $d_{x,t}\phi$ 
never vanishes as soon as $|\overrightarrow{\xi}|\neq 0$ 
which implies
that it defines 
a \textbf{phase function} 
in the sense of H\"ormander. 
We first compute the critical set of 
$\phi$ denoted by 
$\Sigma_\phi$ and defined by the equation $\{d_\xi\phi=0\}$:
$$d_\xi (t\vert\overrightarrow{\xi} \vert - x.\overrightarrow{\xi})=t\sum_{\mu=1}^n\frac{\xi_\mu}{\vert\xi\vert}d\xi_\mu-x^\mu d\xi_\mu=0\Leftrightarrow t=\vert x \vert, x^\mu=\frac{\xi_\mu}{\vert\xi\vert}\vert x\vert .$$
We will later see in Chapter 6 that 
this 
defines 
a Morse family $$\left((\mathbb{R}^{n}\setminus\{0\})\times \mathbb{R}^{n+1}\mapsto\mathbb{R}^{n+1}, \phi\right)$$ and the wave front set is parametrized
by the Lagrange immersion $\lambda_\phi\Sigma_\phi$ 
in $T^\star\mathbb{R}^{n+1}$ 
of the critical set defined by the Morse family:
$$\lambda_\phi\Sigma_\phi=\{ (t,x;\partial_t\phi,\partial_x\phi)| \partial_\xi\phi=0 \}$$ 
$$=\{ (x=0,t=0;|\xi|,-\xi ) \}\cup \{(t,x;\vert\xi\vert,-\xi) \vert t=\vert x\vert,x^\mu=\frac{\xi_\mu}{\vert\xi\vert}\vert x\vert,\xi\neq 0 \} $$
$$=\{ (0,0;|\xi|,-\xi ) \}\cup  \{(\vert x\vert,x^\mu;\lambda,-\frac{\lambda x^\mu}{\vert x\vert} ) \vert  \lambda > 0, \vert x\vert\neq 0 \}.$$
To conclude, we see that the sign in front of $t$ in the phase $\phi_{\pm}(t,x;\xi)=\pm t\vert\xi\vert-x.\xi $ 
will decide of the positivity or negativity of the energy of $WF(\Delta_+)$.
\end{proof}
\section{The holomorphic family $\left((x^0+i0)^2-\sum_{i=1}^n (x^i)^2\right)^s$.}
 We give a detailed derivation of the main steps needed for the computation of the wave front set of the family $\left((x^0+i0)^2-\sum_{i=1}^n (x^i)^2\right)^s,s\in\mathbb{C}$ 
and $\log\left((x^0+i0)^2-\sum_{i=1}^n (x^i)^2\right)$ 
using the general theory of boundary values 
of holomorphic functions along convex sets 
developped by H\"ormander \cite{Hormander}. 
The result is given in \cite{Hormander} p.~322 
without any detail, 
also a similar treatment in the literature can be found in \cite{Vladimirov}.
We carefully follow the exposition of \cite{Hormander} $(8.7)$ but we specialize to the simpler case of the quadratic form $Q=(x^0)^2-\sum_{i=1}^n (x^i)^2$ which makes the explanations 
much clearer 
and allows us to give direct arguments.
 
 Let $C^+$ denote the set 
$\{y| Q(y)>0,y^0>0 \}$, 
$C^+$ is an open cone called
the \emph{future cone}. 
We denote by $q$ the unique
symmetric 
bilinear map associated to the quadratic form
$Q$.
 
\paragraph{Microhyperbolicity.}
Given $\theta=(1,0,0,0)$.
We recall that $Q$ is said to be microhyperbolic (see definition 8.7.1 in \cite{Hormander}) w.r.t. $\theta$ in an open set $\Omega\subset\mathbb{R}^n$ if $\forall x\in\mathbb{R}^n,\exists t(x)>0$, 
such that $\forall t, 0<t<t(x)$, $Q(x+it\theta)\neq 0$.
\begin{prop}\label{muhyp}
The quadratic form $Q(x)=(x^0)^2-\sum_{i=1}^n (x^i)^2 $ is microhyperbolic with respect to any vector $\theta\in C^+$. 
\end{prop}
\begin{proof}
We are supposed first to fix a vector $\theta\in C^+$, and we must check $Q$ is microhyperbolic with respect to $\theta$. In fact, we prove a stronger result:
$\forall x,\forall\varepsilon>0$, $Q(x+i\varepsilon\theta)\neq 0$.
If $Q(x+i\varepsilon\theta)=Q(x)-\varepsilon^2Q(\theta)+2i\varepsilon q(x,\theta) =0$ then the imaginary part $\text{Im}Q(x+i\varepsilon\theta)=0$ must vanish hence $q(x,\theta)=0$.
Hence we would have $Q(x)\leqslant 0$ since $\theta\in C^+$ and finally $Q(x+i\varepsilon\theta)=Q(x)-\varepsilon^2Q(\theta)<0$. Contradiction !
\end{proof}
%Microhyperbolicity implies the lemma 
%\begin{lemm}
%$\forall t\neq 0$, we have
%\begin{eqnarray}
%F(x+it\theta)\neq 0
%\\ F(\theta)\neq 0
%\end{eqnarray} 
%\end{lemm}
%\begin{proof}
%The result of lemma $8.7.2$ applies to our case since $Q$ is homogeneous and $\theta$ microhyperbolic.
%\end{proof}
The domain $T^C=\mathbb{R}^{n+1}+ iC^+$ is called a \emph{tube cone}.
We want to define the limits in the sense of distributions $\lim_{y\rightarrow 0,y\in C^+}Q^s(x+iy)$
of the holomorphic function $Q^s$.

\subsubsection{The Vladimirov approach.}

 In the Vladimirov approach which is similar to H\"ormander's, we have to prove
$Q^s$ is slowly increasing in the algebra $\mathcal{O}(T^C)$ of functions holomorphic
in the tube cone $T^C$ ( see \cite{Vladimirov}). 
In fact, in our case, we would have to prove an estimate of the form
\begin{equation}
\forall z=x+iy, \vert Q^s(z)\vert\leqslant \left(1 + d(y,\partial C^+)^{-2Re(s)} \right).
\end{equation}
where $d(y,\partial C^+)$ is defined as the distance beetween $y\in C^+$ and the boundary $\partial C^+$ of the future cone.
Then we know (see Theorem 4 p.~204 in \cite{Vladimirov}) that 
the Fourier Laplace transform $\mathcal{F}$ 
is an algebra isomorphism 
from $\left(\mathcal{O}(T^C),\times\right)$
to the algebra 
$\left(\mathcal{S}^\prime(C^\circ),\star\right)$ 
of tempered distribution
supported
in the dual cone $C^\circ\subset \mathbb{C}^4$
endowed with the \textbf{convolution product}.
However, both the H\"ormander and Vladimirov approaches 
rely on an estimate
which roughly says the holomorphic function $Q^s(z)$
has moderate growth when 
the imaginary part 
$y$ of $z$ tends to zero
in the Tube cone $T^C$.

\paragraph{Stratification of the space of zeros.}

 For a fixed point $x_0\in\mathbb{R}^{n+1}$, we study the Jets of the map $x\mapsto Q(x)$ at the point $x_0$. The Minkowski space $\mathbb{R}^{n+1}$ is partitioned by the \textbf{lowest order of homogeneity} of the Taylor expansion of $Q$.
Lojasiewicz describes this construction as the \textbf{stratification} of the space $\mathbb{R}^{n+1}$ by the \textbf{orders of the zeros} of $Q$.  
We study the Taylor expansion of $Q$ at $x_0$ by looking at the map $y\mapsto Q(x_0+y)$. We find three distinct situations:
\begin{itemize}
\item $Q(x_0)\neq 0$ thus $ Q(x_0+y)=q(x_0,x_0)+O(\vert y\vert)$, the term of lowest homogeneity is $q(x_0,x_0)$ and is homogeneous of degree $0$ in $y$
\item $Q(x_0)=0,x_0\neq 0$ thus $ Q(x_0+y)=2q(x_0,y)+O(\vert y\vert^2)$, the term of lowest homogeneity is $2q(x_0,y)$ and is homogeneous of degree $1$ in $y$
\item $x_0=0$ thus $ Q(0+y)=q(y,y)+O(\vert y\vert^3)$, the term of lowest homogeneity is $q(y,y)$ and is quadratic hence homogeneous of degree $2$ in $y$.
\end{itemize}
Following H\"ormander, we denote by 
$Q_{x_0}(y)$ the term of lowest homogeneity in $y$. 
The term of lowest homogeneity allows 
to construct a geometric structure over $\mathbb{R}^{n+1}$ called the tuboid. 
\paragraph{Construction of the tuboid.}
To 
every $x_0\in\mathbb{R}^{n+1}$, we associate the cone $\Gamma_{x_0}$ (\cite{Hormander} Lemma 8.7.3 ) defined as the connected component of
\begin{equation}
\{y| Q_{x_0}(y)\neq 0  \}
\end{equation} 
which contains the vector $\theta=(1,0,0,0)$. 
\begin{lemm}
Let $Q=(x^0)^2-\sum_{i=1}^n (x^i)^2$ and $\theta=(1,0,0,0)$.  For every $x_0\in\mathbb{R}^{n+1}$, let $\Gamma_{x_0}$ be the cone defined as above.
\begin{itemize}
\item If $Q(x_0)\neq 0$ then $\Gamma_{x_0}=\{y|Q_{x_0}(y)\neq 0 \}=\mathbb{R}^{n+1}$ since the term of lowest homogeneity $Q(x_0)$ does not depend on $y$.
\item If $Q(x_0)=0,x_0\neq 0$ then $\{y|Q_{x_0}(y)\neq 0 \}=\{y|q(x_0,y)\neq  0 \}=\{y|q(x_0,y)>0 \}\bigcup \{y|q(x_0,y)<0 \} $ contains two connected components the upper and lower half spaces associated to $Q(x_0,.)$, $\Gamma_{x_0}=\{y|q(x_0,\theta)q(x_0,y)>0 \}$.
\item If $x_0=0$ then $\Gamma_{x_0}=\{y|q(y,y)>0, y_0>0 \}$, it is the space of all \textbf{future oriented timelike} vectors.
\end{itemize}
\end{lemm}
%The cone $\Gamma_{x_0}$ is important because it allows to take boundary values of holomorphic functions for $x$ is close to $x_0$ and $y\rightarrow 0\in\Gamma_{x_0}$.
The domain 
$\Lambda=\{x_0+i\Gamma_{x_0}| x_0\in\mathbb{R}^{n+1}\}\subset \mathbb{C}^4$ 
is called
a \emph{tuboid} in the
terminology
of Vladimirov. 
\paragraph{Choice of the branch of the $\log$ function.} 
In order to define the complex powers $Q^s(x+iy)=e^{s\log Q(x+iy)}$ and $\log Q(x+iy)$, we must specify the branch of the $\log$ function that we use.
We choose the branch of the $\log$ in the domain $0 < \arg Q(z)< 2\pi$, for $Q=(x^0)^2-\sum_{i=1}^n (x^i)^2$. For this determination of the $\log$ (see \cite{Joshi} Proposition 4.1),
by the proof
of Proposition \ref{muhyp}, 
we see that $Q(x+i\varepsilon\theta)$ avoids the positive reals .
\begin{prop}\label{trivialcase}
$\lim_{\varepsilon\rightarrow 0}\log Q(.+i\varepsilon\theta)$ converges to a smooth function in the \textbf{nonconnected open set} $Q\neq 0$. 
\end{prop} 
\begin{proof}
We are going to prove that $\lim\log Q(.+i\varepsilon\theta)\in C^\infty(\{Q\neq 0\})$.
We notice that the set $\{Q(x_0)\neq 0\}$ contains three open connected domains,
and we classify the convergence of $\log Q(.+i\varepsilon\theta)$ on each of these
connected domains:  
\begin{eqnarray}
Q(x_0)<0\implies \forall x\in U_{x_0}, \log Q(x+i\varepsilon\theta)\rightarrow \log \vert Q(x)\vert +i\pi\\
Q(x_0)>0, x_0^0>0, \implies \forall x\in U_{x_0}, \log Q(x+i\varepsilon\theta)\rightarrow \log \vert Q(x)\vert\\ 
Q(x_0)>0, x_0^0<0, \implies \forall x\in U_{x_0}, \log Q(x+i\varepsilon\theta)\rightarrow \log \vert Q(x)\vert  +2i\pi.
\end{eqnarray} 
Thus for every $x_0$ such that $Q(x_0)\neq 0$, there is a small neighborhood of $x_0$ in which the family of analytic functions $\log Q(.+i\varepsilon\theta)$ converges uniformly to a \textbf{smooth} function.
\end{proof} 
We only have to study the case $Q(x_0)=0$.
%  Furthermore the conic set $\{(x+iy)|y\in\Gamma_x \}\subset \mathbb{R}^{3+1}+i\mathbb{R}^{3+1}$ will play the role of the wave front set, actually we will later identify this set with $\{(x,y)|y\in\Gamma_x \}\subset T\mathbb{R}^{3+1}$ and the wave front set will be the dual of this set.
% We shall present now a more geometric discussion of the situation. We first define an analog of a ball bundle $\mathbb{B}_\varepsilon\mathbb{R}^{3+1}=\{(x,y)| x\in\mathbb{R}^{3+1},\Vert y\Vert=\varepsilon \}$ over $\mathbb{R}^{3+1}$.
% 
% We consider $$F:(x,y)\mapsto \left(Q(x,x)-Q(y,y),2Q(x,y) \right)$$ 
%as an algebraic map.
%
% Our goal is to find a conic neighborhood of $\{(x,\tau\theta)| x\in\mathbb{R}^{3+1},\tau\leqslant\varepsilon\}$ in such a way that the range of the algebraic map $F$ avoids the set 
%$$\{Q(x,x)-Q(y,y)\geqslant 0, Q(x,y)=0 \} $$
%which is equivalent to the requirement that $Q(x+iy)$ does not meet the real half line $\mathbb{R}_+$.% For $y\in \Gamma_{x}$ we would like that $Q(x+iy)$ avoids the half positive real line.
\subsubsection{The moderate growth estimate along $T^C$.}
H\"ormander proved 
an important estimate in \cite{Hormander} lemma $8.7.4$ which is a specific case of the celebrated Lojasiewicz inequality. We have to slightly modify his result, actually we prove the estimate of lemma $8.7.4$, plus the property that
$Q(x+iy)$ never meets the positive half line $\mathbb{R}_+$ for $x,y$ in appropriate domains.
Let $\theta=(1,0,0,0)$.  For every $x_0\in\mathbb{R}^{n+1}$ such that $Q(x_0)=0$, let $\Gamma_{x_0}$ be the cone defined as the connected component of
\begin{equation}
\{y| Q_{x_0}(y)\neq 0  \}
\end{equation} 
which contains the vector $\theta$.
\begin{prop}\label{proploja}
For any closed conic subset $V_{x_0}\subset \Gamma_{x_0}$, there exists $\delta,\delta^\prime>0$ and
$U_{x_0}$ is a neighborhood of $x_0$ such that for all 
$(x,y)\in U_{x_0}\times V_{x_0},\vert y\vert\leqslant \delta$ 
the following estimate is satisfied:
\begin{equation}
\exists m\in\mathbb{Z}, \delta^\prime \vert y\vert^m\leqslant \vert Q(x+iy)\vert
\end{equation}  
and $Q(x+iy)$ does not meet $\mathbb{R}_+$.
\end{prop}

\begin{proof}
We fix $x_0$. 
We also prove that we can choose $U_{x_0}$ 
in such a way that $U_{x_0}\times V_{x_0}$ 
tends to $\{x_0\}\times \Gamma_{x_0}$ 
for any net of cones $V_{x_0}$ which converges to $\Gamma_{x_0}$.
We study the two usual cases:

 if $Q(x_0)=0,x_0\neq 0$, any closed cone $V_{x_0}$ contained in 
$$\Gamma_{x_0}=\{y|q(x_0,\theta)q(x_0,y)>0 \}$$ 
should be contained in 
$$\{y|q(x_0,\theta)q(x_0,y)\geqslant 2\delta \vert y\vert \}$$ 
for some $\delta>0$ small enough (when $\delta\rightarrow 0$ we recover $\Gamma_{x_0}$).
Let us consider
the continuous map 
$f:=x\mapsto \inf_{y\in V_{x_0},\vert y\vert=1}q(x_0,\theta)q(x,y)$.
By definition of $V_{x_0}$, $f(x_0)\geqslant 2\delta$ 
therefore the set
$f^{-1}[\delta,+\infty)$ 
contains a neighborhood of $x_0$.
%w.l.o.g. we may assume that $\vert y\vert=1$ thus 
%for all $y\in V_{x_0},\vert y\vert=1$, $\sup_{y\in V_{x_0},\vert y\vert=1}q(x_0,\theta)q(x_0,y)\geqslant 2\delta$. ,
We set 
$U_{x_0}=f^{-1}[\delta,+\infty)=\{x | \forall y\in V_{x_0}, q(x_0,\theta)q(x,y)\geqslant \delta \vert y\vert \}$,
then $U_{x_0}$ is a neighborhood of $x_0$.
It is immediate by definition of $U_{x_0}$ 
that for all $(x,y)\in U_{x_0}\times V_{x_0}$, 
we have $\vert q(x,y)\vert
\geqslant \delta\vert q(x_0,\theta)\vert^{-1} \vert y\vert$
which is 
the moderate growth estimate and we also find that 
$\text{Im }Q(x+iy)=2q(x,y)$ never vanishes. 
Thus $Q(x+iy)$ avoids $\mathbb{R}_+$.

 If $x_0=0$ then $ \Gamma_{0}=\{y|q(y,y)>0, y_0>0 \}$ 
is the space of all \textbf{future oriented timelike} vectors. 
If we set $y=t\theta,\theta=(1,0,0,0)$, we find that 
\begin{equation}
\forall x, \vert Q(x+iy)\vert\geqslant \vert Q(y)\vert 
\end{equation}
in fact the unique \textbf{critical point} of the map $(x,t)\mapsto Q(x+it\theta)$ is the point $x=0$.
But then this inequality is invariant by the group $O_{+}^\uparrow(n,1)$ of time and orientation preserving Lorentz transformations. Thus the previous estimate is true for any $y\in \Gamma_{0}$ and reads:
\begin{equation}
\forall x,\forall y\in \Gamma_0,  \vert Q(x+iy)\vert\geqslant \vert Q(y)\vert .
\end{equation}
To properly conclude, we use the fact that $y$ is contained in a closed subcone $V_{0}$ of the interior future cone $q(y,y)>0$, thus there is a constant $\delta<1$ such that
$$(x,y)\in K \implies  \sum_{i=1}^n (y^i)^2\leqslant \delta (y^0)^2 $$ 
this implies the estimates
$$\sum_{\mu=0}^n (y^\mu)^2=(y^0)^2+\sum_{i=1}^n (y^i)^2
\leqslant (1+\delta)(y^0)^2 \implies  (y^0)^2\geqslant\frac{\sum_{\mu=0}^n (y^\mu)^2}{1+\delta} $$
and also the estimate $\forall (x,y)\in U_{0}\times V_{0},$ where $U_0=\vert x\vert\leqslant\delta$: $$q(y,y)=(y^0)^2-\sum_{i=1}^n (y^i)^2\geqslant (y^0)^2-\delta (y^0)^2\implies q(y,y)\geqslant (1-\delta)(y^0)^2 $$
finally, combining the two previous estimates gives 
$$ \frac{(1-\delta)\sum_{\mu=0}^n (y^\mu)^2}{1+\delta} \leqslant q(y,y) ,$$ 
which yields the inequalities, $\forall (x,y)\in U_{0}\times V_{0}$:
\begin{equation}\label{keyineq}  
\frac{(1-\delta)\sum_{\mu=0}^n (y^\mu)^2}{1+\delta} \leqslant q(y,y)\leqslant\vert Q(x+iy)\vert,
\end{equation} 
setting $\delta^\prime=\frac{1-\delta}{1+\delta}$ 
proves the claim.
\end{proof}

\begin{coro}
Thus for all $y\in\Gamma_x$, $\log Q(x+iy)$ and $Q^s(x+iy)$ are well defined analytic functions of the variable $z=x+iy$ for the \textbf{branch} of the $\log$: $0<\arg Q(z)<2\pi $. 
\end{coro}
The tube cone $T^C$ is $O(n,1)_+^{\uparrow}$ \textbf{invariant} thus our arguments would be still valid for any vector $\theta$ in the orbit of $(1,0,0,0)$ by $O(n,1)_+^{\uparrow}$. Thus all results of proposition \ref{proploja} are \textbf{independent} of the choice of $\theta$ in the open cone $Q(\theta)>0,\theta^0>0$.
The key inequality (\ref{keyineq}) also 
appears in a less precise form 
in the proof 
of Proposition 4.1 p.~352 in \cite{Joshi}.
\paragraph{Partial results by the Vladimirov approach.}
In the course of the proof of proposition (\ref{proploja}), 
we rediscovered the Lorentz invariant
inequality $ \forall z=x+iy\in T^C, \vert Q(z)\vert\geqslant \vert Q(y)\vert $. We notice that
$\forall y\in C, Q(y)=2\Delta^2(y)$ where $\Delta(y)=\left(\frac{(y^0)^2-\vert y\vert^2}{2} \right)^{\frac{1}{2}}$ is the Euclidean distance beetween $y$ and the boundary of $C$. Immediately, we deduce that for $Re(s)\leqslant 0$:
$$\vert\left(Q(z)\right)^s\vert \leqslant (2\Delta^2(y))^{Re(s)}\leqslant M(s) (1+\Delta^{2Re(s)}(y)),$$
this means $Q^s$ is in the algebra $H(C)$ of slowly increasing functions in $O(T^C)$ (where $O(T^C)$ is the algebra of holomorphic functions in $T^C$). 
Application of theorems of Vladimirov proves 
the existence of a boundary value $\underset{y\rightarrow 0, z=x+iy\in T^C}{\lim}Q^s(z)$
in the space of tempered distributions 
when $y\rightarrow 0$ in $C$. 
The limit is understood as a tempered distribution and also the Fourier transform of $Q^s$ is a tempered distribution in $\mathcal{S}^\prime(C^\circ)$ which is the algebra for 
the convolution product of Schwartz distributions supported on the dual cone $C^\circ$ of $C$.
In the terminology of Yves Meyer, the boundary value $Q^s(.+i0\theta)$ is $C^\circ$ \emph{holomorphic}.
\paragraph{Existence of the boundary value as a distribution.}  
The previous estimates allow us to prove a moderate growth property which is the requirement to apply Theorems 3.1.15 and 8.4.8 in \cite{Hormander} giving existence of Boundary values and control of the wave front set:
\begin{prop}
For any closed conic subset $V_{x_0}\subset \Gamma_{x_0}$, there exists a sufficiently small neighborhood $U_{x_0}$ of $x_0$ such that for all $x+iy\in U_{x_0}+ iV_{x_0},\vert y\vert\leqslant \delta$, 
\begin{eqnarray}
\vert\log(Q(x+iy))\vert \leqslant \frac{C}{\vert y\vert}\\   
\vert Q^s(x+iy)\vert\leqslant C\vert y\vert^{2Re(s)}
\end{eqnarray}
\end{prop}
Thus the hypothesis of theorem $3.1.15$ of \cite{Hormander} are satisfied for $\log(Q(z)),Q^s(z)$.
\begin{proof} 
Since $\forall (x,y)\in U_{x_0}\times V_{x_0}, 0<\vert y\vert\leqslant \delta$, we have $Q(x+iy)\notin\mathbb{R}_+$, then we must have $\log Q(x+iy)= \log\vert Q(x+iy)\vert + i \text{arg}\left(Q(x+iy)\right) $
where $0<\text{arg}(Q)<2\pi$ which implies $\vert \log Q(x+iy)\vert < \log\vert Q(x+iy)\vert +2\pi $.
Recall that we have estimates of the form $$\forall (x,y)\in U_{x_0}\times V_{x_0}, 0<\vert y\vert\leqslant \delta, \delta\vert y\vert^m \leqslant \vert Q(x+iy)\vert $$
We can assume without loss of generality that $0<C\vert y\vert^m<1$ and $\vert Q(x+iy)\vert \leqslant 1$. 
%We can only restrict to this case when taking boundary values because when $y\rightarrow 0$ in the cone $\Gamma$, if we had $\lim_{y\rightarrow 0}\vert Q(x+iy)\vert\geqslant 1$ then there would be no extension problem, the distributions $\log Q,Q^s$ would already exist as we proved in proposition \ref{trivialcase} !
Then we have 
$$\forall (x,y)\in U_{x_0}\times V_{x_0}, 0<\vert y\vert\leqslant \delta, \delta\vert y\vert^m \leqslant \vert Q(x+iy)\vert  \implies\vert Q^s(x+iy)\vert\leqslant \left(\delta\vert y\vert^m\right)^{Re(s)} $$
for $Re(s)\leqslant 0$.
And also 
$\forall (x,y)\in U_{x_0}\times V_{x_0}, 0<\vert y\vert\leqslant \delta, \delta\vert y\vert^m \leqslant \vert Q(x+iy)\vert$ $$ \implies \log\delta\vert y\vert^m\leqslant\log\vert Q(x+iy)\vert\implies \vert\log\vert Q(x+iy)\vert\vert \leqslant \vert\log\left(\delta\vert y\vert^m\right)\vert.$$
Thus we find 
$$\vert\log Q(x+iy)\vert \leqslant 2\pi+ \vert\log \delta\vert + m\vert\log(\vert y\vert)\vert .$$
\end{proof} 
\begin{coro}\label{fundacoro}
Application of Theorem 3.1.15 in \cite{Hormander} implies 
$Q^s(.+i0y)$ and $\log Q(.+i0y)$ for $y\in\Gamma$ are both well defined on $\mathbb{R}^{n+1}$ as boundary values of holomorphic functions.
\end{coro}
The proof that $Q^s(.+i0y)$ defines a 
tempered distribution is only
sketched in \cite{Joshi} Proposition 4.1 and it is proved in \cite{KKK}
in example 2.4.3 p.~90 that 
these are hyperfunctions 
in the sense of Sato but this is not enough
to prove these are genuine distributions.
Notice that the existence and definition of the boundary values $Q^s(.+i0y)$ and $\log Q(.+i0y)$ \textbf{does not depend} on the choice of $y$ provided $y$ lives in the open cone $C^+$, 
but since this cone is $O(n,1)_+^{\uparrow}$ invariant, the distributions $Q^s(.+i0y)$ and $\log Q(.+i0y)$ are $O(n,1)_+^{\uparrow}$ invariant. 
\paragraph{The wave front set of the boundary value.}
\begin{thm}
The wave front set of $Q^s(.+i0\theta)$ and $\log Q(.+i0\theta)$
is contained in the set:
\begin{equation}
\{(x;\tau dQ) \vert \tau x^0>0, Q(x)=0 \}\bigcup \{(0;\xi) \vert Q(\xi,\xi)\geqslant 0,\xi_0>0 \}.
\end{equation}
\end{thm}
\begin{proof}
We want to apply Theorem 8.7.5 in \cite{Hormander} 
in order to obtain the result explained in \cite{Hormander} on p.~ 322.
More precisely, we want to apply Theorem 8.4.8 of \cite{Hormander} 
which gives the wave front set of boundary values of holomorphic functions. 
Application of Theorem 8.4.8 
of \cite{Hormander} claims that for each point $x_0$ such that $Q(x_0)=0$, 
$$WF(\log Q(U_{x_0}+i0V_{x_0}))\subset U_{x_0}\times V^\circ_{x_0}$$ 
where $V^\circ_{x_0}=\{\eta |\forall y\in V_{x_0}, \eta(y)\geqslant 0\}$ is the dual cone of $V_{x_0}$. 
But since this upper bound is true for \textbf{any} closed subcone $V_{x_0}\subset \Gamma_{x_0}$ and corresponding neighborhood $U_{x_0}$ containing $x_0$, by picking an \emph{increasing} family $V_{\delta,x_0}=\{y| q(x_0,y)\geqslant 2\delta\vert y\vert \}$ and the corresponding \emph{decreasing} family of neighborhoods $U_{x_0,\delta}=\{x|\forall y\in V_{\delta,x_0} ,\vert q(x,y)\vert \geqslant  \delta\vert y\vert, \vert x-x_0 \vert\leqslant \delta \}$, when $\delta \rightarrow 0$, we find that the wave front set of the boundary value over each point $x_0$ should be contained in the \textbf{dual cone} $\Gamma_{x_0}^\circ=\{\eta | \forall y\in \Gamma_{x_0},\eta(y)\geqslant 0 \}$ of $\Gamma_{x_0}$. Our job consists in determining this \textbf{dual cone} $\Gamma_{x_0}^\circ$ for all $x_0$  such that $Q(x_0)=0$ ie in the singular support of $Q^s(.+i0\theta)$.
As usual there are two cases: $Q(x_0)=0,x_0\neq 0$ and $x_0=0$. 

 For $Q(x_0)=0,x_0\neq 0$, consider the cone 
\begin{equation} 
\{y |q(x_0,y)\neq 0 \} 
\end{equation} 
this cone contains $\emph{two connected components}$ separated by the hyperplane $H=\{y |q(x_0,y)=0 \}$, we should set $\Gamma_{x_0}$ equal to the connected component which contains $\theta$, $$\Gamma_{x_0}=\{y |  q(x_0,y)q(x_0,\theta)>0 \}.$$ However, since $q(x_{0},\theta)=x_0^0$ and $dQ_{x_0}(y)=q(x_0,y)$, it is much more convenient to reformulate $\Gamma_{x_0}$ as the half space
\begin{equation}
\Gamma_{x_0}=\{y |\eta(y)>0 \}, \eta=x^0_{0}dQ_{x_0}
\end{equation}
for the linear form $y\mapsto \eta=x^0_{0}dQ_{x_0}(y)$.
By definition, this half space is the convex enveloppe of the linear form $\eta$ thus
the dual cone $\Gamma^\circ_{x_{0}}$ of the half space $\Gamma_{x_0}$ consists in 
the positive scalar multiples of the linear form $\eta$ generating this half space, 
finally $\Gamma^\circ_{x_{0}}=\{\tau dQ_{x_0} |  \tau x_0^0>0 \}$.

 When $x_{0}=0$, consider the cone 
\begin{equation} 
\{y |q(y,y)\neq 0 \} 
\end{equation} 
this cone contains $\emph{three connected components}$ depending on the sign of $Q$ and $y^0$, we should set $\Gamma_0$ equals to the connected component which contains $\theta$:
\begin{equation} 
\Gamma_0=\{y |q(y,y)>0,y^0>0 \}.
\end{equation} 
By a straightforward calculation 
$$\Gamma_0^\circ=\{\eta |\forall y\in\Gamma_0,\eta(y)\geqslant 0 \}=\{\eta | Q(\eta)\geqslant 0, \eta^0\geqslant 0 \}, $$
which is the future cone in dual space. 
Finally, $$WF\log Q(.+i0\theta)\subset \left(\underset{x_{0}\neq 0,Q(x_0)=0}{\bigcup} \Gamma_{x_{0}}^\circ\right)\bigcup \Gamma_0^\circ $$
and we have the same upper bound for $WF Q^s(.+i0\theta)$.
\end{proof}
The proof of this theorem cannot be found in physics textbooks and is not even sketched in \cite{Hormander} (where it is only stated 
as an example of direct application of Theorem 8.7.5 in \cite{Hormander}). 
A nice consequence of theorems proved in this section is that it makes sense of \textbf{complex powers} of the Wightman function $\Delta_+$.
Our work differs from the work of Marcel Riesz 
because the Riesz family $\square^s$ 
does not have the right wave front set, 
for all $s$ $\square^s\neq\Delta_+^s$,
actually $\square^{-1}$ 
is a \textbf{fundamental solution} 
of the wave equation whereas the 
Wightman function $\Delta_+$ is an actual 
\textbf{solution} of the wave equation.
  
\section{Pull-backs and the exponential map.}
\paragraph{The moving frame.}\label{movingframe}
Let $(M,g)$ be a pseudo-Riemannian
manifold and $TM$ its tangent bundle.
We denote by $(p;v)$ an element of $TM$, 
where $p\in M$ and $v\in T_pM$.
Let $\mathcal{N}$ be 
a neighborhood
of the zero section
$\underline{0}$ in $TM$
for which 
the 
map
$(p;v)\in\mathcal{N}\mapsto (p,\exp_p(v))\in M^2 $
is a local
diffeomorphism 
onto its image
($\exp_p:T_pM\mapsto M$ is the exponential geodesic map).
Thus the subset 
$\mathcal{V}=\exp\mathcal{N}\subset M^2$
is a neighborhood
of $d_2$ and
the inverse map
$(p_1,p_2)\in \mathcal{V}\mapsto (p_1;\exp_{p_1}^{-1}(p_2))\in\mathcal{N}$
is a well defined diffeomorphism.
Let
$\Omega$ be an open subset of $M$
and
$(e_0,...,e_n)$ be \textbf{an orthonormal moving frame}
on $\Omega$ 
($\forall p\in \Omega , g_p(e_\mu(p),e_\nu(p))=\eta_{\mu\nu}$), 
and $(\alpha^\mu)_\mu$ 
the corresponding 
orthonormal moving coframe.
\paragraph{The pull-back.}
We denote by
$\epsilon_\mu$ the
canonical basis of
$\mathbb{R}^{n+1}$,
then
the data of the orthonormal moving coframe $(\alpha^\mu)_\mu$
allows to define the submersion
\begin{equation}\label{applipullback}
F:=(p_1,p_2)\in \mathcal{V}\mapsto 
F^\mu(p_1,p_2)\epsilon_\mu=
\underset{\in T_{p_1}^\star M}{\underbrace{\alpha^\mu_{p_1}}}\underset{\in T_{p_1}M}{\underbrace{(\exp_{p_1}^{-1}(p_2))}}\epsilon_\mu\in\mathbb{R}^{n+1}. 
\end{equation}
For any distribution $f$ in
$\mathcal{D}^\prime(\mathbb{R}^{n+1})$, 
the composition 
$$(p_1,p_2)\in \mathcal{V}\mapsto f\circ F(p_1,p_2)=f\circ\left(\alpha^\mu_{p_1}(\exp_{p_1}^{-1}(p_2)) \epsilon_\mu\right)$$
defines the pull-back of $f$ on $\mathcal{V}\subset M^2$.
If $f$ is $O(n,1)_+^\uparrow$ invariant, 
then
the pull-back defined as above
\textbf{does not depend on the choice
of orthonormal moving frame} $(e_\mu)_\mu$ 
and is thus \textbf{intrinsic}
(since all orthonormal moving frames 
are related by gauge transformations 
in $C^\infty(M,O(n,1)_+^\uparrow)$).
We apply this construction to the family 
$Q^{s}(h+i0\theta)\in \mathcal{D}^\prime\left(\mathbb{R}^{n+1}\right)$ 
constructed in Corollary (\ref{fundacoro}) 
as boundary value of holomorphic functions,
and we obtain 
the distribution
$(p_1,p_2)\in\mathcal{V}\mapsto Q^s\circ  \left(\alpha^\mu_{p_1}(\exp_{p_1}^{-1}(p_2)) \epsilon_\mu\right)$.
This 
allows to \textbf{canonically}
pull-back
$O(n,1)_+^\uparrow$ invariant 
distributions to distributions
defined on a neighborhood
of $d_2$.
\begin{ex}\label{Gammaex}
The quadratic function $Q(h)=h^\mu\eta_{\mu\nu} h^\nu$ is $O(n,1)_+^\uparrow$ invariant
in $\mathbb{R}^{n+1}$.
The pull back of $Q$ by $F$ on $\mathcal{V}$ gives 
$$Q\circ F(p_1,p_2)= \alpha^\mu_{p_1}(\exp_{p_1}^{-1}(p_2))\eta_{\mu\nu} \alpha^\nu_{p_1}(\exp_{p_1}^{-1}(p_2))$$
which is the ``square of the pseudodistance'' between
the two points $(p_1,p_2)$ called 
Synge's world function in the physics literature. 
Following \cite{Hadamard}, 
we will denote this function  
by $\Gamma(p_1,p_2)$.
\end{ex}
\subsection{The wave front set of the pull-back.}
We compute the 
wave front set
of $Q^s\circ F$.
\paragraph{The expression 
of $WF(Q^s(.+i0\theta))$ in terms
of $\eta_{\mu\nu}$.}
Notice that $WFQ^s(.+i0\theta)$ can be written in the form:
\begin{equation}\label{WFindices}
WFQ^s(.+i0\theta)=\{ (h^\mu; \lambda\eta_{\mu\nu}h^\nu )| Q(h)=0,h^0\lambda>0  \}\cup \{ (0;k)| Q(k)\geqslant 0,k_0>0 \},
\end{equation}
where 
the condition $h^0\lambda>0$ 
plays an important role in
ensuring that the momentum  
$\lambda\eta_{\mu\nu}h^\nu$ 
has \textbf{positive energy}.
\paragraph{The pull-back theorem of H\"ormander in our case.}
Denote by $t$ the distribution
$Q^s(.+i0\theta)$.
An application of the 
pull-back 
theorem (\cite{Hormander} Theorem 8.2.4) in 
our situation
gives
\begin{equation}\label{pulledbackWF}
WF(F^{\star }t)\subset \{(p_1,p_2; k\circ d_{p_1}F,k\circ d_{p_2}F) | (F(p_1,p_2),k)\in WF(t)\} 
\end{equation}
We denote by $(p_1,p_2;\eta_1,\eta_2)$ an element
of $T^\star \mathcal{V}\subset T^\star M^2$ and
$(h^\mu;k_\mu) $ the coordinates in $T^\star \mathbb{R}^{n+1}$.
The pull-back with 
indices reads:
$$(p_1,p_2; k\circ d_{p_1}F,k\circ d_{p_2}F)=(p_1,p_2; k_\mu d_{p_1}F^\mu,k_\mu d_{p_2}F^\mu).$$
\paragraph{Step 1, we first 
compute
$WF(F^\star t)$
outside
the set
$d_2=\{p_1=p_2\}$.}
The condition 
$(F(p_1,p_2),k)\in WF(t)$ in (\ref{pulledbackWF}), 
reads by (\ref{WFindices})
$(F^\mu(p_1,p_2);k_\mu)=(F^\mu(p_1,p_2);\lambda\eta_{\mu\nu}F^\nu(p_1,p_2))$.
We obtain $$(p_1,p_2; \lambda k\circ d_{p_1}F,\lambda k\circ d_{p_2}F)=(p_1,p_2;\lambda F^\mu\eta_{\mu\nu_2}d_{p_1}F^{\nu_2},\lambda F^\mu\eta_{\mu\nu_2} d_{p_2}F^{\nu_2})$$
and also 
$F^\mu(p_1,p_2)\eta_{\mu\nu}F^\nu(p_1,p_2)=0$.
Now set
$\Gamma(p_1,p_2)=F^\mu(p_1,p_2)\eta_{\mu\nu}F^\nu(p_1,p_2)$.
The key
observation is that
$d_{p_1}\Gamma=2F^\mu\eta_{\mu\nu}d_{p_1}F^\nu$ and $d_{p_2}\Gamma=
2F^\mu\eta_{\mu\nu}d_{p_2}F^\nu$, 
hence:
$$WF(F^{\star }t)\subset \{(p_1,p_2; \lambda d_{p_1}\Gamma,\lambda d_{p_2}\Gamma) | \Gamma(p_1,p_2)=0, \lambda F^{0}(p_1,p_2)>0 ,\lambda\in\mathbb{R}  \}$$
$$\cup  \{(p_1,p_2; k\circ d_{p_1}F,k\circ d_{p_2}F) | p_1=p_2, Q(k)\geqslant 0,k_0>0 \}.$$
\paragraph{The geometric interpretation of the last formula.}
\begin{defi}
A distribution $t\in \mathcal{D}^\prime\left(M^2\right)$
satisfies the Hadamard condition,
if and only if
$WF(t)\subset \{(p_1,p_2;-\eta_1,\eta_2) | (x_1;\eta_1)\sim (x_2;\eta_2),\eta_2^0>0 \}$.
\end{defi}
Our convention for the 
Hadamard
condition is the opposite
of the convention of Theorem
3.9 p.~33 in \cite{Junker}.
The Hadamard condition
is a condition on the wave front set
of a distributional bisolution
of the wave equation
which ensures it represents
a quasi free state 
of the free quantum field
theory in curved space time
(\cite{Junker}).

The function $\Gamma$ is the pseudo Riemannian analogue 
of the square geodesic distance and will be discussed 
in paragraph (\ref{Gammapart}).
We first interpret the term 
$$\{(p_1,p_2; \lambda d_{p_1}\Gamma,\lambda d_{p_2}\Gamma) | \Gamma(p_1,p_2)=0, \lambda F^{0}(p_1,p_2)>0   \}$$ 
appearing in the last formula
as the subset
of all elements
in $T^\star \mathcal{V}$
of the conormal bundle
of the conoid $\{\Gamma=0\}$ 
such 
that 
$(\eta_2)_0$ has
\textbf{constant} sign:
this is exactly the 
\textbf{Hadamard condition}.
If we use the metric 
to lift the indices,
$d_{p_1}\Gamma\left(e_\mu(p_1)\right)\eta^{\mu\nu}e_\nu(p_1)$
and $d_{p_2}\Gamma\left(e_\mu(p_2)\right)\eta^{\mu\nu}e_\nu(p_2)$
are the Euler vector fields $\nabla_1\Gamma,\nabla_2\Gamma$ defined by Hadamard. 
We will later prove in proposition (\ref{Hadamardsimpleform}) 
that the vectors 
$\nabla_1\Gamma,-\nabla_2\Gamma$ 
are parallel along the null geodesic
connecting $p_1$ and $p_2$, proving
$(d_{p_1}\Gamma,-d_{p_2}\Gamma)$ are in fact \textbf{coparallel}
along this null geodesic.
\paragraph{Step 2, ``Diagonal''.}
For any function
$F$ on $M^2$, we
uniquely decompose the total differential
in two
parts
as follows
$$dF=d_{p_1}F+d_{p_2}F,\text{ where }d_{p_1}F|_{\{0\}\times T_{p_2}M}=0,d_{p_2}F|_{T_{p_1}M\times\{0\}}=0.$$
Let $i$ be the inclusion
map $i:=p\in M\mapsto (p,p)\in d_2\subset M$
then 
$\forall p\in M, F\circ i(p)=0\implies 
d_p F\circ i=0 \Leftrightarrow d_{p_1}F\circ di+d_{p_2}F\circ di=0$.
Since
$$d_{p_2}F^\mu(p,p)=d_{p_2}\alpha^\mu_{p_1}\left(\exp_{p_1}^{-1}(p_2)\right)|_{p_1=p_2=p}=
\alpha^\mu_{p_1}\left(d_{p_2}\exp_{p_1}^{-1}(p_2)\right)|_{p_1=p_2=p}=\alpha^{\mu}(p),$$ because $d_{p_2}\exp^{-1}_{p_1}(p_2)|_{p_1=p_2=p}=Id_{T_pM\mapsto T_pM}=e_{\mu}(p)\alpha^{\mu}(p)$.
Thus
$d_{p_1}F^\mu(p,p)=-\alpha^\mu(p)$ and
$$ \{(p_1,p_2; k\circ d_{p_1}F,k\circ d_{p_2}F) | p_1=p_2, Q(k)\geqslant 0,k_0>0 \}$$
$$=\{(p,p; -k_\mu \alpha^\mu(p),k_\mu \alpha^\mu(p)) | p\in M, Q(k)\geqslant 0,k_0>0\}.$$
Then summarizing
step 1 and step 2, 
let us denote by $\Lambda\subset T^\bullet \left(M^2\setminus d_2\right)$ the 
conormal bundle of the set $\{\Gamma=0\}$ with the zero section removed:
\begin{thm}\label{Wavefrontpullback}
The wave front set 
of the distributions 
$Q^s(\cdot+i0\theta )\circ F$ and 
$\log Q(\cdot+i0\theta) \circ F$
is contained in
\begin{equation}
\left(\Lambda\bigcup \{(p,p;-\eta,\eta)|g_p(\eta,\eta)\geqslant 0\}\right) \bigcap \{(p_1,p_2;\eta_1,\eta_2)| \eta_2^0>0\},
\end{equation}
where $\Lambda\subset T^\bullet \left(M^2\setminus d_2\right)$ is the conormal of $\{\Gamma=0\}$ with the zero section removed.
\end{thm}
Remarks:\\
a)If we denote by $\overline{\Lambda}$ the closure
of the conormal $\Lambda\subset T^\bullet \left(M^2\setminus d_2\right)$ 
in $T^\bullet M^2$,
then $\left(\Lambda\bigcup \{(p,p;-\eta,\eta)|g_p(\eta,\eta)\geqslant 0\}\right)=\overline{\Lambda}+\overline{\Lambda}$.

b) $\{(p,p;-\eta,\eta)|g_p(\eta,\eta)\geqslant 0\}$ is contained in the conormal $(Td_2)^\perp$
of $d_2$.

\begin{coro}
The families $Q^s(.+i0\theta)\circ F$ and $\log Q(.+i0\theta)\circ F$ 
satisfy the 
\textbf{Hadamard condition}.
\end{coro}
\paragraph{Discussion of the sign convention for the energy.}
We want to discuss some sign conventions.
Recall that if $(h;k)\in WF(Q(.+i0\theta)^s)$ (resp $WF(Q(.-i0\theta)^s)$) then $k$ has positive (resp negative) energy.
Denote $(p_1,p_2;\eta_1,\eta_2)$ an element of the wave front set of $F^\star Q^s(\cdot\pm i0\theta)$.
If we want $\eta$ to be a covector of \textbf{positive energy} (resp negative energy), then we must consider the distribution 
$F^{\star}Q^s(.+i0\theta)$ (resp $F^{\star}Q^s(.-i0\theta)$).
 
 Notice that in the physics literature, the boundary value is determined using a Cauchy hypersurface determined by a function 
$T:M\mapsto \mathbb{R}$: 
$$\left(\Gamma(p_1,p_2)+i\varepsilon (T(p_1)-T(p_2))+\varepsilon^2\right)^s.$$ 
The proof that it defines 
a well defined distribution is never given 
and the wave front set of this boundary value was never computed. 
Furthermore, 
the formula is not obviously covariant 
since it relies on the existence of 
a foliation of space-times by Cauchy hypersurfaces. 
\subsection{The pull back of the phase function.} 
In order to connect with the interpretation
of the wave front set in terms of Lagrangian
manifold, we 
imitate 
what we did for $((x^0\pm i0)^2-\sum_{i=1}^n (x^i)^2)^{-1}$, 
we pull-back the oscillatory integral representation on $\mathcal{V}\subset M^2$ by the smooth map $F$. 
\begin{thm} 
The distribution $F^{*}\left(Q(.+i0\theta) \right)^{-1}$ 
is the Lagrangian distribution given by the formula
$$C_n\int_{\mathbb{R}^n} d^n\xi e^{i\left(\phi_{\pm}\circ F\right)(p_1,p_2;\xi) }\frac{1}{|\xi|},$$ 
this Lagrangian distribution 
with phase function $\phi_{\pm}\circ F$ 
has a wave front set which satisfies the Hadamard condition.
\end{thm} 
\begin{proof}
Let us only sketch the proof. First we use Proposition (\ref{oscilldefi}) to determine
the wave front set of the oscillatory integral $C_n\int_{\mathbb{R}^n} d^n\xi e^{i\phi_{\pm}(h;\xi) }\frac{1}{|\xi|} $. It is the same wave front set as for $((h^0\pm i0)^2-\sum_1^n (h^{i})^2)^{-1}$, then we apply the pull-back theorem of
H\"ormander in order to define the wave front set on the curved space and it exactly follows the
same proof as for the pull back theorem (\ref{Wavefrontpullback}).
\end{proof}
\section{The construction of the parametrix.}
Our parametrix construction is based on the work of Hadamard \cite{Hadamard} (see also \cite{Duistermaathad}).
The construction is done in the neighborhood 
$\mathcal{V}$ of $d_2$. Recall by \ref{applipullback} that $F(p_1,p_2)=e_{p_1}^\mu\left(\exp^{-1}_{p_1}(p_2) \right)\epsilon_\mu$.
\paragraph{The Hadamard expansion.}
We construct the parametrix locally in $\mathcal{V}$ by successive approximations. 
Inspired by the flat case, we look for an expansion of the form $$\Delta_+= U(p_1,p_2) \left(Q^{-1}\circ F \right)(p_1,p_2) $$
$$ +\sum_{k=0}^\infty V_k(p_1,p_2)\Gamma^k(p_1,p_2)\left(\log Q\circ F\right)(p_1,p_2)$$
where $\Gamma(p_1,p_2)=Q\circ F$ 
is the square of the pseudodistance and 
each term of the asymptotic expansion 
has an intrinsic meaning. 
\subsection{The meaning of the asymptotic expansions.}
Our goal is to construct $U,V_k$ in $C^\infty(\mathcal{V})$. 
First, we would like to make an important remark. 
The series 
$\sum_k V_k\Gamma^k$
does not usually converge.
However, 
we can still make sense 
of the asymptotic expansion 
$\sum_k V_k\Gamma^k$ 
as the asymptotic expansion of the \textbf{composite function} $V(.,.;\Gamma)$ in $C^\infty(\mathcal{V}\times\mathbb{R})$
where only the germs of map $r\mapsto V(.,.;r)$
at $r=0$ are defined 
($V$ is not uniquely defined).

\paragraph{The Borel lemma.}
\begin{prop}
For any sequence of smooth functions $\left(V_k\right)_k$ in $\left(C^\infty(\mathcal{V})\right)^{\mathbb{N}}$, there exists a smooth function $r\mapsto V(.,.;r)$ in $C^\infty(\mathcal{V}\times \mathbb{R})$ such that the coefficients of the Taylor series in the variable $r$ of $V$ is equal to the sequence $V_k$:
\begin{eqnarray} 
V_k(p_1,p_2)=\frac{1}{k!}\frac{\partial^kV}{\partial r^k}(p_1,p_2;0).
\end{eqnarray} 
\end{prop}
\begin{proof}
The proof is an application of the idea of the proof of the Borel lemma 
which states that any sequence $(a_k)_k$ can be realized as the Taylor series of a smooth function at $0$.
The proof we give is due to Malgrange \cite{Malgrange}. 
Let $\Omega\subset \mathcal{V}$ be an open subset with compact closure, 
then $\sup_{\Omega}\vert V_k\vert= a_k<\infty$. 
Let $\chi(r)$ be a cut-off function near $r=0$, $\chi=1$ in a neighborhood of zero and vanishes when $r\geqslant 1$.
We fix any 
sequence $b_k$, s.t. 
$b_k>0$ 
growing sufficiently fast 
such that $\forall k, \sup_{r\in\mathbb{R}^+,\alpha\leqslant k-1}  \vert\partial_r^\alpha a_k\chi(rb_k)r^k\vert\leqslant\frac{1}{2^k} $.
Then $\sum V_k\chi(\frac{r}{b_k})r^k$ 
is a smooth function 
whose Taylor coefficients are the $V_k$. 
%Actually notice $\forall p,\chi(\frac{r}{b_k})r^p\leqslant b_k^p$
%thus if $b_k=\frac{1}{2}(1+a_k)^{-1}$, we find
%$$\forall r, \sum a_k\chi(\frac{r}{b_k})r^k \leqslant \sum a_k\frac{1}{2^{k}}(1+a_k)^{-k}  \leqslant 2+ a_0$$
%and the series converges absolutely. Furthermore 
%$$\forall N,\forall r , \sum_0^\infty U_k\chi(\frac{r}{b_k})r^k =\sum_0^N U_k r^k +R_{N+1}$$ where the remainder $$\vert R_{N+1}\vert\leqslant r^{N+1}\sum_{0}^\infty a_{k+N+1}\chi(\frac{r}{b_{k+N+1}})r^k=r^{N+1}\left(a_{N+1}+\sum_{1}^\infty a_{k+N+1}\chi(\frac{r}{b_{k+N+1}})r^k \right)$$ 
%$$\leqslant r^{N+1}\left(a_{N+1}+\sum_{1}^\infty a_{k+N+1}\chi(\frac{r}{b_{k+N+1}})(\frac{1}{2}(1+a_{k+N+1})^{-1})^k \right)\leqslant r^{N+1}(a_{N+1}+2)  $$
%and we are almost done.
The series $\sum V_k\chi(\frac{r}{b_k})r^k$ 
is bounded and defines 
a smooth function 
only on the set 
$\Omega$.
Let $(\varphi_j)_{j\in J}$ 
be a collection
of compactly supported
functions in $M^2$
such that $\sum_{j=J} \varphi_j=1$ in a 
neighborhood of $d_2$ and vanishes
outside $\mathcal{V}$.
For each $j\in J$,
since $\text{supp }\varphi_j$
is compact 
the previous construction
gives us a sequence 
$(b_{k_j})_{k_j}$.
This gives us a final series 
$U=\sum_{j\in J,k\in\mathbb{N}} \varphi_jV_k \chi(\frac{r}{b_{kj}})r^k  $
which is a smooth function
supported in $\mathcal{V}$
such that
$$V(.,.;\Gamma)= \sum_{j\in J,k\in\mathbb{N}} \varphi_jV_k \chi(\frac{\Gamma}{b_{kj}})\Gamma^k\sim \sum_{k\in\mathbb{N}} V_k\Gamma^k.$$
\end{proof}
This remark cannot be found in any physics textbook. It is given in \cite{Friedlander} Lemma 4.3.2.
Finally, if we know the sequence of coefficients $V_k$, we find a function $V$ such that 
$V(p_1,p_2;r)=\sum V_k(p_1,p_2)r^k$, thus $V(p_1,p_2;\Gamma)$ is a well defined smooth function. 
\subsection{The invariance properties of the Beltrami operator $\square^g$ and of gradient vector fields.}
Let $(M,g)$ be a pseudo Riemannian manifold
and let us define
the Dirichlet energy
$\mathcal{E}\left(u;g\right)$
by the equation: 
\begin{equation} 
\mathcal{E}\left(u;g\right)=\int_M \frac{1}{2}\left\langle\nabla u,\nabla u \right\rangle_{g}d\text{vol}_g.
\end{equation} 
We will follow the
exposition of H\'elein (see \cite{Helein}) 
and
define the Beltrami operator 
$\square^g$ for a general metric $g$
by the first variation of the Dirichlet energy:
\begin{equation} 
\delta \mathcal{E}\left(u,g\right)\left(\varphi\right)=\int_M \left\langle\nabla u,\nabla \varphi \right\rangle_{g}d\text{vol}_g=-\int_M \left(\square^g u\right) \varphi d\text{vol}_g,
\end{equation}
(see \cite{Helein} Equation (1.5) p.~3).

\paragraph{The operator $\square^{g}$.}
Let $\Phi$ be a diffeomorphism of $M$, and 
$$\Phi: (M,\Phi^\star g) \mapsto (M,g) $$
the associated isometry,
then
the Dirichlet energy satisfies
the invariance equation by 
the action of diffeomorphisms:
$\forall\Phi \in Diff(M), \mathcal{E}\left(u;g\right)=\mathcal{E}\left(u\circ \Phi;\Phi^*g\right)$ (see \cite{Helein} p.~18-19 for the proof).
Thus the Beltrami operator $\square^g$ 
obeys the equation
\begin{equation}\label{functorialbox}
\forall\Phi \in \text{Diff}(M), \left(\square^g u\right)\circ \Phi =\square^{\Phi^\star g}\left(u\circ \Phi\right)
\end{equation}
\paragraph{The gradient operator $\nabla^g$.}
We want to prove that gradient vector fields w.r.t.
the metric $g$ also behave in a natural way.
Let $f\in C^\infty(M)$
then
\begin{eqnarray}\label{functorialgrad}
\forall\Phi \in \text{Diff}(M), \forall f\in C^\infty(M), 
\nabla^{\Phi^\star g} \left(f\circ\Phi \right)=\Phi^\star\left(\nabla^g f\right) \\
\left\langle \nabla^{g} f,\nabla^{g} f  \right\rangle_g=\left\langle \nabla^{\Phi^\star g} \left(f\circ\Phi\right),\nabla^{\Phi^\star g} \left(f\circ\Phi\right)\right\rangle_{\Phi^\star g} 
\end{eqnarray}
The first equation is equivalent
to the equation 
$\Phi_\star\left(\nabla^{\Phi^\star g} \left(f\circ\Phi \right)\right)
=\nabla^g f$ (\cite{Lee} p.~92--93).
We use the coordinate convention:
$$\Phi:x^\alpha\in (M,\Phi^\star g) \mapsto \phi^\gamma(x)\in (M,g)$$
We start from the definition:
$$ \nabla^{\Phi^\star g} \left(f\circ\Phi \right) = \left(g^{\gamma\delta}\frac{\partial x^\alpha}{\partial \phi^\gamma}\frac{\partial x^\beta}{\partial \phi^\delta}\right)\circ\Phi \frac{\partial \left(f\circ\Phi\right)}{\partial x^\alpha} \frac{\partial}{\partial x^\beta}$$ $$=\left(g^{\gamma\delta}\frac{\partial x^\alpha}{\partial \phi^\gamma}\frac{\partial x^\beta}{\partial \phi^\delta}\frac{\partial f}{\partial \phi^\mu}\right)\circ\Phi \frac{\partial \phi^\mu}{\partial x^\alpha}\frac{\partial}{\partial x^\beta} =\left(g^{\gamma\delta}\frac{\partial x^\beta}{\partial \phi^\delta}\frac{\partial f}{\partial \phi^\gamma}\right)\circ\Phi\frac{\partial}{\partial x^\beta}  $$
then we push-forward this vector field
$$\Phi_\star\left(\nabla^{\Phi^\star g} \left(f\circ\Phi \right)\right)=\left(g^{\gamma\delta}\frac{\partial x^\beta}{\partial \phi^\delta}\frac{\partial f}{\partial \phi^\gamma}\circ\Phi\right)\circ\Phi^{-1}\frac{\partial \phi^\mu}{\partial x^\beta}\frac{\partial}{\partial \phi^\mu} $$
$$=g^{\gamma\delta}\frac{\partial f}{\partial \phi^\gamma}\frac{\partial}{\partial \phi^\delta}=\nabla^g f  $$
The proof of the second identity 
can be simply deduced from the first one 
and one can also 
look at \cite{Helein} p.~19 for a similar proof.
In the sequel, we write $\nabla$ instead of 
$\nabla^g$ where
it will be obvious
we take the gradient
w.r.t. the intrinsic metric $g$
which 
does not depend on the chart
chosen. 
Recall that
we denote by 
$e_{\mu}$ 
the orthonormal 
moving frame
on $M$.
We define two gradient operators $\nabla_1,\nabla_2$ on $M^2$ as follows:
\begin{eqnarray}
\forall f\in C^\infty(M^2), \nabla_1 f(p_1,p_2)=d_{p_1}f\left(e_{\mu}(p_1)\right)\eta^{\mu\nu}e_{\nu }(p_1)\\
\forall f\in C^\infty(M^2), \nabla_2 f(p_1,p_2)=d_{p_2}f\left(e_{\mu}(p_2)\right)\eta^{\mu\nu}e_{\nu}(p_2).
\end{eqnarray}
\paragraph{The exponential map and lifting on tangent spaces.}  
Let us justify microlocally 
the philosophy of the Hadamard construction 
which consists in treating $Q^{-1}\circ F$ 
and $\log Q\circ F$ as distributions of $p_2$ 
where $p_1$
is viewed as a parameter:
let $f\in\mathcal{D}^\prime(\mathcal{V})$ 
be any distribution in $\mathcal{V}\subset M^2$.
We fix $p_1\in M$, then
we can make sense of
the restriction of $f$,
$f(p_1,.):=p_2\in M\mapsto f(p_1,p_2)$
as a distribution
on $\{p_1\}\times M$
if 
$$\text{Conormal }\left(\{p_1\}\times M\right)\bigcap WF(f)=\emptyset.$$
Let $\pi_{1}$ be the projection $\pi_1:=(p_1,p_2)\in M^2\mapsto p_1\in M$, 
if we have 
$$\forall p_1\in M, \text{Conormal }\left(\{p_1\}\times M\right)\bigcap WF(f)=\emptyset,$$
then for 
any test density $\omega\in\mathcal{D}^{n+1}(M)$,
the map $\pi_{1\star}\left(f\omega\right)$ 
defined by: 
$$\left[\pi_{1\star}\left(\omega f \right)=p_1\mapsto 
\underset{\text{partial integration}}{\underbrace{\int_M \omega(p_2)f(p_1,p_2)}}\right]$$ 
is \textbf{smooth} since 
$WF\left(\pi_{1\star}f\right)=\emptyset$
by Proposition 1.3.4 in \cite{DuistermaatFIO}. 
These conditions are
satisfied in our case
since the wave front set 
of $Q^{-1}\circ F$ and $\log Q\circ F$
are \textbf{transverse} to the conormal of 
$\left(\{p_1\}\times M\right)$ by Theorem
\ref{pulledbackWF}.  
We pull back $f(p_1,.)$ on $\mathbb{R}^{n+1}$ 
by the map $E_{p_1}$ defined as follows:
$$E_{p_1}: (h^\mu)_\mu\in\mathbb{R}^{n+1}\mapsto E_{p_1}(h)=\exp_{p_1}(h^\mu e_\mu(p_1)))\in M.$$
The orthonormal frame $\left(e_\mu(p_1)\right)_{\mu}$ fixes 
the isomorphism beetween $T_{p_1}M$ and $\mathbb{R}^{n+1}$.
\subsection{The function $\Gamma$ and the vectors $\rho_1,\rho_2$.}
In the Hadamard construction, everything is expanded in powers of the function $\Gamma$ which is the ``square of the pseudo Riemannian distance''.
$\Gamma$ is a solution of the nonlinear equation (\ref{nonlinearequationHadamard}).
In the physics literature, the function $\Gamma$ is called Synge world's function but the definition and the key equation (\ref{nonlinearequationHadamard}) satisfied by $\Gamma$ can already by found in Hadamard (see the equation $(32)$ 
in \cite{Hadamard} and the 
Lam\'e Beltrami differential 
parameters for $\Gamma$).
\subsubsection{The function $\Gamma$.}\label{Gammapart}
We already defined 
the function 
$\Gamma(p_1,p_2)=\alpha^\mu_{p_1}(\exp_{p_1}^{-1}(p_2))\eta_{\mu\nu} \alpha^\nu_{p_1}(\exp_{p_1}^{-1}(p_2))$
in example (\ref{Gammaex}). In the following proposition,
we explain which differential equation this function satisfies. 
\begin{prop}\label{gammathm}
Let us define the function $$\Gamma(p_1,p_2)=\left\langle\exp^{-1}_{p_1}(p_2),\exp^{-1}_{p_1}(p_2)\right\rangle_{g_{p_1}}$$ in $\mathcal{V}\subset M^2$. 
Then $\Gamma$ satisfies the equation
\begin{equation}\label{nonlinearequationHadamard}
\forall p_1, \left\langle\nabla_2\Gamma,\nabla_2\Gamma\right\rangle_{g(p_2)}(p_2)=4\Gamma
\end{equation}
\end{prop}
\begin{proof}
Denote by $E_{p_1}^\star g$ the metric
in the geodesic exponential chart centered
at $p_1$. 
We give a purely pseudo Riemannian geometry proof of the claim. 
Since $\Gamma(p_1,p_2)=\left\langle\exp^{-1}_{p_1}(p_2),\exp^{-1}_{p_1}(p_2)\right\rangle_{g_{p_1}}$, we know that 
$$\forall p_1\in M,\forall h\in\mathbb{R}^{n+1},  E_{p_1}^\star\Gamma(p_1,\cdot)(h)=h^\mu\eta_{\mu\nu}h^\nu .$$
Then by equation (\ref{functorialgrad}), writing 
$\left(E^\star_{p_1}g\right)^{\mu\nu}(h)
=\left(E^\star_{p_1}g\right)^{\mu\nu}$
for shortness:
$$\forall p_1\in M,  \left\langle \nabla\Gamma(p_1,.),\nabla\Gamma(p_1,.)\right\rangle_g = \left\langle \nabla_2^{E_{p_1}^\star g}\left(E_{p_1}^\star\Gamma\right),\nabla_2^{E_{p_1}^\star g}\left(E_{p_1}^\star\Gamma\right)\right\rangle_{E_{p_1}^\star g} $$
$$=\left(E_{p_1}^\star g\right)^{\mu\nu}\partial_{h^\mu}(h^{\mu_1}\eta_{\mu_1\nu_1}h^{\nu_1})\partial_{h^\nu}(h^{\mu_2}\eta_{\mu_2\nu_2}h^{\nu_2})$$ $$=\left(E_{p_1}^\star g\right)^{\mu\nu} 2\delta_{\mu}^{\mu_1} \eta_{\mu_1 \mu_2}h^{\mu_2}2\delta_{\nu}^{\nu_1} \eta_{\nu_1\nu_2}h^{\nu_2} $$ 
$$=4(E_{p_1}^\star g)^{\mu\nu}(h)  (E_{p_1}^\star g)_{\mu \mu_2}h^{\mu_2} (E_{p_1}^\star g)_{\nu\nu_2}h^{\nu_2}$$ 
$$=4(E_{p_1}^\star g)_{\mu_2\nu_2}h^{\mu_2}h^{\nu_2}=4\eta_{\mu_2\nu_2}h^{\mu_2}h^{\nu_2} ,$$
by repeated application of the Gauss lemma: 
$(E_{p_1}^\star g)_{\mu \nu} h^\nu=\eta_{\mu\nu}h^\nu$. 
\end{proof}

\paragraph{The Euler fields defined by Hadamard.}
Once we have defined the geometric function $\Gamma$, 
we can define a pair of scaling vector fields:
\begin{defi}\label{Hadamdef}
Let $(p_1,p_2)\in \mathcal{V}\subset M^2$, we define the pair of vector fields
\begin{eqnarray}
\rho_2=\frac{1}{2}\nabla_1 \Gamma=d_{p_2}\Gamma(e_{\mu}(p_2))\eta^{\mu\nu}e_{\nu}(p_2)\\
\rho_1=\frac{1}{2}\nabla_2 \Gamma=d_{p_1}\Gamma(e_{\mu}(p_1))\eta^{\mu\nu}e_{\nu}(p_1).
\end{eqnarray}
\end{defi}
$\rho_1,\rho_2$ are Euler vector fields in the sense of Chapter $1$ for the diagonal $d_2\subset \mathcal{V}$.
The situation is reminiscent of Morse theory. If we freeze the variable $p_1$, the vector field $\rho_1=\frac{1}{2}\nabla_2 \Gamma$ is the \textbf{gradient} (w.r.t. $p_2$ and metric $g$) of the \textbf{Morse function} $p_2\mapsto \Gamma(p_1,p_2)$ which has a critical point at $p_1=p_2$.
The Hadamard equation (\ref{nonlinearequationHadamard}) takes the simple form
\begin{eqnarray}
\rho_2\Gamma(p_1,p_2)=\rho_1\Gamma(p_1,p_2)=2\Gamma(p_1,p_2)
\end{eqnarray}
thus $\Gamma$ is homogeneous of degree $2$ with respect to the geometric scaling defined by these Euler vector fields.
\paragraph{Useful relations beetween $\Gamma,\rho_2$ and $Q^s\circ F$.} 
Around $p_1$,
the manifold $M$ is locally parametrized by the map 
$E_{p_1}:h\in\mathbb{R}^{n+1}\mapsto \exp_{p_1}(h^\mu e_\mu(p_1))$.
$\rho_2=\nabla_2\Gamma$
is an Euler vector field
in $M$ and
we want to study
its pull-back 
$E^\star_{p_1}\rho_2$
by $E_{p_1}$.
\begin{prop}\label{Hadamardsimpleform} 
We have the identity 
$\forall p_1\in M,   E^\star_{p_1}\rho_2=2h^j\partial_{h^j}$ 
and this identity
is independent of the
choice
of orthonormal
moving frame.
\end{prop}
\begin{proof}
Denote by $E_{p_1}^\star g$ the metric
in the geodesic exponential chart centered
at $p_1$. 
By naturality (\ref{functorialgrad}), we have 
setting
$\left(E^\star_{p_1}g\right)^{\mu\nu}(h)
=\left(E^\star_{p_1}g\right)^{\mu\nu}$
$$ E^\star_{p_1}\rho_2= E^\star_{p_1}\left(\nabla_2\Gamma\right)=\nabla\left(E^\star_{p_1}\Gamma\right)$$
$$= (E^\star_{p_1}g)^{\mu\nu}\partial_{h^\mu}\left(\eta_{kl}h^kh^l\right)\partial_{h^\nu}
=(E^\star_{p_1}g)^{\mu\nu}\left(\eta_{kl}\delta_\mu^k h^l + \eta_{kl}h^k\delta_\mu^l\right)\partial_{h^\nu}$$
$$=2(E^\star_{p_1}g)^{\mu\nu}\eta_{\mu l}h^l\partial_{h^\nu}
=2(E^\star_{p_1}g)^{\mu\nu}(E^\star_{p_1}g)_{\mu l}h^l\partial_{h^\nu}
=2h^\nu\partial_{h^\nu} $$
by application of the Gauss lemma.  
\end{proof}
This proposition allows us to interpret $\frac{1}{2}\nabla_2\Gamma$ as the vector $\dot{\gamma}(1)$ where $s\mapsto \gamma(s)$ is the unique geodesic with boundary condition $\gamma(0)=p_1,\gamma(1)=p_2$:
in exponential chart, this geodesic is given 
by the simple equation
$t\mapsto \gamma(t)=th^j$
and for all $t$ the vector $\dot{\gamma}(t)=h^j\frac{\partial}{\partial h^j}$ is \textbf{parallel} along this geodesic. 
By symmetry of the whole construction, we can interchange the roles of $p_1$ and $p_2$ and we deduce that $\rho_1\in T_{p_1}M,-\rho_2\in T_{p_2}M$ are 
\textbf{parallel} vectors along $\gamma$ (see the same remark in \cite{Zelditch} p.~18). 
A similar proof can be found in \cite{Duistermaathad} Lemma 8.4.
 
 We denote by $\Gamma^s$ the distribution $F^{*}\left(\left(Q(.+i0\theta) \right)^s\right)$
and observe that $\forall n\in\mathbb{N}$, $\Gamma^n=F^\star Q^n$. 
\begin{prop} 
The relation 
\begin{equation}
\forall s\in\mathbb{R},\forall n\in\mathbb{N},\Gamma^n\Gamma^s=\Gamma^{n+s}
\end{equation}
holds 
in the distributional sense.
\end{prop}
\begin{proof}
Recall $E_{p_1}^\star\Gamma^s(h)=\left(Q(h+i0\theta) \right)^s$.
For $\varepsilon>0$,
$$Q^n(h) Q^s(h+i\varepsilon\theta) $$ 
$$=Q^{n+s}(h+i\varepsilon\theta) +
\left(\left(Q^n(h)-(Q(h+i\varepsilon\theta))^{n}\right)Q^{s}(h+i\varepsilon\theta) \right)$$ where $\left(\left(Q^n(h)-Q^{n}(h+i\varepsilon\theta)\right)Q^{s}(h+i\varepsilon\theta) \right)$ is an error term which converges weakly to zero when $\varepsilon\rightarrow 0$. Thus we should have $Q^n(h)\left(Q(h+i0\theta) \right)^s=\left(Q(h+i0\theta) \right)^{s+n}$ in the distributional sense.
\end{proof}
\subsection{The main theorem.}
We first prove a lemma which implies that $WF(\Delta_+)$
satisfies the soft landing condition.
\begin{lemm}\label{Xislc}
Let $\Xi$ be the wave front set of $F^{*}\left(\left(Q(\cdot+i0\theta) \right)^s\right)$
then $\Xi$ satisfies the soft landing condition. 
\end{lemm}
\begin{proof}
First note that, by Theorem \ref{Wavefrontpullback},
$\Xi\cap T_{d_2}^\star M^2$
is contained
in $\left(Td_2\right)^\perp$
and $T^\bullet M^2\setminus d_2\subset\Lambda$
hence it suffices
to prove 
that the conormal $\Lambda$ of
the conoid $\{\Gamma=0\}$ 
satisfies the soft landing condition. 
Let $p:x\in \Omega\mapsto p(x)\in M$ be a local parametrization
of $M$,
using the local diffeomorphism $(x,h)\in \Omega\times \mathbb{R}^{n+1}
\mapsto \left(p(x),\exp_{p(x)}(h^\mu e_{\mu}(p(x)))\right)\in \mathcal{V}$ (recall that $(e_\mu)_\mu$
is the orthonormal moving frame), 
we can parametrize 
the neighborhood $\mathcal{V}$ of $d_2$ 
with some neighborhood of $\Omega\times\{0\}$ 
in $\Omega\times \mathbb{R}^{n+1}$.
In coordinates $(x,h)$, the conoid is parametrized
by the simple 
equation $\eta_{\mu\nu}h^\mu h^\nu=0$, thus it is immediate
that its conormal 
$\{(x,h;0,\xi)| \eta_{\mu\nu}h^\mu h^\nu=0, \xi_\mu=\lambda\eta_{\mu\nu}h^\nu,\lambda\in\mathbb{R}\}$ 
satisfies the soft landing condition.
\end{proof}
From the previous Lemma, we deduce the main theorem of this chapter.
The motivation
for this theorem
is that it proves 
that the two point function
satisfies the 
hypothesis
of Theorem \ref{mainthm} 
of Chapter 3 which allows
us to initialize the inductive
proof of Chapter 6 of 
renormalizability of all
$n$-point functions.
We denote by
$\Gamma^{-1}$, $\log\Gamma$ the distributions $F^{*}Q^{-1}(\cdot+i0\theta)$, $F^{*}\log Q(\cdot+i0\theta)$. Recall for any open set $U$, 
$E_s^\mu(U)$ defined in \ref{defEsmuloc}
was 
the space of distributions 
microlocally weakly homogeneous
of degree $s$. 
\begin{thm}\label{delta+bounded}
For any pair $U,V$ of smooth functions in $C^\infty(\mathcal{V})$,
the distribution
$$U\Gamma^{-1}+V\log\Gamma$$ 
is in $E_{-2}^\mu(\mathcal{V})$.
\end{thm} 
\begin{proof}
Let $\rho$ be one of the Euler vector fields defined in (\ref{Hadamdef}).
For any pair $U,V$ of smooth functions in $C^\infty(\mathcal{V})$,
by Theorem \ref{thminvmuloc}, it suffices to prove that
the family 
$$\lambda^2e^{\log\lambda\rho*}\left(U\Gamma^{-1}+V\log\Gamma\right)_{\lambda}$$ 
is bounded in $\mathcal{D}^\prime_\Xi$.
First $\Gamma^{-1}$ is homogeneous of degree $-2$ w.r.t. scaling: $\lambda^2e^{\log\lambda\rho_2*}\Gamma^{-1}=\lambda^2\lambda^{-2}\Gamma^{-1}=\Gamma^{-1}$ and $\lambda e^{\log\lambda\rho_2*}\log\Gamma=\lambda\log\lambda^{-2}\Gamma=-2\lambda\log\lambda+\lambda\log\Gamma$.
Then from these equations, we deduce that the families
$\left(\lambda^2e^{\log\lambda\rho_2*}\Gamma^{-1}\right)_{\lambda\in(0,1]}$ and 
$\left(\lambda^2 e^{\log\lambda\rho_2*}\log\Gamma\right)_{\lambda\in(0,1]}$ are bounded in $\mathcal{D}^\prime_\Xi$. 
Finally, we use that $U,V$ being smooth, the families $(U_\lambda)_\lambda,(V_\lambda)_\lambda$ are bounded in the $C^\infty$ topology in the sense that on any compact set $K$, the sup norms of the derivatives of arbitrary orders of $(U_\lambda)_\lambda,(V_\lambda)_\lambda$ are bounded.
%This means that for any function $\varphi \in D(\mathbb{R}^{2d})$, $\forall N>0, \pi_{N}\left(U_\lambda\varphi\right)<\infty$.
We can conclude using 
the estimate  
\ref{corollaire en or wanted by Christian} 
of Theorem \ref{producttheoremeskin}
to deduce
$(\lambda^{2}U_\lambda \Gamma^{-1}_\lambda)_\lambda=(U_\lambda\Gamma)$ and $(\lambda^{2}V_\lambda\log\Gamma_\lambda)_\lambda=(\lambda^{2}V_\lambda \log\Gamma+2\lambda^{2}V_\lambda\log\lambda)$ are bounded in $\mathcal{D}^\prime_{\Xi}$.
%, for any $\varphi\in D(\mathbb{R}^{2d})$ such that $\varphi|_{\text{supp }\chi}=1$:
% $$\forall N, \Vert \lambda^{2}\Gamma_\lambda^{-1}U_\lambda\Vert_{N,V,\chi}=\Vert \Gamma^{-1} U_\lambda\varphi\Vert_{N,V,\chi} $$ $$ \leqslant C\pi_{2N}(U_\lambda\varphi) (\Vert \Gamma^{-1}\Vert_{N,W,\chi}+\Vert (1+\vert\xi\vert)^{-m} \widehat{\Gamma^{-1}\chi} \Vert_{L^\infty}\pi_m(\chi))<\infty $$ 
%where we use the fact that $\lambda^{2}\Gamma_\lambda^{-1}=\Gamma$. 
\end{proof} 
\begin{coro}
Consequently, if $\Delta_+-\left(U\Gamma^{-1}+V\log\Gamma\right)\in C^\infty(\mathcal{V})$ 
for some $U,V$ 
in $C^\infty(\mathcal{V})$ 
then $\Delta_+\in E_{-2}^\mu(\mathcal{V})$. 
\end{coro}
Then we can construct 
the Hadamard Riesz coefficients
from which we can deduce suitable $U,V$ 
(see the above discussion on the Borel lemma),
however this construction is really
classical and one can look at 
\cite{Zelditch} and \cite{Gar} 
Chapter 5.2 
for the construction
of these coefficients. 

\chapter{The recursive construction of the renormalization.}
\subsection{Introduction.}

This chapter deals with the construction 
of a perturbative quantum field theory 
using the algebraic formalism developed 
in (\cite{BrouderQFT},\cite{Borcherds}) 
and proves their 
renormalisability 
using all 
the analytical tools 
developped in the previous chapters. 
In the first part, 
we describe 
the Hopf algebraic formalism for 
QFT relying heavily on a paper by Christian Brouder \cite{BrouderQFT} and a paper by R. Borcherds \cite{Borcherds}. The end goal of this first part is the construction of the operator product of quantum fields denoted by $\star$. 
Then in the second part, we introduce the important concept of causality which allows 
to axiomatically define the time ordered product denoted by $T$: 
$T$ solves the causality equations
and  $T$ satisfies the Wick expansion property 
which is a Hopf algebraic formulation of the Wick theorem. 
Once we have a $T$-product, we can define quantities such as
$t_n=\left\langle 0| T\phi^{n_1}(x_1)\dots \phi^{n_k}(x_k)  | 0\right\rangle $ 
where $t_n$ is a distribution defined on configuration space $M^n$.
We prove that if $T$ satisfies 
our predefined 
axioms, 
then the collection of distributions $(t_I)_I$ 
indexed by finite subsets $I$ of $\mathbb{N}$ 
satisfies an equation which intuitively
says that
on the whole configuration space minus the 
thin diagonal $M^n \setminus d_n$, the 
distribution $t_n\in\mathcal{D}^\prime(M^n\setminus d_n)$ can be expressed in terms 
of distributions $(t_I)_I$ for 
$I\varsubsetneq \{1,\dots,n\}$.
However, this expression involves products of distributions, thus we prove a recursion theorem which states that these products of distributions are well defined and $t_n\in\mathcal{D}^\prime(M^n\setminus d_n)$ can be extended in $\mathcal{D}^\prime(M^n)$. 
This allows us to recursively construct all 
the distributions $t_n$ for all configuration spaces $(M^n)_{n\in\mathbb{N}}$.

\section{Hopf algebra, $T$ product and $\star$ product.}

In this part, we use the formalism of \cite{BrouderQFT}.
 
\subsection{The polynomial algebra of fields.}

\subsubsection{The Hopf algebra bundle over $M$.}

Let $M$ be a smooth manifold which represents space time.
% Our goal is to construct a \emph{Hopf algebra bundle} denoted 
%$M\times \underset{\mathbb{R}\oplus \mathbb{R}\phi\oplus...}{\underbrace{\mathbb{R}[\phi]}}$ 
%over the
%manifold $M$.
%\\
%\\
%\\
% At each point $x\in M$, the fiber of this bundle over $x$ is a \textbf{polynomial Hopf algebra} $$\mathbb{R}[\phi(x)]=\mathbb{R}\oplus\mathbb{R}\phi(x)\oplus...$$
%over the field $\mathbb{R}$.  
We will denote by $H=\mathbb{R}[\phi]$ 
the \textbf{polynomial algebra} in the 
\textbf{indeterminate} $\phi$ 
and we use the notation 
$\underline{H}$ 
for the trivial bundle $\underline{H}=M\times \mathbb{R}[\phi]$. 
The space of sections 
$$\Gamma\left(M,\underline{H}\right)$$ 
of this vector bundle 
will be denoted  
by the letter $\mathcal{H}$.
$\phi$ is a formal \textbf{indeterminate}
and we denote by 
$\underline{\phi}^n$
the section of $\underline{H}$
which is the 
\textbf{constant section 
equal} 
to $\phi^n$.
Any section of $\underline{H}$
(ie any element of $\mathcal{H}$) 
will be a finite combination
$\sum_{n<+\infty} a_n\underline{\phi}^n$
where $a_n\in C^\infty(M)$.
The space of section $\mathcal{H}$ is
a Hopf \textbf{module} over the \textbf{algebra}
$C^\infty(M)$.
Actually, most of the
theory of Hopf algebras
is still valid on rings
and does not require fields. 
In order to match with the physical convention, $\phi^n(x):=(x,\phi^n)$ 
denotes 
the section 
$\underline{\phi}^n=(x\mapsto \phi^n(x))$ 
evaluated at the point $x\in M$.
$\underline{1}$ 
is the unit section of this module $\mathcal{H}$.
 
The module $\mathcal{H}$ has an algebra and coalgebra structure, 
the product and coproduct of $\mathcal{H}$ are induced from the product and coproduct of $H$, for instance the product $\underline{\phi}_1\underline{\phi}_2$ of two sections is just the product computed fiber by fiber in $H$, and the coproduct $\Delta$ in $\mathcal{H}$ is just the fiberwise coproduct.

%  Thus we can use the formalism of Hopf algebras on flat space time. We define a polynomial algebra $\mathbb{R}[\phi]$: $1,\phi,...,\phi^n,...$ which consists in powers of the quantum field. Geometrically, we are considering a trivial vector bundle $M\times \mathbb{R}[\phi]$ with infinite dimensional fiber $\mathbb{R}[\phi]$. $\phi(x)$ should be written $(x,\phi)$ and refers to a section of this bundle evaluated at the point $x$.
%  
%  
%   The Hopf algebra $H$ of sections of $M\times \mathbb{R}[\phi]$ is $C^\infty(M)$-algebra.

\paragraph{The product.}

 The rule for the product is simple
$$\underline{\phi}^k \underline{\phi}^l=\underline{\phi}^{k+l} $$
which means that the sections $\underline{\phi}^k$ and $\underline{\phi}^l$  
multiply pointwise 
$$\phi^k(x)\phi^l(x)=\phi^{k+l}(x) $$

\paragraph{The coproduct.}

The coproduct on the primitive element $\underline{\phi}$ is given by: 
$$\Delta\underline{\phi}=\underline{1}\otimes_{C^\infty(M)}\underline{\phi}+\underline{\phi}\otimes_{C^\infty(M)} \underline{1}$$
and it can be extended to powers of the field $\underline{\phi}^n$ by the binomial formula:
$$\Delta \underline{\phi}^n=\sum_{k=0}^n \left(\begin{array}{c} n \\ k \end{array}\right)\underline{\phi}^k\otimes_{C^\infty(M)}\underline{\phi}^{n-k}$$

\paragraph{Some comments and the Sweedler notation.}
A special case of coassociativity will be: \begin{equation}\sum a_{(11)}\otimes a_{(12)}\otimes a_{(2)}=\sum a_{(1)}\otimes a_{(21)}\otimes a_{(22)}=\sum a_{(1)}\otimes a_{(2)}\otimes a_{(3)},\end{equation} 
in tensor notation this reads $\left(\Delta\otimes Id\right)\Delta=\left(Id\otimes \Delta\right)\Delta$
which justifies Sweedler's notation: $\Delta^{k-1}a=\sum a_{(1)}\otimes...\otimes a_{(k)}$.
\paragraph{The counit}

 The counit is the Hopf algebra analog of the vacuum expectation value
in QFT:

$\varepsilon((x,\phi^n))= \left\langle 0 |\phi^n(x) | 0 \right\rangle=\delta^n_0$.

\begin{defi}
The counit is a linear map $\varepsilon:H\mapsto C^\infty(M)$ 
which satisfies the following properties:

\begin{itemize}
\item $\varepsilon$ is an algebra morphism: $\varepsilon(ab)=\varepsilon(a)\varepsilon(b)$

\item $\varepsilon(\phi^n(x))=\delta_{n0}\underline{1}$. 
\end{itemize}
\begin{equation}\sum\varepsilon(a_1)a_2=\sum a_1\varepsilon(a_2)=a.\end{equation}
\end{defi}
We make the identification 
$\underline{\phi^0}=\underline{1}$.
\begin{ex}
We want to give an example of the defining equation $$\sum a_1\varepsilon(a_2)=a$$ 
for $a=\underline{\phi}^n$:
$\sum_{k=0}^n \left(\begin{array}{c} n \\ k \end{array}\right)\underline{\phi}^{n-k} \underset{=0\text{ if }k\neq 0}{\underbrace{\varepsilon(\underline{\phi}^{k})}}=\underline{\phi}^n\varepsilon(\underline{1})=\underline{\phi}^n .$ 
\end{ex}

\subsection{Comparison of our formalism and the classical 
formalism from physics textbooks.}

 In QFT textbooks, the fields $\phi$ are thought of as operator valued distributions.
In our formalism, the field $\phi$ is merely an \textbf{indeterminate}. 
In QFT textbooks, 
the noncommutative operator product is defined first 
and the operator product of two fields 
$\phi(x)$ and $\phi(y)$ is written $\phi(x)\phi(y)$. 
Then using the representation of $\phi$ in terms of 
annihilation and creation operators, 
physicists define the normal ordered product denoted by $:\phi(x)\phi(y):$ 
which corresponds to the commutative product of the Hopf module $\mathcal{H}$. 
Whereas in our formalism, 
we start from the commutative product and 
then use a procedure called twisting 
to define the operator product $\star$.

$$
\begin{array}{|c|c|c|c|}
\hline
&\text{Standard QFT} & \text{Our approach} & \text{Borcherds}\\
\hline
\text{Commutative product}&:\phi(x)\phi(y): & \phi(x)\phi(y) & \phi(x)\phi(y)\\ 
\hline
\text{``Operator product''}&\phi(x)\phi(y) & \phi(x)\star \phi(y) & \\
\hline
\text{VEV}& \left\langle 0|\,\ |0 \right\rangle  & \varepsilon & \\
\hline
\text{Correlation functions} & \left\langle 0|T\dots|0 \right\rangle & t=\varepsilon \circ T & \text{Feynman measure }\omega \\
\hline
 & &\text{Laplace coupling}\left(.|.\right)  & \text{Bicharacter } \Delta \\
\hline 
\end{array}
$$

\subsection{Hopf algebra bundle over $M^n$.}
A further step in the construction is to pass from the manifold $M$ to the configuration space $M^n$ of $n$ points.
In order to define products of quantum fields over $n$ points, it is natural to construct 
an algebraic setting on configuration space $M^n$.
We start again from
the fiber 
$H=\mathbb{R}[\phi]$ and consider the $n$-fold tensor product $H^{\otimes n}=\mathbb{R}[\phi]\otimes \dots \otimes \mathbb{R}[\phi]$.
Then $H^{\otimes n}$ can be generated as a polynomial algebra by the $n$ elements:
$$\phi\otimes 1\otimes\dots\otimes 1=\phi_1$$ $$1\otimes \phi\otimes 1 \otimes \dots=\phi_2$$ 
$$\dots $$
thus we deduce that $H^{\otimes n}\simeq\mathbb{R}[\phi_1,...,\phi_n]$. 
Then we denote by 
$\underline{H^{\otimes n}}$
the bundle $M^n\times\mathbb{R}[\phi_1,...,\phi_n] $ 
living over configuration space $M^n$. 
As we did in the previous part, we must consider a module over $C^\infty(M^n)$ which contains products of fields of the form
$$ \underline{\phi}^{k_1}\otimes \dots \otimes \underline{\phi}^{k_n}  ,$$ 
hence we will consider 
the $C^\infty(M^n)$-module 
of sections $\Gamma\left(M^n,\underline{H^{\otimes n}}\right)$. 
This module over the ring $C^\infty(M^n)$ will be denoted $\mathcal{H}^n$. 
Similarly, for any finite subset $I$ of the integers, let $M^I$ be
the set of maps from $I$ to $M$, we define
$$\underline{H^{\otimes I}}=M^I\times \mathbb{R}[\phi_i]_{i\in I}=
M^I\times \mathbb{R}[\phi_i]_{i\in I}.$$
Then $\mathcal{H}^I$ is defined
as the $C^\infty(M^I)$-module
of sections $\Gamma(M^I,\underline{H^{\otimes I}})$.
To consider $\mathcal{H}^n$ over the ring $C^\infty(M^n)$ is not sufficient since in QFT textbooks, 
the 
operator product of fields 
denoted by $\star$ generates 
distributions as we can see in the following example:
\begin{ex}
$\phi(x)\star\phi(y)=\Delta_+(x,y)+\phi(x)\phi(y).$
\end{ex} 
We will have to extend the ring $C^\infty(M^n)$ of smooth functions living on configuration space $M^n$ to a ring which contains \textbf{distributions}.
In order to include sections of $\underline{H}$ with distributional coefficients,
we use a tensor product technique.
This idea already appeared in the previous work of Borcherds \cite{Borvertex}, in which he 
constructs 
a vertex algebra with value in some sort of ring with singular coefficients. 
If we have an algebra $A$ of polynomials over a ring $R$ and $V$ a $R$-module,
it is always possible to define the tensor product $A\otimes_R V$ over the ring $R$.
Here we apply this construction: let $V$ be a left $C^\infty(M^n)$-module of distributions, then the tensor product $\mathcal{H}^n\otimes_{C^\infty(M^n)}V$ makes sense.
Warning: even if $\mathcal{H}^n$ is an algebra, it is no longer true that $\mathcal{H}^n\otimes_{C^\infty(M^n)}V$ is still 
an algebra since
we cannot always multiply distributions. 
\paragraph{The Rota Feynman convention.} 
Following Rota and Feynman, we write $\underline{\phi}_1\underline{\phi}_2$ instead of $\phi(x)\otimes\phi(y)$. 
We drop the tensor product symbol $\otimes$, and the elements of $\mathcal{H}^n$ are $C^\infty(M^n)$-linear combinations of products of powers of fields $\underline{\phi}_1^{i_1}\dots\underline{\phi}_n^{i_n}$.
Hence elements on the $j$-th factor of the tensor product is written $\underline{\phi}_j^{i_j}$.
Sometimes, to make our proofs look even simpler, we write $a_1...a_n$ instead of $\underline{\phi}_1^{i_1}\dots\underline{\phi}_n^{i_n}$.

\paragraph{Extending the product and coproduct.}

 To extend the product and coproduct to $\mathcal{H}^n$, we just compute products and coproducts "point by point".

\begin{defi}
We give the formula of the product for the generators of $\mathcal{H}^n$
$$\left(\underline{\phi}_1^{n_1}\dots \underline{\phi}_k^{n_k}\right)\left(\underline{\phi}_1^{l_1}\dots \underline{\phi}_k^{l_k}\right)$$
$$=\left(\underline{\phi}_1^{n_1+l_1}\dots \underline{\phi}_k^{n_k+l_k}\right) $$
and the formula of the coproduct: 
$$\Delta \left(\underline{\phi}_1^{n_1}\dots \underline{\phi}_k^{n_k}\right)=$$
$$\Delta\underline{\phi}_1^{n_1}\dots\Delta\underline{\phi}_k^{n_k} $$
\end{defi}

 Although the definition is given in terms of sections $\underline{\phi}^{n_i}_i$, 
we will sometimes follow the physics folklore and write $\phi^{n_i}(x_i)$.

\paragraph{Fundamental example}

 If we compute explicitly the coproduct for the generators, we obtain the formula:

$\Delta \left(\underline{\phi}_1^{n_1}\dots\underline{\phi}_k^{n_k}\right)=$
\begin{equation}
\sum \left(\begin{array}{c}n_1 \\ i_1 \end{array}\right)\dots\left(\begin{array}{c}n_k \\ i_k \end{array}\right)\underline{\phi}_1^{n_1-i_1}\dots\underline{\phi}_k^{n_k-i_k}\otimes\underline{\phi}_1^{i_1}\dots\underline{\phi}_k^{i_k}  \end{equation}

\paragraph{The counit and the vacuum expectation values.}
The counit is defined on $\mathcal{H}^n$ by extending the counit $$\varepsilon:\mathcal{H}\rightarrow C^\infty(M^n)$$ to $\mathcal{H}^n$ by coalgebra morphism: $$\varepsilon(uv)=\varepsilon(u)\varepsilon(v).$$
\begin{ex}
$$\varepsilon(\underline{1})=1$$
$$\varepsilon(\underline{\phi}_1\underline{\phi}^2_2\underline{1}_3)=
\varepsilon(\underline{\phi}_1)
\varepsilon(\underline{\phi}^2_2)\varepsilon(\underline{1}_3)=0\times 0\times 1=0 $$
$$\varepsilon(\underline{1}_1 \underline{1}_2\underline{1}_3)=1\times 1\times 1=1 $$
\end{ex}
It is the Hopf algebraic version of the vacuum expectation value and is an essential 
ingredient if one wants to define 
``correlation fonctions''
from 
product of fields.

\subsection{Deformation of the polynomial algebra of fields.}
\subsubsection{The non commutative product of QFT.}
\paragraph{Explanation on the notation of physicists.}
In this part, we will make the same 
notational abuse as physicists.
Instead of writing products of section as  
$\underline{\phi}_1\underline{\phi}_2$ or 
the star product of sections as $\underline{\phi}_1\star\underline{\phi}_2$,
we prefer to adopt the conventional physicist notation
$\phi(x_1)\phi(x_2) $ for the commutative product and
$\phi(x_1)\star \phi(x_2) $ for the star product.
The meaning of the formulas is changed,
since 
in the physicist's notation, 
we multiply sections then evaluate them at points $(x_1,x_2)$ of the configuration space $M^2$
whereas in the mathematical notation, we just multiply two sections 
$\underline{\phi}_1$ and
$\underline{\phi}_2$.

\subsubsection{Examples of Wick theorems coming from physics.}

We give the general QFT formula for the star product in the notations of physicists
$$\phi_1^{n_1}(x_1)\star\dots\star \phi_k^{n_k}(x_k)$$
$$=\sum \left(\begin{array}{c}n_1 \\ i_1 \end{array}\right)\dots\left(\begin{array}{c}n_k \\ i_k \end{array}\right)\underset{\text{Distribution on }M^n }{\underbrace{\left\langle 0|\left(\phi_1^{n_1-i_1}(x_1)\star\dots\star\phi_k^{n_k-i_k}(x_k)\right)|0\right\rangle}}\phi_1^{i_1}(x_1)\dots\phi_k^{i_k}(x_k).
$$
In Physics, the product of fields inside the $\left\langle 0|\dots|0\right\rangle$ 
is computed using Wick's theorem. 
Wick's theorem for time ordered product just means:
$T( \phi_1...\phi_n )=:\text{all possible contractions}:$
when we contract two fields, it just means we
choose some pairs of fields 
in all possible ways
and replace them by a propagator
which is a distributional two point function $\Delta_+$.
We will represent a Wick contraction of two fields with the symbol 
$\overbrace{\phi(x_1)\phi(x_2)}$ and by definition
$\overbrace{\phi(x_1)\phi(x_2)}=\Delta_+(x_1,x_2)$.
We then give some simple examples of $\star$ products in order to illustrate the
mechanism at work.
\begin{ex}
$$\phi(x_1)\star\phi(x_2)=\phi(x_1)\phi(x_2)+\overbrace{\phi(x_1)\phi(x_2)}$$
$$=\phi(x_1)\phi(x_2)+\Delta_+(x_1,x_2)$$
$$\phi(x_1)\star\phi(x_2)\star\phi(x_3)=\phi(x_1)\phi(x_2)\phi(x_3) + \left(\overbrace{\phi(x_1)\phi(x_2)}\phi(x_3)+\text{cyclic} \right)$$ $$=\phi(x_1)\phi(x_2)\phi(x_3) + \left(\Delta_+(x_1,x_2)\phi(x_3)+\text{cyclic} \right) $$
$$\phi^2(x_1)\star\phi^2(x_2)=\phi^2(x_1)\phi^2(x_2)+4 \overbrace{\phi(x_1)\phi(x_2)}\phi(x_1)\phi(x_2)
+2\overbrace{\phi(x_1)\phi(x_2)}\overbrace{\phi(x_1)\phi(x_2)}  $$
$$=\phi^2(x_1)\phi^2(x_2)+4 \Delta_+(x_1,x_2)\phi(x_1)\phi(x_2)
+2\Delta_+^2(x_1,x_2) $$
\end{ex}
\subsubsection{Functorial pull-back operation.} 
Let $I$ be a finite subset of $\mathbb{N}$.
Then we 
denote by $M^I$ the configuration space
of points labelled by $I$.   
In order to define the $\star$ product of fields, we need 
to define some
operations 
which allows us to pull-back some products of fields living in configuration space
$M^I,I\subset \{1,...,n\}$, to the larger configuration space
$M^n$.
\begin{ex}
Consider $\phi(x_1)\star \phi(x_2)\in \mathcal{H}^2$, 
we will illustrate the embedding
of the element
$\phi(x_1)\star \phi(x_2)$ in $\mathcal{H}^4$.
$$p_{\{1234\}\mapsto\{12\}}^*\left(\phi(x_1)\star \phi(x_2)\right)=\left(\phi(x_1)\star \phi(x_2)\right) 1(x_3)1(x_4)$$
\end{ex} 
If $J$ is another finite subset of $\mathbb{N}$ such that $I\subset J$, then there is a canonical
projection $p_{J\mapsto I}:M^J\mapsto M^I$ 
which induces by \textbf{pull-back} a morphism 
$$
\begin{array}{cccc}
p_{J\mapsto I}^*:& C^\infty(M^I)&\mapsto & C^\infty(M^J) 
\\ & f(x_i)_{i\in I}&\mapsto &p_{J\mapsto I}^*f(x_j)_{j\in J}=1(x_j)_{j\in J\setminus I}\otimes_{C^\infty(M^J)} f(x_i)_{i\in I},
\end{array}
$$
$p^*$ is an algebra morphism.
To each configuration space $M^I$, 
we first define
the bundle 
$\underline{H}^I=M^I\times \mathbb{R}[\phi_i]_{i\in I}$, 
and taking the sections of this bundle, 
we obtain the $C^\infty(M^I)$ module 
$\mathcal{H}^I=\Gamma\left( M^I , \underline{H}^I \right)$.
If $I\subset J$,
the idea is that
the morphism 
$p^\star_{J\mapsto I}$
extends
to Hopf modules
by the pull-back operation 
$p^*_{J\mapsto I}$
lifts functorially to a map $\mathcal{H}^I \mapsto \mathcal{H}^J$ 
given by the formula:
$$
\begin{array}{cccc}
p_{J\mapsto I}^*:& \mathcal{H}^I &\mapsto & \mathcal{H}^J 
\\ & \bigotimes_{i\in I} \underline{a}_i&\mapsto &p_{J\mapsto I}^*\left(\bigotimes_{i\in I} \underline{a}_i\right)=\bigotimes_{j\in J\setminus I} \underline{1}_j\otimes_{C^\infty(M^J)} \bigotimes_{i\in I} \underline{a}_i
\end{array}
$$
where $\underline{1}_j$ is the unit
section of the bundle $\underline{H}$ over
the $j$-th factor manifold $M$.

This pull-back operation
allows us to characterize
collections $(T_I)_{I}$, where
each $T_I$
is
a $C^\infty(M^I)$-module
\textbf{morphism}:
$T_I: \mathcal{H}^I\mapsto 
\mathcal{H}^I\otimes_{C^\infty(M^I)}V^I$,   
which
satisfy 
some good compatibility relations 
with the 
collection
of inclusions 
$p_{J\mapsto I}^*:\mathcal{H}^I\hookrightarrow \mathcal{H}^J$,
which means
that $$\forall A_I\in\mathcal{H}^I, 
T_J\left(p^\star_{J\mapsto I}A_I\right)=p^\star_{J\mapsto I}T_I\left(A_I\right).$$
This can also be formulated
as the commutativity
of the diagram:
$$ 
\begin{array}{cccc}
p_{J\mapsto I}^*: &\mathcal{H}^I& \rightarrow & \mathcal{H}^J\\
&T_I\downarrow & & T_J\downarrow\\
&\mathcal{H}^I & \rightarrow & \mathcal{H}^J,
\end{array}
$$ 
for all $I\subset J$.
\begin{ex}
$T_3(\phi^{i_1}(x_1)\otimes \phi^{i_2}(x_2)\otimes 1(x_3))=T_{2}(\phi^{i_1}(x_1)\otimes\phi^{i_2}(x_2))\otimes 1(x_3)$.
\end{ex}
% There is also an associative product $\circ$, where $\circ$ is called $\textbf{twisted}$ product, this associative product represents the \textbf{time ordered product} of QFT:
%\begin{equation}
%T(a_{1}...a_{n})=a_{1}\circ...\circ a_{n}
%\end{equation}
%but the construction of $\circ$ cannot be rigorously defined.
\paragraph{Domain of definition of the $\star$ product.}
For each $I\subset \mathbb{N}$, we need to twist $\mathcal{H}^I$ with a left $C^\infty(M^I)$ module  $V^I\subset \mathcal{D}^\prime(M^I)$ of distributional coefficients and we consider instead $\mathcal{H}^I\otimes_{C^\infty(M^I)}V^I$. 
% Since we saw that for QFT applications, $\mathcal{H}^I$ should also contain distributional coefficients. 
% 
% For example, $\phi(x_1)\star \phi(x_2)=\underset{\text{distributions}}{\underbrace{\Delta_+(x_1,x_2)}}+\phi(x_1)\phi(x_2)\in\mathcal{H}^2\otimes_{C^\infty(M^2)}V^2$.  
For any finite subsets $I,J$ 
of $\mathbb{N}$,
such that $I\cap J=\emptyset$, 
our star product will 
be well defined as a bilinear map
$$\star:\mathcal{H}^I \times \mathcal{H}^J \mapsto \mathcal{H}^{I\cup J}\otimes_{C^\infty(M^{I\cup J})} V^{I\cup J} $$
where  $V^I ,V^J,V^{I\cup J}$ are respectively the left 
$C^\infty(M^{I}),C^\infty(M^{J}),C^\infty(M^{I\cup J})$-module
which contains the distributional coefficients.
The star product is supposed to satisfy the
following rule
$$\forall (u,v)\in V^I\times V^J, \forall (P,Q)\in \mathcal{H}^I \times \mathcal{H}^J,$$ 
$$ \left(u P \right)\star \left(vQ \right) = \left(p_{J\cup I\mapsto I}^*u\right) \left( p_{J\cup I\mapsto J}^*v\right) \left(P\star Q\right)$$ 

\subsection{The construction of $\star$.}
We will describe a general procedure called twisting, which allows to construct non commutative associative products from the usual commutative product of fields and an object called \emph{Laplace coupling} $(.\vert .)$.
The Laplace coupling is the Hopf algebraic machine which produces "the contractions of pairs of fields" that 
we need in order 
to reproduce the Wick theorem.
In the sequel, we will use capital letters to denote strings of operators
\begin{ex}
$A=a_1 \dots a_n$ 
%and $T(A)=a_{1}\star \dots \star a_{n}$.
where $A\in \mathcal{H}^n$ and each $a_i\in \mathcal{H}^{\{i\}}$ .
\end{ex}

 And for $A=a_1\dots a_n,B=b_1\dots b_n$, the
commutative product $AB$ 
means the commutative product over each point $AB=(a_1b_1)\dots (a_nb_n)$.
\paragraph{The Laplace coupling.} 
For our Hopf algebras, the contraction operation of the Wick theorem in QFT is realised by the Laplace coupling: 
\begin{defi}\label{Laplace}
Let $I,J$ be \textbf{finite disjoint} subsets of $\mathbb{N}$.
The Laplace coupling is defined as a bilinear map
$\left(.\vert . \right): \mathcal{H}^I\otimes \mathcal{H}^J \mapsto V^{I\cup J}$ which satisfies the relations
\begin{eqnarray}
\left(\phi(x_1)\vert\phi(x_2) \right)=\Delta_+(x_1,x_2)
\\ \left(AB\vert C\right)=\sum\left(A\vert C_{(1)}\right)\left(B\vert C_{(2)}\right) 
\\ \left(1\vert A\right)=\left(A\vert 1\right)=\varepsilon(A)
\end{eqnarray}
\\ more generally we have the coassociative version $\left(A^1...A^n\vert B \right)=\sum\Pi_{k=1}^n \left(A^k \vert B_{(k)}\right) $.
\end{defi}
We notice that the Laplace coupling of two fields $\phi(x_1),\phi(x_2)$ is exactly the Wick contraction beetween these two fields: $\left(\phi(x_1)\vert\phi(x_2) \right)=\overbrace{\phi(x_1)\phi(x_2)}=\Delta_+(x_1,x_2)$.
\begin{ex}
$$\left(\phi(x_1)\vert \phi(x_2)\right)=\Delta_+(x_1,x_2).$$
$$\left(\phi^2(x_1)\vert \phi^2(x_2)\right)=2\Delta_+^2(x_1,x_2)$$
$$\left(\phi^2(x_1)\vert \phi(x_2)\phi(x_3)\right)=2\Delta_+(x_1,x_2)\Delta_+(x_1,x_3).$$
\end{ex}
\begin{prop}
Let $\left(.\vert . \right)$ be a Laplace coupling as in the definition (\ref{Laplace}). Then $\left(.\vert . \right)$ is entirely determined by the two point function $\left(\phi(x_1)\vert\phi(x_2) \right)=\Delta_+(x_1,x_2)$.
Furthermore, we have the relation: 
$\left(\phi^k(x_1)\vert\phi^l(x_2)\right)=\delta_{kl}k!\Delta_+^k(x_1,x_2)$.
\end{prop}
\begin{proof}
See \cite{BrouderQFT}.
\end{proof}
The function $\Delta_+(x_1,x_2)$ appearing in the definition of the Laplace coupling should be a \textbf{propagator} for the Wave operator.
In QFT, it is the \emph{Wightman propagator} $\Delta_+$ defined in chapter $5$.
\begin{defi}
The star product $\star$ is defined as follows. Let $I,J$ be any finite disjoint subsets of $\mathbb{N}$, 
for all 
$(A,B)\in \mathcal{H}^I\times \mathcal{H}^J$:
\begin{equation}
A\star B=\sum\left(A_{(1)}\vert B_{(1)} \right) A_{(2)} B_{(2)}
\end{equation}
where  $(.\vert.)$ denotes the Laplace coupling and $A_{(2)} B_{(2)}$ denotes the usual commutative product of fields.
\end{defi}
\begin{ex}
$$\phi^3(x_1)\star  \phi^3(x_2)=6\Delta_+^{3}(x_1,x_2)+6\Delta_+^2(x_1,x_2)\phi(x_1)\phi(x_2)+3\Delta_+(x_1,x_2)\phi(x_1,x_2)+\phi^3(x_1)\phi^3(x_2).$$
$$ \phi^2(x_1)\star (\phi(x_2)\phi(x_3))$$ $$=\phi^2(x_1)\phi(x_2)\phi(x_3)+2\phi(x_1) \phi(x_2)\Delta_+(x_1,x_3)+2\phi(x_1) \phi(x_3)\Delta_+(x_1,x_2)+2\Delta_+(x_1,x_2)\Delta_+(x_1,x_3)  $$
\end{ex}
 From the last example, we notice the important fact that the star product 
$A\star B$ is not automatically well defined 
because the computation of the star product 
involves products of distributions and we have yet to prove that these products are well defined.
\subsubsection{The counit $\varepsilon$.}
As we already said, the counit plays the role of the vacuum expectation value in QFT. 
We first recall the most important result about the counit $\varepsilon$, it is the coassociativity equation:
$$A=\sum \varepsilon(A_{(1)})A_{(2)} $$
\begin{ex}
$$\sum \varepsilon(\underline{\phi}^2_{(1)})\underline{\phi}^2_{(2)}=
\varepsilon(\underline{\phi}^2)\underline{1}+2\underline{\phi}\varepsilon(\underline{\phi})  
+\underline{\phi}^2\varepsilon(\underline{1})=
0+0+\underline{\phi}^2 1=\underline{\phi}^2
$$
\end{ex}
We give an example of the same quantity expressed in the language of
Hopf algebras and the conventional QFT language so that the reader can compare:
\begin{ex}
$$\varepsilon(\underline{\phi}_1^2\star\underline{\phi}_2^2)=\varepsilon\left(\left( \underline{1}|\underline{1} \right)\underline{\phi}_1^2\underline{\phi}_2^2
+4\left( \underline{\phi}_1|\underline{\phi}_2 \right)\underline{\phi}_1\underline{\phi}_2 
+\left( \underline{\phi}^2_1|\underline{\phi}^2_2 \right)\right) $$ 
$$=\varepsilon(
\underline{\phi}_1^2\underline{\phi}_2^2+4\Delta\underline{\phi}_1\underline{\phi}_2
+2\Delta_+^2 )=0+0+2\Delta_+^2  $$
$$\left\langle 0|\phi^2(x_1)\phi^2(x_2) |0\right\rangle=2\Delta_+^2(x_1,x_2) $$
\end{ex}

\subsection{The associativity of $\star$.}
For the moment, the $\star$ product we constructed is just bilinear. 
We have to prove it is associative.
First, let us prove some lemmas.
\begin{lemm} 
The $\star $ product satisfies the identities:
\begin{eqnarray}
\label{first}
\Delta(a\star  b)=\sum\left(a_{(1)}\star  b_{(1)}\right)\otimes a_{(2)} b_{(2)}  
\\ \label{second} \left(a\star  b \vert c \right)=\left(a \vert b \star  c \right)
\\ \label{third} \varepsilon \left(a \star  b\right)=\left(a \vert b\right)
\end{eqnarray}
\end{lemm}
Note that
$\Delta$
is the coproduct
of the commutative product
and not the coproduct
of $\star$.
\begin{proof}
$$\Delta(a\star  b)=\sum\left(a_{(1)}\vert b_{(1)}\right)\Delta(a_{(2)} b_{(2)}) =\sum\left(a_{(1)}\vert b_{(1)}\right) a_{(2)} b_{(2)}\otimes a_{(3)}b_{(3)}$$ $$=\sum\left(a_{(11)}\vert b_{(11)}\right) a_{(12)} b_{(12)}\otimes a_{(2)}b_{(2)}  =\sum\left(a_{(1)}\star  b_{(1)}\right)\otimes a_{(2)} b_{(2)}.$$
$$\left(a\star  b \vert c \right)=\sum\left(a_{(1)} \vert b_{(1)}\right)\left(a_{(2)}b_{(2)} \vert c\right)
=\sum\left(a_{(1)} \vert b_{(1)}\right)\left(a_{(2)} \vert c_{(1)}\right)\left(b_{(2)} \vert c_{(2)}\right)$$
$$=\sum\left(a_{(1)} \vert b_{(2)}\right)\left(a_{(2)} \vert c_{(2)}\right)\left(b_{(1)} \vert c_{(1)}\right) $$ because by cocommutativity of the field coproduct, we can permute $b_{(1)},b_{(2)}$ and $c_{(1)},c_{(2)}$.
\\ $\varepsilon(a\star  b)=\sum\varepsilon(\left(a_{(1)}\vert b_{(1)}\right) a_{(2)} b_{(2)} )=\sum\left(a_{(1)} \varepsilon(a_{(2)}) \vert b_{(1)} \varepsilon(b_{(2)})\right)=\sum\left( a\vert b\right)$
\end{proof}
More generally, we have a distributed version of (\ref{first}):
\begin{prop}
$\star $ satisfies the identity:
\begin{equation}\label{boostedfirst}
\Delta(a_{1}\star \dots \star  a_{n})=\sum(a_{1(1)}\star \dots \star  a_{n(1)})\otimes a_{1(2)}...a_{n(2)}
\end{equation} 
%\begin{equation}
%\Delta T(A)=T(A_{(1)})\otimes A_{(2)}
%\end{equation} 
\end{prop} 
\begin{thm}
The product $\star$
is associative provided that the products
of distributions coming
from the Laplace couplings
make sense.
\end{thm}
\begin{proof}
$$(a\star b)\star c=\sum\left((a \star b)_{(1)} \vert c_{(1)} \right)(a\star b)_{(2)}  c_{(2)}=\sum\left(a_{(1)} \star b_{(1)} \vert c_{(1)} \right)a_{(2)} b_{(2)}  c_{(2)}$$ by (\ref{first}) $$=\sum\left(a_{(1)} \vert b_{(1)} \star c_1 \right)a_{(2)} b_{(2)}  c_{(2)}$$ by (\ref{second}) $$=\sum\left(a_{(1)} \vert (b \star c)_{(1)} \right)a_{(2)} (b \star c)_{(2)}$$ by (\ref{first}) $$=a\star(b\star c)$$
\end{proof}
\begin{coro}
$a_1\star ... \star a_n$ is well defined provided that the products
of distributions coming
from the Laplace couplings
make sense.
\end{coro}
\subsection{Wick's property.}
We give a general QFT formula for the star product in the notations of physicists
$$ 
\phi_1^{n_1}(x_1)\star\dots\star \phi_k^{n_k}(x_k)$$
$$=\sum \left(\begin{array}{c}n_1 \\ i_1 \end{array}\right)\dots\left(\begin{array}{c}n_k \\ i_k \end{array}\right)\underset{\text{Distribution on }M^n }{\underbrace{\left\langle 0|
\phi^{n_1-i_1}(x_1)\star\dots \star\phi^{n_k-i_k}(x_k)|0\right\rangle}}\phi^{i_1}(x_1)\dots\phi^{i_k}(x_k).
$$
And we write the Hopf counterpart of this formula
$$
\underline{a}_1\star\dots\star \underline{a}_n=\sum \underset{\text{distributions}}{\underbrace {\varepsilon(\underline{a}_{1(1)}\star\dots\star \underline{a}_{n(1)})}}\underline{a}_{1(2)}\dots\underline{a}_{n(2)}.
$$  

 We introduce a crucial definition which is the Hopf algebra counterpart of the Wick theorem of QFT. We call this property Wick's expansion.
For any finite subsets $I,J$ 
of $\mathbb{N}$,
such that $I\cap J=\emptyset$, 
let $\star$ be any bilinear map
$$\star: \mathcal{H}^I \times \mathcal{H}^J \mapsto \mathcal{H}^{I\cup J}\otimes_{C^\infty(M^{I\cup J})} V^{I\cup J} .$$

\begin{defi}
A bilinear map $\star$
as above satisfies the Wick expansion property if
for $I\cap J=\emptyset$,
$$\forall A=\left(\prod_{i\in I} a_i\right)\in \mathcal{H}^I\otimes_{C^\infty(M^I)} V^I,\forall B=\left(\prod_{j\in J} b_j\right)\in  \mathcal{H}^J\otimes_{C^\infty(M^J)} V^J,$$
\begin{equation}
A\star B=\sum\varepsilon\left(A_{(1)}\star B_{(1)} \right)A_{(2)} B_{(2)}. 
\end{equation}
\end{defi}
This property encodes in the Hopf algebraic language all the algebro combinatorial properties of the Wick theorem.
We prove that our star product defined from the Laplace coupling does indeed satisfy the
Wick property.
\begin{thm}
Let $\star$ be defined by
\begin{equation}
A\star B=\sum\left(A_{(1)}\vert B_{(1)} \right)A_{(2)} B_{(2)}
\end{equation}
where $\left(.\vert. \right) $ denotes the Laplace coupling, then $\star$ satisfies Wick's expansion:
$$\forall A=\prod_{i\in I} (a_i)\in \mathcal{H}^I\otimes_{C^\infty(M^I)} V^I,\forall B=\prod_{j\in J} (b_j)\in  \mathcal{H}^J\otimes_{C^\infty(M^J)} V^J$$
$$A\star B=\sum \varepsilon(A_{(1)}\star B_{(1)})A_{(2)} B_{(2)}$$
\end{thm}
\begin{proof}
By the identity (\ref{first}), notice that $\varepsilon(A_{(1)}\star B_{(1)})=\left(A_{(1)} |B_{(1)} \right)$ which proves the claim.
\end{proof}
The meaning of this theorem is that any associative product $\star$ constructed by the twisting procedure from the Laplace coupling $(.\vert.)$ should satisfy the Wick expansion property.
\subsection{Recovering Feynman graphs.}
\begin{prop}
For any $\left(p_1,...,p_n\right)$,
$\varepsilon \left(\phi^{p_1}(x_1)\star ... \star \phi^{p_n}(x_n) \right)=$
\begin{equation}
p_1!...p_n!\sum_{\sum_{j=1}^nm_{ij}=p_i}\Pi_{1\leqslant i<j\leqslant n}\frac{\Delta_+^{m_{ij}}(x_i,x_j)}{m_{ij}!},
\end{equation}
where $(m_{ij})_{ij}$ runs over the set of all symmetric matrices with integer entries with vanishing diagonal and such that 
for all $i$, 
the sum of the coefficients 
on the $i$-th row is equal to $p_i$.
\end{prop}
Note that $(m_{ij})_{ij}$ should be interpreted as the adjacency matrix of a Feynman graph.
\begin{proof}
% To each operator $a^j=\phi^{p_j}(x_j)$, we associate a vertex decorated with the index $j$.
The sum is indexed by symmetric matrices with integer coefficicients vanishing diagonals. 
We will prove the theorem by recursion. 
We start by checking the formula at degree $2$.
$$\varepsilon\left(\phi^{p_1}(x_1)\star\phi^{p_2}(x_2)\right)=\left(\phi^{p_1}(x_1)\vert \phi^{p_2}(x_2) \right)=p_1!\delta_{p_1p_2}\Delta_+^{p_1}(x_1,x_2)$$ 
$$=p_1!p_2!\sum_{p_{12}=p_1=p_2}\frac{\Delta_+^{p_{12}}(x_1,x_2)}{p_{12}!}.$$
Assume we know that $\frac{\varepsilon \left(\phi^{p_1}(x_1)\star ... \star \phi^{p_k}(x_k) \right)}{p_1!...p_k!}
=\sum_{\sum_{j=1}^km_{ij}=p_i}\Pi_{1\leqslant i<j\leqslant k}\frac{\Delta_+^{m_{ij}}(x_i,x_j)}{m_{ij}!} $ is true for any $k\leqslant n$.
Set $A=\left(a^1\star...\star a^{n} \right)$ and $B=a^{n+1}$.
We use the identity $\varepsilon(A\star B)=\left(A\vert B\right)=\sum \varepsilon(A_{(1)})\left( A_{(2)}\vert B\right)$.
We use the explicit formula for the coproduct of quantum fields
$$\Delta \phi^{p_j}(x_j)=
\sum_{0\leqslant i_j\leqslant p_j}
\left(\begin{array}{c} p_j \\ i_j \end{array} \right) \phi^{i_j}(x_j)\otimes \phi^{p_j-i_j}(x_j)$$ and
$$\Delta^n \phi^{p_{n+1}}(x_{n+1})=\sum_{i_1+...+i_n=p_{n+1}}
\left(\begin{array}
{ccc} & p_{n+1} & \\ i_1 & ... & i_n
\end{array}\right)
\phi^{i_1}(x_{n+1})\otimes\dots\otimes\phi^{i_n}(x_{n+1}) $$
to deduce
$$\varepsilon(A_{(1)})=\left(\begin{array}{c} p_1 \\ i_1 \end{array} \right)\dots\left(\begin{array}{c} p_n \\ i_n \end{array} \right) \varepsilon\left(\phi^{p_1-i_1}(x_1)\star\dots\star \phi^{p_n-i_n}(x_n)\right) $$
$$\left(A_{(2)}\vert B \right)=\left(A_{(2)}\vert \phi^{p_{n+1}}(x_{n+1}) \right) $$ 
$$=\left(\begin{array}
{ccc} & p_{n+1} & \\ i_1 & \dots & i_n
\end{array}\right) \frac{\Delta_+^{i_1}(x_1,x_{n+1})}{i_1!}\dots\frac{\Delta_+^{i_n}(x_n,x_{n+1})}{i_n!}$$ 
$$=p_{n+1}! \Delta_+^{i_1}(x_1,x_{n+1})\dots\Delta_+^{i_n}(x_n,x_{n+1}) .$$
Each term $\varepsilon(A_{(1)})\left(A_{(2)}\vert B \right)$ has the form:
$$p_1!...p_{n+1}!\frac{\varepsilon\left(\phi^{p_1-i_1}(x_1)\star\dots\star \phi^{p_n-i_n}(x_n)\right)}{(p_1-i_1)!\dots(p_n-i_n)!}\frac{\Delta_+^{i_1}(x_1,x_{n+1})}{i_1!}\dots\frac{\Delta_+^{i_n}(x_n,x_{n+1})}{i_n!}, $$
which ends our proof because the product $(p_1-i_1)!\dots(p_n-i_n)!$ in the denominator kill the unwanted factors. The space of $n+1\times n+1$ symmetric matrices with fixed last row with coefficients $i_1,...,i_k$ and such that the sum of terms on the $k$-th line is equal to $p_k$ is in bijection with the space of $n\times n$ symmetric matrices with sum of $k-th$ line equals $p_k-i_k$.    
\end{proof}
\paragraph{A word of caution and an introduction to the next section.}
From now on, the star product is \textbf{fixed} and is defined as above from the ``twisting procedure'' with the Laplace coupling defined by 
the Wightman propagator $\Delta_+$.
However, 
we \emph{have not yet defined rigorously} 
the product $\star$ for elements 
$$(A,B)\in \left(\mathcal{H}^I\otimes_{C^\infty(M^I)}V^I\right)\times\left( \mathcal{H}^J\otimes_{C^\infty(M^J)}V^J\right)$$
with \textbf{distributional coefficients}. 
We will construct a time ordered product $T$
from $\star$ and
we will prove that
$T(AB)$ is well defined
in the distributional sense. 
This is illustrated by one of our previous example:
\begin{ex}
$$\phi^2(x_1)\star (\phi(x_2)\phi(x_3))=\phi^2(x_1)\phi(x_2)\phi(x_3)
+2\phi(x_1)\Delta(x_1,x_2)\phi(x_3)+2\phi(x_1)\Delta(x_1,x_3)\phi(x_2) $$ $$ +2\underset{\text{product of  distributions}}{\underbrace{\Delta(x_1,x_2)\Delta(x_1,x_3)}}  $$
\end{ex}
In the next section, we are going to use $\star$ to define the time ordered product $T$ which \textbf{only satisfies} 
the Wick expansion property $T(A)=\sum t(A_{(1)})A_{(2)}$ 
and the causality equation.
\section{The causality equation.} 
\paragraph{The geometrical lemma}
The geometrical lemma 
(due to Popineau and Stora \cite{Popineau}) 
essentially states 
that we can partition 
the configuration space minus the thin diagonal 
$M^n\setminus d_n$,
with open sets having nice properties from the point of view of causality.

\begin{lemm}\label{geomlemm}
Let $(M,\geqslant)$ be a causal 
Lorentzian manifold endowed 
with the canonical \emph{poset} structure (i.e. a set equipped with
a partial order) induced by the Lorentzian metric
and the chronological causality on $
M$: $x \leqslant y$ 
if $y$ lies in the future cone of $x$.
Define the relation $\nleqslant$
by: $x \nleqslant y$ if and only if $x \leqslant y$ does not hold
Let $[n]=\{1,\dots,n\}$ and $I$
a proper subset of $[n]$.
If $I^c$
is the complement of $I$ in $[n]$ (i.e. $I\sqcup I^c=[n]$),
we define the subset $M_{I,I^c}$ of
$M^n$ by
\begin{eqnarray*}
M_{I,I^c} = \{ (x_1,\dots,x_n)\in M^n |\forall (i,j)\in I\times I^c, x_i\nleqslant x_j\}.
\end{eqnarray*}
Then, 
\begin{eqnarray}
\bigcup_I M_{I,I^c} &=& M^n \backslash d_n,
\label{geolemma}
\end{eqnarray}
where $d_n=\{x_1=\dots=x_n\}$ is the thin diagonal of $M^n$
and $I$ runs over the proper subsets of $[n]$.
\end{lemm}
\begin{proof}
It is clear that, for all proper subsets $I$ of $[n]$,
we have $M_{I,I^c} \subset M^n\backslash d_n$, because
if $(x_1,\dots,x_n)\in d_n$, then $x_i \ge x_j$ for
all $i$ and $j$ in $[n]$. It remains to show that any
$X=(x_1,\dots,x_n)\in  M^n\backslash d_n$ belongs to
some $M_{I,I^c}$. In fact we shall determine all the
$M_{I,I^c}$ to which a given $X$ belongs.
For all $X=(x_1,\cdots,x_n)\in M^n$,
we define
$\lambda(X)$ 
as the finite subset $\{a_1,\cdots,a_r\}\subset M$
s.t. $a\in \lambda(X)$ iff $\exists i\in [n], x_i=a$. 
To each $X \in  M^n$, 
we associate
a directed graph known 
as the \emph{Hasse diagram} of
$X$ as follows.
To each distinct $a\in \lambda(X)$, 
we associate
a vertex and we draw a directed line from vertex 
$a$
to vertex $b$ 
if $a \leqslant b$, 
$a\neq b$
and no
other $c\in \lambda(X)$, 
distinct from $a$ and $b$,
is such that $a\leqslant c\leqslant b$.
All indices $i\in [n]$
such that $x_i=a$ decorate
the same vertex.
The Hasse diagram of $X$ has a single vertex 
if and only
if $X\in d_n$. 
Take $X\in M^n\backslash d_n$, 
its Hasse diagram has 
at least two vertices.
If we pick up any vertex of the Hasse diagram, then any
point $x_j$ greater than a point $x_i$ 
of this vertex 
is such that $x_j\ge x_i$. Thus, $j\in I$ if $i\in I$
and, to build a $M_{I,I^c}$, we can select a non-zero number of
vertices of the diagram and add all the vertices that are
greater than the selected ones. 
The points corresponding
to all these vertices determine a subset $I$ of $[n]$.
If $I\not=[n]$, then $X\in M_{I,I^c}$ and it is always possible
to find such a $I$ by picking up a single maximal vertex
in one connected component of the Hasse diagram. 
Conversely,
any $M_{I,I^c}$ is made of the points that are
greater than their minima. 
To see this, consider a point
$x_i\in M_{I,I^c}$ such that $i\in I$. Then, the set
$S_i=\{x_j\in X | x_i \ge x_j\}$ is not empty
because $x_i$ belongs to it. 
Then, $x_i$ is larger
than a minimum of $S_i$, which is also a minimum of the
Hasse diagram of $X$.
\end{proof}
\begin{figure} %on ouvre l'environnement figure
\begin{center}
\includegraphics[width=12cm]{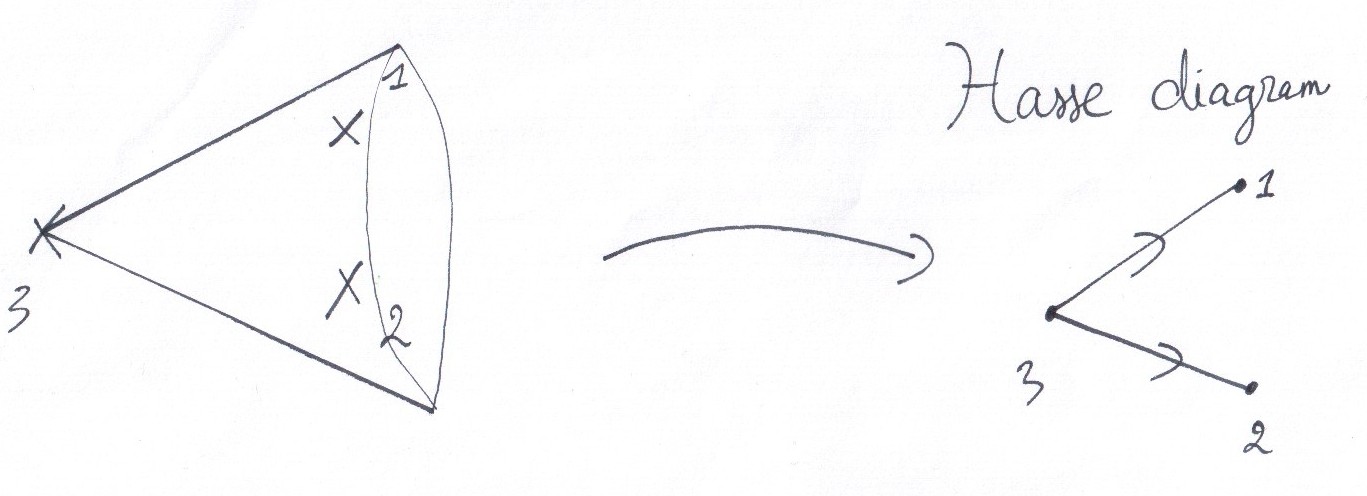} %ou image.png, .jpeg etc.
\caption{A configuration of three points in $C_{\{12\}}\subset M^3$ and the corresponding Hasse diagram.} %la légende
\label{Soften_Stora} %l'étiquette pour faire référence à cette image
\end{center}
\end{figure} %on ferme l'environnement figure 

\subsection{Definition of the time-ordering operator}
In quantum field theory, the poset 
is the Lorentzian manifold $M$ and the fields are,
for example, $\phi^n(x)$.
For any finite subset $I$
of $\mathbb{N}$, 
we defined the 
configuration 
space $M^I$ 
as the set of maps 
from $I$ to $M$ 
and
we introduced
some
vector space 
of distributions 
$V^I$ which contains 
the singularities
of the Feynman amplitudes,
then we introduced
a module $\mathcal{H}^I\otimes_{C^\infty(M^I)}V^I$ 
associated to $I$. 
For all $A\in\mathcal{H}^I$, 
we will denote by $t_I(A)$ 
the element 
$\varepsilon\left( T_I(A)\right)$
and $t:\mathcal{H}^I\mapsto V^I$.
%The only condition required for this product is that
%$A\cdot B=B\cdot A$ is the supports of $A$ and $B$
%are not comparable (i.e. for any $x$ in the support of $A$
%and $y$ in the support of $B$ we have
%$x\ngeqslant y$ and $y\ngeqslant x$).
\paragraph{Axioms for the time ordering operator.}
We are going to define the time-ordering operator
as a collection $(T_I)_I$ of $C^\infty(M^I)$-module morphisms, for all finite subset $I$ of $\mathbb{N}$,
$T_I: \mathcal{H}^I\rightarrow \mathcal{H}^I\otimes_{C^\infty(M^I)}V^I$ 
with the following properties: 
\begin{enumerate}\label{AxiomsTproducts}
\item If $\vert I\vert\leqslant 1,$ the restriction of $T$
to $\mathcal{H}^I$ is the identity map,
\item $T$ 
satisfies the Wick expansion property:
\begin{equation}\label{Wickexp}
T(A)=\sum \varepsilon\circ T(A_{(1)})A_{(2)} 
\end{equation}  
\item The causality equation. Let $A=a_1(x_1)\dots a_n(x_n)\in \mathcal{H}^n$. If there is
a proper subset $I\subset \{1,\dots,n\}$ such that
$x_i \nleqslant x_j$ for $i\in I$ and $j \notin I$, 
denote $A_I=\prod_{i\in I} a_i(x_i)$  and $A_{I^c}=\prod_{j\in I^c} a_j(x_j)$
then 
\begin{equation}\label{causality}  
T(A)=T(A_I)\star T(A_{I^c}). 
\end{equation}  
\end{enumerate}
\paragraph{Remark:}
The equation 
$T(A)=\sum t(A_{(1)})A_{(2)}$ 
implies $T$ is a comodule morphism, 
we denote by $\beta$ the 
coaction defined as follows:
$$\forall (f,A)\in V^I\times\mathcal{H}^I, 
\beta(f\otimes A)=\sum \left(fA_1\otimes A_2\right)=\sum \left(A_1\otimes fA_2\right).$$
$$\beta T(A)= \sum t(A_{(1)})A_{(21)}\otimes A_{(22)}=\sum t(A_{(1)})A_{(2)}\otimes A_{(3)}$$ 
$$=\sum t(A_{(11)})A_{(12)}\otimes A_{(2)}=\sum T(A_{(1)})A_{(2)}=\left(T\otimes Id \right)\beta A .$$
In fact, C Brouder communicated 
to us a proof of $T(A)=\sum t(A_{(1)})A_{(2)}\Leftrightarrow $ 
$T$ is a comodule morphism.
\paragraph{What are we trying to construct ?}
We have a given star product which is the operator product of quantum fields.
The idea is to construct 
all time ordered products 
satisfying the previous set of axioms, 
the most important being causality and the Wick expansion property.
The $T$ product is not unique, actually there are infinitely many $T$-products and there is an infinite dimensional group which acts freely and transitively on the space of all $T$-products
(see equation (4.1) p.~17 in \cite{BrouderQFT}). 
This group is the Bogoliubov renormalization group
which 
was studied
in Hopf algebraic 
terms by C. Brouder 
in \cite{BrouderQFT} p.~17-20.
in \cite{Borcherds}
The problem of construction of a QFT in our sense is reduced to the problem of constructing a $T$-product satisfying the axioms and to \textbf{make sense analytically} of this $T$-product.
We will prove the existence of at least one $T$-product and we will show 
that it is analytically well defined. A crucial ingredient in the existence proof is to establish a recursion equation which expresses the $T$ product $T_n\in Hom(\mathcal{H}^n,\mathcal{H}^n)$ in terms of the elements $T_I\in Hom(\mathcal{H}^I,\mathcal{H}^I)$ for $I\varsubsetneq \{1,\dots,n\} $. We will later see that
the problem of defining the $T$-product reduces to a problem of making sense of \textbf{products of distributions} and a problem of \textbf{extension of distributions}.  
Our approach is related to the one of \cite{Borcherds} but we use 
causality in a more explicit way following Epstein--Glaser. 
However, 
the strategy 
we will adopt 
make essential 
use of ideas of 
Raymond Stora 
which appeared in unpublished form (\cite{Popineau}).   
\subsection{The Causality theorem.}
We give the main structure theorem for the amplitudes coming from perturbative QFT.
This theorem relates $T_n$ and all $(T_I)_I$ for $I\varsubsetneq \{1,\dots,n\}$ on the configuration space minus the thin diagonal $M^n\setminus d_n$.
\begin{thm}
Let $T$ be a collection $(T_I)_I$ of $C^\infty(M^I)$-module morphisms $T_I: \mathcal{H}^I \rightarrow \mathcal{H}^I\otimes_{C^\infty(M^I)}V^I$ which satisfy the collection of axioms (\ref{AxiomsTproducts}).
Then for all $I\varsubsetneq\{1,...,n\}$,  $t=\varepsilon\circ T$ 
satisfies the equation:
\begin{equation}\label{key formula}
t(A)=\sum t(A_{I(1)})t(A_{I^c(1)})\left( A_{I(2)} \vert A_{I^c(2)}\right)
\end{equation} 
$\textbf{on }M_{I,I^c}$. 
We call this equation 
the Hopf algebraic equation 
of causality.
\end{thm}
\begin{proof}
By definition $t=\varepsilon\circ T$, 
$$t(A)=\varepsilon(T(A))
=\varepsilon(T(A_I(x_i)_{i\in I}A_{I^c}(x_i)_{i\in I^c}))$$
$$=\varepsilon(T(A_I) \star T(A_{I^c}))$$ 
because of the causality equation (\ref{causality})
$$t(A)=(T(A_I) \vert T(A_{I^c}))=\sum\left( t(A_{I(1)})A_{I(2)} \vert t(A_{I^c(1)})A_{I^c(2)}\right)$$
because by Wick expansion property (\ref{Wickexp}) $T(A_I)=\sum t(A_{I(1)})A_{I(2)}  $ and  $T(A_{I^c})=t(A_{I^c(1)})A_{I^c(2)}$,
$$t(A)=\sum t(A_{I(1)})t(A_{I^c(1)})\left( A_{I(2)} \vert A_{I^c(2)}\right).$$
\end{proof}
We notice some important facts:
first,
in
Borcherds,
the equation 
\begin{equation}\label{Gaussiancondition}
t(A)=\sum t(A_{I(1)})t(A_{I^c(1)})\left( A_{I(2)} \vert A_{I^c(2)}\right)
\end{equation}
is
called 
the Gaussian
condition
for the Feynman
measure $t$
(Borcherds calls it 
$\omega$),
secondly 
beware that the
above product
is not a priori 
well defined since 
it is a product of distributions. 
Secondly, this theorem says 
that the $T$-product satisfying the axioms
\ref{AxiomsTproducts} \textbf{is not even 
well defined} 
on $d_n$. 
It is only well defined on each $M_{I,I^c}$ 
thus on $M^n\setminus d_n$ because of  
Stora's geometrical Lemma 
(\ref{geomlemm}).
To explain the meaning of the causality equation, we shall quote Borcherds where we changed 
his notation to adapt 
to our case (and 
also inserted some comments): 
``We explain what is going on in this definition. We would like to define the value
of the Feynman measure $t$ to be a sum over Feynman diagrams, formed by joining
up pairs of fields in all possible ways by lines, and then assigning a propagator
to each line and taking the product of all propagators of a diagram. 
This does
not work because of ultraviolet divergences: products of propagators need not be
defined when points coincide. If these products were defined then they would
satisfy the Gaussian condition \ref{Gaussiancondition}, which then says roughly that if the set of vertices $\{1,\dots,n\}$ are divided into two disjoint subsets $I$ and $I^c$,
then a Feynman diagram can be divided
into a subdiagram with vertices $I$, a subdiagram with vertices $I^c$, and some lines
between $I$ and $I^c$. The value $t(A_IA_{I^c})$ of the Feynman diagram would then be the
product of its value $t_I(A_{I(1)})$ on $I$, 
the product 
$\left( A_{I(2)} \vert A_{I^c(2)}\right)$ 
of all the propagators of
lines joining $I$ and $I^c$, and its value $t_{I^c}(A_{I^c(1)})$ on $I^c$.
 The Gaussian condition \ref{Gaussiancondition} need not
make sense if some point of $I$ is equal to some point of $I^c$ because if these points are
joined by a line then the corresponding propagator may have a bad singularity 
[\emph{however this never happens in the domain }$M_{I,I^c}$ \emph{defined in the geometrical lemma}], but
does make sense whenever all points of $I$ are not $\leqslant$ to all points of $I^c$ 
[\emph{this is exactly the definition of the domain} $M_{I,I^c}$].
The definition
above says that a Feynman measure should at least satisfy the Gaussian condition
in this case, when the product is well defined.''
The explanations of Borcherds show that
the geometrical lemma 
gives a very convenient way of
covering $M^n\setminus d_n$ by the sets $M_{I,I^c}$.
\subsection{Consistency condition}

The collection $(M_{I,I^c})_I$ 
forms an open cover of $M^n\setminus d_n$, 
thus there are open domains 
in which a given $M_{I,I^c}$ 
will overlap with a given $C_J$ 
and we must prove the causality equations give the same result on overlapping domains, which justify an eventual gluing by partitions of unity. 
We must check
a sheaf consistency condition:
if $I_1$ and $I_2$ are proper subsets of $\{1,\dots,n\}$
such that $C=C_{I_1} \cap C_{I_2}$ is not empty,
then $T_{I_1}|_C=T_{I_2}|_C$.
Let $u=v_1 w_1$ be the factorization of $u$
corresponding to $I_1$ and $u=v_2 w_2$ the one
corresponding to $I_2$.
We define on $C$
\begin{eqnarray*}
\zdd &=& \prod_{k\in I_1 \cap I_2} a^k(x_k),\\
\zcd &=& \prod_{k\in I_1^c \cap I_2} a^k(x_k),\\
\zdc &=& \prod_{k\in I_1 \cap I_2^c} a^k(x_k),\\
\zcc &=& \prod_{k\in I_1^c \cap I_2^c} a^k(x_k).
\end{eqnarray*}
Therefore, $v_1= \zdd \zdc$, $v_2=\zdd \zcd$,
$w_1=\zcd\zcc$ and $w_2=\zdc\zcc$.
We have
\begin{eqnarray*}
T_{I_1}|_C (u) &=& T(v_1)\cdot T(w_1)
= T(\zdd \zdc)\cdot T(\zcd\zcc).
\end{eqnarray*}
By definition of $C_{I_2}$ we
have $\zdc \ngeqslant \zdd$
and $\zcc \ngeqslant \zcd$, so that
\begin{eqnarray*}
T_{I_1}|_C (u) &=& T(v_1)\cdot T(w_1)
= T(\zdd)\cdot T(\zdc)\cdot T(\zcd)\cdot
T(\zcc).
\end{eqnarray*}
The indices $k$ of $\zcd$ are in $I_1^c$ and those
of $\zdc$ are in $I_1$, thus
$\zcd \ngeqslant \zdc$. On the other hand, the
indices $k$ of $\zcd$ are in $I_2$ and those
of $\zdc$ are in $I_2^c$, thus
$\zdc \ngeqslant \zcd$. In other words
$\zcd \sim \zdc$ so that $T(\zdc)$
and $T(\zcd)$ commute. Therefore,
\begin{eqnarray*}
T_{I_1}|_C (u) &=&
T(\zdd)\cdot T(\zdc)\cdot T(\zcd)\cdot
T(\zcc)
=
T(\zdd)\cdot T(\zcd)\cdot T(\zdc)\cdot
T(\zcc)
\\&=&
T(\zdd \zcd)\cdot T(\zdc\zcc)=
T(v_2)\cdot T(w_2)=
T_{I_2}|_C (u).
\end{eqnarray*}
So we have defined distributions $T_I(u)$ on each $M_{I,I^c}$ in a consistent
way.
We must now show that these $T_I(u)$ extend to a distribution
$T$ on $M^n\backslash D_n$.
If the test function $f$ has its support in $M_{I,I^c}$,
we can define $T(u(f))= T_I(u(f))$. However,
for a test function with a support not included
in a single $M_{I,I^c}$, we need to patch different
$T_I$. To do this we shall use a smooth partition of unity
subordinate to $M_{I,I^c}$.
\section{The geometrical 
lemma for curved space time.}
In this part, 
we need to improve the geometrical 
lemma due to Stora. Why is the geometrical lemma
not enough ? 
We first notice 
that the functions
$\chi_I$ from 
the partition of unity $(\chi_I)_I$ 
subordinate to the open cover $(M_{I,I^c})_I$
of $M^n\setminus d_n$ 
given by the 
geometrical
lemma (\ref{geomlemm})
are \textbf{smooth} 
in $M^n\setminus d_n$ but \textbf{are not smooth} in $M^n$. 
However, we will see
(see formulas \ref{key formula simplified},\ref{key formula}) 
that we are supposed
to multiply $\chi_I$ with some
product of distributions
$t_It_{I^c}\prod \Delta_+^{m_{ij}}$
on $M^n\setminus d_n$, extend
it on $M^n$
and control the wave front set of the extension
$\overline{\chi_It_It_{I^c}\prod \Delta_+^{m_{ij}}}$.
Hence, in order to control the wave front set
of the extension,
we must show
that $\chi_I$
is weakly microlocally bounded
for some $s$. Otherwise
if $\chi_I$ was badly behaving near $d_n$,
we would not be able to control the wave front set
of the extension $\overline{t_n}$!
Actually, 
we explicitely prove 
that for each point 
$\left(x,\dots,x\right)\in d_n$, 
there is a neighborhood 
$U^n$ of $(x,\dots,x)$ in $M^n$ 
where we can construct 
$\chi_I\in C^\infty(U^n)$ 
homogeneous of degree $0$ with respect to 
some specific Euler vector field $\rho$. 
$\chi_I$ is thus scale invariant which implies 
$\forall \lambda\in(0,1],\chi_{I,\lambda}=\chi_I$ 
which means 
that the family $(\chi_{I\lambda})_\lambda$ 
is bounded in 
$C^\infty(U^n\setminus d_n)$ 
hence in 
$D_{\emptyset}^\prime(U^n\setminus d_n)$.
We need these refined properties 
on $(\chi_I)_I$ since
we will have to control
the wave front set of
products of distributions
with these functions $\chi_I$.
\begin{lemm}\label{newStoralemma}
Let $(M_{I,I^c})_I$ 
be the open cover of $M^n\setminus d_n$ 
given by the geometrical lemma
\ref{geomlemm}.
Then there exists a refinement 
$(\tilde{M}_{I,I^c})_I$ of this cover 
and a subordinate 
partition of unity $(\chi_I)_I$ 
where for each 
$I,\chi_I\in C^\infty(M^n\setminus d_n)\bigcap L^1_{loc}(M^n)$ 
and for any Euler vector field $\rho$, 
$e^{\rho\log\lambda*}\left(\chi_I\right)_{\lambda\in(0,1]}$
is a bounded family 
in $\mathcal{D}^\prime_{\emptyset}(M^n\setminus d_n)$.
\end{lemm}
Note that for every
$I$,
$\chi_I$
is in $E_0(M^n)$.

\begin{proof}
For $x_0\in M$, we localize 
in a neighborhood of $(x_0,...,x_0)\in d_n$. 
Using a local chart, we identify some
neighborhood of $x_0$ with $U\subset \mathbb{R}^d$, 
on $U$ the metric reads $g$. 
We pick coordinates $(x^\mu)_\mu$ on $U$
in such
a way that 
$g_{\mu\nu}(0)dx^\mu dx^\nu=\eta_{\mu\nu}dx^\mu dx^\nu$ ($\eta$ is of signature $+,-,-,-$).  

 We are going to 
construct another poset structure
on $U^2$.
For every $x\in U$,
we denote by $C_x=\{y\geqslant x\}\cap U$ the set
of elements of $U$
in the causal future of $x$.
We consider the closed subset $\{x_i\leqslant x_j\}\cap U^2\subset U^2$. 
This set fibers on $U$:
$$\{x_i\leqslant x_j \}\cap U^2=\left(\bigcup_{x_i\in U}\{x_i\}\times C_{x_i}\right)\subset U\times U $$
Then in this local chart $U\subset \mathbb{R}^d$, set the quadratic form $Q=\eta_{\mu\nu}dx^\mu dx^\nu + c^2(dx^0)^2  $ where the aperture of the future cone of $Q$ depends on the parameter $c$. The metric $g$ depends smoothly on $x$ and thus satisfies the estimate $\vert g_{\mu\nu}(x)-\eta_{\mu\nu} \vert\leqslant C\vert x\vert$ on $U$. 
For any strictly positive $c>0$, we have the following estimate at $x_0$:   
$$\xi^0>0\text{ and }g_{\mu\nu}(0)\xi^\mu \xi^\nu=\eta_{\mu\nu}\xi^\mu \xi^\nu \geqslant 0\implies \eta_{\mu\nu}\xi^\mu \xi^\nu + c^2 (\xi^0)^2>0$$   
hence since $g_{\mu\nu}$ is
continuous we can 
find $U$ 
small enough and 
$c$ large enough 
in such
a way that 
\begin{equation}\label{conecont}
\xi^0>0, \sup_{x\in U} g_{\mu\nu}(x)\xi^\mu \xi^\nu \geqslant 0\implies \eta_{\mu\nu}\xi^\mu \xi^\nu + c^2 (\xi^0)^2>0.
\end{equation}  
Set $\tilde{C}$ the future solid cone defined by the constant metric 
$Q_c=\eta_{\mu\nu}dx^\mu dx^\nu+c(dx^0)^2$, $\tilde{C}$ is given by the equations:
\begin{eqnarray}
x^0\geqslant 0\\
Q_c(x)\geqslant 0.
\end{eqnarray}
Intuitively, if $c\rightarrow \infty$, the future cone 
$\tilde{C}$ for the constant metric $Q$ 
has solid angle which tends to $2\pi$. 
Hence for  
$c$ sufficiently large, equation (\ref{conecont}) means that 
the future cone $\tilde{C}$ contains all future conoids $C_x$ for all $x\in U$.
Then: 
$$\{x_i\leqslant x_j\} 
\subset \bigcup_{x_i\in U}\{x_i\}\times\tilde{C}\subset U\times U .$$
\begin{figure} %on ouvre l'environnement figure
\begin{center}
\includegraphics[width=12cm]{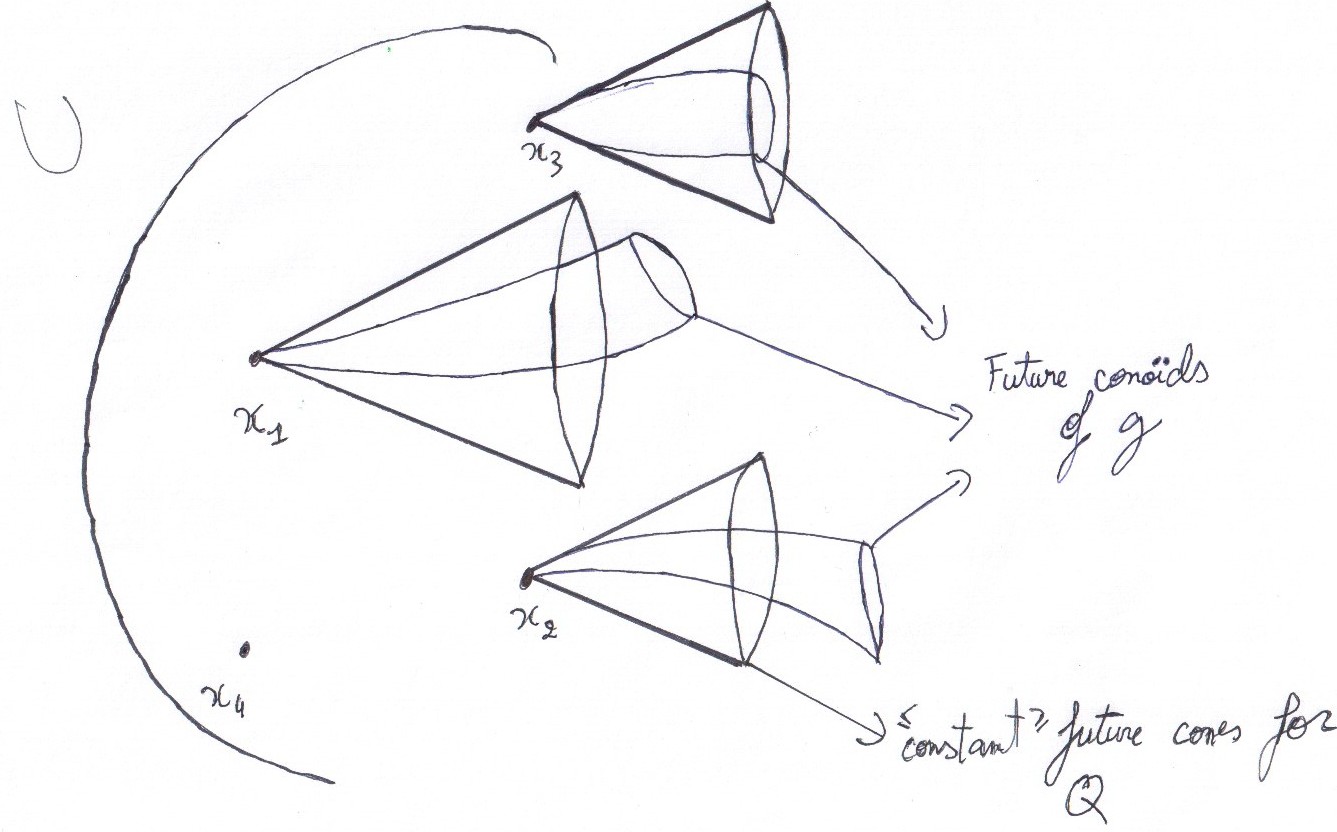} %ou image.png, .jpeg etc.
\caption{$C_{123}$ for the partial order $\leqslant$ and for $\tilde{\leqslant}$} %la légende
\label{Soften_Stora} %l'étiquette pour faire référence à cette image
\end{center}
\end{figure} %on ferme l'environnement figure 

 $\tilde{C}$ defines a 
new partial order relation $\tilde{\geqslant}$, 
hence a new poset structure on 
$U$ defined as follows:
\begin{equation}
x_j\tilde{\geqslant} x_i \text{ if } 
x^0_j-x^0_i\geqslant 0 \text{ and } Q(x_j-x_i)\geqslant 0,
\end{equation} 
where both the cones $\tilde{C}$ 
and the corresponding 
partial order relation are invariant 
(in the configuration 
space
$\mathbb{R}^{nd}$) 
under the action of the group 
$\mathbb{R}^*\ltimes \mathbb{R}^d $:
$$(\lambda,a)\in \mathbb{R}^*\ltimes \mathbb{R}^d : x\in\mathbb{R}^{d}  \mapsto \lambda x + a\in \mathbb{R}^{d}.$$
Define for this 
new order relation 
new open sets 
$\tilde{M}_{I,I^c}=\{ \forall(i,j)\in I\times I^c, 
x_i\tilde{\nleqslant} x_j  \}$.
Notice that if $x_i\leqslant x_j$ for the old order relation, then $x_i\tilde{\leqslant} x_j$ for the new order relation. Consequently, the sets $M_{I,I^c}$ defined for the order relation 
$\leqslant$ are  
larger than the sets $\tilde{M}_{I,I^c}$ 
defined for $\tilde{\leqslant}$.
Applying the geometrical lemma, 
we find: 
$$U^n\setminus d_n \subset \bigcup_{I\subset \{1,...,n\}} \tilde{M}_{I,I^c} .$$ 
The group $\mathbb{R}^*\ltimes \mathbb{R}^d$ 
acts on the configuration space $\mathbb{R}^{dn}$, 
for $(\lambda,a)\in \mathbb{R}^*\ltimes \mathbb{R}^d$,
we define the transformation:
$$(x_1,...,x_n)\in\mathbb{R}^{dn}  
\mapsto 
(\lambda x_1 + a,...,\lambda x_n + a)
\in \mathbb{R}^{dn}.$$
%then inspired by the invariance property of our new partial order, we \textbf{quotient} $\mathbb{R}^{nd}$ by the group of  translations and the quotient space is isomorphic to $\mathbb{R}^{d(n-1)}$. The $\tilde{M}_{I,I^c}$ are translation and dilation invariant.
%
% Then we shall consider the restriction of the open cover $\tilde{C_I} \mod \mathbb{R}^d$ to the submanifold $\sum_{i=2}^n \vert h_i\vert^2=1$ which is an embedded sphere $\mathbb{S}^{(n-1)d-1}$ in the quotient space.

 We describe our 
construction 
in terms of \textbf{fibrations} 
of $\mathbb{R}^{nd}\setminus d_n$. 
%\textbf{by the orbits} of the group of translations over a first quotient space $\mathbb{R}^{d(n-1)}\setminus (0,\dots,0)$. Then we do a 
%second quotient 
%by a fibration of the first quotient space 
%$\mathbb{R}^{d(n-1)}\setminus (0,\dots,0)$ 
%by the orbits of the group of dilations:
$$\begin{array}{ccccc}\mathbb{R}^{nd}\setminus d_n &\longrightarrow & \mathbb{R}^{d(n-1)}\setminus (0,\dots,0) & \longrightarrow & \mathbb{S}^{(n-1)d-1}  
\\ (x_1,\dots,x_n) &\longmapsto  & (h_2=x_2-x_1,\dots,h_n=x_n-x_1) & \longmapsto   &  (\frac{h_2}{\sum_{i=2}^n h^2_i},\dots,\frac{h_n}{\sum_{i=2}^n h^2_i}) \end{array}$$
The first quotient is by the group of translation.
The image of $d_n$ by 
the first projection is the origin $(0,\dots ,0)\in \mathbb{R}^{d(n-1)}$.
The second quotient is by the group of dilations.
We denote by $\pi$ the projection
$\pi:(x_1,\dots,x_n)\mathbb{R}^{nd}\setminus d_n\mapsto  
(\frac{h_2}{\sum_{i=2}^n h^2_i},\dots,\frac{h_n}{\sum_{i=2}^n h^2_i})\in\mathbb{S}^{(n-1)d-1} $.
%Geometrically, $\sum_{i=2}^n \vert h_i\vert^2=1$ intersects the orbits of the group $\mathbb{R}^*$ acting on the quotient configuration space $\mathbb{R}^{d(n-1)}\setminus (0,\dots,0)$.
% The order relation is induced on the submanifold, 
%$$i\leqslant j \Leftrightarrow  h_j^0-h_i^0\geqslant 0, Q(h_j - h_i)\geqslant 0 $$
%hence the $C_I$ are induced on the quotient space $\mathbb{S}^{(n-1)k-1}$, notice that for any pair of points on the $\mathbb{S}^{(n-1)k-1}$ either $h_i\nleqslant h_j $ either $h_j\nleqslant h_i$.
The open cover $(\tilde{M}_{I,I^c})_I$
are inverse images of some open cover $(\tilde{m}_{I,I^c})_I$ of the sphere $\mathbb{S}^{(n-1)d-1}$. Let  $(\varphi_I)_I$ 
be a partition of unity subordinate to 
the open cover $(\tilde{m}_I)_I$
of $\mathbb{S}^{(n-1)d-1}$.
Then we pull-back the functions 
$(\varphi_I)_I$ on $\mathbb{R}^{nd}\setminus d_n$
and set $\forall I, \chi_I=\pi^*\varphi_I$:
$$\chi_I(x_1,\cdots,x_n)=\varphi_I(\frac{x_2-x_1}{\sqrt{\sum_2^n (x_j-x_1)^2}},\cdots,\frac{x_n-x_1}{\sqrt{\sum_2^n (x_j-x_1)^2}}).$$ 

 The collection of functions $\left(\chi_I\right)_I$ are both scale and translation invariant by
the Euler vector field $\rho=\sum_{j=2}^n (x_j-x_1)\left(\partial_{x_j}-\partial_{x_1}\right)$. 
In the relative coordinate system $(x_1,h_{21}=x_2-x_1,...,h_{n1}=x_n-x_1)$, we notice that the collection  
$(\chi_{I})_I$ only depends on the $(h_{i1})_{i\geqslant 2}$. $\chi_{I}$ is smooth in $\mathbb{R}^{nd}\setminus d_n$ hence $\chi_I\in \mathcal{D}^\prime_{\emptyset}(U^n\setminus d_n)$. If we scale linearly, we notice $(\chi_I)_\lambda(h)=\chi_I(\lambda h)=\chi_I(h)$ thus the family $(\chi_I)_\lambda$ is bounded in $\mathcal{D}^\prime_{\emptyset}(U^n\setminus d_n)$. However, we know that 
the boundedness of this family in $\mathcal{D}^\prime_{\emptyset}(U^n\setminus d_n)$ 
and the degree of homogeneity
does not depend on the choice of Euler vector field.

  Let $(U_a)_{a\in A}$ be a locally finite cover of $M$ then the collection of open sets $(U_a)^n_a$ forms an open cover of a neighborhood of $d_n$. 
%In a local chart for $U^n$, we can simply set in local coordinates $(x_1,...,x_n)$, we do this operation for all balls $U\in M$, then by paracompactness of $M$, we can extract a countable and locally finite cover of $M$: $\left(B_{x_i}(\varepsilon_i)\right)_i$, then $\bigcup_{i} B^n_{x_i}(\varepsilon_i)$ forms an open cover of $d_n$.
% 
Let $\varphi_a$ be a partition of unity subordinate to the cover $(U_a)_a^n$.
Then we can patch together the 
various functions $\chi_{I,a}$ 
constructed from the cover 
by the formula 
$$\tilde{\chi}_I=\sum_a \frac{\chi_{I,a}
\varphi^2_a}{\sum_J\sum_a \chi_{J,a}\varphi^2_a}$$
where the sum in the denominator is locally finite.
\end{proof}

\paragraph{Remark.}

The fact that $\chi_I\in C^\infty(U^n\setminus d_n)$ does not immediately imply that the family
$(\chi_I)_{\lambda,\lambda\in [0,1]}$ is \textbf{bounded} in $\mathcal{D}^\prime_{\emptyset}(U^n\setminus d_n)$. For example, 
consider the function $\sin(\frac{1}{x})\in C^\infty(\mathbb{R}\setminus \{0\})$. For any interval $[a,b]\subset \mathbb{R}\setminus \{0\}$, we can construct a sequence $\lambda_n$ which tends to $0$ such that $\frac{d}{dx}\sin(\frac{1}{\lambda_nx})=\frac{1}{\lambda_n x^2}\cos(\frac{1}{\lambda_nx})\rightarrow \infty$ hence the family $\sin(\frac{1}{\lambda x})_\lambda$ is not bounded in $C^1[a,b]$ thus it is not bounded in $\mathcal{D}^\prime_{\emptyset}(\mathbb{R}\setminus \{0\})$.

\section{The recursion.}

\paragraph{Notation, definitions.}
We
denote by
$x\simeq y$ if
$x$ and $y$ 
in $M$ are
connected
by a 
lightlike
geodesic
and 
$(x;\xi)\sim (y;\eta)$
if these two elements
of the cotangent
are connected
by a null bicharacteristic 
curve i.e. a Hamiltonian
curve for the 
Hamiltonian 
$g_{\mu\nu}\xi^\mu\xi^\nu\in C^\infty(T^\star M)$.\\
We denote by
$x>y$ if $x$ 
is in the future 
cone of $y$
and $x\neq y$.\\
Recall the 
configuration space $M^I$
is the set of maps 
from $I$ to $M$ 
then the small diagonal $d_I$
is just the subset 
of constant maps from
$I$ to $M$.\\
We denote by $M_{I,I^c},I\sqcup I^c=[n]$
the set of all elements $(x_1,\dots,x_n)\in M^n$ s.t.
$\forall (i,j)\in I\times I^c$ $x_i\nleqslant x_j$.
By the geometrical lemma 
the collection $\left(M_{I,I^c}\right)_{I}$
forms an open cover
of $M^n\setminus d_n$ and we denote by 
$(\chi_I)_I$ the subordinate
partition of unity.\\
$E_g^+$ is the set of all elements 
in cotangent space having \textbf{positive energy}, 
the concept of
positivity of energy being defined 
relative to the choice of Lorentzian metric $g$. 
\begin{defi}
$E_g^+=\{(x,\xi) |  g_x(\xi,\xi)\geqslant 0,\xi_0>0  \}\subset T^\bullet M$.
\end{defi}
It is a 
\textbf{closed conic convex} set of $T^\bullet M$ and has the property that $E_g^+\cap -E_g^+=\emptyset$. 
We will denote by $E^+_{g,x}$ the component of $E_g^+$ 
living in the fiber $T_{x}^\bullet M$ over $x$.\\
\paragraph{Causality equation and wave front sets.}
The fact that for all $n$, $t_n \in Hom(H^n,\mathcal{D}^\prime \left(M^n\right) )$ satisfies the causality equation imposes some constraints on the wave front set of $t_n$. 
In $M^n$ with coordinates 
$(x_i)_{i\in\{1,\dots,n\}}$, 
$(\chi_I)_I$ is 
the partition of unity subordinate to the cover 
$(M_{I,I^c})_{I}$ of $M^n\setminus d_n$ 
given by the
improved geometrical lemma.
For all $n$, $t_n(A)\in \mathcal{D}^\prime\left(M^n\setminus d_n\right)$ 
satisfies the equation:
\begin{equation}\label{key formula}
t_n(A)=\sum_{M_{I,I^c}}\sum \chi_I t_I(A_{I(1)})t_{I^c}(A_{I^c(1)})\left( A_{I(2)} \vert A_{I^c(2)}\right),
\end{equation} 
where $\left(\phi(x_i) \vert \phi(x_j) \right)=\Delta_+(x_i,x_j)$.
For the sake of simplicity, each of the term 
$t_I(A_{I(1)})t_{I^c}(A_{I^c(1)})\left( A_{I(2)} \vert A_{I^c(2)}\right)$
in the above sum writes:
\begin{equation}\label{key formula simplified}
t_I  
\left(\prod_{ij\in I\times I^c} 
\Delta^{m_{ij}}_+(x_i,x_j) \right)t_{I^c}.
\end{equation}
since each
Laplace coupling 
$((A_I)_{(2)} \vert (A_{I^c})_{(2)})=(\prod_{i\in I}\phi^{k_i}(x_i)\vert \prod_{j\in I^c}\phi^{k_j}(x_j))$
is a product of Wightman propagators: 
$\left(\prod_{ij\in I\times I^c} \Delta^{m_{ij}}_+(x_i,x_j) \right)$,  
$t_I=t_I(A_{I(1)})$ 
and $t_{I^c}=t_{I^c}(A_{I^c(1)})$.
We now face the problem of 
defining $t_n$ 
recursively by using the equation 
(\ref{key formula simplified}),
the difficulty is to make sense 
of the r.h.s. of 
(\ref{key formula simplified}) 
on $M^n\setminus d_n$ 
which is a 
problem of multiplication of distributions 
and the second difficulty 
is to extend the distribution $t_n\in \mathcal{D}^\prime(M^n\setminus d_n)$ (while retaining nice analytical properties) which is only defined on $M^n\setminus d_n$ to a distribution defined on $M^n$.
We prove that renormalisability is local in $M$, 
for all $p\in M$ 
there exists an open neighborhood $\Omega$
of $p$
on 
which all $t_n$ are well defined 
as elements of
$\mathcal{D}^\prime(\Omega^n\setminus d_n)$
and can be extended
as elements of 
$\mathcal{D}^\prime(\Omega^n)$.
In the sequel, using a local chart
around $p$, we will identify $\Omega$
with an open set $U\subset \mathbb{R}^{d}$.
In $U$, the metric reads $g$.
The main theorem we prove is the following
\begin{thm}
The set of equations (\ref{key formula}) can be solved
recursively in $n$, where 
for each $n$,
if all $t_I, I\varsubsetneq [n]$ are given
then the product of distribution
makes sense on $M^n\setminus d_n$ and 
defines a unique element
$t_n\in \mathcal{D}^\prime(M^n\setminus d_n)$ which 
has
some extension 
in $\mathcal{D}^\prime(M^n)$.
\end{thm}
We first treat the problem
of multiplication
of distributions 
outside $d_n$, to do this,
we develop
a machinery
which allows us to describe
wave front sets
of Feynman amplitudes.
\subsection{Polarized conic sets.}
The idea of polarization is inspired by 
the exposition of Yves Meyer of Alberto Calderon's result on 
the product of 
$\Gamma$-holomorphic distributions (\cite{Meyerprod} p.~604 
definition 1). 
In $\mathbb{R}^n$ with coordinates $(x_1,\dots,x_n)$, the $\Gamma$-holomorphic distributions studied by Meyer are tempered distributions having their Fourier transform supported on a closed convex cone $\Gamma$ in the Fourier domain which is contained in the upper half plane $\xi_n>0$.
The beautiful remark of Meyer is that $\Gamma$-holomorphic distributions can always be multiplied (the product extends to $\Gamma$-holomorphic distributions) and form an algebra for the extended product (because of the convexity of $\Gamma$ the convolution product in the Fourier domain preserves is still supported on $\Gamma$)!  
For QFT, we are let to introduce
the concept
of \textbf{polarization}
to describe
subsets
of the cotangent of configuration 
spaces 
$T^\bullet M^n$ for all $n$:
this generalizes
the concept of positivity
of energy 
for 
the cotangent space 
of configuration space.
 
In order to generalize this condition 
to the wave front set of $n$-point functions, 
we define the right concept of 
positivity of energy which is adapted 
to conic sets in $T^\bullet M^n$:
\begin{defi}\label{polarized}
We define a \textbf{reduced polarized part} 
(resp \textbf{reduced strictly polarized part}) 
as a conical subset
$\Xi \subset T^*M$ such that, 
if $\pi:T^*M\longrightarrow M$ 
is the natural projection, 
then $\pi(\Xi)$ 
is a finite subset $A=\{a_1,\cdots, a_r\}\subset M$
and, 
if $a\in A$ is maximal (in the sense 
there is no element $\tilde{a}$ in $A$ s.t. $\tilde{a}>a$), 
then $\Xi\cap T_a^*M \subset (-E_g^+\cup \{0\})$ 
(resp $\Xi\cap T_a^*M \subset (-E_g^+)$)
where $E_g^+$ is the subset of elements
of $T^\star M$ of positive energy.
\end{defi}
We define
the trace
operation
as a map
which 
associates
to each element 
$p=(x_1,\dots,x_n;\xi_1,\dots,\xi_n)\in\left(T^*M\right)^k$
some finite part
$Tr(p)\subset T^\star M$.
\begin{defi}
For all elements $p=((x_1,\xi_1),\cdots,(x_k,\xi_k)) \in T^*M^k$, 
we define
the \textbf{trace} $Tr(p)\subset T^*M$ 
defined by the set of elements
$(a,\eta)\in T^\star M$
such that 
$\exists i\in [1,k]$ 
with the property that 
$x_i=a$, $\xi_i\neq 0$ and
$\eta=\sum_{i;x_i=a}\xi_i$.
\end{defi}
Then 
finally, 
we can define 
polarized subsets $\Gamma \subset T^*M^k$:
\begin{defi}
A conical subset $\Gamma \subset T^*M^k$ 
is \textbf{polarized} (resp strictly polarized) 
if for all $p\in\Gamma$, 
its trace $Tr(p)$  
is a reduced polarized part 
(resp reduced strictly polarized part) of $T^*M$.
\end{defi}
The union of two polarized (resp strictly polarized) subsets 
is polarized (resp strictly polarized) and if a conical subset
is contained in a polarized subset
it is also polarized.
%This definition
%which seems
%very artificial
%and complicated
%at first sight
%is important since
%we will see that
%it is satisfied by
%the two point functions 
%$t_2$, the Laplace couplings
%and the polarization 
%property
%allows to multiply distributions 
%and is stable by products.
\begin{figure} %on ouvre l'environnement figure
\begin{center}
\includegraphics[width=12cm]{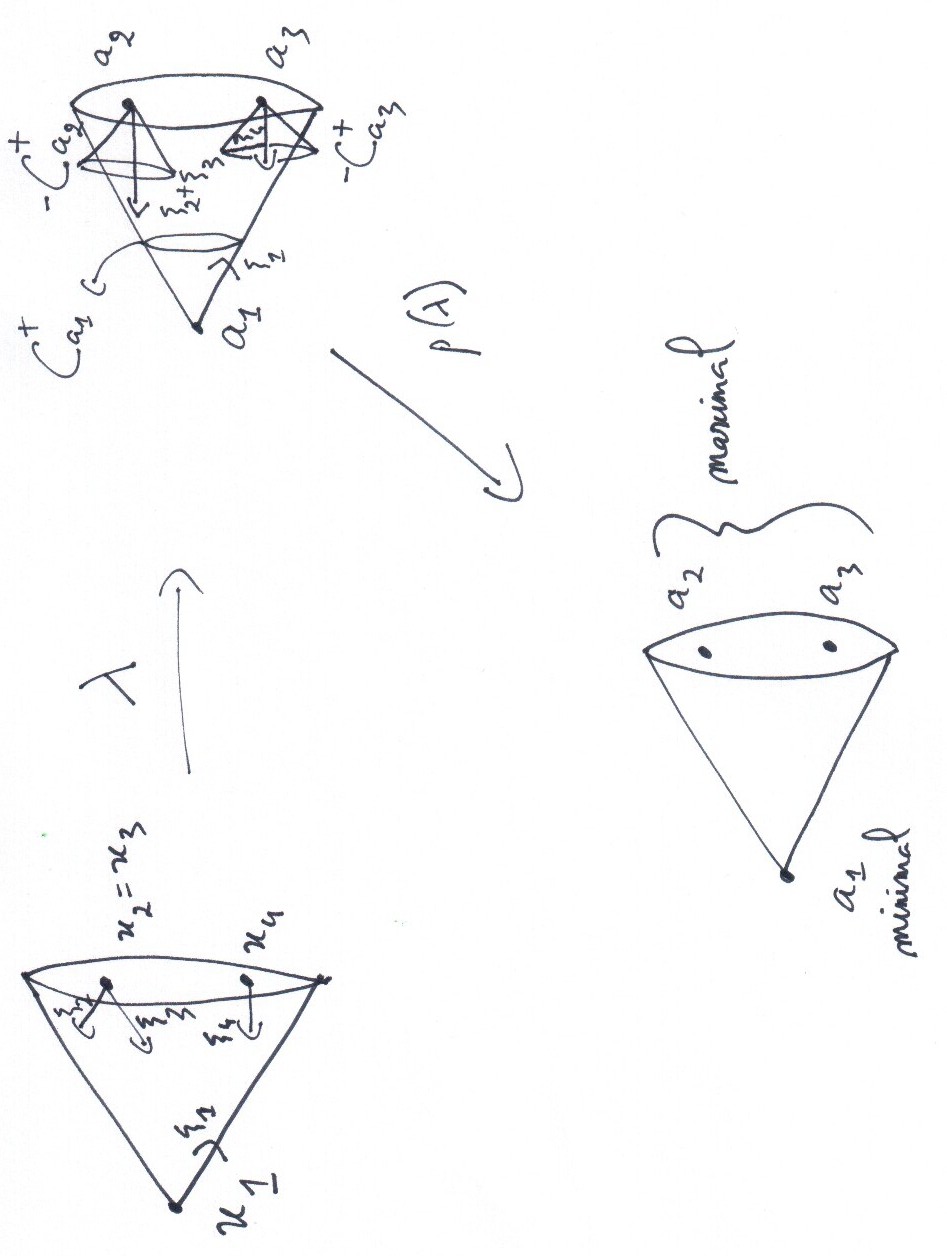} %ou image.png, .jpeg etc.
\caption{A polarized set, the trace $Tr$ and the projection $\pi\circ Tr$.} %la légende
\label{polarization} %l'étiquette pour faire référence à cette image
\end{center}
\end{figure} %on ferme l'environnement figure 

The role of polarization is to control
the wave front set of the
distributions of the form
$\left\langle 0| T \phi^{i_1}(x_1)\dots\phi^{i_n}(x_n)|0 \right\rangle $.
\paragraph{The wave front set of $\Delta_+$.}

In Theorem 
\ref{Wavefrontpullback}, 
we proved that 
for all $m\in\mathbb{N}$, 
$WF(\Delta^m_+)|_{U^2\setminus d_2}
\subset \{\text{Conormal }\Gamma=0\}
\cap \left(-E^+_g\times E^+_g\right)$
where $E^+_g$ is the set of elements 
of positive energy in $T^\bullet M$.
\begin{figure} %on ouvre l'environnement figure
\begin{center}
\includegraphics[width=12cm]{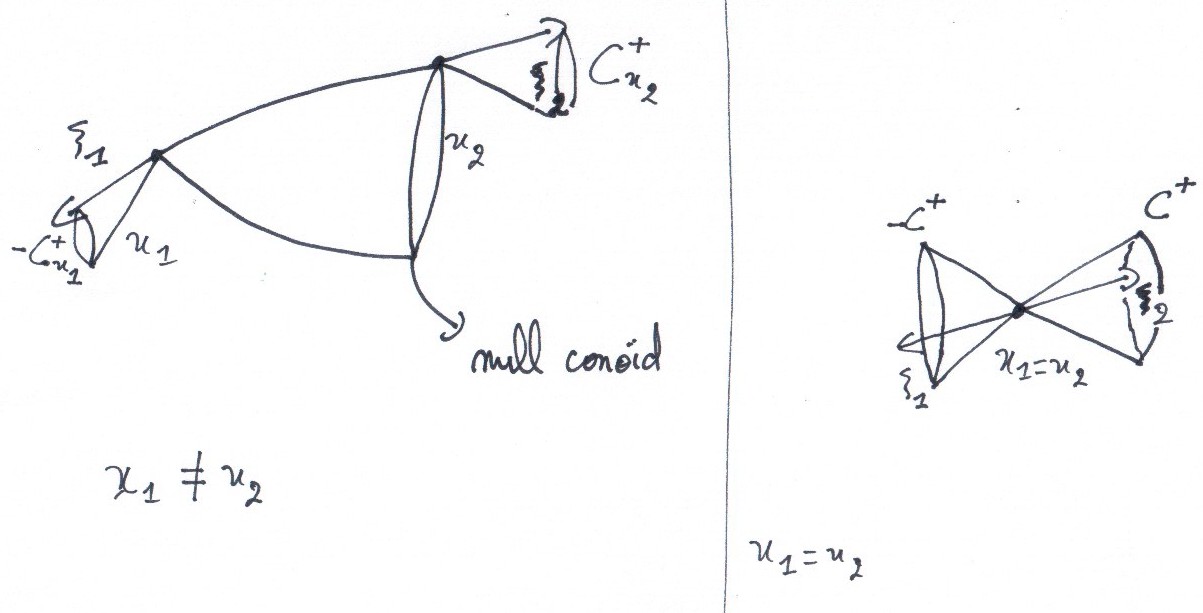} %ou image.png, .jpeg etc.
\caption{wave front set of $\Delta_+^m$.} %la légende
\label{Wavefront_delta_plus_m} %l'étiquette pour faire référence à cette image
\end{center}
\end{figure} %on ferme l'environnement figure 
Thus if $(x_1,x_2;\xi_1,\xi_2)\in WF(\Delta^m_+(x_1,x_2))|_{U^2\setminus d_2}$, two cases arise:
\begin{itemize}
\item if $x_1\nleqslant x_2$ then
we actually have $x_2\leqslant x_1 $ where $x_1\in M$ is maximal in $\{x_1,x_2\}$ and $\xi_1\in -E^+_{g,x_1}$
thus $WF(\Delta^m_+(x_1,x_2))|_{x_1\nleqslant x_2}$
is strictly polarized, 
\item if $x_2\nleqslant x_1$ then
we actually have $x_1\leqslant x_2 $ where $x_2\in M$ is maximal in $\{x_1,x_2\}$ and $\xi_2\in E^+_{g,x_1}$ thus $WF(\Delta^m_+(x_1,x_2))|_{x_2\nleqslant x_1}$
is \textbf{not polarized}.
\end{itemize}
\begin{coro}
For all $(i,j)\in I\times I^c,I\sqcup I^c=[n]$,
$WF(\Delta^m_+(x_i,x_j))|_{M_{I,I^c}}$ 
is strictly polarized.
\end{coro}
We have to check that the 
conormals of the diagonals 
$d_I$
are polarized since they are
the wave front sets
of counterterms
from the extension procedure.
\begin{prop}\label{conormpolar}
The conormal of the diagonal $d_I\subset M^I$ is polarized.
\end{prop}
\begin{proof}
Let $(x_i;\xi_i)_{i\in I}$ be 
in the conormal of
$d_I$,
let $a\in M$ s.t. 
$a=x_i,\forall i\in I$,
and $\eta=\sum \xi_i=0$
is in $-E^+_{g,a}\cup \{0\}$. Thus
the trace $Tr(x_i;\xi_i)_{i\in I}=(a;0)$
of the element $(x_i;\xi_i)_{i\in I}$ in the conormal of $d_I$
is a reduced polarized part of $T^\star M$.
\end{proof}
\begin{prop}\label{t2polar}
For all $m\in\mathbb{N}$,
if $t_2(\phi^m(x_1)\phi^m(x_2))$ satisfies the causality equation 
(\ref{key formula})
on $U^2\setminus d_2$ 
and
$WF\left(t_2(\phi^m(x_1)\phi^m(x_2))\right)|_{d_2}$ is contained 
in the conormal of $d_2$,
then the 
wave front set of $t_2(\phi^m(x_1)\phi^m(x_2))_{U^2}$ is polarized.
\end{prop}
Notice that it
is enough
to prove the 
proposition 
for $t_2(\phi(x_1)\phi(x_2))$
since $t_2(\phi^m(x_1)\phi^m(x_2))=m!t_2(\phi(x_1)\phi(x_2))^m$
on $U^2\setminus d_2$ thus
$WF\left(t_2(\phi^m(x_1)\phi^m(x_2))\right)\subset WF\left(t_2(\phi(x_1)\phi(x_2))\right)+ WF\left(t_2(\phi(x_1)\phi(x_2))\right)$ 
on $U^2\setminus d_2$.

\begin{proof}
Notice that if $x_1\nleqslant x_2$ then
$T\phi(x_1)\phi(x_2)=\phi(x_1)\star\phi(x_2)$
by the definition of causality 
i.e. $T(AB)=TA\star TB$ if $A\nleqslant B$.
Thus
the field $\phi(x_1)$
associated with 
the element $x_1$, 
where $x_1$ is not in the causal past
of $x_2$, stands on the left
of the product 
$\phi(x_1)\star\phi(x_2)$.
Causality
reads from right to left 
when we write products of
fields i.e. $T(AB)=TA\star TB$ if $A\nleqslant B$.
$$t_2(\phi(x_1)\phi(x_2))=\varepsilon\left( T_2\phi(x_1)\phi(x_2)\right)=\varepsilon\left( \phi(x_1)\star\phi(x_2)\right)  $$
$$ 
\begin{array}{c}
=\Delta_+(x_1,x_2) \text{ if } x_1\nleqslant x_2 \\ 
=\Delta_+(x_2,x_1) \text{ if } x_2\nleqslant x_1,
\end{array}
$$
which implies $WF\left(t_2\right)|_{U^2\setminus d_2}$
is polarized.
Using Proposition \ref{conormpolar} and the fact
that $WF(t_2)|_{d_2}$ is contained in the conormal of $d_2$, 
it is immediate to deduce $WF(t_2)$
is polarized.
\end{proof}

Now we will prove the 
key theorem 
which allows 
to multiply 
two distributions
under some conditions 
of polarization on 
their
wave front sets
and deduces specific 
properties
of the wave front set
of the product:
\begin{thm}\label{lemm2} 
Let $u,v$ be two distributions
in $\mathcal{D}^\prime(\Omega)$, for some subset $\Omega\subset M^n$,
s.t. 
$WF(u)\cap T^\bullet \Omega$ is polarized
and $WF(v)\cap T^\bullet \Omega$ is strictly
polarized. Then
the
product $uv$
makes sense in
$\mathcal{D}^\prime(\Omega)$
and
$WF(uv)\cap T^\star \Omega$ 
is polarized. Moreover,
if $WF(u)$ is also 
strictly
polarized then
$WF(uv)$ is strictly
polarized.
\end{thm}
\begin{proof}
Step 1: we prove $WF(u)+WF(v)\cap T^\star \Omega$ 
does not meet 
the zero section.
For any element $p=(x_1,\dots,x_n;\xi_1,\dots,\xi_n)\in T^\star M^n$
we denote by $-p$ 
the element $(x_1,\dots,x_n;-\xi_1,\dots,-\xi_n)\in T^\star M^n$.
Let $p_1=(x_1,\dots,x_n;\xi_1,\dots,\xi_n)\in WF(u)$
and $p_2=(x_1,\dots,x_n;\eta_1,\dots,\eta_n)\in 
WF(v)$,
necessarily we must have 
$(\xi_1,\dots,\xi_n)\neq 0,(\eta_1,\dots,\eta_n)\neq 0$.
We will show
by a contradiction 
argument 
that the sum 
$p_1+p_2=(x_1,\dots,x_n;\xi_1+\eta_1,\dots,\xi_n+\eta_n)$
does not meet the zero section.
Assume that $\xi_1+\eta_1=0,\dots,\xi_n+\eta_n=0$
i.e. $p_1=-p_2$ then we would 
have $\xi_i=-\eta_i\neq 0$
for some $i\in\{1,\dots,n\}$ since
$(\xi_1,\dots,\xi_n)\neq 0,(\eta_1,\dots,\eta_n)\neq 0$.
We assume w.l.o.g. that $\eta_1\neq 0$, thus
$Tr(p_2)$ is non empty ! Let $B=\pi(Tr(p_1)),C=\pi(Tr(p_2))$,
we first notice $B=C$ since 
$p_2=-p_1\implies Tr(p_1)=-Tr(p_2)\implies \pi\circ Tr(p_1)=\pi\circ Tr(p_2)$.
Thus if $a$ is maximal
in $B$, $a$ is
also maximal in $C$ and we have
$$0=\sum_{x_i=a} \xi_i+\eta_i = \sum_{x_i=a} \xi_i +\sum_{x_i=a}\eta_i \in \left(E^-_{g,a}\cup\{0\}+E^-_{g,a}\right)=E^-_{g,a},   $$
where we denote $E^-_{g,a}=-E^+_{g,a}$ 
for notational clarity,
(since $p_1$ is  
polarized and $p_2$ is strictly polarized) 
contradiction ! 

 Step 2, we prove that
the set
$$\left(WF(u)+WF(v)\right) \cap T^\star \Omega$$ 
is strictly polarized.
Recall $B=\pi\circ Tr(p_1)$, $C=\pi\circ Tr(p_2)$
and we denote by $A=\pi\circ Tr(p_1+p_2)$
hence in particular $A\subset B\cup C$.
We denote by $\max A$ (resp $\max B,\max C$) the set of maximal elements in $A$
(resp $B,C$). 
The key argument is to prove that $\max A=\max B\cap \max C$.
Because if $\max A=\max B\cap \max C$ holds
then for any $a\in \max A$, $\sum_{x_i=a} \xi_i+\eta_i=\sum_{x_i=a} \xi_i+\sum_{x_i=a}\eta_i\in -E^+_{g,a}$ since $a\in \max B\cap \max C$ and $Tr(p_1)$ is
a reduced polarized part and $Tr(p_2)$ is reduced strictly polarized. 
Thus $\max A=\max B\cap \max C$ implies that 
$p_1+p_2$ is strictly polarized.\\ 
We first establish the inclusion $\left(\max B\cap \max C\right)\subset \max A$. 
Let $a\in \max B\cap \max C$, then $\sum_{x_i=a} \xi_i\in E^-_{g,a}\cup\{0\}$ and 
$\sum_{x_i=a} \eta_i\in E^-_{g,a}$. 
Thus  
$\sum_{x_i=a} \xi_i+\eta_i\in E^-_{g,a}\implies \sum_{x_i=a} \xi_i+\eta_i\neq 0$ 
so there must exist some $i$
for which $x_i=a$ and $\xi_i+\eta_i\neq 0$. Hence $a\in A$. Since $A\subset B\cup C$, $a\in \max B\cap \max C$, 
we deduce that $a\in\max A$ (if there were $\tilde{a}$ in $A$
greater than $a$ then $\tilde{a}\in B$ or $\tilde{a}\in C$ and $a$
would not be maximal in $B$ and $C$).

 We show the converse inclusion $\max A\subset\left(\max B\cap \max C\right)$ by contraposition.
Assume $a\notin \max B$, then there exists $x_{j_1}\in \max B$ 
s.t.
$x_{j_1}>a$ 
and $\xi_{j_1}\neq 0$. There are two cases
\begin{itemize}
\item either $x_{j_1}\in \max C$ as well, then
$\sum_{x_{j_1}=x_i}\xi_i+\eta_i\in -E^{+}_{g,x_{j_1}}
\implies \sum_{x_{j_1}=x_i}\xi_i+\eta_i\neq 0$
and there is some $i$ for which
$x_i=x_{j_1}$ and $\xi_i+\eta_i\neq 0$ thus 
$x_{j_1}\in A$ and $x_{j_1}>a$ hence $a\notin\max A$.
\item or $x_{j_1}\notin \max C$ 
then there exists $x_{j_2}\in \max C$ s.t. $x_{j_2}>x_{j_1}$ and $\eta_{j_2}\neq 0$.
Since $x_{j_1}\in \max B$,  
we must have $\xi_{j_2}=0$ so that $x_{j_2}\notin B$. 
But we also have $\xi_{j_2}+\eta_{j_2}=\eta_{j_2}\neq 0$ 
so that $x_{j_2}\in A$.
Thus $x_{j_2}\in A$ is greater than $a$ 
hence $a\notin\max A$.
\end{itemize}
We thus proved
$$a\notin\max B \implies a\notin \max A $$
and 
by symmetry of
the above arguments
in $B$ and $C$, 
we also have
$$a\notin\max C\implies a\notin \max A .$$
We established that $(\max B)^c\subset (\max A)^c$ and $(\max C)^c\subset (\max A)^c$,
thus $(\max B)^c\cup (\max C)^c\subset  (\max A)^c$ therefore 
$\max A\subset \max B\cap \max C$,
from which we deduce the equality $\max A=\max B\cap \max C$ which implies that 
$WF(u)+WF(v)$ is strictly polarized and $WF(uv)$ is polarized.
\end{proof}
\begin{lemm}\label{lemm1}
For all $I\sqcup I^c=[n]$, $(k_i)_{i\in I},(k_j)_{j\in I^c}$ 
s.t. 
$\sum_{i\in I} k_i=\sum_{j\in I^c} k_j$ 
the Laplace coupling
$$\left(\prod_{i\in I}\phi^{k_i}(x_i) \vert \prod_{j\in I^c}\phi^{k_j}(x_j) \right) $$
is well defined
in the sense of distributions of $\mathcal{D}^\prime\left(M_{I,I^c}\right)$
and 
its wave front set
is strictly polarized.
\end{lemm}
\begin{proof}
First
the coupling 
$\left(\prod_{i\in I}\phi^{k_i}(x_i) \vert \prod_{j\in I^c}\phi^{k_j}(x_j) \right) $
is a finite sum 
of terms
of the form
$\underset{(i,j)\in I\times I^c}{\prod}\Delta_+^{m_{ij}}(x_i,x_j),m_{ij}\in\mathbb{N}$.
However 
$$WF(\underset{(i,j)\in I\times I^c}{\prod}\Delta_+^{m_{ij}}(x_i,x_j)|_{M_{I,I^c}})$$
is strictly polarized
by application of lemma
\ref{lemm2} 
since $WF(\Delta_+^{m_{ij}}|_{M_{I,I^c}})$
is strictly polarized.
%We denote by
%$p=(x_1,\dots,x_n; \eta_1,\dots,\eta_n)$
%the element $\underset{(i,j)\in A\times B}{\sum}p_{ij}$.
%Let $a$ be maximal 
%in $\pi(Tr(p))$ then $\exists i\in I, x_i=a$.
%Furthermore for all $i\in I$ s.t.
%$x_i=a$,
%$$\eta_i=\sum_{j\in B}\eta_{i}^{ij}\in 
%\sum_{j\in B} \left(-E^+_{g,a}\right)=\left(-E^+_{g,a}\right) $$
%which means $\eta=\sum_{x_i=a}\eta_i\in \left(-E^+_{g,a}\right)$
%and $Tr(p)$ is strictly
%polarized. 
\end{proof}
\begin{figure} %on ouvre l'environnement figure
\begin{center}
\includegraphics[width=10cm]{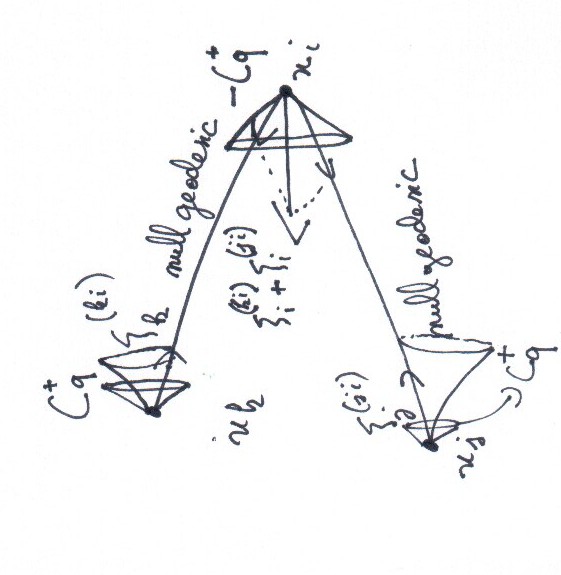} %ou image.png, .jpeg etc.
\caption{The Wavefront of Laplace couplings is strictly polarized.} %la légende
\label{Wavefront_Laplace_Coupling} %l'étiquette pour faire référence à cette image
\end{center}
\end{figure} %on ferme l'environnement figure 
\begin{lemm}\label{lemm4}
Let $t_I,t_{I^c}$
be in $\mathcal{D}^\prime(M^I),\mathcal{D}^\prime(M^{I^c})$
respectively
s.t. $WF(t_I)$
and 
$WF(t_{I^c})$
are polarized
then $WF(t_I  t_{I^c})|_{M_{I,I^c}}$
is polarized.
\end{lemm}
\begin{proof}
%Let us recall
%that
%$WF(t_It_{I^c})=(WF(t_I)\cup \underline{0})\times (\underline{0}\cup WF(t_{I^c}))\setminus \underline{0}$.
%The only thing that we must
%check is that $WF(t_I)\cup \underline{0}$
%is polarized but this 
%is obvious since for all
%$p\in WF(t_I)$, $p\cup \underline{0}\in WF(t_I)\cup \underline{0}$
%and $Tr(p)=Tr(p\cup \underline{0})$. The same argument
%applies to $\underline{0}\cup WF(t_{I^c})$
%and the claim of the lemma follows from
%lemma \ref{lemm3}.
For all $(x_i,x_j;\xi_{i}^I,\xi_j^{I^c})_{(i,j)\in I\times I^c}\in WF( t_I t_{I^c})|_{M_{I,I^c}}$,
$Tr(x_i,x_j;\xi_{i}^I,\xi_j^{I^c})=Tr(x_i;\xi_{i}^I)\cup Tr(x_j;\xi_j^{I^c})$
because for all $(x_1,\dots,x_n)\in M_{I,I^c}$ for all $(i,j)\in I\times I^c$, $x_i\neq x_j$.
Then using the fact that
$Tr(x_i;\xi_{i}^I)$ and $Tr(x_j;\xi_j^{I^c})$
are polarized,
for all 
$a$ maximal in 
$\pi\circ Tr(x_i,x_j;\xi_{i}^I,\xi_j^{I^c})$:

 either $a$ is maximal in  $Tr(x_i;\xi_{i}^I)$ in which
case $\eta=\sum_{x_i=a}\xi_i^I\in -E^+_{g,a}\cup\{0\}$
since $Tr(x_i;\xi_{i}^I)$ is polarized,

 either $a$ is maximal in  $Tr(x_j;\xi_{j}^{I^c})$
and we deduce the 
same kind of result
$$\eta=\sum_{x_j=a}\xi_j^{I^c}\in -E^+_{g,a}\cup\{0\}$$
since $Tr(x_j;\xi_{j}^{I^c})$ is polarized.
\end{proof}

\begin{thm}\label{goodthm1} 
Let $t_I,t_{I^c}$ be distributions
in $\mathcal{D}^\prime_{\Gamma_I},\mathcal{D}^\prime_{\Gamma_{I^c}}$
where $\Gamma_I,\Gamma_{I^c}$ are polarized
in $M^I$ and $M^{I^c}$ and $m_{ij}$
be a collection of integers.
Then the product
$$t_It_{I^c}\prod_{(ij)\in I\times I^c} \Delta_{+}^{m_{ij}}(x_i,x_j)$$
is well defined as a distribution of 
$D_{\Gamma_n}^\prime(M_{I,I^c})$ 
for 
$$\Gamma_n=\sum_{I}\left(\Gamma^0_I
+\Gamma^0_{I^c}
+\sum_{ij}\Gamma^0_{ij}\right)\bigcap T^\bullet M_{I,I^c}$$
and
$\Gamma_n$
is polarized.
Furthermore, $t_n$ defined by the relation
(\ref{key formula}) is well defined
in $\mathcal{D}^\prime(U^n\setminus d_n)$
and its wave front set is polarized in $M^n\setminus d_n$. 
\end{thm}
\begin{proof}
$WF(t_It_{I^c})$
is polarized in $M_{I,I^c}$ 
by Lemma
\ref{lemm4}, each Laplace coupling
is strictly polarized in $M_{I,I^c}$
by Lemma \ref{lemm1} hence by 
Theorem \ref{lemm2}
the product
$$t_It_{I^c}\prod_{(ij)\in I\times I^c} \Delta_{+}^{m_{ij}}(x_i,x_j)$$
exists and its wave front set is polarized over $M_{I,I^c}$.
We sum and multiply each term $\sum t_I(A_{I1})t_{I^c}(A_{I^c1})(A_{I2}|A_{I^c2})$ by the functions
$\chi_I$ of the partition of unity from the geometrical lemma 
which does not affect the wave front set
since they are 
smooth
on $M^n\setminus d_n$, 
thus the wave 
front set of $t_n$ defined by (\ref{key formula}) 
is the finite union
of polarized conical subsets
thus polarized.
\end{proof}

\subsection{Localization and enlarging the polarization.} 
In the previous
part,
we were able to justify
the products
of distributions 
on $M^n\setminus d_n$ in equation
\ref{key formula}
but have not yet
extended 
the distribution $t_n$
on $M^n$.
The goal
of this part 
is to
prove that we can
construct some
polarized cone $\Gamma_I$,
slightly larger than $WF(t_I)$,
which is scale invariant for 
some family
of linear Euler vector fields
and satisfies the soft landing condition.
The 
drawback of working
with the cone $E^+_g\subset T^\bullet U$
is that the cones $E^+_{gx}\subset T_x^\bullet U$  
depend
on the point $x$.
We will construct a larger closed
convex conic 
$E^+_q$ for a constant metric
$q$ which contains $E^+_g$
 and which has fibers
$E^+_{qx}$ that 
do not depend 
on $x\in U$.

We identify an open set $\Omega\subset M$ with $U\subset \mathbb{R}^d$, 
in $U$ the metric reads $g$.  
Then we soften the poset relation 
in a similar way to the step 
2 and 3 in the proof of the improved geometrical lemma 
(\ref{newStoralemma}). 
We use a constant metric 
$Q$ to define a 
new partial order denoted
by $\tilde{\leqslant}$.
Recall $E_g^+ \subset T^\star M$ is 
the subset of elements in cotangent space of positive energy. 
We prove a lemma 
which says we can localize in a domain 
$U\subset \mathbb{R}^d$ 
in which we can control the wave front set of the family $(\Delta_+)_\lambda$, $$\forall\lambda\in(0,1],WF\left(\Delta_{+\lambda}\right)\subset \left(-E_q^+ \right) \times \left(E^+_q\right) $$ 
by a scale and translation invariant set $E_q^+$ 
living in cotangent space $T^\bullet U$. 

\begin{lemm}
For any $x_0\in U$, 
we can always make 
$U$ smaller
around $x_0$ so as to be able
to construct a \textbf{closed conic convex} 
set $E_q^+\subset T^\bullet U$ s.t. $E_g^+\subset E^+_q$, $E^+_q$ does not depend on $x\in U$ and such that $E_q^+\cap -E_q^+=\emptyset$. 
\end{lemm}
\begin{proof}
We enlarge 
the cone of positive energy $E_g^+\subset T^\bullet U$.
Recall we defined $E_g^+$ as $E_g^+=\{(x;\xi) |  g_x(\xi,\xi)\geqslant 0,\xi_0>0  \}\subset T^\bullet M$. 
But the drawback of this definition 
lies in the fact that the fibers $E_{gx}^+$ of the set $E_g^+$ depend on the base point $x$ since $g$ is variable.
We localize the construction in a sufficiently small open ball $U$ in $\mathbb{R}^{d}$ and pick a constant metric $q$ on this ball $U$ in such a way that 
\begin{equation}\label{ineq}
\forall x \in U, g_x(\xi,\xi)\geqslant 0,\xi_0>0\implies   q(\xi,\xi)>0. 
\end{equation} 
Such a metric is easy to construct, 
following the arguments of the proof of the 
improved geometrical lemma, 
we assume $g_{x_0}^{\mu\nu}=\eta^{\mu\nu}$ and by setting
$q=\eta^{\mu\nu}+\lambda^2\delta^{00}$, we can always choose $\lambda$ large enough so that the inequality (\ref{ineq}) is satisfied for all $x\in U$.
\begin{defi}
We set $E_q^+=\{(x,\xi) | q(\xi,\xi)\geqslant 0,\xi_0>0,x\in U \}$.
\end{defi}
It is immediate by construction 
that our new 
\textbf{closed, conic, convex} 
set $E_q^+\subset T^\bullet M$ 
contains the old set $E^+_g$. 
It is also obvious by construction that $E_q^+$ is both \emph{scale and translation invariant} in $U$, since the metric $q$ is constant in $\mathbb{R}^{d}$.
\begin{figure} %on ouvre l'environnement figure
\begin{center}
\includegraphics[width=10cm]{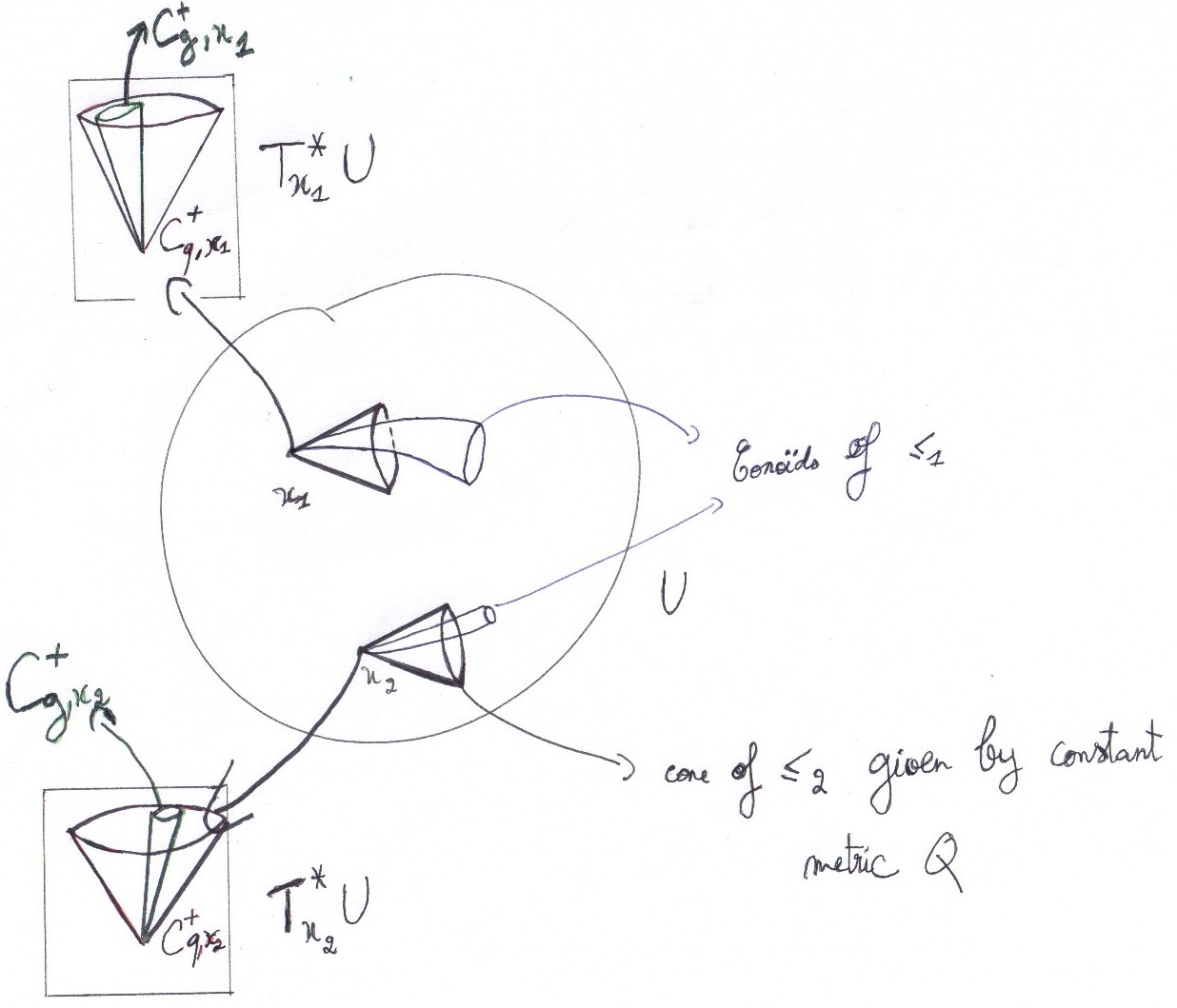} %ou image.png, .jpeg etc.
\caption{Picture of the new poset structure together with the new polarization.} %la légende
\label{Soften_metrics} %l'étiquette pour faire référence à cette image
\end{center}
\end{figure} %on ferme l'environnement figure 
\end{proof}

 We have a new definition of polarization by applying Definition (\ref{polarized}) for 
the new conic set $E_q^+$ and the partial order $\tilde{\leqslant}$ ( $\tilde{\leqslant}$ affects the choices of maximal points). Hence the metric $Q$ 
controls the order relation $\tilde{\leqslant}$
and exploits the finite propagation speed
of light, 
whereas the metric $q$ controls the
cone of positive energy.
\paragraph{Scaling in configuration spaces.}
On $U^I$, we denote the coordinates by
$(x_i)_{i\in I}$, then we define the collection $\rho_{x_i},i\in I$ of $\vert I\vert$ linear Euler fields $\rho_{x_i}=\sum_{j\neq i,j\in I}(x_j-x_i)\frac{\partial}{\partial x_j}$. $\rho_{x_i}$ scales relative to the element $x_i$ in configuration space $U^I$.
\begin{ex}
In $U^n$, the vector field
$\sum_{j\neq 1}(x_j-x_1)\partial_{x_j}$ is Euler since
$\left(\sum_{j\neq 1}(x_j-x_1)\partial_{x_j} (x_i-x_1)\right)-(x_i-x_1)=(x_i-x_1)-(x_i-x_1)=0$
and this implies that 
$\sum_{j\neq 1}(x_j-x_1)\partial_{x_j}f-f\in\mathcal{I}^2$ 
for all 
$f\in\mathcal{I}$
the ideal of functions vanishing on $d_n$.
If we scale by
$f_\lambda(x_1,\dots,x_n)=f(x_1,\lambda(x_2-x_1)+x_1,\dots,\lambda(x_n-x_1)+x_1)$
then 
this corresponds to
the Euler
vector field
$\sum_{j=1}^{n-1} h_j\frac{\partial}{\partial h^j}$.
The cotangent lift of this vector field
equals 
$$ \sum_{j\neq 1} (x_j-x_1)\partial_{x_j} -\xi_j(\partial_{\xi_j}-\partial_{\xi_1}) .$$
The vector field
$\sum_{j=2}^n\xi_j(\partial_{\xi_j}-\partial_{\xi_1})$ 
corresponds to
the system
of ODE's
$$\forall j\geqslant 2, \frac{d\xi^j}{dt}=\xi^j,\frac{d\xi^1}{dt}=\sum_{j=2}^{n}\xi^j,  $$
thus integrating the vector field
$ \sum_{j\neq 1} (x_j-x_1)\partial_{x_j} -\xi_j(\partial_{\xi_j}-\partial_{\xi_1})$
in cotangent space yields
the flow :
$$ (x_1,\lambda(x_2-x_1)+x_1,\dots;\xi_1+(1-\lambda^{-1})\sum_{j=2}^n\xi_j,\lambda^{-1}\xi_2,\dots,\lambda^{-1}\xi_n) .$$
Finally, we 
compute
the coordinate transformation
in cotangent space
which passes from regular coordinates
in cotangent space $T^\star U^n$
to the system of coordinates
$(x,h;k,\xi)$
used in Chapters 1,2,3,4:
\begin{eqnarray}\label{slccite}
(x_1,\dots,x_n;\xi_1,\dots,\xi_n)\mapsto (x,h_1,\dots,h_{n-1};k,\eta_1,\dots,\eta_{n-1})\\
x_1=x,h_j=x_{j+1}-x_1 \\
k=\sum_{i=1}^n \xi_i, \eta_j=\xi_{j+1}.
\end{eqnarray}
\end{ex}
\paragraph{The soft landing condition on configuration space.}
We saw in Chapter 2 and 3 
that the soft landing condition 
was an essential condition
on the wave front set 
of a distribution 
which allows to control the wave front set of extensions
of distributions.
Before we state the
soft landing condition
in $T^\star U^n$, 
we first give the equation
of the conormal of $d_n\subset U^n$
in coordinates $(x_1,\dots,x_n;\xi_1,\dots,\xi_n)$.
The collection $dh_1=dx_2-dx_1,\dots,dh_{n-1}=dx_n-dx_1$ 
of 1-forms spans a basis
of orthogonal forms 
to the
tangent space of $d_n$,
thus a 1-form 
$\xi_1dx_1+\dots +\xi_ndx_n$ belongs
to the conormal if it writes
$\sum_{i=2}^n a_idh_i$ for some $(a_i)_i$ which 
implies $\xi_1=-\sum_{i=2}^n \xi_i$,
thus the equation of the conormal in $U^n$
is $x_1=x_2=\dots=x_n , \xi_1+\dots+\xi_n=0$. 
If we write
the equation of the soft landing condition in $T^\bullet U^n$
for the coordinates, we obtain 
\begin{equation}\label{slcmn}
\vert \sum_{i=1}^n\xi_i\vert\leqslant\delta \left(\sum_{i=2}^n\vert x_1-x_i \vert\right)\left( \sum_{i=2}^n\vert\xi_i\vert \right)
\end{equation}
since $k= \sum_{i=1}^n\xi_i$ and $\forall i\geqslant 2,\eta_{i}=\xi_{i+1}$ by \ref{slccite}, 
the inequality \ref{slcmn} is clearly invariant by the flow
$\lambda\mapsto (x_1,\lambda(x_2-x_1)+x_1,\dots;\xi_1+(1-\lambda^{-1})\sum_{j=2}^n\xi_j,\lambda^{-1}\xi_2,\dots,\lambda^{-1}\xi_n)$.

In configuration space $T^\star U^I$ with coordinates $(x_i;\xi_i)_{i\in I}$, 
the soft landing condition takes the following form: 
a conic set $\Gamma\subset T^\bullet U^I$ satisfies the \textbf{soft landing condition} w.r.t. to $d_I$ if
for all compact set $K\subset U^I$, there exists $\varepsilon>0$ and 
$\delta>0$, such that 
\begin{equation}\label{softlandingconfig}
\Gamma|_{K\cap\{\sum_{i\in I,i\neq j}\vert x_j-x_i \vert\leqslant\varepsilon \}} \subset\{\vert \sum_{i\in I} \xi_i\vert\leqslant \delta\left(\sum_{i\in I,i\neq j}\vert x_j-x_i \vert\right)\left( \sum_{i\in I;i\neq j}\vert\xi_i\vert \right)\}.
\end{equation} 

\subsection{We have $\left(WF\left(e^{\log\lambda\rho_{x_i}*}\Delta_+\right)\bigcap T^\bullet U^2\right)\subset (-E_q^+)\times E_q^+ $.} 

% We recall two useful tools discussed in the previous chapters.
%
%\paragraph{The Euler vector fields of Hadamard, Chapter 5.} 
%Let $U$ be a geodesically convex ball and $\Gamma(x_1,x_2)=\left\langle \exp_{x_1}^{-1}(x_2) , \exp_{x_1}^{-1}(x_2) \right\rangle_{g(x_1)}$ which is well defined on $U^2$. The Hadamard Euler vector fields $\rho^H_{x_i},i=1,2$ are defined as follows
%\begin{eqnarray}
%\rho^H_{x_1}=\frac{1}{2} \nabla_{x_2} \Gamma\\
%\rho^H_{x_2}=\frac{1}{2} \nabla_{x_1} \Gamma
%\end{eqnarray}
%
%\begin{prop}
%$\rho^H_{x_i},i=1,2$ are Euler vector fields w.r.t. to the diagonal $d_2$ in $M^2$.
%\end{prop}
%
%
%\paragraph{Conjuguating two Euler vector fields on cotangent space, Chapter $1,4$.}
%
%\begin{thm}\label{Heleintrickidentity}
%Let $(\rho_1,\rho_2)\in\mathfrak{g}^2$ be two Euler vector fields
%which scale w.r.t. the thin diagonal $d_n$.
%If $\Phi(\lambda)=e^{-\log\lambda \rho_1}\circ e^{\log\lambda \rho_2} $
%then $\Phi$ depends smoothly in $\lambda\in[0,1]$
%\end{thm}
%\begin{thm}
%If $\Phi(\lambda)$ satisfies the previous hypothesis, the cotangent lift $T^\star \Phi_\lambda$
%restricted to the conormal $(Td_n)^\perp|_{d_n}\subset T^\star M^n$ is the identity map.
%\end{thm}
The next lemma aims to use our cone $E^+_q\subset T^\bullet U$ to control
the wave front set of the family $\left(e^{\log\lambda\rho_{x_i}\star}\Delta_+\right)_{\lambda\in(0,1]}, i=(1,2)$. 
\begin{lemm}
We can choose $q$ and $U$ in such a way that  $$\forall\lambda\in(0,1], \left(WF\left(e^{\log\lambda\rho_{x_i}*}\Delta_+\right)\bigcap T^\bullet U^2\right)\subset \left(-E_q^+\right)\times E_q^+ .$$
\end{lemm}
\begin{proof}
By construction of $E^+_q$,
$\left(WF\left(\Delta_+\right)\bigcap T^\bullet U^2\right)\subset \left(-E_q^+\right)\times E_q^+$.
If $(x_1;\xi_1),(x_2;\xi_2)\in -E^+_q\times E^+_q$ 
then $\forall \lambda\in(0,1], (x_1;\xi_1+(1-\lambda)\xi_2),(\lambda^{-1}(x_2-x_1)+x_1;\lambda\xi_2)\in -E^+_q\times E^+_q$ by invariance and convexity 
of $E^+_q$ which immediately 
yields the result.
\end{proof}
\subsection{The scaling properties of translation invariant conic sets.} 
The next lemma we prove also 
has a geometric flavor. 
%Let $\Gamma_I\subset T^\bullet M^I$ be a \textbf{translation invariant} conic set
%with the property that for some index $i\in I$, 
%the set $\Gamma_{I}$ is
%stable by the action of $e^{\rho_{x_i}\log\lambda}$
%then we prove that $\Gamma_I$ is stable by the flow generated by \textbf{any} linear Euler $\rho_{x_j},j\in I$. 
\begin{lemm} 
Let $\Gamma_I\subset T^\bullet M^I$ be a \textbf{translation invariant} conic set.
Then $\Gamma_I$ is stable under $e^{\log\lambda\rho_i}$ for some
$i\in I$ is equivalent
to $\Gamma_I$ is stable by $e^{\log\lambda\rho_i}$ for all
$i\in I$.
\end{lemm}
\begin{proof}
%Let us start by noticing that the scaling relative to $a$ can be decomposed as a composition of translations and dilations:
%$$(x;\xi)\mapsto (x-a;\xi) \mapsto (\lambda(x-a);\lambda^{-1}\xi)\mapsto (\lambda(x-a)+a;\lambda^{-1}\xi) .$$
Following the approach
of Chapter $1$, we try
to find a flow $\Phi(\lambda)$
relating the two linear scalings by $\rho_{x_i}$ and $\rho_{x_j}$.
This flow is given 
by the formula 
$\Phi(\lambda)=e^{-\log\lambda\rho_{x_i}} \circ e^{\log\lambda\rho_{x_j}}$ 
and the lifted flow 
$T^\star\Phi(\lambda)$ 
on cotangent space
is given by the formula
$T^\star \Phi(\lambda)=T^\star e^{-\log\lambda\rho_{x_i}} 
\circ
T^\star e^{\log\lambda\rho_{x_j}}$.
In our specific case, for each $\lambda$,
$\Phi(\lambda)$ is 
a flow by linear translation.
The map $\Phi(\lambda)$ results
from the composition
of two \emph{scalings} relative to two elements $(x_i,x_j)$
with ratio $(\lambda,\lambda^{-1})$ 
respectively. It can be computed explicitely
$$\Phi_\lambda: x\mapsto  \lambda(x-x_i)+x_i$$
$$\mapsto\lambda^{-1}\left(\left(\lambda(x-x_i)+x_i\right)-\left(\lambda(x_j-x_i)+x_i \right) \right) + \left(\lambda(x_j-x_i)+x_i \right)$$ $$=(x-x_j)+ \left(\lambda(x_j-x_i)+x_i \right) =x+\underset{\text{translation vector}}{\underbrace{(\lambda-1)(x_j-x_i)}},$$
which proves $\Phi(\lambda)=e^{-\log\lambda\rho_{x_i}} \circ e^{\log\lambda\rho_{x_j}}$
is a translation of vector $(\lambda-1)(x_j-x_i)$. 
We also have  $T^\star \Phi(\lambda):(x;\xi)\mapsto (x+(\lambda-1)(x_j-x_i);\xi)$.
This computation proves the following fundamental fact: if a translation invariant set $\Gamma_{I}$ is stable by the cotangent lift of scaling relative to one given $a\in\mathbb{R}^d$ then $\Gamma_{I}$ is invariant by the cotangent lift of linear scalings relative to \emph{any element} $a\in\mathbb{R}^d$ which implies the claimed result.  
 \begin{figure} %on ouvre l'environnement figure
\begin{center}
\includegraphics[width=8cm]{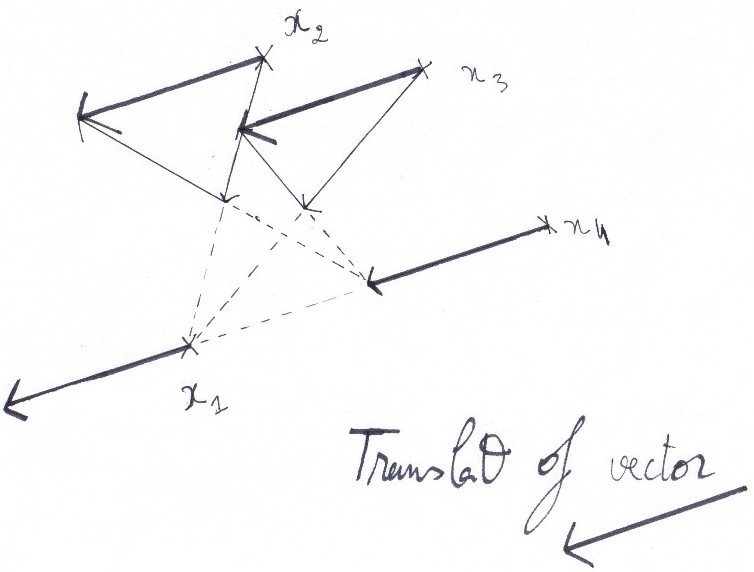} %ou image.png, .jpeg etc.
\caption{Action on configuration space $\left(\mathbb{R}^d\right)^4$ of the map $\Phi(\lambda)=e^{-\log\lambda\rho_{x_4}}\circ e^{\log\lambda\rho_{x_1}}$ for $\lambda=\frac{1}{2}$ as a translation.} %la légende
\label{Change_scaling} %l'étiquette pour faire référence à cette image
\end{center}
\end{figure} %on ferme l'environnement figure 
\end{proof} 
This lemma motivates the following definition: 
\emph{a translation invariant conic set $\Gamma_I\subset T^\bullet M^I$ is said to be 
scale invariant if it is stable by scaling w.r.t. the vector field $\rho_{x_i}$ for some 
$i\in I$}.
\subsection{Thickening sets.}
\begin{lemm}\label{thickeninglemm}
If $\Gamma_I$ satisfies the soft landing condition and is (strictly) polarized, then
there exists a translation and scale invariant $\tilde{\Gamma}_I$ such that $\Gamma_I\subset \tilde{\Gamma}_I$, $\tilde{\Gamma}_I$ is still (strictly) polarized and satisfies the soft landing condition.
\end{lemm}
We call \emph{good}, any conic set that is 
translation invariant, scale invariant, polarized 
and satisfies the 
soft landing condition. 
 
\begin{proof}
%We proved in Chapter $2$ (see the proof of Theorem \ref{gammaconorm}) 
%that if a conic set $\Gamma_I$ satisfies the soft landing condition w.r.t. the conormal bundle $(Td_I)^\perp$, then the enveloppe $\left(\Gamma_{I}\right)_M$ 
%which consists 
%of all curves $t\mapsto e^{t\rho^\star}(p)$ 
%which intersect $\Gamma_I$ also satisfies the soft landing condition. 
Notice that the formulation of 
the soft landing condition 
on configuration space by the equation
\begin{equation}
\vert \sum_{i\in I} \xi_i\vert\leqslant \delta\left(\sum_{i\in I,i\neq j}\vert x_j-x_i \vert\right)\left( \sum_{i\in I,i\neq j}\vert\xi_i\vert \right),
\end{equation} 
is clearly translation and scale
invariant. 
But $E^+_q$ and 
$\tilde{\leqslant}$ are also translation and scale invariant 
thus the concept of polarization is translation and scale invariant. 
So if a set $\Gamma_I\subset T^\bullet U^I$ 
is polarized and satisfies the soft landing condition, then the union
$\tilde{\Gamma}_I$
of 
all orbits of 
the group of translations and dilations which intersect 
$\Gamma_I$ 
satisfies the same properties and contains $\Gamma_I$.
\end{proof}
\subsection{The $\mu$local properties of the two point function.}
Let us consider the configuration space $U^2$ with coordinates $(x_1,x_2)$.
Let $\Xi$ be the wave front set of $\Delta_+$. 
In Chapter 5, 
we proved that 
$$\Xi\subset \left(\Lambda\bigcup \{(x,x;-\eta,\eta)| g_x(\eta,\eta)\geqslant 0\}\right) \bigcap \{(x_1,x_2;\eta_1,\eta_2)| (\eta_2)_0>0\}$$ where $\Lambda$
is the conormal bundle of the conoid $\Gamma=0$ (Theorem \ref{Wavefrontpullback})
and we proved that $\Delta_+$ is microlocally weakly
homogeneous of degree $-2$, 
$\Delta_+\in E_{-2}^\mu(U^2)$ (Theorem \ref{delta+bounded}).
%\paragraph{The soft landing condition is stable by variable family, chapter 4.}
%\begin{prop}
%For any smooth map $\Phi:[0,1]\times M \mapsto [0,1]\times  M$ such that for each $\lambda$, $\Phi_\lambda$ is a diffeomorphism of $M$ which fixes $I$, define $\sigma_\lambda=T^\bullet \Phi_\lambda$ the corresponding family of lifted symplectomorphisms, then $(\sigma_\lambda)_\lambda$ fixes the conormal $C\subset T^\bullet M $.
%\end{prop}
%\begin{prop}\label{HeleinmicroscopicI} 
%For any closed cone $\Gamma$ which satisfies the \textbf{soft landing condition} and $\Phi$ satisfying the previous hypothesis, the set $\bigcup_{\lambda\in[0,1]}\sigma_\lambda\circ \Gamma$ satisfies the \textbf{soft landing condition}.
%\end{prop} 
%\paragraph{The $\mu$local boundedness does not depend on choice of Euler, chapter $4$.}
%Let $t$ be a distribution in $\mathcal{D}^\prime(U^2\setminus d_2)$.
%\begin{thm}\label{variablediff}
%If $\exists V$ a neighborhood of $d_2$ and a closed conic set $\Gamma\subset T^\bullet(U^2\setminus d_2)$ which satisfies the soft landing condition such that $(\lambda^{-s}e^{\log\lambda\rho*}t)_{\lambda\in(0,1]}$ is bounded in $\mathcal{D}^\prime_{\Gamma}(U^2\setminus d_2)$, then for any Euler vector field $\tilde{\rho}$, there exists a closed conic set $\tilde{\Gamma}$ satisfying the soft landing condition and there exists a sufficiently small neighborhood $\tilde{V}$ of $I$ such that $\lambda^{-s}e^{\log\lambda\tilde{\rho}*}t|_{\tilde{V}}$
%is bounded in $\mathcal{D}^\prime_{\tilde{\Gamma}}(U^2\setminus d_2)$.
%\end{thm}
Here, we initialize the recursion 
for $t_2(\phi^m(x_1)\phi^m(x_2))=\varepsilon\circ T(\phi^m(x_1)\phi^m(x_2))$, 
and prove that 
$\lambda^{2m}t_{2,\lambda}$ 
is bounded in $\mathcal{D}^\prime_{\Gamma_2}(U^2\setminus d_2)$ 
where $\Gamma_2$ is a \emph{good} 
cone 
(recall \emph{good} means polarized, 
satisfies the soft landing condition, translation and scaling invariant).
We denote 
by $(\chi_I)_I$
the partition of unity
subordinate
to the cover
$\left(\tilde{M}_{I,I^c}\right)_I$ 
given by the improved 
geometrical lemma. 
\begin{thm} 
Let $t_2(\phi^m(x_1)\phi^m(x_2))
=\chi_{1}\Delta^m_+(x_1,x_2)
+\chi_{2}\Delta^m_+(x_2,x_1)$. 
Then $t_2\in E^\mu_{-2m}(U^2\setminus 
d_2)$ and
there exists a \emph{good cone} $\Gamma_2\subset T^\bullet U^2$ such that for each $\rho_{x_i},i=(1,2)$,  
%  polarized, $\Gamma_2$ satisfies the soft landing condition, $\Gamma_2$ is scaling, translation invariant and
the family
$\left(\lambda^{-2m} e^{\log\lambda\rho_{x_i}*}t_2\right)_{\lambda\in(0,1]} $
is bounded in $\mathcal{D}^\prime_{\Gamma_2}(U^2\setminus d_2)$. 
\end{thm} 
\begin{proof}
%Without loss of generality, 
%we assume in this proof that we scale w.r.t. $x_1$.
%By Theorem \ref{Wavefrontpullback} and by the 
%geometrical lemma:
%$$WF(t_2(\phi^m(x_1)\phi^m(x_2)))|_{U^2\setminus d_2}=WF(\Delta^m_+(x_1,x_2))|_{M_{\{1\}\{2\}}}
%\cup WF(\Delta^m_+(x_2,x_1))|_{M_{\{2\}\{1\}}}$$
%$$\subset \underset{\text{strictly polarized}}{\underbrace{(\Lambda\bigcap \{(x_1,x_2;\xi_1,\xi_2)| ((x_1)^0-(x_2)^0)(\xi_2)_0>0\})}},$$ 
%where $\Lambda$
%is the conormal bundle of the conoid $\Gamma=0$.
%Thus $$WF(t_2(\phi^m(x_1)\phi^m(x_2)))|_{U^2\setminus d_2}\subset \underset{\text{strictly polarized}}{\underbrace{\overline{\Lambda+\Lambda}\bigcap \{(x_1,x_2;\xi_1,\xi_2)| ((x_1)^0-(x_2)^0)(\xi_2)_0>0\}}}.$$
On the one hand $WF(\Delta^m_+)$ satisfies the soft landing condition
by Lemma \ref{Xislc} which implies
$WF(t_2)|_{U^2}$ also does. On the other hand, 
we already proved in
proposition (\ref{t2polar}) that $WF(t_2)$
is polarized then
applying Lemma \ref{thickeninglemm}, 
we find that
the enveloppe $\Gamma_2$ of $WF(t_2)$
is a good cone.
\end{proof}

\subsection{Pull-back of \emph{good} cones.}
Since we always pull-back distributions 
living on configuration spaces $U^I$ 
to higher configuration spaces $U^n$, 
we want the pull-back operation 
to preserve all the nice properties 
of the wave front set.
Let $p_{[n]\mapsto I}$
be the canonical projection 
$p_{[n]\mapsto I}:U^n\mapsto U^I$. 
\begin{lemm}
If $\Gamma_I\subset T^\bullet U^I$ 
is a \emph{good} cone
then $p_{[n]\mapsto I}^*\Gamma_I\subset T^\bullet U^n$ 
is also a \emph{good} cone. 
\end{lemm}
\begin{proof}
By definition 
$p_{[n]\mapsto I}^*\Gamma_I$
is polarized 
in $T^\bullet U^n$
since
the trace $Tr(x_i;\xi_i)_{i\in I}\subset T^\bullet U$
of an element
$(x_i;\xi_i)_{i\in I}\in \Gamma_I$
and of its pulled back
element 
$$((x_i;\xi_i),(x_j;0))_{i\in I,j\in I^c}\in p_{[n]\mapsto I}^*\Gamma_I$$
are the same.
$p_{[n]\mapsto I}^*\Gamma_I$ 
is also translation, 
scale invariant
by invariance of $\Gamma_I$ 
and the projection 
$p_{[n]\mapsto I}$. 
The only subtle point 
is to prove that $p_{[n]\mapsto I}^*\Gamma_I $ 
still satisfies 
the soft landing condition. 
Start from the assumption
that $\Gamma_I$ satisfies
the soft landing condition, 
then for all compact  
$K\subset U^I$, $\exists\varepsilon>0,\exists\delta>0$:
$$\Gamma|_{K\cap\{\sum_{i\in I,i\neq j}\vert x_j-x_i \vert\leqslant\varepsilon\}}\subset\{\vert \sum_{i\in I} \xi_i\vert\leqslant \delta\left(\sum_{i\in I,i\neq j}\vert x_j-x_i \vert\right)\left( \sum_{i\in I,j\neq i}\vert\xi_i\vert \right)\}$$  
then notice   
$$(x_i;\xi_i)_{i\in [n]} \in p_{[n]\mapsto I}^*\Gamma_I\implies (x_i;\xi_i)_{i\in I}\in \Gamma_I $$ 
$$\implies \vert\sum_{i=1}^n \xi_i\vert=\vert\sum_{i\in I} \xi_i\vert \leqslant \delta\left(\sum_{i\in I,i\neq j}\vert x_j-x_i \vert\right)\left( \sum_{i\in I,i\neq j}\vert\xi_i\vert \right)$$ 
$$\leqslant\delta\left(\sum_{i\in [n],i\neq j}\vert x_j-x_i \vert\right)\left( \sum_{i\in [n],i\neq j}\vert\xi_i\vert \right) $$  
which implies  $p_{[n]\mapsto I}^*\Gamma_I \subset\{\vert\sum_{i=1}^n \xi_i\vert\leqslant\delta\left(\sum_{i\in [n],i\neq j}\vert x_j-x_i \vert\right)\left( \sum_{i\in [n],i\neq j}\vert\xi_i\vert \right)  \}$ which is exactly the soft landing condition.
\end{proof}  
In the sequel, we denote by $\Gamma_I$ 
the set $p_{[n]\mapsto I}^*\Gamma_I$
making a slight notational abuse.

%\subsection{The distributional product makes sense and $WF(t_n)$ is polarized.}
\paragraph{The soft landing condition is stable by summation: }
We proved in Proposition \ref{sumstable} 
that
for $\Gamma_1,\Gamma_2$ two closed conic sets which both satisfy the
\textbf{soft landing condition}
and s.t. $\Gamma_1\cap -\Gamma_2=\emptyset$,
the cone $\Gamma_1\cup \Gamma_2\cup \left(\Gamma_1+\Gamma_2\right)$ satisfies the \textbf{soft landing condition}.

For all subsets $I\subset \{1,\dots,n\}$, let $\Lambda_I\subset T^\bullet M^I$ be the 
set of all elements in $T^\bullet U^I$ \textbf{polarized} by $E_q^+$.
Since the cone $E_q^+$, 
the partial order
relation $\tilde{\leqslant}$ 
and the trace operation 
are translation and dilation invariant, 
by Definition \ref{polarized}, 
the subset $\Lambda_I$ 
is also translation and dilation invariant. 
%Thus the set $\Lambda_n$ is invariant under the action of the group $\mathbb{R}^*\ltimes \mathbb{R}^d $ lifted to cotangent space:
%$$(\lambda,a)\in \mathbb{R}^*\ltimes \mathbb{R}^d : (x;\xi)\in \mathbb{R}^d\times\mathbb{R}^{d*}  \mapsto (\lambda x + a;\lambda^{-1}\xi)\in \mathbb{R}^d\times\mathbb{R}^{d*} .$$
For any manifold $M$, 
for any closed cone $\Gamma\subset T^\bullet M$ 
in the cotangent cone $T^\bullet M$, we denote
by $\Gamma^0=\Gamma\cup \underline{0}\subset T^\star M$ 
where $\underline{0}$ is 
the zero section of $T^\star M$. 
\subsection{The wave front set of the product $t_n$ is contained in a good cone $\Gamma_n$.}

\begin{thm}\label{goodthm2}
We assume the hypothesis of theorem (\ref{goodthm1}) 
is valid and keep the same notations.
If furthermore we assume all elements $\Gamma_I,I\varsubsetneq \{1,\dots,n\}$ are \emph{good} conic sets then $\Gamma_n$ is a \emph{good} conic set. 
\end{thm}
\begin{proof}
It is immediate since translation and scale invariance,
the 
polarization 
property
and the soft landing conditions 
are
stable by sums. 
\end{proof}

\subsection{We define the extension $\overline{t}_n$ and control $WF(\overline{t}_n)$.}
We saw in Chapter 4 that the product of distributions satisfying the H\"ormander condition was bounded:
let $\Gamma_1,\Gamma_2$ be two cones, assume $\Gamma_1\cap -\Gamma_2=\emptyset$.
Set $\Gamma=\Gamma_1\cup\Gamma_2\cup(\Gamma_1+\Gamma_2)$,
then the product 
$$(t_1,t_2)\in D_{\Gamma_1}^\prime\times D_{\Gamma_2}^\prime \mapsto t_1t_2\in D_{\Gamma}^\prime $$
is well defined and bounded (Theorem \ref{productbounded}). 
We also concluded Chapter 4 with a general extension theorem (\ref{thmfin2}):
if $t\in E_s^\mu(U^n\setminus d_n)$ then an extension $\overline{t}$ exists in 
$E_{s^\prime}^\mu(U^n)$ for all $s^\prime< s$.
Now we prove a theorem that gives conditions for which the extension $\overline{t}_n$ exists, has finite scaling degree and has \emph{good} wave front set.
\begin{thm}\label{lasttheorem}
Assume that the assumptions of Theorems 
(\ref{goodthm1}) and (\ref{goodthm2}) 
are satisfied and that
the family 
$\lambda^{-s_I}e^{\log\lambda\rho_{x_i}*}t_I$ is bounded in $\mathcal{D}^\prime_{\Gamma_I}$ 
for some $s_I$ where $\Gamma_I$
is good.
Then $t_n$ has a well defined extension $\overline{t}_n$ in $\mathcal{D}^\prime_{WF(t_n)\cup (Td_n)^\perp}(U^n)$ and 
there is a \emph{good} conic set $\Gamma_n$ such that for any $l\in\{1,\dots,n\}$, the family $\left(\lambda^{-s^\prime}e^{\log\lambda\rho_{x_l}*}\overline{t_n}\right)_\lambda,$
is bounded in $\mathcal{D}^\prime_{\Gamma_n\cup (Td_n)^\perp}(U^n)$ 
for
all $s^\prime<s_I+s_{I^c}+\sum_{(i,j)\in I\times I^c}2m_{ij}$.
\end{thm}
\begin{proof}
For any $l \in\{1,\dots,n\}$,
the family
$$\lambda^{-s_I}e^{\log\lambda\rho_{x_l}*}t_I $$ 
is bounded in $\mathcal{D}^\prime_{\Gamma_I}$
where $\Gamma_I$ is a \emph{good} cone.
Let us set 
\begin{equation}
\Gamma_n=\bigcup_{I} \left(\Gamma^0_I+\Gamma^0_{I^c}
+\Gamma^0_{ij}\right)|_{\tilde{M}_{I,I^c}}.
\end{equation}
Then the last step 
of the proof is a mere repetition of 
the proof 
of Theorems 
(\ref{goodthm1}) and (\ref{goodthm2}),
but instead of considering a "static" product 
$t_It_{I^c}\prod_{(i,j)\in I\times I^c} \Delta^{m_{ij}}_+(x_i,x_j)\chi_I$ 
on a given $\tilde{M}_{I,I^c}$, 
we will instead scale the whole product w.r.t. 
to some linear Euler vector field $\rho_{x_l}$: 
$$\underset{\text{bounded in }\mathcal{D}^\prime_{\Gamma_I}(U^n\setminus d_n)}{\underbrace{\left(\lambda^{-s_I}e^{\log\lambda\rho_{x_l}*}t_I\right)}}\underset{ \text{in }\mathcal{D}^\prime_{\Gamma_{I^c}}(U^n\setminus d_n) }{\underbrace{\left(\lambda^{-s_{I^c}}e^{\log\lambda\rho_{x_l}*}t_{I^c}\right)}}$$
$$\prod_{(i,j)\in I\times I^c}\underset{\text{in }\mathcal{D}^\prime_{\Gamma_{ij}}(U^n\setminus d_n) }{\underbrace{\left( \lambda^{-2m_{ij}}e^{\log\lambda\rho_{x_l}*}\Delta^{m_{ij}}_+(x_i,x_j)\right)}} 
\underset{\text{in }\mathcal{D}^\prime_\emptyset(U^n\setminus d_n)}{\underbrace{\chi_I}} . $$
Then we use the 
boundedness of the product
(Theorem \ref{productbounded})
to repeat the arguments of the proof of Theorem
\ref{goodthm1}
for bounded families of distributions.
Notice that it is very convenient for us 
that the functions $\chi_I$ constructed 
in the improved geometric 
lemma
are 
smooth scale invariant functions 
since they are going to be
bounded in $\mathcal{D}^\prime_\emptyset(U^n\setminus d_n)$.
The product 
$$\lambda^{-s_I-s_{I^c}-2\sum_{(ij)\in I\times I^c}m_{ij}} 
e^{\log\lambda\rho_{x_l}*}
\left(t_It_{I^c}
\prod_{(i,j)\in I\times I^c}
\Delta^{m_{ij}}_+(x_i,x_j)\right)_{\lambda\in(0,1]} $$
is well defined 
and bounded in $\mathcal{D}^\prime_{\Gamma_n}(U^n\setminus d_n)$
(by Theorem \ref{productbounded})
where $$\Gamma_n=\bigcup_{I} \left(\Gamma^0_{I}+\Gamma^0_{I^c}+
\Gamma^0_{ij}\right)\setminus\{\underline{0}\}|_{\tilde{M}_{I,I^c}} $$ 
is \emph{good}
by Theorem \ref{goodthm2}.
Then the distribution $$t_n=\left(t_It_{I^c}
\prod_{(i,j)\in I\times I^c}
\Delta^{m_{ij}}_+(x_i,x_j)\right)$$ is in $E^\mu_{s_n}(U^n\setminus d_n)$
(since $\Gamma_n$ satisfies the soft landing condition and
the family
of distributions
$\left(\lambda^{-s_n}t_\lambda\right)_{\lambda\in(0,1]}$
is
bounded in $\mathcal{D}^\prime_{\Gamma_n}$) for 
$s_n=s_I+s_{I^c}+2\sum_{(ij)\in I\times I^c}m_{ij}$. 
We can conclude by the extension theorem
(\ref{thmfin2}), 
which provides an extension $\overline{t}_n$ in $E^\mu_{s^\prime}(U^n)$
for all $s^\prime< s_I+s_{I^c}+2\sum_{(ij)\in I\times I^c}m_{ij}$ 
with the constraint
$WF(\overline{t}_n)\subset WF(t_n)\bigcup (Td_n)^\perp$ on the wave front set of the extension. The wave front set $WF(t_n)$ is polarized and so is the conormal $(Td_n)^\perp$ hence the union $WF(t_n)\bigcup (Td_n)^\perp$ is also
polarized.
And the family
$\left(\lambda^{-s^\prime}\overline{t}_n\right)_{\lambda\in(0,1]}$ 
should be bounded in $\mathcal{D}^\prime_{\Gamma_n\bigcup (Td_n)^\perp}(U^n)$
where $\Gamma_n\bigcup (Td_n)^\perp$ is a \emph{good} conic set.
\end{proof}
The last theorem allows to conclude the recursion 
since we were able to initialize 
the recursion at the step $n=2$: 
$WF(t_2)$ is contained in a \emph{good cone} $\Gamma_2$ and  $\lambda^{2m}e^{\rho\log\lambda*}t_2(\phi^m\phi^m)$ is always bounded in $\mathcal{D}^\prime_{\Gamma_2}(U^2\setminus d_2)$, 
however beware that $t_2(\phi^m\phi^m)$
is in $E_{s^\prime}(U^2)$ for all $s^\prime<2m$, 
hence repeated applications of theorem (\ref{lasttheorem}) 
allows to define all extensions 
$\overline{t}_n\in \mathcal{D}^\prime(U^n)$ for all $n$.

\chapter{A conjecture by Bennequin.}
\section{Parametrizing the wave front set
of the extended distributions.}
In this 
short chapter,
we solve 
a conjecture of 
Daniel
Bennequin
stating
that the 
wave front set
of the extensions 
$\overline{t_n}$
are
singular Lagrangian
manifolds.

Lagrangians often appears in
quantum mechanics as the 
geometrical object living
in cotangent space which represents 
the semiclassical limit of 
quantum states (\cite{BatesWeinstein} p.~16, 35, 60-63 and \cite{EvansZworski} p.~103). 
Our theorem
might help 
us to give a similar geometric
interpretation 
of the wave front set of $n$-point functions
in quantum field theory:
each element
of the Lagrangian 
could represents
the ``trajectory 
of a
process"
in cotangent space.
For instance:
\begin{enumerate}
\item an element of the wave front set
of $t_2(\phi(x)\phi(y))$ represents
a null geodesic lifted
to the cotangent space,
\item an element 
of the wave front set
of $t_3(\phi(x_1)\phi(x_2)\phi^3(y)\phi(x_3))$
represents the interaction
of three null geodesics
intersecting at one point.
\end{enumerate}

The proof also
clarifies the fact 
that the wave front set
of these extensions 
can be
parametrized
by
objects (generalizing
the graph of a gradient)
called Morse families
which were introduced
by Weinstein
and H\"ormander.
\section{Morse families and Lagrangians.}
Let us start by recalling some simple definitions.
We introduce the concept 
(due to Weinstein see \cite{BatesWeinstein} Definition 4.17) 
of a Morse family
(with some modifications of our own):
\begin{defi}
A Morse family is a triple $\mathcal{S}=(\pi:B\mapsto M,S)$ 
satisfying
the following conditions:
\begin{itemize}
\item a) $\left(\pi:B\mapsto M\right)$ is
such that any connected component of $B$ 
is of the form 
$\left(\mathbb{R}^k\setminus \{0\}\right)\times \Omega$ 
for some $k$ and some set $\Omega\subset M$, 
this endows $B$
with the structure of a smooth cone and the restriction
of $\pi$ to this connected component 
is the canonical projection,
\item b) $S\in C^\infty(B)$ is 
homogeneous of degree $1$ 
w.r.t. vertical scaling, 
\item c) $dS\neq 0$. 
\end{itemize}
%the intersection $\mathbf{Gr} dS \cap \mathcal{E}^\perp $ is clean in $T^*B$ and $\lambda_{\mathcal{S}}\Sigma_{\mathcal{S}}$ is a Lagrange immersion of the smooth manifold $\Sigma_{\mathcal{S}}$ (which is the \textbf{critical set}) in $T^\star M$.  
\end{defi}
%for $\pi:B\mapsto M$ a fibration, introduce the exact sequence:
%$$0\rightarrow \mathcal{E} \rightarrow TB \rightarrow_{T\pi} TM\rightarrow 0 $$
%where $T\pi: TB \rightarrow TM$ and $\mathcal{E}=\ker T\pi$.
% Dually, we have the contravariant exact sequence
%$$ 0\leftarrow \mathcal{E}^* \leftarrow T^*B \leftarrow \mathcal{E}^\perp \leftarrow 0 $$ 
%where $\mathcal{E}^\perp$ is coisotropic in $T^*B$ and fibers over $T^*M$.
%Geometrically, $\mathcal{E}^\perp$ is a fiber bundle over $B$ where the fiber over a point $p\in B$ is the conormal of the fiber of $\pi:B\mapsto M$ intersecting $p$.  
Daniel Bennequin 
pointed out to us 
that this definition 
is actually very general since $B$ 
is \emph{not necessarily connected} 
thus we could have several connected components of $B$ 
living over some given point in $M$, like branches of a cover.
The second nice point of the definition of Alan Weinstein is that the map $\pi$ is not necessarily surjective.
%Our definition is slightly different from the definition of Alan Weinstein in that we require that the intersection of $dS$ with $\mathcal{E}^\perp$ to be \emph{clean} and not necessarily transverse, 
%we also demand that $\lambda_{\mathcal{S}}$ to be an immersion and not an embedding. 
%This is why a clean generating function $S$ can parametrize singular Lagrangians.  
Denote by $x$ the coordinates in $M$ and by $(x;\theta)$ the coordinates in $B$ where $\theta$ is the vertical variable. Denote by $\Sigma_S=\{\frac{\partial S}{\partial \theta}=0\}\subset B$ the critical set of $S$. The smooth projection $\pi$ defines a set $\pi(\Sigma_S)$ which
is the projection of the critical set.
\begin{defi}
We denote by $\mathcal{T}\pi(\Sigma_S)$ the \textbf{tangent cone} of $\pi(\Sigma_S)$ which is
defined as follows, for $x\in \pi(\Sigma_S)$, 
$$\mathcal{T}_x\pi(\Sigma_S)=\{d\pi|_{(x,\theta)}(X)|\exists \gamma\in C^1([0,1],\Sigma_S) \text{ s.t. }  \gamma(0)=(x,\theta),\dot{\gamma}(0)=X \},$$
then $\mathcal{T}\pi(\Sigma_S)=\bigcup_{x\in \pi(\Sigma_S)}\mathcal{T}_x\pi(\Sigma_S)$.
\end{defi}
\begin{ex}
For
$\mathcal{S}=\left(\mathbb{R}_{>0}\times (U^2\setminus d_2)\mapsto (U^2\setminus d_2), \theta \Gamma(x,y) \right)$, the set $\Sigma_S$ is equal to  
$\left(\{\Gamma=0\}\cap \left(U^2\setminus d_2\right)\right)\times \mathbb{R}_{>0}$ where $\{\Gamma=0\}$ is the null 
conoid in $U^2\setminus d_2$ i.e. 
the subset of pairs of points connected
by a null geodesic. Thus $\pi\left(\Sigma_S\right)=\{\Gamma=0\}|_{U^2\setminus d_2}$
is an open submanifold
and $\mathcal{T}\pi(\Sigma_S)$ is just the tangent space to the submanifold $\pi\left(\Sigma_S\right)=\{\Gamma=0\}\cap U^2\setminus d_2$.
\end{ex}
It is possible
to define a notion 
of tangent cone 
for very general sets
but we will not need such theory
here.
\begin{defi}
We denote by $N_{\pi(\Sigma_S)}$ 
the \textbf{normal} to $\pi(\Sigma_S)$ 
which is defined as the subset 
$\{(x,\xi)\in T^\star M|x\in \pi(\Sigma_S), 
\xi\left(\mathcal{T}_x\pi(\Sigma_S)\right)\geqslant 0 \}
\subset T^\star M$.  
\end{defi}
Throughout this section, for any cone $C$ in a vector space $E$, we denote by $C^\circ$ the cone
in dual space $E^\star$ defined as $\{\xi| \xi(C)\geqslant 0 \}$
(it is sometimes called the \emph{polar} of $C$).
This definition can be extended to cones in tangent space 
and we denote 
by $\mathcal{T}\pi(\Sigma_S)^\circ$ 
the subset 
$\underset{x\in \pi(\Sigma_S)}{\bigcup} \left(\mathcal{T}_x\pi(\Sigma_S)\right)^\circ$
living in $T^\bullet M$.
Geometrically, $N_{\pi(\Sigma_S)}$ is the \textbf{dual cone} $\mathcal{T}\pi(\Sigma_S)^\circ$ of the tangent cone $\mathcal{T}\pi(\Sigma_S)$. 
If $\pi$ is a smooth embedding, 
$N_\pi$ is just the conormal bundle of $\pi(\Sigma)$. 
\begin{defi}
We denote by $\lambda_{\mathcal{S}}$ the map $\lambda_{\mathcal{S}}:(x;\theta)\in B \mapsto (x;d_x S)(x,\theta)\in T^\star M $.
\end{defi}
In nice situations, $\lambda_S\left(\Sigma_S\right)$ is a smooth Lagrange immersion and coincides with $N_{\pi(\Sigma_S)}$. However in our general situation, we always have the following upperbound:
\begin{prop}
$\lambda_S\left(\Sigma_S\right)\subset N_{\pi(\Sigma_S)}$.
\end{prop}
\begin{proof}
Any vector field in $X\in Vect(B)$ 
decomposes uniquely as a sum 
$X=X_h+X_v=f^\mu\partial_{x^\mu}+f^i\partial_{\theta^i}$ 
where $X_h$ is the horizontal part 
and $X_v$ the vertical part
since $B$ is a 
trivial cone. 
Thus it suffices to 
prove that 
if $d\left(\partial_{\theta^i}S \right)(X)|_{\Sigma_S}=0$ then 
$dS(X_h)|_{\Sigma_S}=0$ because $dS(X_h)_{x,\theta}= d_xS(d\pi_{x,\theta}(X))$.
The key observations
are:
\begin{itemize}
\item a) $\frac{\partial S}{\partial \theta^i}=0\implies \theta^i\frac{\partial S}{\partial \theta^i}=S=0 $,
since $S$
is homogeneous of degree $1$ in $\theta$, thus
$\Sigma_S\subset \{S=0\}$ and
$d\left(\partial_{\theta^i}S \right)(X)|_{\Sigma_S}=0\implies dS(X)|_{\Sigma_S}=0 $,
\item b) for all vertical vector field $X_v$, $dS(X_v)|_{\Sigma_S}=0$.
\end{itemize}
From these observations, 
we deduce that:
$$d\left(\frac{\partial S}{\partial \theta^i} \right)(X)|_{\Sigma_S}=0\implies dS(X)|_{\Sigma_S}=0 \implies  dS(X)|_{\Sigma_S}$$ $$= dS(X_h)|_{\Sigma_S}+ \underset{=0}{\underbrace{dS(X_v)|_{\Sigma_S}}=0}
\implies dS(X_h)|_{\Sigma_S}=0.$$
%If $d\pi|_{(x,\theta)}(X)\in \mathcal{T}_x\pi(\Sigma_S)$, then by definition $(X\frac{\partial S}{\partial \theta^i} )_i=0 \implies \sum_i \theta^iX\frac{\partial S}{\partial \theta^i}=0 $, thus
%$$0=\sum_i \theta^iX\frac{\partial S}{\partial \theta^i}=\sum_i \theta^if^\mu\frac{\partial^2 S}{\partial x^\mu\partial \theta^i}+ \theta^i f^j\frac{\partial^2 S}{\partial \theta^j\partial \theta^i} =\sum_i \theta^if^\mu\frac{\partial^2 S}{\partial x^\mu\partial \theta^i} +
%\underset{=0}{\underbrace{f^j\frac{\partial S}{\partial \theta^j} }}$$
%since $\theta^i\frac{\partial S}{\partial \theta^i}=S$ because $S$
%homogeneous of degree $1$,
%$$\quad =f^\mu \frac{\partial S}{\partial x^\mu}=d_xS|_{(x;\theta)}\left(d\pi|_{(x,\theta)}(X)\right)$$ since $(x;\theta)\in \Sigma_S$. This means $d_xS|_{(x;\theta)}\left(\mathcal{T}_x\pi(\Sigma_S)\right)=0$ ie $d_xS|_{(x;\theta)}\in N_{\pi(\Sigma_S),x}$.
\end{proof}
We want to prove that $\lambda_S\left(\Sigma_S\right)$ is isotropic 
in the sense 
that
the tangent cone of $\lambda_S\left(\Sigma_S\right)$
is symplectic orthogonal to
itself.
We denote by $\mathcal{T}_p\left(\lambda_S\left(\Sigma_S\right)\right)$ the subset defined as
$$\{d\lambda_S|_{(x,\theta)}(X)|\exists \gamma\in C^1([0,1],\Sigma_S) \text{ s.t. }  \gamma(0)=(x,\theta),\dot{\gamma}(0)=X\},$$
and $\mathcal{T}\left(\lambda_S\left(\Sigma_S\right)\right)=\bigcup_{p\in \lambda_S\left(\Sigma_S\right)}\mathcal{T}_p\left(\lambda_S\left(\Sigma_S\right)\right)$.
Let $\omega$ be the natural symplectic form in $T^\star M$:
\begin{prop}
$\omega|_{\lambda_{\mathcal{S}}\left(\Sigma_{\mathcal{S}}\right)}=0$.
\end{prop}
\begin{proof}
We actually prove 
that 
$\lambda_{\mathcal{S}}^\star\omega|_{\Sigma_{\mathcal{S}}}=0$
which implies
$\omega|_{\lambda_{\mathcal{S}}\left(\Sigma_{\mathcal{S}}\right)}=0$.
Let us denote by 
$\alpha=\xi_idx^i\in \Omega^1(T^\star M)$
the Liouville $1$-form
which is the primitive of $\omega$ 
i.e.
$d\alpha=\omega$. 
We decompose
uniquely
the differential
$d$ acting on $\Omega^\bullet (B)$
as a sum $d=d_x+d_\theta$. 
The key observation
is that
$d_\theta S|_{\Sigma_{\mathcal{S}}}=0$.
$$\lambda_{\mathcal{S}}^\star\omega|_{\Sigma_{\mathcal{S}}}= \lambda_{\mathcal{S}}^\star d\alpha|_{\Sigma_{\mathcal{S}}}=d \left(\lambda_{\mathcal{S}}^\star\alpha\right)|_{\Sigma_{\mathcal{S}}}=d \left(\lambda_{\mathcal{S}}^\star\xi_idx^i\right)|_{\Sigma_{\mathcal{S}}}
$$
$$= d\left(\frac{\partial S}{\partial x^i}dx^i\right)|_{\Sigma_{\mathcal{S}}}=d(d_xS)|_{\Sigma_{\mathcal{S}}}=d(d_xS+d_\theta S)|_{\Sigma_{\mathcal{S}}} $$
since $d_\theta S|_{\Sigma_{\mathcal{S} } }=0 $
$$=d^2 S|_{\Sigma_{\mathcal{S}}}=0.$$
%
%In local coordinates $(x^\mu;p_\mu)$ in $T^\star M$ the symplectic form $\omega$ reads
%$\omega=dx^\mu\wedge dp_\mu$.
%Let $(X_1,X_2)$ be such that $\left(d\lambda_S|_{(x,\theta)}(X_1),d\lambda_S|_{(x,\theta)}(X_2)\right)\in \mathcal{T}_p\left(\lambda_S\left(\Sigma_S\right)\right)^2$. In local coordinates $X_1=X_1^\mu\frac{\partial}{\partial x^\mu}+X_1^i\frac{\partial}{\partial \theta^i} $ thus
%$d\lambda_S|_{(x,\theta)}(X_1)= X_1^\mu\frac{\partial^2 S}{\partial x^\mu\partial x^\nu}\frac{\partial}{\partial p_\nu}+ X_1^\mu\frac{\partial}{\partial x^\mu}$ and
%$$\omega\left(d\lambda_S|_{(x,\theta)}(X_1),d\lambda_S|_{(x,\theta)}(X_2)\right)=X_1^\mu X_2^\nu\underset{=0}{\underbrace{\left(\frac{\partial^2 S}{\partial x^\mu\partial x^\nu}-\frac{\partial^2 S}{\partial x^\nu\partial x^\mu}\right)}}=0.$$
\end{proof}
This means that $\lambda_S\left(\Sigma_S\right)$ is \textbf{isotropic}.
At each point $x\in \pi(\Sigma_S)$ where
$\lambda_S\left(\Sigma_S\right)|_x=N_{\pi_S(\Sigma_S),x}$ we will say that 
$\lambda_S\left(\Sigma_S\right)$ is \textbf{Lagrangian} at $x$ because 
it is 
\textbf{isotropic} 
of maximal dimension. If it is \textbf{Lagrangian} at every $x\in \pi(\Sigma_S)$ 
(or on an open dense subset of $\pi(\Sigma_S)$)
then we call it \textbf{Lagrangian}, in nice situations this coincides with the usual
notion of Lagrange immersion (see \cite{Hormander} vol 3 p.~291,292 and \cite{BatesWeinstein}).
We will later consider Morse families $\mathcal{S}$ with the supplementary
requirements that
$\Sigma_S\subset B$ is a finite union of smooth submanifolds 
and
$\lambda_S\left(\Sigma_S\right)$ is \textbf{Lagrangian}.

We work out a fundamental example of Morse family 
which generates the conormal bundle of a submanifold.
\begin{ex}
Let $I\subset M$ be a submanifold. We shall work in local chart where the manifold is given by a system of $d$ equations $f_1=\dots=f_d=0$. Then the Morse triple $((\mathbb{R}^d\setminus\{0\})\times M\mapsto M, \sum_{i=1}^d\theta^if_i)$ parametrizes the conormal bundle $(TI)^\perp$. 
Indeed, $\Sigma_S=\{f_i=0\}\times (\mathbb{R}^d\setminus\{0\})=I\times (\mathbb{R}^d\setminus\{0\})$ and $\lambda_S\left(\Sigma_S\right)=\{\theta_i df_i|_{\Sigma_S},\theta\in \mathbb{R}^d\setminus\{0\}\}$.
The key observation is that
any element in the conormal
of $I$ should decompose in
the basis of 1-forms $(df_i)_i$
thus $\lambda_S\left(\Sigma_S\right)$ parametrizes
the conormal of $I$.
\end{ex} 
\paragraph{An analytic interpretation of $\lambda_S\left(\Sigma_S\right)$.}
We interpret $\lambda_S\left(\Sigma_S\right)$ in terms of the wave front set of an oscillatory integral $t$.
We can understand it as a parametrization of $WF(t)$ by the Morse family $\mathcal{S}$. 
\begin{prop}\label{Feynmanlag} 
Let $\mathcal{S}=(\pi:M\times \mathbb{R}^k\mapsto M,S)$ be a Morse family over the manifold $M$ and $(x;\theta)$ where $\theta\in\mathbb{R}^k$ a system of coordinates in $M\times \mathbb{R}^k$, 
for any asymptotic symbol $a$ (\cite{RS} vol 2 p.~99):
$$WF\left(\int_{\mathbb{R}^k} d\theta a(\cdot;\theta) e^{iS(\cdot,\theta)} \right)\subset \lambda_{\mathcal{S}}\Sigma_{S}.$$ 
\end{prop} 
\begin{proof}
In local coordinates $(x,\theta)$ for $B$, it is just a consequence of Theorem 9.47, p.~102 in \cite{RS}.
\end{proof}  
\paragraph{Functorial behaviour of Morse families.}
In microlocal geometry, we need the following fundamental operations on distributions 
\begin{itemize} 
\item the pull-back $t\mapsto f^\star t$ by a smooth map $f:M\rightarrow N$ which is not always well defined
for distributions
\item the exterior tensor product $(t_1,t_2)\mapsto t_1\boxtimes t_2$ which is always well defined
\item for our purpose, it will be important to add the product of distributions when it is well defined.
\end{itemize}
Assume that the wave front sets 
of given distributions $t$ 
are parametrized 
by Morse families, 
we already know how the wave front sets 
transform under 
these functorial operations on distributions, 
the question is 
whether we can find a new Morse family 
to parametrize the wave front set
of the distribution
obtained by one
of the previous 
operations.
The functorial behaviour of Lagrangians under geometric transformations is already studied in \cite{Geomasympt} Chapter 4, however it is not described in terms of generating functions and our point of view is more explicit and more oriented towards applications.
\subsubsection{Formal operations on Morse families.}
First 
introduce 
operations 
on cones
as follows.
Let $B\mapsto M$ be a smooth cone,
for any smooth map $f:N\mapsto M$, 
$f^\star B\mapsto f^\star M$ 
is a smooth cone 
(Appendix 2 of \cite{Geomasympt})   
with fibers defined as follows $ f^\star B|_x=B|_{f(x)}$.
We also introduce a suitable generalization
of the fiber product for cones, 
recall the fiber product of $\pi_1:B_1\mapsto M$ and $\pi_2:B_2\mapsto M$ denoted by $B_1\times_M B_2$ is defined by $\{(p_1,p_2)\in B_1\times B_2|\pi_1(p_1)=\pi_2(p_2) \}$.
\begin{defi}
Let $B_1,B_2$ be two smooth cones over a given base manifold $M$. Then we define the product
$B_1\overline{\times}_MB_2$ as the cone $$\left(\left(B_1\cup \underline{0}_{1}\right)\times_M \left( B_2\cup \underline{0}_{2}\right)\right)\setminus \left(\underline{0}_{1}\times_M\underline{0}_{2}\right)
=(B_1\times_M \underline{0}_{2})\cup(\underline{0}_{1}\times_M B_2)\cup \left(B_1\times_M B_2\right).$$
\end{defi}
The key point 
of this product is that
we add the zero section so
that our trivial cones 
become trivial vector bundles
we compute
the fiber product
and remove the zero section
at the end.
\paragraph{The QFT case.}
In our recursion, we only need to pull-back by smooth projections. 
For instance, by the canonical projection maps $M^n\mapsto M^I$ for $I\subset [n]$.
In this case, if we still denote by $f$ the submersion $f:N\mapsto M$, the Morse family 
can be chosen extremely simple 
\begin{defi}\label{pull-backMorse}
Let $\mathcal{S}=(\pi:B\mapsto M,S)$ be a Morse family over the manifold $M$,
for any smooth \textbf{projection} $f:N\mapsto M$,
we define the pulled back Morse family as the triple
\begin{equation}
f^\star\mathcal{S}=(f^\star\pi:f^\star B\mapsto f^\star M,f^\star S).
\end{equation}
\end{defi}
It is obvious that
$df^\star S\neq 0$ since $dS\neq 0$ and $df$ is \textbf{surjective}.
When $f$ is a smooth map, we prove that the pull-back by $f$ of  $\lambda_{\mathcal{S}}\Sigma_{\mathcal{S}}$ is parametrized by the Morse family $f^\star\mathcal{S}$: 
\begin{prop}\label{pull-back}
Let $f:=N\mapsto M$ be a smooth projection and $\mathcal{S}=(\pi:B\mapsto M,S)$ a Morse family over the manifold $M$. Then:
\begin{equation}
f^\star\lambda_{\mathcal{S}}\Sigma_{\mathcal{S}}
=\lambda_{f^\star\mathcal{S}}\Sigma_{f^\star\mathcal{S}}.
\end{equation}
\end{prop}
\begin{proof}
We denote by $(y;\eta)$ the coordinates in $T^\star N$ and 
$(x;\xi)$ the coordinates in $T^\star M$.
We have 
$$f^*\left(\lambda_{S}\Sigma_{S}\right)= \{(y;\eta\circ df) | (f(y);\eta) \in \lambda_{\mathcal{S}}\Sigma_{\mathcal{S}} \}$$
by the definition of pull-back in \cite{Hormander} and \cite{Geomasympt} 
$$ =\{(y;d_xS_{(f(y);\theta)}\circ df) |\, d_\theta S(f(y);\theta)=0 \} $$
$$=\{(y;d\left(S\circ f\right)_{(y;\theta)} |\, d_\theta \left(S\circ f\right)(y;\theta)=0 \}$$ 
$$=\lambda_{f^*\mathcal{S}}\Sigma_{f^*\mathcal{S}}  $$
by definition of $\lambda_{f^*\mathcal{S}}\Sigma_{f^*\mathcal{S}}$. 
\end{proof}
\begin{prop}
Under the assumptions of proposition (\ref{pull-back}), 
if $\lambda_S\left(\Sigma_S\right)$ is Lagrangian then
$\lambda_{f^\star S}\Sigma_{f^\star S}$ is Lagrangian.
\end{prop}
\begin{proof}
$$\lambda_{f^\star\mathcal{S}}\Sigma_{f^\star\mathcal{S}}
=f^\star\lambda_{\mathcal{S}}\Sigma_{\mathcal{S}}\text{ by the above proposition}$$ 
$$\quad =f^\star N_{\pi(\Sigma_S)}\text{ because $\lambda_{\mathcal{S}}\Sigma_{\mathcal{S}}$ Lagrangian} $$
$$\quad=N_{\pi(\Sigma_S)}\circ df \text{ by definition of the pull-back} $$
$$\quad=\mathcal{T}\pi(\Sigma_{f^\star S})^\circ  \circ df\text{ by definition of $N_{\pi(\Sigma_S)}$} $$
$$\quad = \mathcal{T}\pi_{f^\star S}(\Sigma_{f^\star S})^\circ \text{ since $\mathcal{T}\pi_S(\Sigma_S)=Df\mathcal{T}\pi_{f^\star S}(\Sigma_{f^\star S})$}$$
$$\quad = N_{\pi(\Sigma_{f^\star S})} \text{ by definition of } N_{\pi(\Sigma_{f^\star S})}  .$$
Finally, $\lambda_{f^\star\mathcal{S}}\Sigma_{f^\star\mathcal{S}}
=N_{\pi(\Sigma_{f^\star S})}$ means, by definition, that $\lambda_{f^\star\mathcal{S}}\Sigma_{f^\star\mathcal{S}}$
is Lagrangian.
\end{proof}
\begin{prop}
Under the assumptions of proposition (\ref{pull-back}), 
if $\Sigma_S$ is a smooth submanifold (resp finite union of smooth submanifolds) in $B$
then
$\Sigma_{f^\star S}$ is also a smooth submanifold (resp finite union of smooth submanifolds)
in $f^\star B$.
\end{prop}
\begin{proof}
This is immediate since $d_{y;\theta}(d_\theta(S\circ f))$ has the same rank as
$d_{x,\theta}S$.
\end{proof}
%For any smooth proper submersion $f:M\mapsto N$,
%we define the "pushforward Morse family" as the triple
%\begin{equation}
%f_\star\mathcal{S}=(\pi\circ f: B\mapsto N , S)
%\end{equation}

%Let $\mathcal{S}_{i}=(\pi_i:B_i\mapsto M_i,S_i),i=(1,2)$ be a pair Morse family over the manifolds $M_1,M_2$, then we define the exterior product of the Morse family $(\mathcal{S}_{1},\mathcal{S}_{2})$ as the triple
%\begin{equation}
%\mathcal{S}_{1}\boxtimes\mathcal{S}_{2}=(\pi_1\overline{\times}\pi_2: B_1 \overline{\times} B_2\mapsto M_1\times M_2 , S_1+S_2)
%\end{equation}
%and this is always clean Morse family.
Let $\mathcal{S}_{i}=(\pi_i:B_i\mapsto M,S_i),i=(1,2)$ be a pair of Morse families over the manifold $M$, then we define the ``sum of the Morse families'' $\mathcal{S}_{1}+\mathcal{S}_{2}$ as the triple
\begin{equation}
\mathcal{S}_{1} + \mathcal{S}_{2}=(\pi_1\overline{\times}_M\pi_2: B_1 \overline{\times}_M B_2\mapsto M , S_1+S_2).
\end{equation}
We put quotation marks ``'' to stress the fact that this operation still defines a triple (cone, base manifold, function) but this triple is not necessarily a Morse family since we do not know if $d(S_1+S_2)\neq 0$, we will see that a necessary and sufficient condition for $\mathcal{S}_{1}+\mathcal{S}_{2}$ to be a Morse family is that
$\lambda_{S_1}\Sigma_{S_1}\cap -\lambda_{S_2}\Sigma_{S_2}=\emptyset$ which is the H\"ormander condition.
%We call such triple a formal Morse family. 
\paragraph{Remark on sums of Morse families.}
Notice by definition that if the cone $B_{i},i=(1,2)$ corresponding to the Morse family $\mathcal{S}_{i}$ has $n_i$ connected components, then $B_1 \overline{\times}_M B_2$ has
$(n_1+1)(n_2+1)-1$ connected components. An immediate recursion yields that
the cone corresponding to the sum $\mathcal{S}_{1} +\dots+ \mathcal{S}_{k}$
has $\left((n_1+1)\dots (n_k+1)\right)-1$ connected components.
\subsubsection{Transversality lemmas.}
We recall the classical notion of transversality 
in differential geometry in our context (see \cite{Lee} Definition 2.48 p.~80).
Let $\Sigma_i,i=(1,2)$ be a pair of smooth manifolds and $\pi_i:\Sigma_i\mapsto M,i=(1,2)$ be a pair of smooth maps. In such case for every $x\in \pi_i(\Sigma_i)$, the tangent cones $\mathcal{T}_x\pi_i(\Sigma_i),i=(1,2)$ are \textbf{vector
subspaces} of $T_xM$ 
(a vector subspace has less structure 
than a cone).
\begin{defi}
$\pi_1$ and $\pi_2$ are called \textbf{transverse} if for all $x\in \pi_1(\Sigma_1)\cap\pi_2(\Sigma_2)$, $\mathcal{T}_x\pi_1(\Sigma_1)+\mathcal{T}_x\pi_2(\Sigma_2)=T_xM$.
\end{defi}
\begin{lemm}\label{transv1}
Let $\Sigma_i,i=(1,2)$ be a pair of smooth submanifolds in $B_i$ and $\pi_i:B_i\mapsto M,i=(1,2)$ be a pair of smooth maps.
If $\pi_1$ and $\pi_2$ are \textbf{transverse} then $\Sigma_1\times_M\Sigma_2 $ is a smooth submanifold in $B_1\times_M B_2$.
\end{lemm}
Lemma \ref{transv1} 
obviously generalizes
to the case $\Sigma_i$ is 
a finite union of submanifolds,
in which case $\Sigma_1\times_M \Sigma_2$
is a finite union
of submanifolds.

\begin{proof}
Denote by $\Delta$ the diagonal in $M\times M$. Then $B_1\times_MB_2$ can be identified with the inverse image $(\pi_1\times\pi_2)^{-1} (\Delta)=B_1\times_\Delta B_2\subset B_1\times B_2 $ which is always a submanifold
of $B_1\times B_2$ and the fiber product $\Sigma_1\times_M\Sigma_2$ is just the intersection
$(\Sigma_1\times\Sigma_2)\bigcap (B_1\times_\Delta B_2)$ in $B_1\times B_2$. So we view both $\Sigma_1\times\Sigma_2$ and $B_1\times_\Delta B_2$ as submanifolds sitting inside $B_1\times B_2$,  
a sufficient condition
for $(\Sigma_1\times\Sigma_2)\bigcap (B_1\times_\Delta B_2)$ to be a submanifold of 
$B_1\times_\Delta B_2$ is that the intersection is transverse (it is a classical result of transversality theory that the transversal intersection of two submanifolds is a submanifold of the two initial submanifolds, it is a particular case of Theorem 2.47 in \cite{Lee} for an embedding also 
see Theorem 3.3 p.~22 in \cite{Hirsch}). It is immediate to check
that at every point $(p_1,p_2)$ of the intersection $(\Sigma_1\times\Sigma_2)\bigcap (B_1\times_\Delta B_2)$, 
$T_{p_1,p_2}(\Sigma_1\times\Sigma_2)+T_{p_1,p_2}(B_1\times_\Delta B_2)=T_{p_1,p_2}(B_1\times B_2)$ since $D(\pi_1\times \pi_2)(\Sigma_1\times \Sigma_2)=T_x\Delta$ by transversality of $\pi_1(\Sigma_1),\pi_2(\Sigma_2)$ and $T_{p_1,p_2}(B_1\times_{\Delta}B_2)$ spans the vertical tangent space of the bundle $B_1\times B_2$.
\end{proof}
For each smooth map $\pi:\Sigma_{S}\mapsto M$, we recall the definition of 
the normal to $\pi(\Sigma)$: $N_{\pi(\Sigma)}\subset T M$ as the subset $\bigcup_{x\in\pi(\Sigma)} \mathcal{T}_x\pi(\Sigma)^\circ$ in $T^\star M$ which is the dual cone in cotangent space of the tangent cone $\mathcal{T}\pi(\Sigma)$. 
We set $N^\bullet_{\pi(\Sigma)}=N_{\pi(\Sigma)}\cap T^\bullet M$.
\begin{lemm}\label{transv2}
Assume $\Sigma_i,i=(1,2)$ are smooth manifolds and
$\pi_i:\Sigma_i\mapsto M$ are smooth maps, 
then $\pi_1,\pi_2$ are \textbf{transverse} 
if and only if 
$N^\bullet_{\pi_1}\cap -N^\bullet_{\pi_2}=\emptyset$.
\end{lemm}
Lemma \ref{transv2} 
obviously generalizes
to the case $\Sigma_i$ is 
a finite union of submanifolds,
in which case 
every submanifold in $\Sigma_1$ 
shall be transverse to any submanifold
of $\Sigma_2$.

\begin{proof}
To prove the lemma, we just work infinitesimally. We fix a pair $(p_1,p_2)\in \Sigma_1\times\Sigma_2$ such that $\pi_1(p_1)=\pi_2(p_2)=x$.  $\pi_1$ and $\pi_2$ are transverse at $x\in M$ implies by definition
that $\mathcal{T}_x\pi_1(\Sigma_1)+\mathcal{T}_x\pi_2(\Sigma_2)=T_xM$. Then 
by a classical result in the duality theory of cones,
$$\{0\}=\overline{T_xM}^\circ
=\overline{\mathcal{T}_x\pi_1(\Sigma_1)+\mathcal{T}_x\pi_2(\Sigma_2)}^\circ$$
$$=\mathcal{T}_x\pi_1(\Sigma_1)^\circ\cap\mathcal{T}_x\pi_2(\Sigma_2)^\circ=N_{\pi_1}\cap -N_{\pi_2}.$$
\end{proof} 
We illustrate the last lemma in the figure (\ref{Intersection}) for the case 
of two curves
intersecting 
transversally
in the plane
and we 
represent
the corresponding
spaces $N_{\pi_i}$.
\begin{figure} %on ouvre l'environnement figure
\begin{center}
\includegraphics[width=14cm]{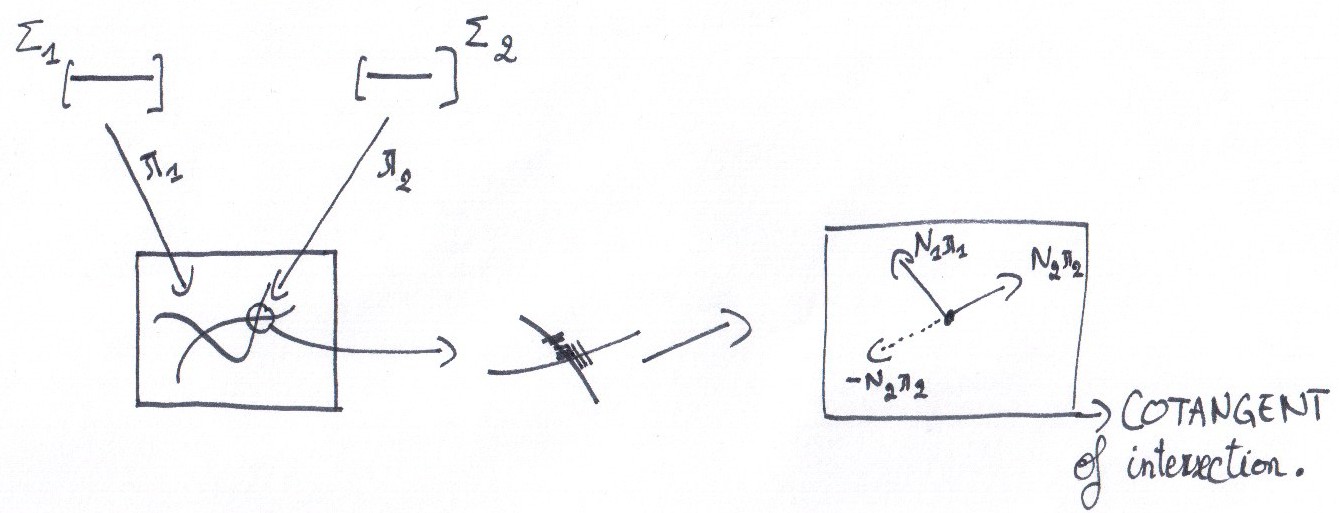} %ou image.png, .jpeg etc.
\caption{Transverse intersection of curves and their conormals.} %la légende
\label{Intersection} %l'étiquette pour faire référence à cette image
\end{center}
\end{figure} %on ferme l'environnement figure  
The meaning of this lemma 
is that
the condition
$N^\bullet_{\pi_1}\cap -N^\bullet_{\pi_2}=\emptyset$
of H\"ormander generalizes
the 
classical differential 
geometric 
transversality when
$\Sigma_i$
are not necessarily smooth submanifolds
in $B_i$.

\begin{prop}\label{suminv}
Let $\mathcal{S}_{i}=(\pi_i:B_i\mapsto M,S_i),i=(1,2)$ be 
a pair of Morse families over the manifold $M$.
If $\lambda_{\mathcal{S}_1}\left(\Sigma_{\mathcal{S}_1}\right)\bigcap (-\lambda_{\mathcal{S}_2}\left(\Sigma_{\mathcal{S}_2})\right)=\emptyset$, then $\left(\lambda_{\mathcal{S}_1}\left(\Sigma_{\mathcal{S}_1}\right)
+\lambda_{\mathcal{S}_2}\left(\Sigma_{\mathcal{S}_2}
\right)\right)\cup \lambda_{\mathcal{S}_1}\left(\Sigma_{\mathcal{S}_1}\right)\cup
\lambda_{\mathcal{S}_2}\left(\Sigma_{\mathcal{S}_2}\right)$ 
is parametrized by the Morse family $\mathcal{S}_{1} + \mathcal{S}_{2}=(\pi_1\overline{\times}_M\pi_2: B_1 \overline{\times}_M B_2\mapsto M, S_1+S_2)$.
\end{prop} 
\begin{proof}
It is sufficient to find the Morse family parametrizing $\lambda_{\mathcal{S}_1}\left(\Sigma_{\mathcal{S}_1}\right)+\lambda_{\mathcal{S}_2}
\left(\Sigma_{\mathcal{S}_2}\right)$.
We will make some local computation in coordinates where we assume w.l.o.g. that
$B_i$ is
equal to the cartesian product $M\times \Theta_i$ with coordinates $(x,\theta_i)$ where $\Theta_i$ is a vector space 
with the origin removed.
Let us consider the Morse family $(\pi_1\times_M\pi_2: B_1 \times_M B_2\mapsto M , S_1+S_2)$, where we use the local coordinates $(x;\theta_1,\theta_2)$ for $B_1 \times_M B_2$. 
Then the critical set of this Morse family is by definition
$\{d_{\theta_1,\theta_2}(S_1+S_2)=0 \} =\{d_{\theta_1}S_1=0 \}\cap \{d_{\theta_2}S_2=0 \}=\Sigma_{\mathcal{S}_1}\times_M \Sigma_{\mathcal{S}_2}\subset B_1\times_M B_2$,
and the image of this subset by $\lambda_{S_1+S_2}$ is given by
$$\lambda_{S_1+S_2}\left(\Sigma_{\mathcal{S}_1}\times_M \Sigma_{\mathcal{S}_2} \right)=\{\left(x;d_{x}\left(S_1+S_2\right)\right)(x;\theta)| d_{\theta_1}S_1=0, d_{\theta_2}S_2=0 \}$$ $$=\{\left(x;d_{x}S_1+d_{x}S_2\right)| 
(x;\theta_1,\theta_2)\in \Sigma_{\mathcal{S}_1}\times_M \Sigma_{\mathcal{S}_2}\}
=\lambda_{\mathcal{S}_1}\Sigma_{\mathcal{S}_1}+
\lambda_{\mathcal{S}_2}\left(\Sigma_{\mathcal{S}_2}\right),$$
which proves $(\pi_1\times_M\pi_2: B_1 \times_M B_2\mapsto M , S_1+S_2)$ parametrizes $\lambda_{\mathcal{S}_1}\Sigma_{\mathcal{S}_1}+\lambda_{\mathcal{S}_2}\left(\Sigma_{\mathcal{S}_2}\right)$,
thus if we add all other components, $\lambda_{\mathcal{S}_1}\Sigma_{\mathcal{S}_1}+\lambda_{\mathcal{S}_2}\left(\Sigma_{\mathcal{S}_2}\right)\cup \lambda_{\mathcal{S}_1}\Sigma_{\mathcal{S}_1}\cup\lambda_{\mathcal{S}_2}\left(\Sigma_{\mathcal{S}_2}\right)$ is parametrized by the family $\mathcal{S}_{1} + \mathcal{S}_{2}=(\pi_1\overline{\times}_M\pi_2: B_1 \overline{\times}_M B_2\mapsto M, S_1+S_2)$.
 
 It remains to prove that $d(S_1+S_2)\neq 0$ in $B_1\times_M B_2$.
If both $d_{\theta_1}S_1(x;\theta_1)=0$ and $d_{\theta_2}S_2(x;\theta_2)=0$ then necessarily $d_x(S_1+S_2)(x;\theta_1,\theta_2)\neq 0$ since $\lambda_{\mathcal{S}_1}\left(\Sigma_{\mathcal{S}_1}\right)\bigcap -\lambda_{\mathcal{S}_2}\left(\Sigma_{\mathcal{S}_2}\right)=\emptyset$.
\end{proof}
%\begin{coro}
%Under the assumptions of the proposition (\ref{suminv}),
%for all distributions $(t_i)_{i=1,2}$ such that $WF(t_i)\subset \lambda_{\mathcal{S}_i}\Sigma_{\mathcal{S}_i},i=(1,2)$, the product
%$t_1t_2$ is a well defined distribution and $WF(t_1t_2)\subset \lambda_{S_1+S_2}\Sigma_{S_1+S_2}$.
%\end{coro}
For the moment our results and statements are for general Morse families
and we did not assume $\lambda_S\left(\Sigma_S\right)$ was Lagrangian (recall Lagrangian means $\lambda_S\left(\Sigma_S\right)=N_{\pi(\Sigma_S)}$ for us) nor that the critical set 
$\Sigma_S$ was a finite union of submanifolds. 
\begin{prop}\label{sumlag}
Under the assumptions of Proposition \ref{suminv}, 
if $\left(\lambda_{S_i}\left(\Sigma_{S_i}\right)\right)_{i=(1,2)}$ are 
\textbf{Lagrangians}
then
$\lambda_{\mathcal{S}_1+\mathcal{S}_2}\left(\Sigma_{\mathcal{S}_1+\mathcal{S}_2}\right)$
is Lagrangian.
\end{prop}
\begin{proof}
One can check from the definitions that
$\mathcal{T}((\pi_1\times_M\pi_2)(\Sigma_1\times_M\Sigma_2))=\mathcal{T}(\pi_1\Sigma_1)\cap \mathcal{T}(\pi_2\Sigma_2)$. Hence by linear algebra,
$$N_{(\pi_1\times_M\pi_2)(\Sigma_1\times_M\Sigma_2)}= \mathcal{T}((\pi_1\times_M\pi_2)(\Sigma_1\times_M\Sigma_2))^\circ=(\mathcal{T}(\pi_1\Sigma_1)\cap \mathcal{T}(\pi_2\Sigma_2))^\circ$$ 
$$=\overline{(\mathcal{T}(\pi_1\Sigma_1))^\circ + (\mathcal{T}(\pi_2\Sigma_2))^\circ}
=\overline{N_{\pi_1(\Sigma_1)} + N_{\pi_2(\Sigma_2)}}
=\overline{\lambda_{\mathcal{S}_1}\left(\Sigma_{\mathcal{S}_1}\right)
+\lambda_{\mathcal{S}_2}\left(\Sigma_{\mathcal{S}_2}\right)},$$ 
finally $N_{(\pi_1\times_M\pi_2)(\Sigma_1\times_M\Sigma_2)}=
\overline{\lambda_{\mathcal{S}_1}\left(\Sigma_{\mathcal{S}_1}\right)
+\lambda_{\mathcal{S}_2}\left(\Sigma_{\mathcal{S}_2}\right)}$
means that 
$$N_{(\pi_1\times_M\pi_2)(\Sigma_1\overline{\times_M}\Sigma_2)}$$
$$=N_{(\pi_1\times_M\pi_2)(\Sigma_1\times_M\Sigma_2)}\cup N_{(\pi_1\times_M\pi_2)(\Sigma_1\times_M\underline{0}_2)}\cup N_{(\pi_1\times_M\pi_2)(\underline{0}_1\times_M\Sigma_2)}
$$
$$=\overline{\lambda_{\mathcal{S}_1}\left(\Sigma_{\mathcal{S}_1}\right)
+\lambda_{\mathcal{S}_2}\left(\Sigma_{\mathcal{S}_2}\right)}
\cup 
\lambda_{\mathcal{S}_1}\left(\Sigma_{\mathcal{S}_1}\right)
\cup
\lambda_{\mathcal{S}_2}\left(\Sigma_{\mathcal{S}_2}\right) $$
$$ 
=\lambda_{\mathcal{S}_1}\left(\Sigma_{\mathcal{S}_1}\right)
+\lambda_{\mathcal{S}_2}\left(\Sigma_{\mathcal{S}_2}\right)
\cup 
\lambda_{\mathcal{S}_1}\left(\Sigma_{\mathcal{S}_1}\right)
\cup
\lambda_{\mathcal{S}_2}\left(\Sigma_{\mathcal{S}_2}\right)$$
$$
=\lambda_{\mathcal{S}_1+\mathcal{S}_2}\left(\Sigma_{\mathcal{S}_1+\mathcal{S}_2}\right),
$$ 
which by definition 
means 
$\lambda_{\mathcal{S}_1+\mathcal{S}_2}\left(\Sigma_{\mathcal{S}_1+\mathcal{S}_2}\right)$
is Lagrangian.  
\end{proof}
\begin{prop}
If under the assumptions of Proposition (\ref{sumlag}), 
each $\Sigma_{S_i}$
is a finite union 
of smooth submanifolds in $B_i$ 
then $\Sigma_{\mathcal{S}_1}\times_M \Sigma_{\mathcal{S}_2}$ 
is a finite union of smooth submanifolds 
of $B_1\times_M B_2$.
\end{prop}
\begin{proof}
It suffices to recognize that the assumption 
$\lambda_{\mathcal{S}_1}\left(\Sigma_{\mathcal{S}_1}\right)
\bigcap -\lambda_{\mathcal{S}_2}\left(\Sigma_{\mathcal{S}_2}\right)=\emptyset$ 
is equivalent to 
$N^\bullet_{\pi_1(\Sigma_{S_1})}
\cap -N^\bullet_{\pi_2(\Sigma_{S_2})}=\emptyset$ 
(by our definition of being Lagrangian) 
which implies the transversality of the two maps 
$\pi_1:\Sigma_{\mathcal{S}_1}\mapsto M$, 
$\pi_2:\Sigma_{\mathcal{S}_2}\mapsto M$ 
by lemma (\ref{transv2}),
which means by application of lemma (\ref{transv1}) that the fiber product 
$\Sigma_{\mathcal{S}_1}\times_M \Sigma_{\mathcal{S}_2}$ 
is a finite union of smooth submanifolds of $B_1\times_M B_2$.
\end{proof}

 To summarize all the results we proved if $t_1$ and $t_2$ are distributions whith wave front set
$WF(t_i)$ parametrized by the Morse family $\mathcal{S}_i$ and
$(\lambda_{S_i}\left(\Sigma_{S_i}\right))_{i=(1,2)}$ satisfy the H\"ormander condition $\lambda_{S_1}\left(\Sigma_{S_1}\right)\cap -\lambda_{S_2}\left(\Sigma_{S_2}\right)=\emptyset$ then the distributional product $t_1t_2$ makes sense
and has wave front set contained in the set $\lambda_{\mathcal{S}_1+\mathcal{S}_2}\left(\Sigma_{\mathcal{S}_1+\mathcal{S}_2}\right)$ parametrized by the Morse family $\mathcal{S}_1+\mathcal{S}_2$. 
Furthermore, we proved 
that 
if $(\lambda_{S_i}\Sigma_{S_i})_{i=(1,2)}$ are Lagrangians
and $(\Sigma_{S_i})_{i=(1,2)}$ 
are finite union 
of smooth submanifolds then
the same properties hold 
for the Morse family $\mathcal{S}_1+\mathcal{S}_2$.
If $f:N\mapsto M$ is a smooth submersion and $t\in \mathcal{D}^\prime(M)$ whith wave front set
$WF(t)$ parametrized by the Morse family $\mathcal{S}$ then the pull-back $f^\star t$ makes sense
and has wave front set contained in the set $\lambda_{f^\star\mathcal{S}}\Sigma_{f^\star\mathcal{S}}$ parametrized by the Morse family $f^\star\mathcal{S}$. Furthermore, we proved 
that if $\lambda_{S}\Sigma_{S}$ is Lagrangian
and $\Sigma_{S}$ is a finite union of 
smooth submanifolds then
the same properties hold 
for the Morse family $f^\star\mathcal{S}$.
\begin{thm}
Let $\overline{t}_n$ be the distributions defined by the recursion theorem.
Then $WF(\overline{t}_n)$ is parametrized by a Morse family and is 
a union of smooth Lagrangian manifolds.
\end{thm}
\begin{proof}
We use the notation and formalism of the section 3 in Chapter 5. To inject this condition in our recursion theorem, it will be sufficient to check that $WF(\Delta_+)|_{C_i},i\in\{1,2\}$ or equivalently $WF t_2(\phi(x)\phi(y))|_{U^2\setminus d_2}$ and all conormal bundles 
$(Td_I)^\perp$ are parametrized by Morse families. 
For $t_2(\phi(x)\phi(y))$, by Theorem \ref{Wavefrontpullback} of Chapter $5$ 
and causality, 
$$WF(t_2(\phi(x)\phi(y)))=WF(\Delta_+(x,y))|_{x\geqslant y}\cup WF(\Delta_+(y,x))|_{y\geqslant x}  $$
$$=\text{conormal }\{\Gamma=0\}\cap \{(x,y;\xi,\eta) | (x^0-y^0)\eta^0>0   \} .$$
Thus we can write the Morse family in a local chart $U^2\setminus d_2$: $$\mathcal{S}=\left(\mathbb{R}_{>0}\times (U^2\setminus d_2)\mapsto (U^2\setminus d_2), \theta \Gamma(x,y) \right)$$ 
and the fact that it parametrizes $WF(t_2)$ results 
from the fact that:
$$\{(x,y;\theta d_x\Gamma,\theta d_y\Gamma | \Gamma(x,y)=0, \theta>0 \}=\text{conormal }\{\Gamma=0\}\cap \{(x,y;\xi,\eta) | (x^0-y^0)\eta^0>0   \}.$$ 
Furthermore the critical set $$\Sigma_S=\{(x,y)\in U^2\setminus d_2| \Gamma(x,y)=0 \}$$ is a \textbf{smooth submanifold} and $\lambda_S\left(\Sigma_S\right)\subset T^\bullet(U^2\setminus d_2)$ is \textbf{Lagrangian}.
Also for the conormal of the diagonals, it was already treated in our examples, they can always be generated by Morse families. Then we inject 
these hypotheses in the recursion and we easily get the result.
\end{proof}
\begin{ex}
In order to illustrate
the mechanism
at work, we choose to study
the example
of the wave front set of
the product 
$$\delta_{x^1=0}\delta_{x^2=0}\delta_{x^3=0}(x^1,x^2,x^3)$$
of three delta functions
$\delta_{x^i=0},i=(1,2,3)$ in $\mathbb{R}^3$.
Each $\delta_{x^i=0}$ is 
supported on 
the hyperplane $x^i=0$. 
One should have in mind 
the boundary
of a cube in a small 
neighborhood of
one vertex !
Each $\delta_{x^i=0}$ 
has wave front set
equal to the conormal bundle 
of the corresponding face 
$x^i=0$
of a cube, 
which is parametrized
by the Morse 
family
$$\mathcal{S}_i=\left((\theta_i;x)\in\left(\mathbb{R}\setminus\{0\}\right)\times\mathbb{R}^3
\mapsto x\in\mathbb{R}^3, S_i(x,\theta_i)= x^i\theta_i \right).$$
We represented in the figure
some vectors $\nabla_xS_i$ 
standing for
the momentum component of the
conormal of the face $x^i=0$.
When two faces 
$F_i,F_j$ are adjacent to an edge
$F_i\cap F_j$, 
the convex sum of 
the wave front sets
supported over the edge
is the conormal of the edge 
(represented in the figure as a \emph{tube})
which is parametrized
by the Morse family
$$\left((\theta_i,\theta_j;x)\in\left(\mathbb{R}\setminus\{0\}\right)^{2}\times\mathbb{R}^3\mapsto x\in\mathbb{R}^3,\left(S_i+S_j\right)(x,\theta_i,\theta_j)=x^i\theta_i+x^j\theta_j\right).$$
Finally the origin
is a vertex adjacent to all faces
and the wave front set
over $(0,0,0)$ is parametrised by
$$\left((\theta_1,\theta_2,\theta_3;x)
\in\left(\mathbb{R}\setminus\{0\}\right)^{3}\times\mathbb{R}^3
\mapsto x\in\mathbb{R}^3,
\left(S_1+S_2+S_3\right)=x^1\theta_1+x^2\theta_2+x^3\theta_3\right),$$
and represents the conormal
at the origin 
(represented in 
the figure as 
the sphere).
In total,
the Wavefronset
has seven smooth components
indexed by the strata of the 
cube boundary: ($3$ faces, $3$ edges, $1$ vertex).
The reader can check
that the wave front set 
of $\delta_{x_1=0}\delta_{x_2=0}\delta_{x_3=0}(x_1,x_2,x_3)$
is parametrized
by the Morse family
$\mathcal{S}_1+\mathcal{S}_2+\mathcal{S}_3$ 
(all seven cases are covered since by definition the sum of Morse families ``contains zero sections'')
which is equal to
$$\{\pi:\left(\mathbb{R}^3\setminus \{0,0,0\}\right)\times\mathbb{R}^3\mapsto \mathbb{R}^3, S(x;\theta)=x^1\theta_1+x^2\theta_2+x^3\theta_3 \}.$$
The morality of this example
is that
the conormal
of a union
of manifolds
\textbf{is not the union
of the conormals}!
One should take into 
account the
informations contained
in the ``strata'' and our
formalism does it for the
most elementary
example.
\end{ex}
\begin{figure} %on ouvre l'environnement figure
\begin{center}
\includegraphics[width=16cm,height=12cm]{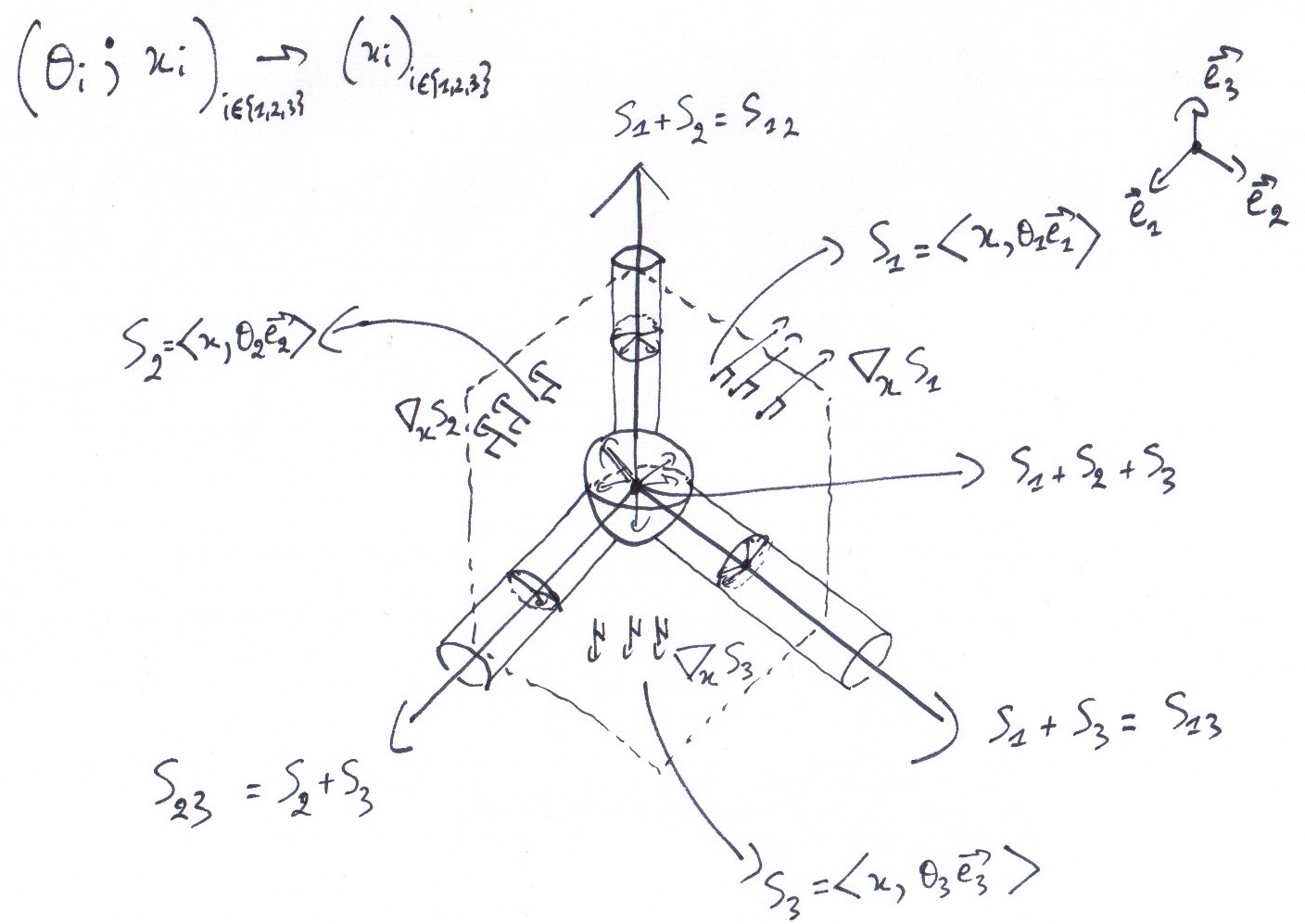} %ou image.png, .jpeg etc.
\caption{The wave front set of $\delta_{x_1=0}\delta_{x_2=0}\delta_{x_3=0}$ as a union of $7$ Lagrange immersions.} %la légende
\label{Cube} %l'étiquette pour faire référence à cette image
\end{center}
\end{figure} %on ferme l'environnement figure  
\section{A conjectural formula.}

We conjecture a formula
which should give
an upper bound
of the
wave front set of
any Feynman amplitude
corresponding 
to a Feynman 
diagram $\Gamma$.

Let $\Gamma$ be a 
graph with $n$ vertices 
which are indexed
by $[n]$.
Let $E(\Gamma)$
denote the set of
edges of $\Gamma$,
to each element
$e\in E(\Gamma)$
corresponds a
unique injective map
$e:\{1,2\}\mapsto [n]$
s.t.
the edge $e$
connects 
the vertices 
$e(1)$ and $e(2)$.
To $\Gamma$,
we associate
the Morse
family
\begin{eqnarray}
\left(\pi:\left(\mathbb{R}_{\geqslant 0}^{E(\gamma)}\times\left(\mathbb{R}^d\right)^{\frac{n(n-1)}{2}}\times U^n\right)\setminus \left(\underline{0}\cup \{dS=0\}\right)\mapsto U^n, \, 
S \right)\\
\pi:(\tau_e)_e,(\theta_{ij})_{ij},(x_1,\cdots,x_n)\mapsto (x_1,\cdots,x_n) \\
S=\sum_{e\in E(\Gamma)} \tau_e\Gamma(x_{e(1)},x_{e(2)})+\sum_{1\leqslant i<j\leqslant n}\theta_{ij}.(x_i-x_j).
\end{eqnarray}
We conjecture
that this Morse family
parametrizes
the wave front set
of the 
Feynman amplitude
corresponding
to $\Gamma$.\\
We also
conjecture that 
the wave front set
of all
$n$-point functions
$t_n$ are contained
in the set parametrized
by the Morse family:
\begin{eqnarray}
\left(\pi:\left(\mathbb{R}^{\frac{n(n-1)}{2}}_{\geqslant 0}\times\left(\mathbb{R}^d\right)^{\frac{n(n-1)}{2}}\times U^n\right)\setminus \left(\underline{0}\cup \{dS=0\}\right)\mapsto U^n, \, 
S \right)\\
\pi:(\tau_{ij}),(\theta_{ij})_{ij},(x_1,\cdots,x_n)\mapsto (x_1,\cdots,x_n) \\
S=\sum_{1\leqslant i<j\leqslant n} \tau_{ij}\Gamma(x_{i},x_{j})+\theta_{ij}.(x_i-x_j).
\end{eqnarray}

\chapter{Anomalies and residues.}
\section{Introduction.}
\paragraph{The plan of the chapter.}

 First, we will generalize the notion of weak homogeneity of Yves Meyer \cite{Meyer} to the setting of currents, then show how the results of Chapter $1$ naturally transfer to this new setting.
However, we need to discuss the notion of Taylor expansion for test forms to give a suitable 
meaning to the notions of Taylor polynomial and Taylor remainder of a test form.
We spend some 
time to discuss the notion of currents supported on a submanifold $I$ and their \emph{representation} in the current theoretic setting. 
Following physics terminology, we will call local counterterms the currents supported on $I$:
actually in the causal approach to QFT, 
all ambiguities of the renormalization schemes 
can be described by \textbf{local counterterms}, 
more precisely the difference between two renormalizations
is a current supported on $I$.

 One natural example of ambiguity originates from the work of Yves Meyer \cite{Meyer}. We call $R$ the composite operation of restriction of a distribution defined on $M$ to $M\setminus I$ followed by any extension operation. We explain why this operation differs from the identity because of the non-uniqueness of the extension procedure.
We describe explicitely the ambiguity of this operation $R$ by giving an explicit formula for $T-RT$ and we show that this difference is a \textbf{local counterterm}. We give an interpretation of this ambiguity in terms of the notion of ``generalized moment'' for currents.

 Then we will describe the dependance of the regularization operator $R$ defined in Chapter $1$, that might be called the \emph{Hadamard regularization operator}, on the choice of bump function $\chi$ (which is equal to $1$ in a neighborhood of $I$) and the choice of Euler vector field $\rho$. 
Without surprise, we will prove that a change in the function $\chi$ or the vector field $\rho$ will result in a change of $R$ by a \textbf{local counterterm}, these are explicit ambiguities.
 In QFT, a fundamental
question is to ask if the symmetries or the exactness of currents can be preserved by the renormalization scheme. However since all continuous symmetries of QFT can be encoded by Lie algebras of vector fields it is natural to wonder if the Lie derivatives commute with the renormalization $R$. 
The symmetry is not always preserved and the quantity which measures this defect will be called \textbf{residue} of $T$. 
In the following, $\mathfrak{Res}$ is defined by generalizing Griffiths--Harris's definition
(\cite{Griffiths} p.~368)  
by the chain homotopy equation  
\begin{equation} 
dR T - Rd T=\mathfrak{Res}[T] 
\end{equation}
and is a \textbf{local counterterm}. 
However $\mathfrak{Res}$ is a special type of counterterm since $\mathfrak{Res}$
is always closed in $\mathcal{D}^\prime(M)$ and is exact when $T$ is closed.
We show that the 
regularization techniques 
of Meyer allows us to extend the notion of residues in the sense of Griffiths--Harris 
(see the section $3$ in \cite{Griffiths}) 
and our resulting 
definition has nothing to do with complex analysis.  
The residue in \cite{Griffiths} is only well defined for \emph{functions} $T\in L_{loc}^q\left(\mathbb{R}^n \right)$ (\cite{Griffiths} p.~369) smooth outside a given singular set $S$, whereas our notion of residue works for distributions in $E_s$ which are weakly homogeneous of degree $s$ for arbitrary $s$. Somehow, our regularity hypothesis on the current $T$ which guarantees the existence of residues is minimal because any current defined \textbf{globally} on $M$ will live in some scale space $E_s$ for some $s$.
%But we want to emphasize the fact that our residue formalism is \textbf{necessary} in QFT since the Feynman amplitude $(\Delta_+(x,y))^{10000} $ (where $\Delta_+$ is the Schwartz kernel of a Hadamard state) does not belong to any function space, for instance it does not belong to any Sobolev space $W^{s,p}$ with positive $s,1\leqslant p\leqslant \infty$. It is a genuine distribution.
%Representing $Res[T]$ as a current supported on $I$,
%we will also prove the residue map sends exact currents to exact currents and thus induces a map from the cohomology of currents in $M\setminus I$ to the cohomology of currents in $I$.   
The residue theory provides a very flexible and general framework to study anomalies.  
We repeat the construction of geometric residues for infinite dimensional Lie algebras of symmetries, for $X$ a vector field which commutes with $\rho$, we study the residue equation
$$L_XRT-RL_XT=\mathfrak{Res}_X[T] $$ 
and we interpret $\mathfrak{Res}_X[T]$ as an obstruction
to the fact that quantization 
(in our sense quantization consists in an operation of extension of distributions)
preserves classical symmetries.
More precisely, if we assume
that
we have 
an infinite dimensional Lie algebra of vector fields
 $\mathfrak{g}$, and that
$\forall X\in \mathfrak{g},L_XT=0$ 
($\mathfrak{g}$ is the Lie algebra of classical symmetries) 
then
$X\mapsto \mathfrak{Res}_X[T]$ is a coboundary
for the infinite dimensional Lie algebra of vector fields.
It can be thought in terms
of a quantum version of the Noether theorem.
$$\begin{array}{|c|c|}
\hline
\text{Physics terminology} &\text{Our interpretation} \\
\hline
\text{renormalization scheme}& \text{Extension operator }R:\mathcal{D}^\prime(M\setminus I)\mapsto \mathcal{D}^\prime(M)\\
\hline
\text{local counterterm} & \text{currents supported on }I\\
\hline
\text{ambiguity} & R_1T-R_2T\\
\hline 
\text{Symmetry} &\text{Lie algebra of vector fields }\mathfrak{g} \\
\hline
\text{anomaly} & \text{residue }L_XR-RL_X \\
\hline
\end{array}$$
% This is quite unusual in the litterature and comparable works involving residues are not common, at the exception of Harvey Lawson (give ref), Eells Allendoerfer (on the cohomology of smooth manifolds) and Leray.
\paragraph{Relationship to other work.}
During the
preparation 
of our work
appeared a very interesting 
preprint
of Todorov, Nikolov and Stora
\cite{TNS}
whose approach 
is
close to the spirit
of the present work.
The difference is that 
the authors 
of \cite{TNS}
work 
on flat space
time
and deal
with 
\emph{associate homogeneous distributions}
in the terminology
of \cite{Gelfand}.
They
found 
the same notion
of residues
as
poles
of the meromorphic
regularization
and 
as anomaly
of the scaling equations.
However their
anomaly
residue
is not as general
as ours since
it only applies
to associate homogeneous
distributions
whereas ours applies
to \textbf{all weakly homogeneous}
distributions
and our formulation has a more 
homological flavour
with
the Schwartz, De Rham theory
of currents.
Our definition
of anomaly is broader
since it applies 
for all vector fields
of symmetries
and we make
more explicit
the connection with
the concept of
\textbf{periods}.
This work 
complements
nicely
the work 
of Dorothea
Bahns
and Michal
Wrochna \cite{WB}
which gives 
very explicit
anomaly
formulas in 
Minkowski
space-time.
We also learned
recently that
the problem
of extension
of currents
was also studied
in Complex analytic geometry
(\cite{Sib,SibDinh1}).

\section{Currents and renormalisation.}
\subsection{Notation and definitions.}
Let us denote by $\mathcal{D}^\prime_k(M)$ the topological dual of the space $\mathcal{D}^k(M)$ of compactly supported
test forms of degree $k$.
Elements
of $\mathcal{D}^\prime_k(M)$ are called
\textbf{currents}. 
If 
$\alpha\in \Omega^{n-k}(M)$ is a smooth form of degree
$n-k$, then \textbf{integration} on $M$ gives a linear map 
$\omega\in \mathcal{D}^k(M)\mapsto \left\langle \alpha,\omega\right\rangle= \int_M \alpha\wedge \omega$ which allows to interpret $\alpha$ as an element 
of $\mathcal{D}^\prime_k$. Thus we have the continuous injection 
$\Omega^{n-k}(M)\hookrightarrow  \mathcal{D}^\prime_k(M)$
and the symbol $\left\langle \alpha,\omega\right\rangle$ extends integration on $M$
to arbitrary $\alpha\in \mathcal{D}_k^\prime(M)$.
Finally, an important structure theorem states that the topological dual space
of the space of smooth compactly supported sections of a vector
bundle $E$ are just distributional sections of the dual bundle $E^\prime$, 
in our specific case 
$\mathcal{D}_k^\prime(M)=
\mathcal{D}^\prime(M)\otimes_{C^\infty(M)}\Omega^{n-k}(M)$ 
(for more on distributional sections
see \cite{Bar-Ginoux-Pfaffle,Gros}).
In the book of Laurent Schwartz 
\cite{Schwartz}, it is explained 
why currents can be treated as exterior forms, for instance the usual operations of contraction with a vector field (interior product), exterior differentiation, exterior product with a smooth form and Lie derivatives are well defined for currents.  
For $U\subset M$, we will denote by $H_k\left( \mathcal{D}^\prime(U) \right)$
the subspace of currents of $\mathcal{D}_k^\prime(U)$ 
which are closed in $U$, and
we denote by $B_k\left( \mathcal{D}^\prime(U) \right)$
the space of exact currents in $U$.
We can define a differential $d$ on the graded $C^\infty(U)$-module $H_\star\left( \mathcal{D}^\prime(U) \right)$
which extends the exterior derivative
of smooth forms to currents, thus 
$\left(H_\star\left( \mathcal{D}^\prime(U) \right),d\right)$ is a chain complex: 
$$H_{\star+1}\left( \mathcal{D}^\prime(U) \right)\overset{d}{\mapsto} H_\star\left( \mathcal{D}^\prime(U) \right) .$$

\paragraph{From $t$ to vector valued currents.}

 Let $\omega\in \mathcal{D}^{k}(M)$ be a test form, then 
the scaling of $\omega$ is defined by pull-back $\omega_\lambda=e^{\log\lambda\rho\star}\omega$. 
Therefore, we define scaling of currents by the following formula, 
for all current $T\in \mathcal{D}_k^\prime(M)$ 
and test forms $\omega\in \mathcal{D}^{k}(M)$:
$$T_\lambda(\omega)=T(\omega_{\lambda^{-1}}) .$$
\begin{defi}
Let $U$ be a $\rho$-convex subset of $M$.
A current $T\in \mathcal{D}_k^\prime(U)$ is in $E_{s}(\mathcal{D}_k^\prime(U))$ iff for all test forms $\omega\in \mathcal{D}^{k}(U)$
$$\sup_{\lambda\in(0,1]}\vert\lambda^{-s}T_\lambda(\omega)\vert <\infty .$$ 
\end{defi}
fortunately, this definition coincides with the definition of  \cite{Meyer}
because in the work of Meyer: $\lambda^{-d}\int_{\mathbb{R}^d} T\varphi_{\lambda^{-1}} d^dx=\int_{\mathbb{R}^d} T\left(\varphi d^dx\right)_{\lambda^{-1}}=\int_{\mathbb{R}^d} T_\lambda\varphi d^dx ,$
Meyer views distributions as dual
of test forms $\omega=\varphi d^dx$
and the theory of Chapter $1$ applies verbatim to this case.

\paragraph{The Taylor formula for test forms.}
It is important to understand 
the formalism of Taylor expansion 
for currents 
because we need to subtract Taylor polynomials 
in order to define certain renormalized extensions of distributions.
Let $\omega$ be a smooth test form in $\mathcal{D}^{k}(M)$, then for a given $\rho$ using the normal form theorem of chapter $1$, we find that there exists a local coordinate chart
around each point of $I$ in which $\rho=h^j\partial_{h^j}$ 
and $\omega=\sum_{\vert I\vert+\vert J\vert=k}\omega_{IJ}(x,h)dx^I\wedge dh^J$ where $I,J$ are multi-indices.
We immediately see that $\omega_{IJ}$ have various homogeneities w.r.t. $\rho$ depending on the length $\vert J\vert$.
Thus, it is wiser to view $\omega$ as a function
of $(x,h;dx,dh)$ 
smooth in $(x,h)$ 
and polynomial in the Grassmann variables $(dx,dh)$
which are treated on an equal footing
as the variables $(x,h)$,
a function $\omega$
is said to be homogeneous 
of degree $n$
if $\omega(x,\lambda h,dx,\lambda dh)=\lambda^n\omega(x,h,dx,dh)
$. Consider the decomposition:  
$$\omega=\sum_{0\leqslant n \leqslant m} \omega_{n} + I_m(\omega)=P_m(\omega)+I_m(\omega) $$
in the sense of the Taylor expansion of Chapter $1$:
$$\omega_n= \frac{1}{n!}\left(\left(\frac{d}{dt}\right)^{n}e^{\log t\rho*}\omega\right)|_{t=0} $$
where $\omega_{n}$ is homogeneous of degree $n$.
We also have the formula for the Taylor remainder:
$$I_m(\omega)=\frac{1}{m!}\int_0^1 dt(1-t)^m \left(\frac{d}{dt} \right)^{m+1}\left(e^{\log t\rho*}\omega\right) $$  
 
\begin{ex} 
In this formalism $dh$ is homogeneous of degree $1$, $\left(\left(\frac{d}{dt}\right)e^{\log t\rho*}dh\right)|_{t=0}=\frac{d}{dt} tdh|_{t=0}=dh$. 
\end{ex} 
 
\paragraph{Conceptual meaning of the Taylor expansion.}

 We give an equivalent formula for $\omega_n$ due to F H\'elein: 
\begin{equation}
\omega_n=\lim_{t\rightarrow 0} \frac{1}{t^nn!}\left(\rho\right)...\left(\rho-n+1\right)e^{\log t\rho*}\omega=\lim_{t\rightarrow 0}t^{-n}\left(\begin{array}{c} n \\ \rho \end{array}\right)e^{\log t\rho*}\omega
\end{equation} 
which allows to give the following conceptual remark:
$$\lim_{t\rightarrow 0}\frac{1}{t^nn!}\left(\rho\right)...\left(\rho-n+1\right)e^{\log t\rho*}\omega(p)=\lim_{t\rightarrow 0}t^{-n}\left(\begin{array}{c} n \\ \rho \end{array}\right)e^{\log t\rho*}\omega(p)$$ 
depends \textbf{linearly} on the $n$-jet of $\omega$ at the point $e^{\rho\log t}p$. 
But it also depends \textbf{polynomially} 
on the $(n-1)$-jet of the smooth Euler vector field $\rho$ at the point $e^{\rho\log t}p$. 
Finally, $\omega_n$ depends \textbf{linearly} on the $n$-jet of $\omega$, and depends polynomially on the $(n-1)$-jet of $\rho$ at the point $\lim_{t\rightarrow 0}e^{\rho\log t}p\in I$. Since the $n$-jet of $\omega$ at the point $\lim_{t\rightarrow 0}e^{\rho\log t}p\in I$ is independent of $\rho$, we deduce that
the Taylor polynomial $P_m(\omega)=\sum_{n\leqslant m} \omega_n $
depends linearly on the $m$-jet of $\omega$ along $I$, but it depends polynomially in the $(m-1)$-jets of $\rho$ along $I$. As noticed by H\'elein, in an arbitrary local chart, $P_m(\omega)$ is in general \textbf{not a polynomial} hence
the term Taylor polynomial is somewhat abusive, however in the coordinates in which $\rho$ takes the normal form $\rho=h^j\partial_{h^j}$, $P_m(\omega)$ is a genuine polynomial in the variables $h^j,dh^j$.
Let us discuss the expression of the Taylor polynomial $P$ in coordinates. Let $\omega$ be a test $k$-form which reads $\omega=\sum_{\vert I\vert+\vert J\vert=k}\omega_{IJ}dx^I\wedge dh^J$, then
$$P_m(\omega)= \sum_{\vert I\vert+\vert J\vert=k,\vert\gamma\vert+\vert J\vert\leqslant m}\frac{h^\gamma}{\gamma !}\partial^\gamma_h\omega_{IJ}(x,0)dx^I\wedge dh^J.$$
\subsection{From Taylor polynomials to local counterterms via the notion of moments of a compactly supported distribution $T$.}
\paragraph{The representation theorem.}
  Before we discuss the results of Chapter $1$ in the current theoretic setting, we would like to discuss the issue of \textbf{local counterterms}. But even before we discuss the problem of local counterterms, we must recall the representation theorem for currents supported on $I$ (see \cite{Melrose}).
%On $\mathbb{R}^{n+d}$ with coordinates $(x,h)$ where $I=\{h=0\}\simeq \mathbb{R}^n$, we recall 
%the distribution $\partial_h^\alpha\delta_I$ is defined as follows: $\forall \omega \in \mathcal{D}^{d}(\mathbb{R}^{n+d})$, $\left\langle \partial_h^\alpha\delta_I,\omega \right\rangle=(-1)^\alpha\int_{\mathbb{R}^n} \partial_h^\alpha\omega(x,0) d^nx $. 
For any distribution $t_{\alpha J}\in\mathcal{D}^\prime(I)$, if we denote by 
$i:I\hookrightarrow M$ the canonical embedding of $I$ in $M$ then 
$i_\star t_{\alpha J}$ is the push-forward of $t_{\alpha J}$ in $M$:
$$\forall\varphi\in\mathcal{D}(M), \left\langle i_\star t_{\alpha J}, \varphi \right\rangle=\left\langle t_{\alpha J}, \varphi\circ i \right\rangle.$$ 
Let $I\subset M$ be a closed embedded submanifold of $M$.
\begin{thm}\label{repsthm}
Let us consider a current $t\in \mathcal{D}_*\left( M\right)$ supported on $I$.
Then for any local 
system of coordinates $(h^j)_j$ transversal to $I$,
$t$ has a unique 
decomposition as locally finite
linear combinations of transversal derivatives
of push-forward to $M$ of currents $t_{\alpha J}$ 
in $\mathcal{D}_*^\prime (I)$:
\begin{equation}
t=\sum_{\alpha,J} \partial^\alpha_h\left(i_\star t_{\alpha J}\right)\wedge dh^J. 
\end{equation}
\end{thm}
\begin{proof}
We first use the decomposition of a current $t\in \mathcal{D}^\prime_k(M)$ as a sum 
$t_{I,J} dx^I\wedge dh^J$ where $t_{I,J}\in \mathcal{D}^\prime_0(M)$ are $0$-currents (see \cite{Giaquinta} 2.3 p.~123 and \cite{Rham} Chapter 3 p.~36).
Then the $0$-currents $t_{I,J}$ are in fact distributions supported on $I$, then we apply the structure theorem 37 p.~102 \cite{Schwartz} which describes distributions supported on a submanifold, which gives the desired result (also see \ref{lemmbound}).
\end{proof}

Let us explain the ideas of the concept of moments, first we fix a coordinate system which gives a basis $dx^i,dh^j$.
Then we define the \textbf{moments} $c_{\alpha I}\in D_*^\prime\left(I\right)$ of $T\in\mathcal{D}_*^\prime(M)$ by the push-forward formula,
if the projection $\pi:(x,h)\mapsto x$ is proper on 
$\text{supp }T$: 
\begin{eqnarray}
\forall \omega\in \mathcal{D}(I), \left\langle c_{\alpha I}(T),\omega\right\rangle
=\int_I\int_h\left(T \wedge\frac{h^\alpha}{\alpha!} \left(\frac{\partial}{\partial h^I}\lrcorner dh^d\right)\wedge \omega(x)\right). 
\end{eqnarray}
These moments are indexed by the multi-indices $(\alpha,I)$ and satisfy the identity
\begin{eqnarray}
\left\langle T , P_m(\omega)\right\rangle=\sum_{\vert\alpha\vert+d-\vert I\vert\leqslant m}\left\langle c_{\alpha,I}\wedge dh^I \partial_h^\alpha\delta_I ,\omega \right\rangle
\end{eqnarray}
In the case $n=0$, and $I=\{0\}$ is the \textbf{origin} of $\mathbb{R}^d$ and $T(h)$ 
is an integrable function in $L^1(\mathbb{R}^d)$,  
this definition coincides with the moment of the function $T\in L^1(\mathbb{R}^d)$ (see \cite{DuistermaatKolk} Proposition 6.3 p.~52). 
Now, we notice that when $t\in D_*^\prime(M)$ is supported on $I$, the moments $c_{\alpha,J}(t)$ of $t$ exactly coincide with the coefficients $t_{\alpha,J}$ in the representation (\ref{repsthm}).
The concepts of moments are crucial when we wish to represent currents supported on $I$ or residues.
\subsection{The results of Chapter $1$.}
Now that we have the 
suitable 
language to describe 
local counterterms, 
we can recall the results of Chapter $1$ in this new current theoretic setting: 
\begin{prop}
Let $T\in E_{s}\left(\mathcal{D}_k^\prime(M\setminus I)\right)$ and $p=\sup(0,k-n)$. 
If $s+p>0$ then for all $\omega\in \mathcal{D}_k(M)$
and $\chi$ is some smooth function which is equal to $1$
in a neighborhood of $I$: 
\begin{equation}
\lim_{\varepsilon\rightarrow 0}\left\langle T\left(\chi-e^{-\log\varepsilon\rho*}\chi\right), \omega\right\rangle
\end{equation}
exists.\\
If $s+p\leqslant 0$ and let $m\in\mathbb{N}$ s.t. $-m-1<s\leqslant -m$,
then for all $\omega\in \mathcal{D}_k(M)$:
\begin{equation}
\lim_{\varepsilon\rightarrow 0}\left\langle T\left(\chi-e^{-\log\varepsilon\rho*}\chi\right), I_m(\omega)\right\rangle
%=\lim_{\varepsilon\rightarrow 0}\left\langle T\left(\chi-e^{-\log\varepsilon\rho*}\chi\right), \omega \right\rangle-\sum_{\vert\alpha\vert+\vert I\vert\leqslant m}\left\langle c_{\alpha,I}(\varepsilon)\wedge dh^I D^\alpha\delta_I ,\omega \right\rangle
\end{equation}
exists where $I_m(\omega)$ is the generalized Taylor remainder
\begin{equation}
I_m(\omega)=\frac{1}{m!}\int_0^1 dt(1-t)^m \left(\frac{d}{dt} \right)^{m+1}\left(e^{\log t\rho*}\omega\right) 
\end{equation}
%and
%\begin{equation}
%c_{\alpha,I}(\varepsilon)=\pi_*\left(T\left(\chi-e^{-\log\varepsilon\rho*}\chi\right) \wedge\frac{h^\alpha}{\alpha!} \left(i_{\frac{\partial}{\partial h^I}}\lrcorner dh^d\right)\right)  
%\end{equation}
\end{prop}

\begin{proof}
We decompose the test forms $\omega$ in local coordinates $(x,h,dx,dh)$ then we reduce the proof exactly to the same proofs as in Chapter $1$. 
There are differences
involved because we are dealing with forms. In normal coordinates $(x,h)$ for $\rho$
%where for any frozen pair $\alpha,\beta$
%$$\sum_{\vert\gamma\vert+\vert\beta\vert\leqslant m}\partial_h^\gamma\omega_{\alpha\beta}(x,0)\frac{h^\gamma}{\gamma !} dx^\alpha dh^\beta + I_{\alpha\beta(m-\vert\beta\vert)}(x,h)dx^\alpha dh^\beta=\omega_{\alpha\beta}(x,h)dx^\alpha dh^\beta $$
If $\omega$ is a $k$ form, in the decomposition $\omega=\sum_{\vert I\vert+\vert J\vert=k}\omega_{IJ} dx^I dh^J$ the length $\vert J\vert$ of the multi-index $J$ is at least equal to 
$k-n$ because there are $n$ coordinate functions $(x^i)_{i=1\cdots n}$. 
Thus $\omega$ is in fact weakly homogeneous of degree $k-n$ which explains the criteria $s+k-n>0$. 
Now the second case is simple since $I_m(\omega)$ is weakly homogeneous of degree $m+1$.
\end{proof}

 We would like to introduce a new notation for the operation of regularization, we call it $R_\varepsilon$, and we define it as follows:
 
\begin{defi}\label{meyercurrents}
We define the continuous linear operator $R_\varepsilon$  on $E_{s}\left(\mathcal{D}_k^\prime(M\setminus I)\right)$ as follows. 
Let $p=\sup(0,k-n)$.
\begin{itemize}
\item If $s+p>0$ then for all $\omega\in \mathcal{D}_k(M)$: 
\begin{equation}
\left\langle R_{\varepsilon}T, \omega\right\rangle=\left\langle T\left(1-e^{-\log\varepsilon\rho*}\chi\right), \omega\right\rangle,
\end{equation}
and $\lim_{\varepsilon\rightarrow 0} R_\varepsilon T=RT$ exists in $\mathcal{D}_k^\prime(M)$ and defines an extension $RT$ of $T$.
\item If $s+p\leqslant 0$ and let $m\in\mathbb{N}$ s.t. $-m-1<s\leqslant -m$, then for all $\omega\in \mathcal{D}_k(M)$:
\begin{equation}
\left\langle R_\varepsilon T, \omega\right\rangle=\left\langle T\left(\chi-e^{-\log\varepsilon\rho*}\chi\right), I_m(\omega) \right\rangle+\left\langle T\left(1-\chi\right), \omega\right\rangle,
\end{equation}
and $\lim_{\varepsilon\rightarrow 0} R_\varepsilon T=RT$ exists in $\mathcal{D}_k^\prime(M)$ and defines an extension $RT$ of $T$.
\end{itemize}
%or by duality, we can introduce the formula
%\begin{eqnarray}
% \left\langle\frac{1}{m!}\int_\varepsilon^1 dt(1-t)^m \left(\frac{d}{dt} \right)^{m+1} e^{-\log t\rho*}\left(T\left(\chi-e^{-\log\varepsilon\rho*}\chi\right)\right),\omega\right\rangle+\left\langle T\left(1-\chi\right), \omega\right\rangle
%\\  R_\varepsilon T= \frac{1}{m!}\int_\varepsilon^1 dt(1-t)^m \left(\frac{d}{dt} \right)^{m+1} e^{-\log t\rho*}\left(T\left(\chi-e^{-\log\varepsilon\rho*}\chi\right)\right)+T\left(1-\chi\right)
%\end{eqnarray}
\end{defi}

\section{Renormalization, local counterterms and residues.}

% In this section, we show the regularization scheme of Meyer produces residues in the sense of Griffiths Harris (see the section $3$ in Principles of algebraic geometry) and our definition has nothing to do with complex analysis. 
% 
%
%
%
% This is quite unusual in the litterature and comparable works involving residues are not common, at the exception of Harvey Lawson (give ref), Eells Allendoerfer (on the cohomology of smooth manifolds) and Leray.
% 
% We give the minimal regularity hypothesis on the current $T$ which guaranties the existence of the residues. 
\subsection{The ambiguities of the operator $R_\varepsilon$ and the moments of a distribution $T$.}

 Actually, first notice that any current $T$ in $\mathcal{D}_*^\prime(M)$ is also an element of $\mathcal{D}_*^\prime(M\setminus I)$ by the pull-back $i^*T$ by the restriction map $i:M\setminus I\hookrightarrow M$. Thus we ask ourselves a very natural question, does the restriction followed by the extension operation allows to reconstruct the element $T$, in other words do we have
$\lim_{\varepsilon\rightarrow 0}R_\varepsilon i^*T=T  $ ?
The answer is no ! A distribution supported on $I$ is automatically killed by $R_\varepsilon,\forall \varepsilon>0$ thus  if $T$ is supported on $I$
$\lim_{\varepsilon\rightarrow 0}R_\varepsilon i^*T=0$.
This idea is strongly related to the discussion in \cite{Meyer} Chapter $1$, 
let $t$ be a tempered distribution, does 
the Littlewood--Paley series $\sum_{j=-N}^\infty \Delta_j(t)$ 
converges \emph{weakly} to $t$ when $N\rightarrow +\infty$ ?
The answer is no! There is convergence modulo floating polynomials in Fourier space (see \cite{Meyer} Proposition 1.5 p.~15). 
The floating polynomials in Fourier space are in fact corrections that we have to subtract
from
the Littlewood--Paley series 
in order to make it convergent and 
these polynomials 
should be related 
to vanishing moments conditions (see Meyer chapter $2$ p.~45).
We introduce a linear operator 
$A$ which describes the ambiguities of the restriction-extension operation on the distribution $T$.
\begin{defi}
Let $T\in \mathcal{D}_k^\prime(M)$, then we define the ambiguity as
$$AT=\lim_{\varepsilon \rightarrow 0} \left(T-R_\varepsilon T\right) .$$
The operator $A$ depends on $\chi$.
\end{defi}
The ambiguity is a non trivial operator because of the example discussed previously. 
As usual, we motivate our theorem with the simplest fundamental example
\begin{ex}
$\delta\in \mathcal{D}^\prime(\mathbb{R})$ is a well defined distribution.
But $\forall \varepsilon>0, R_\varepsilon\delta=0$ because $0$ never meets the support of the cut-off hence
$$A\delta=\lim_{\varepsilon\rightarrow 0} \left(\delta-R_{\varepsilon}\delta\right) =\delta $$
\end{ex}
We state a simple theorem which expresses the ambiguity $A$ in terms of the moments of $T\chi$. 
\begin{thm}\label{ambiguitythm}
Let $T\in E_{s}\left(\mathcal{D}^\prime_k(M)\right)$ where $-(m+1)<s\leqslant -m,m\in\mathbb{N}$, 
then 
the ambiguity $AT$ is given by the following formula:
\begin{eqnarray}
\forall \omega\in\mathcal{D}^k(M), AT(\omega)=\left\langle T\chi, P_m(\omega) \right\rangle,
\end{eqnarray}
where $P_m(\omega)=\sum_{k\leqslant m}\omega_k$.
\end{thm}
\begin{proof}
Yves Meyer defines the ambiguity by the Bernstein theorem. We will give a more direct in space proof which does not use the Fourier transform. The first idea is the concept of moments of a current $T\chi\in \mathcal{D}_k^\prime\left(M\right)$. First write the duality coupling in simple form: 
 $$\left\langle T,\omega \right\rangle=\left\langle T(1-\chi),\omega \right\rangle+\left\langle T\chi,\omega \right\rangle=\left\langle T(1-\chi),\omega \right\rangle+\left\langle T\chi,P_m(\omega)\right\rangle+\left\langle T\chi,I_m(\omega) \right\rangle $$
where $P$ is the Taylor polynomial $\sum_{k\leqslant m} \omega_k$.
We remind the definition of $R_\varepsilon T$   
$$\left\langle R_\varepsilon T, \omega\right\rangle=
\left\langle T\left(1-\chi\right), \omega\right\rangle
+\left\langle T\left(\chi-e^{-\log\varepsilon\rho*}\chi\right), I_m(\omega)\right\rangle .$$   
Then we immediately find:
$$\left\langle T,\omega \right\rangle-\left\langle R_\varepsilon T, \omega\right\rangle=\left\langle T\chi,P_m(\omega)\right\rangle+\left\langle Te^{-\log\varepsilon\rho*}\chi, I_m(\omega) \right\rangle $$
now notice that 
$$\left\langle Te^{-\log\varepsilon\rho*}\chi, I_m(\omega) \right\rangle =\left\langle \left(e^{\log\varepsilon\rho*}T\right)\chi, e^{\log\varepsilon\rho*}I_m(\omega) \right\rangle =\left\langle T_\varepsilon\chi, (I_m(\omega))_{\varepsilon} \right\rangle$$
where
$$\exists C>0, \vert\left\langle T_\varepsilon\chi, (I_m(\omega))_{\varepsilon} \right\rangle\vert\leqslant C\varepsilon^{s+m+1}\rightarrow 0$$
since 
$\chi(I_m(\omega))_{\varepsilon} $
is a bounded family of test
forms,
thus 
$$AT(\omega)=\left\langle T\chi,P_m(\omega)\right\rangle$$
where $\omega=P_m(\omega)+I_m(\omega)$ and the final result follows from the definition of the notion of moment of the distribution $T\chi$.
\end{proof} 
% We can formalize the previous result in the form of an exact sequence
% 
%\begin{coro}
% Let $E_{sk}\left(I\subset M\right)$ be the currents of degree $k$ supported on $I$ and weakly homogeneous of degree $s$. $E_{s}\left(\mathcal{D}_k^\prime(M\setminus I)\right),E_{sk}\left(M\right)$ are the $k$-currents weakly homogeneous of degree $s$ defined on $M\setminus I,M$ respectively.
% 
% We have the short exact sequence
%\begin{equation}
%0\rightarrow E_{s}\left(\mathcal{D}_k^\prime(M\setminus I)\right)\overset{R}\rightarrow E_{sk}\left(M\right)\overset{A}\rightarrow E_{sk}\left(I\subset M\right) \rightarrow 0
%\end{equation}
%\end{coro} 
\subsubsection{The dependence of $R$ on the choice of $\chi,\rho$.}
We would also like to describe the dependance of 
the operator $R$ on the choice of 
$\chi$ and $\rho$.
As usual, the result will be expressed in terms of \textbf{local counterterms}.
\paragraph{Changing $\chi$.} 
Let $\chi_1,\chi_2$ be two functions such that $\chi_i=1,i=1,2$ 
in a neighborhood of $I$ and $\rho\chi_i$
is uniformly 
supported in 
an annulus
domain of $M$.
Let $R^i_\varepsilon,i=1,2$ be 
the corresponding regularization operators on 
$E_{s}\left(\mathcal{D}_k^\prime(M\setminus I)\right)$
defined as follows:
for $p=\sup(k-n,0)$,
if $s+p\leqslant 0$ let $m\in\mathbb{N}$ 
s.t. $-m-1<s\leqslant m$, then 
the regularization operator $R^i$
corresponding to each $\chi_i,i=(1,2)$ is given by the formula 
\begin{equation}\left\langle R^i_\varepsilon T, \omega\right\rangle=\left\langle T\left(\chi_i-e^{-\log\varepsilon\rho*}\chi_i\right), I_m\left(\omega\right) \right\rangle+\left\langle T\left(1-\chi_i\right), \omega\right\rangle,
\end{equation}
and $\lim_{\varepsilon\rightarrow 0} R^i_\varepsilon T=R^iT$ exists in $\mathcal{D}_k^\prime(M)$ and defines an extension $R^iT$ of $T$
if otherwise $s+p\geqslant 0$ then 
$R^iT=\lim_{\varepsilon\rightarrow 0} T(1-\chi_{i\varepsilon^{-1}})$.

\begin{thm}\label{Changingchi}
Let $T\in E_{s}\left(\mathcal{D}_k^\prime(M\setminus I)\right)$.
If $s+p>0$ then $ R^1T=R^2T$ (i.e. 
$R$ does not depend on the choice of $\chi$).
If $s+p\leqslant 0$ then 
\begin{eqnarray}
\left\langle \left(R^1-R^2\right) T, \omega\right\rangle=\left\langle T\left(\chi_2-\chi_1\right), P_m(\omega)\right\rangle,
\end{eqnarray}
where $m\in\mathbb{N}$ is s.t. $-m-1<s\leqslant -m$.
\end{thm}
\begin{proof}
By definition, we have:
$$\left\langle R_\varepsilon^i T, \omega\right\rangle=\left\langle T(\chi_i-\chi_{i\varepsilon^{-1}}), I_m(\omega)\right\rangle+\left\langle T\left(1-\chi_i\right), \omega\right\rangle$$

 The only thing we have to do is to compute the difference $\left(R_\varepsilon^1-R_\varepsilon^2\right) T$. First notice that
$$ \left\langle T\left(1-\chi_1\right), \omega\right\rangle=\left\langle T\left(1-\chi_2\right), \omega\right\rangle+\left\langle T\left(\chi_2-\chi_1\right), \omega\right\rangle$$ $$=\left\langle T\left(1-\chi_2\right), \omega\right\rangle+\left\langle T\left(\chi_2-\chi_1\right), P_m(\omega)\right\rangle+\left\langle T\left(\chi_2-\chi_1\right), I_m(\omega)\right\rangle $$
thus 
$$  \left\langle R_\varepsilon^1 T, \omega\right\rangle=\left\langle T\left(1-\chi_1\right), \omega\right\rangle+\left\langle T(\chi_1-\chi_{1\varepsilon^{-1}}), I_m(\omega)\right\rangle$$ $$=\left\langle T\left(1-\chi_2\right), \omega\right\rangle+\left\langle T\left(\chi_2-\chi_1\right), P_m(\omega)\right\rangle+\left\langle T\left(\chi_2-\chi_1\right), I_m(\omega)\right\rangle+\left\langle T(\chi_1-\chi_{1\varepsilon^{-1}}), I_m(\omega)\right\rangle $$
$$=\left\langle T\left(1-\chi_2\right), \omega\right\rangle+\left\langle T\left(\chi_2-\chi_1\right), P_m(\omega)\right\rangle+\left\langle T(\chi_2-\chi_{1\varepsilon^{-1}}), I_m(\omega)\right\rangle $$
then computing the difference
$$ \left\langle \left(R_\varepsilon^1-R_\varepsilon^2\right) T, \omega\right\rangle=\left\langle R_\varepsilon^1T, \omega\right\rangle-\left\langle R_\varepsilon^2 T, \omega\right\rangle$$ $$=\left\langle T\left(1-\chi_2\right), \omega\right\rangle+\left\langle T\left(\chi_2-\chi_1\right), P_m(\omega)\right\rangle+\left\langle T(\chi_2-\chi_{1\varepsilon^{-1}}), I_m(\omega)\right\rangle$$ 
$$-\left\langle T(\chi_2-\chi_{2\varepsilon^{-1}}), I_m(\omega)\right\rangle-\left\langle T\left(1-\chi_2\right), \omega\right\rangle$$
$$=\left\langle T\left(\chi_2-\chi_1\right), P_m(\omega)\right\rangle+\left\langle T\left(\chi_2-\chi_1\right)_{\varepsilon^{-1}}, I_m(\omega)\right\rangle $$
 As in the proof of theorem (\ref{ambiguitythm}), we can take the limit $\varepsilon\rightarrow 0$ and we find  that the term $\left\langle T\left(\chi_2-\chi_1\right)_{\varepsilon^{-1}}, I_m(\omega)\right\rangle$ will vanish when $\varepsilon\rightarrow 0$.
\end{proof}
\paragraph{Changing $\rho$.} 
We say that $\chi$ is compatible with $\rho$ iff
for each $p\in I$, there is a neighborhood $V_p$ 
of $p$
and a local chart $(x,h):V_p\mapsto \mathbb{R}^{n+d}$ 
on this neighborhood
on which $\rho=h^j\frac{\partial}{\partial h^j}$,
$\chi=0$ when $\vert h\vert \geqslant b$
and $\chi=1$ when $\vert h\vert\leqslant a$ for some pair $0<a<b$.
Let $\rho_1,\rho_2$ be two 
Euler vector fields 
and $\chi$ which is \emph{compatible} with $\rho_1$ and $\rho_2$.
Let $R^i_\varepsilon,i=1,2$ be 
the corresponding regularization operators 
on $E_{s}\left(\mathcal{D}_k^\prime(M\setminus I)\right)$ 
defined as follows:
for $p=\sup(k-n,0)$. 
If $s+p\leqslant 0$,
let $m$ s.t. $-m-1<s\leqslant m$, 
the regularization operator $R^i$
corresponding to each $\rho_i,i=(1,2)$ is given by the formula 
\begin{equation}\left\langle R^i_\varepsilon T, \omega\right\rangle
=\left\langle T\left(\chi-e^{-\log\varepsilon\rho_i*}\chi\right),I_{im} \right\rangle 
+ \left\langle T\left(1-\chi\right), \omega\right\rangle
\end{equation}
where $\omega=P_{im}(\omega)+I_{im}(\omega),i=1,2$, $P_{im}(\omega)$ is the ``Taylor polynomial of order $m$''
of $\omega$ 
for the Euler vector field $\rho_i$ 
and $\lim_{\varepsilon\rightarrow 0} R^i_\varepsilon T=R^iT$ exists in $\mathcal{D}_k^\prime(M)$ and defines an extension $R^iT$ of $T$.
Otherwise, if $s+p>0$ then
$R^iT=\lim_{\varepsilon\rightarrow 0} T(1-e^{-\log\varepsilon\rho_i\star}\chi)$.

\begin{thm}
Let $T\in E_{s}\left(\mathcal{D}_k^\prime(M\setminus I)\right)$ 
and $p=\sup (0,k-n)$.
If $s+p>0$ then $R^1T=R^2T$.
If $s+p\leqslant 0$ let 
$m\in\mathbb{N}$ s.t. $-(m+1)<s\leqslant -m$, 
then
for any Euler vector field 
$\rho$ such that $\chi$ is \emph{compatible}
with $\rho$,
\begin{eqnarray}
\left\langle \left(R^1-R^2\right) T,\omega\right\rangle=\lim_{\varepsilon\rightarrow 0} \left\langle T\left(\chi-e^{-\log\varepsilon\rho\star}\chi\right), P_{2m}(\omega)-P_{1m}(\omega)\right\rangle.
\end{eqnarray}
\end{thm}
Notice that in the conclusion 
of this theorem the vector field 
$\rho$ is chosen independently of 
$\rho_1,\rho_2$.

\begin{proof}
Before we prove our claim, we would like to give some important remarks. 
\begin{itemize}
\item First, no matter what Euler vector field $\rho_i$ we choose,
the Taylor remainder $I_{im}(\omega)$ \emph{always vanishes at order} $m$ on the submanifold $I$.
The key point is that if a smooth form $\varpi$ vanishes at order $m$ at $I$,
then the limit $\lim_{\varepsilon\rightarrow 0} \left\langle T\left(\chi-e^{-\log\varepsilon\rho\star}\chi\right), \varpi\right\rangle$
does not depend on the choice
of Euler vector field $\rho$
provided $\chi$ is $\rho$ admissible.
Hence by choosing some 
Euler vector field $\rho$ for which $\chi$ is $\rho$ admissible,
we still have
$$\forall i, \lim_{\varepsilon\rightarrow 0}\left\langle T (\chi-e^{-\log\varepsilon\rho_i}\chi), I_{im}(\omega)  \right\rangle=\lim_{\varepsilon\rightarrow 0}\left\langle T (\chi-e^{-\log\varepsilon\rho}\chi), I_{im}(\omega)  \right\rangle$$ $$=\lim_{\varepsilon\rightarrow 0}\left\langle T (\chi-\chi_{\varepsilon^{-1}}), I_{im}(\omega)  \right\rangle \text{ where } \chi_{\varepsilon^{-1}}=e^{-\log\varepsilon\rho}\chi.$$ 
\item Secondly, if we denote by $P_{im}(\omega),i=1,2$ (resp $I_{im}(\omega),i=1,2$) the ``Taylor polynomials'' (resp ``Taylor remainders'') associated with $\rho_i,i=1,2$, then
from $\omega=P_{1m}(\omega)+I_{1m}(\omega)=P_{2m}(\omega)+I_{2m}(\omega)$ we deduce that $I_{1m}(\omega)-I_{2m}(\omega)=P_{2m}(\omega)-P_{1m}(\omega)$, hence $P_{2m}(\omega)-P_{1m}(\omega)$ depends only on some finite jet of $\omega,\rho_1,\rho_2$ and
vanishes at order $m$ at $I$ 
(it is in general \textbf{not a polynomial in 
arbitrary local charts}).
\end{itemize}

We can now compute $\left(R^1-R^2\right) T$: 
$$\left\langle \left(R^1-R^2\right) T, \omega\right\rangle=\lim_{\varepsilon\rightarrow 0} \left\langle T(\chi-\chi_{\varepsilon^{-1}}), I_{1m}(\omega)-I_{2m}(\omega)\right\rangle .$$
Using $I_{1m}(\omega)-I_{2m}(\omega)=P_{2m}(\omega)-P_{1m}(\omega)$, we finally get: 
$$\left\langle \left(R^1-R^2\right) T, \omega\right\rangle=\lim_{\varepsilon\rightarrow 0} \left\langle T(\chi-\chi_{\varepsilon^{-1}}),P_{2m}(\omega)-P_{1m}(\omega)\right\rangle $$ 
where 
the above limit makes sense since 
$P_{2m}(\omega)-P_{1m}(\omega)$ 
vanishes 
at order $m$ on the submanifold $I$.
\end{proof}

\subsection{The geometric residues.}

\subsubsection{The residues and the coboundary $d$ of currents.}
  We want to describe the ambiguities of the restriction-extension operation on closed currents $T\in H_*\left(\mathcal{D}_*^\prime(M\setminus I),d \right)$ defined on $M\setminus I$ and on exact currents $dT\in B_*\left(\mathcal{D}_*^\prime(M\setminus I),d \right)$ defined on $M\setminus I$.
In other words
one could ask is how
does our extension
procedure
behaves when applied to
closed currents ?
The notion
of \textbf{residue}
(following \cite{Griffiths} and Eells--Allendoerfer \cite{Eells})
that we define below answers this question,
$\mathfrak{Res}[T]$ is defined as the solution of the chain homotopy equation:
\begin{eqnarray}
\mathfrak{Res}[T]=d RT-Rd T.
\end{eqnarray}

Recall $E_{s}\left(\mathcal{D}_k^\prime(M\setminus I)\right)$ 
is the space of $k$-currents in 
$\mathcal{D}^\prime_k(M\setminus I)$ 
which are weakly homogeneous of degree $s$
and we work on $M\setminus I$
where $\dim M=n+d$ and $\dim I=n$.
\begin{thm}
Let $T\in E_{s}\left(\mathcal{D}_k^\prime(M\setminus I)\right)$, 
and $p=\sup(0,k-n-1)$.
If $s+p>0 $
then
$\mathfrak{Res}[T]=0$. 
%is a \textbf{current supported on }$I$ given by the formula
%\begin{equation}
%\forall \omega\in \mathcal{D}^{k-1}(M), \mathfrak{Res}[T](\omega)=-\lim_{\varepsilon\rightarrow 0}\left\langle  d\chi,T_{\varepsilon}\wedge\omega_{\varepsilon}\right\rangle
%\end{equation}
\end{thm}
\begin{proof}
The key remark is that
$dT\in E_{s}\left(\mathcal{D}_{k-1}^\prime(M\setminus I)\right)$
since $d$ is scale invariant.
The residue equals $d(RT)-R(dT)$ by definition. If $s+p>0$ then
by definition of $R$ (\ref{meyercurrents}): 
$$\left\langle d(R T)-R(dT),\omega\right\rangle=
\lim_{\varepsilon\rightarrow 0}\left\langle d((1-\chi_{\varepsilon^{-1}})T)-(1-\chi_{\varepsilon^{-1}})(dT),\omega\right\rangle$$
since there are no counterterms to subtract
$$=-\lim_{\varepsilon\rightarrow 0}\left\langle  d\chi_{\varepsilon^{-1}},T\wedge\omega\right\rangle =0.$$
Since $\vert\left\langle  d\chi_{\varepsilon^{-1}},T\wedge\omega\right\rangle\vert\leqslant C\varepsilon^{s+p}$ for some $C>0$ by the hypothesis of homogeneity on $T$ and the degree of $T$.
%The locality of $\mathfrak{Res}$ is simple to establish, since if we assume we work in local coordinates $(x,h)$ and
%$\text{supp }d\chi\subset \{a\leqslant h\leqslant b\}$, then $\left\langle  d\chi,T_{\varepsilon}\wedge\omega_{\varepsilon}\right\rangle$ depends only on finite jets of $\omega$ on $\{\varepsilon a\leqslant h\leqslant \varepsilon b\}$.
\end{proof}
Let us give the 
fundamental example 
of residue from Griffiths--Harris see \cite{Griffiths} p.~367 
and Laurent Schwartz \cite{Schwartz} p.~345-347.
\begin{ex}
Let $H$ be the Heaviside function on $\mathbb{R}$. $H$ is a smooth closed $0$-form on $\mathbb{R}\setminus \{0\}$. The local integrability around $0$ guarantees it extends in a unique way as a current denoted $RH\in \mathcal{D}_1^\prime\left(\mathbb{R}\right)$. 
By integration by parts and by the fact that $d H|_{\mathbb{R}\setminus \{0\}}=0$ since $H$ is \textbf{closed}, it is immediate that
$$ d RH -R \underset{=0}{\underbrace{d H}}=d RH = \delta_0(x) dx  $$ 
So the current $\delta_0(x) dx\in \mathcal{D}_0^\prime(\mathbb{R})$ 
is the residue of the Heaviside function $H$ 
which is closed on $\mathbb{R}\setminus \{0\}$.
\end{ex}
In the above example, 
the residue measures the jump at $0$.
However in the case 
of renormalization theory, 
our residues must generalize the 
``classical'' 
notion of residue to
take into account
\textbf{more singular} distributions (see \cite{Griffiths} p.369,371).
\begin{thm}\label{residuethm}
Let $T\in E_{s}\left(\mathcal{D}_k^\prime(M\setminus I)\right)$, 
and $p=\sup(0,k-n-1)$.
If $s+p\leqslant 0$
let for $m\in\mathbb{N}$ s.t. $-m-1<s\leqslant -m$,
then
$\mathfrak{Res}$ is a \textbf{current supported on }$I$ given by the formula
\begin{equation}
\forall \omega\in \mathcal{D}^{k-1}(M), \mathfrak{Res}[T](\omega)=(-1)^{n-k-1}\left\langle T,d\chi\wedge P_m(\omega)\right\rangle.
\end{equation}
\end{thm}
\begin{proof} 
Let $T$ be a current in $\mathcal{D}^\prime_k$ and $\omega\in \mathcal{D}^{k-1}(M)$ a $k-1$ test form. We want to compute the difference $\left\langle d \left(R_\varepsilon T\right),\omega \right\rangle-\left\langle  \left(R_\varepsilon d T\right),\omega \right\rangle$.
%By duality, we will think of the operator $R_\varepsilon$ as acting on the test form $\omega$ as 
% $$R_\varepsilon \omega =(1-\chi)\omega  + \left(\chi-e^{-\log\varepsilon\rho*}\chi\right) \frac{1}{m!}\int_\varepsilon^1 dt(1-t)^m \left(\frac{d}{dt} \right)^{m+1}\left(e^{\log t\rho*}\omega\right) $$ 
There are two cases for this theorem.

 Either both $T$ and $dT$ need a renormalization. We first treat this case.
By definition of the coboundary $d$ of a current (\cite{Schwartz}, \cite{Giaquinta}),
we find that $$\left\langle d \left(R_\varepsilon T\right),\omega \right\rangle-\left\langle  \left(R_\varepsilon d T\right),\omega \right\rangle=(-1)^{n-k-1}\left\langle R_\varepsilon T,   d\omega \right\rangle -\left\langle  R_\varepsilon dT ,\omega \right\rangle.$$
On the one hand, we have: 
$$\left\langle R_\varepsilon T,   d\omega \right\rangle=\left\langle T,(1-\chi)d\omega \right\rangle+\left\langle T\left(\chi-e^{-\log\varepsilon\rho*}\chi\right), I_m(d\omega) \right\rangle$$ $$=\left\langle T,(1-\chi)d\omega \right\rangle+\left\langle T\left(\chi-e^{-\log\varepsilon\rho*}\chi\right) , dI_m(\omega) \right\rangle $$
since $\frac{1}{m!}\int_\varepsilon^1 dt(1-t)^m \left(\frac{d}{dt} \right)^{m+1}\left(e^{\log t\rho*}d\omega\right)=d\frac{1}{m!}\int_\varepsilon^1 dt(1-t)^m \left(\frac{d}{dt} \right)^{m+1}\left(e^{\log t\rho*}\omega\right) $
because $d$ commutes with the pull-back operator $e^{\log t\rho*}$. 
We hence notice the important fact that if we view $I_m$ and $P_m$ as projections
in $Hom\left(\mathcal{D}^\star(M), \mathcal{D}^\star(M)\right)$, 
then they \textbf{commute} with $d$.
On the other hand:
$$\left\langle  R_\varepsilon dT, \omega \right\rangle=\left\langle dT,\left(1-\chi\right) \omega\right\rangle +\left\langle dT,\left(\chi-e^{-\log\varepsilon\rho*}\chi\right) I_m(\omega)\right\rangle ,$$
then following the definition of the coboundary $d$ of a current, 
we differentiate the test form: 
$$\left\langle T, d\left(\left(1-\chi\right) \omega\right)\right\rangle +\left\langle T,d\left(\left(\chi-e^{-\log\varepsilon\rho*}\chi\right) I_m(\omega)\right) \right\rangle=\left\langle T,\left(1-\chi\right)d \omega\right\rangle -\left\langle T,\left(d\chi\right)\wedge \omega\right\rangle$$ $$ +\left\langle T,\left(d\chi\right)\wedge I_m(\omega) \right\rangle - \left\langle T,\left(d\chi\right)_{\varepsilon^{-1}}\wedge I_m(\omega) \right\rangle + \left\langle T,\left(\chi-e^{-\log\varepsilon\rho*}\chi\right) dI_m(\omega)\right) .$$
Thus  
$$(-1)^{n-k-1}\left\langle R_\varepsilon dT, \omega\right\rangle=\left\langle T,\left(1-\chi\right)d \omega\right\rangle - \left\langle T,\left(d\chi\right)\wedge P_m(\omega)\right\rangle$$
$$- \left\langle T,\left(d\chi\right)_{\varepsilon^{-1}}\wedge I_m(\omega) \right\rangle + \left\langle T,\left(\chi-e^{-\log\varepsilon\rho*}\chi\right) dI_m(\omega)\right)  $$
where $\omega=P_m(\omega)+I_m(\omega)$ by the Taylor formula. 
Then we find:
$$\left\langle dR_\varepsilon T,   \omega \right\rangle-\left\langle R_\varepsilon dT,\omega \right\rangle=(-1)^{n-k-1}\left(\left\langle T,\left(d\chi\right)\wedge P_m(\omega)\right\rangle + \left\langle T,\left(d\chi\right)_{\varepsilon^{-1}}\wedge I_m(\omega) \right\rangle\right).$$
Now notice that 
$$\left\langle T,\left(d\chi\right)_{\varepsilon^{-1}}\wedge I_m(\omega) \right\rangle=\left\langle T, \left(e^{-\log\varepsilon\rho*}d\chi\right)\wedge I_m(\omega) \right\rangle$$ 
$$ =\left\langle e^{\log\varepsilon\rho*}T,\left(d\chi\right)\wedge e^{\log\varepsilon\rho*}I_m(\omega) \right\rangle =\left\langle T_\varepsilon,\left(d\chi\right)\wedge I_m(\omega)_{\varepsilon} \right\rangle$$
and the above term 
satisfies the following estimate:
$$\exists C>0, \vert\left\langle T_\varepsilon,d\chi\wedge I_m(\omega)_{\varepsilon} \right\rangle\vert\leqslant C\varepsilon^{s+m+1}\underset{\varepsilon\rightarrow 0}{\rightarrow} 0$$
since $-m-1<s$, $T$ is weakly homogeneous
of degree $s$ and
the family of test
forms $d\chi\wedge I_m(\omega)_{\varepsilon},\varepsilon\in[0,1]$
is bounded.
Thus 
$$\lim_{\varepsilon\rightarrow 0}\left\langle  \left(d\circ R_\varepsilon 
- R_\varepsilon\circ d\right) T, \omega \right\rangle= (-1)^{n-k-1}\left\langle T,\left(d\chi\right)\wedge P_m(\omega)\right\rangle .$$
Finally, we find
$$\mathfrak{Res}[T](\omega)=(-1)^{n-k-1}\left\langle T,d\chi\wedge P_m(\omega)\right\rangle.$$

 Either
$T$ is s.t. $s+\sup (0,k-n)>0$ 
thus $RT$ does not need a renormalization 
and $s+\sup (k-n-1,0)\leqslant 0$ 
which implies that the definition 
of the extension $RdT$ 
needs a renormalization
and that $k-n-1\geqslant 0$, thus
$p=k-n-1$.
Actually since $-p-1<s\leqslant -p$,
we must subtract a counterterm
$P_p(\omega)$ to the
$k-1$ form $\omega$
to define the extension: $RdT$. 
The key fact
is to notice that
$d\omega$
is polynomial 
in $dh$ of degree at least 
$p+1$
thus $d\omega= I_p(d\omega)=d\omega-P_p(d\omega)$
and $$\left\langle R_\varepsilon T,d\omega\right\rangle=\left\langle (1-\chi)T,d\omega\right\rangle 
+ \left\langle (\chi-\chi_{\varepsilon^{-1}})T,I_p(\omega)\right\rangle$$
and we are reduced to the first case.
\end{proof}

We give the most fundamental example illustrative of our approach
\begin{ex}
We set $T=\frac{1}{\vert x\vert}$ and we will show how to compute 
the residue for this simple example. 
$RT$ is defined by the formula $\left\langle RT, \varphi dx\right\rangle=\int_{-\infty}^\infty  \frac{1}{\vert x\vert}\chi(x) (\varphi(x)-\varphi(0))dx + \int_{-\infty}^\infty  \frac{1}{\vert x\vert}(1-\chi(x))\varphi(x)dx$.
The residue is given by the simple formula
$$\mathfrak{Res}[\frac{1}{\vert x\vert}]=-\left(\int_{-\infty}^\infty  \frac{1}{\vert x\vert}(\partial_x\chi)(x)dx\right) \delta_0 .$$ 
\end{ex}
We give a second example 
which illustrates the limit
case where $RT$
does not need a renormalization
but $RdT$ does.
\begin{ex}
Let us work in $\mathbb{R}^d$ and $n=0$.
Let $T$ be a $d-1$ form
in $\mathbb{R}^d\setminus \{0\}$
which is homogeneous
of degree $0$, 
i.e. $T\in E_0(\mathcal{D}^{\prime}_1(\mathbb{R}^d\setminus \{0\}))$, 
then
$p=\sup(0,1-0)=1$ and $s+1=0+1>0$ thus
$RT$ does not 
need a renormalization.
$dT$ is a $d$ form
which is still 
homogeneous of degree $0$ but
$dT \in E_0(\mathcal{D}^{\prime}_0(\mathbb{R}^d\setminus \{0\}))$ thus $s+0=0+0\leqslant 0$
and $dT$
needs a renormalization 
with subtraction
of the form $\omega_0$.
$$\left\langle RdT,\omega\right\rangle =\lim_{\varepsilon\rightarrow 0}\left\langle dT,(1-\chi_{\varepsilon^{-1}})
\omega\right\rangle-\left\langle dT,(\chi-\chi_{\varepsilon^{-1}})
\omega_0\right\rangle
$$
but notice that
$\left\langle dT,
(\chi-\chi_{\varepsilon^{-1}})
\omega_0\right\rangle
=0$ by scale invariance
of $\omega_0$ and $dT$
thus in this 
example
the 
\textbf{counterterm vanishes}.
%$\left\langle RT,d\omega\right\rangle
%=\left\langle T(1-\chi_{\varepsilon^{-1}}),d\omega\right\rangle
%=\left\langle T(1-\chi_{\varepsilon^{-1}}),d\omega\right\rangle
%-\left\langle T(\chi-\chi_{\varepsilon^{-1}}),(d\omega)_0\right\rangle$
%since $d\omega_0=0$ !
Finally, 
the residue
satisfies the simple equation:
$$\lim_{\varepsilon\rightarrow 0}
\left\langle 
d\left(T(1-\chi_{\varepsilon^{-1}})\right),\omega \right\rangle -
\left\langle dT,(1-\chi_{\varepsilon^{-1}})\omega \right\rangle=\left\langle T , \omega_0d\chi \right\rangle.$$
\end{ex}
% Using the result of proposition (\ref{momentsprop}) and following the theorem (\ref{repsthm}), let us represent Res by the formula:
%\begin{eqnarray} 
%\mathfrak{Res}[T]=\sum c_{\alpha,I}(Res[T]) D^\alpha\delta_I\wedge  dh^I
%\\c_{\alpha,I}(Res[T])= \pi_*\left(T\wedge d\chi \wedge\frac{h^\alpha}{\alpha!} \left(i_{\frac{\partial}{\partial h^I}}\lrcorner dh^d\right)\right)
%\end{eqnarray}
%where the elements $(c_{\alpha,I}(Res[T]))_{\alpha,I}\in \mathcal{D}^\prime_*(I)$ are the \textbf{moments} of Res. 
For $T$ a closed current in $ E_{s}(\mathcal{D}^\prime_k(M\setminus I))$, 
we associated
a current $\mathfrak{Res}[T]\in \mathcal{D}^\prime_*(M)$ supported on $I$.
If $T$ is closed, what can be said about $\mathfrak{Res}[T]$?
\begin{prop}\label{exact}
Let $V$ be some neighborhood of $I$, 
$\pi:V\setminus I\mapsto I$ a submersion 
and $T\in  E_{s}(\mathcal{D}^\prime_k(V\setminus I))$.
If $T\in H_k\left(\mathcal{D}_*\left(V\setminus I\right),d \right)$ is a cycle in the complex of currents and $\pi$ is proper
on the support of $T$ 
then $\mathfrak{Res}[T]\in B_*\left(\mathcal{D}^\prime(M)\right)$.
\end{prop}
\begin{proof}
We first notice that if $T$ is closed then 
$$dRT-R\underset{=0}{\underbrace{dT}}=
d\left(RT\right)=\mathfrak{Res}[T]$$ 
implies  
$\mathfrak{Res}[T]\in B_*\left(\mathcal{D}_{*k}(M),d \right)$ is an exact current.
\end{proof}
Can we relate $\mathfrak{Res}[T]\in \mathcal{D}_\star^\prime(M)$ 
with a current in $\mathcal{D}_*^\prime(I)$ 
in the spirit of the representation theorem (\ref{repsthm})?
The answer is yes but
the naive idea  
to ``restrict'' $\mathfrak{Res}[T]$ 
to the submanifold $I$ does not make sense!
We need another idea which is 
explained in the following example.
\begin{ex}
Let $\delta(h) d^dh$ be the current supported by the point $0$. In this case, $I=\{0\}\subset \mathbb{R}^{d}$.
Then the corresponding current of $\mathcal{D}^\prime(I)$ is just the function $1$, and it can be recovered
by integrating over the ``fiber'' $\mathbb{R}^{d}$, $1=\int_{\mathbb{R}^d} \delta(h) d^dh$. 
\end{ex}
Let $N(I\subset M)$ be the normal bundle of $I$ in $M$. We can identify the closed smooth forms in $H^\star(V,d)$, which are supported in some neighborhood $V$ of $I$ which is \textbf{homotopy retract} to $I$, with
the closed smooth forms in 
$H_v^\star(N(I\subset M),d)$
which have 
compact vertical support 
(see \cite{Bott-Tu} for more on these forms).
The proof is a straightforward application
of the tubular neighborhood theorem
which gives a diffeomorphism
between a neighborhood 
of the zero section 
of $N(I\subset M)$ and $V$ 
and the fact 
that this diffeomorphism  
induces an isomorphism
in cohomology $H_v^\star(N(I\subset M),d)\simeq H^\star(V,d)$.
We denote by $i$ the embedding $i:I\hookrightarrow M$.
For any submersion $\pi:V\setminus I\mapsto I$ and any
current $T\in\mathcal{D}_\star^\prime(V)$ s.t. $\pi$
is proper on its support, the push-forward
$\pi_\star T $ is defined by the formula
$$\forall\omega\in \mathcal{D}(I), \left\langle \pi_\star T,\omega\right\rangle_I=\left\langle T,\pi^\star\omega\right\rangle_M .$$ 
\begin{thm}
Let $V$ be some neighborhood of $I$, 
$\pi:V\setminus I\mapsto I$ a submersion 
and $T\in  E_{s}(\mathcal{D}^\prime_k(V\setminus I))$.
If $T\in H_k\left(\mathcal{D}_*\left(V\setminus I\right),d \right)$ is a cycle in the complex of currents and $\pi$ is proper
on the support of $T$ 
then the push-forward
$$\pi_{\star}\left(\mathfrak{Res}[T]\right)\in B_\star(\mathcal{D}^\prime(I),d).$$
In particular, the current $\mathfrak{Res}[T]\in \mathcal{D}^\prime(M)$ 
is represented by the push-forward of $\pi_{\star}\left(\mathfrak{Res}[T]\right)$:
$\mathfrak{Res}[T]=i_\star\left(\pi_{\star}\left(\mathfrak{Res}[T]\right)\right)$.
\end{thm}
Remark that in this theorem, 
the map $T\mapsto \pi_{\star}\left(\mathfrak{Res}[T]\right)$
is the inverse
of the Leray coboundary 
$\delta$ 
(see \cite{Pham} p.~59--61).

\begin{proof}
Proposition \ref{exact} 
gave us the exactness of $\mathfrak{Res}[T]$. Thus by pull back on the normal bundle,
$\mathfrak{Res}[T]\in B_\star(\mathcal{D}^\prime_\star(N(I\subset M)))$ is exact and supported on the zero section of the normal bundle $N(I\subset M)$. 
Then we pushforward $\mathfrak{Res}[T]$ along the fibers of $\pi:N(I\subset M)\mapsto I$. 
Recall that pushforward $\pi_\star$ commutes with the coboundary operator $d$, 
hence $\pi_{\star}\left(\mathfrak{Res}[T]\right)=\pi_\star d\left(RT\right)=d\pi_\star\left(RT\right)$ by \ref{exact} which
yields the result. 
\end{proof} 
This means that the residue map induces a map on the level of cohomology.

\subsubsection{The residues and symmetries.}
The previous theorem gave us a formula which measured the defect of 
commutativity of the operator $R$ with the coboundary operator $d$. 
Now we study the loss of commutativity of $R$ 
with the operator of Lie derivation $L_X$ for any vector field $X$ such that $[X,\rho]=0$ and $X$ is tangent to $I$ in the sense of H\"ormander (Lemma 18.2.5 in \cite{Hormander} volume 3).
We first notice that the vector space $\mathfrak{g}$ forms an \textbf{infinite dimensional Lie algebra}.   
However, despite the infinite dimensionality 
of this Lie algebra $\mathfrak{g}$, 
it has the following structure:
\begin{prop}
Let $\mathcal{A}\subset C^\infty(M)$ be the subalgebra of the algebra of smooth functions which are killed by $\rho$.
Let us fix a local chart where $I=\{h=0\}\subset \mathbb{R}^{n+d}$ in which the Euler vector field has the form $\rho=h^j\partial_{h^j}$. Then
$\mathfrak{g}$ is a \textbf{finitely generated} left $\mathcal{A}$-\textbf{module} with generators
$h^i\partial_{h^j},\partial_{x^i} $. 
\end{prop}
Any vector field $X$ in $\mathfrak{g}$ is tangent to $I$
thus it decomposes as $a_{i}^j h^i\partial_{h^j}+b^i \partial_{x^i}$
where $a_i^j,b^i$ are smooth
functions by Lemma 18.2.5 in \cite{Hormander}.
Now if $X$ commutes with $\rho$, an elementary computation forces the functions $a_i^j,b^i$ 
to be $\rho$-invariant.

All our symmetries will be Lie subalgebras of $\mathfrak{g}$. 
As usual, we discuss here 
the most 
important example
for QFT which 
comes from our understanding 
of an article of Hollands and Wald
\cite{HW}.
We study the neighborhood 
of the thin diagonal $d_n$ of a configuration space $M^n$ 
where $(M,g)$ is a pseudoriemannian manifold
of dimension $p+1$ and the signature
of $g$ is $(1,p)$. 
By the tubular neighborhood theorem, 
it is always possible to 
identify this neighborhood 
with a neighborhood 
of the zero section of the normal bundle $N(d_n\subset M^n)$.  
Another trick consists in using the exponential map
(see Chapter $5$ section $3$) 
to identify the
normal bundle 
with the metric vector bundle
$\underset{(n-1)\text{ times}}{\underbrace{TM\times_M\dots\times_M M}}$ 
of rank $(n-1)(p+1)$, 
the fiber of this bundle over $x$
is $\underset{n-1}{\underbrace{T_xM\times \dots \times T_xM}}$
which has a canonical metric $\gamma_x$
of 
signature $n-1,(n-1)p$. 
Then the 
Lie algebra of infinitesimal
gauge transformations
of this vector bundle is
the suitable Lie algebra of symmetries.

\begin{ex}
Let $\pi:(P,\gamma)\mapsto I$ be a metric vector bundle of rank $d$ with metric $\gamma$ on the fibers
(in the Hollands Wald discussion $P$ is the normal bundle $N(d_n\subset M^n)$ and $d=(n-1)(p+1)$). 
We construct a trivialisation of $P$ 
by the moving frame technique. Let $U\subset I$ be an open set. Let $(e_0,...,e_n)$ be \textbf{an orthonormal moving frame} ($\forall x\in U , \gamma_x(e_\mu,e_\nu)=\eta_{\mu\nu}$)
and let
$$ (x,h):\pi^{-1}\left(U\right)\rightarrow U\times \mathbb{R}^{d}  $$   
$$(p,v) \mapsto \left(x(p),h(p,v) \right) $$
such that $v=\sum_0^d h^\mu(p,v)e_\mu(p)$, for $p\in U$ and $v\in \pi_p^{-1}(U)$. We use the coordinate system $(x,h)$ on $P$. 
All orthonormal moving frames are related by gauge transformations which are maps in $C^\infty(I,O(\eta))$ where $O(\eta)$ is the orthogonal group of the quadratic form $\eta$. The gauge group $C^\infty(I,O(\eta))$ is a subgroup of the group of diffeomorphism of the total space $M$ preserving the zero section $\underline{0}$ (the zero section $\underline{0}$ being isomorphic to $I$). The Euler vector field $\rho=h^j\frac{\partial}{\partial h^j}$ 
which scales linearly in the fibers w.r.t. the zero section $\underline{0}$ is canonically given and the gauge Lie algebra consists of vector fields  of the form $\underset{\text{antisymmetric}}{\underbrace{a_{\mu\nu}}(x)}\left(h^\mu\partial^\nu_{h} \right)$, where $\forall \nu, \partial^\nu_h=\gamma^{\mu\nu}\partial_{h^\mu}$,
hence $\left(h^\mu\partial^\nu_{h} \right)-\left(h^\nu\partial^\mu_{h} \right)$
is an infinitesimal generator of
the Lie algebra $o(\eta)$
which commutes with $\rho$ and vanishes at $\underline{0}$.
\end{ex}

Before we state and prove the 
residue theorem for vector fields with symmetries, let us pick again our simplest fundamental example 
(again due to Laurent Schwartz) to illustrate the anomaly phenomenon:
\begin{ex}
The Heaviside current $T=H(x)dx$ is smooth in $\mathbb{R}\setminus \{0\}$ 
and satisfies the symmetry equation $L_{\partial_x}T=0$
on $\mathbb{R}\setminus \{0\}$,
i.e. it is translation invariant outside the singularity.
Again, let $R$ be 
the extension operator, 
recall the extension $RT$ is unique 
for this example
and again by integration by parts,
we obtain the residue equation:
$$L_{\partial_x} \left(RT\right) - R\left(L_{\partial_x} T\right)=L_{\partial_x}\left( RT\right)=\delta_0 dx.$$
\end{ex} 
Recall $E_{s}\left(\mathcal{D}_k^\prime(M\setminus I)\right)$ 
is 
the space of $k$-currents 
in $\mathcal{D}^\prime_k(M\setminus I)$ 
which are weakly homogeneous of degree $s$.
For any
vector field $X\in\mathfrak{g}$, 
we denote by 
$L_X$ the operator of Lie derivation. 
We define 
the residue
of $T$
w.r.t. 
the 
vector field $X\in\mathfrak{g}$
as the current
defined by the equation:
\begin{equation}
\mathfrak{Res}_X[T]=L_X\left(RT\right)-R\left(L_XT\right).
\end{equation}

\begin{thm}
Let $ T\in E_{s}\left(\mathcal{D}_k^\prime(M\setminus I)\right)$,
$p=\sup(0,k-n)$ and $X\in\mathfrak{g}$.
If $p+s\leqslant 0$, let $m\in\mathbb{N}$ s.t. 
$-m-1<s\leqslant -m$, then 
we have the residue equation:
\begin{eqnarray}
\mathfrak{Res}_X[T](\omega)=(-1)^{n-k-1}\left\langle i_X\left(T\wedge P_m(\omega)\right), d\chi \right\rangle,
\end{eqnarray}
where $i_X$ denotes contraction of the current 
$\left(T\wedge P_m(\omega)\right)$ 
with
the vector field $X$. Note that
$\mathfrak{Res}_X[T](\omega)$ is a local counterterm 
in the sense it is a \textbf{current supported on }$I$. 
%and it depends
%only on the restriction
%on the submanifold $I$
%of \textbf{finite jets} 
%of the vector field $X$.

% For the regularization operator $\left\langle RT,\omega\right\rangle=\lim_{\varepsilon\rightarrow 0}\left\langle T,(1-\chi)\omega\right\rangle+\left\langle T (\chi-\chi_{\varepsilon^{-1}}) ,I\right\rangle   $. 

\end{thm}
The proof is exactly the same as in Theorem \ref{residuethm}, just replace the boundary operator $d$ by $L_X$ and we obtain $\mathfrak{Res}_X[T]=(-1)\left\langle T\left(L_X\chi\right), P_m(\omega)  \right\rangle$.
Then we use exterior differential calculus to convert
this expression 
$$\left\langle T\left(L_X\chi\right), P_m(\omega)  \right\rangle=\left\langle T i_X d\chi, P_m(\omega)  \right\rangle$$ 
$$=\left\langle T \wedge P_m(\omega),i_X d\chi  \right\rangle=(-1)^{n-k-1}\left\langle i_X\left(T \wedge P_m(\omega)\right), d\chi  \right\rangle .$$
 
\subsection{Stability of geometric residues.}

 Now the natural questions we should ask ourselves are: what are the conditions for which the residue vanishes ? Is the residue independent of $\chi$ ? In general, we would like to know what are the stability properties of residues.  
In the case of symmetries, what should replace the closed or exact currents in the De Rham complex of currents ? 
 
 There is a cohomological analogue of the De Rham complex in the case of symmetries generated by infinite dimensional Lie algebras of vector fields on $M$ denoted by $\mathfrak{g}$. This is the theory of continuous cohomology of infinite dimensional Lie algebras developped by I M Gelfand and D Fuchs.  
Fortunately for us, we only need basic definitions of this theory following \cite{Fuchs}.
For any left $\mathfrak{g}$-module $\mathcal{M}$, we define the complex (\cite{Fuchs} Chapter 1,  ``The standard chain complex of a Lie algebra'', p.~137,138)
$$C^k(\mathfrak{g},\mathcal{M})=Hom\left(\bigwedge^k\mathfrak{g},\mathcal{M} \right)  $$
with the differential
$\delta :C^k(\mathfrak{g},\mathcal{M})\rightarrow C^{k+1}(\mathfrak{g},\mathcal{M}) $
which for $k=0$ reads
$$\delta \Theta(X)=L_X\Theta,$$
$\Theta\in C^0(\mathcal{M})
\simeq \mathcal{M}$ 
and $L_X$ 
denotes the 
left action of $X$ on the module $\mathcal{M}$.
$\left(C^\bullet(\mathfrak{g},\mathcal{M}),\delta\right)$ is called \emph{the standard cochain complex of the Lie algebra} $\mathfrak{g}$ with coefficient in the module $\mathcal{M}$. Now, the choice of topological module $\mathcal{M}$ dictated by our problem is the space of currents $\mathcal{D}^\prime_*(M)$ with the natural weak topology defined on it and the left action of $\mathfrak{g}$ on $\mathcal{D}^\prime_*(M)$ is the action by \textbf{Lie derivatives}.
Then without surprise, 
the formula for $\delta$ is the classical Cartan formula in differential geometry.
The Lie algebra of smooth vector fields on $M$ has a natural $C^\infty$ topology, this topology induces on $\mathfrak{g}$ a $C^\infty$ topology:
the space of smooth vector fields
is endowed with the topology
of $C^\infty$ convergence of the components and some finite number of derivatives over compact sets. Then we require our cochains $T\in C^\star(\mathfrak{g},\mathcal{M})=Hom\left(\bigwedge^\star\mathfrak{g},\mathcal{M} \right)$ to
be 
continuous 
for 
the $C^\infty$ topology of 
$\mathfrak{g}$
and the weak topology of $\mathcal{M}$.  
% We introduce another module $\mathcal{M}_I$ of currents supported on $I$.
\begin{thm}\label{Stabilitythm}
Let $T\in E_s\left(\mathcal{D}_k^\prime(M\setminus I)\right)$ and 
$\omega\in\mathcal{D}^k(M)$.
If $\exists X\in\mathfrak{g}$ such that $L_X\left(T\wedge P_m(\omega) \right)=0$, then
for all smooth closed forms 
$[C]\in H^1\left(\left(\Omega^*\left(M\setminus I \right),d \right) \right)$ 
such that
$[C]=[-d\chi]$, we have the identity  
\begin{equation}
\mathfrak{Res}_X[T](\omega)=(-1)^{n-k}\left\langle i_X\left(T\wedge P_m(\omega)\right), [C] \right\rangle
\end{equation} 
and $\mathfrak{Res}_X[T](\omega)$ is a \textbf{period}. 
\end{thm}
\begin{proof}
If $T$ is a current in $\mathcal{D}_k^\prime(M\setminus I)$
and $\omega\in \mathcal{D}^k(M)$ is a test $k$-form, 
then
the Taylor polynomial $P_m(\omega)\in \Omega^k(M)$ 
is also a smooth $k$-form but is no longer compactly supported. 
Thus the exterior product $T\wedge P_m(\omega)$ 
is well defined as a current in $\mathcal{D}^\prime_{0}(M\setminus I)$ (\cite{Schwartz} p.~341).
Currents in $\mathcal{D}^\prime_{0}(M\setminus I)$
are similar to \textbf{forms
of maximal degree} and are always closed, 
thus
$T\wedge P_m(\omega)$ is closed on $\text{supp }d\chi\subset \left(M\setminus I\right)$. 
But from the Lie Cartan formula for currents (\cite{Schwartz}), $0=L_X\left(T\wedge P_m(\omega)\right)=(i_X d+di_X)\left(T\wedge P_m(\omega)\right)=di_X\left(T\wedge P_m(\omega)\right)$ because $T\wedge P_m(\omega)$ is closed. 
We find that $di_X\left(T\wedge P_m(\omega)\right)=0$ which means $i_X\left(T\wedge P_m(\omega)\right)$ is a closed curent 
and $\mathfrak{Res}_X[T](\omega)$ 
is the \textbf{period}
of the \emph{closed form} $d\chi$ 
relative to \emph{the cycle} 
$i_X\left(T\wedge P_m(\omega)\right)$ 
in the sense of Hodge and De Rham 
(see \cite{Rham} p.~135 and \cite{Giaquinta} p.~585). 
\end{proof}
\begin{coro}\label{nondep}
Under the assumptions of Theorem \ref{Stabilitythm},
$\mathfrak{Res}_X[T](\omega)$ \textbf{does not depend} on the choice of $\chi$. 
%and:
%\begin{equation}
%\mathfrak{Res}_X[T](\omega)=\lim_{\varepsilon\rightarrow 0}-\left\langle i_X\lrcorner\left(T\wedge P_m(\omega)\right), (d\chi)_{\varepsilon^{-1}} \right\rangle.
%\end{equation}
\end{coro}
\begin{proof}
$\mathfrak{Res}_X[T](\omega)$ 
does not depend on the choice of $\chi$ because if $\chi_1,\chi_2$ are two smooth functions such that $\chi_i=1$ in a neighborhood of $I$, then $\chi_1-\chi_2=0$ in a neighborhood of $I$, thus $[d\chi_1]-[d\chi_2]=[d(\chi_1-\chi_2)]=0$.   
\end{proof}
\begin{thm}
Let $T\in E_s\left(\mathcal{D}_k^\prime(M\setminus I)\right)$ and 
$\omega\in\mathcal{D}^k(M)$.
If $\exists X\in\mathfrak{g}$ such that $L_X\left(T\wedge P_m(\omega) \right)=0$,
then
$\mathfrak{Res}_X[T]$ 
is local in the sense 
it is a \textbf{current supported on }$I$ 
and it depends
only on the restriction
on $I$
of \textbf{finite jets} 
of the vector field $X$.
\end{thm}
\begin{proof}
 To prove the locality in the vector field $X$, 
the key point
is to notice that
$\forall\varepsilon>0, [d\chi]=[d\chi_{\varepsilon^{-1}}]$
in $H^1(M\setminus I)$ since
$d\chi-d\chi_{\varepsilon^{-1}}=d(\chi-\chi_{\varepsilon^{-1}})$
where $(\chi-\chi_{\varepsilon^{-1}})\in C^\infty(M\setminus I)$
vanishes in a neighborhood of $I$
thus $$\forall\varepsilon>0,
\mathfrak{Res}_X[T](\omega)=(-1)^{n-k}\left\langle i_X\left(T\wedge P_m(\omega)\right), [-d\chi_{\varepsilon^{-1}}] \right\rangle .$$
Since $T\wedge P_m(\omega)$
is a distribution in $\mathcal{D}_{0}^\prime(M\setminus I)$
we can assume it is a distribution of order $m_i$ on each open ball
$U_i$ of a given cover $(U_i)_i$ of $M$.
Let $(\varphi_i)_i$ be a partition of
unity subordinated to the cover $(U_i)_i$.
Then we decompose the duality coupling:
$$ \left\langle T\wedge P_m(\omega), L_X\chi\right\rangle=\sum_i \left\langle T\wedge P_m(\omega), \varphi_i L_X\chi\right\rangle $$
On each ball $U_i$, the distribution $T\wedge P_m(\omega)$ can be represented as a continuous linear form $\ell_i$ 
acting on the $m_i$-jet of $\varphi_i L_X\chi$ (this is the structure theorem of Laurent Schwartz for distributions \cite{Schwartz})
$$\left\langle T\wedge P_m(\omega), L_X\chi\right\rangle=\sum_i \ell_i\left( j^{m_i}(\varphi_i L_X\chi)\right) $$
Hence we deduce from this result that $\mathfrak{Res}_X[T]$ depends locally on finite jets of $X$.
We can conclude by taking the limit
$$\left\langle T\wedge P_m(\omega), L_X\chi\right\rangle=\lim_{\varepsilon\rightarrow 0}\left\langle T\wedge P_m(\omega), L_X\chi_{\varepsilon^{-1}}\right\rangle$$ 
$$=\lim_{\varepsilon\rightarrow 0}\sum_i \ell_i\left( j^{m_i}(\varphi_i L_X\chi_{\varepsilon^{-1}})\right) $$
which localizes the dependence on the jets of $X$ restricted on $I$.  
\end{proof}
 
 We know that $\mathfrak{Res}_X[T]$ is a local coboundary supported on $I$, but we don't know if 
$\mathfrak{Res}_X[T]$ is the coboundary of a cochain supported on $I$.
We prove a theorem
which
gives a 
cohomological
formulation
of the existence
of a 
$\mathfrak{g}$-invariant 
extension
of the current $T$
in terms of the residue
of the extension $R$.
\begin{thm} 
Let $T\in E_s\left(\mathcal{D}^\prime_0(M\setminus I)\right)$ and $T$
is $\mathfrak{g}$ invariant i.e.
$\forall X\in \mathfrak{g},L_XT=0$. Then 
there exists an extension 
$\overline{T}$ of $T$ which is $\mathfrak{g}$-invariant
\textbf{if and only if} 
$X\mapsto\mathfrak{Res}_X[T]$ 
is the $1$-coboundary 
of a current supported on $I$.
\end{thm} 

\begin{proof}
We just follow the definitions. 
We view the map $X\mapsto RT$ 
as an element 
in $C^0(\mathfrak{g},\mathcal{M})$
because it does not depend 
on $\mathfrak{g}$. 
Then $\Theta=\delta RT$ is the coboundary of $RT$.
Let $\overline{T}$ be a $\mathfrak{g}$ invariant extension of $T$. Then $c=\overline{T}-RT$ is a current supported by $I$. 
 $$\forall X\in\mathfrak{g}, L_Xc=L_X\left( \overline{T}-RT\right)=-L_XRT $$  
because $L_X \overline{T}=0$. But this means that we were able to write $\Theta$ as minus the coboundary
of the cochain $c$ supported on $I$.
Conversely, if $\Theta$ is the coboundary of a local cochain $c$ supported on $I$, then setting $\overline{T}=RT-c$ gives a $\mathfrak{g}$-invariant extension of $T$.
\end{proof}

\paragraph{Anomalies in QFT and relation with the work of Costello.}

 The author wants to stress that the suitable language to speak about anomalies in QFT 
is to write them as cocycles for the Lie algebra $\mathfrak{g}$ of symmetries with value in a
certain module $\mathcal{M}$ 
which depends on the formalism 
in which we work. Usually, the Lie algebra $\mathfrak{g}$ is infinite dimensional.

 In recent works of Kevin Costello, anomalies appear under the form of a \textbf{character} $\chi$ and 
constitute 
a central extension of the Lie algebra $\mathfrak{g}$ of symmetries, this is the content of the ``Noether theorem'' for factorization algebras discovered by Costello Gwilliam. They also require that this cocycle be local ie the cocycle $\chi$ is bilinear in $\mathfrak{g}$ with value in the module $\mathcal{M}$ and is represented by integration against a Schwartz kernel.

$$\chi(X_1,X_2)=\int_{M^2} \left\langle \chi(x_1,x_2),X_1(x_1) \otimes X_2(x_2) \right\rangle $$ 
where $\chi(x_1,x_2)$ is \textbf{supported on the diagonal} $d_2\subset M^2$.
In our work, we exhibit a purely \textbf{analytic way} 
to produce such local cocycles as residues.
The residue $\mathfrak{Res}_X[T]$ 
is \textbf{local} 
in the sense 
it is a \textbf{current supported on }$I$ 
and it depends
only on the restriction
on the submanifold $I$
of \textbf{finite jets} 
of the vector field $X$.

\chapter{The meromorphic regularization.}
\section{Introduction.}
\paragraph{The plan of the chapter.}

 In this part, we would like to revisit the theory of meromorphic regularization using the techniques of chapter $1$.
We will show the advantages of the
continuous partition of unity over the dyadic methods
because it allows us to define an extension of distributions, that we call Riesz extension, 
using meromorphic techniques
as in the ``dimensional regularization''
used in physics textbooks.
The first step is to define some suitable space of distributions on which we can apply the meromorphic regularization procedure. It was suggested to the author by L Boutet de Monvel that such 
spaces 
are the spaces of distributions having asymptotic expansions with moderate growth in the transversal directions to $I$.

  Given the canonical Euler vector field $\rho$, we define a simple notion of constant coefficient Fuchsian differential equation and first order Fuchsian system $P$, the solutions $t$ of the constant coefficient Fuchsian systems are vectors with distributional entries. 
For instance a Fuchsian operator $P$ in the vector case is of the form $P=\rho-\Omega $  
where $\Omega$ is a constant square matrix. These Fuchs operators are adaptation of the concept of Fuchsian systems appearing in complex analysis. We first motivate the reason why we have to introduce asymptotic expansions in the space of distributions and the relationship with Fuchsian systems.   
  
\subsubsection{QFT example of $\Delta_+$ and motivations.}

 In curved space times, the Hadamard states $\Delta_+(x,y)$ viewed as a two point distribution in $\mathcal{D}^\prime\left(M^2\right)$ is not an exact solution of any \textbf{constant coefficient} Fuchsian equation that would come to our mind. Actually, we would like to study $\Delta_+$ and its powers $\Delta_+^k$.

 For the Euler vector field $\rho=\frac{1}{2} \nabla_x \Gamma$
we have the following asymptotic expansion of $\Delta_+$:
\begin{equation}
\Delta_+=\sum_{n=0}^\infty U_n\Gamma^{-1}+V_n\log\Gamma + W_n
\end{equation}
where $U_n,V_n,W_n$ are homogeneous of degree $n$ wrt $\rho$.
\begin{prop}
Let $\Delta_+$ be the Hadamard parametrix and $\rho=\frac{1}{2} \nabla_x \Gamma$, then
$\Delta_+$ satisfies the equation:
\begin{equation}
(\rho+2)(\rho+1)\rho^2\Delta_+ \in E_{0}.
\end{equation}
\end{prop}
\begin{proof}
Notice that if $U_n$ is homogeneous
of degree $n$ since $\Gamma^{-1}$ is homogeneous 
of degree $-2$ then we must have 
$(\rho-n+2)U_n\Gamma^{-1}=0$
and also 
$\rho V_n\log\Gamma=nV_n\log\Gamma+2V_n$
which implies 
$(\rho-n)^2V_n\log\Gamma=2(\rho-n)V_n=0$.
We deduce the system
of equations:
\begin{eqnarray}
(\rho+2)U_0\Gamma^{-1}=0 \\
(\rho+1)U_1\Gamma^{-1}=0 \\
\rho U_2\Gamma^{-1}=0 \\
\rho^2 V_0\log\Gamma=0. 
\end{eqnarray}
Thus if we act on $\sum_{n=0}^\infty U_n\Gamma^{-1}+V_n\log\Gamma + W_n$
by the differential operator
$(\rho+2)(\rho+1)\rho^2$,
the above system
of equations
shows that
we will kill all singular terms
in the sum $\sum_{n=0}^\infty U_n\Gamma^{-1}+V_n\log\Gamma + W_n$.
\end{proof}

From this typical quantum field theoretic example, we understand that it is not possible to find \textbf{constant coefficients} Fuchsian operators that kills exactly the Feynman amplitudes. However, we can kill them with constant coefficients Fuchsian operators modulo an \emph{error term which lives in nicer space} and go on successively.
We define the space $F_\Omega$ of $\textbf{Fuchsian symbols}$ which consists of distributions $t$
having 
asymptotic expansions of the form
$t=\sum_0^\infty t_k$ i.e.
$\exists s\in\mathbb{R}, \forall N, t-\sum_0^N t_k \in E_{s+N}$,
where we used 
the property that the scale spaces $E_s$ are filtered, 
$s^\prime\geqslant s\implies E_{s^\prime}\subset E_s$.
Intuitively, we would say that these are spaces of distributions which are killed by constant coefficients Fuchsian operators modulo an error term which can be made ``arbitrarily nice'', 
the price to pay for a nice error term is that we must use constant coefficients Fuchsian operators   
of arbitrary order.

\paragraph{The meromorphic regularization and the Mellin transform.}

 We modify the extension formula of H\"ormander $\int_0^1 d\lambda \lambda^{-1} t\psi_{\lambda^{-1}}+(1-\chi)t$ and define a regularization of the extension $t^\mu=\int_0^1 d\lambda\lambda^{\mu-1} t\psi_{\lambda^{-1}}$ depending on a parameter $\mu$. We relate the new regularization formula to the Mellin transform. The idea actually goes back to Gelfand who considered Mellin transform of functions averaged on hypersurfaces (see \cite{Gelfand} (4.5) Chapter 3 p.~326 and \cite{AVG} (7.2.1) p.~218).
When $t\in E_s(U\setminus I)$, we prove that $t^\mu$ has an extension in $E_{s+\mu}$ and is holomorphic in $\mu$ for $Re(\mu)$ large enough, intuitively, when $Re(\mu)$ is large enough the integral $\int_0^1 d\lambda\lambda^{\mu-1} t\psi_{\lambda^{-1}}$ has better chances to converge. 
Moreover, we can already prove that if there is any meromorphic extension $\mu\mapsto t^\mu$, then \textbf{the tail of the Laurent series must be local counterterms}.
 Now if we know that $t\in F_\Omega$, which is a much stronger assumption than $t\in E_s$, we then establish a nice identity satisfied by the regularized extension
\begin{eqnarray}\label{identitymeromcontinuation}
\forall N, \left\langle T^\mu,\varphi \right\rangle =\sum_{j\leqslant N}(\mu+j+\Omega)^{-1} \left\langle (T\varphi)_{j} ,\psi    \right\rangle+\left\langle \left(I_N(T\varphi)\right)^\mu,\psi\right\rangle ,
\end{eqnarray}
where $ I_N(T\varphi)
=\frac{1}{N!}\int_0^1 ds(1-s)^N\left(\frac{\partial}{\partial s}\right)^{N+1}s^{-\Omega}\left(T\varphi\right)_s
$ is the remainder of the expansion $\left(T\varphi\right)_s=\sum_{j\leqslant N}s^{j+\Omega}(T\wedge\omega)_j+I_N(T\varphi)_s$,
and we prove that the regularization $\mu\mapsto t^\mu$ can be extended meromorphically in $\mu$ with poles located in $Spec\left(\Omega\right)+\mathbb{N}$. 
We write explicit formulas for the poles of $t^\mu$. 
\paragraph{The Riesz extension.}
To go back to the \textbf{interesting case}, we have to take the limit of $t^\mu$ when $\mu=0$. However, if $\mu=0$ is a pole of finite order of $t^\mu$, then we must remove the tail of the Laurent series which are 
\textbf{local counterterms}, i.e. distributions supported on $I$.
Then we will prove that the operation of meromorphic regularization then removing the poles at $\mu=0$ and finally taking the limit $\mu\rightarrow 0$ defines an \textbf{extension operation} which is called the Riesz extension and is a specific case of all the extensions defined in Chapter 1.  
Then we will show that the Fuchsian symbols renormalized  
by the Riesz extension are still Fuchsian symbols.   
Finally, we will explain
how to introduce a length scale
$\ell$
in the Riesz extension and how the one parameter
renormalization group
emerges in this picture and involves
only polynomials of $\log\ell$.  
%\paragraph{The construction of the residue by three point of views.}
%The poles which appeared in the process of meromorphic regularization will be related to the \textbf{residues} of Chapter 8.
%In general, we propose two different point of views on the theory of residues, the \textbf{current theoretic} point of view of \cite{Schwartz} and \cite{Griffiths} and 
%a more complex analytic point of view a la Gelfand--Shilov \cite{Gelfand} as \textbf{certain poles} appearing in the meromorphic regularization.
\paragraph{Relationship to other works.}
In this Chapter,
we give general definitions
of Fuchsian symbols
which are adapted
to QFT in curved
space times as
we illustrated
in our example.
To our knowledge,
these definitions
were first given 
by Kashiwara--Kawai
\cite{KK}.
They also appear
in the work
of Richard Melrose
\cite{Melroseasympt}.
We undertake
the task of 
meromorphic
regularizing
Fuchsian symbols
which are asymptotic expansions
of a more
general nature
than associate
homogeneous distributions.
\section{Fuchsian symbols.}
In QFT, scalings 
of distributions is not
necessarily homogeneous,
there are $\log$ terms. 
Distributions encountered in QFT
are not solutions of equations
of the form
$(\rho-d)t=0$
but they might be 
solutions
of equations
of the form
$(\rho-d)^nt=0$. 
We work in flat space $\mathbb{R}^{n+d}$ with coordinates $(x,h)\in\mathbb{R}^n\times\mathbb{R}^d$ and where $I=\{h=0\}$.
The scaling is defined by the Euler vector field $\rho=h^j\partial_{h^j}$. 
\subsection{Constant coefficients Fuchsian operators.}

Given the canonical Euler vector field $\rho$, 
we give a simple definition of a \textbf{constant coefficient} Fuchsian differential operator of order $n$:
\begin{defi}
A \textbf{constant coefficient Fuchsian operator} of degree $n$ is an operator of the form $b(\rho)$ where
$b\in \mathbb{C}[X]$ is a polynomial of degree $n$
with \textbf{real roots}.  
\end{defi}  
In QFT, these roots will often be integers.
\begin{ex}
Consider the one variable case where $\rho=h\frac{d}{dh}$. 
The monomial $h^d$ is solution of the equation
$(\rho-d)h^d=0$, hence $b(X)=(X-d)$.
On the other hand
$\log h$ is solution of the equation
$\rho^2\log h=0$ hence $b(X)=X^2$.
Lastly,
$h^d\log h$ is solution of the
equation
$(\rho-d)^2h^d\log h=0$.
\end{ex}  
Next define first order \textbf{constant coefficient} 
Fuchsian operators of rank $n$:
\begin{defi}
A Fuchsian system of rank $n$ is a differential operator of the form $P=\rho-\Omega $  
where $\Omega=\left(\omega_{ij}\right)_{1\leqslant ij\leqslant n}\in M_n(\mathbb{C}) $ is a \textbf{constant} $n\times n$ matrix
with real eigenvalues.     
\end{defi}  
\begin{ex}
The column 
$\left(\begin{array}{c} \log h \\
1 \end{array}\right)$
is solution
of the system
$$\rho \left(\begin{array}{c} \log h \\
1 \end{array}\right)=\left(\begin{array}{cc} 0 & 1 \\
0 & 0\end{array}\right)\left(\begin{array}{c} \log h \\
1 \end{array}\right) $$
\end{ex}
Let $U$ be an arbitrary open domain which is $\rho$-convex.
For $b$ a $n$-th order operator (resp $P=\rho-\Omega$ a system), we give a fairly general definition of some new subspaces $F_b(U)$ (resp $F_\Omega(U)$) which are associated to the differential operators $b$ (resp $P$) and which are different from the space $E_s(U)$ defined by Yves Meyer.
However their definition uses the spaces $E_s(U)$ defined by Meyer.
We define the space 
$F_b(U)$ of $\textbf{Fuchsian symbols}$ 
associated to a Fuchsian operator $b$: 
\begin{defi} 
Let $b(\rho)$ be a constant coefficients Fuchsian differential operator of order $n$. 
Then the space $F_b(U)$ of Fuchsian symbols is defined as the space distributions 
$t$ s.t. there exists 
some neighborhood $V$ of $I\cap \overline{U}$ 
and a sequence $(t_k)_k$ of distributions
such that
\begin{eqnarray} 
\forall N, t=\sum_{k=0}^N t_k+ R_N
\\ \forall k, b(\rho-k)t_k|_V=0
\end{eqnarray}
where $\forall N, R_N\in E_{s+N+1}(U)$, $s=\inf Spec(b)$.
\end{defi}
\begin{ex}
Let us consider the series $\sum_{k=0}^\infty a_k h^{d+k}$, then each term $a_k h^{d+k}$ is killed
by the operator $(\rho-d-k)$.
\end{ex}
\begin{defi}
 Let $ \Omega=\left(\omega_{ij}\right)_{1\leqslant ij\leqslant n}\in M_n(\mathbb{C}) $ be a $n\times n$ matrix and $P=\rho-\Omega$ be a Fuchsian operator of first order and rank $n$. Then the space of Fuchsian symbols $F_\Omega(U)$ 
is the space of vector valued distributions $t=(t_i)_{1\leqslant i\leqslant n}$
such that there exists 
some neighborhood $V$ of $I\cap \overline{U}$ 
and a sequence $(t_k)_k$ of distributions
such that
\begin{eqnarray} 
\forall N, t=\sum_{k=0}^N t_k+ R_N
\\ \forall k, \left(\rho-(\Omega+k)\right)t_k|_V=0
\end{eqnarray}
where $\forall N, R_N\in E_{s+N+1}(U)$, $s=\inf Spec(\Omega)$.
\end{defi}
\paragraph{Some remarks on scalings.}
Assume $t\in F_\Omega$.
Notice that for all test functions $\varphi$, the function $\lambda\mapsto\lambda^{-\Omega}\left\langle t_\lambda,\varphi \right\rangle$ is smooth
in $(0,1]$ since $\left\langle t_\lambda,\varphi \right\rangle=\lambda^{-d}\left\langle t,\varphi_{\lambda^{-1}} \right\rangle$ and has a \textbf{unique asymptotic expansion} at $\lambda=0$,
$$\lambda^{-\Omega}\left\langle t_\lambda,\varphi \right\rangle\sim \sum_{k=0}^\infty \lambda^k \left\langle t_k,\varphi \right\rangle  .$$
But this does not mean that $\lambda\mapsto \lambda^{-\Omega}\left\langle t_\lambda,\varphi \right\rangle$ is smooth at $\lambda=0$
as the following counterexample illustrates:
\begin{ex}\label{exHeleinmoi}
The function $f(\lambda)=e^{\frac{-1}{\lambda^2}}\sin(e^{\frac{1}{\lambda^2}})$ has asymptotic expansion $e^{\frac{-1}{\lambda^2}}\sin(e^{\frac{1}{\lambda^2}})\sim 0$ and is smooth in $(0,1]$, however it 
is not smooth in $[0,1]$ since the first derivative of this function does not converge to zero when
$\lambda\rightarrow 0$. 
\end{ex}
However, we have a condition
which implies the smoothness
on $[0,1]$:
\begin{lemm}
Let $\lambda\mapsto f(\lambda)$ be a function which is smooth on $(0,1]$ and which has an asymptotic expansion
at $\lambda=0$. Then if $\forall n$, $f^{(n)}$ has asymptotic expansion at $0$ which is obtained by formally differentiating $n$ times the expansion of $f$ then $f$ extends smoothly at $\lambda=0$.
\end{lemm}
The proof can be found in \cite{Godement} lemme 1 p.~120.

We want to remind the reader there is a standard way to go from Fuchsian differential operators of order $n$ to $1$st order Fuchsian systems of rank $n$, this is called the companion system (see \cite{Yakovenko} 19B p.~332, 19E p.~342 for this classical construction).
\paragraph{Asymptotic expansions.}

 We explain the connection with asymptotic expansions of distributions.
\begin{defi} 
The distribution $t$ 
admits an asymptotic expansion if $t\in E_s(U)$ and
if 
there exists a strictly increasing 
sequence of real numbers $(s_i)_i$ such that $s\leq s_0$   
\begin{eqnarray}
\exists (t_i)_i, t_i \in E_{s_i}(U)
\\ \forall N, \left(t-\sum_{i=1}^{N} t_i\right) \in E_{s_{N+1}}(U)
\end{eqnarray}
\end{defi}

 In concrete applications, the sequence $(s_i)_i$ 
is equal to $A+\mathbb{N}$
where 
$A$ is a finite set of real numbers.
So we see that our space of Fuchsian symbols
is just a subspace of the space
of distributions having asymptotic expansions.
However, these spaces are less general than 
the spaces $E_s$ defined by Yves Meyer
as we shall illustrate in the following example
\begin{ex}
$\sin(\frac{1}{x})$ is weakly homogeneous of degree $0$ on $\mathbb{R}$, thus it lives in $E_0(\mathbb{R})$. However, it admits no asymptotic expansion ! 
\end{ex}
We want to insist on the fact that our spaces $F_\Omega$
are defined in the smooth category
and does not require any analyticity
hypothesis.

\subsection{Fuchsian symbols currents.}

 For a given Fuchsian operator $P=\rho-\Omega$ of first order and rank $n$,
$F_\Omega(U)$ is the space of vector valued currents $T$ such that there exists a sequence $(T_k)_k$ of distributions such that in a certain neighborhood $V$ of $I\cap U$
 \begin{eqnarray} 
\forall N, T=\sum_{k=0}^N T_k+ R_N
\\ \forall k, \left(\rho-(\Omega+k)\right)T_k=0
\end{eqnarray}
where $\forall N, R_N\in E_{s+N+1}(U)$, $s=\inf Spec(\Omega)$.
Recall also that we are able to decompose test forms $\omega$ as a sum 
$$\omega=\sum_{n=0}^m \omega_n + I_m(\omega) $$
where the $\omega_n$ are homogeneous of degree $n$.
 
 Notice that for any compactly supported test form $\omega$, the exterior product $T\wedge \omega$ is a Fuchsian symbol and $T_k\wedge\omega_n$ satisfies the following exact equation:
\begin{equation} 
\rho \left(T_k\omega_{n}\right)= \left(n+k+\Omega\right)T_k\omega_{n}. 
\end{equation}

\paragraph{On the relationship with 
the standard notion
of Fuchsian differential equations.}
The
theory
of Fuchsian
differential
equation
has an old story
which goes back
to great names
such as
Poincar\'e,
Riemann and 
Fuchs.
More recently,
there was
a resurgence
of activities
around these
equations in the context of
PDE's
with famous 
works 
in analysis 
by 
Malgrange,
Kashiwara,
Leray, Pham.
Some very nice
surveys and textbooks
now exist
on the subjects,
and our work
is
particularly 
inspired
by 
(\cite{Pham,Zoladek,AVG,Yakovenko,YN})
which give very nice expositions
of this topic.
Distributions
solution
to Fuchsian 
differential
operators
have several names.
They were called
`associate
homogeneous
distributions''
by \cite{Gelfand}.
 
These distributions 
are also 
called
``hyperfunctions of the Nilsson Class''
by Pham \cite{Pham},
for instance 
a similar proof
of Proposition (3.2) p.~18 
in \cite{TNS}
can be found in \cite{Pham} p.~153,154.

\subsection{The solution of a variable coefficients Fuchsian equation is a Fuchsian symbol.}

 The idea is that we want to deal with perturbations of the Euler equation $(\rho-\Omega)t=0$ where $\Omega$ is a constant matrix. Let $\mathcal{I}\subset C^\infty(M)$ denote the ideal of smooth functions vanishing on $I$.
Let $\tilde{\Omega}$ be a perturbation of $\Omega$:  $\tilde{\Omega}-\Omega\in M_n\left(\mathcal{I} \right)$,
note that this implies $\tilde{\Omega}|_I$ is constant
and equals $\Omega$. 
We are then 
able to prove that solutions of the Fuchsian operator with variable coefficients $P=\rho-\Omega$ are Fuchsian symbols. The space of Fuchsian symbols is thus the natural space of solutions of perturbed Euler equation.

Let us work in a local chart in $\mathbb{R}^{n+d}$ with coordinates $(x,h)$ where $I=\{h=0\}$
and $\rho=h^j\frac{\partial}{\partial h^j}$. Let $P=\rho-\tilde{\Omega}$ where $\tilde{\Omega}\in \Omega + M_n\left(\mathcal{I} \right)$ and $\rho-\Omega $  
is a first order Fuchsian system of rank $n$ with constant coefficients.

For any complex number $\lambda$ and matrix $\Omega$,
we define $\lambda^\Omega$ by the equation
$$\lambda^\Omega=\exp(\log\lambda\Omega) $$
for the branch $0\leqslant \arg\log<2\pi$ of the logarithm.

\begin{ex}
Before we state and prove the theorem, let us give an example in the holomorphic case on $\mathbb{C}$. Assume $t(z)$ is holomorphic in $\mathbb{C}\setminus \{0\}$
and solves the equation $z\frac{d}{dz}t-(\Omega-zh(z))t=0$ where $h$ is holomorphic in a neighborhood of $\{0\}$.
Then $f(z)=z^{\Omega}t(z)$ solves the equation $z\frac{d}{dz}f-zh(z)f=0\implies \frac{d}{dz}f-h(z)f=0$. But this means 
that $f(z)=e^{\int_{z_0}^z h(t)dt} f(z_0)$ is \textbf{holomorphic} in a neighborhood of zero. Hence by the principle of analytic continuation, we can extend the function $f$ holomorphically at $0$ ! Finally, $t(z)=\sum_{k=0}^\infty \frac{1}{k!} f^{(k)}(0)z^{k+\Omega}$ has the asymptotic expansion of Fuchsian symbols.
\end{ex}
However, in contrast
with the previous example our
theorem
\textbf{does not assume
any hypothesis
of analyticity}
since our perturbed
operator $\rho-\Omega$
is an operator
with smooth coefficients.
\begin{thm}
Let $\tilde{\Omega}\in M_n(C^\infty(M))$ s.t. there exists $\Omega\in M_n(\mathbb{C})$
with real roots satisfying $\tilde{\Omega}-\Omega|_I=0$. 
If $t\in \mathcal{D}^\prime(U\setminus I)$ is a solution of the equation $(\rho-\tilde{\Omega})t=0$ then $t$ is a Fuchsian symbol in the space $F_{\Omega}(U\setminus I)$ and $t=\sum_0^\infty t_k$
where $(t_k)_\lambda=\lambda^{\Omega+k}t_k$.
\end{thm}
\begin{proof}
The idea consists in proving that
$\lambda^{-\Omega}  t_\lambda$ 
is smooth in $\lambda$, 
then the Taylor expansion about $\lambda=0$ 
of $\lambda^{-\Omega}  t_\lambda$
will give us the expansion as Fuchsian symbol.
We restrict to a set $K^\prime=\{(x,h)| \vert h\vert\leqslant R\}$ which is stable by scaling.
We can pick a function $\chi$ which vanishes outside a compact neighborhood $K$ of $K^\prime$, $\chi|_{K^\prime}=1$, then the distribution $t\chi$ equals $t$ on $K^\prime$ and is an element of the dual space $(C^m(K))^\prime$ of the \emph{Banach space} $C^m(K)$ where $m$ is the order of the distribution $t$ (see Eskin theorem $6.4$ page $22$).
The \emph{topological dual} $(C^m(K))^\prime$ of the Banach space $C^m(K)$ is also a \emph{Banach
space} for the operator norm.  
We want to prove that
$\Vert\lambda^{-\Omega}  t_\lambda\chi\Vert_{\left(C^m(K)\right)^\prime}$ 
is bounded 
for the Banach space norm 
$\Vert . \Vert_{\left(C^m(K)\right)^\prime}$ 
of $\left(C^m(K)\right)^\prime$
and we also want to prove that the map
$\lambda\mapsto \lambda^{-\Omega}  t_\lambda\chi$ is a smooth map
for $\lambda\in[0,1]$ with value in the \emph{Banach space} $(C^m(K))^\prime$.  
We must precise the regularity of $\lambda^{-\Omega}  t_\lambda\chi$ in $\lambda\in(0,1]$.
From the identity 
$\left\langle  t_{\lambda}\chi,\varphi\right\rangle= \left\langle t,\chi_{\lambda^{-1}}\varphi_{\lambda^{-1}}\right\rangle$
we can easily prove the 
$C^0$ regularity on $\lambda\in(0,1]$ 
with value distribution of order $m$.
Then the derivative in $\lambda$ is given by the formula
$\partial_\lambda\left( \lambda^{-\Omega}  t_\lambda\chi\right)=\lambda^{-1-\Omega} \left((\rho-\Omega) t_\lambda\right)\chi$
where $\left(\rho t_\lambda\right)\chi$ is of order $m+1$.
This implies $\lambda\in(0,1]\mapsto \lambda^{-\Omega}  t_\lambda\chi\in C^1\left((0,1],(C^{m+1}(K))^\prime\right)$ then by recursion $\lambda\in(0,1]\mapsto \lambda^{-\Omega}  t_\lambda\chi\in C^k\left((0,1],(C^{m+k}(K))^\prime\right)$ where $t$ is a distribution of order $m$.
We see that at each time we increase the order of regularity in $\lambda$ of one unit, we lose regularity of  
$\lambda^{-\Omega}  t_\lambda\chi$ as a compactly supported distribution. For the moment, we know  $\lambda^{-\Omega}t_\lambda$ is smooth in $\lambda\in(0,1]$ with value distribution but the difficulty is to prove that there is no blow up at $\lambda=0$ and that it has a $C^\infty$ extension for $\lambda\in[0,1]$. The idea is to exploit the fact it satisfies a differential equation and use a version of the Gronwall lemma for Banach space valued ODE.
$f_\lambda=\lambda^{-\Omega}  t_\lambda\chi$ is a solution of the linear ODE
\begin{equation}\label{FuchsODE} 
\frac{d}{d\lambda}f_\lambda=\frac{\left(\tilde{\Omega}
-\Omega\right)_\lambda}{\lambda}f_\lambda, f_1= t\chi 
\end{equation}
where $\frac{\left(\tilde{\Omega}-\Omega\right)_\lambda}{\lambda}=\frac{e^{\log\lambda\rho\star}\left(\tilde{\Omega}-\Omega\right)}{\lambda}$ is smooth in $(\lambda,x,h)\in[0,1]\times \mathbb{R}^{n+d}$ since $\tilde{\Omega}-\Omega\in M_n(\mathcal{I})$. 
We want to prove that
there is no blow up at $\lambda=0$ which would give a unique
extension of $\lambda^{-\Omega}  t_\lambda\chi$ to $\lambda\in[0,1]$ by ODE uniqueness.
We notice that there exists a constant $C$ such that
$$\forall\lambda\in[0,1], \Vert\frac{\left(\tilde{\Omega}-\Omega\right)_\lambda}{\lambda}\lambda^{-\Omega}  t_\lambda\chi\Vert_{\left(C^m(K)\right)^\prime}\leqslant C \Vert\lambda^{-\Omega}  t_\lambda\chi\Vert_{\left(C^m(K)\right)^\prime}$$
since $\tilde{\Omega}-\Omega\in M_n(\mathcal{I})$ which means $\left(\tilde{\Omega}-\Omega\right)_\lambda=O(\lambda)$ and $\frac{\left(\tilde{\Omega}-\Omega\right)_\lambda}{\lambda}$ is bounded in $\lambda$ in the space of smooth functions for usual $C^\infty$ topology. 
Actually, we only need the simple estimate
$\forall \lambda\in[0,1], \sup_{\lambda\in[0,1]} \Vert\frac{\left(\tilde{\Omega}-\Omega\right)_\lambda}{\lambda} \Vert_{C^m(K)}<\infty $, thus  
$$f_\tau=f_1+\int_1^\tau d\lambda\frac{\left(\tilde{\Omega}-\Omega\right)_\lambda}{\lambda}f_\lambda   $$
and 
$$\Vert f_\tau\Vert_{\left(C^m(K)\right)^\prime}\leqslant \Vert f_1\Vert_{\left(C^m(K)\right)^\prime}+\Vert\int_1^\tau d\lambda\frac{\left(\tilde{\Omega}-\Omega\right)_\lambda}{\lambda}  f_\lambda \Vert_{\left(C^m(K)\right)^\prime}$$ 
by the triangle inequality
$$\Vert f_\tau\Vert_{\left(C^m(K)\right)^\prime}\leqslant  \Vert f_1\Vert_{\left(C^m(K)\right)^\prime}+\int_\tau^1 d\lambda\Vert \frac{\left(\tilde{\Omega}-\Omega\right)_\lambda}{\lambda}f_\lambda \Vert_{\left(C^m(K)\right)^\prime}  $$
by Minkowski inequality
$$\Vert f_\tau\Vert_{\left(C^m(K)\right)^\prime}\leqslant  \Vert f_1\Vert_{\left(C^m(K)\right)^\prime}+C\int_\tau^1 d\lambda \Vert f_\lambda \Vert_{\left(C^m(K)\right)^\prime}  $$
and we can conclude by an application of the Gronwall lemma.  
We deduce that $\forall \lambda\in [0,1], \Vert f_\lambda\Vert_{\left(C^m(K)\right)^\prime}\leqslant e^{C(1-\lambda)}\Vert f_1 \Vert_{\left(C^m(K)\right)^\prime}$. 
Hence $f_\lambda$ exists on $[0,1]$ (for more on Gronwall see \cite{Tao} Theorem 1.17 p.~14) otherwise there would be blow up at $\lambda=0$ but the Gronwall lemma prevents $f_\lambda$ from blowing up at $\lambda=0$. Since the ODE (\ref{FuchsODE}) has smooth coefficients the value
of its solution is smooth in $\lambda$.
%The application of the Gronwall lemma was just an a priori estimate assuming $\lambda^{-\Omega_0}t_\lambda$ was a solution of the ODE. 
%To prove the smoothness property in $\lambda$, we use a classical bootstrap argument. 
%We start from the fact that $\lambda\mapsto \lambda^{-\Omega_0}t_\lambda\in C^0\left((0,1],(C^m(K))^\prime\right)\cap L^\infty\left([0,1],(C^m(K))^\prime\right)$ and using the integral version of the ODE, we find: 
% $$\underset{C^{k+1}(0,1]\cap L^{k+1,\infty}[0,1]}{\underbrace{\tau^{-\Omega_0}  t_\tau}}=t+\int_1^\tau d\lambda\underset{C^k(0,1]\cap L^{k,\infty}[0,1] }{\underbrace{\frac{\left(\tilde{\Omega}-\Omega\right)_\lambda}{\lambda}\lambda^{-\Omega_0}  t_\lambda}} $$
%(where $L^{k,\infty}[0,1]$ denotes the space of functions with bounded $k$-th derivatives) 
%we deduce that  
%$\lambda\mapsto \lambda^{-\Omega_0}t_\lambda\in C^\infty\left([0,1],(C^m(K))^\prime\right)$.
To conclude, we Taylor expand $\lambda^{-\Omega}t_\lambda\chi$ in $\lambda$
$$\lambda^{-\Omega}t_\lambda\chi=\sum_{k=0}^\infty \frac{\lambda^k}{k!} u_k $$
hence using $\chi|_K=1$: $$t_\lambda|_K=\sum_{k=0}^\infty \frac{\lambda^{k+\Omega}}{k!} u_k|_K .$$
Hence we deduce the conclusion with $t_k|_K=\frac{u_k}{k!}$. 
\end{proof}

\subsection{Stability of the concept of approximate Fuchsians.}
 
 First, the space $F_\Omega$ is stable by left product with elements in $C^\infty(M)$, the proof is simple by Taylor expanding the smooth function.
%Thus it is stable by the left action of differential operators in $\mathbb{M}$ (the proof is easy when we work in local coordinates). 
Let $G$ be the space of diffeomorphisms of $M$ \textbf{fixing} $I$.
Before we end this section, let us prove a theorem which shows that the space $F_\Omega(U)$ of Fuchsian symbols is stable by action of $G$. This
result will imply that
$F_\Omega(U)$ does not depend on the
choice of Euler $\rho$. Before proving the theorem we give some useful lemmas:
\begin{lemm}
Let $\Phi(\lambda)=S(\lambda)^{-1}\circ 
\Phi\circ S(\lambda)$ where 
$S(\lambda)=e^{\log\lambda\rho}$ 
and $\Phi=e^{X}$ for some vector field $X$
which vanishes on $I$.
Then $\Phi(\lambda)$ 
is smooth in $\lambda\in[0,1]$
and 
$\Phi(0)$ is a
diffeomorphism fixing
$I$ which commutes 
with $\rho$
and $\Phi,\Phi(0)$ 
have the same 1-jet on $I$.
\end{lemm}
\begin{proof}
Let $\Phi(\lambda)=S(\lambda)^{-1} \circ \Phi \circ S(\lambda) $. 
We assume $\Phi=e^{X}\in G$ where $X\in\mathfrak{g}$ is 
a vector field vanishing on $I$ 
thus $\Phi(\lambda)=S(\lambda)^{-1} \circ \Phi \circ S(\lambda) = S(\lambda)^{-1} \circ e^X \circ S(\lambda)=e^{S(\lambda)^{-1} \circ X \circ S(\lambda)}=e^{X(\lambda)} $ 
where $X(\lambda)=S(\lambda)^{-1} \circ X \circ S(\lambda)$.
$\lim_{\lambda\rightarrow 0} X(\lambda)=X(0)$ exists since $X=h^ia_{i}^j(x,h)\partial_{h^j}+h^ib_i^j(x,h)\partial_{x^j}$ hence $X(\lambda)=h^ia_{i}^j(x,\lambda h)\partial_{h^j} + \lambda h^ib_i^j(x,\lambda h)\partial_{x^j}$ and $X(0)=h^ia_{i}^j(x,0)\partial_{h^j}$.
We recall 
the following 
important fact, 
$X(0)$ is in 
fact scale invariant 
i.e. 
it commutes with $\rho$. 
thus $\Phi(0)=e^{X(0)}$ 
commutes with $\rho$. 
Moreover an easy computation:
$$\left(X-X(0) \right) h^iH_i(x,h)$$ 
$$=\left(h^i(a_{i}^j(x,h)-a_i^j(x,0))\partial_{h^j}+h^ib_i^j(x,h)\partial_{x^j}\right) h^iH_i(x,h)
$$
and the fact that
$a_i^j(x,h)-a_i^j(x,0)\in\mathcal{I} $
prove that $(X-X(0))h^iH_i(x,h)=O(\vert h\vert^2)$.
Thus $(X-X(0))\mathcal{I}\subset \mathcal{I}^2$
which implies $\left(e^X-e^{X(0)}\right)^\star \mathcal{I}=\left(\Phi-\Phi(0)\right)^\star \mathcal{I}\subset \mathcal{I}^2$.
This is enough to prove that 
$\Phi$ and $\Phi(0)$ have same 1-jet along $I$.
\end{proof}
\begin{lemm}
Under the hypothesis of the above lemma, the pull-back operator $\Phi(\lambda)^\star$
admits a Taylor expansion of the following form:
$$\Phi(\lambda)^\star=\sum_{k=0}^N \frac{\lambda^k}{k!}\mathbb{D}_k\Phi_0^*+I_N(\Phi,\lambda)^\star $$
where $\mathbb{D}_k$ 
is a differential operator 
which depends polynomially on finite jets of $X$ and $\rho$ at $I$.
\end{lemm}
\begin{proof}
We start from the identity 
$\lambda\frac{d}{d\lambda} X(\lambda)
= \lambda\frac{d}{d\lambda} Ad_{S(\lambda)}X
=-\left[\rho,X(\lambda)\right]$. 
This implies 
$$\partial^i_\lambda X(\lambda)
=\frac{1}{\lambda^i} \lambda^i\partial^i_{\lambda}X(\lambda) 
=\frac{1}{\lambda^i i!}\lambda\frac{d}{d\lambda}\dots 
\left(\lambda\frac{d}{d\lambda}-i+1 \right)X(\lambda)$$
$$=\frac{1}{\lambda^i i!}(-ad_\rho) \dots (-ad_\rho -i+1)X(\lambda)
$$ 
$$\implies \partial^i_\lambda X(0)
=\lim_{\lambda\rightarrow 0}\frac{1}{\lambda^i i!}(-ad_\rho) \dots (-ad_\rho -i+1)X(\lambda).$$ 
Hence the derivatives $\partial^i_\lambda X(0)$ only depend polynomially 
on finite jets of $X$ and $\rho$ at $(x,0)$. 
Then we Taylor expand the map $\Phi(\lambda)$ at $\lambda=0$:
$$\Phi(\lambda)=\sum_{k\leqslant N} \frac{\lambda^k}{k!} \left(\partial_\lambda^k e^{X(\lambda)\star}\right)_{\lambda=0}+ I_N(\Phi,\lambda)^\star $$ 
by definition of the exponential map and successive differentiation, 
the terms 
$\left(\partial_\lambda^k e^{X(\lambda)\star}\right)_{\lambda=0}$ 
are all 
of the form 
$\mathbb{D}_k\Phi_0^* $ where each $\mathbb{D}_k$ 
is a differential operator in 
$\mathbb{C}\left\langle \partial_\lambda^i X(0)\right\rangle_i$, for instance: 
$$\mathbb{D}_1=\partial_\lambda X(0), 
\mathbb{D}_2=\partial_\lambda^2X(0)+(\partial_\lambda X)^2(0).$$
\end{proof}
A consequence of the above lemma 
is that
for all distribution $t$, for all $N,\lambda$, 
the pull-back
$I_N(\Phi,\lambda)^\star t$ exists and 
we can bound its wave front set: 
$$WF(I_N(\Phi,\lambda)^\star t)
\subset \Phi(0)^\star WF(t)\cup  \Phi(\lambda)^\star WF(t).$$
\begin{thm}
Let $t\in F^\rho_\Omega$ for a choice of $\rho$, $t$ has the asymptotic expansion 
$t=\sum_l t_l$, and $\Phi=e^X\in G$ for $X$ vanishing on $I$. 
Then we have
$\Phi^\star t\in F_\Omega^{\rho}$ and $\Phi^\star t=\sum_{n=0}^\infty \tilde{t}_n$
where $\tilde{t}_n$ depends only on $t_l,l\leqslant n$ and polynomially on finite jets of $\rho, X$ at $I$.
\end{thm}

\begin{proof}
Since $\Phi(\lambda)$ depends smoothly in $\lambda$ and  $\lambda^{-\Omega}t_\lambda$ admits an asymptotic expansion at $\lambda=0$, the pulled back family $\Phi(\lambda)^*(\lambda^{-\Omega}t_\lambda)=\lambda^{-\Omega}\left(\Phi^*t\right)_\lambda$ admits an asymptotic expansion at $\lambda=0$. In order to conclude, we expand $\lambda^{-\Omega}t_\lambda=\sum_{l=0}^\infty \lambda^{-\Omega+l}t_l $ and $\Phi(\lambda)=\sum_{k=0}^\infty \frac{\lambda^k}{k!} \mathbb{D}_k\Phi_0^*$
and we obtain the general expansion
$$\Phi(\lambda)^*\left(\lambda^{-\Omega}t_\lambda\right)=\sum_{n=0}^\infty\lambda^{-\Omega+n}
\sum_{k+l=n} \frac{1}{k!} \mathbb{D}_k\Phi_0^*t_l.$$ 
\end{proof}
We keep the notation and hypothesis of the above theorem
\begin{coro}
Let $\Gamma$ be a cone in $T^\bullet (M\setminus I)$.
If $\forall k$, $WF(t_k)\subset\Gamma$ then
$\forall n, WF(\tilde{t}_n)\subset \Phi_0^\star \Gamma$.
\end{coro}
We deduce from the previous theorem an important corollary which is that the class of 
Fuchsian symbols $F_\Omega$ is \textbf{independent} of the choice of Euler vector field. 
\begin{coro}
Let $t\in F^\rho_\Omega$ for a choice of $\rho$, then for any other generalized Euler $\tilde{\rho}$, we have  
$t\in F_\Omega^{\tilde{\rho}}$.
\end{coro}
\begin{proof}
By the result of chapter $1$, for any other vector $\tilde{\rho}$, we have $\Phi^{-1*}\tilde{\rho}=\rho$ for a diffeomorphism $\Phi$ fixing $I$.
$$ 0=\rho t-\Omega t= \Phi^{-1*}\tilde{\rho} \Phi^{*}t-\Phi^{-1*}\Omega\Phi^{*} t \implies \tilde{\rho} \Phi^{*}t-\Omega\Phi^{*} t=0   $$
this means $\Phi^*t $ is killed by the Fuchsian operator $\tilde{\rho}-\Omega$ thus $\Phi^*t \in F_\Omega^{\tilde{\rho}}$. 
\end{proof}

\section{Meromorphic regularization as a Mellin transform.}

 In this section, for pedagogical reasons, we work in local charts in order to make as explicit as possible the relationship with the Mellin transform. More precisely, we work in a given fixed compact subset $K=K_1\times K_2\subset \mathbb{R}^{n+d}$, the compact set is geodesically convex for $\rho=h^j\partial_{h^j}$. 
All test functions are supported in $K$.
Let $\chi\in C^\infty_0(\mathbb{R}^{n+d})$, $\chi\geqslant 0$ and $\chi|_{K\cap \{\vert h\vert\leqslant a\}}=1, \chi|_{K\cap \{\vert h\vert\geqslant b\}}=0$ where $b>a>0$.

\begin{equation}
\left\langle T,\omega\right\rangle =\int_{0}^1\frac{d\lambda}{\lambda} \left\langle T \psi_{\lambda^{-1}} ,\omega\right\rangle+ \left\langle T(1-\chi),\omega\right\rangle
\end{equation}
\paragraph{The meromorphic regularization formula.}

 We modify the extension formula of H\"ormander by introducing a weight $\lambda^\mu$ in the integral over the scale $\lambda$:
\begin{equation}\label{meromreg}
\left\langle T^\mu,\omega\right\rangle =\int_{0}^1\frac{d\lambda}{\lambda}\lambda^\mu \left\langle T \psi_{\lambda^{-1}} ,\omega\right\rangle,
\end{equation}
this defines a regularization of the extension depending on a parameter $\mu$. 
We would like to call the attention of the reader on the fact that if the test form $\omega$ was not supported on $I$, we would have a well defined extension at the limit $\mu\rightarrow 0$.

\paragraph{The philosophy of meromorphic regularization.}

 The goal is to prove that $T^\mu$ can be extended to a family of current in $\mathcal{D}^\prime_k(U)$ depending holomorphically in $\mu$ for $Re(\mu)$ large enough. 
Then under the hypothesis that $T$ is a Fuchsian symbol, $T^\mu$ should extend \textbf{meromorphically} in $\mu$ with poles at $\mu=0$ which are currents supported on $I$ (ie local counterterms).
Then the meromorphic regularization will be given by the formula
\begin{equation}
\lim_{\mu\rightarrow 0}\left(T^\mu + T (1-\chi)-\text{poles at $\mu=0$ with value current supported on $I$}\right)
\end{equation}

\begin{defi}
A family $(T^\mu)_\mu$ of currents in $\mathcal{D}_k^\prime(U)$ is said to be holomorphic (resp meromorphic) in $\mu$ iff for all test forms $\omega\in \mathcal{D}^k(U)$, $\mu\mapsto \left\langle T^\mu  ,\omega\right\rangle\in\mathbb{C}$ is holomorphic (resp meromorphic).
\end{defi}
If $\mu\mapsto T^\mu$ is holomorphic in a domain $B_r(\mu_0)\setminus \{\mu_0\}$, for all test functions $\varphi$, the map
$\mu\mapsto \left\langle T^\mu  ,\varphi\right\rangle $ 
has an expansion in Laurent series in $\mu$ around $\mu_0$, $\left\langle T^\mu  ,\varphi\right\rangle=\sum_{k=-\infty}^{k=+\infty} (\mu-\mu_0)^k\left\langle T^{\mu_0(k)}   ,\varphi\right\rangle$
where each coefficient of the Laurent series is a distribution tested against $\varphi$ (there is a similar discussion in \cite{Gelfand} Chapter $1$ appendix $2$).

\begin{proof}
By the Cauchy formula and by the holomorphicity of $\left\langle T^\mu  ,\varphi\right\rangle$,
for all test function $\varphi$, 
we must have
$$\forall k\in\mathbb{Z},
\left\langle T^{\mu_0(k)}   ,\varphi\right\rangle=\frac{1}{2i\pi}\int_{\partial B_r(\mu_0)}\frac{d\mu}{(\mu-\mu_0)^{k+1}} \left\langle T^\mu  ,\varphi\right\rangle .$$
Thus we define $T^{\mu_0(k)}=\frac{1}{2i\pi}\int_{\partial B_r(\mu_0)}\frac{d\mu}{(\mu-\mu_0)^{k+1}} T^\mu $ which is a linear map on $\mathcal{D}(U)$.
To prove the continuity, we just use the Banach Steinhaus theorem, for all compact $K\subset U$, there exists $C>0$
and a seminorm $\pi_m$ s.t. for all $\varphi\in\mathcal{D}_K(U)$
$$\forall\mu\in\partial B_r(\mu_0), \vert\left\langle T^\mu  ,\varphi\right\rangle\vert\leqslant C\pi_m(\varphi) ,$$
thus 
$$\forall \varphi\in\mathcal{D}_K(U), \vert\left\langle T^{\mu_0(k)}   ,\varphi\right\rangle \vert\leqslant
Cr^{-k}\pi_m(\varphi),$$
which proves the continuity of $T^{\mu_0(k)}$ for all $k$.
\end{proof}
 Thus we can write the Laurent series expansion of $\mu\mapsto T^\mu$ around $\mu_0$ as a series in powers of $(\mu-\mu_0)$ with distributional coefficients:
$$ T^\mu =\sum_{k=-\infty}^{k=+\infty} (\mu-\mu_0)^k T^{\mu_0(k)} .$$
\begin{defi}
We say that $\mu\mapsto T^\mu$ is meromorphic with poles of order $N$ at $\mu_0$ when  $\mu\mapsto T^\mu$ is holomorphic in a domain $B_r(\mu_0)\setminus \{\mu_0\}$ and 
$T^\mu =\sum_{k=-N}^{k=+\infty} (\mu-\mu_0)^k T^{\mu_0(k)} $.
\end{defi}
Using this definition, it makes sense to speak about the support of the poles, it just means the support of the distributions $T^{\mu_0(k)}$ for $k<0$.

\subsubsection{The holomorphicity theorem.}
 Recall that $T^\mu$ is defined by the formula
$\left\langle T^\mu,\omega\right\rangle =\int_{0}^1\frac{d\lambda}{\lambda}\lambda^\mu \left\langle T \psi_{\lambda^{-1}} ,\omega\right\rangle$.

\begin{lemm}\label{convlemm}
If $T\in E_{s}\left(\mathcal{D}_k^\prime\left(U\setminus I\right)\right)$,
then $T^\mu$ has a well defined extension in $\mathcal{D}_k^\prime(U)$ for $Re(\mu)+s+k-n>0$ and $T^\mu\in E_{s+Re(\mu)}(\mathcal{D}_k^\prime(U))$.
\end{lemm}
\begin{proof}
We keep the notation of the proof of theorem $(1.2)$ and we recall the main facts.
In the proof of theorem $(1.2)$, we proved that if $(c_\lambda)_\lambda$ is a bounded family of distributions supported on a fixed annulus $a\leqslant\vert h\vert\leqslant b$, then $\lambda^{-d}c_\lambda(.,\lambda.)$ is a bounded family of distributions. Hence from the boundedness of the family $(c_\lambda=\lambda^{-s}t_\lambda\psi)_\lambda$, we deduced the boundedness of the family $(\lambda^{-d}c_\lambda(.,\lambda.)=\lambda^{-s-d} t \psi_{\lambda^{-1}})_\lambda$.
By reasoning as in the proof of theorem $(1.2)$ in Chapter $1$, the function $\lambda\mapsto f(\lambda)=\lambda^{-s-(k-n)}\left\langle T \psi_{\lambda^{-1}} ,\omega\right\rangle$ is a bounded function supported on the interval $[0,1]$.
Thus we find 
$$
\left\langle T^\mu,\omega\right\rangle =
\int_{0}^1\frac{d\lambda}{\lambda}\lambda^\mu \left\langle T \psi_{\lambda^{-1}} ,\omega\right\rangle $$ 
$$=\int_{0}^1\frac{d\lambda}{\lambda}\lambda^{\mu+s+k-n} \lambda^{-s-(k-n)}\left\langle T \psi_{\lambda^{-1}} ,\omega\right\rangle=\int_{0}^{+\infty}\frac{d\lambda}{\lambda}\lambda^{\mu+s+k-n}f(\lambda)
$$
The last integral \textbf{converges} when $Re(\mu)+s+k-n>0$ because $f$ is bounded on $[0,1]$. This already tells us that the family of currents $(T^\mu)_\mu$ is \textbf{well defined} in $\mathcal{D}_k^\prime(U)$ when $Re(\mu)+s+k-n>0$.
To prove that $T^\mu\in E_{s+Re(\mu)}$, we use the theorem (2.1) proved in Chapter $1$ for the bounded family of currents $\left(c_\lambda=\lambda^{-s} T_\lambda\psi\right)_\lambda $ supported on a fixed annulus.
\end{proof}
We establish a neat result namely that the function $\lambda\mapsto \left\langle T\psi_{\lambda^{-1}},\omega \right\rangle$ is in fact always smooth in $\lambda\in(0,1]$.
But of course that does not mean it should be $L_{loc}^1$ at $\lambda=0$. 
\begin{lemm}
$\lambda\mapsto \lambda^\mu \left\langle T\psi_{\lambda^{-1}},\omega\right\rangle$ is smooth in $0<\lambda\leqslant 1$.
\end{lemm}
\begin{proof} 
There is a compact set $K=\text{supp }\omega$ such that if $x\notin K$, $\psi_{\lambda^{-1}}\omega(x)=0$, $\forall \lambda\in(0,1]$. Also $\lambda\mapsto \psi_{\lambda^{-1}}\omega$ is smooth in $\lambda$. Then the result follows from application of Theorem 2.1.3 in \cite{Hormander}.
\end{proof}

\begin{thm}\label{holomorphicitythm}  
We keep the notation and hypothesis of lemma (\ref{convlemm}), then 
$\forall\omega\in \mathcal{D}^k(U)$ (resp $\omega\in \mathcal{D}^k(U\setminus I)$), the map
$\mu\mapsto \left\langle T^\mu,\omega \right\rangle$
is holomorphic in the $\textbf{half-plane}$ $Re(\mu)+s+k-n>0$ (resp holomorphic in $\mathbb{C}$).\end{thm}  

\begin{proof}
We relate the regularization formulas to the Mellin transform.
By definition, the Mellin transform of a distribution $f\in \mathcal{D}^\prime(\mathbb{R}^+)$
is given by the formula (see ``The Mellin Transformation and Other
Useful Analytic Techniques'' by Don Zagier in \cite{Zeidler} p.~305 and \cite{Jeanquartier})
\begin{equation}
\tilde{f}(\mu)=\int_0^\infty \frac{d\lambda}{\lambda} \lambda^\mu f(\lambda).
\end{equation}  
Actually, in the notation of Zagier, we study the half-Mellin transform:
\begin{equation}
\tilde{f}_{\leqslant 1}(\mu)=\int_0^1\frac{d\lambda}{\lambda} \lambda^\mu f(\lambda)
\end{equation}  
The regularization formula (\ref{meromreg}) is the \textbf{Mellin transform} of the  function $\lambda\mapsto \left\langle T \psi_{\lambda^{-1}} ,\omega\right\rangle\chi_{[0,1]}$, where $\chi$ is the characteristic function of the interval $[0,1]$.
The function $\lambda\mapsto f(\lambda)=\lambda^{-s-(k-n)}\left\langle T \psi_{\lambda^{-1}} ,\omega\right\rangle\chi_{[0,1]}$ is a function in $C^\infty(0,1]\cap L^\infty[0,1]$ ( however, it is not smooth at $0$), $\left\langle T^\mu,\omega\right\rangle$ is thus reinterpreted as the Mellin transform $\Gamma_f(\mu+s+k-n)$ of $f\in C^\infty(0,1]\cap L^\infty[0,1]\implies f\in L^1[0,1]$.
Then we use the classical holomorphic properties of the Mellin transform as explained in \cite{Taylor} appendix A p.~308,309.
To understand the holomorphicity properties of the Mellin transform, we relate the Mellin transform with the Fourier Laplace transform in the complex plane by the variable change $e^{t}=\lambda $ (see \cite{Taylor} appendix A formula $\text{A}.18$ p.~308)
$$\int_0^1 \frac{d\lambda}{\lambda} \lambda^s f(\lambda)=\int_{-\infty}^0 dt e^{ts}f(e^t)=\int_{-\infty}^\infty dt e^{-ts}f(e^{-t})H(t)  $$ 
where $H$ is the Heaviside function and where $t\mapsto f(e^{-t})H(t)$ is bounded. For any $\varepsilon>0$, $t\mapsto e^{-t\varepsilon}H(t)f(e^{-t})$ is in $L^p(\mathbb{R}),\forall p\in[1,\infty]$, especially in $L^2(\mathbb{R})$ hence
$$s\mapsto\int_{-\infty}^\infty dt e^{-t(s+\varepsilon)}f(e^{-t})H(t)  $$ 
is holomorphic in $s$ for $Re(s)\geqslant 0$ by the properties of the holomorphic Fourier transform. As this is true for any $\varepsilon >0$, the Mellin transform is holomorphic on $Re(s)>0$.
%
%
%  Or we use the Paley Wiener theorem (Theorem $3.1$ in Jeanquartier transformation de Mellin et developpements asymptotiques) for the Mellin transform :
%the Mellin transform $\tilde{f}(\mu)=\int_0^\infty \frac{d\lambda}{\lambda}\lambda^\mu f(\lambda)$ of any bounded function $f$ supported on an interval $[0,1]$ is entire hence holomorphic in some domain of non empty interior. 
%But this theorem does not give a bound on the domain of holomorphicity. However, if $f$ is bounded on $[0,1]$, then $\forall\varepsilon>0,\lambda\mapsto\lambda^{\varepsilon}f(\lambda) $ is square integrable
%$$ \int_0^\infty \frac{d\lambda}{\lambda }(\lambda^{\varepsilon}f(\lambda))^2 <\infty $$
%and by application of theorem $3.2$ in Jeanquartier, the Mellin transform $\tilde{\lambda^\varepsilon f}(\mu)=\tilde{f}(\mu+\varepsilon)$ is $\textbf{holomorphic}$ in the half-plane $Re(\mu)>0$. Hence, $\forall\varepsilon>0$ the map $\mu\mapsto\tilde{f}(\mu) $ is $\textbf{holomorphic}$ in the half-plane $Re(\mu)>\varepsilon$, hence it is holomorphic in $Re(\mu)>0$.
%We deduce $\mu\mapsto \left\langle T^\mu,\omega\right\rangle=\tilde{f}(\mu+s+k-n)$ is holomorphic in $Re(\mu)+s+k-n>0$.
\end{proof}
Let us keep the notations of the previous theorem and consider the family $\mu\mapsto T^\mu$  holomorphic for $Re(\mu)+s+k-n>0$.
We prove a lemma which states that if there is a meromorphic extension of the holomorphic family $\mu\mapsto T^\mu$, then this meromorphic extension must have poles supported on $I$ (ie locality of counterterms). 
\begin{lemm}\label{locpoles}
If $\mu\mapsto T^\mu$ is a meromorphic extension of the holomorphic family $\mu\mapsto T^\mu$, then the poles of $T^\mu$ are distributions in $\mathcal{D}^\prime(U)$ \textbf{supported} on $U\cap I$
i.e. they are local counterterms.
\end{lemm}
\begin{proof}
$\forall\omega\in \mathcal{D}^k(U), \mu\mapsto \left\langle T^\mu,\omega \right\rangle$ is holomorphic in the $\textbf{half-plane}$ $Re(\mu)+s+k-n>0$. 
Let us notice that if $\omega\in \mathcal{D}^k(U\setminus I)$, the function $\lambda\mapsto \lambda^\mu \left\langle T \psi_{\lambda^{-1}} ,\omega\right\rangle$ is smooth in $\lambda$ and vanishes in a neighborhood of $\lambda=0$, hence the formula (\ref{meromreg}) makes sense for all $\mu\in\mathbb{C}$ and is holomorphic in $\mu$. If $T^\mu$ had a meromorphic expansion, then we write the Laurent series expansion of $\mu\mapsto T^\mu$ around some value $\mu_0\in\mathbb{C}$:
$$T^\mu =\sum_{k=-N}^{k=+\infty} (\mu-\mu_0)^k T^{\mu_0(k)} $$
but for all $\omega$ supported on $U\setminus I$,
$ \left\langle T^\mu,\omega \right\rangle$ is holomorphic at $\mu_0$ 
thus all the poles $( \left\langle T^{\mu_0(k)},\omega \right\rangle)_{k<0}$ must vanish ! $\forall\omega \in \mathcal{D}^k(U\setminus I),\forall k<0, \left\langle T^{\mu_0(k)},\omega \right\rangle=0$ which means $\forall k<0,\text{supp }T^{\mu_0(k)}$ does not meet $U\setminus I$ which yields the conclusion.
\end{proof}
\subsection{The meromorphic extension.}
 We set the stage for our next theorem which states that if $T$ is a Fuchsian symbol, then the holomorphic regularization formula of H\"ormander $\mu\mapsto T^\mu$ has a \textbf{meromorphic extension} in the complex parameter $\mu$. 
Let $T\in \mathcal{D}^\prime_k(U\setminus I)$ and if $T\in F_\Omega(U\setminus I)$ then we have by definition $T=\sum_0^N T_k+R_N$ where the error term $R_N\in E_{s+N+1}$ where $s=\inf Spec(\Omega)$. Notice that for any compactly supported test form $\omega$, the current $T\wedge\omega$ is also a Fuchsian symbol, and we have the expansion
$\forall N,\left(T\wedge\omega\right)=\sum_{j\leqslant N}(T\wedge\omega)_j+I_N(T\wedge\omega)$ where $(T\wedge\omega)_{js}=s^{j+\Omega}(T\wedge\omega)_j$ and the remainder $I_N(T\wedge\omega)\in E_{s+N+1}$. Following the notations of Chapter $1$, we denote by $\psi$ the function $(-\rho\chi)$.
%Set $\left\langle T^\mu,\omega \right\rangle=\int_{0}^1\frac{d\lambda}{\lambda}\lambda^\mu \left\langle T \psi_{\lambda^{-1}} ,\omega\right\rangle $.
\begin{thm}\label{meromthm}
If $T\in F_\Omega(U\setminus I)$ then $\mu\mapsto T^\mu $ has an extension as a distribution in $\mathcal{D}^\prime(U)$ and depends meromorphically in $\mu$ with poles in $-Spec(\Omega)-\mathbb{N}$.
\begin{eqnarray}\label{identitymeromcontinuation}
\forall p, \exists N,\left\langle T^\mu,\omega \right\rangle =\sum_{j\leqslant N}(\mu+j+\Omega)^{-1} \left\langle (T\wedge\omega)_{j} ,\psi    \right\rangle+\left\langle I^\mu_N(T\wedge\omega),\psi\right\rangle 
\end{eqnarray}
 where the identity is meromorphic in the domain $ \{Re(\mu)+p>0\}$.
\end{thm}

\begin{proof}
Before we start proving anything, let us make a small comment on the principle used here. The key idea is \textbf{analytic continuation}, when two holomorphic functions $f_1,f_2$ defined on respective domains $U_1,U_2$ coincide on an open set, then there is a \textbf{unique} function $f$ (unique in the sense that any analytic continuation of $f_i,i=1,2$ must coincide with $f$ on their common domain of definition) defined on $U_1\bigcup U_2$ which extends $f_1,f_2$.
Recall that the exterior product $\left(T\wedge\omega\right)$ is a Fuchsian symbol
since $T\in F_{\Omega}$ is Fuchsian and $\omega$ is a smooth test form.
Thus $\lambda^{-\Omega}(T\wedge\omega)_\lambda$ has an asymptotic expansion in $\lambda$.
We expand $(T\wedge\omega)$ in order to extract the relevant first terms and the remainder of the asymptotic expansion.
$$ T\wedge\omega= \sum_{k=0}^N \underset{\text{killed by }\rho-k-\Omega}{\underbrace{(T\wedge\omega)_k}}+\underset{\in E_{N+\Omega+1}}{\underbrace{I_N(T\wedge\omega)}} $$
we replace 
this decomposition in the integral formula $\int_{0}^1\frac{d\lambda}{\lambda}\lambda^\mu \left\langle T \psi_{\lambda^{-1}} ,\omega\right\rangle $.
The computation gives:
$$\forall N, \int_{0}^1\frac{d\lambda}{\lambda}\lambda^\mu \left\langle T \psi_{\lambda^{-1}} ,\omega\right\rangle = \int_{0}^1\frac{d\lambda}{\lambda}\lambda^\mu \left\langle (T\wedge\omega), \psi_{\lambda^{-1}} \right\rangle  $$
$$=\int_{0}^1\frac{d\lambda}{\lambda}\lambda^\mu \left\langle (T\wedge\omega)_{\lambda}, \psi \right\rangle=\int_{0}^1\frac{d\lambda}{\lambda}\lambda^\mu \left(\sum_{j\leqslant N}\left\langle (T\wedge\omega)_{j\lambda}, \psi \right\rangle +\left\langle (I_N(T\wedge\omega))_{\lambda}, \psi \right\rangle\right) $$
$$=\sum_{j\leqslant N} \int_{0}^1\frac{d\lambda}{\lambda}\lambda^{\mu+\Omega+j} \left\langle (T\wedge\omega)_{j}, \psi \right\rangle + \int_{0}^1\frac{d\lambda}{\lambda}\lambda^{\mu}\left\langle (I_N(T\wedge\omega))_{\lambda}, \psi \right\rangle.$$
Then for $Re(\mu)$ large enough, the first $N+1$ integrals converge and can be computed
$$=\sum_{j\leqslant N} \underset{\text{poles when }\det(\mu+\Omega+j)=0 }{\underbrace{\left(\mu+\Omega+j\right)^{-1}}} \left\langle (T\wedge\omega)_{j}, \psi \right\rangle + \int_{0}^1\frac{d\lambda}{\lambda}\underset{O(\lambda^{N+1+\Omega+Re(\mu)})}{\underbrace{\lambda^{\mu}\left\langle (I_N(T\wedge\omega))_{\lambda}, \psi \right\rangle}} $$
where the remainder is integrable and holomorphic in $\mu$ in the half plane $Re(\mu)+N+1+\Omega>0$ by theorem (\ref{holomorphicitythm}).
Finally for all $N$, $\left\langle T^\mu,\omega \right\rangle$ has meromorphic continuation on  $Re(\mu)+N+1+\Omega>0$ hence it has meromorphic continuation everywhere on $\mathbb{C}$. 
\end{proof}
By a matrix conjuguation, 
we can always reduce $\Omega$ to its Jordan normal form $\Omega=G^{-1}(D+N)G$
where $D$ is diagonal and $N$ is a nilpotent matrix which commutes with $D$.
We set $(-d_i,n_i)_{i\in I}$ 
the eigenvalues of $\Omega$ 
with their respective multiplicities, 
hence $D$ is a diagonal matrix with eigenvalues $(-d_i)_i$. 
Note that if $0\in -Spec(\Omega)-\mathbb{N}$, then
$\mu=0$ is a pole of the meromorphic extension: $0=d_i-j$
where $j\in\mathbb{N}$ and $d_i$ is an eigenvalue of $\Omega$
with multiplicity $n_i$. 
\begin{prop}\label{ordn}
Let $\Omega\in M_n(\mathbb{C})$ 
and $T\in F_\Omega(U\setminus I)$. 
If $Spec\left(\Omega\right)\cap -\mathbb{N}=\emptyset$
then $T^\mu$ is holomorphic at $\mu=0$.
If $Spec\left(\Omega\right)\cap -\mathbb{N}\neq\emptyset$ 
then $T^\mu$ has a pole at $\mu=0$ of order
at most $n$.
\end{prop}
\begin{proof}
We assume that $d_i-j=0$ 
for some eigenvalue $d_i\in Spec(\Omega)$ 
and some integer $j$.
Up to conjuguation and projection, 
the proof reduces to an elementary computation
in a generalized eigenspace $E_i$ of dimension $n_i$ 
associated
to the eigenvalue $-d_i$ s.t. 
$d_i-j=0$. 
Indeed, $\Omega|_{E_i}=-d_i+N_i$ 
where $N_i$ is a \emph{nilpotent} matrix of fixed order $n_i$.
$(\mu+\Omega+j)^{-1}|_{E_i}=(\mu+N_i)^{-1}=\mu^{-1}\left(\sum_{k=0}^{n_i-1}(-1)^k\mu^{-k} N_i^k\right)=\mu^{-1}+\dots+\mu^{-n_i}(-1)^{n_i-1}N_i^{n_i-1}$, so the worst singularity
is a pole of order at most $n_i$ in $\mu$.
\end{proof}
\begin{prop}\label{propimp}
The extension $T^\mu$ defined in the previous theorem satisfies the property $T^\mu\in F_{\Omega+\mu}$.
\end{prop}
\begin{proof} 
To prove that $T^\mu\in F_{\Omega+\mu}$, it is enough to prove that if $T$ is a solution of $(\rho-\Omega)T=0$, then the meromorphic extension $T^\mu$ is solution of the equation
$(\rho-\Omega-\mu)T^\mu=0$ \textbf{on the domain} $\chi=1$. We try to scale $T^\mu$ and we compute $\tau^{-\Omega-\mu}  T^\mu(.,\tau.)$ where $T\in \mathcal{D}_k^\prime(U\setminus I)$ is exact Fuchsian $T_\lambda=\lambda^\Omega T$. First, it is not true that $T^\mu$ will scale exactly like $T^\mu_\tau=\tau^{\Omega+\mu}T^\mu$
everywhere in $U\setminus I$. However, in any $\rho$-stable domain $U$ for $\rho=h^j\partial_{h^j}$ in which $\chi|_U=1$, we will be able to find that $\forall \tau\in(0,1], T^\mu_\tau|_U=\tau^{\Omega+\mu}T^\mu|_U$. This can be understood in terms of section $T^\mu|_U$ of the \textbf{sheaf of currents} over the open set $U$. A typical example of such nice domains would be $K\times \{\vert h\vert\leqslant a\}\subset \mathbb{R}^n\times\mathbb{R}^d$ in the local chart $\mathbb{R}^{n+d}$ where the plateau function $\chi$ satisfies the support condition:
\begin{equation}
\chi_{K\times \{\vert h\vert\leqslant a\}}=1,\chi_{K\times \{\vert h\vert\geqslant b\}}=0 
\end{equation}  
for $0<a<b$.  
We pick a test form $\omega\in \mathcal{D}^\prime(U)$.$$\forall 0< \tau\leqslant 1  , \tau^{-\Omega-\mu} \left\langle T^\mu_\tau ,\omega \right\rangle=\tau^{-\Omega-\mu}  \left\langle T^\mu ,\omega_{\tau^{-1}} \right\rangle=\int_{0}^1\frac{d\lambda}{\lambda} \lambda^\mu\tau^{-\Omega-\mu} \left\langle T\psi_{\lambda^{-1}},\omega_{\tau^{-1}} \right\rangle $$ 
$$=\int_{0}^1\frac{d\lambda}{\lambda} \left(\frac{\lambda}{\tau}\right)^{\mu}\tau^{-\Omega} \left\langle T_\lambda\psi,\omega_{\lambda\tau^{-1}} \right\rangle =\int_{0}^1\frac{d\lambda}{\lambda} \left(\frac{\lambda}{\tau}\right)^{\mu} \left\langle T_{\lambda\tau^{-1}}\psi,\omega_{\lambda\tau^{-1}} \right\rangle $$
because $T$ is exact Fuchsian. Then by a change of variable, we obtain
 $$\tau^{-\Omega-\mu} \left\langle T^\mu_\tau ,\omega \right\rangle=\int_{0}^{\frac{1}{\tau}}\frac{d\lambda}{\lambda} \lambda^{\mu} \left\langle T\psi_{\lambda^{-1}},\omega \right\rangle $$
We notice that the condition on the support of $\chi$ implies $\psi=-\rho\chi$ is supported in $\{a\leqslant\vert h\vert\leqslant b\}\cap U$. Since $\psi$ is supported in $\{a\leqslant\vert h\vert\leqslant b\}\cap U$, 
$\psi_{\lambda^{-1}}$ is supported in $\{\lambda a\leqslant\vert h\vert\leqslant \lambda b\}\cap U$.
However, we also recall 
that $\omega$ is supported 
inside the domain $\{\vert h\vert\leqslant a\}$. 
$T\psi_{\lambda^{-1}}$ is supported in $\{\lambda a\leqslant \vert h\vert\leqslant \lambda b\}$
hence $\left\langle T \psi_{\lambda^{-1}} ,\omega \right\rangle $ vanishes when
$ \lambda\geqslant 1$.
Finally:
$$\tau^{-\Omega-\mu}  \left\langle T^\mu_\tau,\omega\right\rangle=\int_{0}^{\frac{1}{\tau}}\frac{d\lambda}{\lambda} \lambda^{\mu} \left\langle T\psi_{\lambda^{-1}},\omega \right\rangle=\int_{0}^{1}\frac{d\lambda}{\lambda} \lambda^{\mu} \left\langle T\psi_{\lambda^{-1}},\omega \right\rangle=\left\langle T^\mu,\omega\right\rangle$$
Notice that for $Re(\mu)$ large enough, all our integrals make sense when $\tau>0$ because the integrand viewed as a function of $\lambda$ is in $L^1([0,1])$.
Then by the principle of analytic continuation
$$\rho T^\mu-(\Omega+\mu)T^\mu=0 \text{ on }U$$
for $Re(\mu)$ large enough thus the same equation is satisfied by any meromorphic continuation of $T^\mu$ and the r.h.s. of the equation \ref{identitymeromcontinuation}
satifies the Fuchs equation $\rho T^\mu-(\Omega+\mu)T^\mu=0 $. 
%If $T$ is not pure Fuchsian, then decompose $T=T_s+R$ where $T_s$ is pure Fuchsian.
%Then we compute the action of the differential operator $\left(\rho-\mu-\Omega\right)\in\mathbb{M}$ on the extension $T^\mu$: $$\left(\rho-\mu-\Omega\right)T^\mu=\left(\rho-\mu-\Omega\right)T^\mu_s
% +\left(\rho-\mu-\Omega\right)R^\mu =0+ \left(\rho-\mu-\Omega\right)R^\mu$$
%because we saw that the extension $T_s^\mu$ is killed by $\left(\rho-\mu-\Omega\right)$.
%But then we use the proposition \ref{stabilitymelrose} to conclude that $\left(\rho-\mu-\Omega\right)T^\mu=\left(\rho-\mu-\Omega\right)R^\mu\in E_{s+Re(\mu)}$.
%This proves the last claim. 
\end{proof}

\section{The Riesz regularization.}
\subsubsection{Preliminary discussion.} 
Up to now, the meromorphic regularization operation seems not very interesting since it does not define an extension of the original current $T\in  F_\Omega(U\setminus I)$ from which
we started. In order to recover a genuine extension, we must somehow make $\mu$ tend to $0$ in
the meromorphic regularization of H\"ormander. In order to do this, we will have to subtract poles 
but fortunately these poles are local counterterms hence the subtraction operation does not affect
the extension outside the submanifold $I$. The procedure we are going to describe will be called Riesz regularization.
Let us consider a given $T\in F_\Omega(U\setminus I)$. 
If $-m-1<s\leqslant-m$, the extension procedure defined in Chapter $1$ which could be called the Hadamard finite part procedure is given by
\begin{equation} 
\left\langle \overline{T}_{\text{Hadamard}} ,\omega\right\rangle=\lim_{\varepsilon\rightarrow 0}\left\langle T(\chi-\chi_{\varepsilon^{-1}})  , I_m(\omega) \right\rangle+ \left\langle T(1-\chi) , \omega\right\rangle 
\end{equation}
whereas in the Riesz regularization, 
we first extend meromorphically in $\mu$, 
then we subtract the poles at $\mu=0$, 
and finally take the limit $\mu\rightarrow 0$.
\paragraph{Fundamental example.}
\begin{ex}\label{exRiesz}
To illustrate this section, 
we give our favorite example: 
we are going to Riesz regularize
the function 
$\frac{1}{h^n}$ following the classical approach of \cite{Gelfand}. 
First, we regularize by the formula 
$$\int_0^1 \frac{d\lambda}{\lambda}\lambda^\mu \left\langle\frac{1}{h^{n}}\psi_{\lambda^{-1}},\varphi \right\rangle + \left\langle\frac{1}{h^{n}}(1-\chi),\varphi \right\rangle $$
where $\mu\in\mathbb{C}$. 
We shall concentrate only
on the term $\int_0^1 \frac{d\lambda}{\lambda}\lambda^\mu \left\langle\frac{1}{h^{n}}\psi_{\lambda^{-1}},\varphi \right\rangle$:
$$\int_0^1 \frac{d\lambda}{\lambda}\lambda^\mu \left\langle\frac{1}{h^{n}}\psi_{\lambda^{-1}},\varphi \right\rangle
=
\int_0^1 \frac{d\lambda}{\lambda}\lambda^{\mu-n+1} \left\langle\frac{1}{h^{n}}\psi,\varphi_\lambda \right\rangle$$
$$=\int_0^1 \frac{d\lambda}{\lambda}\sum_{k=0}^N\frac{\lambda^{\mu-n+1+k}}{k!} \left\langle\frac{1}{h^{n}}\psi,h^k\partial^k_h\varphi(0)\right\rangle
+\int_0^1\frac{d\lambda}{\lambda} \lambda^{\mu-n+1}
\left\langle\frac{1}{h^{n}}\psi,I_{N,\lambda}\right\rangle
$$
Then for $Re(\mu)$ small enough, 
we can integrate the first $N$ terms:
$$\int_0^1 \frac{d\lambda}{\lambda}\lambda^\mu \left\langle\frac{1}{h^{n}}\psi_{\lambda^{-1}},\varphi \right\rangle + \left\langle\frac{1}{h^{n}}(1-\chi),\varphi \right\rangle $$ 
$$=\sum_{k=0}^N\frac{1}{(\mu-n+1+k)k!} \left\langle\frac{1}{h^{n}}\psi,h^k\partial^k_h\varphi(0)\right\rangle
+\text{ nice terms } .$$
At $\mu=0$, when $k=n-1$, we have a pole
$\frac{1}{\mu(n-1)!} \left\langle\frac{1}{h}\psi,\partial^{n-1}_h\varphi(0)\right\rangle $ 
of the Laurent series, and subtracting it 
allows us to define
the regularization:
$$\lim_{\mu\rightarrow 0}\int_0^1 \frac{d\lambda}{\lambda}\lambda^\mu \left\langle\frac{1}{h^{n}}\psi_{\lambda^{-1}},\varphi \right\rangle-\frac{1}{\mu(n-1)!} \left\langle\frac{1}{h}\psi,\partial^{n-1}_h\varphi(0)\right\rangle+\left\langle\frac{1}{h^{n}}(1-\chi),\varphi \right\rangle.$$ 
\end{ex}
We recall that
if $T^\mu$ is meromorphic at $\mu=0$
then the pole has order at most $n$ and 
$T^\mu$ is holomorphic in $B_r(0)\setminus \{0\}$
for $r$ small enough (since the poles of $T^\mu$
are located in $-Spec(\Omega)-\mathbb{N}$), 
then
$T^\mu=\sum_{k=-n}^{+\infty} \mu^k T^k$ where 
$\forall k\in\mathbb{Z}, T^k=\frac{1}{2i\pi} \int_{\partial B_r(0)} \frac{d\mu}{\mu^{k+1}}T^\mu.$
\begin{defi}
Let $T\in \mathcal{D}_k^\prime(U\setminus I)$ and $T\in F_\Omega(U\setminus I)$.
%Let $T_{i}$ be the projection of $T$ on the generalized eigenspace $E_i$ of the matrix $\Omega$.
%Let $\Omega|_{E_i}=-d_i+N_i$ where $N_i$ is \emph{nilpotent} matrix of fixed order $n_i$.
Then $T^\mu$ is \textbf{meromorphic} in $\mu$ by Theorem \ref{meromthm}
and the Riesz regularization is defined as
%\begin{equation}
%\left\langle R_{Riesz}T,\omega \right\rangle=\lim_{\mu\rightarrow 0}\left(\left\langle T^\mu,\omega \right\rangle -\sum_{-d_i \in Spec(\Omega),j+l=d_i}(\mu+N_i)^{-1}\left\langle T_{ij}\wedge\omega_{l},\psi \right\rangle\right)+ \left\langle T(1-\chi),\omega \right\rangle
%\end{equation}
\begin{equation}
\left\langle R_{Riesz}T,\omega \right\rangle=\lim_{\mu\rightarrow 0}\left(\left\langle T^\mu,\omega \right\rangle -\sum_{k=-n}^{-1}\mu^k \left\langle T^k,\omega\right\rangle\right) + \left\langle T(1-\chi),\omega \right\rangle.
\end{equation}
\end{defi}  
It is not completely obvious from its definition 
that $R_{Riesz}$ defines an extension operator.
\begin{prop}
For all $T\in \mathcal{D}_k^\prime(U\setminus I)\cap F_\Omega(U\setminus I)$,
$R_{Riesz}T$ is an extension of $T$.
\end{prop}
\begin{proof}
Let $\omega$ be a test form supported in $U\setminus I$.
Then by lemma \ref{locpoles}, 
all poles of $\left\langle T^\mu,\omega \right\rangle$
vanish hence $\left\langle T^\mu,\omega \right\rangle$ 
is holomorphic in $\mu$ and $$\left\langle R_{Riesz}T,\omega \right\rangle
=\lim_{\mu\rightarrow 0}\left(\left\langle T^\mu,\omega \right\rangle -\sum_{k=-n}^{-1}\mu^k \left\langle T^k,\omega\right\rangle\right) + \left\langle T(1-\chi),\omega \right\rangle$$ $$=\lim_{\mu\rightarrow 0}\left(\left\langle T^\mu,\omega \right\rangle \right) 
+ \left\langle T(1-\chi),\omega \right\rangle
=\left\langle T\chi,\omega \right\rangle  + \left\langle T(1-\chi),\omega \right\rangle
=\left\langle T,\omega \right\rangle,$$
since $\lim_{\mu\rightarrow 0} \int_0^1 \frac{d\lambda}{\lambda} \lambda^\mu \left\langle T\psi_{\lambda^{-1}},\omega \right\rangle=\int_0^1 \frac{d\lambda}{\lambda}\left\langle T\psi_{\lambda^{-1}},\omega \right\rangle=\left\langle T\chi,\omega \right\rangle $.
\end{proof}
\paragraph{The anomalous scaling.}
Our next theorem 
is fundamental 
for quantum field theory
since it implies that 
if $T$ is a Fuchsian symbol
then its extension $R_{Riesz}T$ 
is also a Fuchsian symbol.
\begin{thm}
Let $\Omega\in M_n(\mathbb{C})$
where $Spec(\Omega)\in\mathbb{R}$.
For all $T\in \mathcal{D}_k^\prime(U\setminus I)\cap F_\Omega(U\setminus I)$,
if $(\rho-\Omega) T=0$
then
$R_{Riesz}T$ 
satisfies the equation 
$(\rho-\Omega)R_{Riesz}T=0$ when  $Spec(\Omega)\cap-\mathbb{N}=\emptyset$ and 
$(\rho-\Omega)^{n+1} R_{Riesz}T=0$ when $Spec(\Omega)\cap-\mathbb{N}\neq \emptyset$.
\end{thm}
\begin{proof}
By the proof of \ref{propimp}, 
we know that 
$(\rho-\Omega)T=0$ implies
\begin{equation}\label{equacle}
(\rho-\mu-\Omega)T^\mu=0 
\end{equation}
on some neighborhood $V$ of $I$ 
provided $V$ is stable by scaling 
and $\chi|_U=1$. 
Then the trick consists in replacing
$T^\mu$ by its Laurent series expansion 
in equation \ref{equacle}.
$$(\rho-\Omega-\mu)T^\mu=(\rho-\Omega-\mu)\left(\sum_{k=-n}^{+\infty} \mu^k T^k\right)  $$
\begin{equation}\label{equadiff}
=(\rho-\Omega-\mu)\left(\sum_{k=-n}^{-1} \mu^k T^k + T^0 + O(\mu) \right)=0.
\end{equation}
Notice that the constant term in the Laurent series expansion $T^0=R_{Riesz}T-T(1-\chi)$
therefore on $V$, we have $T^0=R_{Riesz}T$ since $1-\chi|_V=0$.
By \textbf{uniqueness of the Laurent series expansion}, we expand the equation 
(\ref{equadiff})
in powers of $\mu$:
$$(\rho-\Omega)T^{-n}\mu^{-n}+\sum_{k=-n+1}^{0} \mu^k\left((\rho-\Omega)T^{k}-T^{k-1}\right)+ O(\mu)=0$$
and we require that all coefficients of the Laurent series expansion
should vanish.
Hence we find a system of equations:
\begin{eqnarray}\label{syst}
(\rho-\Omega)T^{-n}=0
\\\forall k, -n+1\leqslant k\leqslant 0, \left((\rho-\Omega)T^{k}-T^{k-1}\right)=0.
\end{eqnarray} 
Then for $T^0=R_{Riesz}T$ on $V$, 
we have $(\rho-\Omega)T^0=(\rho-\Omega)R_{Riesz}T=T^{-1}$.
Also note that 
on the complement of $V$, $(\rho-\Omega)R_{Riesz}T=0$ since $R_{Riesz}T=T$ because $R_{Riesz}T$ is an extension of $T$. Thus we have globally
$(\rho-\Omega)R_{Riesz}T=T^{-1}$. Now the key fact is that if $Spec(\Omega)\cap-\mathbb{N}=\emptyset$
then
$T^{-1}=0$ since $T^\mu$ has no poles at $\mu=0$.
Finally, if $Spec(\Omega)\cap-\mathbb{N}\neq\emptyset$ then by an easy recursion:
$$(\rho-\Omega)^{n+1}R_{Riesz}T=(\rho-\Omega)^{n}T^{-1}=(\rho-\Omega)^{n-1}T^{-2}=\dots
=(\rho-\Omega)T^{-n}=0, $$
which is the final equation we wanted to find. 
\end{proof}
\begin{ex}
We pick again our example of $T=\frac{1}{h^n}$, its 
Riesz extension satisfies the differential equations
$$(\rho+n)R_{Riesz}T=\left\langle\frac{1}{h} ,\psi\right\rangle\frac{1}{(n-1)!}\partial_h^{n-1}\delta_0$$
and
$$(\rho+n)^2R_{Riesz}T=0. $$
\end{ex}
\paragraph{The residue equation.}
  
A small comment before we state anything. 
The role of the poles seems to disappear since 
we subtract them in order 
to define the Riesz regularization, 
however they come back with a revenge
when we compute the
residue or anomaly of the Riesz
regularization. Following 
the philosophy of Chapter 8,
we \textbf{define} the residues of $R_{Riesz}$
for the vector field $\rho$
by the simple equation:
$\mathfrak{Res}_\rho[T]= \rho(R_{Riesz} T)- R_{Riesz} (\rho T)$. 
\begin{thm}
Let $T\in F_\Omega(U\setminus I)$  
and $T^{-1}$ is the coefficient of $\mu^{-1}$ in the Laurent series
expansion of 
the 
meromorphic function 
$T^\mu$ around $\mu=0$. 
Then 
$R_{Riesz}$ satisfies the residue equation
\begin{equation}
\mathfrak{Res}_\rho[T]=T^{-1}.
%\sum_{-d_i\in Spec(\Omega),j+l=d_i} (\mu+N_i)^{-1}\left\langle T_{ij}\psi,\omega_{l} \right\rangle
\end{equation}
In 
particular the residue vanishes 
when $Spec(\Omega)\cap-\mathbb{N}
=\emptyset$.
\end{thm}
Comment: the residue $\mathfrak{Res}[T]$ is 
the \emph{holomorphic residue} 
of $T^\mu$ at $\mu=0$.
 
\begin{proof}
By Proposition (\ref{ordn}), the residue vanishes
if $-Spec(\Omega)\cap\mathbb{N}=\emptyset$ because $T_k^\mu$ admits no pole at $\mu=0$
thus $R_{Riesz}T_k$ satisfies the same equation as $T_k$, thus 
$(\rho-\Omega-k)R_{Riesz}T_k=0=\rho R_{Riesz}T_k-R_{Riesz}\rho T_k$.
If $k\in -Spec(\Omega)\cap\mathbb{N}$,
then by equation \ref{syst}, $\rho R_{Riesz}T_k- R_{Riesz}\rho T_k=(\rho-\Omega-k)R_{Riesz}T_k=T_k^{-1} $ which yields the result.
\end{proof}
\section{The $\log$ and the 1-parameter RG.}  
%  In the general case where $\Omega$ is not necessarily diagonalizable, we \textbf{cannot hope} that the poles are going to be $\chi$ independent.
%However, this is not a curse ! The nilpotent part $N_i$ will bring in $\log$ terms which are  responsible for the one parameter renormalization group (according to B Delamotte).
Let us fix $\rho$ and a current $T\in \mathcal{D}^\prime_k(U\setminus I)\cap  F_\Omega(U\setminus I)$. 
Once we fix the function $\chi$ and the Euler vector field $\rho$, we can renormalize following the Riesz extension since $T\in F_\Omega(U\setminus I)$, this is called choosing a \emph{renormalization scheme}.
% 
% We will describe this arbitrariness as follows.
%The difference beetween any extension procedure 
%is a continuous linear map from the given space of currents 
%with value currents supported on $I$.   
% 
% We need to refine the Hadamard extension 
%for currents  
%in $F_\Omega(U\setminus I)$:  
%\begin{defi}  
%
% Let $F_\Omega(U\setminus I)$ be the infinite dimensional space of approximate Fuchsian currents for 
%the Fuchsian equation $\rho-\Omega$ where $\Omega$ is assumed to be in the Jordan normal form. 
%
% For each decomposition of this space as a direct sum $F_\Omega(U\setminus I)=F^s_\Omega(U\setminus I)\oplus F^r_\Omega(U\setminus I)$ of a space of singular currents and regular currents where 
%$F^r_\Omega(U\setminus I)\subset E_s$ where $s>()$.
%
% Then we can compose the Hadamard scheme with projector on $F^s_\Omega(U\setminus I)$. The remainder is extended by continuity (without subtraction of counterterms).
% 
% For $T=T_s+R$ where $T_s\in F^s_\Omega(U\setminus I), R\in F^r_\Omega(U\setminus I)$:
% 
%\begin{equation}
%\left\langle R_{Hadamard}T,\omega \right\rangle=\lim_{\varepsilon\rightarrow 0}\left(\left\langle  T_s(\chi-\chi_{\varepsilon^{-1}}) , I_m(\omega)\right\rangle+\left\langle  R(1-\chi_{\varepsilon^{-1}}) , \omega\right\rangle\right) +\left\langle  T_s(1-\chi) , \omega \right\rangle}
%\end{equation} 
%  
%\end{defi}  
%  
But in contrary to the flat case, if we change the Euler field $\rho$ and the function $\chi$, we change the renormalization scheme, and the extensions will differ by a \textbf{local counterterm} which is a distribution supported on $I$. We thus have some infinite dimensional space of choices. But if $\chi,\rho$ and the extension $R_{Riesz}$ is choosed, 
then we still have a one dimensional degree of freedom left when we \textbf{scale the cut-off function} $\chi$ by the flow $\chi \mapsto e^{\rho\log\ell*} \chi,\ell \in\mathbb{R}^{+\star}$
which changes the \emph{length scale} of our renormalization.
The idea of scaling the function $\chi$ by the one parameter group $e^{\log\ell \rho}$ was inspired
by the reading of unpublished lecture notes of John Cardy \cite{Cardynotes} and \cite{Cardy} Chapter 5 section (5.2). 
The mechanism we are going to explain allows to relate the Bogoliubov, Epstein-Glaser technique with the  1-parameter renormalization group of Bogoliubov Shirkov.   
  
\begin{ex}
 Let us give some important comment on the physical meaning of the variable $\ell$ in the case where the manifold is a configuration space $M^2$ and $I=d_2$ is the diagonal of $M^2$. When $\ell\rightarrow \infty$, the function $\chi_\ell$ will have a \textbf{support shrinking}
to the diagonal $d_2$. This means that we must think of $\ell^{-1}$ in terms of characteristic length beetween pair of points $(x,y)\in M^2$ (think of them in terms of particles in the hard ball model, see \cite{Cardy} p.~88). Then according to this interpretation $\ell\rightarrow \infty$ should be called UV flow whereas $\ell\rightarrow 0$ is the IR flow. We describe the simple example of the amplitude
$\left\langle \phi^2(x)\phi^2(y) \right\rangle $ in the flat Euclidean case:
$$\begin{array}{|c|c|c|}
\hline \text{Cardy poor man's renorm}  & \text{Our approach} & \text{Costello Heat kernel}  \\
\hline \int_{M^2\setminus \{\vert x-y\vert\geqslant 
\ell \}}\Delta^2(x,y)g(x)g(y)d^4xd^4y &  \left\langle R^\ell\Delta^2 , g\otimes g\right\rangle & \frac{1}{2}\int_\ell^\infty \frac{dt}{t}t^2 \left\langle K_t , g\otimes g\right\rangle \\
\hline
\end{array}$$
In Costello's approach (\cite{Costello} (4.2) p.~43), 
$K_t$ is the Heat kernel and 
the UV regularized two point function in 
the massless case is 
given by the formula $\int_\ell^\infty dt K_t$ .
\end{ex}  
  
 Let $T$ be a given current $T\in \mathcal{D}^\prime_k(U\setminus I)$. For each function $\chi$ such that $\chi=1$ 
in a neighborhood of $I$ 
and vanishes 
outside a tubular neighborhood of $I$, we denote by
$R^\ell_{Riesz} $ the corresponding Riesz regularization operator 
constructed with $\chi_\ell$:
$$\left\langle R^\ell_{Riesz}T,\omega \right\rangle=\lim_{\mu\rightarrow 0}\left(1-\sum_{k=-n}^{-1}\int_{\partial B(0,r)}\frac{d\mu}{2i\pi\mu^{k+1}}\right)\int_0^1\frac{d\lambda}{\lambda}\lambda^\mu T\psi_{\ell\lambda^{-1}}  +T(1-\chi_\ell).$$
We shall state the renormalization group flow theorem for the Riesz regularization. The residue $\mathfrak{Res}$ 
appears when we scale the bump function $\chi$.
\begin{thm}
Let $T\in F_\Omega(U\setminus I)$ and $\forall\ell\in\mathbb{R}_{>0}$, the
residue
$\mathfrak{Res}_\rho[T](\ell)=\rho R_{Riesz}^{\ell}T- R_{Riesz}^{\ell}\rho T$. 
Then 
both 
$R_{Riesz}^{\ell},\mathfrak{Res}_\rho[T](\ell)$
satisfy the differential 
equations
\begin{eqnarray}
\ell\frac{d}{d\ell}R_{Riesz}^{\ell}T= \mathfrak{Res}_\rho[T](\ell)\\
\left(\ell\frac{d}{d\ell}\right)^{n}\mathfrak{Res}_\rho[T](\ell)
=0.
\end{eqnarray}
Thus $R_{Riesz}^{\ell}T$ 
scales like a polynomial
of $\log\ell$ of degree $n$:
\begin{equation}
R_{Riesz}^{\ell}T=R_{Riesz}^{1}T+\sum_{k=1}^n\frac{(\log\ell)^k}{k!}\left(\ell\frac{d}{d\ell} \right)^k\mathfrak{Res}_\rho[T](1)
\end{equation}
where the divergent part is a polynomial of degree $n$ in $\log\ell$ 
with coefficients local counterterms. 
\end{thm}
\begin{proof}
From the decomposition $T=\sum_0^\infty T_j$
where $\forall j, (\rho-\Omega-j)T_j=0$, by linearity of the Riesz extension and by the fact that $Res_\rho[T_j]$
vanishes for $j$ large enough,
we can
reduce the proof to an element $T\in F_\Omega(U\setminus I)$ killed by $\rho-\Omega$.
$$ \ell\frac{d}{d\ell}\left(T^{\mu,\ell} +  T(1-\chi_\ell)  \right)
= \ell\frac{d}{d\ell}T^{\mu,\ell}$$ $$ =\ell\frac{d}{d\ell} 
\int_0^1\frac{d\lambda}{\lambda}\lambda^\mu  T \psi_{\ell\lambda^{-1}}  =\int_0^1\frac{d\lambda}{\lambda}\lambda^\mu  T (\rho\psi)_{\lambda^{-1}\ell}  $$
$$=\int_0^1\frac{d\lambda}{\lambda}\lambda^\mu  \rho (T \psi)_{\lambda^{-1}\ell} -\int_0^1\frac{d\lambda}{\lambda}\lambda^\mu  (\rho T) \psi_{\lambda^{-1}\ell} $$
$$=\rho T^{\mu\ell}-\Omega T^{\mu,\ell}=(\Omega+\mu)T^{\mu,\ell}-\Omega T^{\mu\ell}=\mu T^{\mu\ell} .$$
We obtain the simple equation $\ell\frac{d}{d\ell}T^{\mu,\ell}=\mu T^{\mu,\ell}  $.
Expanding the l.h.s and the r.h.s. of this equation
in Laurent series
and identifying 
the different terms 
in the Laurent series expansion,
$$\sum_{k=-n}^{+\infty} \ell\frac{d}{d\ell}T^{k,\ell} \mu^k=\sum_{k=-n}^{+\infty}T^{k,\ell} \mu^{k+1}  $$
we deduce a system of linear equations: 
\begin{equation}
\forall k\geqslant -n+1,\ell\frac{d}{d\ell}T^{k,\ell}=T^{k-1,\ell}\text{ and }
\ell\frac{d}{d\ell}T^{-n,\ell}=0. 
\end{equation} 

But since $\ell\frac{d}{d\ell}T^{0,\ell}=\ell\frac{d}{d\ell}R_{Riesz}^{\ell}T$ and from the fact
that $\left(\ell\frac{d}{d\ell}\right)^{n+1}T^{0,\ell}=\left(\ell\frac{d}{d\ell}\right)^{n}T^{-1,\ell}=\left(\ell\frac{d}{d\ell}\right)^{n}\mathfrak{Res}_\rho[T](\ell)=\cdots=\ell\frac{d}{d\ell}T^{-n,\ell}=0 $, we must have $\left(\ell\frac{d}{d\ell}\right)^{n+1}R_{Riesz}^{\ell}T=0$ which implies
$R_{Riesz}^{\ell}T $ scales like a \emph{polynomial}
of $\log\ell$ 
of degree $n$.
\end{proof}

\bibliographystyle{plain}
\bibliography{biblio}
\end{document}